\documentclass[12pt]{article}
\pdfoutput=1
\usepackage{amsmath}
\usepackage{graphicx,psfrag,epsf}
\usepackage{enumerate}
\usepackage{newclude} 
\usepackage[square,sort]{natbib}
\usepackage{url} 

\newcommand{\blind}{1}

\RequirePackage[paper = a4paper,         
            hmargin = 25mm,          
            vmargin = 25mm,          
            dvipdfm]{geometry}

\clearpage{}
\RequirePackage[english]{babel}
\usepackage[autostyle=true,english=british]{csquotes}

\RequirePackage[]{color}
 \makeatletter \def\maxwidth{ 
   \ifdim\Gin@nat@width>\linewidth \linewidth \else \Gin@nat@width \fi
 } \makeatother

\definecolor{fgcolor}{rgb}{0, 0, 0}
\newcommand{\hlnum}[1]{\textcolor[rgb]{0.502,0,0.502}{\textbf{#1}}}
\newcommand{\hlstr}[1]{\textcolor[rgb]{0.651,0.522,0}{#1}}

\newcommand{\hlstd}[1]{\textcolor[rgb]{0,0,0}{#1}}

\newcommand{\hlkwb}[1]{\textcolor[rgb]{0.502,0.502,0.753}{\textbf{#1}}}
\newcommand{\hlkwc}[1]{\textcolor[rgb]{0,0.502,0.753}{#1}}
\newcommand{\hlkwd}[1]{\textcolor[rgb]{0,0.267,0.4}{#1}}

\RequirePackage{framed}
\makeatletter
\newenvironment{kframe}{
 \def\at@end@of@kframe{}
 \ifinner\ifhmode
  \def\at@end@of@kframe{\end{minipage}}
  \begin{minipage}{\columnwidth}
 \fi\fi
 \def\FrameCommand##1{\hskip\@totalleftmargin \hskip-\fboxsep
 \colorbox{shadecolor}{##1}\hskip-\fboxsep
 \hskip-\linewidth \hskip-\@totalleftmargin \hskip\columnwidth}
\MakeFramed {\advance\hsize-\width \@totalleftmargin\z@
  \linewidth\hsize \@setminipage}}
{\par\unskip\endMakeFramed
 \at@end@of@kframe}
\makeatother

\definecolor{shadecolor}{rgb}{.97, .97, .97}
\definecolor{messagecolor}{rgb}{0, 0, 0}
\definecolor{warningcolor}{rgb}{1, 0, 1}
\definecolor{errorcolor}{rgb}{1, 0, 0}
\newenvironment{knitrout}{}{}

\RequirePackage{alltt}

\RequirePackage{pdflscape}
\RequirePackage{mathtools}

\RequirePackage{fancyhdr}
\RequirePackage{datetime} 
\RequirePackage{last page}

\numberwithin{equation}{section}

\RequirePackage{verbatim}  
\RequirePackage{bbm}     
\RequirePackage[final]{showkeys} 
\RequirePackage{bm} 
\RequirePackage{xspace} 
\RequirePackage{framed}

\RequirePackage{lineno}  

\RequirePackage{enumitem} 
\setlist{nolistsep} 

\RequirePackage[
            colorlinks,      
            citecolor=blue,  
            linkcolor=magenta,  
            urlcolor=blue,   
            ]{hyperref}

\hypersetup{citecolor=blue, urlcolor=magenta, colorlinks=true} 

\RequirePackage[sf,medium,bf,noindentafter]{titlesec}  
\titlespacing{\section}{0pt}{12pt plus 4pt minus 2pt}{0pt plus 2pt minus 2pt}
\titlespacing{\subsection}{0pt}{12pt plus 4pt minus 2pt}{0pt plus 2pt minus 2pt}
\titlespacing{\subsubsection}{0pt}{12pt plus 4pt minus 2pt}{0pt plus 2pt minus 2pt}
\titlespacing{\paragraph}{0pt}{\parskip}{-\parskip}

\RequirePackage{tensor} 
\RequirePackage{etoolbox} 

\RequirePackage{xparse} 

\RequirePackage{trimclip} 
\RequirePackage{adjustbox}

\RequirePackage{fouridx}

\IfFileExists{upquote.sty}{\RequirePackage{upquote}}{}
\RequirePackage{amsthm}

\makeatletter
\newcommand\appendix@section[1]{
  \refstepcounter{section}
  \orig@section*{Appendix \@Alph\c@section: #1}
  \addcontentsline{toc}{section}{Appendix \@Alph\c@section: #1}
}
\let\orig@section\section
\g@addto@macro\appendix{\let\section\appendix@section}
\makeatother

\clearpage{}
\clearpage{}
\pdfoutput=1

\graphicspath{{figure/}}
\graphicspath{{/}{../}}

\makeatletter{}

\RequirePackage{amssymb}
\RequirePackage{amsfonts}

\RequirePackage{xstring}

\newcounter{definition}
\newcounter{corollary}

\newtheorem{theorem}{Theorem}[section]
\newtheorem{lemma}[theorem]{Lemma}
\newtheorem{corollary}[theorem]{Corollary}

\newtheorem{definition}[theorem]{Definition}

\let\oldLemma\lemma
\renewcommand{\lemma}{  \crefalias{theorem}{lemma}  \oldLemma}

\let\oldDefinition\definition
\renewcommand{\definition}{  \crefalias{theorem}{definition}  \oldDefinition}

\let\oldCorollary\corollary
\renewcommand{\corollary}{  \crefalias{theorem}{corollary}  \oldCorollary}

\RequirePackage[]{cleveref}

\newcommand{\myref}[2]{\mbox{\cref{#1}\ref{#2}}\xspace
}

\newcommand{\Myref}[2]{\mbox{\Cref{#1}\ref{#2}}\xspace
}

\newrobustcmd{\MacroNote}[1]{}

\MacroNote{Text written like this is contained in the macro MacroNote,
  and this is intended as a reminder concerning how to call the
  different commands. Might never be used...}

\newrobustcmd{\InternalNote}[2][Internal note:]{
  \rule{\linewidth}{.25pt} \\
  \textbf{#1} #2\\[-.25cm]
  \rule{\linewidth}{.25pt}
  \vspace*{\lineskip}
}

\newcommand{\E}[2][]{\ensuremath{
    \operatorname{E}_{#1}\! \left[#2\right]}}
\newcommand{\Prob}[1]{\ensuremath{
    \operatorname{P}\! \left(#1\right)}}
\newcommand{\Var}[2][]{\ensuremath{
    \operatorname{Var}_{#1}\!\left(#2\right)}}
\newcommand{\Cov}[2]{\ensuremath{
    \operatorname{Cov}\!\left(#1,#2\right)}}

\newcommand{\Ind}[1]{\ensuremath{ \mathbbm{1}{\left\{#1\right\}}}}

\renewcommand{\d}[1]{\operatorname{d}\!#1}

\newcommand{\UVN}[2]{\ensuremath{
    \operatorname{N}\!\left(#1,#2\right)}}

\newcommand{\R}{\texttt{R}\xspace}
\newcommand{\Rfunction}{\mbox{\R-function}\xspace}
\newcommand{\Rpackage}{\mbox{\R-package}\xspace}
\newcommand{\Rcode}{\mbox{\R code}\xspace}

\newcommand{\Rref}[1]{\mbox{\texttt{#1}\xspace}}

\newcommand{\lgsdRpackage}{\Rref{localgaussSpec}}
\newcommand{\devtoolgithublgsdRpackage}{\Rref{devtools:\!:install\_github("LAJordanger/localgaussSpec")}}

\newcommand{\NN}{\ensuremath{\mathbb{N}}} \newcommand{\RR}{\ensuremath{\mathbb{R}}}  \newcommand{\ZZ}{\ensuremath{\mathbb{Z}}}

\newcommand{\defarrow}{\stackrel{\textup{\tiny def}}{\Longleftrightarrow}}
\newcommand{\defeq}{\coloneqq}

\newcommand{\inv}{\ensuremath{{}^{-1}}}

\newcommand{\ssi}[1]{\mbox{\footnotesize
    \IfStrEq{#1}{}
    {}{\ensuremath{{}_{#1}}}}}

\newcommand{\floor}[1]{\ensuremath{\left\lfloor #1 \right\rfloor}}
\newcommand{\ceil}[1]{\ensuremath{\left\lceil #1 \right\rceil}}

\newcommand{\ssub}[3]{\mbox{\tiny
        \IfStrEq{#2}{}
                {}{\ensuremath{{#1#2#3}}}}} 
\newcommand{\ssup}[3]{\mbox{\tiny
        \IfStrEq{#2}{}
                {}{\ensuremath{{#1#2#3}}}}}

\newcommand{\SSI}[7]{\ensuremath{
        \tensor*{#1}{^{\ssup{#2}{#3}{#4}}_{\ssub{#5}{#6}{#7}}}}}

\newcommand{\pSSI}[9]{\ensuremath{
        \tensor*[^{\ssup{#8}{p}{#9}}]{#1}{^{\ssup{#2}{#3}{#4}}_{\ssub{#5}{#6}{#7}}}}}

\newcommand{\fs}[3][]{\SSI{f}{}{#1}{}{}{#2}{}
        \IfStrEq{#3}{}{}{\left(#3\right)}}
\newcommand{\Fs}[3][]{\SSI{F}{}{#1}{}{}{#2}{}
        \IfStrEq{#3}{}{}{\left(#3\right)}}
\newcommand{\fbs}[3][]{\SSI{\bm{f}}{}{#1}{}{}{#2}{}
        \IfStrEq{#3}{}{}{\left(#3\right)}}
\newcommand{\Fbs}[3][]{\SSI{\bm{F}}{}{#1}{}{}{#2}{}
        \IfStrEq{#3}{}{}{\!\! \left(#3\right)}}

\newcommand{\cs}[3][]{\SSI{c}{}{#1}{}{}{#2}{}
        \IfStrEq{#3}{}{}{\left(#3\right)}}
\newcommand{\Cs}[3][]{\SSI{C}{}{#1}{}{}{#2}{}
        \IfStrEq{#3}{}{}{\left(#3\right)}}
\newcommand{\cbs}[3][]{\SSI{\bm{c}}{}{#1}{}{}{#2}{}
        \IfStrEq{#3}{}{}{\left(#3\right)}}
\newcommand{\Cbs}[3][]{\SSI{\bm{C}}{}{#1}{}{}{#2}{}
        \IfStrEq{#3}{}{}{\left(#3\right)}}

\newcommand{\pFs}[3][]{\pSSI{F}{}{#1}{}{}{#2}{}{}{}
        \IfStrEq{#3}{}{}{\left(#3\right)}}

\newrobustcmd{\subp}[5]{ 
      \def\somethingThere{#2#3#4#5}
  \IfStrEq{\somethingThere}{}{#1}{
    \ensuremath{#1_{\scriptscriptstyle #2#3}^{        \scriptscriptstyle #4#5}}}
}

\newrobustcmd{\partialz}[2][]{
  \ensuremath{\subp{\partial}{}{#2}{}{#1}}}

\newrobustcmd{\onez}[2][]{
  \ensuremath{\subp{1}{}{#2}{}{#1}}}
\newrobustcmd{\bmonez}[2][]{
  \ensuremath{\subp{\bm{1}}{}{#2}{}{#1}}}
\newrobustcmd{\zeroz}[2][]{
  \ensuremath{\subp{0}{}{#2}{}{#1}}}
\newrobustcmd{\bmzeroz}[2][]{
  \ensuremath{\subp{\bm{0}}{}{#2}{}{#1}}}

\newrobustcmd{\az}[2][]{\ensuremath{\subp{a}{}{#2}{}{#1}}}
\newrobustcmd{\hataz}[2][]{\ensuremath{\subp{\hat{a}}{}{#2}{}{#1}}}
\newrobustcmd{\widehataz}[2][]{\ensuremath{\subp{\widehat{a}}{}{#2}{}{#1}}}
\newrobustcmd{\checkaz}[2][]{\ensuremath{\subp{\check{a}}{}{#2}{}{#1}}}
\newrobustcmd{\tildeaz}[2][]{\ensuremath{\subp{\tilde{a}}{}{#2}{}{#1}}}
\newrobustcmd{\widetildeaz}[2][]{\ensuremath{\subp{\widetilde{a}}{}{#2}{}{#1}}}
\newrobustcmd{\acuteaz}[2][]{\ensuremath{\subp{\acute{a}}{}{#2}{}{#1}}}
\newrobustcmd{\graveaz}[2][]{\ensuremath{\subp{\grave{a}}{}{#2}{}{#1}}}
\newrobustcmd{\dotaz}[2][]{\ensuremath{\subp{\dot{a}}{}{#2}{}{#1}}}
\newrobustcmd{\ddotaz}[2][]{\ensuremath{\subp{\ddot{a}}{}{#2}{}{#1}}}
\newrobustcmd{\breveaz}[2][]{\ensuremath{\subp{\breve{a}}{}{#2}{}{#1}}}
\newrobustcmd{\baraz}[2][]{\ensuremath{\subp{\bar{a}}{}{#2}{}{#1}}}
\newrobustcmd{\vecaz}[2][]{\ensuremath{\subp{\vec{a}}{}{#2}{}{#1}}}
\newrobustcmd{\bmaz}[2][]{\ensuremath{\subp{\bm{a}}{}{#2}{}{#1}}}
\newrobustcmd{\hatbmaz}[2][]{\ensuremath{\subp{\hat{\bm{a}}}{}{#2}{}{#1}}}
\newrobustcmd{\widehatbmaz}[2][]{\ensuremath{\subp{\widehat{\bm{a}}}{}{#2}{}{#1}}}
\newrobustcmd{\checkbmaz}[2][]{\ensuremath{\subp{\check{\bm{a}}}{}{#2}{}{#1}}}
\newrobustcmd{\tildebmaz}[2][]{\ensuremath{\subp{\tilde{\bm{a}}}{}{#2}{}{#1}}}
\newrobustcmd{\widetildebmaz}[2][]{\ensuremath{\subp{\widetilde{\bm{a}}}{}{#2}{}{#1}}}
\newrobustcmd{\acutebmaz}[2][]{\ensuremath{\subp{\acute{\bm{a}}}{}{#2}{}{#1}}}
\newrobustcmd{\gravebmaz}[2][]{\ensuremath{\subp{\grave{\bm{a}}}{}{#2}{}{#1}}}
\newrobustcmd{\dotbmaz}[2][]{\ensuremath{\subp{\dot{\bm{a}}}{}{#2}{}{#1}}}
\newrobustcmd{\ddotbmaz}[2][]{\ensuremath{\subp{\ddot{\bm{a}}}{}{#2}{}{#1}}}
\newrobustcmd{\brevebmaz}[2][]{\ensuremath{\subp{\breve{\bm{a}}}{}{#2}{}{#1}}}
\newrobustcmd{\barbmaz}[2][]{\ensuremath{\subp{\bar{\bm{a}}}{}{#2}{}{#1}}}
\newrobustcmd{\vecbmaz}[2][]{\ensuremath{\subp{\vec{\bm{a}}}{}{#2}{}{#1}}}
\newrobustcmd{\bz}[2][]{\ensuremath{\subp{b}{}{#2}{}{#1}}}
\newrobustcmd{\hatbz}[2][]{\ensuremath{\subp{\hat{b}}{}{#2}{}{#1}}}
\newrobustcmd{\widehatbz}[2][]{\ensuremath{\subp{\widehat{b}}{}{#2}{}{#1}}}
\newrobustcmd{\checkbz}[2][]{\ensuremath{\subp{\check{b}}{}{#2}{}{#1}}}
\newrobustcmd{\tildebz}[2][]{\ensuremath{\subp{\tilde{b}}{}{#2}{}{#1}}}
\newrobustcmd{\widetildebz}[2][]{\ensuremath{\subp{\widetilde{b}}{}{#2}{}{#1}}}
\newrobustcmd{\acutebz}[2][]{\ensuremath{\subp{\acute{b}}{}{#2}{}{#1}}}
\newrobustcmd{\gravebz}[2][]{\ensuremath{\subp{\grave{b}}{}{#2}{}{#1}}}
\newrobustcmd{\dotbz}[2][]{\ensuremath{\subp{\dot{b}}{}{#2}{}{#1}}}
\newrobustcmd{\ddotbz}[2][]{\ensuremath{\subp{\ddot{b}}{}{#2}{}{#1}}}
\newrobustcmd{\brevebz}[2][]{\ensuremath{\subp{\breve{b}}{}{#2}{}{#1}}}
\newrobustcmd{\barbz}[2][]{\ensuremath{\subp{\bar{b}}{}{#2}{}{#1}}}
\newrobustcmd{\vecbz}[2][]{\ensuremath{\subp{\vec{b}}{}{#2}{}{#1}}}
\newrobustcmd{\bmbz}[2][]{\ensuremath{\subp{\bm{b}}{}{#2}{}{#1}}}
\newrobustcmd{\hatbmbz}[2][]{\ensuremath{\subp{\hat{\bm{b}}}{}{#2}{}{#1}}}
\newrobustcmd{\widehatbmbz}[2][]{\ensuremath{\subp{\widehat{\bm{b}}}{}{#2}{}{#1}}}
\newrobustcmd{\checkbmbz}[2][]{\ensuremath{\subp{\check{\bm{b}}}{}{#2}{}{#1}}}
\newrobustcmd{\tildebmbz}[2][]{\ensuremath{\subp{\tilde{\bm{b}}}{}{#2}{}{#1}}}
\newrobustcmd{\widetildebmbz}[2][]{\ensuremath{\subp{\widetilde{\bm{b}}}{}{#2}{}{#1}}}
\newrobustcmd{\acutebmbz}[2][]{\ensuremath{\subp{\acute{\bm{b}}}{}{#2}{}{#1}}}
\newrobustcmd{\gravebmbz}[2][]{\ensuremath{\subp{\grave{\bm{b}}}{}{#2}{}{#1}}}
\newrobustcmd{\dotbmbz}[2][]{\ensuremath{\subp{\dot{\bm{b}}}{}{#2}{}{#1}}}
\newrobustcmd{\ddotbmbz}[2][]{\ensuremath{\subp{\ddot{\bm{b}}}{}{#2}{}{#1}}}
\newrobustcmd{\brevebmbz}[2][]{\ensuremath{\subp{\breve{\bm{b}}}{}{#2}{}{#1}}}
\newrobustcmd{\barbmbz}[2][]{\ensuremath{\subp{\bar{\bm{b}}}{}{#2}{}{#1}}}
\newrobustcmd{\vecbmbz}[2][]{\ensuremath{\subp{\vec{\bm{b}}}{}{#2}{}{#1}}}
\newrobustcmd{\cz}[2][]{\ensuremath{\subp{c}{}{#2}{}{#1}}}
\newrobustcmd{\hatcz}[2][]{\ensuremath{\subp{\hat{c}}{}{#2}{}{#1}}}
\newrobustcmd{\widehatcz}[2][]{\ensuremath{\subp{\widehat{c}}{}{#2}{}{#1}}}
\newrobustcmd{\checkcz}[2][]{\ensuremath{\subp{\check{c}}{}{#2}{}{#1}}}
\newrobustcmd{\tildecz}[2][]{\ensuremath{\subp{\tilde{c}}{}{#2}{}{#1}}}
\newrobustcmd{\widetildecz}[2][]{\ensuremath{\subp{\widetilde{c}}{}{#2}{}{#1}}}
\newrobustcmd{\acutecz}[2][]{\ensuremath{\subp{\acute{c}}{}{#2}{}{#1}}}
\newrobustcmd{\gravecz}[2][]{\ensuremath{\subp{\grave{c}}{}{#2}{}{#1}}}
\newrobustcmd{\dotcz}[2][]{\ensuremath{\subp{\dot{c}}{}{#2}{}{#1}}}
\newrobustcmd{\ddotcz}[2][]{\ensuremath{\subp{\ddot{c}}{}{#2}{}{#1}}}
\newrobustcmd{\brevecz}[2][]{\ensuremath{\subp{\breve{c}}{}{#2}{}{#1}}}
\newrobustcmd{\barcz}[2][]{\ensuremath{\subp{\bar{c}}{}{#2}{}{#1}}}
\newrobustcmd{\veccz}[2][]{\ensuremath{\subp{\vec{c}}{}{#2}{}{#1}}}
\newrobustcmd{\bmcz}[2][]{\ensuremath{\subp{\bm{c}}{}{#2}{}{#1}}}
\newrobustcmd{\hatbmcz}[2][]{\ensuremath{\subp{\hat{\bm{c}}}{}{#2}{}{#1}}}
\newrobustcmd{\widehatbmcz}[2][]{\ensuremath{\subp{\widehat{\bm{c}}}{}{#2}{}{#1}}}
\newrobustcmd{\checkbmcz}[2][]{\ensuremath{\subp{\check{\bm{c}}}{}{#2}{}{#1}}}
\newrobustcmd{\tildebmcz}[2][]{\ensuremath{\subp{\tilde{\bm{c}}}{}{#2}{}{#1}}}
\newrobustcmd{\widetildebmcz}[2][]{\ensuremath{\subp{\widetilde{\bm{c}}}{}{#2}{}{#1}}}
\newrobustcmd{\acutebmcz}[2][]{\ensuremath{\subp{\acute{\bm{c}}}{}{#2}{}{#1}}}
\newrobustcmd{\gravebmcz}[2][]{\ensuremath{\subp{\grave{\bm{c}}}{}{#2}{}{#1}}}
\newrobustcmd{\dotbmcz}[2][]{\ensuremath{\subp{\dot{\bm{c}}}{}{#2}{}{#1}}}
\newrobustcmd{\ddotbmcz}[2][]{\ensuremath{\subp{\ddot{\bm{c}}}{}{#2}{}{#1}}}
\newrobustcmd{\brevebmcz}[2][]{\ensuremath{\subp{\breve{\bm{c}}}{}{#2}{}{#1}}}
\newrobustcmd{\barbmcz}[2][]{\ensuremath{\subp{\bar{\bm{c}}}{}{#2}{}{#1}}}
\newrobustcmd{\vecbmcz}[2][]{\ensuremath{\subp{\vec{\bm{c}}}{}{#2}{}{#1}}}
\newrobustcmd{\dz}[2][]{\ensuremath{\subp{d}{}{#2}{}{#1}}}
\newrobustcmd{\hatdz}[2][]{\ensuremath{\subp{\hat{d}}{}{#2}{}{#1}}}
\newrobustcmd{\widehatdz}[2][]{\ensuremath{\subp{\widehat{d}}{}{#2}{}{#1}}}
\newrobustcmd{\checkdz}[2][]{\ensuremath{\subp{\check{d}}{}{#2}{}{#1}}}
\newrobustcmd{\tildedz}[2][]{\ensuremath{\subp{\tilde{d}}{}{#2}{}{#1}}}
\newrobustcmd{\widetildedz}[2][]{\ensuremath{\subp{\widetilde{d}}{}{#2}{}{#1}}}
\newrobustcmd{\acutedz}[2][]{\ensuremath{\subp{\acute{d}}{}{#2}{}{#1}}}
\newrobustcmd{\gravedz}[2][]{\ensuremath{\subp{\grave{d}}{}{#2}{}{#1}}}
\newrobustcmd{\dotdz}[2][]{\ensuremath{\subp{\dot{d}}{}{#2}{}{#1}}}
\newrobustcmd{\ddotdz}[2][]{\ensuremath{\subp{\ddot{d}}{}{#2}{}{#1}}}
\newrobustcmd{\brevedz}[2][]{\ensuremath{\subp{\breve{d}}{}{#2}{}{#1}}}
\newrobustcmd{\bardz}[2][]{\ensuremath{\subp{\bar{d}}{}{#2}{}{#1}}}
\newrobustcmd{\vecdz}[2][]{\ensuremath{\subp{\vec{d}}{}{#2}{}{#1}}}
\newrobustcmd{\bmdz}[2][]{\ensuremath{\subp{\bm{d}}{}{#2}{}{#1}}}
\newrobustcmd{\hatbmdz}[2][]{\ensuremath{\subp{\hat{\bm{d}}}{}{#2}{}{#1}}}
\newrobustcmd{\widehatbmdz}[2][]{\ensuremath{\subp{\widehat{\bm{d}}}{}{#2}{}{#1}}}
\newrobustcmd{\checkbmdz}[2][]{\ensuremath{\subp{\check{\bm{d}}}{}{#2}{}{#1}}}
\newrobustcmd{\tildebmdz}[2][]{\ensuremath{\subp{\tilde{\bm{d}}}{}{#2}{}{#1}}}
\newrobustcmd{\widetildebmdz}[2][]{\ensuremath{\subp{\widetilde{\bm{d}}}{}{#2}{}{#1}}}
\newrobustcmd{\acutebmdz}[2][]{\ensuremath{\subp{\acute{\bm{d}}}{}{#2}{}{#1}}}
\newrobustcmd{\gravebmdz}[2][]{\ensuremath{\subp{\grave{\bm{d}}}{}{#2}{}{#1}}}
\newrobustcmd{\dotbmdz}[2][]{\ensuremath{\subp{\dot{\bm{d}}}{}{#2}{}{#1}}}
\newrobustcmd{\ddotbmdz}[2][]{\ensuremath{\subp{\ddot{\bm{d}}}{}{#2}{}{#1}}}
\newrobustcmd{\brevebmdz}[2][]{\ensuremath{\subp{\breve{\bm{d}}}{}{#2}{}{#1}}}
\newrobustcmd{\barbmdz}[2][]{\ensuremath{\subp{\bar{\bm{d}}}{}{#2}{}{#1}}}
\newrobustcmd{\vecbmdz}[2][]{\ensuremath{\subp{\vec{\bm{d}}}{}{#2}{}{#1}}}
\newrobustcmd{\ez}[2][]{\ensuremath{\subp{e}{}{#2}{}{#1}}}
\newrobustcmd{\hatez}[2][]{\ensuremath{\subp{\hat{e}}{}{#2}{}{#1}}}
\newrobustcmd{\widehatez}[2][]{\ensuremath{\subp{\widehat{e}}{}{#2}{}{#1}}}
\newrobustcmd{\checkez}[2][]{\ensuremath{\subp{\check{e}}{}{#2}{}{#1}}}
\newrobustcmd{\tildeez}[2][]{\ensuremath{\subp{\tilde{e}}{}{#2}{}{#1}}}
\newrobustcmd{\widetildeez}[2][]{\ensuremath{\subp{\widetilde{e}}{}{#2}{}{#1}}}
\newrobustcmd{\acuteez}[2][]{\ensuremath{\subp{\acute{e}}{}{#2}{}{#1}}}
\newrobustcmd{\graveez}[2][]{\ensuremath{\subp{\grave{e}}{}{#2}{}{#1}}}
\newrobustcmd{\dotez}[2][]{\ensuremath{\subp{\dot{e}}{}{#2}{}{#1}}}
\newrobustcmd{\ddotez}[2][]{\ensuremath{\subp{\ddot{e}}{}{#2}{}{#1}}}
\newrobustcmd{\breveez}[2][]{\ensuremath{\subp{\breve{e}}{}{#2}{}{#1}}}
\newrobustcmd{\barez}[2][]{\ensuremath{\subp{\bar{e}}{}{#2}{}{#1}}}
\newrobustcmd{\vecez}[2][]{\ensuremath{\subp{\vec{e}}{}{#2}{}{#1}}}
\newrobustcmd{\bmez}[2][]{\ensuremath{\subp{\bm{e}}{}{#2}{}{#1}}}
\newrobustcmd{\hatbmez}[2][]{\ensuremath{\subp{\hat{\bm{e}}}{}{#2}{}{#1}}}
\newrobustcmd{\widehatbmez}[2][]{\ensuremath{\subp{\widehat{\bm{e}}}{}{#2}{}{#1}}}
\newrobustcmd{\checkbmez}[2][]{\ensuremath{\subp{\check{\bm{e}}}{}{#2}{}{#1}}}
\newrobustcmd{\tildebmez}[2][]{\ensuremath{\subp{\tilde{\bm{e}}}{}{#2}{}{#1}}}
\newrobustcmd{\widetildebmez}[2][]{\ensuremath{\subp{\widetilde{\bm{e}}}{}{#2}{}{#1}}}
\newrobustcmd{\acutebmez}[2][]{\ensuremath{\subp{\acute{\bm{e}}}{}{#2}{}{#1}}}
\newrobustcmd{\gravebmez}[2][]{\ensuremath{\subp{\grave{\bm{e}}}{}{#2}{}{#1}}}
\newrobustcmd{\dotbmez}[2][]{\ensuremath{\subp{\dot{\bm{e}}}{}{#2}{}{#1}}}
\newrobustcmd{\ddotbmez}[2][]{\ensuremath{\subp{\ddot{\bm{e}}}{}{#2}{}{#1}}}
\newrobustcmd{\brevebmez}[2][]{\ensuremath{\subp{\breve{\bm{e}}}{}{#2}{}{#1}}}
\newrobustcmd{\barbmez}[2][]{\ensuremath{\subp{\bar{\bm{e}}}{}{#2}{}{#1}}}
\newrobustcmd{\vecbmez}[2][]{\ensuremath{\subp{\vec{\bm{e}}}{}{#2}{}{#1}}}
\newrobustcmd{\fz}[2][]{\ensuremath{\subp{f}{}{#2}{}{#1}}}
\newrobustcmd{\hatfz}[2][]{\ensuremath{\subp{\hat{f}}{}{#2}{}{#1}}}
\newrobustcmd{\widehatfz}[2][]{\ensuremath{\subp{\widehat{f}}{}{#2}{}{#1}}}
\newrobustcmd{\checkfz}[2][]{\ensuremath{\subp{\check{f}}{}{#2}{}{#1}}}
\newrobustcmd{\tildefz}[2][]{\ensuremath{\subp{\tilde{f}}{}{#2}{}{#1}}}
\newrobustcmd{\widetildefz}[2][]{\ensuremath{\subp{\widetilde{f}}{}{#2}{}{#1}}}
\newrobustcmd{\acutefz}[2][]{\ensuremath{\subp{\acute{f}}{}{#2}{}{#1}}}
\newrobustcmd{\gravefz}[2][]{\ensuremath{\subp{\grave{f}}{}{#2}{}{#1}}}
\newrobustcmd{\dotfz}[2][]{\ensuremath{\subp{\dot{f}}{}{#2}{}{#1}}}
\newrobustcmd{\ddotfz}[2][]{\ensuremath{\subp{\ddot{f}}{}{#2}{}{#1}}}
\newrobustcmd{\brevefz}[2][]{\ensuremath{\subp{\breve{f}}{}{#2}{}{#1}}}
\newrobustcmd{\barfz}[2][]{\ensuremath{\subp{\bar{f}}{}{#2}{}{#1}}}
\newrobustcmd{\vecfz}[2][]{\ensuremath{\subp{\vec{f}}{}{#2}{}{#1}}}
\newrobustcmd{\bmfz}[2][]{\ensuremath{\subp{\bm{f}}{}{#2}{}{#1}}}
\newrobustcmd{\hatbmfz}[2][]{\ensuremath{\subp{\hat{\bm{f}}}{}{#2}{}{#1}}}
\newrobustcmd{\widehatbmfz}[2][]{\ensuremath{\subp{\widehat{\bm{f}}}{}{#2}{}{#1}}}
\newrobustcmd{\checkbmfz}[2][]{\ensuremath{\subp{\check{\bm{f}}}{}{#2}{}{#1}}}
\newrobustcmd{\tildebmfz}[2][]{\ensuremath{\subp{\tilde{\bm{f}}}{}{#2}{}{#1}}}
\newrobustcmd{\widetildebmfz}[2][]{\ensuremath{\subp{\widetilde{\bm{f}}}{}{#2}{}{#1}}}
\newrobustcmd{\acutebmfz}[2][]{\ensuremath{\subp{\acute{\bm{f}}}{}{#2}{}{#1}}}
\newrobustcmd{\gravebmfz}[2][]{\ensuremath{\subp{\grave{\bm{f}}}{}{#2}{}{#1}}}
\newrobustcmd{\dotbmfz}[2][]{\ensuremath{\subp{\dot{\bm{f}}}{}{#2}{}{#1}}}
\newrobustcmd{\ddotbmfz}[2][]{\ensuremath{\subp{\ddot{\bm{f}}}{}{#2}{}{#1}}}
\newrobustcmd{\brevebmfz}[2][]{\ensuremath{\subp{\breve{\bm{f}}}{}{#2}{}{#1}}}
\newrobustcmd{\barbmfz}[2][]{\ensuremath{\subp{\bar{\bm{f}}}{}{#2}{}{#1}}}
\newrobustcmd{\vecbmfz}[2][]{\ensuremath{\subp{\vec{\bm{f}}}{}{#2}{}{#1}}}
\newrobustcmd{\gz}[2][]{\ensuremath{\subp{g}{}{#2}{}{#1}}}
\newrobustcmd{\hatgz}[2][]{\ensuremath{\subp{\hat{g}}{}{#2}{}{#1}}}
\newrobustcmd{\widehatgz}[2][]{\ensuremath{\subp{\widehat{g}}{}{#2}{}{#1}}}
\newrobustcmd{\checkgz}[2][]{\ensuremath{\subp{\check{g}}{}{#2}{}{#1}}}
\newrobustcmd{\tildegz}[2][]{\ensuremath{\subp{\tilde{g}}{}{#2}{}{#1}}}
\newrobustcmd{\widetildegz}[2][]{\ensuremath{\subp{\widetilde{g}}{}{#2}{}{#1}}}
\newrobustcmd{\acutegz}[2][]{\ensuremath{\subp{\acute{g}}{}{#2}{}{#1}}}
\newrobustcmd{\gravegz}[2][]{\ensuremath{\subp{\grave{g}}{}{#2}{}{#1}}}
\newrobustcmd{\dotgz}[2][]{\ensuremath{\subp{\dot{g}}{}{#2}{}{#1}}}
\newrobustcmd{\ddotgz}[2][]{\ensuremath{\subp{\ddot{g}}{}{#2}{}{#1}}}
\newrobustcmd{\brevegz}[2][]{\ensuremath{\subp{\breve{g}}{}{#2}{}{#1}}}
\newrobustcmd{\bargz}[2][]{\ensuremath{\subp{\bar{g}}{}{#2}{}{#1}}}
\newrobustcmd{\vecgz}[2][]{\ensuremath{\subp{\vec{g}}{}{#2}{}{#1}}}
\newrobustcmd{\bmgz}[2][]{\ensuremath{\subp{\bm{g}}{}{#2}{}{#1}}}
\newrobustcmd{\hatbmgz}[2][]{\ensuremath{\subp{\hat{\bm{g}}}{}{#2}{}{#1}}}
\newrobustcmd{\widehatbmgz}[2][]{\ensuremath{\subp{\widehat{\bm{g}}}{}{#2}{}{#1}}}
\newrobustcmd{\checkbmgz}[2][]{\ensuremath{\subp{\check{\bm{g}}}{}{#2}{}{#1}}}
\newrobustcmd{\tildebmgz}[2][]{\ensuremath{\subp{\tilde{\bm{g}}}{}{#2}{}{#1}}}
\newrobustcmd{\widetildebmgz}[2][]{\ensuremath{\subp{\widetilde{\bm{g}}}{}{#2}{}{#1}}}
\newrobustcmd{\acutebmgz}[2][]{\ensuremath{\subp{\acute{\bm{g}}}{}{#2}{}{#1}}}
\newrobustcmd{\gravebmgz}[2][]{\ensuremath{\subp{\grave{\bm{g}}}{}{#2}{}{#1}}}
\newrobustcmd{\dotbmgz}[2][]{\ensuremath{\subp{\dot{\bm{g}}}{}{#2}{}{#1}}}
\newrobustcmd{\ddotbmgz}[2][]{\ensuremath{\subp{\ddot{\bm{g}}}{}{#2}{}{#1}}}
\newrobustcmd{\brevebmgz}[2][]{\ensuremath{\subp{\breve{\bm{g}}}{}{#2}{}{#1}}}
\newrobustcmd{\barbmgz}[2][]{\ensuremath{\subp{\bar{\bm{g}}}{}{#2}{}{#1}}}
\newrobustcmd{\vecbmgz}[2][]{\ensuremath{\subp{\vec{\bm{g}}}{}{#2}{}{#1}}}
\newrobustcmd{\hz}[2][]{\ensuremath{\subp{h}{}{#2}{}{#1}}}
\newrobustcmd{\hathz}[2][]{\ensuremath{\subp{\hat{h}}{}{#2}{}{#1}}}
\newrobustcmd{\widehathz}[2][]{\ensuremath{\subp{\widehat{h}}{}{#2}{}{#1}}}
\newrobustcmd{\checkhz}[2][]{\ensuremath{\subp{\check{h}}{}{#2}{}{#1}}}
\newrobustcmd{\tildehz}[2][]{\ensuremath{\subp{\tilde{h}}{}{#2}{}{#1}}}
\newrobustcmd{\widetildehz}[2][]{\ensuremath{\subp{\widetilde{h}}{}{#2}{}{#1}}}
\newrobustcmd{\acutehz}[2][]{\ensuremath{\subp{\acute{h}}{}{#2}{}{#1}}}
\newrobustcmd{\gravehz}[2][]{\ensuremath{\subp{\grave{h}}{}{#2}{}{#1}}}
\newrobustcmd{\dothz}[2][]{\ensuremath{\subp{\dot{h}}{}{#2}{}{#1}}}
\newrobustcmd{\ddothz}[2][]{\ensuremath{\subp{\ddot{h}}{}{#2}{}{#1}}}
\newrobustcmd{\brevehz}[2][]{\ensuremath{\subp{\breve{h}}{}{#2}{}{#1}}}
\newrobustcmd{\barhz}[2][]{\ensuremath{\subp{\bar{h}}{}{#2}{}{#1}}}
\newrobustcmd{\vechz}[2][]{\ensuremath{\subp{\vec{h}}{}{#2}{}{#1}}}
\newrobustcmd{\bmhz}[2][]{\ensuremath{\subp{\bm{h}}{}{#2}{}{#1}}}
\newrobustcmd{\hatbmhz}[2][]{\ensuremath{\subp{\hat{\bm{h}}}{}{#2}{}{#1}}}
\newrobustcmd{\widehatbmhz}[2][]{\ensuremath{\subp{\widehat{\bm{h}}}{}{#2}{}{#1}}}
\newrobustcmd{\checkbmhz}[2][]{\ensuremath{\subp{\check{\bm{h}}}{}{#2}{}{#1}}}
\newrobustcmd{\tildebmhz}[2][]{\ensuremath{\subp{\tilde{\bm{h}}}{}{#2}{}{#1}}}
\newrobustcmd{\widetildebmhz}[2][]{\ensuremath{\subp{\widetilde{\bm{h}}}{}{#2}{}{#1}}}
\newrobustcmd{\acutebmhz}[2][]{\ensuremath{\subp{\acute{\bm{h}}}{}{#2}{}{#1}}}
\newrobustcmd{\gravebmhz}[2][]{\ensuremath{\subp{\grave{\bm{h}}}{}{#2}{}{#1}}}
\newrobustcmd{\dotbmhz}[2][]{\ensuremath{\subp{\dot{\bm{h}}}{}{#2}{}{#1}}}
\newrobustcmd{\ddotbmhz}[2][]{\ensuremath{\subp{\ddot{\bm{h}}}{}{#2}{}{#1}}}
\newrobustcmd{\brevebmhz}[2][]{\ensuremath{\subp{\breve{\bm{h}}}{}{#2}{}{#1}}}
\newrobustcmd{\barbmhz}[2][]{\ensuremath{\subp{\bar{\bm{h}}}{}{#2}{}{#1}}}
\newrobustcmd{\vecbmhz}[2][]{\ensuremath{\subp{\vec{\bm{h}}}{}{#2}{}{#1}}}
\newrobustcmd{\iz}[2][]{\ensuremath{\subp{i}{}{#2}{}{#1}}}
\newrobustcmd{\hatiz}[2][]{\ensuremath{\subp{\hat{\imath}}{}{#2}{}{#1}}}
\newrobustcmd{\widehatiz}[2][]{\ensuremath{\subp{\widehat{\imath}}{}{#2}{}{#1}}}
\newrobustcmd{\checkiz}[2][]{\ensuremath{\subp{\check{\imath}}{}{#2}{}{#1}}}
\newrobustcmd{\tildeiz}[2][]{\ensuremath{\subp{\tilde{\imath}}{}{#2}{}{#1}}}
\newrobustcmd{\widetildeiz}[2][]{\ensuremath{\subp{\widetilde{\imath}}{}{#2}{}{#1}}}
\newrobustcmd{\acuteiz}[2][]{\ensuremath{\subp{\acute{\imath}}{}{#2}{}{#1}}}
\newrobustcmd{\graveiz}[2][]{\ensuremath{\subp{\grave{\imath}}{}{#2}{}{#1}}}
\newrobustcmd{\dotiz}[2][]{\ensuremath{\subp{\dot{\imath}}{}{#2}{}{#1}}}
\newrobustcmd{\ddotiz}[2][]{\ensuremath{\subp{\ddot{\imath}}{}{#2}{}{#1}}}
\newrobustcmd{\breveiz}[2][]{\ensuremath{\subp{\breve{\imath}}{}{#2}{}{#1}}}
\newrobustcmd{\bariz}[2][]{\ensuremath{\subp{\bar{\imath}}{}{#2}{}{#1}}}
\newrobustcmd{\veciz}[2][]{\ensuremath{\subp{\vec{\imath}}{}{#2}{}{#1}}}
\newrobustcmd{\bmiz}[2][]{\ensuremath{\subp{\bm{i}}{}{#2}{}{#1}}}
\newrobustcmd{\hatbmiz}[2][]{\ensuremath{\subp{\hat{\bm{\imath}}}{}{#2}{}{#1}}}
\newrobustcmd{\widehatbmiz}[2][]{\ensuremath{\subp{\widehat{\bm{\imath}}}{}{#2}{}{#1}}}
\newrobustcmd{\checkbmiz}[2][]{\ensuremath{\subp{\check{\bm{\imath}}}{}{#2}{}{#1}}}
\newrobustcmd{\tildebmiz}[2][]{\ensuremath{\subp{\tilde{\bm{\imath}}}{}{#2}{}{#1}}}
\newrobustcmd{\widetildebmiz}[2][]{\ensuremath{\subp{\widetilde{\bm{\imath}}}{}{#2}{}{#1}}}
\newrobustcmd{\acutebmiz}[2][]{\ensuremath{\subp{\acute{\bm{\imath}}}{}{#2}{}{#1}}}
\newrobustcmd{\gravebmiz}[2][]{\ensuremath{\subp{\grave{\bm{\imath}}}{}{#2}{}{#1}}}
\newrobustcmd{\dotbmiz}[2][]{\ensuremath{\subp{\dot{\bm{\imath}}}{}{#2}{}{#1}}}
\newrobustcmd{\ddotbmiz}[2][]{\ensuremath{\subp{\ddot{\bm{\imath}}}{}{#2}{}{#1}}}
\newrobustcmd{\brevebmiz}[2][]{\ensuremath{\subp{\breve{\bm{\imath}}}{}{#2}{}{#1}}}
\newrobustcmd{\barbmiz}[2][]{\ensuremath{\subp{\bar{\bm{\imath}}}{}{#2}{}{#1}}}
\newrobustcmd{\vecbmiz}[2][]{\ensuremath{\subp{\vec{\bm{\imath}}}{}{#2}{}{#1}}}
\newrobustcmd{\jz}[2][]{\ensuremath{\subp{j}{}{#2}{}{#1}}}
\newrobustcmd{\hatjz}[2][]{\ensuremath{\subp{\hat{\jmath}}{}{#2}{}{#1}}}
\newrobustcmd{\widehatjz}[2][]{\ensuremath{\subp{\widehat{\jmath}}{}{#2}{}{#1}}}
\newrobustcmd{\checkjz}[2][]{\ensuremath{\subp{\check{\jmath}}{}{#2}{}{#1}}}
\newrobustcmd{\tildejz}[2][]{\ensuremath{\subp{\tilde{\jmath}}{}{#2}{}{#1}}}
\newrobustcmd{\widetildejz}[2][]{\ensuremath{\subp{\widetilde{\jmath}}{}{#2}{}{#1}}}
\newrobustcmd{\acutejz}[2][]{\ensuremath{\subp{\acute{\jmath}}{}{#2}{}{#1}}}
\newrobustcmd{\gravejz}[2][]{\ensuremath{\subp{\grave{\jmath}}{}{#2}{}{#1}}}
\newrobustcmd{\dotjz}[2][]{\ensuremath{\subp{\dot{\jmath}}{}{#2}{}{#1}}}
\newrobustcmd{\ddotjz}[2][]{\ensuremath{\subp{\ddot{\jmath}}{}{#2}{}{#1}}}
\newrobustcmd{\brevejz}[2][]{\ensuremath{\subp{\breve{\jmath}}{}{#2}{}{#1}}}
\newrobustcmd{\barjz}[2][]{\ensuremath{\subp{\bar{\jmath}}{}{#2}{}{#1}}}
\newrobustcmd{\vecjz}[2][]{\ensuremath{\subp{\vec{\jmath}}{}{#2}{}{#1}}}
\newrobustcmd{\bmjz}[2][]{\ensuremath{\subp{\bm{j}}{}{#2}{}{#1}}}
\newrobustcmd{\hatbmjz}[2][]{\ensuremath{\subp{\hat{\bm{\jmath}}}{}{#2}{}{#1}}}
\newrobustcmd{\widehatbmjz}[2][]{\ensuremath{\subp{\widehat{\bm{\jmath}}}{}{#2}{}{#1}}}
\newrobustcmd{\checkbmjz}[2][]{\ensuremath{\subp{\check{\bm{\jmath}}}{}{#2}{}{#1}}}
\newrobustcmd{\tildebmjz}[2][]{\ensuremath{\subp{\tilde{\bm{\jmath}}}{}{#2}{}{#1}}}
\newrobustcmd{\widetildebmjz}[2][]{\ensuremath{\subp{\widetilde{\bm{\jmath}}}{}{#2}{}{#1}}}
\newrobustcmd{\acutebmjz}[2][]{\ensuremath{\subp{\acute{\bm{\jmath}}}{}{#2}{}{#1}}}
\newrobustcmd{\gravebmjz}[2][]{\ensuremath{\subp{\grave{\bm{\jmath}}}{}{#2}{}{#1}}}
\newrobustcmd{\dotbmjz}[2][]{\ensuremath{\subp{\dot{\bm{\jmath}}}{}{#2}{}{#1}}}
\newrobustcmd{\ddotbmjz}[2][]{\ensuremath{\subp{\ddot{\bm{\jmath}}}{}{#2}{}{#1}}}
\newrobustcmd{\brevebmjz}[2][]{\ensuremath{\subp{\breve{\bm{\jmath}}}{}{#2}{}{#1}}}
\newrobustcmd{\barbmjz}[2][]{\ensuremath{\subp{\bar{\bm{\jmath}}}{}{#2}{}{#1}}}
\newrobustcmd{\vecbmjz}[2][]{\ensuremath{\subp{\vec{\bm{\jmath}}}{}{#2}{}{#1}}}
\newrobustcmd{\kz}[2][]{\ensuremath{\subp{k}{}{#2}{}{#1}}}
\newrobustcmd{\hatkz}[2][]{\ensuremath{\subp{\hat{k}}{}{#2}{}{#1}}}
\newrobustcmd{\widehatkz}[2][]{\ensuremath{\subp{\widehat{k}}{}{#2}{}{#1}}}
\newrobustcmd{\checkkz}[2][]{\ensuremath{\subp{\check{k}}{}{#2}{}{#1}}}
\newrobustcmd{\tildekz}[2][]{\ensuremath{\subp{\tilde{k}}{}{#2}{}{#1}}}
\newrobustcmd{\widetildekz}[2][]{\ensuremath{\subp{\widetilde{k}}{}{#2}{}{#1}}}
\newrobustcmd{\acutekz}[2][]{\ensuremath{\subp{\acute{k}}{}{#2}{}{#1}}}
\newrobustcmd{\gravekz}[2][]{\ensuremath{\subp{\grave{k}}{}{#2}{}{#1}}}
\newrobustcmd{\dotkz}[2][]{\ensuremath{\subp{\dot{k}}{}{#2}{}{#1}}}
\newrobustcmd{\ddotkz}[2][]{\ensuremath{\subp{\ddot{k}}{}{#2}{}{#1}}}
\newrobustcmd{\brevekz}[2][]{\ensuremath{\subp{\breve{k}}{}{#2}{}{#1}}}
\newrobustcmd{\barkz}[2][]{\ensuremath{\subp{\bar{k}}{}{#2}{}{#1}}}
\newrobustcmd{\veckz}[2][]{\ensuremath{\subp{\vec{k}}{}{#2}{}{#1}}}
\newrobustcmd{\bmkz}[2][]{\ensuremath{\subp{\bm{k}}{}{#2}{}{#1}}}
\newrobustcmd{\hatbmkz}[2][]{\ensuremath{\subp{\hat{\bm{k}}}{}{#2}{}{#1}}}
\newrobustcmd{\widehatbmkz}[2][]{\ensuremath{\subp{\widehat{\bm{k}}}{}{#2}{}{#1}}}
\newrobustcmd{\checkbmkz}[2][]{\ensuremath{\subp{\check{\bm{k}}}{}{#2}{}{#1}}}
\newrobustcmd{\tildebmkz}[2][]{\ensuremath{\subp{\tilde{\bm{k}}}{}{#2}{}{#1}}}
\newrobustcmd{\widetildebmkz}[2][]{\ensuremath{\subp{\widetilde{\bm{k}}}{}{#2}{}{#1}}}
\newrobustcmd{\acutebmkz}[2][]{\ensuremath{\subp{\acute{\bm{k}}}{}{#2}{}{#1}}}
\newrobustcmd{\gravebmkz}[2][]{\ensuremath{\subp{\grave{\bm{k}}}{}{#2}{}{#1}}}
\newrobustcmd{\dotbmkz}[2][]{\ensuremath{\subp{\dot{\bm{k}}}{}{#2}{}{#1}}}
\newrobustcmd{\ddotbmkz}[2][]{\ensuremath{\subp{\ddot{\bm{k}}}{}{#2}{}{#1}}}
\newrobustcmd{\brevebmkz}[2][]{\ensuremath{\subp{\breve{\bm{k}}}{}{#2}{}{#1}}}
\newrobustcmd{\barbmkz}[2][]{\ensuremath{\subp{\bar{\bm{k}}}{}{#2}{}{#1}}}
\newrobustcmd{\vecbmkz}[2][]{\ensuremath{\subp{\vec{\bm{k}}}{}{#2}{}{#1}}}
\newrobustcmd{\lz}[2][]{\ensuremath{\subp{l}{}{#2}{}{#1}}}
\newrobustcmd{\hatlz}[2][]{\ensuremath{\subp{\hat{l}}{}{#2}{}{#1}}}
\newrobustcmd{\widehatlz}[2][]{\ensuremath{\subp{\widehat{l}}{}{#2}{}{#1}}}
\newrobustcmd{\checklz}[2][]{\ensuremath{\subp{\check{l}}{}{#2}{}{#1}}}
\newrobustcmd{\tildelz}[2][]{\ensuremath{\subp{\tilde{l}}{}{#2}{}{#1}}}
\newrobustcmd{\widetildelz}[2][]{\ensuremath{\subp{\widetilde{l}}{}{#2}{}{#1}}}
\newrobustcmd{\acutelz}[2][]{\ensuremath{\subp{\acute{l}}{}{#2}{}{#1}}}
\newrobustcmd{\gravelz}[2][]{\ensuremath{\subp{\grave{l}}{}{#2}{}{#1}}}
\newrobustcmd{\dotlz}[2][]{\ensuremath{\subp{\dot{l}}{}{#2}{}{#1}}}
\newrobustcmd{\ddotlz}[2][]{\ensuremath{\subp{\ddot{l}}{}{#2}{}{#1}}}
\newrobustcmd{\brevelz}[2][]{\ensuremath{\subp{\breve{l}}{}{#2}{}{#1}}}
\newrobustcmd{\barlz}[2][]{\ensuremath{\subp{\bar{l}}{}{#2}{}{#1}}}
\newrobustcmd{\veclz}[2][]{\ensuremath{\subp{\vec{l}}{}{#2}{}{#1}}}
\newrobustcmd{\bmlz}[2][]{\ensuremath{\subp{\bm{l}}{}{#2}{}{#1}}}
\newrobustcmd{\hatbmlz}[2][]{\ensuremath{\subp{\hat{\bm{l}}}{}{#2}{}{#1}}}
\newrobustcmd{\widehatbmlz}[2][]{\ensuremath{\subp{\widehat{\bm{l}}}{}{#2}{}{#1}}}
\newrobustcmd{\checkbmlz}[2][]{\ensuremath{\subp{\check{\bm{l}}}{}{#2}{}{#1}}}
\newrobustcmd{\tildebmlz}[2][]{\ensuremath{\subp{\tilde{\bm{l}}}{}{#2}{}{#1}}}
\newrobustcmd{\widetildebmlz}[2][]{\ensuremath{\subp{\widetilde{\bm{l}}}{}{#2}{}{#1}}}
\newrobustcmd{\acutebmlz}[2][]{\ensuremath{\subp{\acute{\bm{l}}}{}{#2}{}{#1}}}
\newrobustcmd{\gravebmlz}[2][]{\ensuremath{\subp{\grave{\bm{l}}}{}{#2}{}{#1}}}
\newrobustcmd{\dotbmlz}[2][]{\ensuremath{\subp{\dot{\bm{l}}}{}{#2}{}{#1}}}
\newrobustcmd{\ddotbmlz}[2][]{\ensuremath{\subp{\ddot{\bm{l}}}{}{#2}{}{#1}}}
\newrobustcmd{\brevebmlz}[2][]{\ensuremath{\subp{\breve{\bm{l}}}{}{#2}{}{#1}}}
\newrobustcmd{\barbmlz}[2][]{\ensuremath{\subp{\bar{\bm{l}}}{}{#2}{}{#1}}}
\newrobustcmd{\vecbmlz}[2][]{\ensuremath{\subp{\vec{\bm{l}}}{}{#2}{}{#1}}}
\newrobustcmd{\mz}[2][]{\ensuremath{\subp{m}{}{#2}{}{#1}}}
\newrobustcmd{\hatmz}[2][]{\ensuremath{\subp{\hat{m}}{}{#2}{}{#1}}}
\newrobustcmd{\widehatmz}[2][]{\ensuremath{\subp{\widehat{m}}{}{#2}{}{#1}}}
\newrobustcmd{\checkmz}[2][]{\ensuremath{\subp{\check{m}}{}{#2}{}{#1}}}
\newrobustcmd{\tildemz}[2][]{\ensuremath{\subp{\tilde{m}}{}{#2}{}{#1}}}
\newrobustcmd{\widetildemz}[2][]{\ensuremath{\subp{\widetilde{m}}{}{#2}{}{#1}}}
\newrobustcmd{\acutemz}[2][]{\ensuremath{\subp{\acute{m}}{}{#2}{}{#1}}}
\newrobustcmd{\gravemz}[2][]{\ensuremath{\subp{\grave{m}}{}{#2}{}{#1}}}
\newrobustcmd{\dotmz}[2][]{\ensuremath{\subp{\dot{m}}{}{#2}{}{#1}}}
\newrobustcmd{\ddotmz}[2][]{\ensuremath{\subp{\ddot{m}}{}{#2}{}{#1}}}
\newrobustcmd{\brevemz}[2][]{\ensuremath{\subp{\breve{m}}{}{#2}{}{#1}}}
\newrobustcmd{\barmz}[2][]{\ensuremath{\subp{\bar{m}}{}{#2}{}{#1}}}
\newrobustcmd{\vecmz}[2][]{\ensuremath{\subp{\vec{m}}{}{#2}{}{#1}}}
\newrobustcmd{\bmmz}[2][]{\ensuremath{\subp{\bm{m}}{}{#2}{}{#1}}}
\newrobustcmd{\hatbmmz}[2][]{\ensuremath{\subp{\hat{\bm{m}}}{}{#2}{}{#1}}}
\newrobustcmd{\widehatbmmz}[2][]{\ensuremath{\subp{\widehat{\bm{m}}}{}{#2}{}{#1}}}
\newrobustcmd{\checkbmmz}[2][]{\ensuremath{\subp{\check{\bm{m}}}{}{#2}{}{#1}}}
\newrobustcmd{\tildebmmz}[2][]{\ensuremath{\subp{\tilde{\bm{m}}}{}{#2}{}{#1}}}
\newrobustcmd{\widetildebmmz}[2][]{\ensuremath{\subp{\widetilde{\bm{m}}}{}{#2}{}{#1}}}
\newrobustcmd{\acutebmmz}[2][]{\ensuremath{\subp{\acute{\bm{m}}}{}{#2}{}{#1}}}
\newrobustcmd{\gravebmmz}[2][]{\ensuremath{\subp{\grave{\bm{m}}}{}{#2}{}{#1}}}
\newrobustcmd{\dotbmmz}[2][]{\ensuremath{\subp{\dot{\bm{m}}}{}{#2}{}{#1}}}
\newrobustcmd{\ddotbmmz}[2][]{\ensuremath{\subp{\ddot{\bm{m}}}{}{#2}{}{#1}}}
\newrobustcmd{\brevebmmz}[2][]{\ensuremath{\subp{\breve{\bm{m}}}{}{#2}{}{#1}}}
\newrobustcmd{\barbmmz}[2][]{\ensuremath{\subp{\bar{\bm{m}}}{}{#2}{}{#1}}}
\newrobustcmd{\vecbmmz}[2][]{\ensuremath{\subp{\vec{\bm{m}}}{}{#2}{}{#1}}}
\newrobustcmd{\nz}[2][]{\ensuremath{\subp{n}{}{#2}{}{#1}}}
\newrobustcmd{\hatnz}[2][]{\ensuremath{\subp{\hat{n}}{}{#2}{}{#1}}}
\newrobustcmd{\widehatnz}[2][]{\ensuremath{\subp{\widehat{n}}{}{#2}{}{#1}}}
\newrobustcmd{\checknz}[2][]{\ensuremath{\subp{\check{n}}{}{#2}{}{#1}}}
\newrobustcmd{\tildenz}[2][]{\ensuremath{\subp{\tilde{n}}{}{#2}{}{#1}}}
\newrobustcmd{\widetildenz}[2][]{\ensuremath{\subp{\widetilde{n}}{}{#2}{}{#1}}}
\newrobustcmd{\acutenz}[2][]{\ensuremath{\subp{\acute{n}}{}{#2}{}{#1}}}
\newrobustcmd{\gravenz}[2][]{\ensuremath{\subp{\grave{n}}{}{#2}{}{#1}}}
\newrobustcmd{\dotnz}[2][]{\ensuremath{\subp{\dot{n}}{}{#2}{}{#1}}}
\newrobustcmd{\ddotnz}[2][]{\ensuremath{\subp{\ddot{n}}{}{#2}{}{#1}}}
\newrobustcmd{\brevenz}[2][]{\ensuremath{\subp{\breve{n}}{}{#2}{}{#1}}}
\newrobustcmd{\barnz}[2][]{\ensuremath{\subp{\bar{n}}{}{#2}{}{#1}}}
\newrobustcmd{\vecnz}[2][]{\ensuremath{\subp{\vec{n}}{}{#2}{}{#1}}}
\newrobustcmd{\bmnz}[2][]{\ensuremath{\subp{\bm{n}}{}{#2}{}{#1}}}
\newrobustcmd{\hatbmnz}[2][]{\ensuremath{\subp{\hat{\bm{n}}}{}{#2}{}{#1}}}
\newrobustcmd{\widehatbmnz}[2][]{\ensuremath{\subp{\widehat{\bm{n}}}{}{#2}{}{#1}}}
\newrobustcmd{\checkbmnz}[2][]{\ensuremath{\subp{\check{\bm{n}}}{}{#2}{}{#1}}}
\newrobustcmd{\tildebmnz}[2][]{\ensuremath{\subp{\tilde{\bm{n}}}{}{#2}{}{#1}}}
\newrobustcmd{\widetildebmnz}[2][]{\ensuremath{\subp{\widetilde{\bm{n}}}{}{#2}{}{#1}}}
\newrobustcmd{\acutebmnz}[2][]{\ensuremath{\subp{\acute{\bm{n}}}{}{#2}{}{#1}}}
\newrobustcmd{\gravebmnz}[2][]{\ensuremath{\subp{\grave{\bm{n}}}{}{#2}{}{#1}}}
\newrobustcmd{\dotbmnz}[2][]{\ensuremath{\subp{\dot{\bm{n}}}{}{#2}{}{#1}}}
\newrobustcmd{\ddotbmnz}[2][]{\ensuremath{\subp{\ddot{\bm{n}}}{}{#2}{}{#1}}}
\newrobustcmd{\brevebmnz}[2][]{\ensuremath{\subp{\breve{\bm{n}}}{}{#2}{}{#1}}}
\newrobustcmd{\barbmnz}[2][]{\ensuremath{\subp{\bar{\bm{n}}}{}{#2}{}{#1}}}
\newrobustcmd{\vecbmnz}[2][]{\ensuremath{\subp{\vec{\bm{n}}}{}{#2}{}{#1}}}
\newrobustcmd{\oz}[2][]{\ensuremath{\subp{o}{}{#2}{}{#1}}}
\newrobustcmd{\hatoz}[2][]{\ensuremath{\subp{\hat{o}}{}{#2}{}{#1}}}
\newrobustcmd{\widehatoz}[2][]{\ensuremath{\subp{\widehat{o}}{}{#2}{}{#1}}}
\newrobustcmd{\checkoz}[2][]{\ensuremath{\subp{\check{o}}{}{#2}{}{#1}}}
\newrobustcmd{\tildeoz}[2][]{\ensuremath{\subp{\tilde{o}}{}{#2}{}{#1}}}
\newrobustcmd{\widetildeoz}[2][]{\ensuremath{\subp{\widetilde{o}}{}{#2}{}{#1}}}
\newrobustcmd{\acuteoz}[2][]{\ensuremath{\subp{\acute{o}}{}{#2}{}{#1}}}
\newrobustcmd{\graveoz}[2][]{\ensuremath{\subp{\grave{o}}{}{#2}{}{#1}}}
\newrobustcmd{\dotoz}[2][]{\ensuremath{\subp{\dot{o}}{}{#2}{}{#1}}}
\newrobustcmd{\ddotoz}[2][]{\ensuremath{\subp{\ddot{o}}{}{#2}{}{#1}}}
\newrobustcmd{\breveoz}[2][]{\ensuremath{\subp{\breve{o}}{}{#2}{}{#1}}}
\newrobustcmd{\baroz}[2][]{\ensuremath{\subp{\bar{o}}{}{#2}{}{#1}}}
\newrobustcmd{\vecoz}[2][]{\ensuremath{\subp{\vec{o}}{}{#2}{}{#1}}}
\newrobustcmd{\bmoz}[2][]{\ensuremath{\subp{\bm{o}}{}{#2}{}{#1}}}
\newrobustcmd{\hatbmoz}[2][]{\ensuremath{\subp{\hat{\bm{o}}}{}{#2}{}{#1}}}
\newrobustcmd{\widehatbmoz}[2][]{\ensuremath{\subp{\widehat{\bm{o}}}{}{#2}{}{#1}}}
\newrobustcmd{\checkbmoz}[2][]{\ensuremath{\subp{\check{\bm{o}}}{}{#2}{}{#1}}}
\newrobustcmd{\tildebmoz}[2][]{\ensuremath{\subp{\tilde{\bm{o}}}{}{#2}{}{#1}}}
\newrobustcmd{\widetildebmoz}[2][]{\ensuremath{\subp{\widetilde{\bm{o}}}{}{#2}{}{#1}}}
\newrobustcmd{\acutebmoz}[2][]{\ensuremath{\subp{\acute{\bm{o}}}{}{#2}{}{#1}}}
\newrobustcmd{\gravebmoz}[2][]{\ensuremath{\subp{\grave{\bm{o}}}{}{#2}{}{#1}}}
\newrobustcmd{\dotbmoz}[2][]{\ensuremath{\subp{\dot{\bm{o}}}{}{#2}{}{#1}}}
\newrobustcmd{\ddotbmoz}[2][]{\ensuremath{\subp{\ddot{\bm{o}}}{}{#2}{}{#1}}}
\newrobustcmd{\brevebmoz}[2][]{\ensuremath{\subp{\breve{\bm{o}}}{}{#2}{}{#1}}}
\newrobustcmd{\barbmoz}[2][]{\ensuremath{\subp{\bar{\bm{o}}}{}{#2}{}{#1}}}
\newrobustcmd{\vecbmoz}[2][]{\ensuremath{\subp{\vec{\bm{o}}}{}{#2}{}{#1}}}
\newrobustcmd{\pz}[2][]{\ensuremath{\subp{p}{}{#2}{}{#1}}}
\newrobustcmd{\hatpz}[2][]{\ensuremath{\subp{\hat{p}}{}{#2}{}{#1}}}
\newrobustcmd{\widehatpz}[2][]{\ensuremath{\subp{\widehat{p}}{}{#2}{}{#1}}}
\newrobustcmd{\checkpz}[2][]{\ensuremath{\subp{\check{p}}{}{#2}{}{#1}}}
\newrobustcmd{\tildepz}[2][]{\ensuremath{\subp{\tilde{p}}{}{#2}{}{#1}}}
\newrobustcmd{\widetildepz}[2][]{\ensuremath{\subp{\widetilde{p}}{}{#2}{}{#1}}}
\newrobustcmd{\acutepz}[2][]{\ensuremath{\subp{\acute{p}}{}{#2}{}{#1}}}
\newrobustcmd{\gravepz}[2][]{\ensuremath{\subp{\grave{p}}{}{#2}{}{#1}}}
\newrobustcmd{\dotpz}[2][]{\ensuremath{\subp{\dot{p}}{}{#2}{}{#1}}}
\newrobustcmd{\ddotpz}[2][]{\ensuremath{\subp{\ddot{p}}{}{#2}{}{#1}}}
\newrobustcmd{\brevepz}[2][]{\ensuremath{\subp{\breve{p}}{}{#2}{}{#1}}}
\newrobustcmd{\barpz}[2][]{\ensuremath{\subp{\bar{p}}{}{#2}{}{#1}}}
\newrobustcmd{\vecpz}[2][]{\ensuremath{\subp{\vec{p}}{}{#2}{}{#1}}}
\newrobustcmd{\bmpz}[2][]{\ensuremath{\subp{\bm{p}}{}{#2}{}{#1}}}
\newrobustcmd{\hatbmpz}[2][]{\ensuremath{\subp{\hat{\bm{p}}}{}{#2}{}{#1}}}
\newrobustcmd{\widehatbmpz}[2][]{\ensuremath{\subp{\widehat{\bm{p}}}{}{#2}{}{#1}}}
\newrobustcmd{\checkbmpz}[2][]{\ensuremath{\subp{\check{\bm{p}}}{}{#2}{}{#1}}}
\newrobustcmd{\tildebmpz}[2][]{\ensuremath{\subp{\tilde{\bm{p}}}{}{#2}{}{#1}}}
\newrobustcmd{\widetildebmpz}[2][]{\ensuremath{\subp{\widetilde{\bm{p}}}{}{#2}{}{#1}}}
\newrobustcmd{\acutebmpz}[2][]{\ensuremath{\subp{\acute{\bm{p}}}{}{#2}{}{#1}}}
\newrobustcmd{\gravebmpz}[2][]{\ensuremath{\subp{\grave{\bm{p}}}{}{#2}{}{#1}}}
\newrobustcmd{\dotbmpz}[2][]{\ensuremath{\subp{\dot{\bm{p}}}{}{#2}{}{#1}}}
\newrobustcmd{\ddotbmpz}[2][]{\ensuremath{\subp{\ddot{\bm{p}}}{}{#2}{}{#1}}}
\newrobustcmd{\brevebmpz}[2][]{\ensuremath{\subp{\breve{\bm{p}}}{}{#2}{}{#1}}}
\newrobustcmd{\barbmpz}[2][]{\ensuremath{\subp{\bar{\bm{p}}}{}{#2}{}{#1}}}
\newrobustcmd{\vecbmpz}[2][]{\ensuremath{\subp{\vec{\bm{p}}}{}{#2}{}{#1}}}
\newrobustcmd{\qz}[2][]{\ensuremath{\subp{q}{}{#2}{}{#1}}}
\newrobustcmd{\hatqz}[2][]{\ensuremath{\subp{\hat{q}}{}{#2}{}{#1}}}
\newrobustcmd{\widehatqz}[2][]{\ensuremath{\subp{\widehat{q}}{}{#2}{}{#1}}}
\newrobustcmd{\checkqz}[2][]{\ensuremath{\subp{\check{q}}{}{#2}{}{#1}}}
\newrobustcmd{\tildeqz}[2][]{\ensuremath{\subp{\tilde{q}}{}{#2}{}{#1}}}
\newrobustcmd{\widetildeqz}[2][]{\ensuremath{\subp{\widetilde{q}}{}{#2}{}{#1}}}
\newrobustcmd{\acuteqz}[2][]{\ensuremath{\subp{\acute{q}}{}{#2}{}{#1}}}
\newrobustcmd{\graveqz}[2][]{\ensuremath{\subp{\grave{q}}{}{#2}{}{#1}}}
\newrobustcmd{\dotqz}[2][]{\ensuremath{\subp{\dot{q}}{}{#2}{}{#1}}}
\newrobustcmd{\ddotqz}[2][]{\ensuremath{\subp{\ddot{q}}{}{#2}{}{#1}}}
\newrobustcmd{\breveqz}[2][]{\ensuremath{\subp{\breve{q}}{}{#2}{}{#1}}}
\newrobustcmd{\barqz}[2][]{\ensuremath{\subp{\bar{q}}{}{#2}{}{#1}}}
\newrobustcmd{\vecqz}[2][]{\ensuremath{\subp{\vec{q}}{}{#2}{}{#1}}}
\newrobustcmd{\bmqz}[2][]{\ensuremath{\subp{\bm{q}}{}{#2}{}{#1}}}
\newrobustcmd{\hatbmqz}[2][]{\ensuremath{\subp{\hat{\bm{q}}}{}{#2}{}{#1}}}
\newrobustcmd{\widehatbmqz}[2][]{\ensuremath{\subp{\widehat{\bm{q}}}{}{#2}{}{#1}}}
\newrobustcmd{\checkbmqz}[2][]{\ensuremath{\subp{\check{\bm{q}}}{}{#2}{}{#1}}}
\newrobustcmd{\tildebmqz}[2][]{\ensuremath{\subp{\tilde{\bm{q}}}{}{#2}{}{#1}}}
\newrobustcmd{\widetildebmqz}[2][]{\ensuremath{\subp{\widetilde{\bm{q}}}{}{#2}{}{#1}}}
\newrobustcmd{\acutebmqz}[2][]{\ensuremath{\subp{\acute{\bm{q}}}{}{#2}{}{#1}}}
\newrobustcmd{\gravebmqz}[2][]{\ensuremath{\subp{\grave{\bm{q}}}{}{#2}{}{#1}}}
\newrobustcmd{\dotbmqz}[2][]{\ensuremath{\subp{\dot{\bm{q}}}{}{#2}{}{#1}}}
\newrobustcmd{\ddotbmqz}[2][]{\ensuremath{\subp{\ddot{\bm{q}}}{}{#2}{}{#1}}}
\newrobustcmd{\brevebmqz}[2][]{\ensuremath{\subp{\breve{\bm{q}}}{}{#2}{}{#1}}}
\newrobustcmd{\barbmqz}[2][]{\ensuremath{\subp{\bar{\bm{q}}}{}{#2}{}{#1}}}
\newrobustcmd{\vecbmqz}[2][]{\ensuremath{\subp{\vec{\bm{q}}}{}{#2}{}{#1}}}
\newrobustcmd{\rz}[2][]{\ensuremath{\subp{r}{}{#2}{}{#1}}}
\newrobustcmd{\hatrz}[2][]{\ensuremath{\subp{\hat{r}}{}{#2}{}{#1}}}
\newrobustcmd{\widehatrz}[2][]{\ensuremath{\subp{\widehat{r}}{}{#2}{}{#1}}}
\newrobustcmd{\checkrz}[2][]{\ensuremath{\subp{\check{r}}{}{#2}{}{#1}}}
\newrobustcmd{\tilderz}[2][]{\ensuremath{\subp{\tilde{r}}{}{#2}{}{#1}}}
\newrobustcmd{\widetilderz}[2][]{\ensuremath{\subp{\widetilde{r}}{}{#2}{}{#1}}}
\newrobustcmd{\acuterz}[2][]{\ensuremath{\subp{\acute{r}}{}{#2}{}{#1}}}
\newrobustcmd{\graverz}[2][]{\ensuremath{\subp{\grave{r}}{}{#2}{}{#1}}}
\newrobustcmd{\dotrz}[2][]{\ensuremath{\subp{\dot{r}}{}{#2}{}{#1}}}
\newrobustcmd{\ddotrz}[2][]{\ensuremath{\subp{\ddot{r}}{}{#2}{}{#1}}}
\newrobustcmd{\breverz}[2][]{\ensuremath{\subp{\breve{r}}{}{#2}{}{#1}}}
\newrobustcmd{\barrz}[2][]{\ensuremath{\subp{\bar{r}}{}{#2}{}{#1}}}
\newrobustcmd{\vecrz}[2][]{\ensuremath{\subp{\vec{r}}{}{#2}{}{#1}}}
\newrobustcmd{\bmrz}[2][]{\ensuremath{\subp{\bm{r}}{}{#2}{}{#1}}}
\newrobustcmd{\hatbmrz}[2][]{\ensuremath{\subp{\hat{\bm{r}}}{}{#2}{}{#1}}}
\newrobustcmd{\widehatbmrz}[2][]{\ensuremath{\subp{\widehat{\bm{r}}}{}{#2}{}{#1}}}
\newrobustcmd{\checkbmrz}[2][]{\ensuremath{\subp{\check{\bm{r}}}{}{#2}{}{#1}}}
\newrobustcmd{\tildebmrz}[2][]{\ensuremath{\subp{\tilde{\bm{r}}}{}{#2}{}{#1}}}
\newrobustcmd{\widetildebmrz}[2][]{\ensuremath{\subp{\widetilde{\bm{r}}}{}{#2}{}{#1}}}
\newrobustcmd{\acutebmrz}[2][]{\ensuremath{\subp{\acute{\bm{r}}}{}{#2}{}{#1}}}
\newrobustcmd{\gravebmrz}[2][]{\ensuremath{\subp{\grave{\bm{r}}}{}{#2}{}{#1}}}
\newrobustcmd{\dotbmrz}[2][]{\ensuremath{\subp{\dot{\bm{r}}}{}{#2}{}{#1}}}
\newrobustcmd{\ddotbmrz}[2][]{\ensuremath{\subp{\ddot{\bm{r}}}{}{#2}{}{#1}}}
\newrobustcmd{\brevebmrz}[2][]{\ensuremath{\subp{\breve{\bm{r}}}{}{#2}{}{#1}}}
\newrobustcmd{\barbmrz}[2][]{\ensuremath{\subp{\bar{\bm{r}}}{}{#2}{}{#1}}}
\newrobustcmd{\vecbmrz}[2][]{\ensuremath{\subp{\vec{\bm{r}}}{}{#2}{}{#1}}}
\newrobustcmd{\sz}[2][]{\ensuremath{\subp{s}{}{#2}{}{#1}}}
\newrobustcmd{\hatsz}[2][]{\ensuremath{\subp{\hat{s}}{}{#2}{}{#1}}}
\newrobustcmd{\widehatsz}[2][]{\ensuremath{\subp{\widehat{s}}{}{#2}{}{#1}}}
\newrobustcmd{\checksz}[2][]{\ensuremath{\subp{\check{s}}{}{#2}{}{#1}}}
\newrobustcmd{\tildesz}[2][]{\ensuremath{\subp{\tilde{s}}{}{#2}{}{#1}}}
\newrobustcmd{\widetildesz}[2][]{\ensuremath{\subp{\widetilde{s}}{}{#2}{}{#1}}}
\newrobustcmd{\acutesz}[2][]{\ensuremath{\subp{\acute{s}}{}{#2}{}{#1}}}
\newrobustcmd{\gravesz}[2][]{\ensuremath{\subp{\grave{s}}{}{#2}{}{#1}}}
\newrobustcmd{\dotsz}[2][]{\ensuremath{\subp{\dot{s}}{}{#2}{}{#1}}}
\newrobustcmd{\ddotsz}[2][]{\ensuremath{\subp{\ddot{s}}{}{#2}{}{#1}}}
\newrobustcmd{\brevesz}[2][]{\ensuremath{\subp{\breve{s}}{}{#2}{}{#1}}}
\newrobustcmd{\barsz}[2][]{\ensuremath{\subp{\bar{s}}{}{#2}{}{#1}}}
\newrobustcmd{\vecsz}[2][]{\ensuremath{\subp{\vec{s}}{}{#2}{}{#1}}}
\newrobustcmd{\bmsz}[2][]{\ensuremath{\subp{\bm{s}}{}{#2}{}{#1}}}
\newrobustcmd{\hatbmsz}[2][]{\ensuremath{\subp{\hat{\bm{s}}}{}{#2}{}{#1}}}
\newrobustcmd{\widehatbmsz}[2][]{\ensuremath{\subp{\widehat{\bm{s}}}{}{#2}{}{#1}}}
\newrobustcmd{\checkbmsz}[2][]{\ensuremath{\subp{\check{\bm{s}}}{}{#2}{}{#1}}}
\newrobustcmd{\tildebmsz}[2][]{\ensuremath{\subp{\tilde{\bm{s}}}{}{#2}{}{#1}}}
\newrobustcmd{\widetildebmsz}[2][]{\ensuremath{\subp{\widetilde{\bm{s}}}{}{#2}{}{#1}}}
\newrobustcmd{\acutebmsz}[2][]{\ensuremath{\subp{\acute{\bm{s}}}{}{#2}{}{#1}}}
\newrobustcmd{\gravebmsz}[2][]{\ensuremath{\subp{\grave{\bm{s}}}{}{#2}{}{#1}}}
\newrobustcmd{\dotbmsz}[2][]{\ensuremath{\subp{\dot{\bm{s}}}{}{#2}{}{#1}}}
\newrobustcmd{\ddotbmsz}[2][]{\ensuremath{\subp{\ddot{\bm{s}}}{}{#2}{}{#1}}}
\newrobustcmd{\brevebmsz}[2][]{\ensuremath{\subp{\breve{\bm{s}}}{}{#2}{}{#1}}}
\newrobustcmd{\barbmsz}[2][]{\ensuremath{\subp{\bar{\bm{s}}}{}{#2}{}{#1}}}
\newrobustcmd{\vecbmsz}[2][]{\ensuremath{\subp{\vec{\bm{s}}}{}{#2}{}{#1}}}
\newrobustcmd{\tz}[2][]{\ensuremath{\subp{t}{}{#2}{}{#1}}}
\newrobustcmd{\hattz}[2][]{\ensuremath{\subp{\hat{t}}{}{#2}{}{#1}}}
\newrobustcmd{\widehattz}[2][]{\ensuremath{\subp{\widehat{t}}{}{#2}{}{#1}}}
\newrobustcmd{\checktz}[2][]{\ensuremath{\subp{\check{t}}{}{#2}{}{#1}}}
\newrobustcmd{\tildetz}[2][]{\ensuremath{\subp{\tilde{t}}{}{#2}{}{#1}}}
\newrobustcmd{\widetildetz}[2][]{\ensuremath{\subp{\widetilde{t}}{}{#2}{}{#1}}}
\newrobustcmd{\acutetz}[2][]{\ensuremath{\subp{\acute{t}}{}{#2}{}{#1}}}
\newrobustcmd{\gravetz}[2][]{\ensuremath{\subp{\grave{t}}{}{#2}{}{#1}}}
\newrobustcmd{\dottz}[2][]{\ensuremath{\subp{\dot{t}}{}{#2}{}{#1}}}
\newrobustcmd{\ddottz}[2][]{\ensuremath{\subp{\ddot{t}}{}{#2}{}{#1}}}
\newrobustcmd{\brevetz}[2][]{\ensuremath{\subp{\breve{t}}{}{#2}{}{#1}}}
\newrobustcmd{\bartz}[2][]{\ensuremath{\subp{\bar{t}}{}{#2}{}{#1}}}
\newrobustcmd{\vectz}[2][]{\ensuremath{\subp{\vec{t}}{}{#2}{}{#1}}}
\newrobustcmd{\bmtz}[2][]{\ensuremath{\subp{\bm{t}}{}{#2}{}{#1}}}
\newrobustcmd{\hatbmtz}[2][]{\ensuremath{\subp{\hat{\bm{t}}}{}{#2}{}{#1}}}
\newrobustcmd{\widehatbmtz}[2][]{\ensuremath{\subp{\widehat{\bm{t}}}{}{#2}{}{#1}}}
\newrobustcmd{\checkbmtz}[2][]{\ensuremath{\subp{\check{\bm{t}}}{}{#2}{}{#1}}}
\newrobustcmd{\tildebmtz}[2][]{\ensuremath{\subp{\tilde{\bm{t}}}{}{#2}{}{#1}}}
\newrobustcmd{\widetildebmtz}[2][]{\ensuremath{\subp{\widetilde{\bm{t}}}{}{#2}{}{#1}}}
\newrobustcmd{\acutebmtz}[2][]{\ensuremath{\subp{\acute{\bm{t}}}{}{#2}{}{#1}}}
\newrobustcmd{\gravebmtz}[2][]{\ensuremath{\subp{\grave{\bm{t}}}{}{#2}{}{#1}}}
\newrobustcmd{\dotbmtz}[2][]{\ensuremath{\subp{\dot{\bm{t}}}{}{#2}{}{#1}}}
\newrobustcmd{\ddotbmtz}[2][]{\ensuremath{\subp{\ddot{\bm{t}}}{}{#2}{}{#1}}}
\newrobustcmd{\brevebmtz}[2][]{\ensuremath{\subp{\breve{\bm{t}}}{}{#2}{}{#1}}}
\newrobustcmd{\barbmtz}[2][]{\ensuremath{\subp{\bar{\bm{t}}}{}{#2}{}{#1}}}
\newrobustcmd{\vecbmtz}[2][]{\ensuremath{\subp{\vec{\bm{t}}}{}{#2}{}{#1}}}
\newrobustcmd{\uz}[2][]{\ensuremath{\subp{u}{}{#2}{}{#1}}}
\newrobustcmd{\hatuz}[2][]{\ensuremath{\subp{\hat{u}}{}{#2}{}{#1}}}
\newrobustcmd{\widehatuz}[2][]{\ensuremath{\subp{\widehat{u}}{}{#2}{}{#1}}}
\newrobustcmd{\checkuz}[2][]{\ensuremath{\subp{\check{u}}{}{#2}{}{#1}}}
\newrobustcmd{\tildeuz}[2][]{\ensuremath{\subp{\tilde{u}}{}{#2}{}{#1}}}
\newrobustcmd{\widetildeuz}[2][]{\ensuremath{\subp{\widetilde{u}}{}{#2}{}{#1}}}
\newrobustcmd{\acuteuz}[2][]{\ensuremath{\subp{\acute{u}}{}{#2}{}{#1}}}
\newrobustcmd{\graveuz}[2][]{\ensuremath{\subp{\grave{u}}{}{#2}{}{#1}}}
\newrobustcmd{\dotuz}[2][]{\ensuremath{\subp{\dot{u}}{}{#2}{}{#1}}}
\newrobustcmd{\ddotuz}[2][]{\ensuremath{\subp{\ddot{u}}{}{#2}{}{#1}}}
\newrobustcmd{\breveuz}[2][]{\ensuremath{\subp{\breve{u}}{}{#2}{}{#1}}}
\newrobustcmd{\baruz}[2][]{\ensuremath{\subp{\bar{u}}{}{#2}{}{#1}}}
\newrobustcmd{\vecuz}[2][]{\ensuremath{\subp{\vec{u}}{}{#2}{}{#1}}}
\newrobustcmd{\bmuz}[2][]{\ensuremath{\subp{\bm{u}}{}{#2}{}{#1}}}
\newrobustcmd{\hatbmuz}[2][]{\ensuremath{\subp{\hat{\bm{u}}}{}{#2}{}{#1}}}
\newrobustcmd{\widehatbmuz}[2][]{\ensuremath{\subp{\widehat{\bm{u}}}{}{#2}{}{#1}}}
\newrobustcmd{\checkbmuz}[2][]{\ensuremath{\subp{\check{\bm{u}}}{}{#2}{}{#1}}}
\newrobustcmd{\tildebmuz}[2][]{\ensuremath{\subp{\tilde{\bm{u}}}{}{#2}{}{#1}}}
\newrobustcmd{\widetildebmuz}[2][]{\ensuremath{\subp{\widetilde{\bm{u}}}{}{#2}{}{#1}}}
\newrobustcmd{\acutebmuz}[2][]{\ensuremath{\subp{\acute{\bm{u}}}{}{#2}{}{#1}}}
\newrobustcmd{\gravebmuz}[2][]{\ensuremath{\subp{\grave{\bm{u}}}{}{#2}{}{#1}}}
\newrobustcmd{\dotbmuz}[2][]{\ensuremath{\subp{\dot{\bm{u}}}{}{#2}{}{#1}}}
\newrobustcmd{\ddotbmuz}[2][]{\ensuremath{\subp{\ddot{\bm{u}}}{}{#2}{}{#1}}}
\newrobustcmd{\brevebmuz}[2][]{\ensuremath{\subp{\breve{\bm{u}}}{}{#2}{}{#1}}}
\newrobustcmd{\barbmuz}[2][]{\ensuremath{\subp{\bar{\bm{u}}}{}{#2}{}{#1}}}
\newrobustcmd{\vecbmuz}[2][]{\ensuremath{\subp{\vec{\bm{u}}}{}{#2}{}{#1}}}
\newrobustcmd{\vz}[2][]{\ensuremath{\subp{v}{}{#2}{}{#1}}}
\newrobustcmd{\hatvz}[2][]{\ensuremath{\subp{\hat{v}}{}{#2}{}{#1}}}
\newrobustcmd{\widehatvz}[2][]{\ensuremath{\subp{\widehat{v}}{}{#2}{}{#1}}}
\newrobustcmd{\checkvz}[2][]{\ensuremath{\subp{\check{v}}{}{#2}{}{#1}}}
\newrobustcmd{\tildevz}[2][]{\ensuremath{\subp{\tilde{v}}{}{#2}{}{#1}}}
\newrobustcmd{\widetildevz}[2][]{\ensuremath{\subp{\widetilde{v}}{}{#2}{}{#1}}}
\newrobustcmd{\acutevz}[2][]{\ensuremath{\subp{\acute{v}}{}{#2}{}{#1}}}
\newrobustcmd{\gravevz}[2][]{\ensuremath{\subp{\grave{v}}{}{#2}{}{#1}}}
\newrobustcmd{\dotvz}[2][]{\ensuremath{\subp{\dot{v}}{}{#2}{}{#1}}}
\newrobustcmd{\ddotvz}[2][]{\ensuremath{\subp{\ddot{v}}{}{#2}{}{#1}}}
\newrobustcmd{\brevevz}[2][]{\ensuremath{\subp{\breve{v}}{}{#2}{}{#1}}}
\newrobustcmd{\barvz}[2][]{\ensuremath{\subp{\bar{v}}{}{#2}{}{#1}}}
\newrobustcmd{\vecvz}[2][]{\ensuremath{\subp{\vec{v}}{}{#2}{}{#1}}}
\newrobustcmd{\bmvz}[2][]{\ensuremath{\subp{\bm{v}}{}{#2}{}{#1}}}
\newrobustcmd{\hatbmvz}[2][]{\ensuremath{\subp{\hat{\bm{v}}}{}{#2}{}{#1}}}
\newrobustcmd{\widehatbmvz}[2][]{\ensuremath{\subp{\widehat{\bm{v}}}{}{#2}{}{#1}}}
\newrobustcmd{\checkbmvz}[2][]{\ensuremath{\subp{\check{\bm{v}}}{}{#2}{}{#1}}}
\newrobustcmd{\tildebmvz}[2][]{\ensuremath{\subp{\tilde{\bm{v}}}{}{#2}{}{#1}}}
\newrobustcmd{\widetildebmvz}[2][]{\ensuremath{\subp{\widetilde{\bm{v}}}{}{#2}{}{#1}}}
\newrobustcmd{\acutebmvz}[2][]{\ensuremath{\subp{\acute{\bm{v}}}{}{#2}{}{#1}}}
\newrobustcmd{\gravebmvz}[2][]{\ensuremath{\subp{\grave{\bm{v}}}{}{#2}{}{#1}}}
\newrobustcmd{\dotbmvz}[2][]{\ensuremath{\subp{\dot{\bm{v}}}{}{#2}{}{#1}}}
\newrobustcmd{\ddotbmvz}[2][]{\ensuremath{\subp{\ddot{\bm{v}}}{}{#2}{}{#1}}}
\newrobustcmd{\brevebmvz}[2][]{\ensuremath{\subp{\breve{\bm{v}}}{}{#2}{}{#1}}}
\newrobustcmd{\barbmvz}[2][]{\ensuremath{\subp{\bar{\bm{v}}}{}{#2}{}{#1}}}
\newrobustcmd{\vecbmvz}[2][]{\ensuremath{\subp{\vec{\bm{v}}}{}{#2}{}{#1}}}
\newrobustcmd{\wz}[2][]{\ensuremath{\subp{w}{}{#2}{}{#1}}}
\newrobustcmd{\hatwz}[2][]{\ensuremath{\subp{\hat{w}}{}{#2}{}{#1}}}
\newrobustcmd{\widehatwz}[2][]{\ensuremath{\subp{\widehat{w}}{}{#2}{}{#1}}}
\newrobustcmd{\checkwz}[2][]{\ensuremath{\subp{\check{w}}{}{#2}{}{#1}}}
\newrobustcmd{\tildewz}[2][]{\ensuremath{\subp{\tilde{w}}{}{#2}{}{#1}}}
\newrobustcmd{\widetildewz}[2][]{\ensuremath{\subp{\widetilde{w}}{}{#2}{}{#1}}}
\newrobustcmd{\acutewz}[2][]{\ensuremath{\subp{\acute{w}}{}{#2}{}{#1}}}
\newrobustcmd{\gravewz}[2][]{\ensuremath{\subp{\grave{w}}{}{#2}{}{#1}}}
\newrobustcmd{\dotwz}[2][]{\ensuremath{\subp{\dot{w}}{}{#2}{}{#1}}}
\newrobustcmd{\ddotwz}[2][]{\ensuremath{\subp{\ddot{w}}{}{#2}{}{#1}}}
\newrobustcmd{\brevewz}[2][]{\ensuremath{\subp{\breve{w}}{}{#2}{}{#1}}}
\newrobustcmd{\barwz}[2][]{\ensuremath{\subp{\bar{w}}{}{#2}{}{#1}}}
\newrobustcmd{\vecwz}[2][]{\ensuremath{\subp{\vec{w}}{}{#2}{}{#1}}}
\newrobustcmd{\bmwz}[2][]{\ensuremath{\subp{\bm{w}}{}{#2}{}{#1}}}
\newrobustcmd{\hatbmwz}[2][]{\ensuremath{\subp{\hat{\bm{w}}}{}{#2}{}{#1}}}
\newrobustcmd{\widehatbmwz}[2][]{\ensuremath{\subp{\widehat{\bm{w}}}{}{#2}{}{#1}}}
\newrobustcmd{\checkbmwz}[2][]{\ensuremath{\subp{\check{\bm{w}}}{}{#2}{}{#1}}}
\newrobustcmd{\tildebmwz}[2][]{\ensuremath{\subp{\tilde{\bm{w}}}{}{#2}{}{#1}}}
\newrobustcmd{\widetildebmwz}[2][]{\ensuremath{\subp{\widetilde{\bm{w}}}{}{#2}{}{#1}}}
\newrobustcmd{\acutebmwz}[2][]{\ensuremath{\subp{\acute{\bm{w}}}{}{#2}{}{#1}}}
\newrobustcmd{\gravebmwz}[2][]{\ensuremath{\subp{\grave{\bm{w}}}{}{#2}{}{#1}}}
\newrobustcmd{\dotbmwz}[2][]{\ensuremath{\subp{\dot{\bm{w}}}{}{#2}{}{#1}}}
\newrobustcmd{\ddotbmwz}[2][]{\ensuremath{\subp{\ddot{\bm{w}}}{}{#2}{}{#1}}}
\newrobustcmd{\brevebmwz}[2][]{\ensuremath{\subp{\breve{\bm{w}}}{}{#2}{}{#1}}}
\newrobustcmd{\barbmwz}[2][]{\ensuremath{\subp{\bar{\bm{w}}}{}{#2}{}{#1}}}
\newrobustcmd{\vecbmwz}[2][]{\ensuremath{\subp{\vec{\bm{w}}}{}{#2}{}{#1}}}
\newrobustcmd{\xz}[2][]{\ensuremath{\subp{x}{}{#2}{}{#1}}}
\newrobustcmd{\hatxz}[2][]{\ensuremath{\subp{\hat{x}}{}{#2}{}{#1}}}
\newrobustcmd{\widehatxz}[2][]{\ensuremath{\subp{\widehat{x}}{}{#2}{}{#1}}}
\newrobustcmd{\checkxz}[2][]{\ensuremath{\subp{\check{x}}{}{#2}{}{#1}}}
\newrobustcmd{\tildexz}[2][]{\ensuremath{\subp{\tilde{x}}{}{#2}{}{#1}}}
\newrobustcmd{\widetildexz}[2][]{\ensuremath{\subp{\widetilde{x}}{}{#2}{}{#1}}}
\newrobustcmd{\acutexz}[2][]{\ensuremath{\subp{\acute{x}}{}{#2}{}{#1}}}
\newrobustcmd{\gravexz}[2][]{\ensuremath{\subp{\grave{x}}{}{#2}{}{#1}}}
\newrobustcmd{\dotxz}[2][]{\ensuremath{\subp{\dot{x}}{}{#2}{}{#1}}}
\newrobustcmd{\ddotxz}[2][]{\ensuremath{\subp{\ddot{x}}{}{#2}{}{#1}}}
\newrobustcmd{\brevexz}[2][]{\ensuremath{\subp{\breve{x}}{}{#2}{}{#1}}}
\newrobustcmd{\barxz}[2][]{\ensuremath{\subp{\bar{x}}{}{#2}{}{#1}}}
\newrobustcmd{\vecxz}[2][]{\ensuremath{\subp{\vec{x}}{}{#2}{}{#1}}}
\newrobustcmd{\bmxz}[2][]{\ensuremath{\subp{\bm{x}}{}{#2}{}{#1}}}
\newrobustcmd{\hatbmxz}[2][]{\ensuremath{\subp{\hat{\bm{x}}}{}{#2}{}{#1}}}
\newrobustcmd{\widehatbmxz}[2][]{\ensuremath{\subp{\widehat{\bm{x}}}{}{#2}{}{#1}}}
\newrobustcmd{\checkbmxz}[2][]{\ensuremath{\subp{\check{\bm{x}}}{}{#2}{}{#1}}}
\newrobustcmd{\tildebmxz}[2][]{\ensuremath{\subp{\tilde{\bm{x}}}{}{#2}{}{#1}}}
\newrobustcmd{\widetildebmxz}[2][]{\ensuremath{\subp{\widetilde{\bm{x}}}{}{#2}{}{#1}}}
\newrobustcmd{\acutebmxz}[2][]{\ensuremath{\subp{\acute{\bm{x}}}{}{#2}{}{#1}}}
\newrobustcmd{\gravebmxz}[2][]{\ensuremath{\subp{\grave{\bm{x}}}{}{#2}{}{#1}}}
\newrobustcmd{\dotbmxz}[2][]{\ensuremath{\subp{\dot{\bm{x}}}{}{#2}{}{#1}}}
\newrobustcmd{\ddotbmxz}[2][]{\ensuremath{\subp{\ddot{\bm{x}}}{}{#2}{}{#1}}}
\newrobustcmd{\brevebmxz}[2][]{\ensuremath{\subp{\breve{\bm{x}}}{}{#2}{}{#1}}}
\newrobustcmd{\barbmxz}[2][]{\ensuremath{\subp{\bar{\bm{x}}}{}{#2}{}{#1}}}
\newrobustcmd{\vecbmxz}[2][]{\ensuremath{\subp{\vec{\bm{x}}}{}{#2}{}{#1}}}
\newrobustcmd{\yz}[2][]{\ensuremath{\subp{y}{}{#2}{}{#1}}}
\newrobustcmd{\hatyz}[2][]{\ensuremath{\subp{\hat{y}}{}{#2}{}{#1}}}
\newrobustcmd{\widehatyz}[2][]{\ensuremath{\subp{\widehat{y}}{}{#2}{}{#1}}}
\newrobustcmd{\checkyz}[2][]{\ensuremath{\subp{\check{y}}{}{#2}{}{#1}}}
\newrobustcmd{\tildeyz}[2][]{\ensuremath{\subp{\tilde{y}}{}{#2}{}{#1}}}
\newrobustcmd{\widetildeyz}[2][]{\ensuremath{\subp{\widetilde{y}}{}{#2}{}{#1}}}
\newrobustcmd{\acuteyz}[2][]{\ensuremath{\subp{\acute{y}}{}{#2}{}{#1}}}
\newrobustcmd{\graveyz}[2][]{\ensuremath{\subp{\grave{y}}{}{#2}{}{#1}}}
\newrobustcmd{\dotyz}[2][]{\ensuremath{\subp{\dot{y}}{}{#2}{}{#1}}}
\newrobustcmd{\ddotyz}[2][]{\ensuremath{\subp{\ddot{y}}{}{#2}{}{#1}}}
\newrobustcmd{\breveyz}[2][]{\ensuremath{\subp{\breve{y}}{}{#2}{}{#1}}}
\newrobustcmd{\baryz}[2][]{\ensuremath{\subp{\bar{y}}{}{#2}{}{#1}}}
\newrobustcmd{\vecyz}[2][]{\ensuremath{\subp{\vec{y}}{}{#2}{}{#1}}}
\newrobustcmd{\bmyz}[2][]{\ensuremath{\subp{\bm{y}}{}{#2}{}{#1}}}
\newrobustcmd{\hatbmyz}[2][]{\ensuremath{\subp{\hat{\bm{y}}}{}{#2}{}{#1}}}
\newrobustcmd{\widehatbmyz}[2][]{\ensuremath{\subp{\widehat{\bm{y}}}{}{#2}{}{#1}}}
\newrobustcmd{\checkbmyz}[2][]{\ensuremath{\subp{\check{\bm{y}}}{}{#2}{}{#1}}}
\newrobustcmd{\tildebmyz}[2][]{\ensuremath{\subp{\tilde{\bm{y}}}{}{#2}{}{#1}}}
\newrobustcmd{\widetildebmyz}[2][]{\ensuremath{\subp{\widetilde{\bm{y}}}{}{#2}{}{#1}}}
\newrobustcmd{\acutebmyz}[2][]{\ensuremath{\subp{\acute{\bm{y}}}{}{#2}{}{#1}}}
\newrobustcmd{\gravebmyz}[2][]{\ensuremath{\subp{\grave{\bm{y}}}{}{#2}{}{#1}}}
\newrobustcmd{\dotbmyz}[2][]{\ensuremath{\subp{\dot{\bm{y}}}{}{#2}{}{#1}}}
\newrobustcmd{\ddotbmyz}[2][]{\ensuremath{\subp{\ddot{\bm{y}}}{}{#2}{}{#1}}}
\newrobustcmd{\brevebmyz}[2][]{\ensuremath{\subp{\breve{\bm{y}}}{}{#2}{}{#1}}}
\newrobustcmd{\barbmyz}[2][]{\ensuremath{\subp{\bar{\bm{y}}}{}{#2}{}{#1}}}
\newrobustcmd{\vecbmyz}[2][]{\ensuremath{\subp{\vec{\bm{y}}}{}{#2}{}{#1}}}
\newrobustcmd{\zz}[2][]{\ensuremath{\subp{z}{}{#2}{}{#1}}}
\newrobustcmd{\hatzz}[2][]{\ensuremath{\subp{\hat{z}}{}{#2}{}{#1}}}
\newrobustcmd{\widehatzz}[2][]{\ensuremath{\subp{\widehat{z}}{}{#2}{}{#1}}}
\newrobustcmd{\checkzz}[2][]{\ensuremath{\subp{\check{z}}{}{#2}{}{#1}}}
\newrobustcmd{\tildezz}[2][]{\ensuremath{\subp{\tilde{z}}{}{#2}{}{#1}}}
\newrobustcmd{\widetildezz}[2][]{\ensuremath{\subp{\widetilde{z}}{}{#2}{}{#1}}}
\newrobustcmd{\acutezz}[2][]{\ensuremath{\subp{\acute{z}}{}{#2}{}{#1}}}
\newrobustcmd{\gravezz}[2][]{\ensuremath{\subp{\grave{z}}{}{#2}{}{#1}}}
\newrobustcmd{\dotzz}[2][]{\ensuremath{\subp{\dot{z}}{}{#2}{}{#1}}}
\newrobustcmd{\ddotzz}[2][]{\ensuremath{\subp{\ddot{z}}{}{#2}{}{#1}}}
\newrobustcmd{\brevezz}[2][]{\ensuremath{\subp{\breve{z}}{}{#2}{}{#1}}}
\newrobustcmd{\barzz}[2][]{\ensuremath{\subp{\bar{z}}{}{#2}{}{#1}}}
\newrobustcmd{\veczz}[2][]{\ensuremath{\subp{\vec{z}}{}{#2}{}{#1}}}
\newrobustcmd{\bmzz}[2][]{\ensuremath{\subp{\bm{z}}{}{#2}{}{#1}}}
\newrobustcmd{\hatbmzz}[2][]{\ensuremath{\subp{\hat{\bm{z}}}{}{#2}{}{#1}}}
\newrobustcmd{\widehatbmzz}[2][]{\ensuremath{\subp{\widehat{\bm{z}}}{}{#2}{}{#1}}}
\newrobustcmd{\checkbmzz}[2][]{\ensuremath{\subp{\check{\bm{z}}}{}{#2}{}{#1}}}
\newrobustcmd{\tildebmzz}[2][]{\ensuremath{\subp{\tilde{\bm{z}}}{}{#2}{}{#1}}}
\newrobustcmd{\widetildebmzz}[2][]{\ensuremath{\subp{\widetilde{\bm{z}}}{}{#2}{}{#1}}}
\newrobustcmd{\acutebmzz}[2][]{\ensuremath{\subp{\acute{\bm{z}}}{}{#2}{}{#1}}}
\newrobustcmd{\gravebmzz}[2][]{\ensuremath{\subp{\grave{\bm{z}}}{}{#2}{}{#1}}}
\newrobustcmd{\dotbmzz}[2][]{\ensuremath{\subp{\dot{\bm{z}}}{}{#2}{}{#1}}}
\newrobustcmd{\ddotbmzz}[2][]{\ensuremath{\subp{\ddot{\bm{z}}}{}{#2}{}{#1}}}
\newrobustcmd{\brevebmzz}[2][]{\ensuremath{\subp{\breve{\bm{z}}}{}{#2}{}{#1}}}
\newrobustcmd{\barbmzz}[2][]{\ensuremath{\subp{\bar{\bm{z}}}{}{#2}{}{#1}}}
\newrobustcmd{\vecbmzz}[2][]{\ensuremath{\subp{\vec{\bm{z}}}{}{#2}{}{#1}}}

\newrobustcmd{\Az}[2][]{\ensuremath{\subp{A}{}{#2}{}{#1}}}
\newrobustcmd{\hatAz}[2][]{\ensuremath{\subp{\hat{A}}{}{#2}{}{#1}}}
\newrobustcmd{\widehatAz}[2][]{\ensuremath{\subp{\widehat{A}}{}{#2}{}{#1}}}
\newrobustcmd{\checkAz}[2][]{\ensuremath{\subp{\check{A}}{}{#2}{}{#1}}}
\newrobustcmd{\tildeAz}[2][]{\ensuremath{\subp{\tilde{A}}{}{#2}{}{#1}}}
\newrobustcmd{\widetildeAz}[2][]{\ensuremath{\subp{\widetilde{A}}{}{#2}{}{#1}}}
\newrobustcmd{\acuteAz}[2][]{\ensuremath{\subp{\acute{A}}{}{#2}{}{#1}}}
\newrobustcmd{\graveAz}[2][]{\ensuremath{\subp{\grave{A}}{}{#2}{}{#1}}}
\newrobustcmd{\dotAz}[2][]{\ensuremath{\subp{\dot{A}}{}{#2}{}{#1}}}
\newrobustcmd{\ddotAz}[2][]{\ensuremath{\subp{\ddot{A}}{}{#2}{}{#1}}}
\newrobustcmd{\breveAz}[2][]{\ensuremath{\subp{\breve{A}}{}{#2}{}{#1}}}
\newrobustcmd{\barAz}[2][]{\ensuremath{\subp{\bar{A}}{}{#2}{}{#1}}}
\newrobustcmd{\vecAz}[2][]{\ensuremath{\subp{\vec{A}}{}{#2}{}{#1}}}
\newrobustcmd{\bmAz}[2][]{\ensuremath{\subp{\bm{A}}{}{#2}{}{#1}}}
\newrobustcmd{\hatbmAz}[2][]{\ensuremath{\subp{\hat{\bm{A}}}{}{#2}{}{#1}}}
\newrobustcmd{\widehatbmAz}[2][]{\ensuremath{\subp{\widehat{\bm{A}}}{}{#2}{}{#1}}}
\newrobustcmd{\checkbmAz}[2][]{\ensuremath{\subp{\check{\bm{A}}}{}{#2}{}{#1}}}
\newrobustcmd{\tildebmAz}[2][]{\ensuremath{\subp{\tilde{\bm{A}}}{}{#2}{}{#1}}}
\newrobustcmd{\widetildebmAz}[2][]{\ensuremath{\subp{\widetilde{\bm{A}}}{}{#2}{}{#1}}}
\newrobustcmd{\acutebmAz}[2][]{\ensuremath{\subp{\acute{\bm{A}}}{}{#2}{}{#1}}}
\newrobustcmd{\gravebmAz}[2][]{\ensuremath{\subp{\grave{\bm{A}}}{}{#2}{}{#1}}}
\newrobustcmd{\dotbmAz}[2][]{\ensuremath{\subp{\dot{\bm{A}}}{}{#2}{}{#1}}}
\newrobustcmd{\ddotbmAz}[2][]{\ensuremath{\subp{\ddot{\bm{A}}}{}{#2}{}{#1}}}
\newrobustcmd{\brevebmAz}[2][]{\ensuremath{\subp{\breve{\bm{A}}}{}{#2}{}{#1}}}
\newrobustcmd{\barbmAz}[2][]{\ensuremath{\subp{\bar{\bm{A}}}{}{#2}{}{#1}}}
\newrobustcmd{\vecbmAz}[2][]{\ensuremath{\subp{\vec{\bm{A}}}{}{#2}{}{#1}}}
\newrobustcmd{\Bz}[2][]{\ensuremath{\subp{B}{}{#2}{}{#1}}}
\newrobustcmd{\hatBz}[2][]{\ensuremath{\subp{\hat{B}}{}{#2}{}{#1}}}
\newrobustcmd{\widehatBz}[2][]{\ensuremath{\subp{\widehat{B}}{}{#2}{}{#1}}}
\newrobustcmd{\checkBz}[2][]{\ensuremath{\subp{\check{B}}{}{#2}{}{#1}}}
\newrobustcmd{\tildeBz}[2][]{\ensuremath{\subp{\tilde{B}}{}{#2}{}{#1}}}
\newrobustcmd{\widetildeBz}[2][]{\ensuremath{\subp{\widetilde{B}}{}{#2}{}{#1}}}
\newrobustcmd{\acuteBz}[2][]{\ensuremath{\subp{\acute{B}}{}{#2}{}{#1}}}
\newrobustcmd{\graveBz}[2][]{\ensuremath{\subp{\grave{B}}{}{#2}{}{#1}}}
\newrobustcmd{\dotBz}[2][]{\ensuremath{\subp{\dot{B}}{}{#2}{}{#1}}}
\newrobustcmd{\ddotBz}[2][]{\ensuremath{\subp{\ddot{B}}{}{#2}{}{#1}}}
\newrobustcmd{\breveBz}[2][]{\ensuremath{\subp{\breve{B}}{}{#2}{}{#1}}}
\newrobustcmd{\barBz}[2][]{\ensuremath{\subp{\bar{B}}{}{#2}{}{#1}}}
\newrobustcmd{\vecBz}[2][]{\ensuremath{\subp{\vec{B}}{}{#2}{}{#1}}}
\newrobustcmd{\bmBz}[2][]{\ensuremath{\subp{\bm{B}}{}{#2}{}{#1}}}
\newrobustcmd{\hatbmBz}[2][]{\ensuremath{\subp{\hat{\bm{B}}}{}{#2}{}{#1}}}
\newrobustcmd{\widehatbmBz}[2][]{\ensuremath{\subp{\widehat{\bm{B}}}{}{#2}{}{#1}}}
\newrobustcmd{\checkbmBz}[2][]{\ensuremath{\subp{\check{\bm{B}}}{}{#2}{}{#1}}}
\newrobustcmd{\tildebmBz}[2][]{\ensuremath{\subp{\tilde{\bm{B}}}{}{#2}{}{#1}}}
\newrobustcmd{\widetildebmBz}[2][]{\ensuremath{\subp{\widetilde{\bm{B}}}{}{#2}{}{#1}}}
\newrobustcmd{\acutebmBz}[2][]{\ensuremath{\subp{\acute{\bm{B}}}{}{#2}{}{#1}}}
\newrobustcmd{\gravebmBz}[2][]{\ensuremath{\subp{\grave{\bm{B}}}{}{#2}{}{#1}}}
\newrobustcmd{\dotbmBz}[2][]{\ensuremath{\subp{\dot{\bm{B}}}{}{#2}{}{#1}}}
\newrobustcmd{\ddotbmBz}[2][]{\ensuremath{\subp{\ddot{\bm{B}}}{}{#2}{}{#1}}}
\newrobustcmd{\brevebmBz}[2][]{\ensuremath{\subp{\breve{\bm{B}}}{}{#2}{}{#1}}}
\newrobustcmd{\barbmBz}[2][]{\ensuremath{\subp{\bar{\bm{B}}}{}{#2}{}{#1}}}
\newrobustcmd{\vecbmBz}[2][]{\ensuremath{\subp{\vec{\bm{B}}}{}{#2}{}{#1}}}
\newrobustcmd{\Cz}[2][]{\ensuremath{\subp{C}{}{#2}{}{#1}}}
\newrobustcmd{\hatCz}[2][]{\ensuremath{\subp{\hat{C}}{}{#2}{}{#1}}}
\newrobustcmd{\widehatCz}[2][]{\ensuremath{\subp{\widehat{C}}{}{#2}{}{#1}}}
\newrobustcmd{\checkCz}[2][]{\ensuremath{\subp{\check{C}}{}{#2}{}{#1}}}
\newrobustcmd{\tildeCz}[2][]{\ensuremath{\subp{\tilde{C}}{}{#2}{}{#1}}}
\newrobustcmd{\widetildeCz}[2][]{\ensuremath{\subp{\widetilde{C}}{}{#2}{}{#1}}}
\newrobustcmd{\acuteCz}[2][]{\ensuremath{\subp{\acute{C}}{}{#2}{}{#1}}}
\newrobustcmd{\graveCz}[2][]{\ensuremath{\subp{\grave{C}}{}{#2}{}{#1}}}
\newrobustcmd{\dotCz}[2][]{\ensuremath{\subp{\dot{C}}{}{#2}{}{#1}}}
\newrobustcmd{\ddotCz}[2][]{\ensuremath{\subp{\ddot{C}}{}{#2}{}{#1}}}
\newrobustcmd{\breveCz}[2][]{\ensuremath{\subp{\breve{C}}{}{#2}{}{#1}}}
\newrobustcmd{\barCz}[2][]{\ensuremath{\subp{\bar{C}}{}{#2}{}{#1}}}
\newrobustcmd{\vecCz}[2][]{\ensuremath{\subp{\vec{C}}{}{#2}{}{#1}}}
\newrobustcmd{\bmCz}[2][]{\ensuremath{\subp{\bm{C}}{}{#2}{}{#1}}}
\newrobustcmd{\hatbmCz}[2][]{\ensuremath{\subp{\hat{\bm{C}}}{}{#2}{}{#1}}}
\newrobustcmd{\widehatbmCz}[2][]{\ensuremath{\subp{\widehat{\bm{C}}}{}{#2}{}{#1}}}
\newrobustcmd{\checkbmCz}[2][]{\ensuremath{\subp{\check{\bm{C}}}{}{#2}{}{#1}}}
\newrobustcmd{\tildebmCz}[2][]{\ensuremath{\subp{\tilde{\bm{C}}}{}{#2}{}{#1}}}
\newrobustcmd{\widetildebmCz}[2][]{\ensuremath{\subp{\widetilde{\bm{C}}}{}{#2}{}{#1}}}
\newrobustcmd{\acutebmCz}[2][]{\ensuremath{\subp{\acute{\bm{C}}}{}{#2}{}{#1}}}
\newrobustcmd{\gravebmCz}[2][]{\ensuremath{\subp{\grave{\bm{C}}}{}{#2}{}{#1}}}
\newrobustcmd{\dotbmCz}[2][]{\ensuremath{\subp{\dot{\bm{C}}}{}{#2}{}{#1}}}
\newrobustcmd{\ddotbmCz}[2][]{\ensuremath{\subp{\ddot{\bm{C}}}{}{#2}{}{#1}}}
\newrobustcmd{\brevebmCz}[2][]{\ensuremath{\subp{\breve{\bm{C}}}{}{#2}{}{#1}}}
\newrobustcmd{\barbmCz}[2][]{\ensuremath{\subp{\bar{\bm{C}}}{}{#2}{}{#1}}}
\newrobustcmd{\vecbmCz}[2][]{\ensuremath{\subp{\vec{\bm{C}}}{}{#2}{}{#1}}}

\newrobustcmd{\Dz}[2][]{\ensuremath{\subp{D}{}{#2}{}{#1}}}
\newrobustcmd{\hatDz}[2][]{\ensuremath{\subp{\hat{D}}{}{#2}{}{#1}}}
\newrobustcmd{\widehatDz}[2][]{\ensuremath{\subp{\widehat{D}}{}{#2}{}{#1}}}
\newrobustcmd{\checkDz}[2][]{\ensuremath{\subp{\check{D}}{}{#2}{}{#1}}}
\newrobustcmd{\tildeDz}[2][]{\ensuremath{\subp{\tilde{D}}{}{#2}{}{#1}}}
\newrobustcmd{\widetildeDz}[2][]{\ensuremath{\subp{\widetilde{D}}{}{#2}{}{#1}}}
\newrobustcmd{\acuteDz}[2][]{\ensuremath{\subp{\acute{D}}{}{#2}{}{#1}}}
\newrobustcmd{\graveDz}[2][]{\ensuremath{\subp{\grave{D}}{}{#2}{}{#1}}}
\newrobustcmd{\dotDz}[2][]{\ensuremath{\subp{\dot{D}}{}{#2}{}{#1}}}
\newrobustcmd{\ddotDz}[2][]{\ensuremath{\subp{\ddot{D}}{}{#2}{}{#1}}}
\newrobustcmd{\breveDz}[2][]{\ensuremath{\subp{\breve{D}}{}{#2}{}{#1}}}
\newrobustcmd{\barDz}[2][]{\ensuremath{\subp{\bar{D}}{}{#2}{}{#1}}}
\newrobustcmd{\vecDz}[2][]{\ensuremath{\subp{\vec{D}}{}{#2}{}{#1}}}
\newrobustcmd{\bmDz}[2][]{\ensuremath{\subp{\bm{D}}{}{#2}{}{#1}}}
\newrobustcmd{\hatbmDz}[2][]{\ensuremath{\subp{\hat{\bm{D}}}{}{#2}{}{#1}}}
\newrobustcmd{\widehatbmDz}[2][]{\ensuremath{\subp{\widehat{\bm{D}}}{}{#2}{}{#1}}}
\newrobustcmd{\checkbmDz}[2][]{\ensuremath{\subp{\check{\bm{D}}}{}{#2}{}{#1}}}
\newrobustcmd{\tildebmDz}[2][]{\ensuremath{\subp{\tilde{\bm{D}}}{}{#2}{}{#1}}}
\newrobustcmd{\widetildebmDz}[2][]{\ensuremath{\subp{\widetilde{\bm{D}}}{}{#2}{}{#1}}}
\newrobustcmd{\acutebmDz}[2][]{\ensuremath{\subp{\acute{\bm{D}}}{}{#2}{}{#1}}}
\newrobustcmd{\gravebmDz}[2][]{\ensuremath{\subp{\grave{\bm{D}}}{}{#2}{}{#1}}}
\newrobustcmd{\dotbmDz}[2][]{\ensuremath{\subp{\dot{\bm{D}}}{}{#2}{}{#1}}}
\newrobustcmd{\ddotbmDz}[2][]{\ensuremath{\subp{\ddot{\bm{D}}}{}{#2}{}{#1}}}
\newrobustcmd{\brevebmDz}[2][]{\ensuremath{\subp{\breve{\bm{D}}}{}{#2}{}{#1}}}
\newrobustcmd{\barbmDz}[2][]{\ensuremath{\subp{\bar{\bm{D}}}{}{#2}{}{#1}}}
\newrobustcmd{\vecbmDz}[2][]{\ensuremath{\subp{\vec{\bm{D}}}{}{#2}{}{#1}}}
\newrobustcmd{\Ez}[2][]{\ensuremath{\subp{E}{}{#2}{}{#1}}}
\newrobustcmd{\hatEz}[2][]{\ensuremath{\subp{\hat{E}}{}{#2}{}{#1}}}
\newrobustcmd{\widehatEz}[2][]{\ensuremath{\subp{\widehat{E}}{}{#2}{}{#1}}}
\newrobustcmd{\checkEz}[2][]{\ensuremath{\subp{\check{E}}{}{#2}{}{#1}}}
\newrobustcmd{\tildeEz}[2][]{\ensuremath{\subp{\tilde{E}}{}{#2}{}{#1}}}
\newrobustcmd{\widetildeEz}[2][]{\ensuremath{\subp{\widetilde{E}}{}{#2}{}{#1}}}
\newrobustcmd{\acuteEz}[2][]{\ensuremath{\subp{\acute{E}}{}{#2}{}{#1}}}
\newrobustcmd{\graveEz}[2][]{\ensuremath{\subp{\grave{E}}{}{#2}{}{#1}}}
\newrobustcmd{\dotEz}[2][]{\ensuremath{\subp{\dot{E}}{}{#2}{}{#1}}}
\newrobustcmd{\ddotEz}[2][]{\ensuremath{\subp{\ddot{E}}{}{#2}{}{#1}}}
\newrobustcmd{\breveEz}[2][]{\ensuremath{\subp{\breve{E}}{}{#2}{}{#1}}}
\newrobustcmd{\barEz}[2][]{\ensuremath{\subp{\bar{E}}{}{#2}{}{#1}}}
\newrobustcmd{\vecEz}[2][]{\ensuremath{\subp{\vec{E}}{}{#2}{}{#1}}}
\newrobustcmd{\bmEz}[2][]{\ensuremath{\subp{\bm{E}}{}{#2}{}{#1}}}
\newrobustcmd{\hatbmEz}[2][]{\ensuremath{\subp{\hat{\bm{E}}}{}{#2}{}{#1}}}
\newrobustcmd{\widehatbmEz}[2][]{\ensuremath{\subp{\widehat{\bm{E}}}{}{#2}{}{#1}}}
\newrobustcmd{\checkbmEz}[2][]{\ensuremath{\subp{\check{\bm{E}}}{}{#2}{}{#1}}}
\newrobustcmd{\tildebmEz}[2][]{\ensuremath{\subp{\tilde{\bm{E}}}{}{#2}{}{#1}}}
\newrobustcmd{\widetildebmEz}[2][]{\ensuremath{\subp{\widetilde{\bm{E}}}{}{#2}{}{#1}}}
\newrobustcmd{\acutebmEz}[2][]{\ensuremath{\subp{\acute{\bm{E}}}{}{#2}{}{#1}}}
\newrobustcmd{\gravebmEz}[2][]{\ensuremath{\subp{\grave{\bm{E}}}{}{#2}{}{#1}}}
\newrobustcmd{\dotbmEz}[2][]{\ensuremath{\subp{\dot{\bm{E}}}{}{#2}{}{#1}}}
\newrobustcmd{\ddotbmEz}[2][]{\ensuremath{\subp{\ddot{\bm{E}}}{}{#2}{}{#1}}}
\newrobustcmd{\brevebmEz}[2][]{\ensuremath{\subp{\breve{\bm{E}}}{}{#2}{}{#1}}}
\newrobustcmd{\barbmEz}[2][]{\ensuremath{\subp{\bar{\bm{E}}}{}{#2}{}{#1}}}
\newrobustcmd{\vecbmEz}[2][]{\ensuremath{\subp{\vec{\bm{E}}}{}{#2}{}{#1}}}
\newrobustcmd{\Fz}[2][]{\ensuremath{\subp{F}{}{#2}{}{#1}}}
\newrobustcmd{\hatFz}[2][]{\ensuremath{\subp{\hat{F}}{}{#2}{}{#1}}}
\newrobustcmd{\widehatFz}[2][]{\ensuremath{\subp{\widehat{F}}{}{#2}{}{#1}}}
\newrobustcmd{\checkFz}[2][]{\ensuremath{\subp{\check{F}}{}{#2}{}{#1}}}
\newrobustcmd{\tildeFz}[2][]{\ensuremath{\subp{\tilde{F}}{}{#2}{}{#1}}}
\newrobustcmd{\widetildeFz}[2][]{\ensuremath{\subp{\widetilde{F}}{}{#2}{}{#1}}}
\newrobustcmd{\acuteFz}[2][]{\ensuremath{\subp{\acute{F}}{}{#2}{}{#1}}}
\newrobustcmd{\graveFz}[2][]{\ensuremath{\subp{\grave{F}}{}{#2}{}{#1}}}
\newrobustcmd{\dotFz}[2][]{\ensuremath{\subp{\dot{F}}{}{#2}{}{#1}}}
\newrobustcmd{\ddotFz}[2][]{\ensuremath{\subp{\ddot{F}}{}{#2}{}{#1}}}
\newrobustcmd{\breveFz}[2][]{\ensuremath{\subp{\breve{F}}{}{#2}{}{#1}}}
\newrobustcmd{\barFz}[2][]{\ensuremath{\subp{\bar{F}}{}{#2}{}{#1}}}
\newrobustcmd{\vecFz}[2][]{\ensuremath{\subp{\vec{F}}{}{#2}{}{#1}}}
\newrobustcmd{\bmFz}[2][]{\ensuremath{\subp{\bm{F}}{}{#2}{}{#1}}}
\newrobustcmd{\hatbmFz}[2][]{\ensuremath{\subp{\hat{\bm{F}}}{}{#2}{}{#1}}}
\newrobustcmd{\widehatbmFz}[2][]{\ensuremath{\subp{\widehat{\bm{F}}}{}{#2}{}{#1}}}
\newrobustcmd{\checkbmFz}[2][]{\ensuremath{\subp{\check{\bm{F}}}{}{#2}{}{#1}}}
\newrobustcmd{\tildebmFz}[2][]{\ensuremath{\subp{\tilde{\bm{F}}}{}{#2}{}{#1}}}
\newrobustcmd{\widetildebmFz}[2][]{\ensuremath{\subp{\widetilde{\bm{F}}}{}{#2}{}{#1}}}
\newrobustcmd{\acutebmFz}[2][]{\ensuremath{\subp{\acute{\bm{F}}}{}{#2}{}{#1}}}
\newrobustcmd{\gravebmFz}[2][]{\ensuremath{\subp{\grave{\bm{F}}}{}{#2}{}{#1}}}
\newrobustcmd{\dotbmFz}[2][]{\ensuremath{\subp{\dot{\bm{F}}}{}{#2}{}{#1}}}
\newrobustcmd{\ddotbmFz}[2][]{\ensuremath{\subp{\ddot{\bm{F}}}{}{#2}{}{#1}}}
\newrobustcmd{\brevebmFz}[2][]{\ensuremath{\subp{\breve{\bm{F}}}{}{#2}{}{#1}}}
\newrobustcmd{\barbmFz}[2][]{\ensuremath{\subp{\bar{\bm{F}}}{}{#2}{}{#1}}}
\newrobustcmd{\vecbmFz}[2][]{\ensuremath{\subp{\vec{\bm{F}}}{}{#2}{}{#1}}}
\newrobustcmd{\Gz}[2][]{\ensuremath{\subp{G}{}{#2}{}{#1}}}
\newrobustcmd{\hatGz}[2][]{\ensuremath{\subp{\hat{G}}{}{#2}{}{#1}}}
\newrobustcmd{\widehatGz}[2][]{\ensuremath{\subp{\widehat{G}}{}{#2}{}{#1}}}
\newrobustcmd{\checkGz}[2][]{\ensuremath{\subp{\check{G}}{}{#2}{}{#1}}}
\newrobustcmd{\tildeGz}[2][]{\ensuremath{\subp{\tilde{G}}{}{#2}{}{#1}}}
\newrobustcmd{\widetildeGz}[2][]{\ensuremath{\subp{\widetilde{G}}{}{#2}{}{#1}}}
\newrobustcmd{\acuteGz}[2][]{\ensuremath{\subp{\acute{G}}{}{#2}{}{#1}}}
\newrobustcmd{\graveGz}[2][]{\ensuremath{\subp{\grave{G}}{}{#2}{}{#1}}}
\newrobustcmd{\dotGz}[2][]{\ensuremath{\subp{\dot{G}}{}{#2}{}{#1}}}
\newrobustcmd{\ddotGz}[2][]{\ensuremath{\subp{\ddot{G}}{}{#2}{}{#1}}}
\newrobustcmd{\breveGz}[2][]{\ensuremath{\subp{\breve{G}}{}{#2}{}{#1}}}
\newrobustcmd{\barGz}[2][]{\ensuremath{\subp{\bar{G}}{}{#2}{}{#1}}}
\newrobustcmd{\vecGz}[2][]{\ensuremath{\subp{\vec{G}}{}{#2}{}{#1}}}
\newrobustcmd{\bmGz}[2][]{\ensuremath{\subp{\bm{G}}{}{#2}{}{#1}}}
\newrobustcmd{\hatbmGz}[2][]{\ensuremath{\subp{\hat{\bm{G}}}{}{#2}{}{#1}}}
\newrobustcmd{\widehatbmGz}[2][]{\ensuremath{\subp{\widehat{\bm{G}}}{}{#2}{}{#1}}}
\newrobustcmd{\checkbmGz}[2][]{\ensuremath{\subp{\check{\bm{G}}}{}{#2}{}{#1}}}
\newrobustcmd{\tildebmGz}[2][]{\ensuremath{\subp{\tilde{\bm{G}}}{}{#2}{}{#1}}}
\newrobustcmd{\widetildebmGz}[2][]{\ensuremath{\subp{\widetilde{\bm{G}}}{}{#2}{}{#1}}}
\newrobustcmd{\acutebmGz}[2][]{\ensuremath{\subp{\acute{\bm{G}}}{}{#2}{}{#1}}}
\newrobustcmd{\gravebmGz}[2][]{\ensuremath{\subp{\grave{\bm{G}}}{}{#2}{}{#1}}}
\newrobustcmd{\dotbmGz}[2][]{\ensuremath{\subp{\dot{\bm{G}}}{}{#2}{}{#1}}}
\newrobustcmd{\ddotbmGz}[2][]{\ensuremath{\subp{\ddot{\bm{G}}}{}{#2}{}{#1}}}
\newrobustcmd{\brevebmGz}[2][]{\ensuremath{\subp{\breve{\bm{G}}}{}{#2}{}{#1}}}
\newrobustcmd{\barbmGz}[2][]{\ensuremath{\subp{\bar{\bm{G}}}{}{#2}{}{#1}}}
\newrobustcmd{\vecbmGz}[2][]{\ensuremath{\subp{\vec{\bm{G}}}{}{#2}{}{#1}}}
\newrobustcmd{\Hz}[2][]{\ensuremath{\subp{H}{}{#2}{}{#1}}}
\newrobustcmd{\hatHz}[2][]{\ensuremath{\subp{\hat{H}}{}{#2}{}{#1}}}
\newrobustcmd{\widehatHz}[2][]{\ensuremath{\subp{\widehat{H}}{}{#2}{}{#1}}}
\newrobustcmd{\checkHz}[2][]{\ensuremath{\subp{\check{H}}{}{#2}{}{#1}}}
\newrobustcmd{\tildeHz}[2][]{\ensuremath{\subp{\tilde{H}}{}{#2}{}{#1}}}
\newrobustcmd{\widetildeHz}[2][]{\ensuremath{\subp{\widetilde{H}}{}{#2}{}{#1}}}
\newrobustcmd{\acuteHz}[2][]{\ensuremath{\subp{\acute{H}}{}{#2}{}{#1}}}
\newrobustcmd{\graveHz}[2][]{\ensuremath{\subp{\grave{H}}{}{#2}{}{#1}}}
\newrobustcmd{\dotHz}[2][]{\ensuremath{\subp{\dot{H}}{}{#2}{}{#1}}}
\newrobustcmd{\ddotHz}[2][]{\ensuremath{\subp{\ddot{H}}{}{#2}{}{#1}}}
\newrobustcmd{\breveHz}[2][]{\ensuremath{\subp{\breve{H}}{}{#2}{}{#1}}}
\newrobustcmd{\barHz}[2][]{\ensuremath{\subp{\bar{H}}{}{#2}{}{#1}}}
\newrobustcmd{\vecHz}[2][]{\ensuremath{\subp{\vec{H}}{}{#2}{}{#1}}}
\newrobustcmd{\bmHz}[2][]{\ensuremath{\subp{\bm{H}}{}{#2}{}{#1}}}
\newrobustcmd{\hatbmHz}[2][]{\ensuremath{\subp{\hat{\bm{H}}}{}{#2}{}{#1}}}
\newrobustcmd{\widehatbmHz}[2][]{\ensuremath{\subp{\widehat{\bm{H}}}{}{#2}{}{#1}}}
\newrobustcmd{\checkbmHz}[2][]{\ensuremath{\subp{\check{\bm{H}}}{}{#2}{}{#1}}}
\newrobustcmd{\tildebmHz}[2][]{\ensuremath{\subp{\tilde{\bm{H}}}{}{#2}{}{#1}}}
\newrobustcmd{\widetildebmHz}[2][]{\ensuremath{\subp{\widetilde{\bm{H}}}{}{#2}{}{#1}}}
\newrobustcmd{\acutebmHz}[2][]{\ensuremath{\subp{\acute{\bm{H}}}{}{#2}{}{#1}}}
\newrobustcmd{\gravebmHz}[2][]{\ensuremath{\subp{\grave{\bm{H}}}{}{#2}{}{#1}}}
\newrobustcmd{\dotbmHz}[2][]{\ensuremath{\subp{\dot{\bm{H}}}{}{#2}{}{#1}}}
\newrobustcmd{\ddotbmHz}[2][]{\ensuremath{\subp{\ddot{\bm{H}}}{}{#2}{}{#1}}}
\newrobustcmd{\brevebmHz}[2][]{\ensuremath{\subp{\breve{\bm{H}}}{}{#2}{}{#1}}}
\newrobustcmd{\barbmHz}[2][]{\ensuremath{\subp{\bar{\bm{H}}}{}{#2}{}{#1}}}
\newrobustcmd{\vecbmHz}[2][]{\ensuremath{\subp{\vec{\bm{H}}}{}{#2}{}{#1}}}
\newrobustcmd{\Iz}[2][]{\ensuremath{\subp{I}{}{#2}{}{#1}}}
\newrobustcmd{\hatIz}[2][]{\ensuremath{\subp{\hat{I}}{}{#2}{}{#1}}}
\newrobustcmd{\widehatIz}[2][]{\ensuremath{\subp{\widehat{I}}{}{#2}{}{#1}}}
\newrobustcmd{\checkIz}[2][]{\ensuremath{\subp{\check{I}}{}{#2}{}{#1}}}
\newrobustcmd{\tildeIz}[2][]{\ensuremath{\subp{\tilde{I}}{}{#2}{}{#1}}}
\newrobustcmd{\widetildeIz}[2][]{\ensuremath{\subp{\widetilde{I}}{}{#2}{}{#1}}}
\newrobustcmd{\acuteIz}[2][]{\ensuremath{\subp{\acute{I}}{}{#2}{}{#1}}}
\newrobustcmd{\graveIz}[2][]{\ensuremath{\subp{\grave{I}}{}{#2}{}{#1}}}
\newrobustcmd{\dotIz}[2][]{\ensuremath{\subp{\dot{I}}{}{#2}{}{#1}}}
\newrobustcmd{\ddotIz}[2][]{\ensuremath{\subp{\ddot{I}}{}{#2}{}{#1}}}
\newrobustcmd{\breveIz}[2][]{\ensuremath{\subp{\breve{I}}{}{#2}{}{#1}}}
\newrobustcmd{\barIz}[2][]{\ensuremath{\subp{\bar{I}}{}{#2}{}{#1}}}
\newrobustcmd{\vecIz}[2][]{\ensuremath{\subp{\vec{I}}{}{#2}{}{#1}}}
\newrobustcmd{\bmIz}[2][]{\ensuremath{\subp{\bm{I}}{}{#2}{}{#1}}}
\newrobustcmd{\hatbmIz}[2][]{\ensuremath{\subp{\hat{\bm{I}}}{}{#2}{}{#1}}}
\newrobustcmd{\widehatbmIz}[2][]{\ensuremath{\subp{\widehat{\bm{I}}}{}{#2}{}{#1}}}
\newrobustcmd{\checkbmIz}[2][]{\ensuremath{\subp{\check{\bm{I}}}{}{#2}{}{#1}}}
\newrobustcmd{\tildebmIz}[2][]{\ensuremath{\subp{\tilde{\bm{I}}}{}{#2}{}{#1}}}
\newrobustcmd{\widetildebmIz}[2][]{\ensuremath{\subp{\widetilde{\bm{I}}}{}{#2}{}{#1}}}
\newrobustcmd{\acutebmIz}[2][]{\ensuremath{\subp{\acute{\bm{I}}}{}{#2}{}{#1}}}
\newrobustcmd{\gravebmIz}[2][]{\ensuremath{\subp{\grave{\bm{I}}}{}{#2}{}{#1}}}
\newrobustcmd{\dotbmIz}[2][]{\ensuremath{\subp{\dot{\bm{I}}}{}{#2}{}{#1}}}
\newrobustcmd{\ddotbmIz}[2][]{\ensuremath{\subp{\ddot{\bm{I}}}{}{#2}{}{#1}}}
\newrobustcmd{\brevebmIz}[2][]{\ensuremath{\subp{\breve{\bm{I}}}{}{#2}{}{#1}}}
\newrobustcmd{\barbmIz}[2][]{\ensuremath{\subp{\bar{\bm{I}}}{}{#2}{}{#1}}}
\newrobustcmd{\vecbmIz}[2][]{\ensuremath{\subp{\vec{\bm{I}}}{}{#2}{}{#1}}}
\newrobustcmd{\Jz}[2][]{\ensuremath{\subp{J}{}{#2}{}{#1}}}
\newrobustcmd{\hatJz}[2][]{\ensuremath{\subp{\hat{J}}{}{#2}{}{#1}}}
\newrobustcmd{\widehatJz}[2][]{\ensuremath{\subp{\widehat{J}}{}{#2}{}{#1}}}
\newrobustcmd{\checkJz}[2][]{\ensuremath{\subp{\check{J}}{}{#2}{}{#1}}}
\newrobustcmd{\tildeJz}[2][]{\ensuremath{\subp{\tilde{J}}{}{#2}{}{#1}}}
\newrobustcmd{\widetildeJz}[2][]{\ensuremath{\subp{\widetilde{J}}{}{#2}{}{#1}}}
\newrobustcmd{\acuteJz}[2][]{\ensuremath{\subp{\acute{J}}{}{#2}{}{#1}}}
\newrobustcmd{\graveJz}[2][]{\ensuremath{\subp{\grave{J}}{}{#2}{}{#1}}}
\newrobustcmd{\dotJz}[2][]{\ensuremath{\subp{\dot{J}}{}{#2}{}{#1}}}
\newrobustcmd{\ddotJz}[2][]{\ensuremath{\subp{\ddot{J}}{}{#2}{}{#1}}}
\newrobustcmd{\breveJz}[2][]{\ensuremath{\subp{\breve{J}}{}{#2}{}{#1}}}
\newrobustcmd{\barJz}[2][]{\ensuremath{\subp{\bar{J}}{}{#2}{}{#1}}}
\newrobustcmd{\vecJz}[2][]{\ensuremath{\subp{\vec{J}}{}{#2}{}{#1}}}
\newrobustcmd{\bmJz}[2][]{\ensuremath{\subp{\bm{J}}{}{#2}{}{#1}}}
\newrobustcmd{\hatbmJz}[2][]{\ensuremath{\subp{\hat{\bm{J}}}{}{#2}{}{#1}}}
\newrobustcmd{\widehatbmJz}[2][]{\ensuremath{\subp{\widehat{\bm{J}}}{}{#2}{}{#1}}}
\newrobustcmd{\checkbmJz}[2][]{\ensuremath{\subp{\check{\bm{J}}}{}{#2}{}{#1}}}
\newrobustcmd{\tildebmJz}[2][]{\ensuremath{\subp{\tilde{\bm{J}}}{}{#2}{}{#1}}}
\newrobustcmd{\widetildebmJz}[2][]{\ensuremath{\subp{\widetilde{\bm{J}}}{}{#2}{}{#1}}}
\newrobustcmd{\acutebmJz}[2][]{\ensuremath{\subp{\acute{\bm{J}}}{}{#2}{}{#1}}}
\newrobustcmd{\gravebmJz}[2][]{\ensuremath{\subp{\grave{\bm{J}}}{}{#2}{}{#1}}}
\newrobustcmd{\dotbmJz}[2][]{\ensuremath{\subp{\dot{\bm{J}}}{}{#2}{}{#1}}}
\newrobustcmd{\ddotbmJz}[2][]{\ensuremath{\subp{\ddot{\bm{J}}}{}{#2}{}{#1}}}
\newrobustcmd{\brevebmJz}[2][]{\ensuremath{\subp{\breve{\bm{J}}}{}{#2}{}{#1}}}
\newrobustcmd{\barbmJz}[2][]{\ensuremath{\subp{\bar{\bm{J}}}{}{#2}{}{#1}}}
\newrobustcmd{\vecbmJz}[2][]{\ensuremath{\subp{\vec{\bm{J}}}{}{#2}{}{#1}}}
\newrobustcmd{\Kz}[2][]{\ensuremath{\subp{K}{}{#2}{}{#1}}}
\newrobustcmd{\hatKz}[2][]{\ensuremath{\subp{\hat{K}}{}{#2}{}{#1}}}
\newrobustcmd{\widehatKz}[2][]{\ensuremath{\subp{\widehat{K}}{}{#2}{}{#1}}}
\newrobustcmd{\checkKz}[2][]{\ensuremath{\subp{\check{K}}{}{#2}{}{#1}}}
\newrobustcmd{\tildeKz}[2][]{\ensuremath{\subp{\tilde{K}}{}{#2}{}{#1}}}
\newrobustcmd{\widetildeKz}[2][]{\ensuremath{\subp{\widetilde{K}}{}{#2}{}{#1}}}
\newrobustcmd{\acuteKz}[2][]{\ensuremath{\subp{\acute{K}}{}{#2}{}{#1}}}
\newrobustcmd{\graveKz}[2][]{\ensuremath{\subp{\grave{K}}{}{#2}{}{#1}}}
\newrobustcmd{\dotKz}[2][]{\ensuremath{\subp{\dot{K}}{}{#2}{}{#1}}}
\newrobustcmd{\ddotKz}[2][]{\ensuremath{\subp{\ddot{K}}{}{#2}{}{#1}}}
\newrobustcmd{\breveKz}[2][]{\ensuremath{\subp{\breve{K}}{}{#2}{}{#1}}}
\newrobustcmd{\barKz}[2][]{\ensuremath{\subp{\bar{K}}{}{#2}{}{#1}}}
\newrobustcmd{\vecKz}[2][]{\ensuremath{\subp{\vec{K}}{}{#2}{}{#1}}}
\newrobustcmd{\bmKz}[2][]{\ensuremath{\subp{\bm{K}}{}{#2}{}{#1}}}
\newrobustcmd{\hatbmKz}[2][]{\ensuremath{\subp{\hat{\bm{K}}}{}{#2}{}{#1}}}
\newrobustcmd{\widehatbmKz}[2][]{\ensuremath{\subp{\widehat{\bm{K}}}{}{#2}{}{#1}}}
\newrobustcmd{\checkbmKz}[2][]{\ensuremath{\subp{\check{\bm{K}}}{}{#2}{}{#1}}}
\newrobustcmd{\tildebmKz}[2][]{\ensuremath{\subp{\tilde{\bm{K}}}{}{#2}{}{#1}}}
\newrobustcmd{\widetildebmKz}[2][]{\ensuremath{\subp{\widetilde{\bm{K}}}{}{#2}{}{#1}}}
\newrobustcmd{\acutebmKz}[2][]{\ensuremath{\subp{\acute{\bm{K}}}{}{#2}{}{#1}}}
\newrobustcmd{\gravebmKz}[2][]{\ensuremath{\subp{\grave{\bm{K}}}{}{#2}{}{#1}}}
\newrobustcmd{\dotbmKz}[2][]{\ensuremath{\subp{\dot{\bm{K}}}{}{#2}{}{#1}}}
\newrobustcmd{\ddotbmKz}[2][]{\ensuremath{\subp{\ddot{\bm{K}}}{}{#2}{}{#1}}}
\newrobustcmd{\brevebmKz}[2][]{\ensuremath{\subp{\breve{\bm{K}}}{}{#2}{}{#1}}}
\newrobustcmd{\barbmKz}[2][]{\ensuremath{\subp{\bar{\bm{K}}}{}{#2}{}{#1}}}
\newrobustcmd{\vecbmKz}[2][]{\ensuremath{\subp{\vec{\bm{K}}}{}{#2}{}{#1}}}
\newrobustcmd{\Lz}[2][]{\ensuremath{\subp{L}{}{#2}{}{#1}}}
\newrobustcmd{\hatLz}[2][]{\ensuremath{\subp{\hat{L}}{}{#2}{}{#1}}}
\newrobustcmd{\widehatLz}[2][]{\ensuremath{\subp{\widehat{L}}{}{#2}{}{#1}}}
\newrobustcmd{\checkLz}[2][]{\ensuremath{\subp{\check{L}}{}{#2}{}{#1}}}
\newrobustcmd{\tildeLz}[2][]{\ensuremath{\subp{\tilde{L}}{}{#2}{}{#1}}}
\newrobustcmd{\widetildeLz}[2][]{\ensuremath{\subp{\widetilde{L}}{}{#2}{}{#1}}}
\newrobustcmd{\acuteLz}[2][]{\ensuremath{\subp{\acute{L}}{}{#2}{}{#1}}}
\newrobustcmd{\graveLz}[2][]{\ensuremath{\subp{\grave{L}}{}{#2}{}{#1}}}
\newrobustcmd{\dotLz}[2][]{\ensuremath{\subp{\dot{L}}{}{#2}{}{#1}}}
\newrobustcmd{\ddotLz}[2][]{\ensuremath{\subp{\ddot{L}}{}{#2}{}{#1}}}
\newrobustcmd{\breveLz}[2][]{\ensuremath{\subp{\breve{L}}{}{#2}{}{#1}}}
\newrobustcmd{\barLz}[2][]{\ensuremath{\subp{\bar{L}}{}{#2}{}{#1}}}
\newrobustcmd{\vecLz}[2][]{\ensuremath{\subp{\vec{L}}{}{#2}{}{#1}}}
\newrobustcmd{\bmLz}[2][]{\ensuremath{\subp{\bm{L}}{}{#2}{}{#1}}}
\newrobustcmd{\hatbmLz}[2][]{\ensuremath{\subp{\hat{\bm{L}}}{}{#2}{}{#1}}}
\newrobustcmd{\widehatbmLz}[2][]{\ensuremath{\subp{\widehat{\bm{L}}}{}{#2}{}{#1}}}
\newrobustcmd{\checkbmLz}[2][]{\ensuremath{\subp{\check{\bm{L}}}{}{#2}{}{#1}}}
\newrobustcmd{\tildebmLz}[2][]{\ensuremath{\subp{\tilde{\bm{L}}}{}{#2}{}{#1}}}
\newrobustcmd{\widetildebmLz}[2][]{\ensuremath{\subp{\widetilde{\bm{L}}}{}{#2}{}{#1}}}
\newrobustcmd{\acutebmLz}[2][]{\ensuremath{\subp{\acute{\bm{L}}}{}{#2}{}{#1}}}
\newrobustcmd{\gravebmLz}[2][]{\ensuremath{\subp{\grave{\bm{L}}}{}{#2}{}{#1}}}
\newrobustcmd{\dotbmLz}[2][]{\ensuremath{\subp{\dot{\bm{L}}}{}{#2}{}{#1}}}
\newrobustcmd{\ddotbmLz}[2][]{\ensuremath{\subp{\ddot{\bm{L}}}{}{#2}{}{#1}}}
\newrobustcmd{\brevebmLz}[2][]{\ensuremath{\subp{\breve{\bm{L}}}{}{#2}{}{#1}}}
\newrobustcmd{\barbmLz}[2][]{\ensuremath{\subp{\bar{\bm{L}}}{}{#2}{}{#1}}}
\newrobustcmd{\vecbmLz}[2][]{\ensuremath{\subp{\vec{\bm{L}}}{}{#2}{}{#1}}}
\newrobustcmd{\Mz}[2][]{\ensuremath{\subp{M}{}{#2}{}{#1}}}
\newrobustcmd{\hatMz}[2][]{\ensuremath{\subp{\hat{M}}{}{#2}{}{#1}}}
\newrobustcmd{\widehatMz}[2][]{\ensuremath{\subp{\widehat{M}}{}{#2}{}{#1}}}
\newrobustcmd{\checkMz}[2][]{\ensuremath{\subp{\check{M}}{}{#2}{}{#1}}}
\newrobustcmd{\tildeMz}[2][]{\ensuremath{\subp{\tilde{M}}{}{#2}{}{#1}}}
\newrobustcmd{\widetildeMz}[2][]{\ensuremath{\subp{\widetilde{M}}{}{#2}{}{#1}}}
\newrobustcmd{\acuteMz}[2][]{\ensuremath{\subp{\acute{M}}{}{#2}{}{#1}}}
\newrobustcmd{\graveMz}[2][]{\ensuremath{\subp{\grave{M}}{}{#2}{}{#1}}}
\newrobustcmd{\dotMz}[2][]{\ensuremath{\subp{\dot{M}}{}{#2}{}{#1}}}
\newrobustcmd{\ddotMz}[2][]{\ensuremath{\subp{\ddot{M}}{}{#2}{}{#1}}}
\newrobustcmd{\breveMz}[2][]{\ensuremath{\subp{\breve{M}}{}{#2}{}{#1}}}
\newrobustcmd{\barMz}[2][]{\ensuremath{\subp{\bar{M}}{}{#2}{}{#1}}}
\newrobustcmd{\vecMz}[2][]{\ensuremath{\subp{\vec{M}}{}{#2}{}{#1}}}
\newrobustcmd{\bmMz}[2][]{\ensuremath{\subp{\bm{M}}{}{#2}{}{#1}}}
\newrobustcmd{\hatbmMz}[2][]{\ensuremath{\subp{\hat{\bm{M}}}{}{#2}{}{#1}}}
\newrobustcmd{\widehatbmMz}[2][]{\ensuremath{\subp{\widehat{\bm{M}}}{}{#2}{}{#1}}}
\newrobustcmd{\checkbmMz}[2][]{\ensuremath{\subp{\check{\bm{M}}}{}{#2}{}{#1}}}
\newrobustcmd{\tildebmMz}[2][]{\ensuremath{\subp{\tilde{\bm{M}}}{}{#2}{}{#1}}}
\newrobustcmd{\widetildebmMz}[2][]{\ensuremath{\subp{\widetilde{\bm{M}}}{}{#2}{}{#1}}}
\newrobustcmd{\acutebmMz}[2][]{\ensuremath{\subp{\acute{\bm{M}}}{}{#2}{}{#1}}}
\newrobustcmd{\gravebmMz}[2][]{\ensuremath{\subp{\grave{\bm{M}}}{}{#2}{}{#1}}}
\newrobustcmd{\dotbmMz}[2][]{\ensuremath{\subp{\dot{\bm{M}}}{}{#2}{}{#1}}}
\newrobustcmd{\ddotbmMz}[2][]{\ensuremath{\subp{\ddot{\bm{M}}}{}{#2}{}{#1}}}
\newrobustcmd{\brevebmMz}[2][]{\ensuremath{\subp{\breve{\bm{M}}}{}{#2}{}{#1}}}
\newrobustcmd{\barbmMz}[2][]{\ensuremath{\subp{\bar{\bm{M}}}{}{#2}{}{#1}}}
\newrobustcmd{\vecbmMz}[2][]{\ensuremath{\subp{\vec{\bm{M}}}{}{#2}{}{#1}}}
\newrobustcmd{\Nz}[2][]{\ensuremath{\subp{N}{}{#2}{}{#1}}}
\newrobustcmd{\hatNz}[2][]{\ensuremath{\subp{\hat{N}}{}{#2}{}{#1}}}
\newrobustcmd{\widehatNz}[2][]{\ensuremath{\subp{\widehat{N}}{}{#2}{}{#1}}}
\newrobustcmd{\checkNz}[2][]{\ensuremath{\subp{\check{N}}{}{#2}{}{#1}}}
\newrobustcmd{\tildeNz}[2][]{\ensuremath{\subp{\tilde{N}}{}{#2}{}{#1}}}
\newrobustcmd{\widetildeNz}[2][]{\ensuremath{\subp{\widetilde{N}}{}{#2}{}{#1}}}
\newrobustcmd{\acuteNz}[2][]{\ensuremath{\subp{\acute{N}}{}{#2}{}{#1}}}
\newrobustcmd{\graveNz}[2][]{\ensuremath{\subp{\grave{N}}{}{#2}{}{#1}}}
\newrobustcmd{\dotNz}[2][]{\ensuremath{\subp{\dot{N}}{}{#2}{}{#1}}}
\newrobustcmd{\ddotNz}[2][]{\ensuremath{\subp{\ddot{N}}{}{#2}{}{#1}}}
\newrobustcmd{\breveNz}[2][]{\ensuremath{\subp{\breve{N}}{}{#2}{}{#1}}}
\newrobustcmd{\barNz}[2][]{\ensuremath{\subp{\bar{N}}{}{#2}{}{#1}}}
\newrobustcmd{\vecNz}[2][]{\ensuremath{\subp{\vec{N}}{}{#2}{}{#1}}}
\newrobustcmd{\bmNz}[2][]{\ensuremath{\subp{\bm{N}}{}{#2}{}{#1}}}
\newrobustcmd{\hatbmNz}[2][]{\ensuremath{\subp{\hat{\bm{N}}}{}{#2}{}{#1}}}
\newrobustcmd{\widehatbmNz}[2][]{\ensuremath{\subp{\widehat{\bm{N}}}{}{#2}{}{#1}}}
\newrobustcmd{\checkbmNz}[2][]{\ensuremath{\subp{\check{\bm{N}}}{}{#2}{}{#1}}}
\newrobustcmd{\tildebmNz}[2][]{\ensuremath{\subp{\tilde{\bm{N}}}{}{#2}{}{#1}}}
\newrobustcmd{\widetildebmNz}[2][]{\ensuremath{\subp{\widetilde{\bm{N}}}{}{#2}{}{#1}}}
\newrobustcmd{\acutebmNz}[2][]{\ensuremath{\subp{\acute{\bm{N}}}{}{#2}{}{#1}}}
\newrobustcmd{\gravebmNz}[2][]{\ensuremath{\subp{\grave{\bm{N}}}{}{#2}{}{#1}}}
\newrobustcmd{\dotbmNz}[2][]{\ensuremath{\subp{\dot{\bm{N}}}{}{#2}{}{#1}}}
\newrobustcmd{\ddotbmNz}[2][]{\ensuremath{\subp{\ddot{\bm{N}}}{}{#2}{}{#1}}}
\newrobustcmd{\brevebmNz}[2][]{\ensuremath{\subp{\breve{\bm{N}}}{}{#2}{}{#1}}}
\newrobustcmd{\barbmNz}[2][]{\ensuremath{\subp{\bar{\bm{N}}}{}{#2}{}{#1}}}
\newrobustcmd{\vecbmNz}[2][]{\ensuremath{\subp{\vec{\bm{N}}}{}{#2}{}{#1}}}
\newrobustcmd{\Oz}[2][]{\ensuremath{\subp{O}{}{#2}{}{#1}}}
\newrobustcmd{\hatOz}[2][]{\ensuremath{\subp{\hat{O}}{}{#2}{}{#1}}}
\newrobustcmd{\widehatOz}[2][]{\ensuremath{\subp{\widehat{O}}{}{#2}{}{#1}}}
\newrobustcmd{\checkOz}[2][]{\ensuremath{\subp{\check{O}}{}{#2}{}{#1}}}
\newrobustcmd{\tildeOz}[2][]{\ensuremath{\subp{\tilde{O}}{}{#2}{}{#1}}}
\newrobustcmd{\widetildeOz}[2][]{\ensuremath{\subp{\widetilde{O}}{}{#2}{}{#1}}}
\newrobustcmd{\acuteOz}[2][]{\ensuremath{\subp{\acute{O}}{}{#2}{}{#1}}}
\newrobustcmd{\graveOz}[2][]{\ensuremath{\subp{\grave{O}}{}{#2}{}{#1}}}
\newrobustcmd{\dotOz}[2][]{\ensuremath{\subp{\dot{O}}{}{#2}{}{#1}}}
\newrobustcmd{\ddotOz}[2][]{\ensuremath{\subp{\ddot{O}}{}{#2}{}{#1}}}
\newrobustcmd{\breveOz}[2][]{\ensuremath{\subp{\breve{O}}{}{#2}{}{#1}}}
\newrobustcmd{\barOz}[2][]{\ensuremath{\subp{\bar{O}}{}{#2}{}{#1}}}
\newrobustcmd{\vecOz}[2][]{\ensuremath{\subp{\vec{O}}{}{#2}{}{#1}}}
\newrobustcmd{\bmOz}[2][]{\ensuremath{\subp{\bm{O}}{}{#2}{}{#1}}}
\newrobustcmd{\hatbmOz}[2][]{\ensuremath{\subp{\hat{\bm{O}}}{}{#2}{}{#1}}}
\newrobustcmd{\widehatbmOz}[2][]{\ensuremath{\subp{\widehat{\bm{O}}}{}{#2}{}{#1}}}
\newrobustcmd{\checkbmOz}[2][]{\ensuremath{\subp{\check{\bm{O}}}{}{#2}{}{#1}}}
\newrobustcmd{\tildebmOz}[2][]{\ensuremath{\subp{\tilde{\bm{O}}}{}{#2}{}{#1}}}
\newrobustcmd{\widetildebmOz}[2][]{\ensuremath{\subp{\widetilde{\bm{O}}}{}{#2}{}{#1}}}
\newrobustcmd{\acutebmOz}[2][]{\ensuremath{\subp{\acute{\bm{O}}}{}{#2}{}{#1}}}
\newrobustcmd{\gravebmOz}[2][]{\ensuremath{\subp{\grave{\bm{O}}}{}{#2}{}{#1}}}
\newrobustcmd{\dotbmOz}[2][]{\ensuremath{\subp{\dot{\bm{O}}}{}{#2}{}{#1}}}
\newrobustcmd{\ddotbmOz}[2][]{\ensuremath{\subp{\ddot{\bm{O}}}{}{#2}{}{#1}}}
\newrobustcmd{\brevebmOz}[2][]{\ensuremath{\subp{\breve{\bm{O}}}{}{#2}{}{#1}}}
\newrobustcmd{\barbmOz}[2][]{\ensuremath{\subp{\bar{\bm{O}}}{}{#2}{}{#1}}}
\newrobustcmd{\vecbmOz}[2][]{\ensuremath{\subp{\vec{\bm{O}}}{}{#2}{}{#1}}}
\newrobustcmd{\Pz}[2][]{\ensuremath{\subp{P}{}{#2}{}{#1}}}
\newrobustcmd{\hatPz}[2][]{\ensuremath{\subp{\hat{P}}{}{#2}{}{#1}}}
\newrobustcmd{\widehatPz}[2][]{\ensuremath{\subp{\widehat{P}}{}{#2}{}{#1}}}
\newrobustcmd{\checkPz}[2][]{\ensuremath{\subp{\check{P}}{}{#2}{}{#1}}}
\newrobustcmd{\tildePz}[2][]{\ensuremath{\subp{\tilde{P}}{}{#2}{}{#1}}}
\newrobustcmd{\widetildePz}[2][]{\ensuremath{\subp{\widetilde{P}}{}{#2}{}{#1}}}
\newrobustcmd{\acutePz}[2][]{\ensuremath{\subp{\acute{P}}{}{#2}{}{#1}}}
\newrobustcmd{\gravePz}[2][]{\ensuremath{\subp{\grave{P}}{}{#2}{}{#1}}}
\newrobustcmd{\dotPz}[2][]{\ensuremath{\subp{\dot{P}}{}{#2}{}{#1}}}
\newrobustcmd{\ddotPz}[2][]{\ensuremath{\subp{\ddot{P}}{}{#2}{}{#1}}}
\newrobustcmd{\brevePz}[2][]{\ensuremath{\subp{\breve{P}}{}{#2}{}{#1}}}
\newrobustcmd{\barPz}[2][]{\ensuremath{\subp{\bar{P}}{}{#2}{}{#1}}}
\newrobustcmd{\vecPz}[2][]{\ensuremath{\subp{\vec{P}}{}{#2}{}{#1}}}
\newrobustcmd{\bmPz}[2][]{\ensuremath{\subp{\bm{P}}{}{#2}{}{#1}}}
\newrobustcmd{\hatbmPz}[2][]{\ensuremath{\subp{\hat{\bm{P}}}{}{#2}{}{#1}}}
\newrobustcmd{\widehatbmPz}[2][]{\ensuremath{\subp{\widehat{\bm{P}}}{}{#2}{}{#1}}}
\newrobustcmd{\checkbmPz}[2][]{\ensuremath{\subp{\check{\bm{P}}}{}{#2}{}{#1}}}
\newrobustcmd{\tildebmPz}[2][]{\ensuremath{\subp{\tilde{\bm{P}}}{}{#2}{}{#1}}}
\newrobustcmd{\widetildebmPz}[2][]{\ensuremath{\subp{\widetilde{\bm{P}}}{}{#2}{}{#1}}}
\newrobustcmd{\acutebmPz}[2][]{\ensuremath{\subp{\acute{\bm{P}}}{}{#2}{}{#1}}}
\newrobustcmd{\gravebmPz}[2][]{\ensuremath{\subp{\grave{\bm{P}}}{}{#2}{}{#1}}}
\newrobustcmd{\dotbmPz}[2][]{\ensuremath{\subp{\dot{\bm{P}}}{}{#2}{}{#1}}}
\newrobustcmd{\ddotbmPz}[2][]{\ensuremath{\subp{\ddot{\bm{P}}}{}{#2}{}{#1}}}
\newrobustcmd{\brevebmPz}[2][]{\ensuremath{\subp{\breve{\bm{P}}}{}{#2}{}{#1}}}
\newrobustcmd{\barbmPz}[2][]{\ensuremath{\subp{\bar{\bm{P}}}{}{#2}{}{#1}}}
\newrobustcmd{\vecbmPz}[2][]{\ensuremath{\subp{\vec{\bm{P}}}{}{#2}{}{#1}}}
\newrobustcmd{\Qz}[2][]{\ensuremath{\subp{Q}{}{#2}{}{#1}}}
\newrobustcmd{\hatQz}[2][]{\ensuremath{\subp{\hat{Q}}{}{#2}{}{#1}}}
\newrobustcmd{\widehatQz}[2][]{\ensuremath{\subp{\widehat{Q}}{}{#2}{}{#1}}}
\newrobustcmd{\checkQz}[2][]{\ensuremath{\subp{\check{Q}}{}{#2}{}{#1}}}
\newrobustcmd{\tildeQz}[2][]{\ensuremath{\subp{\tilde{Q}}{}{#2}{}{#1}}}
\newrobustcmd{\widetildeQz}[2][]{\ensuremath{\subp{\widetilde{Q}}{}{#2}{}{#1}}}
\newrobustcmd{\acuteQz}[2][]{\ensuremath{\subp{\acute{Q}}{}{#2}{}{#1}}}
\newrobustcmd{\graveQz}[2][]{\ensuremath{\subp{\grave{Q}}{}{#2}{}{#1}}}
\newrobustcmd{\dotQz}[2][]{\ensuremath{\subp{\dot{Q}}{}{#2}{}{#1}}}
\newrobustcmd{\ddotQz}[2][]{\ensuremath{\subp{\ddot{Q}}{}{#2}{}{#1}}}
\newrobustcmd{\breveQz}[2][]{\ensuremath{\subp{\breve{Q}}{}{#2}{}{#1}}}
\newrobustcmd{\barQz}[2][]{\ensuremath{\subp{\bar{Q}}{}{#2}{}{#1}}}
\newrobustcmd{\vecQz}[2][]{\ensuremath{\subp{\vec{Q}}{}{#2}{}{#1}}}
\newrobustcmd{\bmQz}[2][]{\ensuremath{\subp{\bm{Q}}{}{#2}{}{#1}}}
\newrobustcmd{\hatbmQz}[2][]{\ensuremath{\subp{\hat{\bm{Q}}}{}{#2}{}{#1}}}
\newrobustcmd{\widehatbmQz}[2][]{\ensuremath{\subp{\widehat{\bm{Q}}}{}{#2}{}{#1}}}
\newrobustcmd{\checkbmQz}[2][]{\ensuremath{\subp{\check{\bm{Q}}}{}{#2}{}{#1}}}
\newrobustcmd{\tildebmQz}[2][]{\ensuremath{\subp{\tilde{\bm{Q}}}{}{#2}{}{#1}}}
\newrobustcmd{\widetildebmQz}[2][]{\ensuremath{\subp{\widetilde{\bm{Q}}}{}{#2}{}{#1}}}
\newrobustcmd{\acutebmQz}[2][]{\ensuremath{\subp{\acute{\bm{Q}}}{}{#2}{}{#1}}}
\newrobustcmd{\gravebmQz}[2][]{\ensuremath{\subp{\grave{\bm{Q}}}{}{#2}{}{#1}}}
\newrobustcmd{\dotbmQz}[2][]{\ensuremath{\subp{\dot{\bm{Q}}}{}{#2}{}{#1}}}
\newrobustcmd{\ddotbmQz}[2][]{\ensuremath{\subp{\ddot{\bm{Q}}}{}{#2}{}{#1}}}
\newrobustcmd{\brevebmQz}[2][]{\ensuremath{\subp{\breve{\bm{Q}}}{}{#2}{}{#1}}}
\newrobustcmd{\barbmQz}[2][]{\ensuremath{\subp{\bar{\bm{Q}}}{}{#2}{}{#1}}}
\newrobustcmd{\vecbmQz}[2][]{\ensuremath{\subp{\vec{\bm{Q}}}{}{#2}{}{#1}}}
\newrobustcmd{\Rz}[2][]{\ensuremath{\subp{R}{}{#2}{}{#1}}}
\newrobustcmd{\hatRz}[2][]{\ensuremath{\subp{\hat{R}}{}{#2}{}{#1}}}
\newrobustcmd{\widehatRz}[2][]{\ensuremath{\subp{\widehat{R}}{}{#2}{}{#1}}}
\newrobustcmd{\checkRz}[2][]{\ensuremath{\subp{\check{R}}{}{#2}{}{#1}}}
\newrobustcmd{\tildeRz}[2][]{\ensuremath{\subp{\tilde{R}}{}{#2}{}{#1}}}
\newrobustcmd{\widetildeRz}[2][]{\ensuremath{\subp{\widetilde{R}}{}{#2}{}{#1}}}
\newrobustcmd{\acuteRz}[2][]{\ensuremath{\subp{\acute{R}}{}{#2}{}{#1}}}
\newrobustcmd{\graveRz}[2][]{\ensuremath{\subp{\grave{R}}{}{#2}{}{#1}}}
\newrobustcmd{\dotRz}[2][]{\ensuremath{\subp{\dot{R}}{}{#2}{}{#1}}}
\newrobustcmd{\ddotRz}[2][]{\ensuremath{\subp{\ddot{R}}{}{#2}{}{#1}}}
\newrobustcmd{\breveRz}[2][]{\ensuremath{\subp{\breve{R}}{}{#2}{}{#1}}}
\newrobustcmd{\barRz}[2][]{\ensuremath{\subp{\bar{R}}{}{#2}{}{#1}}}
\newrobustcmd{\vecRz}[2][]{\ensuremath{\subp{\vec{R}}{}{#2}{}{#1}}}
\newrobustcmd{\bmRz}[2][]{\ensuremath{\subp{\bm{R}}{}{#2}{}{#1}}}
\newrobustcmd{\hatbmRz}[2][]{\ensuremath{\subp{\hat{\bm{R}}}{}{#2}{}{#1}}}
\newrobustcmd{\widehatbmRz}[2][]{\ensuremath{\subp{\widehat{\bm{R}}}{}{#2}{}{#1}}}
\newrobustcmd{\checkbmRz}[2][]{\ensuremath{\subp{\check{\bm{R}}}{}{#2}{}{#1}}}
\newrobustcmd{\tildebmRz}[2][]{\ensuremath{\subp{\tilde{\bm{R}}}{}{#2}{}{#1}}}
\newrobustcmd{\widetildebmRz}[2][]{\ensuremath{\subp{\widetilde{\bm{R}}}{}{#2}{}{#1}}}
\newrobustcmd{\acutebmRz}[2][]{\ensuremath{\subp{\acute{\bm{R}}}{}{#2}{}{#1}}}
\newrobustcmd{\gravebmRz}[2][]{\ensuremath{\subp{\grave{\bm{R}}}{}{#2}{}{#1}}}
\newrobustcmd{\dotbmRz}[2][]{\ensuremath{\subp{\dot{\bm{R}}}{}{#2}{}{#1}}}
\newrobustcmd{\ddotbmRz}[2][]{\ensuremath{\subp{\ddot{\bm{R}}}{}{#2}{}{#1}}}
\newrobustcmd{\brevebmRz}[2][]{\ensuremath{\subp{\breve{\bm{R}}}{}{#2}{}{#1}}}
\newrobustcmd{\barbmRz}[2][]{\ensuremath{\subp{\bar{\bm{R}}}{}{#2}{}{#1}}}
\newrobustcmd{\vecbmRz}[2][]{\ensuremath{\subp{\vec{\bm{R}}}{}{#2}{}{#1}}}
\newrobustcmd{\Sz}[2][]{\ensuremath{\subp{S}{}{#2}{}{#1}}}
\newrobustcmd{\hatSz}[2][]{\ensuremath{\subp{\hat{S}}{}{#2}{}{#1}}}
\newrobustcmd{\widehatSz}[2][]{\ensuremath{\subp{\widehat{S}}{}{#2}{}{#1}}}
\newrobustcmd{\checkSz}[2][]{\ensuremath{\subp{\check{S}}{}{#2}{}{#1}}}
\newrobustcmd{\tildeSz}[2][]{\ensuremath{\subp{\tilde{S}}{}{#2}{}{#1}}}
\newrobustcmd{\widetildeSz}[2][]{\ensuremath{\subp{\widetilde{S}}{}{#2}{}{#1}}}
\newrobustcmd{\acuteSz}[2][]{\ensuremath{\subp{\acute{S}}{}{#2}{}{#1}}}
\newrobustcmd{\graveSz}[2][]{\ensuremath{\subp{\grave{S}}{}{#2}{}{#1}}}
\newrobustcmd{\dotSz}[2][]{\ensuremath{\subp{\dot{S}}{}{#2}{}{#1}}}
\newrobustcmd{\ddotSz}[2][]{\ensuremath{\subp{\ddot{S}}{}{#2}{}{#1}}}
\newrobustcmd{\breveSz}[2][]{\ensuremath{\subp{\breve{S}}{}{#2}{}{#1}}}
\newrobustcmd{\barSz}[2][]{\ensuremath{\subp{\bar{S}}{}{#2}{}{#1}}}
\newrobustcmd{\vecSz}[2][]{\ensuremath{\subp{\vec{S}}{}{#2}{}{#1}}}
\newrobustcmd{\bmSz}[2][]{\ensuremath{\subp{\bm{S}}{}{#2}{}{#1}}}
\newrobustcmd{\hatbmSz}[2][]{\ensuremath{\subp{\hat{\bm{S}}}{}{#2}{}{#1}}}
\newrobustcmd{\widehatbmSz}[2][]{\ensuremath{\subp{\widehat{\bm{S}}}{}{#2}{}{#1}}}
\newrobustcmd{\checkbmSz}[2][]{\ensuremath{\subp{\check{\bm{S}}}{}{#2}{}{#1}}}
\newrobustcmd{\tildebmSz}[2][]{\ensuremath{\subp{\tilde{\bm{S}}}{}{#2}{}{#1}}}
\newrobustcmd{\widetildebmSz}[2][]{\ensuremath{\subp{\widetilde{\bm{S}}}{}{#2}{}{#1}}}
\newrobustcmd{\acutebmSz}[2][]{\ensuremath{\subp{\acute{\bm{S}}}{}{#2}{}{#1}}}
\newrobustcmd{\gravebmSz}[2][]{\ensuremath{\subp{\grave{\bm{S}}}{}{#2}{}{#1}}}
\newrobustcmd{\dotbmSz}[2][]{\ensuremath{\subp{\dot{\bm{S}}}{}{#2}{}{#1}}}
\newrobustcmd{\ddotbmSz}[2][]{\ensuremath{\subp{\ddot{\bm{S}}}{}{#2}{}{#1}}}
\newrobustcmd{\brevebmSz}[2][]{\ensuremath{\subp{\breve{\bm{S}}}{}{#2}{}{#1}}}
\newrobustcmd{\barbmSz}[2][]{\ensuremath{\subp{\bar{\bm{S}}}{}{#2}{}{#1}}}
\newrobustcmd{\vecbmSz}[2][]{\ensuremath{\subp{\vec{\bm{S}}}{}{#2}{}{#1}}}
\newrobustcmd{\Tz}[2][]{\ensuremath{\subp{T}{}{#2}{}{#1}}}
\newrobustcmd{\hatTz}[2][]{\ensuremath{\subp{\hat{T}}{}{#2}{}{#1}}}
\newrobustcmd{\widehatTz}[2][]{\ensuremath{\subp{\widehat{T}}{}{#2}{}{#1}}}
\newrobustcmd{\checkTz}[2][]{\ensuremath{\subp{\check{T}}{}{#2}{}{#1}}}
\newrobustcmd{\tildeTz}[2][]{\ensuremath{\subp{\tilde{T}}{}{#2}{}{#1}}}
\newrobustcmd{\widetildeTz}[2][]{\ensuremath{\subp{\widetilde{T}}{}{#2}{}{#1}}}
\newrobustcmd{\acuteTz}[2][]{\ensuremath{\subp{\acute{T}}{}{#2}{}{#1}}}
\newrobustcmd{\graveTz}[2][]{\ensuremath{\subp{\grave{T}}{}{#2}{}{#1}}}
\newrobustcmd{\dotTz}[2][]{\ensuremath{\subp{\dot{T}}{}{#2}{}{#1}}}
\newrobustcmd{\ddotTz}[2][]{\ensuremath{\subp{\ddot{T}}{}{#2}{}{#1}}}
\newrobustcmd{\breveTz}[2][]{\ensuremath{\subp{\breve{T}}{}{#2}{}{#1}}}
\newrobustcmd{\barTz}[2][]{\ensuremath{\subp{\bar{T}}{}{#2}{}{#1}}}
\newrobustcmd{\vecTz}[2][]{\ensuremath{\subp{\vec{T}}{}{#2}{}{#1}}}
\newrobustcmd{\bmTz}[2][]{\ensuremath{\subp{\bm{T}}{}{#2}{}{#1}}}
\newrobustcmd{\hatbmTz}[2][]{\ensuremath{\subp{\hat{\bm{T}}}{}{#2}{}{#1}}}
\newrobustcmd{\widehatbmTz}[2][]{\ensuremath{\subp{\widehat{\bm{T}}}{}{#2}{}{#1}}}
\newrobustcmd{\checkbmTz}[2][]{\ensuremath{\subp{\check{\bm{T}}}{}{#2}{}{#1}}}
\newrobustcmd{\tildebmTz}[2][]{\ensuremath{\subp{\tilde{\bm{T}}}{}{#2}{}{#1}}}
\newrobustcmd{\widetildebmTz}[2][]{\ensuremath{\subp{\widetilde{\bm{T}}}{}{#2}{}{#1}}}
\newrobustcmd{\acutebmTz}[2][]{\ensuremath{\subp{\acute{\bm{T}}}{}{#2}{}{#1}}}
\newrobustcmd{\gravebmTz}[2][]{\ensuremath{\subp{\grave{\bm{T}}}{}{#2}{}{#1}}}
\newrobustcmd{\dotbmTz}[2][]{\ensuremath{\subp{\dot{\bm{T}}}{}{#2}{}{#1}}}
\newrobustcmd{\ddotbmTz}[2][]{\ensuremath{\subp{\ddot{\bm{T}}}{}{#2}{}{#1}}}
\newrobustcmd{\brevebmTz}[2][]{\ensuremath{\subp{\breve{\bm{T}}}{}{#2}{}{#1}}}
\newrobustcmd{\barbmTz}[2][]{\ensuremath{\subp{\bar{\bm{T}}}{}{#2}{}{#1}}}
\newrobustcmd{\vecbmTz}[2][]{\ensuremath{\subp{\vec{\bm{T}}}{}{#2}{}{#1}}}
\newrobustcmd{\Uz}[2][]{\ensuremath{\subp{U}{}{#2}{}{#1}}}
\newrobustcmd{\hatUz}[2][]{\ensuremath{\subp{\hat{U}}{}{#2}{}{#1}}}
\newrobustcmd{\widehatUz}[2][]{\ensuremath{\subp{\widehat{U}}{}{#2}{}{#1}}}
\newrobustcmd{\checkUz}[2][]{\ensuremath{\subp{\check{U}}{}{#2}{}{#1}}}
\newrobustcmd{\tildeUz}[2][]{\ensuremath{\subp{\tilde{U}}{}{#2}{}{#1}}}
\newrobustcmd{\widetildeUz}[2][]{\ensuremath{\subp{\widetilde{U}}{}{#2}{}{#1}}}
\newrobustcmd{\acuteUz}[2][]{\ensuremath{\subp{\acute{U}}{}{#2}{}{#1}}}
\newrobustcmd{\graveUz}[2][]{\ensuremath{\subp{\grave{U}}{}{#2}{}{#1}}}
\newrobustcmd{\dotUz}[2][]{\ensuremath{\subp{\dot{U}}{}{#2}{}{#1}}}
\newrobustcmd{\ddotUz}[2][]{\ensuremath{\subp{\ddot{U}}{}{#2}{}{#1}}}
\newrobustcmd{\breveUz}[2][]{\ensuremath{\subp{\breve{U}}{}{#2}{}{#1}}}
\newrobustcmd{\barUz}[2][]{\ensuremath{\subp{\bar{U}}{}{#2}{}{#1}}}
\newrobustcmd{\vecUz}[2][]{\ensuremath{\subp{\vec{U}}{}{#2}{}{#1}}}
\newrobustcmd{\bmUz}[2][]{\ensuremath{\subp{\bm{U}}{}{#2}{}{#1}}}
\newrobustcmd{\hatbmUz}[2][]{\ensuremath{\subp{\hat{\bm{U}}}{}{#2}{}{#1}}}
\newrobustcmd{\widehatbmUz}[2][]{\ensuremath{\subp{\widehat{\bm{U}}}{}{#2}{}{#1}}}
\newrobustcmd{\checkbmUz}[2][]{\ensuremath{\subp{\check{\bm{U}}}{}{#2}{}{#1}}}
\newrobustcmd{\tildebmUz}[2][]{\ensuremath{\subp{\tilde{\bm{U}}}{}{#2}{}{#1}}}
\newrobustcmd{\widetildebmUz}[2][]{\ensuremath{\subp{\widetilde{\bm{U}}}{}{#2}{}{#1}}}
\newrobustcmd{\acutebmUz}[2][]{\ensuremath{\subp{\acute{\bm{U}}}{}{#2}{}{#1}}}
\newrobustcmd{\gravebmUz}[2][]{\ensuremath{\subp{\grave{\bm{U}}}{}{#2}{}{#1}}}
\newrobustcmd{\dotbmUz}[2][]{\ensuremath{\subp{\dot{\bm{U}}}{}{#2}{}{#1}}}
\newrobustcmd{\ddotbmUz}[2][]{\ensuremath{\subp{\ddot{\bm{U}}}{}{#2}{}{#1}}}
\newrobustcmd{\brevebmUz}[2][]{\ensuremath{\subp{\breve{\bm{U}}}{}{#2}{}{#1}}}
\newrobustcmd{\barbmUz}[2][]{\ensuremath{\subp{\bar{\bm{U}}}{}{#2}{}{#1}}}
\newrobustcmd{\vecbmUz}[2][]{\ensuremath{\subp{\vec{\bm{U}}}{}{#2}{}{#1}}}
\newrobustcmd{\Vz}[2][]{\ensuremath{\subp{V}{}{#2}{}{#1}}}
\newrobustcmd{\hatVz}[2][]{\ensuremath{\subp{\hat{V}}{}{#2}{}{#1}}}
\newrobustcmd{\widehatVz}[2][]{\ensuremath{\subp{\widehat{V}}{}{#2}{}{#1}}}
\newrobustcmd{\checkVz}[2][]{\ensuremath{\subp{\check{V}}{}{#2}{}{#1}}}
\newrobustcmd{\tildeVz}[2][]{\ensuremath{\subp{\tilde{V}}{}{#2}{}{#1}}}
\newrobustcmd{\widetildeVz}[2][]{\ensuremath{\subp{\widetilde{V}}{}{#2}{}{#1}}}
\newrobustcmd{\acuteVz}[2][]{\ensuremath{\subp{\acute{V}}{}{#2}{}{#1}}}
\newrobustcmd{\graveVz}[2][]{\ensuremath{\subp{\grave{V}}{}{#2}{}{#1}}}
\newrobustcmd{\dotVz}[2][]{\ensuremath{\subp{\dot{V}}{}{#2}{}{#1}}}
\newrobustcmd{\ddotVz}[2][]{\ensuremath{\subp{\ddot{V}}{}{#2}{}{#1}}}
\newrobustcmd{\breveVz}[2][]{\ensuremath{\subp{\breve{V}}{}{#2}{}{#1}}}
\newrobustcmd{\barVz}[2][]{\ensuremath{\subp{\bar{V}}{}{#2}{}{#1}}}
\newrobustcmd{\vecVz}[2][]{\ensuremath{\subp{\vec{V}}{}{#2}{}{#1}}}
\newrobustcmd{\bmVz}[2][]{\ensuremath{\subp{\bm{V}}{}{#2}{}{#1}}}
\newrobustcmd{\hatbmVz}[2][]{\ensuremath{\subp{\hat{\bm{V}}}{}{#2}{}{#1}}}
\newrobustcmd{\widehatbmVz}[2][]{\ensuremath{\subp{\widehat{\bm{V}}}{}{#2}{}{#1}}}
\newrobustcmd{\checkbmVz}[2][]{\ensuremath{\subp{\check{\bm{V}}}{}{#2}{}{#1}}}
\newrobustcmd{\tildebmVz}[2][]{\ensuremath{\subp{\tilde{\bm{V}}}{}{#2}{}{#1}}}
\newrobustcmd{\widetildebmVz}[2][]{\ensuremath{\subp{\widetilde{\bm{V}}}{}{#2}{}{#1}}}
\newrobustcmd{\acutebmVz}[2][]{\ensuremath{\subp{\acute{\bm{V}}}{}{#2}{}{#1}}}
\newrobustcmd{\gravebmVz}[2][]{\ensuremath{\subp{\grave{\bm{V}}}{}{#2}{}{#1}}}
\newrobustcmd{\dotbmVz}[2][]{\ensuremath{\subp{\dot{\bm{V}}}{}{#2}{}{#1}}}
\newrobustcmd{\ddotbmVz}[2][]{\ensuremath{\subp{\ddot{\bm{V}}}{}{#2}{}{#1}}}
\newrobustcmd{\brevebmVz}[2][]{\ensuremath{\subp{\breve{\bm{V}}}{}{#2}{}{#1}}}
\newrobustcmd{\barbmVz}[2][]{\ensuremath{\subp{\bar{\bm{V}}}{}{#2}{}{#1}}}
\newrobustcmd{\vecbmVz}[2][]{\ensuremath{\subp{\vec{\bm{V}}}{}{#2}{}{#1}}}
\newrobustcmd{\Wz}[2][]{\ensuremath{\subp{W}{}{#2}{}{#1}}}
\newrobustcmd{\hatWz}[2][]{\ensuremath{\subp{\hat{W}}{}{#2}{}{#1}}}
\newrobustcmd{\widehatWz}[2][]{\ensuremath{\subp{\widehat{W}}{}{#2}{}{#1}}}
\newrobustcmd{\checkWz}[2][]{\ensuremath{\subp{\check{W}}{}{#2}{}{#1}}}
\newrobustcmd{\tildeWz}[2][]{\ensuremath{\subp{\tilde{W}}{}{#2}{}{#1}}}
\newrobustcmd{\widetildeWz}[2][]{\ensuremath{\subp{\widetilde{W}}{}{#2}{}{#1}}}
\newrobustcmd{\acuteWz}[2][]{\ensuremath{\subp{\acute{W}}{}{#2}{}{#1}}}
\newrobustcmd{\graveWz}[2][]{\ensuremath{\subp{\grave{W}}{}{#2}{}{#1}}}
\newrobustcmd{\dotWz}[2][]{\ensuremath{\subp{\dot{W}}{}{#2}{}{#1}}}
\newrobustcmd{\ddotWz}[2][]{\ensuremath{\subp{\ddot{W}}{}{#2}{}{#1}}}
\newrobustcmd{\breveWz}[2][]{\ensuremath{\subp{\breve{W}}{}{#2}{}{#1}}}
\newrobustcmd{\barWz}[2][]{\ensuremath{\subp{\bar{W}}{}{#2}{}{#1}}}
\newrobustcmd{\vecWz}[2][]{\ensuremath{\subp{\vec{W}}{}{#2}{}{#1}}}
\newrobustcmd{\bmWz}[2][]{\ensuremath{\subp{\bm{W}}{}{#2}{}{#1}}}
\newrobustcmd{\hatbmWz}[2][]{\ensuremath{\subp{\hat{\bm{W}}}{}{#2}{}{#1}}}
\newrobustcmd{\widehatbmWz}[2][]{\ensuremath{\subp{\widehat{\bm{W}}}{}{#2}{}{#1}}}
\newrobustcmd{\checkbmWz}[2][]{\ensuremath{\subp{\check{\bm{W}}}{}{#2}{}{#1}}}
\newrobustcmd{\tildebmWz}[2][]{\ensuremath{\subp{\tilde{\bm{W}}}{}{#2}{}{#1}}}
\newrobustcmd{\widetildebmWz}[2][]{\ensuremath{\subp{\widetilde{\bm{W}}}{}{#2}{}{#1}}}
\newrobustcmd{\acutebmWz}[2][]{\ensuremath{\subp{\acute{\bm{W}}}{}{#2}{}{#1}}}
\newrobustcmd{\gravebmWz}[2][]{\ensuremath{\subp{\grave{\bm{W}}}{}{#2}{}{#1}}}
\newrobustcmd{\dotbmWz}[2][]{\ensuremath{\subp{\dot{\bm{W}}}{}{#2}{}{#1}}}
\newrobustcmd{\ddotbmWz}[2][]{\ensuremath{\subp{\ddot{\bm{W}}}{}{#2}{}{#1}}}
\newrobustcmd{\brevebmWz}[2][]{\ensuremath{\subp{\breve{\bm{W}}}{}{#2}{}{#1}}}
\newrobustcmd{\barbmWz}[2][]{\ensuremath{\subp{\bar{\bm{W}}}{}{#2}{}{#1}}}
\newrobustcmd{\vecbmWz}[2][]{\ensuremath{\subp{\vec{\bm{W}}}{}{#2}{}{#1}}}
\newrobustcmd{\Xz}[2][]{\ensuremath{\subp{X}{}{#2}{}{#1}}}
\newrobustcmd{\hatXz}[2][]{\ensuremath{\subp{\hat{X}}{}{#2}{}{#1}}}
\newrobustcmd{\widehatXz}[2][]{\ensuremath{\subp{\widehat{X}}{}{#2}{}{#1}}}
\newrobustcmd{\checkXz}[2][]{\ensuremath{\subp{\check{X}}{}{#2}{}{#1}}}
\newrobustcmd{\tildeXz}[2][]{\ensuremath{\subp{\tilde{X}}{}{#2}{}{#1}}}
\newrobustcmd{\widetildeXz}[2][]{\ensuremath{\subp{\widetilde{X}}{}{#2}{}{#1}}}
\newrobustcmd{\acuteXz}[2][]{\ensuremath{\subp{\acute{X}}{}{#2}{}{#1}}}
\newrobustcmd{\graveXz}[2][]{\ensuremath{\subp{\grave{X}}{}{#2}{}{#1}}}
\newrobustcmd{\dotXz}[2][]{\ensuremath{\subp{\dot{X}}{}{#2}{}{#1}}}
\newrobustcmd{\ddotXz}[2][]{\ensuremath{\subp{\ddot{X}}{}{#2}{}{#1}}}
\newrobustcmd{\breveXz}[2][]{\ensuremath{\subp{\breve{X}}{}{#2}{}{#1}}}
\newrobustcmd{\barXz}[2][]{\ensuremath{\subp{\bar{X}}{}{#2}{}{#1}}}
\newrobustcmd{\vecXz}[2][]{\ensuremath{\subp{\vec{X}}{}{#2}{}{#1}}}
\newrobustcmd{\bmXz}[2][]{\ensuremath{\subp{\bm{X}}{}{#2}{}{#1}}}
\newrobustcmd{\hatbmXz}[2][]{\ensuremath{\subp{\hat{\bm{X}}}{}{#2}{}{#1}}}
\newrobustcmd{\widehatbmXz}[2][]{\ensuremath{\subp{\widehat{\bm{X}}}{}{#2}{}{#1}}}
\newrobustcmd{\checkbmXz}[2][]{\ensuremath{\subp{\check{\bm{X}}}{}{#2}{}{#1}}}
\newrobustcmd{\tildebmXz}[2][]{\ensuremath{\subp{\tilde{\bm{X}}}{}{#2}{}{#1}}}
\newrobustcmd{\widetildebmXz}[2][]{\ensuremath{\subp{\widetilde{\bm{X}}}{}{#2}{}{#1}}}
\newrobustcmd{\acutebmXz}[2][]{\ensuremath{\subp{\acute{\bm{X}}}{}{#2}{}{#1}}}
\newrobustcmd{\gravebmXz}[2][]{\ensuremath{\subp{\grave{\bm{X}}}{}{#2}{}{#1}}}
\newrobustcmd{\dotbmXz}[2][]{\ensuremath{\subp{\dot{\bm{X}}}{}{#2}{}{#1}}}
\newrobustcmd{\ddotbmXz}[2][]{\ensuremath{\subp{\ddot{\bm{X}}}{}{#2}{}{#1}}}
\newrobustcmd{\brevebmXz}[2][]{\ensuremath{\subp{\breve{\bm{X}}}{}{#2}{}{#1}}}
\newrobustcmd{\barbmXz}[2][]{\ensuremath{\subp{\bar{\bm{X}}}{}{#2}{}{#1}}}
\newrobustcmd{\vecbmXz}[2][]{\ensuremath{\subp{\vec{\bm{X}}}{}{#2}{}{#1}}}
\newrobustcmd{\Yz}[2][]{\ensuremath{\subp{Y}{}{#2}{}{#1}}}
\newrobustcmd{\hatYz}[2][]{\ensuremath{\subp{\hat{Y}}{}{#2}{}{#1}}}
\newrobustcmd{\widehatYz}[2][]{\ensuremath{\subp{\widehat{Y}}{}{#2}{}{#1}}}
\newrobustcmd{\checkYz}[2][]{\ensuremath{\subp{\check{Y}}{}{#2}{}{#1}}}
\newrobustcmd{\tildeYz}[2][]{\ensuremath{\subp{\tilde{Y}}{}{#2}{}{#1}}}
\newrobustcmd{\widetildeYz}[2][]{\ensuremath{\subp{\widetilde{Y}}{}{#2}{}{#1}}}
\newrobustcmd{\acuteYz}[2][]{\ensuremath{\subp{\acute{Y}}{}{#2}{}{#1}}}
\newrobustcmd{\graveYz}[2][]{\ensuremath{\subp{\grave{Y}}{}{#2}{}{#1}}}
\newrobustcmd{\dotYz}[2][]{\ensuremath{\subp{\dot{Y}}{}{#2}{}{#1}}}
\newrobustcmd{\ddotYz}[2][]{\ensuremath{\subp{\ddot{Y}}{}{#2}{}{#1}}}
\newrobustcmd{\breveYz}[2][]{\ensuremath{\subp{\breve{Y}}{}{#2}{}{#1}}}
\newrobustcmd{\barYz}[2][]{\ensuremath{\subp{\bar{Y}}{}{#2}{}{#1}}}
\newrobustcmd{\vecYz}[2][]{\ensuremath{\subp{\vec{Y}}{}{#2}{}{#1}}}
\newrobustcmd{\bmYz}[2][]{\ensuremath{\subp{\bm{Y}}{}{#2}{}{#1}}}
\newrobustcmd{\hatbmYz}[2][]{\ensuremath{\subp{\hat{\bm{Y}}}{}{#2}{}{#1}}}
\newrobustcmd{\widehatbmYz}[2][]{\ensuremath{\subp{\widehat{\bm{Y}}}{}{#2}{}{#1}}}
\newrobustcmd{\checkbmYz}[2][]{\ensuremath{\subp{\check{\bm{Y}}}{}{#2}{}{#1}}}
\newrobustcmd{\tildebmYz}[2][]{\ensuremath{\subp{\tilde{\bm{Y}}}{}{#2}{}{#1}}}
\newrobustcmd{\widetildebmYz}[2][]{\ensuremath{\subp{\widetilde{\bm{Y}}}{}{#2}{}{#1}}}
\newrobustcmd{\acutebmYz}[2][]{\ensuremath{\subp{\acute{\bm{Y}}}{}{#2}{}{#1}}}
\newrobustcmd{\gravebmYz}[2][]{\ensuremath{\subp{\grave{\bm{Y}}}{}{#2}{}{#1}}}
\newrobustcmd{\dotbmYz}[2][]{\ensuremath{\subp{\dot{\bm{Y}}}{}{#2}{}{#1}}}
\newrobustcmd{\ddotbmYz}[2][]{\ensuremath{\subp{\ddot{\bm{Y}}}{}{#2}{}{#1}}}
\newrobustcmd{\brevebmYz}[2][]{\ensuremath{\subp{\breve{\bm{Y}}}{}{#2}{}{#1}}}
\newrobustcmd{\barbmYz}[2][]{\ensuremath{\subp{\bar{\bm{Y}}}{}{#2}{}{#1}}}
\newrobustcmd{\vecbmYz}[2][]{\ensuremath{\subp{\vec{\bm{Y}}}{}{#2}{}{#1}}}
\newrobustcmd{\Zz}[2][]{\ensuremath{\subp{Z}{}{#2}{}{#1}}}
\newrobustcmd{\hatZz}[2][]{\ensuremath{\subp{\hat{Z}}{}{#2}{}{#1}}}
\newrobustcmd{\widehatZz}[2][]{\ensuremath{\subp{\widehat{Z}}{}{#2}{}{#1}}}
\newrobustcmd{\checkZz}[2][]{\ensuremath{\subp{\check{Z}}{}{#2}{}{#1}}}
\newrobustcmd{\tildeZz}[2][]{\ensuremath{\subp{\tilde{Z}}{}{#2}{}{#1}}}
\newrobustcmd{\widetildeZz}[2][]{\ensuremath{\subp{\widetilde{Z}}{}{#2}{}{#1}}}
\newrobustcmd{\acuteZz}[2][]{\ensuremath{\subp{\acute{Z}}{}{#2}{}{#1}}}
\newrobustcmd{\graveZz}[2][]{\ensuremath{\subp{\grave{Z}}{}{#2}{}{#1}}}
\newrobustcmd{\dotZz}[2][]{\ensuremath{\subp{\dot{Z}}{}{#2}{}{#1}}}
\newrobustcmd{\ddotZz}[2][]{\ensuremath{\subp{\ddot{Z}}{}{#2}{}{#1}}}
\newrobustcmd{\breveZz}[2][]{\ensuremath{\subp{\breve{Z}}{}{#2}{}{#1}}}
\newrobustcmd{\barZz}[2][]{\ensuremath{\subp{\bar{Z}}{}{#2}{}{#1}}}
\newrobustcmd{\vecZz}[2][]{\ensuremath{\subp{\vec{Z}}{}{#2}{}{#1}}}
\newrobustcmd{\bmZz}[2][]{\ensuremath{\subp{\bm{Z}}{}{#2}{}{#1}}}
\newrobustcmd{\hatbmZz}[2][]{\ensuremath{\subp{\hat{\bm{Z}}}{}{#2}{}{#1}}}
\newrobustcmd{\widehatbmZz}[2][]{\ensuremath{\subp{\widehat{\bm{Z}}}{}{#2}{}{#1}}}
\newrobustcmd{\checkbmZz}[2][]{\ensuremath{\subp{\check{\bm{Z}}}{}{#2}{}{#1}}}
\newrobustcmd{\tildebmZz}[2][]{\ensuremath{\subp{\tilde{\bm{Z}}}{}{#2}{}{#1}}}
\newrobustcmd{\widetildebmZz}[2][]{\ensuremath{\subp{\widetilde{\bm{Z}}}{}{#2}{}{#1}}}
\newrobustcmd{\acutebmZz}[2][]{\ensuremath{\subp{\acute{\bm{Z}}}{}{#2}{}{#1}}}
\newrobustcmd{\gravebmZz}[2][]{\ensuremath{\subp{\grave{\bm{Z}}}{}{#2}{}{#1}}}
\newrobustcmd{\dotbmZz}[2][]{\ensuremath{\subp{\dot{\bm{Z}}}{}{#2}{}{#1}}}
\newrobustcmd{\ddotbmZz}[2][]{\ensuremath{\subp{\ddot{\bm{Z}}}{}{#2}{}{#1}}}
\newrobustcmd{\brevebmZz}[2][]{\ensuremath{\subp{\breve{\bm{Z}}}{}{#2}{}{#1}}}
\newrobustcmd{\barbmZz}[2][]{\ensuremath{\subp{\bar{\bm{Z}}}{}{#2}{}{#1}}}
\newrobustcmd{\vecbmZz}[2][]{\ensuremath{\subp{\vec{\bm{Z}}}{}{#2}{}{#1}}}

\newrobustcmd{\alphaz}[2][]{\ensuremath{\subp{\alpha}{}{#2}{}{#1}}}
\newrobustcmd{\hatalphaz}[2][]{\ensuremath{\subp{\hat{\alpha}}{}{#2}{}{#1}}}
\newrobustcmd{\widehatalphaz}[2][]{\ensuremath{\subp{\widehat{\alpha}}{}{#2}{}{#1}}}
\newrobustcmd{\checkalphaz}[2][]{\ensuremath{\subp{\check{\alpha}}{}{#2}{}{#1}}}
\newrobustcmd{\tildealphaz}[2][]{\ensuremath{\subp{\tilde{\alpha}}{}{#2}{}{#1}}}
\newrobustcmd{\widetildealphaz}[2][]{\ensuremath{\subp{\widetilde{\alpha}}{}{#2}{}{#1}}}
\newrobustcmd{\acutealphaz}[2][]{\ensuremath{\subp{\acute{\alpha}}{}{#2}{}{#1}}}
\newrobustcmd{\gravealphaz}[2][]{\ensuremath{\subp{\grave{\alpha}}{}{#2}{}{#1}}}
\newrobustcmd{\dotalphaz}[2][]{\ensuremath{\subp{\dot{\alpha}}{}{#2}{}{#1}}}
\newrobustcmd{\ddotalphaz}[2][]{\ensuremath{\subp{\ddot{\alpha}}{}{#2}{}{#1}}}
\newrobustcmd{\brevealphaz}[2][]{\ensuremath{\subp{\breve{\alpha}}{}{#2}{}{#1}}}
\newrobustcmd{\baralphaz}[2][]{\ensuremath{\subp{\bar{\alpha}}{}{#2}{}{#1}}}
\newrobustcmd{\vecalphaz}[2][]{\ensuremath{\subp{\vec{\alpha}}{}{#2}{}{#1}}}
\newrobustcmd{\bmalphaz}[2][]{\ensuremath{\subp{\bm{\alpha}}{}{#2}{}{#1}}}
\newrobustcmd{\hatbmalphaz}[2][]{\ensuremath{\subp{\hat{\bm{\alpha}}}{}{#2}{}{#1}}}
\newrobustcmd{\widehatbmalphaz}[2][]{\ensuremath{\subp{\widehat{\bm{\alpha}}}{}{#2}{}{#1}}}
\newrobustcmd{\checkbmalphaz}[2][]{\ensuremath{\subp{\check{\bm{\alpha}}}{}{#2}{}{#1}}}
\newrobustcmd{\tildebmalphaz}[2][]{\ensuremath{\subp{\tilde{\bm{\alpha}}}{}{#2}{}{#1}}}
\newrobustcmd{\widetildebmalphaz}[2][]{\ensuremath{\subp{\widetilde{\bm{\alpha}}}{}{#2}{}{#1}}}
\newrobustcmd{\acutebmalphaz}[2][]{\ensuremath{\subp{\acute{\bm{\alpha}}}{}{#2}{}{#1}}}
\newrobustcmd{\gravebmalphaz}[2][]{\ensuremath{\subp{\grave{\bm{\alpha}}}{}{#2}{}{#1}}}
\newrobustcmd{\dotbmalphaz}[2][]{\ensuremath{\subp{\dot{\bm{\alpha}}}{}{#2}{}{#1}}}
\newrobustcmd{\ddotbmalphaz}[2][]{\ensuremath{\subp{\ddot{\bm{\alpha}}}{}{#2}{}{#1}}}
\newrobustcmd{\brevebmalphaz}[2][]{\ensuremath{\subp{\breve{\bm{\alpha}}}{}{#2}{}{#1}}}
\newrobustcmd{\barbmalphaz}[2][]{\ensuremath{\subp{\bar{\bm{\alpha}}}{}{#2}{}{#1}}}
\newrobustcmd{\vecbmalphaz}[2][]{\ensuremath{\subp{\vec{\bm{\alpha}}}{}{#2}{}{#1}}}
\newrobustcmd{\betaz}[2][]{\ensuremath{\subp{\beta}{}{#2}{}{#1}}}
\newrobustcmd{\hatbetaz}[2][]{\ensuremath{\subp{\hat{\beta}}{}{#2}{}{#1}}}
\newrobustcmd{\widehatbetaz}[2][]{\ensuremath{\subp{\widehat{\beta}}{}{#2}{}{#1}}}
\newrobustcmd{\checkbetaz}[2][]{\ensuremath{\subp{\check{\beta}}{}{#2}{}{#1}}}
\newrobustcmd{\tildebetaz}[2][]{\ensuremath{\subp{\tilde{\beta}}{}{#2}{}{#1}}}
\newrobustcmd{\widetildebetaz}[2][]{\ensuremath{\subp{\widetilde{\beta}}{}{#2}{}{#1}}}
\newrobustcmd{\acutebetaz}[2][]{\ensuremath{\subp{\acute{\beta}}{}{#2}{}{#1}}}
\newrobustcmd{\gravebetaz}[2][]{\ensuremath{\subp{\grave{\beta}}{}{#2}{}{#1}}}
\newrobustcmd{\dotbetaz}[2][]{\ensuremath{\subp{\dot{\beta}}{}{#2}{}{#1}}}
\newrobustcmd{\ddotbetaz}[2][]{\ensuremath{\subp{\ddot{\beta}}{}{#2}{}{#1}}}
\newrobustcmd{\brevebetaz}[2][]{\ensuremath{\subp{\breve{\beta}}{}{#2}{}{#1}}}
\newrobustcmd{\barbetaz}[2][]{\ensuremath{\subp{\bar{\beta}}{}{#2}{}{#1}}}
\newrobustcmd{\vecbetaz}[2][]{\ensuremath{\subp{\vec{\beta}}{}{#2}{}{#1}}}
\newrobustcmd{\bmbetaz}[2][]{\ensuremath{\subp{\bm{\beta}}{}{#2}{}{#1}}}
\newrobustcmd{\hatbmbetaz}[2][]{\ensuremath{\subp{\hat{\bm{\beta}}}{}{#2}{}{#1}}}
\newrobustcmd{\widehatbmbetaz}[2][]{\ensuremath{\subp{\widehat{\bm{\beta}}}{}{#2}{}{#1}}}
\newrobustcmd{\checkbmbetaz}[2][]{\ensuremath{\subp{\check{\bm{\beta}}}{}{#2}{}{#1}}}
\newrobustcmd{\tildebmbetaz}[2][]{\ensuremath{\subp{\tilde{\bm{\beta}}}{}{#2}{}{#1}}}
\newrobustcmd{\widetildebmbetaz}[2][]{\ensuremath{\subp{\widetilde{\bm{\beta}}}{}{#2}{}{#1}}}
\newrobustcmd{\acutebmbetaz}[2][]{\ensuremath{\subp{\acute{\bm{\beta}}}{}{#2}{}{#1}}}
\newrobustcmd{\gravebmbetaz}[2][]{\ensuremath{\subp{\grave{\bm{\beta}}}{}{#2}{}{#1}}}
\newrobustcmd{\dotbmbetaz}[2][]{\ensuremath{\subp{\dot{\bm{\beta}}}{}{#2}{}{#1}}}
\newrobustcmd{\ddotbmbetaz}[2][]{\ensuremath{\subp{\ddot{\bm{\beta}}}{}{#2}{}{#1}}}
\newrobustcmd{\brevebmbetaz}[2][]{\ensuremath{\subp{\breve{\bm{\beta}}}{}{#2}{}{#1}}}
\newrobustcmd{\barbmbetaz}[2][]{\ensuremath{\subp{\bar{\bm{\beta}}}{}{#2}{}{#1}}}
\newrobustcmd{\vecbmbetaz}[2][]{\ensuremath{\subp{\vec{\bm{\beta}}}{}{#2}{}{#1}}}
\newrobustcmd{\gammaz}[2][]{\ensuremath{\subp{\gamma}{}{#2}{}{#1}}}
\newrobustcmd{\hatgammaz}[2][]{\ensuremath{\subp{\hat{\gamma}}{}{#2}{}{#1}}}
\newrobustcmd{\widehatgammaz}[2][]{\ensuremath{\subp{\widehat{\gamma}}{}{#2}{}{#1}}}
\newrobustcmd{\checkgammaz}[2][]{\ensuremath{\subp{\check{\gamma}}{}{#2}{}{#1}}}
\newrobustcmd{\tildegammaz}[2][]{\ensuremath{\subp{\tilde{\gamma}}{}{#2}{}{#1}}}
\newrobustcmd{\widetildegammaz}[2][]{\ensuremath{\subp{\widetilde{\gamma}}{}{#2}{}{#1}}}
\newrobustcmd{\acutegammaz}[2][]{\ensuremath{\subp{\acute{\gamma}}{}{#2}{}{#1}}}
\newrobustcmd{\gravegammaz}[2][]{\ensuremath{\subp{\grave{\gamma}}{}{#2}{}{#1}}}
\newrobustcmd{\dotgammaz}[2][]{\ensuremath{\subp{\dot{\gamma}}{}{#2}{}{#1}}}
\newrobustcmd{\ddotgammaz}[2][]{\ensuremath{\subp{\ddot{\gamma}}{}{#2}{}{#1}}}
\newrobustcmd{\brevegammaz}[2][]{\ensuremath{\subp{\breve{\gamma}}{}{#2}{}{#1}}}
\newrobustcmd{\bargammaz}[2][]{\ensuremath{\subp{\bar{\gamma}}{}{#2}{}{#1}}}
\newrobustcmd{\vecgammaz}[2][]{\ensuremath{\subp{\vec{\gamma}}{}{#2}{}{#1}}}
\newrobustcmd{\bmgammaz}[2][]{\ensuremath{\subp{\bm{\gamma}}{}{#2}{}{#1}}}
\newrobustcmd{\hatbmgammaz}[2][]{\ensuremath{\subp{\hat{\bm{\gamma}}}{}{#2}{}{#1}}}
\newrobustcmd{\widehatbmgammaz}[2][]{\ensuremath{\subp{\widehat{\bm{\gamma}}}{}{#2}{}{#1}}}
\newrobustcmd{\checkbmgammaz}[2][]{\ensuremath{\subp{\check{\bm{\gamma}}}{}{#2}{}{#1}}}
\newrobustcmd{\tildebmgammaz}[2][]{\ensuremath{\subp{\tilde{\bm{\gamma}}}{}{#2}{}{#1}}}
\newrobustcmd{\widetildebmgammaz}[2][]{\ensuremath{\subp{\widetilde{\bm{\gamma}}}{}{#2}{}{#1}}}
\newrobustcmd{\acutebmgammaz}[2][]{\ensuremath{\subp{\acute{\bm{\gamma}}}{}{#2}{}{#1}}}
\newrobustcmd{\gravebmgammaz}[2][]{\ensuremath{\subp{\grave{\bm{\gamma}}}{}{#2}{}{#1}}}
\newrobustcmd{\dotbmgammaz}[2][]{\ensuremath{\subp{\dot{\bm{\gamma}}}{}{#2}{}{#1}}}
\newrobustcmd{\ddotbmgammaz}[2][]{\ensuremath{\subp{\ddot{\bm{\gamma}}}{}{#2}{}{#1}}}
\newrobustcmd{\brevebmgammaz}[2][]{\ensuremath{\subp{\breve{\bm{\gamma}}}{}{#2}{}{#1}}}
\newrobustcmd{\barbmgammaz}[2][]{\ensuremath{\subp{\bar{\bm{\gamma}}}{}{#2}{}{#1}}}
\newrobustcmd{\vecbmgammaz}[2][]{\ensuremath{\subp{\vec{\bm{\gamma}}}{}{#2}{}{#1}}}
\newrobustcmd{\deltaz}[2][]{\ensuremath{\subp{\delta}{}{#2}{}{#1}}}
\newrobustcmd{\hatdeltaz}[2][]{\ensuremath{\subp{\hat{\delta}}{}{#2}{}{#1}}}
\newrobustcmd{\widehatdeltaz}[2][]{\ensuremath{\subp{\widehat{\delta}}{}{#2}{}{#1}}}
\newrobustcmd{\checkdeltaz}[2][]{\ensuremath{\subp{\check{\delta}}{}{#2}{}{#1}}}
\newrobustcmd{\tildedeltaz}[2][]{\ensuremath{\subp{\tilde{\delta}}{}{#2}{}{#1}}}
\newrobustcmd{\widetildedeltaz}[2][]{\ensuremath{\subp{\widetilde{\delta}}{}{#2}{}{#1}}}
\newrobustcmd{\acutedeltaz}[2][]{\ensuremath{\subp{\acute{\delta}}{}{#2}{}{#1}}}
\newrobustcmd{\gravedeltaz}[2][]{\ensuremath{\subp{\grave{\delta}}{}{#2}{}{#1}}}
\newrobustcmd{\dotdeltaz}[2][]{\ensuremath{\subp{\dot{\delta}}{}{#2}{}{#1}}}
\newrobustcmd{\ddotdeltaz}[2][]{\ensuremath{\subp{\ddot{\delta}}{}{#2}{}{#1}}}
\newrobustcmd{\brevedeltaz}[2][]{\ensuremath{\subp{\breve{\delta}}{}{#2}{}{#1}}}
\newrobustcmd{\bardeltaz}[2][]{\ensuremath{\subp{\bar{\delta}}{}{#2}{}{#1}}}
\newrobustcmd{\vecdeltaz}[2][]{\ensuremath{\subp{\vec{\delta}}{}{#2}{}{#1}}}
\newrobustcmd{\bmdeltaz}[2][]{\ensuremath{\subp{\bm{\delta}}{}{#2}{}{#1}}}
\newrobustcmd{\hatbmdeltaz}[2][]{\ensuremath{\subp{\hat{\bm{\delta}}}{}{#2}{}{#1}}}
\newrobustcmd{\widehatbmdeltaz}[2][]{\ensuremath{\subp{\widehat{\bm{\delta}}}{}{#2}{}{#1}}}
\newrobustcmd{\checkbmdeltaz}[2][]{\ensuremath{\subp{\check{\bm{\delta}}}{}{#2}{}{#1}}}
\newrobustcmd{\tildebmdeltaz}[2][]{\ensuremath{\subp{\tilde{\bm{\delta}}}{}{#2}{}{#1}}}
\newrobustcmd{\widetildebmdeltaz}[2][]{\ensuremath{\subp{\widetilde{\bm{\delta}}}{}{#2}{}{#1}}}
\newrobustcmd{\acutebmdeltaz}[2][]{\ensuremath{\subp{\acute{\bm{\delta}}}{}{#2}{}{#1}}}
\newrobustcmd{\gravebmdeltaz}[2][]{\ensuremath{\subp{\grave{\bm{\delta}}}{}{#2}{}{#1}}}
\newrobustcmd{\dotbmdeltaz}[2][]{\ensuremath{\subp{\dot{\bm{\delta}}}{}{#2}{}{#1}}}
\newrobustcmd{\ddotbmdeltaz}[2][]{\ensuremath{\subp{\ddot{\bm{\delta}}}{}{#2}{}{#1}}}
\newrobustcmd{\brevebmdeltaz}[2][]{\ensuremath{\subp{\breve{\bm{\delta}}}{}{#2}{}{#1}}}
\newrobustcmd{\barbmdeltaz}[2][]{\ensuremath{\subp{\bar{\bm{\delta}}}{}{#2}{}{#1}}}
\newrobustcmd{\vecbmdeltaz}[2][]{\ensuremath{\subp{\vec{\bm{\delta}}}{}{#2}{}{#1}}}
\newrobustcmd{\epsilonz}[2][]{\ensuremath{\subp{\epsilon}{}{#2}{}{#1}}}
\newrobustcmd{\hatepsilonz}[2][]{\ensuremath{\subp{\hat{\epsilon}}{}{#2}{}{#1}}}
\newrobustcmd{\widehatepsilonz}[2][]{\ensuremath{\subp{\widehat{\epsilon}}{}{#2}{}{#1}}}
\newrobustcmd{\checkepsilonz}[2][]{\ensuremath{\subp{\check{\epsilon}}{}{#2}{}{#1}}}
\newrobustcmd{\tildeepsilonz}[2][]{\ensuremath{\subp{\tilde{\epsilon}}{}{#2}{}{#1}}}
\newrobustcmd{\widetildeepsilonz}[2][]{\ensuremath{\subp{\widetilde{\epsilon}}{}{#2}{}{#1}}}
\newrobustcmd{\acuteepsilonz}[2][]{\ensuremath{\subp{\acute{\epsilon}}{}{#2}{}{#1}}}
\newrobustcmd{\graveepsilonz}[2][]{\ensuremath{\subp{\grave{\epsilon}}{}{#2}{}{#1}}}
\newrobustcmd{\dotepsilonz}[2][]{\ensuremath{\subp{\dot{\epsilon}}{}{#2}{}{#1}}}
\newrobustcmd{\ddotepsilonz}[2][]{\ensuremath{\subp{\ddot{\epsilon}}{}{#2}{}{#1}}}
\newrobustcmd{\breveepsilonz}[2][]{\ensuremath{\subp{\breve{\epsilon}}{}{#2}{}{#1}}}
\newrobustcmd{\barepsilonz}[2][]{\ensuremath{\subp{\bar{\epsilon}}{}{#2}{}{#1}}}
\newrobustcmd{\vecepsilonz}[2][]{\ensuremath{\subp{\vec{\epsilon}}{}{#2}{}{#1}}}
\newrobustcmd{\bmepsilonz}[2][]{\ensuremath{\subp{\bm{\epsilon}}{}{#2}{}{#1}}}
\newrobustcmd{\hatbmepsilonz}[2][]{\ensuremath{\subp{\hat{\bm{\epsilon}}}{}{#2}{}{#1}}}
\newrobustcmd{\widehatbmepsilonz}[2][]{\ensuremath{\subp{\widehat{\bm{\epsilon}}}{}{#2}{}{#1}}}
\newrobustcmd{\checkbmepsilonz}[2][]{\ensuremath{\subp{\check{\bm{\epsilon}}}{}{#2}{}{#1}}}
\newrobustcmd{\tildebmepsilonz}[2][]{\ensuremath{\subp{\tilde{\bm{\epsilon}}}{}{#2}{}{#1}}}
\newrobustcmd{\widetildebmepsilonz}[2][]{\ensuremath{\subp{\widetilde{\bm{\epsilon}}}{}{#2}{}{#1}}}
\newrobustcmd{\acutebmepsilonz}[2][]{\ensuremath{\subp{\acute{\bm{\epsilon}}}{}{#2}{}{#1}}}
\newrobustcmd{\gravebmepsilonz}[2][]{\ensuremath{\subp{\grave{\bm{\epsilon}}}{}{#2}{}{#1}}}
\newrobustcmd{\dotbmepsilonz}[2][]{\ensuremath{\subp{\dot{\bm{\epsilon}}}{}{#2}{}{#1}}}
\newrobustcmd{\ddotbmepsilonz}[2][]{\ensuremath{\subp{\ddot{\bm{\epsilon}}}{}{#2}{}{#1}}}
\newrobustcmd{\brevebmepsilonz}[2][]{\ensuremath{\subp{\breve{\bm{\epsilon}}}{}{#2}{}{#1}}}
\newrobustcmd{\barbmepsilonz}[2][]{\ensuremath{\subp{\bar{\bm{\epsilon}}}{}{#2}{}{#1}}}
\newrobustcmd{\vecbmepsilonz}[2][]{\ensuremath{\subp{\vec{\bm{\epsilon}}}{}{#2}{}{#1}}}
\newrobustcmd{\varepsilonz}[2][]{\ensuremath{\subp{\varepsilon}{}{#2}{}{#1}}}
\newrobustcmd{\hatvarepsilonz}[2][]{\ensuremath{\subp{\hat{\varepsilon}}{}{#2}{}{#1}}}
\newrobustcmd{\widehatvarepsilonz}[2][]{\ensuremath{\subp{\widehat{\varepsilon}}{}{#2}{}{#1}}}
\newrobustcmd{\checkvarepsilonz}[2][]{\ensuremath{\subp{\check{\varepsilon}}{}{#2}{}{#1}}}
\newrobustcmd{\tildevarepsilonz}[2][]{\ensuremath{\subp{\tilde{\varepsilon}}{}{#2}{}{#1}}}
\newrobustcmd{\widetildevarepsilonz}[2][]{\ensuremath{\subp{\widetilde{\varepsilon}}{}{#2}{}{#1}}}
\newrobustcmd{\acutevarepsilonz}[2][]{\ensuremath{\subp{\acute{\varepsilon}}{}{#2}{}{#1}}}
\newrobustcmd{\gravevarepsilonz}[2][]{\ensuremath{\subp{\grave{\varepsilon}}{}{#2}{}{#1}}}
\newrobustcmd{\dotvarepsilonz}[2][]{\ensuremath{\subp{\dot{\varepsilon}}{}{#2}{}{#1}}}
\newrobustcmd{\ddotvarepsilonz}[2][]{\ensuremath{\subp{\ddot{\varepsilon}}{}{#2}{}{#1}}}
\newrobustcmd{\brevevarepsilonz}[2][]{\ensuremath{\subp{\breve{\varepsilon}}{}{#2}{}{#1}}}
\newrobustcmd{\barvarepsilonz}[2][]{\ensuremath{\subp{\bar{\varepsilon}}{}{#2}{}{#1}}}
\newrobustcmd{\vecvarepsilonz}[2][]{\ensuremath{\subp{\vec{\varepsilon}}{}{#2}{}{#1}}}
\newrobustcmd{\bmvarepsilonz}[2][]{\ensuremath{\subp{\bm{\varepsilon}}{}{#2}{}{#1}}}
\newrobustcmd{\hatbmvarepsilonz}[2][]{\ensuremath{\subp{\hat{\bm{\varepsilon}}}{}{#2}{}{#1}}}
\newrobustcmd{\widehatbmvarepsilonz}[2][]{\ensuremath{\subp{\widehat{\bm{\varepsilon}}}{}{#2}{}{#1}}}
\newrobustcmd{\checkbmvarepsilonz}[2][]{\ensuremath{\subp{\check{\bm{\varepsilon}}}{}{#2}{}{#1}}}
\newrobustcmd{\tildebmvarepsilonz}[2][]{\ensuremath{\subp{\tilde{\bm{\varepsilon}}}{}{#2}{}{#1}}}
\newrobustcmd{\widetildebmvarepsilonz}[2][]{\ensuremath{\subp{\widetilde{\bm{\varepsilon}}}{}{#2}{}{#1}}}
\newrobustcmd{\acutebmvarepsilonz}[2][]{\ensuremath{\subp{\acute{\bm{\varepsilon}}}{}{#2}{}{#1}}}
\newrobustcmd{\gravebmvarepsilonz}[2][]{\ensuremath{\subp{\grave{\bm{\varepsilon}}}{}{#2}{}{#1}}}
\newrobustcmd{\dotbmvarepsilonz}[2][]{\ensuremath{\subp{\dot{\bm{\varepsilon}}}{}{#2}{}{#1}}}
\newrobustcmd{\ddotbmvarepsilonz}[2][]{\ensuremath{\subp{\ddot{\bm{\varepsilon}}}{}{#2}{}{#1}}}
\newrobustcmd{\brevebmvarepsilonz}[2][]{\ensuremath{\subp{\breve{\bm{\varepsilon}}}{}{#2}{}{#1}}}
\newrobustcmd{\barbmvarepsilonz}[2][]{\ensuremath{\subp{\bar{\bm{\varepsilon}}}{}{#2}{}{#1}}}
\newrobustcmd{\vecbmvarepsilonz}[2][]{\ensuremath{\subp{\vec{\bm{\varepsilon}}}{}{#2}{}{#1}}}
\newrobustcmd{\zetaz}[2][]{\ensuremath{\subp{\zeta}{}{#2}{}{#1}}}
\newrobustcmd{\hatzetaz}[2][]{\ensuremath{\subp{\hat{\zeta}}{}{#2}{}{#1}}}
\newrobustcmd{\widehatzetaz}[2][]{\ensuremath{\subp{\widehat{\zeta}}{}{#2}{}{#1}}}
\newrobustcmd{\checkzetaz}[2][]{\ensuremath{\subp{\check{\zeta}}{}{#2}{}{#1}}}
\newrobustcmd{\tildezetaz}[2][]{\ensuremath{\subp{\tilde{\zeta}}{}{#2}{}{#1}}}
\newrobustcmd{\widetildezetaz}[2][]{\ensuremath{\subp{\widetilde{\zeta}}{}{#2}{}{#1}}}
\newrobustcmd{\acutezetaz}[2][]{\ensuremath{\subp{\acute{\zeta}}{}{#2}{}{#1}}}
\newrobustcmd{\gravezetaz}[2][]{\ensuremath{\subp{\grave{\zeta}}{}{#2}{}{#1}}}
\newrobustcmd{\dotzetaz}[2][]{\ensuremath{\subp{\dot{\zeta}}{}{#2}{}{#1}}}
\newrobustcmd{\ddotzetaz}[2][]{\ensuremath{\subp{\ddot{\zeta}}{}{#2}{}{#1}}}
\newrobustcmd{\brevezetaz}[2][]{\ensuremath{\subp{\breve{\zeta}}{}{#2}{}{#1}}}
\newrobustcmd{\barzetaz}[2][]{\ensuremath{\subp{\bar{\zeta}}{}{#2}{}{#1}}}
\newrobustcmd{\veczetaz}[2][]{\ensuremath{\subp{\vec{\zeta}}{}{#2}{}{#1}}}
\newrobustcmd{\bmzetaz}[2][]{\ensuremath{\subp{\bm{\zeta}}{}{#2}{}{#1}}}
\newrobustcmd{\hatbmzetaz}[2][]{\ensuremath{\subp{\hat{\bm{\zeta}}}{}{#2}{}{#1}}}
\newrobustcmd{\widehatbmzetaz}[2][]{\ensuremath{\subp{\widehat{\bm{\zeta}}}{}{#2}{}{#1}}}
\newrobustcmd{\checkbmzetaz}[2][]{\ensuremath{\subp{\check{\bm{\zeta}}}{}{#2}{}{#1}}}
\newrobustcmd{\tildebmzetaz}[2][]{\ensuremath{\subp{\tilde{\bm{\zeta}}}{}{#2}{}{#1}}}
\newrobustcmd{\widetildebmzetaz}[2][]{\ensuremath{\subp{\widetilde{\bm{\zeta}}}{}{#2}{}{#1}}}
\newrobustcmd{\acutebmzetaz}[2][]{\ensuremath{\subp{\acute{\bm{\zeta}}}{}{#2}{}{#1}}}
\newrobustcmd{\gravebmzetaz}[2][]{\ensuremath{\subp{\grave{\bm{\zeta}}}{}{#2}{}{#1}}}
\newrobustcmd{\dotbmzetaz}[2][]{\ensuremath{\subp{\dot{\bm{\zeta}}}{}{#2}{}{#1}}}
\newrobustcmd{\ddotbmzetaz}[2][]{\ensuremath{\subp{\ddot{\bm{\zeta}}}{}{#2}{}{#1}}}
\newrobustcmd{\brevebmzetaz}[2][]{\ensuremath{\subp{\breve{\bm{\zeta}}}{}{#2}{}{#1}}}
\newrobustcmd{\barbmzetaz}[2][]{\ensuremath{\subp{\bar{\bm{\zeta}}}{}{#2}{}{#1}}}
\newrobustcmd{\vecbmzetaz}[2][]{\ensuremath{\subp{\vec{\bm{\zeta}}}{}{#2}{}{#1}}}
\newrobustcmd{\etaz}[2][]{\ensuremath{\subp{\eta}{}{#2}{}{#1}}}
\newrobustcmd{\hatetaz}[2][]{\ensuremath{\subp{\hat{\eta}}{}{#2}{}{#1}}}
\newrobustcmd{\widehatetaz}[2][]{\ensuremath{\subp{\widehat{\eta}}{}{#2}{}{#1}}}
\newrobustcmd{\checketaz}[2][]{\ensuremath{\subp{\check{\eta}}{}{#2}{}{#1}}}
\newrobustcmd{\tildeetaz}[2][]{\ensuremath{\subp{\tilde{\eta}}{}{#2}{}{#1}}}
\newrobustcmd{\widetildeetaz}[2][]{\ensuremath{\subp{\widetilde{\eta}}{}{#2}{}{#1}}}
\newrobustcmd{\acuteetaz}[2][]{\ensuremath{\subp{\acute{\eta}}{}{#2}{}{#1}}}
\newrobustcmd{\graveetaz}[2][]{\ensuremath{\subp{\grave{\eta}}{}{#2}{}{#1}}}
\newrobustcmd{\dotetaz}[2][]{\ensuremath{\subp{\dot{\eta}}{}{#2}{}{#1}}}
\newrobustcmd{\ddotetaz}[2][]{\ensuremath{\subp{\ddot{\eta}}{}{#2}{}{#1}}}
\newrobustcmd{\breveetaz}[2][]{\ensuremath{\subp{\breve{\eta}}{}{#2}{}{#1}}}
\newrobustcmd{\baretaz}[2][]{\ensuremath{\subp{\bar{\eta}}{}{#2}{}{#1}}}
\newrobustcmd{\vecetaz}[2][]{\ensuremath{\subp{\vec{\eta}}{}{#2}{}{#1}}}
\newrobustcmd{\bmetaz}[2][]{\ensuremath{\subp{\bm{\eta}}{}{#2}{}{#1}}}
\newrobustcmd{\hatbmetaz}[2][]{\ensuremath{\subp{\hat{\bm{\eta}}}{}{#2}{}{#1}}}
\newrobustcmd{\widehatbmetaz}[2][]{\ensuremath{\subp{\widehat{\bm{\eta}}}{}{#2}{}{#1}}}
\newrobustcmd{\checkbmetaz}[2][]{\ensuremath{\subp{\check{\bm{\eta}}}{}{#2}{}{#1}}}
\newrobustcmd{\tildebmetaz}[2][]{\ensuremath{\subp{\tilde{\bm{\eta}}}{}{#2}{}{#1}}}
\newrobustcmd{\widetildebmetaz}[2][]{\ensuremath{\subp{\widetilde{\bm{\eta}}}{}{#2}{}{#1}}}
\newrobustcmd{\acutebmetaz}[2][]{\ensuremath{\subp{\acute{\bm{\eta}}}{}{#2}{}{#1}}}
\newrobustcmd{\gravebmetaz}[2][]{\ensuremath{\subp{\grave{\bm{\eta}}}{}{#2}{}{#1}}}
\newrobustcmd{\dotbmetaz}[2][]{\ensuremath{\subp{\dot{\bm{\eta}}}{}{#2}{}{#1}}}
\newrobustcmd{\ddotbmetaz}[2][]{\ensuremath{\subp{\ddot{\bm{\eta}}}{}{#2}{}{#1}}}
\newrobustcmd{\brevebmetaz}[2][]{\ensuremath{\subp{\breve{\bm{\eta}}}{}{#2}{}{#1}}}
\newrobustcmd{\barbmetaz}[2][]{\ensuremath{\subp{\bar{\bm{\eta}}}{}{#2}{}{#1}}}
\newrobustcmd{\vecbmetaz}[2][]{\ensuremath{\subp{\vec{\bm{\eta}}}{}{#2}{}{#1}}}
\newrobustcmd{\thetaz}[2][]{\ensuremath{\subp{\theta}{}{#2}{}{#1}}}
\newrobustcmd{\hatthetaz}[2][]{\ensuremath{\subp{\hat{\theta}}{}{#2}{}{#1}}}
\newrobustcmd{\widehatthetaz}[2][]{\ensuremath{\subp{\widehat{\theta}}{}{#2}{}{#1}}}
\newrobustcmd{\checkthetaz}[2][]{\ensuremath{\subp{\check{\theta}}{}{#2}{}{#1}}}
\newrobustcmd{\tildethetaz}[2][]{\ensuremath{\subp{\tilde{\theta}}{}{#2}{}{#1}}}
\newrobustcmd{\widetildethetaz}[2][]{\ensuremath{\subp{\widetilde{\theta}}{}{#2}{}{#1}}}
\newrobustcmd{\acutethetaz}[2][]{\ensuremath{\subp{\acute{\theta}}{}{#2}{}{#1}}}
\newrobustcmd{\gravethetaz}[2][]{\ensuremath{\subp{\grave{\theta}}{}{#2}{}{#1}}}
\newrobustcmd{\dotthetaz}[2][]{\ensuremath{\subp{\dot{\theta}}{}{#2}{}{#1}}}
\newrobustcmd{\ddotthetaz}[2][]{\ensuremath{\subp{\ddot{\theta}}{}{#2}{}{#1}}}
\newrobustcmd{\brevethetaz}[2][]{\ensuremath{\subp{\breve{\theta}}{}{#2}{}{#1}}}
\newrobustcmd{\barthetaz}[2][]{\ensuremath{\subp{\bar{\theta}}{}{#2}{}{#1}}}
\newrobustcmd{\vecthetaz}[2][]{\ensuremath{\subp{\vec{\theta}}{}{#2}{}{#1}}}
\newrobustcmd{\bmthetaz}[2][]{\ensuremath{\subp{\bm{\theta}}{}{#2}{}{#1}}}
\newrobustcmd{\hatbmthetaz}[2][]{\ensuremath{\subp{\hat{\bm{\theta}}}{}{#2}{}{#1}}}
\newrobustcmd{\widehatbmthetaz}[2][]{\ensuremath{\subp{\widehat{\bm{\theta}}}{}{#2}{}{#1}}}
\newrobustcmd{\checkbmthetaz}[2][]{\ensuremath{\subp{\check{\bm{\theta}}}{}{#2}{}{#1}}}
\newrobustcmd{\tildebmthetaz}[2][]{\ensuremath{\subp{\tilde{\bm{\theta}}}{}{#2}{}{#1}}}
\newrobustcmd{\widetildebmthetaz}[2][]{\ensuremath{\subp{\widetilde{\bm{\theta}}}{}{#2}{}{#1}}}
\newrobustcmd{\acutebmthetaz}[2][]{\ensuremath{\subp{\acute{\bm{\theta}}}{}{#2}{}{#1}}}
\newrobustcmd{\gravebmthetaz}[2][]{\ensuremath{\subp{\grave{\bm{\theta}}}{}{#2}{}{#1}}}
\newrobustcmd{\dotbmthetaz}[2][]{\ensuremath{\subp{\dot{\bm{\theta}}}{}{#2}{}{#1}}}
\newrobustcmd{\ddotbmthetaz}[2][]{\ensuremath{\subp{\ddot{\bm{\theta}}}{}{#2}{}{#1}}}
\newrobustcmd{\brevebmthetaz}[2][]{\ensuremath{\subp{\breve{\bm{\theta}}}{}{#2}{}{#1}}}
\newrobustcmd{\barbmthetaz}[2][]{\ensuremath{\subp{\bar{\bm{\theta}}}{}{#2}{}{#1}}}
\newrobustcmd{\vecbmthetaz}[2][]{\ensuremath{\subp{\vec{\bm{\theta}}}{}{#2}{}{#1}}}
\newrobustcmd{\varthetaz}[2][]{\ensuremath{\subp{\vartheta}{}{#2}{}{#1}}}
\newrobustcmd{\hatvarthetaz}[2][]{\ensuremath{\subp{\hat{\vartheta}}{}{#2}{}{#1}}}
\newrobustcmd{\widehatvarthetaz}[2][]{\ensuremath{\subp{\widehat{\vartheta}}{}{#2}{}{#1}}}
\newrobustcmd{\checkvarthetaz}[2][]{\ensuremath{\subp{\check{\vartheta}}{}{#2}{}{#1}}}
\newrobustcmd{\tildevarthetaz}[2][]{\ensuremath{\subp{\tilde{\vartheta}}{}{#2}{}{#1}}}
\newrobustcmd{\widetildevarthetaz}[2][]{\ensuremath{\subp{\widetilde{\vartheta}}{}{#2}{}{#1}}}
\newrobustcmd{\acutevarthetaz}[2][]{\ensuremath{\subp{\acute{\vartheta}}{}{#2}{}{#1}}}
\newrobustcmd{\gravevarthetaz}[2][]{\ensuremath{\subp{\grave{\vartheta}}{}{#2}{}{#1}}}
\newrobustcmd{\dotvarthetaz}[2][]{\ensuremath{\subp{\dot{\vartheta}}{}{#2}{}{#1}}}
\newrobustcmd{\ddotvarthetaz}[2][]{\ensuremath{\subp{\ddot{\vartheta}}{}{#2}{}{#1}}}
\newrobustcmd{\brevevarthetaz}[2][]{\ensuremath{\subp{\breve{\vartheta}}{}{#2}{}{#1}}}
\newrobustcmd{\barvarthetaz}[2][]{\ensuremath{\subp{\bar{\vartheta}}{}{#2}{}{#1}}}
\newrobustcmd{\vecvarthetaz}[2][]{\ensuremath{\subp{\vec{\vartheta}}{}{#2}{}{#1}}}
\newrobustcmd{\bmvarthetaz}[2][]{\ensuremath{\subp{\bm{\vartheta}}{}{#2}{}{#1}}}
\newrobustcmd{\hatbmvarthetaz}[2][]{\ensuremath{\subp{\hat{\bm{\vartheta}}}{}{#2}{}{#1}}}
\newrobustcmd{\widehatbmvarthetaz}[2][]{\ensuremath{\subp{\widehat{\bm{\vartheta}}}{}{#2}{}{#1}}}
\newrobustcmd{\checkbmvarthetaz}[2][]{\ensuremath{\subp{\check{\bm{\vartheta}}}{}{#2}{}{#1}}}
\newrobustcmd{\tildebmvarthetaz}[2][]{\ensuremath{\subp{\tilde{\bm{\vartheta}}}{}{#2}{}{#1}}}
\newrobustcmd{\widetildebmvarthetaz}[2][]{\ensuremath{\subp{\widetilde{\bm{\vartheta}}}{}{#2}{}{#1}}}
\newrobustcmd{\acutebmvarthetaz}[2][]{\ensuremath{\subp{\acute{\bm{\vartheta}}}{}{#2}{}{#1}}}
\newrobustcmd{\gravebmvarthetaz}[2][]{\ensuremath{\subp{\grave{\bm{\vartheta}}}{}{#2}{}{#1}}}
\newrobustcmd{\dotbmvarthetaz}[2][]{\ensuremath{\subp{\dot{\bm{\vartheta}}}{}{#2}{}{#1}}}
\newrobustcmd{\ddotbmvarthetaz}[2][]{\ensuremath{\subp{\ddot{\bm{\vartheta}}}{}{#2}{}{#1}}}
\newrobustcmd{\brevebmvarthetaz}[2][]{\ensuremath{\subp{\breve{\bm{\vartheta}}}{}{#2}{}{#1}}}
\newrobustcmd{\barbmvarthetaz}[2][]{\ensuremath{\subp{\bar{\bm{\vartheta}}}{}{#2}{}{#1}}}
\newrobustcmd{\vecbmvarthetaz}[2][]{\ensuremath{\subp{\vec{\bm{\vartheta}}}{}{#2}{}{#1}}}
\newrobustcmd{\iotaz}[2][]{\ensuremath{\subp{\iota}{}{#2}{}{#1}}}
\newrobustcmd{\hatiotaz}[2][]{\ensuremath{\subp{\hat{\iota}}{}{#2}{}{#1}}}
\newrobustcmd{\widehatiotaz}[2][]{\ensuremath{\subp{\widehat{\iota}}{}{#2}{}{#1}}}
\newrobustcmd{\checkiotaz}[2][]{\ensuremath{\subp{\check{\iota}}{}{#2}{}{#1}}}
\newrobustcmd{\tildeiotaz}[2][]{\ensuremath{\subp{\tilde{\iota}}{}{#2}{}{#1}}}
\newrobustcmd{\widetildeiotaz}[2][]{\ensuremath{\subp{\widetilde{\iota}}{}{#2}{}{#1}}}
\newrobustcmd{\acuteiotaz}[2][]{\ensuremath{\subp{\acute{\iota}}{}{#2}{}{#1}}}
\newrobustcmd{\graveiotaz}[2][]{\ensuremath{\subp{\grave{\iota}}{}{#2}{}{#1}}}
\newrobustcmd{\dotiotaz}[2][]{\ensuremath{\subp{\dot{\iota}}{}{#2}{}{#1}}}
\newrobustcmd{\ddotiotaz}[2][]{\ensuremath{\subp{\ddot{\iota}}{}{#2}{}{#1}}}
\newrobustcmd{\breveiotaz}[2][]{\ensuremath{\subp{\breve{\iota}}{}{#2}{}{#1}}}
\newrobustcmd{\bariotaz}[2][]{\ensuremath{\subp{\bar{\iota}}{}{#2}{}{#1}}}
\newrobustcmd{\veciotaz}[2][]{\ensuremath{\subp{\vec{\iota}}{}{#2}{}{#1}}}
\newrobustcmd{\bmiotaz}[2][]{\ensuremath{\subp{\bm{\iota}}{}{#2}{}{#1}}}
\newrobustcmd{\hatbmiotaz}[2][]{\ensuremath{\subp{\hat{\bm{\iota}}}{}{#2}{}{#1}}}
\newrobustcmd{\widehatbmiotaz}[2][]{\ensuremath{\subp{\widehat{\bm{\iota}}}{}{#2}{}{#1}}}
\newrobustcmd{\checkbmiotaz}[2][]{\ensuremath{\subp{\check{\bm{\iota}}}{}{#2}{}{#1}}}
\newrobustcmd{\tildebmiotaz}[2][]{\ensuremath{\subp{\tilde{\bm{\iota}}}{}{#2}{}{#1}}}
\newrobustcmd{\widetildebmiotaz}[2][]{\ensuremath{\subp{\widetilde{\bm{\iota}}}{}{#2}{}{#1}}}
\newrobustcmd{\acutebmiotaz}[2][]{\ensuremath{\subp{\acute{\bm{\iota}}}{}{#2}{}{#1}}}
\newrobustcmd{\gravebmiotaz}[2][]{\ensuremath{\subp{\grave{\bm{\iota}}}{}{#2}{}{#1}}}
\newrobustcmd{\dotbmiotaz}[2][]{\ensuremath{\subp{\dot{\bm{\iota}}}{}{#2}{}{#1}}}
\newrobustcmd{\ddotbmiotaz}[2][]{\ensuremath{\subp{\ddot{\bm{\iota}}}{}{#2}{}{#1}}}
\newrobustcmd{\brevebmiotaz}[2][]{\ensuremath{\subp{\breve{\bm{\iota}}}{}{#2}{}{#1}}}
\newrobustcmd{\barbmiotaz}[2][]{\ensuremath{\subp{\bar{\bm{\iota}}}{}{#2}{}{#1}}}
\newrobustcmd{\vecbmiotaz}[2][]{\ensuremath{\subp{\vec{\bm{\iota}}}{}{#2}{}{#1}}}
\newrobustcmd{\kappaz}[2][]{\ensuremath{\subp{\kappa}{}{#2}{}{#1}}}
\newrobustcmd{\hatkappaz}[2][]{\ensuremath{\subp{\hat{\kappa}}{}{#2}{}{#1}}}
\newrobustcmd{\widehatkappaz}[2][]{\ensuremath{\subp{\widehat{\kappa}}{}{#2}{}{#1}}}
\newrobustcmd{\checkkappaz}[2][]{\ensuremath{\subp{\check{\kappa}}{}{#2}{}{#1}}}
\newrobustcmd{\tildekappaz}[2][]{\ensuremath{\subp{\tilde{\kappa}}{}{#2}{}{#1}}}
\newrobustcmd{\widetildekappaz}[2][]{\ensuremath{\subp{\widetilde{\kappa}}{}{#2}{}{#1}}}
\newrobustcmd{\acutekappaz}[2][]{\ensuremath{\subp{\acute{\kappa}}{}{#2}{}{#1}}}
\newrobustcmd{\gravekappaz}[2][]{\ensuremath{\subp{\grave{\kappa}}{}{#2}{}{#1}}}
\newrobustcmd{\dotkappaz}[2][]{\ensuremath{\subp{\dot{\kappa}}{}{#2}{}{#1}}}
\newrobustcmd{\ddotkappaz}[2][]{\ensuremath{\subp{\ddot{\kappa}}{}{#2}{}{#1}}}
\newrobustcmd{\brevekappaz}[2][]{\ensuremath{\subp{\breve{\kappa}}{}{#2}{}{#1}}}
\newrobustcmd{\barkappaz}[2][]{\ensuremath{\subp{\bar{\kappa}}{}{#2}{}{#1}}}
\newrobustcmd{\veckappaz}[2][]{\ensuremath{\subp{\vec{\kappa}}{}{#2}{}{#1}}}
\newrobustcmd{\bmkappaz}[2][]{\ensuremath{\subp{\bm{\kappa}}{}{#2}{}{#1}}}
\newrobustcmd{\hatbmkappaz}[2][]{\ensuremath{\subp{\hat{\bm{\kappa}}}{}{#2}{}{#1}}}
\newrobustcmd{\widehatbmkappaz}[2][]{\ensuremath{\subp{\widehat{\bm{\kappa}}}{}{#2}{}{#1}}}
\newrobustcmd{\checkbmkappaz}[2][]{\ensuremath{\subp{\check{\bm{\kappa}}}{}{#2}{}{#1}}}
\newrobustcmd{\tildebmkappaz}[2][]{\ensuremath{\subp{\tilde{\bm{\kappa}}}{}{#2}{}{#1}}}
\newrobustcmd{\widetildebmkappaz}[2][]{\ensuremath{\subp{\widetilde{\bm{\kappa}}}{}{#2}{}{#1}}}
\newrobustcmd{\acutebmkappaz}[2][]{\ensuremath{\subp{\acute{\bm{\kappa}}}{}{#2}{}{#1}}}
\newrobustcmd{\gravebmkappaz}[2][]{\ensuremath{\subp{\grave{\bm{\kappa}}}{}{#2}{}{#1}}}
\newrobustcmd{\dotbmkappaz}[2][]{\ensuremath{\subp{\dot{\bm{\kappa}}}{}{#2}{}{#1}}}
\newrobustcmd{\ddotbmkappaz}[2][]{\ensuremath{\subp{\ddot{\bm{\kappa}}}{}{#2}{}{#1}}}
\newrobustcmd{\brevebmkappaz}[2][]{\ensuremath{\subp{\breve{\bm{\kappa}}}{}{#2}{}{#1}}}
\newrobustcmd{\barbmkappaz}[2][]{\ensuremath{\subp{\bar{\bm{\kappa}}}{}{#2}{}{#1}}}
\newrobustcmd{\vecbmkappaz}[2][]{\ensuremath{\subp{\vec{\bm{\kappa}}}{}{#2}{}{#1}}}
\newrobustcmd{\varkappaz}[2][]{\ensuremath{\subp{\varkappa}{}{#2}{}{#1}}}
\newrobustcmd{\hatvarkappaz}[2][]{\ensuremath{\subp{\hat{\varkappa}}{}{#2}{}{#1}}}
\newrobustcmd{\widehatvarkappaz}[2][]{\ensuremath{\subp{\widehat{\varkappa}}{}{#2}{}{#1}}}
\newrobustcmd{\checkvarkappaz}[2][]{\ensuremath{\subp{\check{\varkappa}}{}{#2}{}{#1}}}
\newrobustcmd{\tildevarkappaz}[2][]{\ensuremath{\subp{\tilde{\varkappa}}{}{#2}{}{#1}}}
\newrobustcmd{\widetildevarkappaz}[2][]{\ensuremath{\subp{\widetilde{\varkappa}}{}{#2}{}{#1}}}
\newrobustcmd{\acutevarkappaz}[2][]{\ensuremath{\subp{\acute{\varkappa}}{}{#2}{}{#1}}}
\newrobustcmd{\gravevarkappaz}[2][]{\ensuremath{\subp{\grave{\varkappa}}{}{#2}{}{#1}}}
\newrobustcmd{\dotvarkappaz}[2][]{\ensuremath{\subp{\dot{\varkappa}}{}{#2}{}{#1}}}
\newrobustcmd{\ddotvarkappaz}[2][]{\ensuremath{\subp{\ddot{\varkappa}}{}{#2}{}{#1}}}
\newrobustcmd{\brevevarkappaz}[2][]{\ensuremath{\subp{\breve{\varkappa}}{}{#2}{}{#1}}}
\newrobustcmd{\barvarkappaz}[2][]{\ensuremath{\subp{\bar{\varkappa}}{}{#2}{}{#1}}}
\newrobustcmd{\vecvarkappaz}[2][]{\ensuremath{\subp{\vec{\varkappa}}{}{#2}{}{#1}}}
\newrobustcmd{\bmvarkappaz}[2][]{\ensuremath{\subp{\bm{\varkappa}}{}{#2}{}{#1}}}
\newrobustcmd{\hatbmvarkappaz}[2][]{\ensuremath{\subp{\hat{\bm{\varkappa}}}{}{#2}{}{#1}}}
\newrobustcmd{\widehatbmvarkappaz}[2][]{\ensuremath{\subp{\widehat{\bm{\varkappa}}}{}{#2}{}{#1}}}
\newrobustcmd{\checkbmvarkappaz}[2][]{\ensuremath{\subp{\check{\bm{\varkappa}}}{}{#2}{}{#1}}}
\newrobustcmd{\tildebmvarkappaz}[2][]{\ensuremath{\subp{\tilde{\bm{\varkappa}}}{}{#2}{}{#1}}}
\newrobustcmd{\widetildebmvarkappaz}[2][]{\ensuremath{\subp{\widetilde{\bm{\varkappa}}}{}{#2}{}{#1}}}
\newrobustcmd{\acutebmvarkappaz}[2][]{\ensuremath{\subp{\acute{\bm{\varkappa}}}{}{#2}{}{#1}}}
\newrobustcmd{\gravebmvarkappaz}[2][]{\ensuremath{\subp{\grave{\bm{\varkappa}}}{}{#2}{}{#1}}}
\newrobustcmd{\dotbmvarkappaz}[2][]{\ensuremath{\subp{\dot{\bm{\varkappa}}}{}{#2}{}{#1}}}
\newrobustcmd{\ddotbmvarkappaz}[2][]{\ensuremath{\subp{\ddot{\bm{\varkappa}}}{}{#2}{}{#1}}}
\newrobustcmd{\brevebmvarkappaz}[2][]{\ensuremath{\subp{\breve{\bm{\varkappa}}}{}{#2}{}{#1}}}
\newrobustcmd{\barbmvarkappaz}[2][]{\ensuremath{\subp{\bar{\bm{\varkappa}}}{}{#2}{}{#1}}}
\newrobustcmd{\vecbmvarkappaz}[2][]{\ensuremath{\subp{\vec{\bm{\varkappa}}}{}{#2}{}{#1}}}
\newrobustcmd{\lambdaz}[2][]{\ensuremath{\subp{\lambda}{}{#2}{}{#1}}}
\newrobustcmd{\hatlambdaz}[2][]{\ensuremath{\subp{\hat{\lambda}}{}{#2}{}{#1}}}
\newrobustcmd{\widehatlambdaz}[2][]{\ensuremath{\subp{\widehat{\lambda}}{}{#2}{}{#1}}}
\newrobustcmd{\checklambdaz}[2][]{\ensuremath{\subp{\check{\lambda}}{}{#2}{}{#1}}}
\newrobustcmd{\tildelambdaz}[2][]{\ensuremath{\subp{\tilde{\lambda}}{}{#2}{}{#1}}}
\newrobustcmd{\widetildelambdaz}[2][]{\ensuremath{\subp{\widetilde{\lambda}}{}{#2}{}{#1}}}
\newrobustcmd{\acutelambdaz}[2][]{\ensuremath{\subp{\acute{\lambda}}{}{#2}{}{#1}}}
\newrobustcmd{\gravelambdaz}[2][]{\ensuremath{\subp{\grave{\lambda}}{}{#2}{}{#1}}}
\newrobustcmd{\dotlambdaz}[2][]{\ensuremath{\subp{\dot{\lambda}}{}{#2}{}{#1}}}
\newrobustcmd{\ddotlambdaz}[2][]{\ensuremath{\subp{\ddot{\lambda}}{}{#2}{}{#1}}}
\newrobustcmd{\brevelambdaz}[2][]{\ensuremath{\subp{\breve{\lambda}}{}{#2}{}{#1}}}
\newrobustcmd{\barlambdaz}[2][]{\ensuremath{\subp{\bar{\lambda}}{}{#2}{}{#1}}}
\newrobustcmd{\veclambdaz}[2][]{\ensuremath{\subp{\vec{\lambda}}{}{#2}{}{#1}}}
\newrobustcmd{\bmlambdaz}[2][]{\ensuremath{\subp{\bm{\lambda}}{}{#2}{}{#1}}}
\newrobustcmd{\hatbmlambdaz}[2][]{\ensuremath{\subp{\hat{\bm{\lambda}}}{}{#2}{}{#1}}}
\newrobustcmd{\widehatbmlambdaz}[2][]{\ensuremath{\subp{\widehat{\bm{\lambda}}}{}{#2}{}{#1}}}
\newrobustcmd{\checkbmlambdaz}[2][]{\ensuremath{\subp{\check{\bm{\lambda}}}{}{#2}{}{#1}}}
\newrobustcmd{\tildebmlambdaz}[2][]{\ensuremath{\subp{\tilde{\bm{\lambda}}}{}{#2}{}{#1}}}
\newrobustcmd{\widetildebmlambdaz}[2][]{\ensuremath{\subp{\widetilde{\bm{\lambda}}}{}{#2}{}{#1}}}
\newrobustcmd{\acutebmlambdaz}[2][]{\ensuremath{\subp{\acute{\bm{\lambda}}}{}{#2}{}{#1}}}
\newrobustcmd{\gravebmlambdaz}[2][]{\ensuremath{\subp{\grave{\bm{\lambda}}}{}{#2}{}{#1}}}
\newrobustcmd{\dotbmlambdaz}[2][]{\ensuremath{\subp{\dot{\bm{\lambda}}}{}{#2}{}{#1}}}
\newrobustcmd{\ddotbmlambdaz}[2][]{\ensuremath{\subp{\ddot{\bm{\lambda}}}{}{#2}{}{#1}}}
\newrobustcmd{\brevebmlambdaz}[2][]{\ensuremath{\subp{\breve{\bm{\lambda}}}{}{#2}{}{#1}}}
\newrobustcmd{\barbmlambdaz}[2][]{\ensuremath{\subp{\bar{\bm{\lambda}}}{}{#2}{}{#1}}}
\newrobustcmd{\vecbmlambdaz}[2][]{\ensuremath{\subp{\vec{\bm{\lambda}}}{}{#2}{}{#1}}}
\newrobustcmd{\muz}[2][]{\ensuremath{\subp{\mu}{}{#2}{}{#1}}}
\newrobustcmd{\hatmuz}[2][]{\ensuremath{\subp{\hat{\mu}}{}{#2}{}{#1}}}
\newrobustcmd{\widehatmuz}[2][]{\ensuremath{\subp{\widehat{\mu}}{}{#2}{}{#1}}}
\newrobustcmd{\checkmuz}[2][]{\ensuremath{\subp{\check{\mu}}{}{#2}{}{#1}}}
\newrobustcmd{\tildemuz}[2][]{\ensuremath{\subp{\tilde{\mu}}{}{#2}{}{#1}}}
\newrobustcmd{\widetildemuz}[2][]{\ensuremath{\subp{\widetilde{\mu}}{}{#2}{}{#1}}}
\newrobustcmd{\acutemuz}[2][]{\ensuremath{\subp{\acute{\mu}}{}{#2}{}{#1}}}
\newrobustcmd{\gravemuz}[2][]{\ensuremath{\subp{\grave{\mu}}{}{#2}{}{#1}}}
\newrobustcmd{\dotmuz}[2][]{\ensuremath{\subp{\dot{\mu}}{}{#2}{}{#1}}}
\newrobustcmd{\ddotmuz}[2][]{\ensuremath{\subp{\ddot{\mu}}{}{#2}{}{#1}}}
\newrobustcmd{\brevemuz}[2][]{\ensuremath{\subp{\breve{\mu}}{}{#2}{}{#1}}}
\newrobustcmd{\barmuz}[2][]{\ensuremath{\subp{\bar{\mu}}{}{#2}{}{#1}}}
\newrobustcmd{\vecmuz}[2][]{\ensuremath{\subp{\vec{\mu}}{}{#2}{}{#1}}}
\newrobustcmd{\bmmuz}[2][]{\ensuremath{\subp{\bm{\mu}}{}{#2}{}{#1}}}
\newrobustcmd{\hatbmmuz}[2][]{\ensuremath{\subp{\hat{\bm{\mu}}}{}{#2}{}{#1}}}
\newrobustcmd{\widehatbmmuz}[2][]{\ensuremath{\subp{\widehat{\bm{\mu}}}{}{#2}{}{#1}}}
\newrobustcmd{\checkbmmuz}[2][]{\ensuremath{\subp{\check{\bm{\mu}}}{}{#2}{}{#1}}}
\newrobustcmd{\tildebmmuz}[2][]{\ensuremath{\subp{\tilde{\bm{\mu}}}{}{#2}{}{#1}}}
\newrobustcmd{\widetildebmmuz}[2][]{\ensuremath{\subp{\widetilde{\bm{\mu}}}{}{#2}{}{#1}}}
\newrobustcmd{\acutebmmuz}[2][]{\ensuremath{\subp{\acute{\bm{\mu}}}{}{#2}{}{#1}}}
\newrobustcmd{\gravebmmuz}[2][]{\ensuremath{\subp{\grave{\bm{\mu}}}{}{#2}{}{#1}}}
\newrobustcmd{\dotbmmuz}[2][]{\ensuremath{\subp{\dot{\bm{\mu}}}{}{#2}{}{#1}}}
\newrobustcmd{\ddotbmmuz}[2][]{\ensuremath{\subp{\ddot{\bm{\mu}}}{}{#2}{}{#1}}}
\newrobustcmd{\brevebmmuz}[2][]{\ensuremath{\subp{\breve{\bm{\mu}}}{}{#2}{}{#1}}}
\newrobustcmd{\barbmmuz}[2][]{\ensuremath{\subp{\bar{\bm{\mu}}}{}{#2}{}{#1}}}
\newrobustcmd{\vecbmmuz}[2][]{\ensuremath{\subp{\vec{\bm{\mu}}}{}{#2}{}{#1}}}
\newrobustcmd{\nuz}[2][]{\ensuremath{\subp{\nu}{}{#2}{}{#1}}}
\newrobustcmd{\hatnuz}[2][]{\ensuremath{\subp{\hat{\nu}}{}{#2}{}{#1}}}
\newrobustcmd{\widehatnuz}[2][]{\ensuremath{\subp{\widehat{\nu}}{}{#2}{}{#1}}}
\newrobustcmd{\checknuz}[2][]{\ensuremath{\subp{\check{\nu}}{}{#2}{}{#1}}}
\newrobustcmd{\tildenuz}[2][]{\ensuremath{\subp{\tilde{\nu}}{}{#2}{}{#1}}}
\newrobustcmd{\widetildenuz}[2][]{\ensuremath{\subp{\widetilde{\nu}}{}{#2}{}{#1}}}
\newrobustcmd{\acutenuz}[2][]{\ensuremath{\subp{\acute{\nu}}{}{#2}{}{#1}}}
\newrobustcmd{\gravenuz}[2][]{\ensuremath{\subp{\grave{\nu}}{}{#2}{}{#1}}}
\newrobustcmd{\dotnuz}[2][]{\ensuremath{\subp{\dot{\nu}}{}{#2}{}{#1}}}
\newrobustcmd{\ddotnuz}[2][]{\ensuremath{\subp{\ddot{\nu}}{}{#2}{}{#1}}}
\newrobustcmd{\brevenuz}[2][]{\ensuremath{\subp{\breve{\nu}}{}{#2}{}{#1}}}
\newrobustcmd{\barnuz}[2][]{\ensuremath{\subp{\bar{\nu}}{}{#2}{}{#1}}}
\newrobustcmd{\vecnuz}[2][]{\ensuremath{\subp{\vec{\nu}}{}{#2}{}{#1}}}
\newrobustcmd{\bmnuz}[2][]{\ensuremath{\subp{\bm{\nu}}{}{#2}{}{#1}}}
\newrobustcmd{\hatbmnuz}[2][]{\ensuremath{\subp{\hat{\bm{\nu}}}{}{#2}{}{#1}}}
\newrobustcmd{\widehatbmnuz}[2][]{\ensuremath{\subp{\widehat{\bm{\nu}}}{}{#2}{}{#1}}}
\newrobustcmd{\checkbmnuz}[2][]{\ensuremath{\subp{\check{\bm{\nu}}}{}{#2}{}{#1}}}
\newrobustcmd{\tildebmnuz}[2][]{\ensuremath{\subp{\tilde{\bm{\nu}}}{}{#2}{}{#1}}}
\newrobustcmd{\widetildebmnuz}[2][]{\ensuremath{\subp{\widetilde{\bm{\nu}}}{}{#2}{}{#1}}}
\newrobustcmd{\acutebmnuz}[2][]{\ensuremath{\subp{\acute{\bm{\nu}}}{}{#2}{}{#1}}}
\newrobustcmd{\gravebmnuz}[2][]{\ensuremath{\subp{\grave{\bm{\nu}}}{}{#2}{}{#1}}}
\newrobustcmd{\dotbmnuz}[2][]{\ensuremath{\subp{\dot{\bm{\nu}}}{}{#2}{}{#1}}}
\newrobustcmd{\ddotbmnuz}[2][]{\ensuremath{\subp{\ddot{\bm{\nu}}}{}{#2}{}{#1}}}
\newrobustcmd{\brevebmnuz}[2][]{\ensuremath{\subp{\breve{\bm{\nu}}}{}{#2}{}{#1}}}
\newrobustcmd{\barbmnuz}[2][]{\ensuremath{\subp{\bar{\bm{\nu}}}{}{#2}{}{#1}}}
\newrobustcmd{\vecbmnuz}[2][]{\ensuremath{\subp{\vec{\bm{\nu}}}{}{#2}{}{#1}}}
\newrobustcmd{\xiz}[2][]{\ensuremath{\subp{\xi}{}{#2}{}{#1}}}
\newrobustcmd{\hatxiz}[2][]{\ensuremath{\subp{\hat{\xi}}{}{#2}{}{#1}}}
\newrobustcmd{\widehatxiz}[2][]{\ensuremath{\subp{\widehat{\xi}}{}{#2}{}{#1}}}
\newrobustcmd{\checkxiz}[2][]{\ensuremath{\subp{\check{\xi}}{}{#2}{}{#1}}}
\newrobustcmd{\tildexiz}[2][]{\ensuremath{\subp{\tilde{\xi}}{}{#2}{}{#1}}}
\newrobustcmd{\widetildexiz}[2][]{\ensuremath{\subp{\widetilde{\xi}}{}{#2}{}{#1}}}
\newrobustcmd{\acutexiz}[2][]{\ensuremath{\subp{\acute{\xi}}{}{#2}{}{#1}}}
\newrobustcmd{\gravexiz}[2][]{\ensuremath{\subp{\grave{\xi}}{}{#2}{}{#1}}}
\newrobustcmd{\dotxiz}[2][]{\ensuremath{\subp{\dot{\xi}}{}{#2}{}{#1}}}
\newrobustcmd{\ddotxiz}[2][]{\ensuremath{\subp{\ddot{\xi}}{}{#2}{}{#1}}}
\newrobustcmd{\brevexiz}[2][]{\ensuremath{\subp{\breve{\xi}}{}{#2}{}{#1}}}
\newrobustcmd{\barxiz}[2][]{\ensuremath{\subp{\bar{\xi}}{}{#2}{}{#1}}}
\newrobustcmd{\vecxiz}[2][]{\ensuremath{\subp{\vec{\xi}}{}{#2}{}{#1}}}
\newrobustcmd{\bmxiz}[2][]{\ensuremath{\subp{\bm{\xi}}{}{#2}{}{#1}}}
\newrobustcmd{\hatbmxiz}[2][]{\ensuremath{\subp{\hat{\bm{\xi}}}{}{#2}{}{#1}}}
\newrobustcmd{\widehatbmxiz}[2][]{\ensuremath{\subp{\widehat{\bm{\xi}}}{}{#2}{}{#1}}}
\newrobustcmd{\checkbmxiz}[2][]{\ensuremath{\subp{\check{\bm{\xi}}}{}{#2}{}{#1}}}
\newrobustcmd{\tildebmxiz}[2][]{\ensuremath{\subp{\tilde{\bm{\xi}}}{}{#2}{}{#1}}}
\newrobustcmd{\widetildebmxiz}[2][]{\ensuremath{\subp{\widetilde{\bm{\xi}}}{}{#2}{}{#1}}}
\newrobustcmd{\acutebmxiz}[2][]{\ensuremath{\subp{\acute{\bm{\xi}}}{}{#2}{}{#1}}}
\newrobustcmd{\gravebmxiz}[2][]{\ensuremath{\subp{\grave{\bm{\xi}}}{}{#2}{}{#1}}}
\newrobustcmd{\dotbmxiz}[2][]{\ensuremath{\subp{\dot{\bm{\xi}}}{}{#2}{}{#1}}}
\newrobustcmd{\ddotbmxiz}[2][]{\ensuremath{\subp{\ddot{\bm{\xi}}}{}{#2}{}{#1}}}
\newrobustcmd{\brevebmxiz}[2][]{\ensuremath{\subp{\breve{\bm{\xi}}}{}{#2}{}{#1}}}
\newrobustcmd{\barbmxiz}[2][]{\ensuremath{\subp{\bar{\bm{\xi}}}{}{#2}{}{#1}}}
\newrobustcmd{\vecbmxiz}[2][]{\ensuremath{\subp{\vec{\bm{\xi}}}{}{#2}{}{#1}}}
\newrobustcmd{\piz}[2][]{\ensuremath{\subp{\pi}{}{#2}{}{#1}}}
\newrobustcmd{\hatpiz}[2][]{\ensuremath{\subp{\hat{\pi}}{}{#2}{}{#1}}}
\newrobustcmd{\widehatpiz}[2][]{\ensuremath{\subp{\widehat{\pi}}{}{#2}{}{#1}}}
\newrobustcmd{\checkpiz}[2][]{\ensuremath{\subp{\check{\pi}}{}{#2}{}{#1}}}
\newrobustcmd{\tildepiz}[2][]{\ensuremath{\subp{\tilde{\pi}}{}{#2}{}{#1}}}
\newrobustcmd{\widetildepiz}[2][]{\ensuremath{\subp{\widetilde{\pi}}{}{#2}{}{#1}}}
\newrobustcmd{\acutepiz}[2][]{\ensuremath{\subp{\acute{\pi}}{}{#2}{}{#1}}}
\newrobustcmd{\gravepiz}[2][]{\ensuremath{\subp{\grave{\pi}}{}{#2}{}{#1}}}
\newrobustcmd{\dotpiz}[2][]{\ensuremath{\subp{\dot{\pi}}{}{#2}{}{#1}}}
\newrobustcmd{\ddotpiz}[2][]{\ensuremath{\subp{\ddot{\pi}}{}{#2}{}{#1}}}
\newrobustcmd{\brevepiz}[2][]{\ensuremath{\subp{\breve{\pi}}{}{#2}{}{#1}}}
\newrobustcmd{\barpiz}[2][]{\ensuremath{\subp{\bar{\pi}}{}{#2}{}{#1}}}
\newrobustcmd{\vecpiz}[2][]{\ensuremath{\subp{\vec{\pi}}{}{#2}{}{#1}}}
\newrobustcmd{\bmpiz}[2][]{\ensuremath{\subp{\bm{\pi}}{}{#2}{}{#1}}}
\newrobustcmd{\hatbmpiz}[2][]{\ensuremath{\subp{\hat{\bm{\pi}}}{}{#2}{}{#1}}}
\newrobustcmd{\widehatbmpiz}[2][]{\ensuremath{\subp{\widehat{\bm{\pi}}}{}{#2}{}{#1}}}
\newrobustcmd{\checkbmpiz}[2][]{\ensuremath{\subp{\check{\bm{\pi}}}{}{#2}{}{#1}}}
\newrobustcmd{\tildebmpiz}[2][]{\ensuremath{\subp{\tilde{\bm{\pi}}}{}{#2}{}{#1}}}
\newrobustcmd{\widetildebmpiz}[2][]{\ensuremath{\subp{\widetilde{\bm{\pi}}}{}{#2}{}{#1}}}
\newrobustcmd{\acutebmpiz}[2][]{\ensuremath{\subp{\acute{\bm{\pi}}}{}{#2}{}{#1}}}
\newrobustcmd{\gravebmpiz}[2][]{\ensuremath{\subp{\grave{\bm{\pi}}}{}{#2}{}{#1}}}
\newrobustcmd{\dotbmpiz}[2][]{\ensuremath{\subp{\dot{\bm{\pi}}}{}{#2}{}{#1}}}
\newrobustcmd{\ddotbmpiz}[2][]{\ensuremath{\subp{\ddot{\bm{\pi}}}{}{#2}{}{#1}}}
\newrobustcmd{\brevebmpiz}[2][]{\ensuremath{\subp{\breve{\bm{\pi}}}{}{#2}{}{#1}}}
\newrobustcmd{\barbmpiz}[2][]{\ensuremath{\subp{\bar{\bm{\pi}}}{}{#2}{}{#1}}}
\newrobustcmd{\vecbmpiz}[2][]{\ensuremath{\subp{\vec{\bm{\pi}}}{}{#2}{}{#1}}}
\newrobustcmd{\varpiz}[2][]{\ensuremath{\subp{\varpi}{}{#2}{}{#1}}}
\newrobustcmd{\hatvarpiz}[2][]{\ensuremath{\subp{\hat{\varpi}}{}{#2}{}{#1}}}
\newrobustcmd{\widehatvarpiz}[2][]{\ensuremath{\subp{\widehat{\varpi}}{}{#2}{}{#1}}}
\newrobustcmd{\checkvarpiz}[2][]{\ensuremath{\subp{\check{\varpi}}{}{#2}{}{#1}}}
\newrobustcmd{\tildevarpiz}[2][]{\ensuremath{\subp{\tilde{\varpi}}{}{#2}{}{#1}}}
\newrobustcmd{\widetildevarpiz}[2][]{\ensuremath{\subp{\widetilde{\varpi}}{}{#2}{}{#1}}}
\newrobustcmd{\acutevarpiz}[2][]{\ensuremath{\subp{\acute{\varpi}}{}{#2}{}{#1}}}
\newrobustcmd{\gravevarpiz}[2][]{\ensuremath{\subp{\grave{\varpi}}{}{#2}{}{#1}}}
\newrobustcmd{\dotvarpiz}[2][]{\ensuremath{\subp{\dot{\varpi}}{}{#2}{}{#1}}}
\newrobustcmd{\ddotvarpiz}[2][]{\ensuremath{\subp{\ddot{\varpi}}{}{#2}{}{#1}}}
\newrobustcmd{\brevevarpiz}[2][]{\ensuremath{\subp{\breve{\varpi}}{}{#2}{}{#1}}}
\newrobustcmd{\barvarpiz}[2][]{\ensuremath{\subp{\bar{\varpi}}{}{#2}{}{#1}}}
\newrobustcmd{\vecvarpiz}[2][]{\ensuremath{\subp{\vec{\varpi}}{}{#2}{}{#1}}}
\newrobustcmd{\bmvarpiz}[2][]{\ensuremath{\subp{\bm{\varpi}}{}{#2}{}{#1}}}
\newrobustcmd{\hatbmvarpiz}[2][]{\ensuremath{\subp{\hat{\bm{\varpi}}}{}{#2}{}{#1}}}
\newrobustcmd{\widehatbmvarpiz}[2][]{\ensuremath{\subp{\widehat{\bm{\varpi}}}{}{#2}{}{#1}}}
\newrobustcmd{\checkbmvarpiz}[2][]{\ensuremath{\subp{\check{\bm{\varpi}}}{}{#2}{}{#1}}}
\newrobustcmd{\tildebmvarpiz}[2][]{\ensuremath{\subp{\tilde{\bm{\varpi}}}{}{#2}{}{#1}}}
\newrobustcmd{\widetildebmvarpiz}[2][]{\ensuremath{\subp{\widetilde{\bm{\varpi}}}{}{#2}{}{#1}}}
\newrobustcmd{\acutebmvarpiz}[2][]{\ensuremath{\subp{\acute{\bm{\varpi}}}{}{#2}{}{#1}}}
\newrobustcmd{\gravebmvarpiz}[2][]{\ensuremath{\subp{\grave{\bm{\varpi}}}{}{#2}{}{#1}}}
\newrobustcmd{\dotbmvarpiz}[2][]{\ensuremath{\subp{\dot{\bm{\varpi}}}{}{#2}{}{#1}}}
\newrobustcmd{\ddotbmvarpiz}[2][]{\ensuremath{\subp{\ddot{\bm{\varpi}}}{}{#2}{}{#1}}}
\newrobustcmd{\brevebmvarpiz}[2][]{\ensuremath{\subp{\breve{\bm{\varpi}}}{}{#2}{}{#1}}}
\newrobustcmd{\barbmvarpiz}[2][]{\ensuremath{\subp{\bar{\bm{\varpi}}}{}{#2}{}{#1}}}
\newrobustcmd{\vecbmvarpiz}[2][]{\ensuremath{\subp{\vec{\bm{\varpi}}}{}{#2}{}{#1}}}
\newrobustcmd{\rhoz}[2][]{\ensuremath{\subp{\rho}{}{#2}{}{#1}}}
\newrobustcmd{\hatrhoz}[2][]{\ensuremath{\subp{\hat{\rho}}{}{#2}{}{#1}}}
\newrobustcmd{\widehatrhoz}[2][]{\ensuremath{\subp{\widehat{\rho}}{}{#2}{}{#1}}}
\newrobustcmd{\checkrhoz}[2][]{\ensuremath{\subp{\check{\rho}}{}{#2}{}{#1}}}
\newrobustcmd{\tilderhoz}[2][]{\ensuremath{\subp{\tilde{\rho}}{}{#2}{}{#1}}}
\newrobustcmd{\widetilderhoz}[2][]{\ensuremath{\subp{\widetilde{\rho}}{}{#2}{}{#1}}}
\newrobustcmd{\acuterhoz}[2][]{\ensuremath{\subp{\acute{\rho}}{}{#2}{}{#1}}}
\newrobustcmd{\graverhoz}[2][]{\ensuremath{\subp{\grave{\rho}}{}{#2}{}{#1}}}
\newrobustcmd{\dotrhoz}[2][]{\ensuremath{\subp{\dot{\rho}}{}{#2}{}{#1}}}
\newrobustcmd{\ddotrhoz}[2][]{\ensuremath{\subp{\ddot{\rho}}{}{#2}{}{#1}}}
\newrobustcmd{\breverhoz}[2][]{\ensuremath{\subp{\breve{\rho}}{}{#2}{}{#1}}}
\newrobustcmd{\barrhoz}[2][]{\ensuremath{\subp{\bar{\rho}}{}{#2}{}{#1}}}
\newrobustcmd{\vecrhoz}[2][]{\ensuremath{\subp{\vec{\rho}}{}{#2}{}{#1}}}
\newrobustcmd{\bmrhoz}[2][]{\ensuremath{\subp{\bm{\rho}}{}{#2}{}{#1}}}
\newrobustcmd{\hatbmrhoz}[2][]{\ensuremath{\subp{\hat{\bm{\rho}}}{}{#2}{}{#1}}}
\newrobustcmd{\widehatbmrhoz}[2][]{\ensuremath{\subp{\widehat{\bm{\rho}}}{}{#2}{}{#1}}}
\newrobustcmd{\checkbmrhoz}[2][]{\ensuremath{\subp{\check{\bm{\rho}}}{}{#2}{}{#1}}}
\newrobustcmd{\tildebmrhoz}[2][]{\ensuremath{\subp{\tilde{\bm{\rho}}}{}{#2}{}{#1}}}
\newrobustcmd{\widetildebmrhoz}[2][]{\ensuremath{\subp{\widetilde{\bm{\rho}}}{}{#2}{}{#1}}}
\newrobustcmd{\acutebmrhoz}[2][]{\ensuremath{\subp{\acute{\bm{\rho}}}{}{#2}{}{#1}}}
\newrobustcmd{\gravebmrhoz}[2][]{\ensuremath{\subp{\grave{\bm{\rho}}}{}{#2}{}{#1}}}
\newrobustcmd{\dotbmrhoz}[2][]{\ensuremath{\subp{\dot{\bm{\rho}}}{}{#2}{}{#1}}}
\newrobustcmd{\ddotbmrhoz}[2][]{\ensuremath{\subp{\ddot{\bm{\rho}}}{}{#2}{}{#1}}}
\newrobustcmd{\brevebmrhoz}[2][]{\ensuremath{\subp{\breve{\bm{\rho}}}{}{#2}{}{#1}}}
\newrobustcmd{\barbmrhoz}[2][]{\ensuremath{\subp{\bar{\bm{\rho}}}{}{#2}{}{#1}}}
\newrobustcmd{\vecbmrhoz}[2][]{\ensuremath{\subp{\vec{\bm{\rho}}}{}{#2}{}{#1}}}
\newrobustcmd{\varrhoz}[2][]{\ensuremath{\subp{\varrho}{}{#2}{}{#1}}}
\newrobustcmd{\hatvarrhoz}[2][]{\ensuremath{\subp{\hat{\varrho}}{}{#2}{}{#1}}}
\newrobustcmd{\widehatvarrhoz}[2][]{\ensuremath{\subp{\widehat{\varrho}}{}{#2}{}{#1}}}
\newrobustcmd{\checkvarrhoz}[2][]{\ensuremath{\subp{\check{\varrho}}{}{#2}{}{#1}}}
\newrobustcmd{\tildevarrhoz}[2][]{\ensuremath{\subp{\tilde{\varrho}}{}{#2}{}{#1}}}
\newrobustcmd{\widetildevarrhoz}[2][]{\ensuremath{\subp{\widetilde{\varrho}}{}{#2}{}{#1}}}
\newrobustcmd{\acutevarrhoz}[2][]{\ensuremath{\subp{\acute{\varrho}}{}{#2}{}{#1}}}
\newrobustcmd{\gravevarrhoz}[2][]{\ensuremath{\subp{\grave{\varrho}}{}{#2}{}{#1}}}
\newrobustcmd{\dotvarrhoz}[2][]{\ensuremath{\subp{\dot{\varrho}}{}{#2}{}{#1}}}
\newrobustcmd{\ddotvarrhoz}[2][]{\ensuremath{\subp{\ddot{\varrho}}{}{#2}{}{#1}}}
\newrobustcmd{\brevevarrhoz}[2][]{\ensuremath{\subp{\breve{\varrho}}{}{#2}{}{#1}}}
\newrobustcmd{\barvarrhoz}[2][]{\ensuremath{\subp{\bar{\varrho}}{}{#2}{}{#1}}}
\newrobustcmd{\vecvarrhoz}[2][]{\ensuremath{\subp{\vec{\varrho}}{}{#2}{}{#1}}}
\newrobustcmd{\bmvarrhoz}[2][]{\ensuremath{\subp{\bm{\varrho}}{}{#2}{}{#1}}}
\newrobustcmd{\hatbmvarrhoz}[2][]{\ensuremath{\subp{\hat{\bm{\varrho}}}{}{#2}{}{#1}}}
\newrobustcmd{\widehatbmvarrhoz}[2][]{\ensuremath{\subp{\widehat{\bm{\varrho}}}{}{#2}{}{#1}}}
\newrobustcmd{\checkbmvarrhoz}[2][]{\ensuremath{\subp{\check{\bm{\varrho}}}{}{#2}{}{#1}}}
\newrobustcmd{\tildebmvarrhoz}[2][]{\ensuremath{\subp{\tilde{\bm{\varrho}}}{}{#2}{}{#1}}}
\newrobustcmd{\widetildebmvarrhoz}[2][]{\ensuremath{\subp{\widetilde{\bm{\varrho}}}{}{#2}{}{#1}}}
\newrobustcmd{\acutebmvarrhoz}[2][]{\ensuremath{\subp{\acute{\bm{\varrho}}}{}{#2}{}{#1}}}
\newrobustcmd{\gravebmvarrhoz}[2][]{\ensuremath{\subp{\grave{\bm{\varrho}}}{}{#2}{}{#1}}}
\newrobustcmd{\dotbmvarrhoz}[2][]{\ensuremath{\subp{\dot{\bm{\varrho}}}{}{#2}{}{#1}}}
\newrobustcmd{\ddotbmvarrhoz}[2][]{\ensuremath{\subp{\ddot{\bm{\varrho}}}{}{#2}{}{#1}}}
\newrobustcmd{\brevebmvarrhoz}[2][]{\ensuremath{\subp{\breve{\bm{\varrho}}}{}{#2}{}{#1}}}
\newrobustcmd{\barbmvarrhoz}[2][]{\ensuremath{\subp{\bar{\bm{\varrho}}}{}{#2}{}{#1}}}
\newrobustcmd{\vecbmvarrhoz}[2][]{\ensuremath{\subp{\vec{\bm{\varrho}}}{}{#2}{}{#1}}}
\newrobustcmd{\sigmaz}[2][]{\ensuremath{\subp{\sigma}{}{#2}{}{#1}}}
\newrobustcmd{\hatsigmaz}[2][]{\ensuremath{\subp{\hat{\sigma}}{}{#2}{}{#1}}}
\newrobustcmd{\widehatsigmaz}[2][]{\ensuremath{\subp{\widehat{\sigma}}{}{#2}{}{#1}}}
\newrobustcmd{\checksigmaz}[2][]{\ensuremath{\subp{\check{\sigma}}{}{#2}{}{#1}}}
\newrobustcmd{\tildesigmaz}[2][]{\ensuremath{\subp{\tilde{\sigma}}{}{#2}{}{#1}}}
\newrobustcmd{\widetildesigmaz}[2][]{\ensuremath{\subp{\widetilde{\sigma}}{}{#2}{}{#1}}}
\newrobustcmd{\acutesigmaz}[2][]{\ensuremath{\subp{\acute{\sigma}}{}{#2}{}{#1}}}
\newrobustcmd{\gravesigmaz}[2][]{\ensuremath{\subp{\grave{\sigma}}{}{#2}{}{#1}}}
\newrobustcmd{\dotsigmaz}[2][]{\ensuremath{\subp{\dot{\sigma}}{}{#2}{}{#1}}}
\newrobustcmd{\ddotsigmaz}[2][]{\ensuremath{\subp{\ddot{\sigma}}{}{#2}{}{#1}}}
\newrobustcmd{\brevesigmaz}[2][]{\ensuremath{\subp{\breve{\sigma}}{}{#2}{}{#1}}}
\newrobustcmd{\barsigmaz}[2][]{\ensuremath{\subp{\bar{\sigma}}{}{#2}{}{#1}}}
\newrobustcmd{\vecsigmaz}[2][]{\ensuremath{\subp{\vec{\sigma}}{}{#2}{}{#1}}}
\newrobustcmd{\bmsigmaz}[2][]{\ensuremath{\subp{\bm{\sigma}}{}{#2}{}{#1}}}
\newrobustcmd{\hatbmsigmaz}[2][]{\ensuremath{\subp{\hat{\bm{\sigma}}}{}{#2}{}{#1}}}
\newrobustcmd{\widehatbmsigmaz}[2][]{\ensuremath{\subp{\widehat{\bm{\sigma}}}{}{#2}{}{#1}}}
\newrobustcmd{\checkbmsigmaz}[2][]{\ensuremath{\subp{\check{\bm{\sigma}}}{}{#2}{}{#1}}}
\newrobustcmd{\tildebmsigmaz}[2][]{\ensuremath{\subp{\tilde{\bm{\sigma}}}{}{#2}{}{#1}}}
\newrobustcmd{\widetildebmsigmaz}[2][]{\ensuremath{\subp{\widetilde{\bm{\sigma}}}{}{#2}{}{#1}}}
\newrobustcmd{\acutebmsigmaz}[2][]{\ensuremath{\subp{\acute{\bm{\sigma}}}{}{#2}{}{#1}}}
\newrobustcmd{\gravebmsigmaz}[2][]{\ensuremath{\subp{\grave{\bm{\sigma}}}{}{#2}{}{#1}}}
\newrobustcmd{\dotbmsigmaz}[2][]{\ensuremath{\subp{\dot{\bm{\sigma}}}{}{#2}{}{#1}}}
\newrobustcmd{\ddotbmsigmaz}[2][]{\ensuremath{\subp{\ddot{\bm{\sigma}}}{}{#2}{}{#1}}}
\newrobustcmd{\brevebmsigmaz}[2][]{\ensuremath{\subp{\breve{\bm{\sigma}}}{}{#2}{}{#1}}}
\newrobustcmd{\barbmsigmaz}[2][]{\ensuremath{\subp{\bar{\bm{\sigma}}}{}{#2}{}{#1}}}
\newrobustcmd{\vecbmsigmaz}[2][]{\ensuremath{\subp{\vec{\bm{\sigma}}}{}{#2}{}{#1}}}
\newrobustcmd{\varsigmaz}[2][]{\ensuremath{\subp{\varsigma}{}{#2}{}{#1}}}
\newrobustcmd{\hatvarsigmaz}[2][]{\ensuremath{\subp{\hat{\varsigma}}{}{#2}{}{#1}}}
\newrobustcmd{\widehatvarsigmaz}[2][]{\ensuremath{\subp{\widehat{\varsigma}}{}{#2}{}{#1}}}
\newrobustcmd{\checkvarsigmaz}[2][]{\ensuremath{\subp{\check{\varsigma}}{}{#2}{}{#1}}}
\newrobustcmd{\tildevarsigmaz}[2][]{\ensuremath{\subp{\tilde{\varsigma}}{}{#2}{}{#1}}}
\newrobustcmd{\widetildevarsigmaz}[2][]{\ensuremath{\subp{\widetilde{\varsigma}}{}{#2}{}{#1}}}
\newrobustcmd{\acutevarsigmaz}[2][]{\ensuremath{\subp{\acute{\varsigma}}{}{#2}{}{#1}}}
\newrobustcmd{\gravevarsigmaz}[2][]{\ensuremath{\subp{\grave{\varsigma}}{}{#2}{}{#1}}}
\newrobustcmd{\dotvarsigmaz}[2][]{\ensuremath{\subp{\dot{\varsigma}}{}{#2}{}{#1}}}
\newrobustcmd{\ddotvarsigmaz}[2][]{\ensuremath{\subp{\ddot{\varsigma}}{}{#2}{}{#1}}}
\newrobustcmd{\brevevarsigmaz}[2][]{\ensuremath{\subp{\breve{\varsigma}}{}{#2}{}{#1}}}
\newrobustcmd{\barvarsigmaz}[2][]{\ensuremath{\subp{\bar{\varsigma}}{}{#2}{}{#1}}}
\newrobustcmd{\vecvarsigmaz}[2][]{\ensuremath{\subp{\vec{\varsigma}}{}{#2}{}{#1}}}
\newrobustcmd{\bmvarsigmaz}[2][]{\ensuremath{\subp{\bm{\varsigma}}{}{#2}{}{#1}}}
\newrobustcmd{\hatbmvarsigmaz}[2][]{\ensuremath{\subp{\hat{\bm{\varsigma}}}{}{#2}{}{#1}}}
\newrobustcmd{\widehatbmvarsigmaz}[2][]{\ensuremath{\subp{\widehat{\bm{\varsigma}}}{}{#2}{}{#1}}}
\newrobustcmd{\checkbmvarsigmaz}[2][]{\ensuremath{\subp{\check{\bm{\varsigma}}}{}{#2}{}{#1}}}
\newrobustcmd{\tildebmvarsigmaz}[2][]{\ensuremath{\subp{\tilde{\bm{\varsigma}}}{}{#2}{}{#1}}}
\newrobustcmd{\widetildebmvarsigmaz}[2][]{\ensuremath{\subp{\widetilde{\bm{\varsigma}}}{}{#2}{}{#1}}}
\newrobustcmd{\acutebmvarsigmaz}[2][]{\ensuremath{\subp{\acute{\bm{\varsigma}}}{}{#2}{}{#1}}}
\newrobustcmd{\gravebmvarsigmaz}[2][]{\ensuremath{\subp{\grave{\bm{\varsigma}}}{}{#2}{}{#1}}}
\newrobustcmd{\dotbmvarsigmaz}[2][]{\ensuremath{\subp{\dot{\bm{\varsigma}}}{}{#2}{}{#1}}}
\newrobustcmd{\ddotbmvarsigmaz}[2][]{\ensuremath{\subp{\ddot{\bm{\varsigma}}}{}{#2}{}{#1}}}
\newrobustcmd{\brevebmvarsigmaz}[2][]{\ensuremath{\subp{\breve{\bm{\varsigma}}}{}{#2}{}{#1}}}
\newrobustcmd{\barbmvarsigmaz}[2][]{\ensuremath{\subp{\bar{\bm{\varsigma}}}{}{#2}{}{#1}}}
\newrobustcmd{\vecbmvarsigmaz}[2][]{\ensuremath{\subp{\vec{\bm{\varsigma}}}{}{#2}{}{#1}}}
\newrobustcmd{\tauz}[2][]{\ensuremath{\subp{\tau}{}{#2}{}{#1}}}
\newrobustcmd{\hattauz}[2][]{\ensuremath{\subp{\hat{\tau}}{}{#2}{}{#1}}}
\newrobustcmd{\widehattauz}[2][]{\ensuremath{\subp{\widehat{\tau}}{}{#2}{}{#1}}}
\newrobustcmd{\checktauz}[2][]{\ensuremath{\subp{\check{\tau}}{}{#2}{}{#1}}}
\newrobustcmd{\tildetauz}[2][]{\ensuremath{\subp{\tilde{\tau}}{}{#2}{}{#1}}}
\newrobustcmd{\widetildetauz}[2][]{\ensuremath{\subp{\widetilde{\tau}}{}{#2}{}{#1}}}
\newrobustcmd{\acutetauz}[2][]{\ensuremath{\subp{\acute{\tau}}{}{#2}{}{#1}}}
\newrobustcmd{\gravetauz}[2][]{\ensuremath{\subp{\grave{\tau}}{}{#2}{}{#1}}}
\newrobustcmd{\dottauz}[2][]{\ensuremath{\subp{\dot{\tau}}{}{#2}{}{#1}}}
\newrobustcmd{\ddottauz}[2][]{\ensuremath{\subp{\ddot{\tau}}{}{#2}{}{#1}}}
\newrobustcmd{\brevetauz}[2][]{\ensuremath{\subp{\breve{\tau}}{}{#2}{}{#1}}}
\newrobustcmd{\bartauz}[2][]{\ensuremath{\subp{\bar{\tau}}{}{#2}{}{#1}}}
\newrobustcmd{\vectauz}[2][]{\ensuremath{\subp{\vec{\tau}}{}{#2}{}{#1}}}
\newrobustcmd{\bmtauz}[2][]{\ensuremath{\subp{\bm{\tau}}{}{#2}{}{#1}}}
\newrobustcmd{\hatbmtauz}[2][]{\ensuremath{\subp{\hat{\bm{\tau}}}{}{#2}{}{#1}}}
\newrobustcmd{\widehatbmtauz}[2][]{\ensuremath{\subp{\widehat{\bm{\tau}}}{}{#2}{}{#1}}}
\newrobustcmd{\checkbmtauz}[2][]{\ensuremath{\subp{\check{\bm{\tau}}}{}{#2}{}{#1}}}
\newrobustcmd{\tildebmtauz}[2][]{\ensuremath{\subp{\tilde{\bm{\tau}}}{}{#2}{}{#1}}}
\newrobustcmd{\widetildebmtauz}[2][]{\ensuremath{\subp{\widetilde{\bm{\tau}}}{}{#2}{}{#1}}}
\newrobustcmd{\acutebmtauz}[2][]{\ensuremath{\subp{\acute{\bm{\tau}}}{}{#2}{}{#1}}}
\newrobustcmd{\gravebmtauz}[2][]{\ensuremath{\subp{\grave{\bm{\tau}}}{}{#2}{}{#1}}}
\newrobustcmd{\dotbmtauz}[2][]{\ensuremath{\subp{\dot{\bm{\tau}}}{}{#2}{}{#1}}}
\newrobustcmd{\ddotbmtauz}[2][]{\ensuremath{\subp{\ddot{\bm{\tau}}}{}{#2}{}{#1}}}
\newrobustcmd{\brevebmtauz}[2][]{\ensuremath{\subp{\breve{\bm{\tau}}}{}{#2}{}{#1}}}
\newrobustcmd{\barbmtauz}[2][]{\ensuremath{\subp{\bar{\bm{\tau}}}{}{#2}{}{#1}}}
\newrobustcmd{\vecbmtauz}[2][]{\ensuremath{\subp{\vec{\bm{\tau}}}{}{#2}{}{#1}}}
\newrobustcmd{\upsilonz}[2][]{\ensuremath{\subp{\upsilon}{}{#2}{}{#1}}}
\newrobustcmd{\hatupsilonz}[2][]{\ensuremath{\subp{\hat{\upsilon}}{}{#2}{}{#1}}}
\newrobustcmd{\widehatupsilonz}[2][]{\ensuremath{\subp{\widehat{\upsilon}}{}{#2}{}{#1}}}
\newrobustcmd{\checkupsilonz}[2][]{\ensuremath{\subp{\check{\upsilon}}{}{#2}{}{#1}}}
\newrobustcmd{\tildeupsilonz}[2][]{\ensuremath{\subp{\tilde{\upsilon}}{}{#2}{}{#1}}}
\newrobustcmd{\widetildeupsilonz}[2][]{\ensuremath{\subp{\widetilde{\upsilon}}{}{#2}{}{#1}}}
\newrobustcmd{\acuteupsilonz}[2][]{\ensuremath{\subp{\acute{\upsilon}}{}{#2}{}{#1}}}
\newrobustcmd{\graveupsilonz}[2][]{\ensuremath{\subp{\grave{\upsilon}}{}{#2}{}{#1}}}
\newrobustcmd{\dotupsilonz}[2][]{\ensuremath{\subp{\dot{\upsilon}}{}{#2}{}{#1}}}
\newrobustcmd{\ddotupsilonz}[2][]{\ensuremath{\subp{\ddot{\upsilon}}{}{#2}{}{#1}}}
\newrobustcmd{\breveupsilonz}[2][]{\ensuremath{\subp{\breve{\upsilon}}{}{#2}{}{#1}}}
\newrobustcmd{\barupsilonz}[2][]{\ensuremath{\subp{\bar{\upsilon}}{}{#2}{}{#1}}}
\newrobustcmd{\vecupsilonz}[2][]{\ensuremath{\subp{\vec{\upsilon}}{}{#2}{}{#1}}}
\newrobustcmd{\bmupsilonz}[2][]{\ensuremath{\subp{\bm{\upsilon}}{}{#2}{}{#1}}}
\newrobustcmd{\hatbmupsilonz}[2][]{\ensuremath{\subp{\hat{\bm{\upsilon}}}{}{#2}{}{#1}}}
\newrobustcmd{\widehatbmupsilonz}[2][]{\ensuremath{\subp{\widehat{\bm{\upsilon}}}{}{#2}{}{#1}}}
\newrobustcmd{\checkbmupsilonz}[2][]{\ensuremath{\subp{\check{\bm{\upsilon}}}{}{#2}{}{#1}}}
\newrobustcmd{\tildebmupsilonz}[2][]{\ensuremath{\subp{\tilde{\bm{\upsilon}}}{}{#2}{}{#1}}}
\newrobustcmd{\widetildebmupsilonz}[2][]{\ensuremath{\subp{\widetilde{\bm{\upsilon}}}{}{#2}{}{#1}}}
\newrobustcmd{\acutebmupsilonz}[2][]{\ensuremath{\subp{\acute{\bm{\upsilon}}}{}{#2}{}{#1}}}
\newrobustcmd{\gravebmupsilonz}[2][]{\ensuremath{\subp{\grave{\bm{\upsilon}}}{}{#2}{}{#1}}}
\newrobustcmd{\dotbmupsilonz}[2][]{\ensuremath{\subp{\dot{\bm{\upsilon}}}{}{#2}{}{#1}}}
\newrobustcmd{\ddotbmupsilonz}[2][]{\ensuremath{\subp{\ddot{\bm{\upsilon}}}{}{#2}{}{#1}}}
\newrobustcmd{\brevebmupsilonz}[2][]{\ensuremath{\subp{\breve{\bm{\upsilon}}}{}{#2}{}{#1}}}
\newrobustcmd{\barbmupsilonz}[2][]{\ensuremath{\subp{\bar{\bm{\upsilon}}}{}{#2}{}{#1}}}
\newrobustcmd{\vecbmupsilonz}[2][]{\ensuremath{\subp{\vec{\bm{\upsilon}}}{}{#2}{}{#1}}}
\newrobustcmd{\phiz}[2][]{\ensuremath{\subp{\phi}{}{#2}{}{#1}}}
\newrobustcmd{\hatphiz}[2][]{\ensuremath{\subp{\hat{\phi}}{}{#2}{}{#1}}}
\newrobustcmd{\widehatphiz}[2][]{\ensuremath{\subp{\widehat{\phi}}{}{#2}{}{#1}}}
\newrobustcmd{\checkphiz}[2][]{\ensuremath{\subp{\check{\phi}}{}{#2}{}{#1}}}
\newrobustcmd{\tildephiz}[2][]{\ensuremath{\subp{\tilde{\phi}}{}{#2}{}{#1}}}
\newrobustcmd{\widetildephiz}[2][]{\ensuremath{\subp{\widetilde{\phi}}{}{#2}{}{#1}}}
\newrobustcmd{\acutephiz}[2][]{\ensuremath{\subp{\acute{\phi}}{}{#2}{}{#1}}}
\newrobustcmd{\gravephiz}[2][]{\ensuremath{\subp{\grave{\phi}}{}{#2}{}{#1}}}
\newrobustcmd{\dotphiz}[2][]{\ensuremath{\subp{\dot{\phi}}{}{#2}{}{#1}}}
\newrobustcmd{\ddotphiz}[2][]{\ensuremath{\subp{\ddot{\phi}}{}{#2}{}{#1}}}
\newrobustcmd{\brevephiz}[2][]{\ensuremath{\subp{\breve{\phi}}{}{#2}{}{#1}}}
\newrobustcmd{\barphiz}[2][]{\ensuremath{\subp{\bar{\phi}}{}{#2}{}{#1}}}
\newrobustcmd{\vecphiz}[2][]{\ensuremath{\subp{\vec{\phi}}{}{#2}{}{#1}}}
\newrobustcmd{\bmphiz}[2][]{\ensuremath{\subp{\bm{\phi}}{}{#2}{}{#1}}}
\newrobustcmd{\hatbmphiz}[2][]{\ensuremath{\subp{\hat{\bm{\phi}}}{}{#2}{}{#1}}}
\newrobustcmd{\widehatbmphiz}[2][]{\ensuremath{\subp{\widehat{\bm{\phi}}}{}{#2}{}{#1}}}
\newrobustcmd{\checkbmphiz}[2][]{\ensuremath{\subp{\check{\bm{\phi}}}{}{#2}{}{#1}}}
\newrobustcmd{\tildebmphiz}[2][]{\ensuremath{\subp{\tilde{\bm{\phi}}}{}{#2}{}{#1}}}
\newrobustcmd{\widetildebmphiz}[2][]{\ensuremath{\subp{\widetilde{\bm{\phi}}}{}{#2}{}{#1}}}
\newrobustcmd{\acutebmphiz}[2][]{\ensuremath{\subp{\acute{\bm{\phi}}}{}{#2}{}{#1}}}
\newrobustcmd{\gravebmphiz}[2][]{\ensuremath{\subp{\grave{\bm{\phi}}}{}{#2}{}{#1}}}
\newrobustcmd{\dotbmphiz}[2][]{\ensuremath{\subp{\dot{\bm{\phi}}}{}{#2}{}{#1}}}
\newrobustcmd{\ddotbmphiz}[2][]{\ensuremath{\subp{\ddot{\bm{\phi}}}{}{#2}{}{#1}}}
\newrobustcmd{\brevebmphiz}[2][]{\ensuremath{\subp{\breve{\bm{\phi}}}{}{#2}{}{#1}}}
\newrobustcmd{\barbmphiz}[2][]{\ensuremath{\subp{\bar{\bm{\phi}}}{}{#2}{}{#1}}}
\newrobustcmd{\vecbmphiz}[2][]{\ensuremath{\subp{\vec{\bm{\phi}}}{}{#2}{}{#1}}}
\newrobustcmd{\varphiz}[2][]{\ensuremath{\subp{\varphi}{}{#2}{}{#1}}}
\newrobustcmd{\hatvarphiz}[2][]{\ensuremath{\subp{\hat{\varphi}}{}{#2}{}{#1}}}
\newrobustcmd{\widehatvarphiz}[2][]{\ensuremath{\subp{\widehat{\varphi}}{}{#2}{}{#1}}}
\newrobustcmd{\checkvarphiz}[2][]{\ensuremath{\subp{\check{\varphi}}{}{#2}{}{#1}}}
\newrobustcmd{\tildevarphiz}[2][]{\ensuremath{\subp{\tilde{\varphi}}{}{#2}{}{#1}}}
\newrobustcmd{\widetildevarphiz}[2][]{\ensuremath{\subp{\widetilde{\varphi}}{}{#2}{}{#1}}}
\newrobustcmd{\acutevarphiz}[2][]{\ensuremath{\subp{\acute{\varphi}}{}{#2}{}{#1}}}
\newrobustcmd{\gravevarphiz}[2][]{\ensuremath{\subp{\grave{\varphi}}{}{#2}{}{#1}}}
\newrobustcmd{\dotvarphiz}[2][]{\ensuremath{\subp{\dot{\varphi}}{}{#2}{}{#1}}}
\newrobustcmd{\ddotvarphiz}[2][]{\ensuremath{\subp{\ddot{\varphi}}{}{#2}{}{#1}}}
\newrobustcmd{\brevevarphiz}[2][]{\ensuremath{\subp{\breve{\varphi}}{}{#2}{}{#1}}}
\newrobustcmd{\barvarphiz}[2][]{\ensuremath{\subp{\bar{\varphi}}{}{#2}{}{#1}}}
\newrobustcmd{\vecvarphiz}[2][]{\ensuremath{\subp{\vec{\varphi}}{}{#2}{}{#1}}}
\newrobustcmd{\bmvarphiz}[2][]{\ensuremath{\subp{\bm{\varphi}}{}{#2}{}{#1}}}
\newrobustcmd{\hatbmvarphiz}[2][]{\ensuremath{\subp{\hat{\bm{\varphi}}}{}{#2}{}{#1}}}
\newrobustcmd{\widehatbmvarphiz}[2][]{\ensuremath{\subp{\widehat{\bm{\varphi}}}{}{#2}{}{#1}}}
\newrobustcmd{\checkbmvarphiz}[2][]{\ensuremath{\subp{\check{\bm{\varphi}}}{}{#2}{}{#1}}}
\newrobustcmd{\tildebmvarphiz}[2][]{\ensuremath{\subp{\tilde{\bm{\varphi}}}{}{#2}{}{#1}}}
\newrobustcmd{\widetildebmvarphiz}[2][]{\ensuremath{\subp{\widetilde{\bm{\varphi}}}{}{#2}{}{#1}}}
\newrobustcmd{\acutebmvarphiz}[2][]{\ensuremath{\subp{\acute{\bm{\varphi}}}{}{#2}{}{#1}}}
\newrobustcmd{\gravebmvarphiz}[2][]{\ensuremath{\subp{\grave{\bm{\varphi}}}{}{#2}{}{#1}}}
\newrobustcmd{\dotbmvarphiz}[2][]{\ensuremath{\subp{\dot{\bm{\varphi}}}{}{#2}{}{#1}}}
\newrobustcmd{\ddotbmvarphiz}[2][]{\ensuremath{\subp{\ddot{\bm{\varphi}}}{}{#2}{}{#1}}}
\newrobustcmd{\brevebmvarphiz}[2][]{\ensuremath{\subp{\breve{\bm{\varphi}}}{}{#2}{}{#1}}}
\newrobustcmd{\barbmvarphiz}[2][]{\ensuremath{\subp{\bar{\bm{\varphi}}}{}{#2}{}{#1}}}
\newrobustcmd{\vecbmvarphiz}[2][]{\ensuremath{\subp{\vec{\bm{\varphi}}}{}{#2}{}{#1}}}
\newrobustcmd{\chiz}[2][]{\ensuremath{\subp{\chi}{}{#2}{}{#1}}}
\newrobustcmd{\hatchiz}[2][]{\ensuremath{\subp{\hat{\chi}}{}{#2}{}{#1}}}
\newrobustcmd{\widehatchiz}[2][]{\ensuremath{\subp{\widehat{\chi}}{}{#2}{}{#1}}}
\newrobustcmd{\checkchiz}[2][]{\ensuremath{\subp{\check{\chi}}{}{#2}{}{#1}}}
\newrobustcmd{\tildechiz}[2][]{\ensuremath{\subp{\tilde{\chi}}{}{#2}{}{#1}}}
\newrobustcmd{\widetildechiz}[2][]{\ensuremath{\subp{\widetilde{\chi}}{}{#2}{}{#1}}}
\newrobustcmd{\acutechiz}[2][]{\ensuremath{\subp{\acute{\chi}}{}{#2}{}{#1}}}
\newrobustcmd{\gravechiz}[2][]{\ensuremath{\subp{\grave{\chi}}{}{#2}{}{#1}}}
\newrobustcmd{\dotchiz}[2][]{\ensuremath{\subp{\dot{\chi}}{}{#2}{}{#1}}}
\newrobustcmd{\ddotchiz}[2][]{\ensuremath{\subp{\ddot{\chi}}{}{#2}{}{#1}}}
\newrobustcmd{\brevechiz}[2][]{\ensuremath{\subp{\breve{\chi}}{}{#2}{}{#1}}}
\newrobustcmd{\barchiz}[2][]{\ensuremath{\subp{\bar{\chi}}{}{#2}{}{#1}}}
\newrobustcmd{\vecchiz}[2][]{\ensuremath{\subp{\vec{\chi}}{}{#2}{}{#1}}}
\newrobustcmd{\bmchiz}[2][]{\ensuremath{\subp{\bm{\chi}}{}{#2}{}{#1}}}
\newrobustcmd{\hatbmchiz}[2][]{\ensuremath{\subp{\hat{\bm{\chi}}}{}{#2}{}{#1}}}
\newrobustcmd{\widehatbmchiz}[2][]{\ensuremath{\subp{\widehat{\bm{\chi}}}{}{#2}{}{#1}}}
\newrobustcmd{\checkbmchiz}[2][]{\ensuremath{\subp{\check{\bm{\chi}}}{}{#2}{}{#1}}}
\newrobustcmd{\tildebmchiz}[2][]{\ensuremath{\subp{\tilde{\bm{\chi}}}{}{#2}{}{#1}}}
\newrobustcmd{\widetildebmchiz}[2][]{\ensuremath{\subp{\widetilde{\bm{\chi}}}{}{#2}{}{#1}}}
\newrobustcmd{\acutebmchiz}[2][]{\ensuremath{\subp{\acute{\bm{\chi}}}{}{#2}{}{#1}}}
\newrobustcmd{\gravebmchiz}[2][]{\ensuremath{\subp{\grave{\bm{\chi}}}{}{#2}{}{#1}}}
\newrobustcmd{\dotbmchiz}[2][]{\ensuremath{\subp{\dot{\bm{\chi}}}{}{#2}{}{#1}}}
\newrobustcmd{\ddotbmchiz}[2][]{\ensuremath{\subp{\ddot{\bm{\chi}}}{}{#2}{}{#1}}}
\newrobustcmd{\brevebmchiz}[2][]{\ensuremath{\subp{\breve{\bm{\chi}}}{}{#2}{}{#1}}}
\newrobustcmd{\barbmchiz}[2][]{\ensuremath{\subp{\bar{\bm{\chi}}}{}{#2}{}{#1}}}
\newrobustcmd{\vecbmchiz}[2][]{\ensuremath{\subp{\vec{\bm{\chi}}}{}{#2}{}{#1}}}
\newrobustcmd{\psiz}[2][]{\ensuremath{\subp{\psi}{}{#2}{}{#1}}}
\newrobustcmd{\hatpsiz}[2][]{\ensuremath{\subp{\hat{\psi}}{}{#2}{}{#1}}}
\newrobustcmd{\widehatpsiz}[2][]{\ensuremath{\subp{\widehat{\psi}}{}{#2}{}{#1}}}
\newrobustcmd{\checkpsiz}[2][]{\ensuremath{\subp{\check{\psi}}{}{#2}{}{#1}}}
\newrobustcmd{\tildepsiz}[2][]{\ensuremath{\subp{\tilde{\psi}}{}{#2}{}{#1}}}
\newrobustcmd{\widetildepsiz}[2][]{\ensuremath{\subp{\widetilde{\psi}}{}{#2}{}{#1}}}
\newrobustcmd{\acutepsiz}[2][]{\ensuremath{\subp{\acute{\psi}}{}{#2}{}{#1}}}
\newrobustcmd{\gravepsiz}[2][]{\ensuremath{\subp{\grave{\psi}}{}{#2}{}{#1}}}
\newrobustcmd{\dotpsiz}[2][]{\ensuremath{\subp{\dot{\psi}}{}{#2}{}{#1}}}
\newrobustcmd{\ddotpsiz}[2][]{\ensuremath{\subp{\ddot{\psi}}{}{#2}{}{#1}}}
\newrobustcmd{\brevepsiz}[2][]{\ensuremath{\subp{\breve{\psi}}{}{#2}{}{#1}}}
\newrobustcmd{\barpsiz}[2][]{\ensuremath{\subp{\bar{\psi}}{}{#2}{}{#1}}}
\newrobustcmd{\vecpsiz}[2][]{\ensuremath{\subp{\vec{\psi}}{}{#2}{}{#1}}}
\newrobustcmd{\bmpsiz}[2][]{\ensuremath{\subp{\bm{\psi}}{}{#2}{}{#1}}}
\newrobustcmd{\hatbmpsiz}[2][]{\ensuremath{\subp{\hat{\bm{\psi}}}{}{#2}{}{#1}}}
\newrobustcmd{\widehatbmpsiz}[2][]{\ensuremath{\subp{\widehat{\bm{\psi}}}{}{#2}{}{#1}}}
\newrobustcmd{\checkbmpsiz}[2][]{\ensuremath{\subp{\check{\bm{\psi}}}{}{#2}{}{#1}}}
\newrobustcmd{\tildebmpsiz}[2][]{\ensuremath{\subp{\tilde{\bm{\psi}}}{}{#2}{}{#1}}}
\newrobustcmd{\widetildebmpsiz}[2][]{\ensuremath{\subp{\widetilde{\bm{\psi}}}{}{#2}{}{#1}}}
\newrobustcmd{\acutebmpsiz}[2][]{\ensuremath{\subp{\acute{\bm{\psi}}}{}{#2}{}{#1}}}
\newrobustcmd{\gravebmpsiz}[2][]{\ensuremath{\subp{\grave{\bm{\psi}}}{}{#2}{}{#1}}}
\newrobustcmd{\dotbmpsiz}[2][]{\ensuremath{\subp{\dot{\bm{\psi}}}{}{#2}{}{#1}}}
\newrobustcmd{\ddotbmpsiz}[2][]{\ensuremath{\subp{\ddot{\bm{\psi}}}{}{#2}{}{#1}}}
\newrobustcmd{\brevebmpsiz}[2][]{\ensuremath{\subp{\breve{\bm{\psi}}}{}{#2}{}{#1}}}
\newrobustcmd{\barbmpsiz}[2][]{\ensuremath{\subp{\bar{\bm{\psi}}}{}{#2}{}{#1}}}
\newrobustcmd{\vecbmpsiz}[2][]{\ensuremath{\subp{\vec{\bm{\psi}}}{}{#2}{}{#1}}}
\newrobustcmd{\omegaz}[2][]{\ensuremath{\subp{\omega}{}{#2}{}{#1}}}
\newrobustcmd{\hatomegaz}[2][]{\ensuremath{\subp{\hat{\omega}}{}{#2}{}{#1}}}
\newrobustcmd{\widehatomegaz}[2][]{\ensuremath{\subp{\widehat{\omega}}{}{#2}{}{#1}}}
\newrobustcmd{\checkomegaz}[2][]{\ensuremath{\subp{\check{\omega}}{}{#2}{}{#1}}}
\newrobustcmd{\tildeomegaz}[2][]{\ensuremath{\subp{\tilde{\omega}}{}{#2}{}{#1}}}
\newrobustcmd{\widetildeomegaz}[2][]{\ensuremath{\subp{\widetilde{\omega}}{}{#2}{}{#1}}}
\newrobustcmd{\acuteomegaz}[2][]{\ensuremath{\subp{\acute{\omega}}{}{#2}{}{#1}}}
\newrobustcmd{\graveomegaz}[2][]{\ensuremath{\subp{\grave{\omega}}{}{#2}{}{#1}}}
\newrobustcmd{\dotomegaz}[2][]{\ensuremath{\subp{\dot{\omega}}{}{#2}{}{#1}}}
\newrobustcmd{\ddotomegaz}[2][]{\ensuremath{\subp{\ddot{\omega}}{}{#2}{}{#1}}}
\newrobustcmd{\breveomegaz}[2][]{\ensuremath{\subp{\breve{\omega}}{}{#2}{}{#1}}}
\newrobustcmd{\baromegaz}[2][]{\ensuremath{\subp{\bar{\omega}}{}{#2}{}{#1}}}
\newrobustcmd{\vecomegaz}[2][]{\ensuremath{\subp{\vec{\omega}}{}{#2}{}{#1}}}
\newrobustcmd{\bmomegaz}[2][]{\ensuremath{\subp{\bm{\omega}}{}{#2}{}{#1}}}
\newrobustcmd{\hatbmomegaz}[2][]{\ensuremath{\subp{\hat{\bm{\omega}}}{}{#2}{}{#1}}}
\newrobustcmd{\widehatbmomegaz}[2][]{\ensuremath{\subp{\widehat{\bm{\omega}}}{}{#2}{}{#1}}}
\newrobustcmd{\checkbmomegaz}[2][]{\ensuremath{\subp{\check{\bm{\omega}}}{}{#2}{}{#1}}}
\newrobustcmd{\tildebmomegaz}[2][]{\ensuremath{\subp{\tilde{\bm{\omega}}}{}{#2}{}{#1}}}
\newrobustcmd{\widetildebmomegaz}[2][]{\ensuremath{\subp{\widetilde{\bm{\omega}}}{}{#2}{}{#1}}}
\newrobustcmd{\acutebmomegaz}[2][]{\ensuremath{\subp{\acute{\bm{\omega}}}{}{#2}{}{#1}}}
\newrobustcmd{\gravebmomegaz}[2][]{\ensuremath{\subp{\grave{\bm{\omega}}}{}{#2}{}{#1}}}
\newrobustcmd{\dotbmomegaz}[2][]{\ensuremath{\subp{\dot{\bm{\omega}}}{}{#2}{}{#1}}}
\newrobustcmd{\ddotbmomegaz}[2][]{\ensuremath{\subp{\ddot{\bm{\omega}}}{}{#2}{}{#1}}}
\newrobustcmd{\brevebmomegaz}[2][]{\ensuremath{\subp{\breve{\bm{\omega}}}{}{#2}{}{#1}}}
\newrobustcmd{\barbmomegaz}[2][]{\ensuremath{\subp{\bar{\bm{\omega}}}{}{#2}{}{#1}}}
\newrobustcmd{\vecbmomegaz}[2][]{\ensuremath{\subp{\vec{\bm{\omega}}}{}{#2}{}{#1}}}

\newrobustcmd{\Alphaz}[2][]{\ensuremath{\subp{A}{}{#2}{}{#1}}}
\newrobustcmd{\hatAlphaz}[2][]{\ensuremath{\subp{\hat{A}}{}{#2}{}{#1}}}
\newrobustcmd{\widehatAlphaz}[2][]{\ensuremath{\subp{\widehat{A}}{}{#2}{}{#1}}}
\newrobustcmd{\checkAlphaz}[2][]{\ensuremath{\subp{\check{A}}{}{#2}{}{#1}}}
\newrobustcmd{\tildeAlphaz}[2][]{\ensuremath{\subp{\tilde{A}}{}{#2}{}{#1}}}
\newrobustcmd{\widetildeAlphaz}[2][]{\ensuremath{\subp{\widetilde{A}}{}{#2}{}{#1}}}
\newrobustcmd{\acuteAlphaz}[2][]{\ensuremath{\subp{\acute{A}}{}{#2}{}{#1}}}
\newrobustcmd{\graveAlphaz}[2][]{\ensuremath{\subp{\grave{A}}{}{#2}{}{#1}}}
\newrobustcmd{\dotAlphaz}[2][]{\ensuremath{\subp{\dot{A}}{}{#2}{}{#1}}}
\newrobustcmd{\ddotAlphaz}[2][]{\ensuremath{\subp{\ddot{A}}{}{#2}{}{#1}}}
\newrobustcmd{\breveAlphaz}[2][]{\ensuremath{\subp{\breve{A}}{}{#2}{}{#1}}}
\newrobustcmd{\barAlphaz}[2][]{\ensuremath{\subp{\bar{A}}{}{#2}{}{#1}}}
\newrobustcmd{\vecAlphaz}[2][]{\ensuremath{\subp{\vec{A}}{}{#2}{}{#1}}}
\newrobustcmd{\bmAlphaz}[2][]{\ensuremath{\subp{\bm{A}}{}{#2}{}{#1}}}
\newrobustcmd{\hatbmAlphaz}[2][]{\ensuremath{\subp{\hat{\bm{A}}}{}{#2}{}{#1}}}
\newrobustcmd{\widehatbmAlphaz}[2][]{\ensuremath{\subp{\widehat{\bm{A}}}{}{#2}{}{#1}}}
\newrobustcmd{\checkbmAlphaz}[2][]{\ensuremath{\subp{\check{\bm{A}}}{}{#2}{}{#1}}}
\newrobustcmd{\tildebmAlphaz}[2][]{\ensuremath{\subp{\tilde{\bm{A}}}{}{#2}{}{#1}}}
\newrobustcmd{\widetildebmAlphaz}[2][]{\ensuremath{\subp{\widetilde{\bm{A}}}{}{#2}{}{#1}}}
\newrobustcmd{\acutebmAlphaz}[2][]{\ensuremath{\subp{\acute{\bm{A}}}{}{#2}{}{#1}}}
\newrobustcmd{\gravebmAlphaz}[2][]{\ensuremath{\subp{\grave{\bm{A}}}{}{#2}{}{#1}}}
\newrobustcmd{\dotbmAlphaz}[2][]{\ensuremath{\subp{\dot{\bm{A}}}{}{#2}{}{#1}}}
\newrobustcmd{\ddotbmAlphaz}[2][]{\ensuremath{\subp{\ddot{\bm{A}}}{}{#2}{}{#1}}}
\newrobustcmd{\brevebmAlphaz}[2][]{\ensuremath{\subp{\breve{\bm{A}}}{}{#2}{}{#1}}}
\newrobustcmd{\barbmAlphaz}[2][]{\ensuremath{\subp{\bar{\bm{A}}}{}{#2}{}{#1}}}
\newrobustcmd{\vecbmAlphaz}[2][]{\ensuremath{\subp{\vec{\bm{A}}}{}{#2}{}{#1}}}
\newrobustcmd{\Betaz}[2][]{\ensuremath{\subp{B}{}{#2}{}{#1}}}
\newrobustcmd{\hatBetaz}[2][]{\ensuremath{\subp{\hat{B}}{}{#2}{}{#1}}}
\newrobustcmd{\widehatBetaz}[2][]{\ensuremath{\subp{\widehat{B}}{}{#2}{}{#1}}}
\newrobustcmd{\checkBetaz}[2][]{\ensuremath{\subp{\check{B}}{}{#2}{}{#1}}}
\newrobustcmd{\tildeBetaz}[2][]{\ensuremath{\subp{\tilde{B}}{}{#2}{}{#1}}}
\newrobustcmd{\widetildeBetaz}[2][]{\ensuremath{\subp{\widetilde{B}}{}{#2}{}{#1}}}
\newrobustcmd{\acuteBetaz}[2][]{\ensuremath{\subp{\acute{B}}{}{#2}{}{#1}}}
\newrobustcmd{\graveBetaz}[2][]{\ensuremath{\subp{\grave{B}}{}{#2}{}{#1}}}
\newrobustcmd{\dotBetaz}[2][]{\ensuremath{\subp{\dot{B}}{}{#2}{}{#1}}}
\newrobustcmd{\ddotBetaz}[2][]{\ensuremath{\subp{\ddot{B}}{}{#2}{}{#1}}}
\newrobustcmd{\breveBetaz}[2][]{\ensuremath{\subp{\breve{B}}{}{#2}{}{#1}}}
\newrobustcmd{\barBetaz}[2][]{\ensuremath{\subp{\bar{B}}{}{#2}{}{#1}}}
\newrobustcmd{\vecBetaz}[2][]{\ensuremath{\subp{\vec{B}}{}{#2}{}{#1}}}
\newrobustcmd{\bmBetaz}[2][]{\ensuremath{\subp{\bm{B}}{}{#2}{}{#1}}}
\newrobustcmd{\hatbmBetaz}[2][]{\ensuremath{\subp{\hat{\bm{B}}}{}{#2}{}{#1}}}
\newrobustcmd{\widehatbmBetaz}[2][]{\ensuremath{\subp{\widehat{\bm{B}}}{}{#2}{}{#1}}}
\newrobustcmd{\checkbmBetaz}[2][]{\ensuremath{\subp{\check{\bm{B}}}{}{#2}{}{#1}}}
\newrobustcmd{\tildebmBetaz}[2][]{\ensuremath{\subp{\tilde{\bm{B}}}{}{#2}{}{#1}}}
\newrobustcmd{\widetildebmBetaz}[2][]{\ensuremath{\subp{\widetilde{\bm{B}}}{}{#2}{}{#1}}}
\newrobustcmd{\acutebmBetaz}[2][]{\ensuremath{\subp{\acute{\bm{B}}}{}{#2}{}{#1}}}
\newrobustcmd{\gravebmBetaz}[2][]{\ensuremath{\subp{\grave{\bm{B}}}{}{#2}{}{#1}}}
\newrobustcmd{\dotbmBetaz}[2][]{\ensuremath{\subp{\dot{\bm{B}}}{}{#2}{}{#1}}}
\newrobustcmd{\ddotbmBetaz}[2][]{\ensuremath{\subp{\ddot{\bm{B}}}{}{#2}{}{#1}}}
\newrobustcmd{\brevebmBetaz}[2][]{\ensuremath{\subp{\breve{\bm{B}}}{}{#2}{}{#1}}}
\newrobustcmd{\barbmBetaz}[2][]{\ensuremath{\subp{\bar{\bm{B}}}{}{#2}{}{#1}}}
\newrobustcmd{\vecbmBetaz}[2][]{\ensuremath{\subp{\vec{\bm{B}}}{}{#2}{}{#1}}}
\newrobustcmd{\Gammaz}[2][]{\ensuremath{\subp{\Gamma}{}{#2}{}{#1}}}
\newrobustcmd{\hatGammaz}[2][]{\ensuremath{\subp{\hat{\Gamma}}{}{#2}{}{#1}}}
\newrobustcmd{\widehatGammaz}[2][]{\ensuremath{\subp{\widehat{\Gamma}}{}{#2}{}{#1}}}
\newrobustcmd{\checkGammaz}[2][]{\ensuremath{\subp{\check{\Gamma}}{}{#2}{}{#1}}}
\newrobustcmd{\tildeGammaz}[2][]{\ensuremath{\subp{\tilde{\Gamma}}{}{#2}{}{#1}}}
\newrobustcmd{\widetildeGammaz}[2][]{\ensuremath{\subp{\widetilde{\Gamma}}{}{#2}{}{#1}}}
\newrobustcmd{\acuteGammaz}[2][]{\ensuremath{\subp{\acute{\Gamma}}{}{#2}{}{#1}}}
\newrobustcmd{\graveGammaz}[2][]{\ensuremath{\subp{\grave{\Gamma}}{}{#2}{}{#1}}}
\newrobustcmd{\dotGammaz}[2][]{\ensuremath{\subp{\dot{\Gamma}}{}{#2}{}{#1}}}
\newrobustcmd{\ddotGammaz}[2][]{\ensuremath{\subp{\ddot{\Gamma}}{}{#2}{}{#1}}}
\newrobustcmd{\breveGammaz}[2][]{\ensuremath{\subp{\breve{\Gamma}}{}{#2}{}{#1}}}
\newrobustcmd{\barGammaz}[2][]{\ensuremath{\subp{\bar{\Gamma}}{}{#2}{}{#1}}}
\newrobustcmd{\vecGammaz}[2][]{\ensuremath{\subp{\vec{\Gamma}}{}{#2}{}{#1}}}
\newrobustcmd{\bmGammaz}[2][]{\ensuremath{\subp{\bm{\Gamma}}{}{#2}{}{#1}}}
\newrobustcmd{\hatbmGammaz}[2][]{\ensuremath{\subp{\hat{\bm{\Gamma}}}{}{#2}{}{#1}}}
\newrobustcmd{\widehatbmGammaz}[2][]{\ensuremath{\subp{\widehat{\bm{\Gamma}}}{}{#2}{}{#1}}}
\newrobustcmd{\checkbmGammaz}[2][]{\ensuremath{\subp{\check{\bm{\Gamma}}}{}{#2}{}{#1}}}
\newrobustcmd{\tildebmGammaz}[2][]{\ensuremath{\subp{\tilde{\bm{\Gamma}}}{}{#2}{}{#1}}}
\newrobustcmd{\widetildebmGammaz}[2][]{\ensuremath{\subp{\widetilde{\bm{\Gamma}}}{}{#2}{}{#1}}}
\newrobustcmd{\acutebmGammaz}[2][]{\ensuremath{\subp{\acute{\bm{\Gamma}}}{}{#2}{}{#1}}}
\newrobustcmd{\gravebmGammaz}[2][]{\ensuremath{\subp{\grave{\bm{\Gamma}}}{}{#2}{}{#1}}}
\newrobustcmd{\dotbmGammaz}[2][]{\ensuremath{\subp{\dot{\bm{\Gamma}}}{}{#2}{}{#1}}}
\newrobustcmd{\ddotbmGammaz}[2][]{\ensuremath{\subp{\ddot{\bm{\Gamma}}}{}{#2}{}{#1}}}
\newrobustcmd{\brevebmGammaz}[2][]{\ensuremath{\subp{\breve{\bm{\Gamma}}}{}{#2}{}{#1}}}
\newrobustcmd{\barbmGammaz}[2][]{\ensuremath{\subp{\bar{\bm{\Gamma}}}{}{#2}{}{#1}}}
\newrobustcmd{\vecbmGammaz}[2][]{\ensuremath{\subp{\vec{\bm{\Gamma}}}{}{#2}{}{#1}}}
\newrobustcmd{\Deltaz}[2][]{\ensuremath{\subp{\Delta}{}{#2}{}{#1}}}
\newrobustcmd{\hatDeltaz}[2][]{\ensuremath{\subp{\hat{\Delta}}{}{#2}{}{#1}}}
\newrobustcmd{\widehatDeltaz}[2][]{\ensuremath{\subp{\widehat{\Delta}}{}{#2}{}{#1}}}
\newrobustcmd{\checkDeltaz}[2][]{\ensuremath{\subp{\check{\Delta}}{}{#2}{}{#1}}}
\newrobustcmd{\tildeDeltaz}[2][]{\ensuremath{\subp{\tilde{\Delta}}{}{#2}{}{#1}}}
\newrobustcmd{\widetildeDeltaz}[2][]{\ensuremath{\subp{\widetilde{\Delta}}{}{#2}{}{#1}}}
\newrobustcmd{\acuteDeltaz}[2][]{\ensuremath{\subp{\acute{\Delta}}{}{#2}{}{#1}}}
\newrobustcmd{\graveDeltaz}[2][]{\ensuremath{\subp{\grave{\Delta}}{}{#2}{}{#1}}}
\newrobustcmd{\dotDeltaz}[2][]{\ensuremath{\subp{\dot{\Delta}}{}{#2}{}{#1}}}
\newrobustcmd{\ddotDeltaz}[2][]{\ensuremath{\subp{\ddot{\Delta}}{}{#2}{}{#1}}}
\newrobustcmd{\breveDeltaz}[2][]{\ensuremath{\subp{\breve{\Delta}}{}{#2}{}{#1}}}
\newrobustcmd{\barDeltaz}[2][]{\ensuremath{\subp{\bar{\Delta}}{}{#2}{}{#1}}}
\newrobustcmd{\vecDeltaz}[2][]{\ensuremath{\subp{\vec{\Delta}}{}{#2}{}{#1}}}
\newrobustcmd{\bmDeltaz}[2][]{\ensuremath{\subp{\bm{\Delta}}{}{#2}{}{#1}}}
\newrobustcmd{\hatbmDeltaz}[2][]{\ensuremath{\subp{\hat{\bm{\Delta}}}{}{#2}{}{#1}}}
\newrobustcmd{\widehatbmDeltaz}[2][]{\ensuremath{\subp{\widehat{\bm{\Delta}}}{}{#2}{}{#1}}}
\newrobustcmd{\checkbmDeltaz}[2][]{\ensuremath{\subp{\check{\bm{\Delta}}}{}{#2}{}{#1}}}
\newrobustcmd{\tildebmDeltaz}[2][]{\ensuremath{\subp{\tilde{\bm{\Delta}}}{}{#2}{}{#1}}}
\newrobustcmd{\widetildebmDeltaz}[2][]{\ensuremath{\subp{\widetilde{\bm{\Delta}}}{}{#2}{}{#1}}}
\newrobustcmd{\acutebmDeltaz}[2][]{\ensuremath{\subp{\acute{\bm{\Delta}}}{}{#2}{}{#1}}}
\newrobustcmd{\gravebmDeltaz}[2][]{\ensuremath{\subp{\grave{\bm{\Delta}}}{}{#2}{}{#1}}}
\newrobustcmd{\dotbmDeltaz}[2][]{\ensuremath{\subp{\dot{\bm{\Delta}}}{}{#2}{}{#1}}}
\newrobustcmd{\ddotbmDeltaz}[2][]{\ensuremath{\subp{\ddot{\bm{\Delta}}}{}{#2}{}{#1}}}
\newrobustcmd{\brevebmDeltaz}[2][]{\ensuremath{\subp{\breve{\bm{\Delta}}}{}{#2}{}{#1}}}
\newrobustcmd{\barbmDeltaz}[2][]{\ensuremath{\subp{\bar{\bm{\Delta}}}{}{#2}{}{#1}}}
\newrobustcmd{\vecbmDeltaz}[2][]{\ensuremath{\subp{\vec{\bm{\Delta}}}{}{#2}{}{#1}}}
\newrobustcmd{\Epsilonz}[2][]{\ensuremath{\subp{\Epsilon}{}{#2}{}{#1}}}
\newrobustcmd{\hatEpsilonz}[2][]{\ensuremath{\subp{\hat{\Epsilon}}{}{#2}{}{#1}}}
\newrobustcmd{\widehatEpsilonz}[2][]{\ensuremath{\subp{\widehat{\Epsilon}}{}{#2}{}{#1}}}
\newrobustcmd{\checkEpsilonz}[2][]{\ensuremath{\subp{\check{\Epsilon}}{}{#2}{}{#1}}}
\newrobustcmd{\tildeEpsilonz}[2][]{\ensuremath{\subp{\tilde{\Epsilon}}{}{#2}{}{#1}}}
\newrobustcmd{\widetildeEpsilonz}[2][]{\ensuremath{\subp{\widetilde{\Epsilon}}{}{#2}{}{#1}}}
\newrobustcmd{\acuteEpsilonz}[2][]{\ensuremath{\subp{\acute{\Epsilon}}{}{#2}{}{#1}}}
\newrobustcmd{\graveEpsilonz}[2][]{\ensuremath{\subp{\grave{\Epsilon}}{}{#2}{}{#1}}}
\newrobustcmd{\dotEpsilonz}[2][]{\ensuremath{\subp{\dot{\Epsilon}}{}{#2}{}{#1}}}
\newrobustcmd{\ddotEpsilonz}[2][]{\ensuremath{\subp{\ddot{\Epsilon}}{}{#2}{}{#1}}}
\newrobustcmd{\breveEpsilonz}[2][]{\ensuremath{\subp{\breve{\Epsilon}}{}{#2}{}{#1}}}
\newrobustcmd{\barEpsilonz}[2][]{\ensuremath{\subp{\bar{\Epsilon}}{}{#2}{}{#1}}}
\newrobustcmd{\vecEpsilonz}[2][]{\ensuremath{\subp{\vec{\Epsilon}}{}{#2}{}{#1}}}
\newrobustcmd{\bmEpsilonz}[2][]{\ensuremath{\subp{\bm{\Epsilon}}{}{#2}{}{#1}}}
\newrobustcmd{\hatbmEpsilonz}[2][]{\ensuremath{\subp{\hat{\bm{\Epsilon}}}{}{#2}{}{#1}}}
\newrobustcmd{\widehatbmEpsilonz}[2][]{\ensuremath{\subp{\widehat{\bm{\Epsilon}}}{}{#2}{}{#1}}}
\newrobustcmd{\checkbmEpsilonz}[2][]{\ensuremath{\subp{\check{\bm{\Epsilon}}}{}{#2}{}{#1}}}
\newrobustcmd{\tildebmEpsilonz}[2][]{\ensuremath{\subp{\tilde{\bm{\Epsilon}}}{}{#2}{}{#1}}}
\newrobustcmd{\widetildebmEpsilonz}[2][]{\ensuremath{\subp{\widetilde{\bm{\Epsilon}}}{}{#2}{}{#1}}}
\newrobustcmd{\acutebmEpsilonz}[2][]{\ensuremath{\subp{\acute{\bm{\Epsilon}}}{}{#2}{}{#1}}}
\newrobustcmd{\gravebmEpsilonz}[2][]{\ensuremath{\subp{\grave{\bm{\Epsilon}}}{}{#2}{}{#1}}}
\newrobustcmd{\dotbmEpsilonz}[2][]{\ensuremath{\subp{\dot{\bm{\Epsilon}}}{}{#2}{}{#1}}}
\newrobustcmd{\ddotbmEpsilonz}[2][]{\ensuremath{\subp{\ddot{\bm{\Epsilon}}}{}{#2}{}{#1}}}
\newrobustcmd{\brevebmEpsilonz}[2][]{\ensuremath{\subp{\breve{\bm{\Epsilon}}}{}{#2}{}{#1}}}
\newrobustcmd{\barbmEpsilonz}[2][]{\ensuremath{\subp{\bar{\bm{\Epsilon}}}{}{#2}{}{#1}}}
\newrobustcmd{\vecbmEpsilonz}[2][]{\ensuremath{\subp{\vec{\bm{\Epsilon}}}{}{#2}{}{#1}}}
\newrobustcmd{\Zetaz}[2][]{\ensuremath{\subp{Z}{}{#2}{}{#1}}}
\newrobustcmd{\hatZetaz}[2][]{\ensuremath{\subp{\hat{Z}}{}{#2}{}{#1}}}
\newrobustcmd{\widehatZetaz}[2][]{\ensuremath{\subp{\widehat{Z}}{}{#2}{}{#1}}}
\newrobustcmd{\checkZetaz}[2][]{\ensuremath{\subp{\check{Z}}{}{#2}{}{#1}}}
\newrobustcmd{\tildeZetaz}[2][]{\ensuremath{\subp{\tilde{Z}}{}{#2}{}{#1}}}
\newrobustcmd{\widetildeZetaz}[2][]{\ensuremath{\subp{\widetilde{Z}}{}{#2}{}{#1}}}
\newrobustcmd{\acuteZetaz}[2][]{\ensuremath{\subp{\acute{Z}}{}{#2}{}{#1}}}
\newrobustcmd{\graveZetaz}[2][]{\ensuremath{\subp{\grave{Z}}{}{#2}{}{#1}}}
\newrobustcmd{\dotZetaz}[2][]{\ensuremath{\subp{\dot{Z}}{}{#2}{}{#1}}}
\newrobustcmd{\ddotZetaz}[2][]{\ensuremath{\subp{\ddot{Z}}{}{#2}{}{#1}}}
\newrobustcmd{\breveZetaz}[2][]{\ensuremath{\subp{\breve{Z}}{}{#2}{}{#1}}}
\newrobustcmd{\barZetaz}[2][]{\ensuremath{\subp{\bar{Z}}{}{#2}{}{#1}}}
\newrobustcmd{\vecZetaz}[2][]{\ensuremath{\subp{\vec{Z}}{}{#2}{}{#1}}}
\newrobustcmd{\bmZetaz}[2][]{\ensuremath{\subp{\bm{Z}}{}{#2}{}{#1}}}
\newrobustcmd{\hatbmZetaz}[2][]{\ensuremath{\subp{\hat{\bm{Z}}}{}{#2}{}{#1}}}
\newrobustcmd{\widehatbmZetaz}[2][]{\ensuremath{\subp{\widehat{\bm{Z}}}{}{#2}{}{#1}}}
\newrobustcmd{\checkbmZetaz}[2][]{\ensuremath{\subp{\check{\bm{Z}}}{}{#2}{}{#1}}}
\newrobustcmd{\tildebmZetaz}[2][]{\ensuremath{\subp{\tilde{\bm{Z}}}{}{#2}{}{#1}}}
\newrobustcmd{\widetildebmZetaz}[2][]{\ensuremath{\subp{\widetilde{\bm{Z}}}{}{#2}{}{#1}}}
\newrobustcmd{\acutebmZetaz}[2][]{\ensuremath{\subp{\acute{\bm{Z}}}{}{#2}{}{#1}}}
\newrobustcmd{\gravebmZetaz}[2][]{\ensuremath{\subp{\grave{\bm{Z}}}{}{#2}{}{#1}}}
\newrobustcmd{\dotbmZetaz}[2][]{\ensuremath{\subp{\dot{\bm{Z}}}{}{#2}{}{#1}}}
\newrobustcmd{\ddotbmZetaz}[2][]{\ensuremath{\subp{\ddot{\bm{Z}}}{}{#2}{}{#1}}}
\newrobustcmd{\brevebmZetaz}[2][]{\ensuremath{\subp{\breve{\bm{Z}}}{}{#2}{}{#1}}}
\newrobustcmd{\barbmZetaz}[2][]{\ensuremath{\subp{\bar{\bm{Z}}}{}{#2}{}{#1}}}
\newrobustcmd{\vecbmZetaz}[2][]{\ensuremath{\subp{\vec{\bm{Z}}}{}{#2}{}{#1}}}
\newrobustcmd{\Etaz}[2][]{\ensuremath{\subp{H}{}{#2}{}{#1}}}
\newrobustcmd{\hatEtaz}[2][]{\ensuremath{\subp{\hat{H}}{}{#2}{}{#1}}}
\newrobustcmd{\widehatEtaz}[2][]{\ensuremath{\subp{\widehat{H}}{}{#2}{}{#1}}}
\newrobustcmd{\checkEtaz}[2][]{\ensuremath{\subp{\check{H}}{}{#2}{}{#1}}}
\newrobustcmd{\tildeEtaz}[2][]{\ensuremath{\subp{\tilde{H}}{}{#2}{}{#1}}}
\newrobustcmd{\widetildeEtaz}[2][]{\ensuremath{\subp{\widetilde{H}}{}{#2}{}{#1}}}
\newrobustcmd{\acuteEtaz}[2][]{\ensuremath{\subp{\acute{H}}{}{#2}{}{#1}}}
\newrobustcmd{\graveEtaz}[2][]{\ensuremath{\subp{\grave{H}}{}{#2}{}{#1}}}
\newrobustcmd{\dotEtaz}[2][]{\ensuremath{\subp{\dot{H}}{}{#2}{}{#1}}}
\newrobustcmd{\ddotEtaz}[2][]{\ensuremath{\subp{\ddot{H}}{}{#2}{}{#1}}}
\newrobustcmd{\breveEtaz}[2][]{\ensuremath{\subp{\breve{H}}{}{#2}{}{#1}}}
\newrobustcmd{\barEtaz}[2][]{\ensuremath{\subp{\bar{H}}{}{#2}{}{#1}}}
\newrobustcmd{\vecEtaz}[2][]{\ensuremath{\subp{\vec{H}}{}{#2}{}{#1}}}
\newrobustcmd{\bmEtaz}[2][]{\ensuremath{\subp{\bm{H}}{}{#2}{}{#1}}}
\newrobustcmd{\hatbmEtaz}[2][]{\ensuremath{\subp{\hat{\bm{H}}}{}{#2}{}{#1}}}
\newrobustcmd{\widehatbmEtaz}[2][]{\ensuremath{\subp{\widehat{\bm{H}}}{}{#2}{}{#1}}}
\newrobustcmd{\checkbmEtaz}[2][]{\ensuremath{\subp{\check{\bm{H}}}{}{#2}{}{#1}}}
\newrobustcmd{\tildebmEtaz}[2][]{\ensuremath{\subp{\tilde{\bm{H}}}{}{#2}{}{#1}}}
\newrobustcmd{\widetildebmEtaz}[2][]{\ensuremath{\subp{\widetilde{\bm{H}}}{}{#2}{}{#1}}}
\newrobustcmd{\acutebmEtaz}[2][]{\ensuremath{\subp{\acute{\bm{H}}}{}{#2}{}{#1}}}
\newrobustcmd{\gravebmEtaz}[2][]{\ensuremath{\subp{\grave{\bm{H}}}{}{#2}{}{#1}}}
\newrobustcmd{\dotbmEtaz}[2][]{\ensuremath{\subp{\dot{\bm{H}}}{}{#2}{}{#1}}}
\newrobustcmd{\ddotbmEtaz}[2][]{\ensuremath{\subp{\ddot{\bm{H}}}{}{#2}{}{#1}}}
\newrobustcmd{\brevebmEtaz}[2][]{\ensuremath{\subp{\breve{\bm{H}}}{}{#2}{}{#1}}}
\newrobustcmd{\barbmEtaz}[2][]{\ensuremath{\subp{\bar{\bm{H}}}{}{#2}{}{#1}}}
\newrobustcmd{\vecbmEtaz}[2][]{\ensuremath{\subp{\vec{\bm{H}}}{}{#2}{}{#1}}}
\newrobustcmd{\Thetaz}[2][]{\ensuremath{\subp{\Theta}{}{#2}{}{#1}}}
\newrobustcmd{\hatThetaz}[2][]{\ensuremath{\subp{\hat{\Theta}}{}{#2}{}{#1}}}
\newrobustcmd{\widehatThetaz}[2][]{\ensuremath{\subp{\widehat{\Theta}}{}{#2}{}{#1}}}
\newrobustcmd{\checkThetaz}[2][]{\ensuremath{\subp{\check{\Theta}}{}{#2}{}{#1}}}
\newrobustcmd{\tildeThetaz}[2][]{\ensuremath{\subp{\tilde{\Theta}}{}{#2}{}{#1}}}
\newrobustcmd{\widetildeThetaz}[2][]{\ensuremath{\subp{\widetilde{\Theta}}{}{#2}{}{#1}}}
\newrobustcmd{\acuteThetaz}[2][]{\ensuremath{\subp{\acute{\Theta}}{}{#2}{}{#1}}}
\newrobustcmd{\graveThetaz}[2][]{\ensuremath{\subp{\grave{\Theta}}{}{#2}{}{#1}}}
\newrobustcmd{\dotThetaz}[2][]{\ensuremath{\subp{\dot{\Theta}}{}{#2}{}{#1}}}
\newrobustcmd{\ddotThetaz}[2][]{\ensuremath{\subp{\ddot{\Theta}}{}{#2}{}{#1}}}
\newrobustcmd{\breveThetaz}[2][]{\ensuremath{\subp{\breve{\Theta}}{}{#2}{}{#1}}}
\newrobustcmd{\barThetaz}[2][]{\ensuremath{\subp{\bar{\Theta}}{}{#2}{}{#1}}}
\newrobustcmd{\vecThetaz}[2][]{\ensuremath{\subp{\vec{\Theta}}{}{#2}{}{#1}}}
\newrobustcmd{\bmThetaz}[2][]{\ensuremath{\subp{\bm{\Theta}}{}{#2}{}{#1}}}
\newrobustcmd{\hatbmThetaz}[2][]{\ensuremath{\subp{\hat{\bm{\Theta}}}{}{#2}{}{#1}}}
\newrobustcmd{\widehatbmThetaz}[2][]{\ensuremath{\subp{\widehat{\bm{\Theta}}}{}{#2}{}{#1}}}
\newrobustcmd{\checkbmThetaz}[2][]{\ensuremath{\subp{\check{\bm{\Theta}}}{}{#2}{}{#1}}}
\newrobustcmd{\tildebmThetaz}[2][]{\ensuremath{\subp{\tilde{\bm{\Theta}}}{}{#2}{}{#1}}}
\newrobustcmd{\widetildebmThetaz}[2][]{\ensuremath{\subp{\widetilde{\bm{\Theta}}}{}{#2}{}{#1}}}
\newrobustcmd{\acutebmThetaz}[2][]{\ensuremath{\subp{\acute{\bm{\Theta}}}{}{#2}{}{#1}}}
\newrobustcmd{\gravebmThetaz}[2][]{\ensuremath{\subp{\grave{\bm{\Theta}}}{}{#2}{}{#1}}}
\newrobustcmd{\dotbmThetaz}[2][]{\ensuremath{\subp{\dot{\bm{\Theta}}}{}{#2}{}{#1}}}
\newrobustcmd{\ddotbmThetaz}[2][]{\ensuremath{\subp{\ddot{\bm{\Theta}}}{}{#2}{}{#1}}}
\newrobustcmd{\brevebmThetaz}[2][]{\ensuremath{\subp{\breve{\bm{\Theta}}}{}{#2}{}{#1}}}
\newrobustcmd{\barbmThetaz}[2][]{\ensuremath{\subp{\bar{\bm{\Theta}}}{}{#2}{}{#1}}}
\newrobustcmd{\vecbmThetaz}[2][]{\ensuremath{\subp{\vec{\bm{\Theta}}}{}{#2}{}{#1}}}
\newrobustcmd{\Iotaz}[2][]{\ensuremath{\subp{I}{}{#2}{}{#1}}}
\newrobustcmd{\hatIotaz}[2][]{\ensuremath{\subp{\hat{I}}{}{#2}{}{#1}}}
\newrobustcmd{\widehatIotaz}[2][]{\ensuremath{\subp{\widehat{I}}{}{#2}{}{#1}}}
\newrobustcmd{\checkIotaz}[2][]{\ensuremath{\subp{\check{I}}{}{#2}{}{#1}}}
\newrobustcmd{\tildeIotaz}[2][]{\ensuremath{\subp{\tilde{I}}{}{#2}{}{#1}}}
\newrobustcmd{\widetildeIotaz}[2][]{\ensuremath{\subp{\widetilde{I}}{}{#2}{}{#1}}}
\newrobustcmd{\acuteIotaz}[2][]{\ensuremath{\subp{\acute{I}}{}{#2}{}{#1}}}
\newrobustcmd{\graveIotaz}[2][]{\ensuremath{\subp{\grave{I}}{}{#2}{}{#1}}}
\newrobustcmd{\dotIotaz}[2][]{\ensuremath{\subp{\dot{I}}{}{#2}{}{#1}}}
\newrobustcmd{\ddotIotaz}[2][]{\ensuremath{\subp{\ddot{I}}{}{#2}{}{#1}}}
\newrobustcmd{\breveIotaz}[2][]{\ensuremath{\subp{\breve{I}}{}{#2}{}{#1}}}
\newrobustcmd{\barIotaz}[2][]{\ensuremath{\subp{\bar{I}}{}{#2}{}{#1}}}
\newrobustcmd{\vecIotaz}[2][]{\ensuremath{\subp{\vec{I}}{}{#2}{}{#1}}}
\newrobustcmd{\bmIotaz}[2][]{\ensuremath{\subp{\bm{I}}{}{#2}{}{#1}}}
\newrobustcmd{\hatbmIotaz}[2][]{\ensuremath{\subp{\hat{\bm{I}}}{}{#2}{}{#1}}}
\newrobustcmd{\widehatbmIotaz}[2][]{\ensuremath{\subp{\widehat{\bm{I}}}{}{#2}{}{#1}}}
\newrobustcmd{\checkbmIotaz}[2][]{\ensuremath{\subp{\check{\bm{I}}}{}{#2}{}{#1}}}
\newrobustcmd{\tildebmIotaz}[2][]{\ensuremath{\subp{\tilde{\bm{I}}}{}{#2}{}{#1}}}
\newrobustcmd{\widetildebmIotaz}[2][]{\ensuremath{\subp{\widetilde{\bm{I}}}{}{#2}{}{#1}}}
\newrobustcmd{\acutebmIotaz}[2][]{\ensuremath{\subp{\acute{\bm{I}}}{}{#2}{}{#1}}}
\newrobustcmd{\gravebmIotaz}[2][]{\ensuremath{\subp{\grave{\bm{I}}}{}{#2}{}{#1}}}
\newrobustcmd{\dotbmIotaz}[2][]{\ensuremath{\subp{\dot{\bm{I}}}{}{#2}{}{#1}}}
\newrobustcmd{\ddotbmIotaz}[2][]{\ensuremath{\subp{\ddot{\bm{I}}}{}{#2}{}{#1}}}
\newrobustcmd{\brevebmIotaz}[2][]{\ensuremath{\subp{\breve{\bm{I}}}{}{#2}{}{#1}}}
\newrobustcmd{\barbmIotaz}[2][]{\ensuremath{\subp{\bar{\bm{I}}}{}{#2}{}{#1}}}
\newrobustcmd{\vecbmIotaz}[2][]{\ensuremath{\subp{\vec{\bm{I}}}{}{#2}{}{#1}}}
\newrobustcmd{\Kappaz}[2][]{\ensuremath{\subp{K}{}{#2}{}{#1}}}
\newrobustcmd{\hatKappaz}[2][]{\ensuremath{\subp{\hat{K}}{}{#2}{}{#1}}}
\newrobustcmd{\widehatKappaz}[2][]{\ensuremath{\subp{\widehat{K}}{}{#2}{}{#1}}}
\newrobustcmd{\checkKappaz}[2][]{\ensuremath{\subp{\check{K}}{}{#2}{}{#1}}}
\newrobustcmd{\tildeKappaz}[2][]{\ensuremath{\subp{\tilde{K}}{}{#2}{}{#1}}}
\newrobustcmd{\widetildeKappaz}[2][]{\ensuremath{\subp{\widetilde{K}}{}{#2}{}{#1}}}
\newrobustcmd{\acuteKappaz}[2][]{\ensuremath{\subp{\acute{K}}{}{#2}{}{#1}}}
\newrobustcmd{\graveKappaz}[2][]{\ensuremath{\subp{\grave{K}}{}{#2}{}{#1}}}
\newrobustcmd{\dotKappaz}[2][]{\ensuremath{\subp{\dot{K}}{}{#2}{}{#1}}}
\newrobustcmd{\ddotKappaz}[2][]{\ensuremath{\subp{\ddot{K}}{}{#2}{}{#1}}}
\newrobustcmd{\breveKappaz}[2][]{\ensuremath{\subp{\breve{K}}{}{#2}{}{#1}}}
\newrobustcmd{\barKappaz}[2][]{\ensuremath{\subp{\bar{K}}{}{#2}{}{#1}}}
\newrobustcmd{\vecKappaz}[2][]{\ensuremath{\subp{\vec{K}}{}{#2}{}{#1}}}
\newrobustcmd{\bmKappaz}[2][]{\ensuremath{\subp{\bm{K}}{}{#2}{}{#1}}}
\newrobustcmd{\hatbmKappaz}[2][]{\ensuremath{\subp{\hat{\bm{K}}}{}{#2}{}{#1}}}
\newrobustcmd{\widehatbmKappaz}[2][]{\ensuremath{\subp{\widehat{\bm{K}}}{}{#2}{}{#1}}}
\newrobustcmd{\checkbmKappaz}[2][]{\ensuremath{\subp{\check{\bm{K}}}{}{#2}{}{#1}}}
\newrobustcmd{\tildebmKappaz}[2][]{\ensuremath{\subp{\tilde{\bm{K}}}{}{#2}{}{#1}}}
\newrobustcmd{\widetildebmKappaz}[2][]{\ensuremath{\subp{\widetilde{\bm{K}}}{}{#2}{}{#1}}}
\newrobustcmd{\acutebmKappaz}[2][]{\ensuremath{\subp{\acute{\bm{K}}}{}{#2}{}{#1}}}
\newrobustcmd{\gravebmKappaz}[2][]{\ensuremath{\subp{\grave{\bm{K}}}{}{#2}{}{#1}}}
\newrobustcmd{\dotbmKappaz}[2][]{\ensuremath{\subp{\dot{\bm{K}}}{}{#2}{}{#1}}}
\newrobustcmd{\ddotbmKappaz}[2][]{\ensuremath{\subp{\ddot{\bm{K}}}{}{#2}{}{#1}}}
\newrobustcmd{\brevebmKappaz}[2][]{\ensuremath{\subp{\breve{\bm{K}}}{}{#2}{}{#1}}}
\newrobustcmd{\barbmKappaz}[2][]{\ensuremath{\subp{\bar{\bm{K}}}{}{#2}{}{#1}}}
\newrobustcmd{\vecbmKappaz}[2][]{\ensuremath{\subp{\vec{\bm{K}}}{}{#2}{}{#1}}}
\newrobustcmd{\Lambdaz}[2][]{\ensuremath{\subp{\Lambda}{}{#2}{}{#1}}}
\newrobustcmd{\hatLambdaz}[2][]{\ensuremath{\subp{\hat{\Lambda}}{}{#2}{}{#1}}}
\newrobustcmd{\widehatLambdaz}[2][]{\ensuremath{\subp{\widehat{\Lambda}}{}{#2}{}{#1}}}
\newrobustcmd{\checkLambdaz}[2][]{\ensuremath{\subp{\check{\Lambda}}{}{#2}{}{#1}}}
\newrobustcmd{\tildeLambdaz}[2][]{\ensuremath{\subp{\tilde{\Lambda}}{}{#2}{}{#1}}}
\newrobustcmd{\widetildeLambdaz}[2][]{\ensuremath{\subp{\widetilde{\Lambda}}{}{#2}{}{#1}}}
\newrobustcmd{\acuteLambdaz}[2][]{\ensuremath{\subp{\acute{\Lambda}}{}{#2}{}{#1}}}
\newrobustcmd{\graveLambdaz}[2][]{\ensuremath{\subp{\grave{\Lambda}}{}{#2}{}{#1}}}
\newrobustcmd{\dotLambdaz}[2][]{\ensuremath{\subp{\dot{\Lambda}}{}{#2}{}{#1}}}
\newrobustcmd{\ddotLambdaz}[2][]{\ensuremath{\subp{\ddot{\Lambda}}{}{#2}{}{#1}}}
\newrobustcmd{\breveLambdaz}[2][]{\ensuremath{\subp{\breve{\Lambda}}{}{#2}{}{#1}}}
\newrobustcmd{\barLambdaz}[2][]{\ensuremath{\subp{\bar{\Lambda}}{}{#2}{}{#1}}}
\newrobustcmd{\vecLambdaz}[2][]{\ensuremath{\subp{\vec{\Lambda}}{}{#2}{}{#1}}}
\newrobustcmd{\bmLambdaz}[2][]{\ensuremath{\subp{\bm{\Lambda}}{}{#2}{}{#1}}}
\newrobustcmd{\hatbmLambdaz}[2][]{\ensuremath{\subp{\hat{\bm{\Lambda}}}{}{#2}{}{#1}}}
\newrobustcmd{\widehatbmLambdaz}[2][]{\ensuremath{\subp{\widehat{\bm{\Lambda}}}{}{#2}{}{#1}}}
\newrobustcmd{\checkbmLambdaz}[2][]{\ensuremath{\subp{\check{\bm{\Lambda}}}{}{#2}{}{#1}}}
\newrobustcmd{\tildebmLambdaz}[2][]{\ensuremath{\subp{\tilde{\bm{\Lambda}}}{}{#2}{}{#1}}}
\newrobustcmd{\widetildebmLambdaz}[2][]{\ensuremath{\subp{\widetilde{\bm{\Lambda}}}{}{#2}{}{#1}}}
\newrobustcmd{\acutebmLambdaz}[2][]{\ensuremath{\subp{\acute{\bm{\Lambda}}}{}{#2}{}{#1}}}
\newrobustcmd{\gravebmLambdaz}[2][]{\ensuremath{\subp{\grave{\bm{\Lambda}}}{}{#2}{}{#1}}}
\newrobustcmd{\dotbmLambdaz}[2][]{\ensuremath{\subp{\dot{\bm{\Lambda}}}{}{#2}{}{#1}}}
\newrobustcmd{\ddotbmLambdaz}[2][]{\ensuremath{\subp{\ddot{\bm{\Lambda}}}{}{#2}{}{#1}}}
\newrobustcmd{\brevebmLambdaz}[2][]{\ensuremath{\subp{\breve{\bm{\Lambda}}}{}{#2}{}{#1}}}
\newrobustcmd{\barbmLambdaz}[2][]{\ensuremath{\subp{\bar{\bm{\Lambda}}}{}{#2}{}{#1}}}
\newrobustcmd{\vecbmLambdaz}[2][]{\ensuremath{\subp{\vec{\bm{\Lambda}}}{}{#2}{}{#1}}}
\newrobustcmd{\Muz}[2][]{\ensuremath{\subp{M}{}{#2}{}{#1}}}
\newrobustcmd{\hatMuz}[2][]{\ensuremath{\subp{\hat{M}}{}{#2}{}{#1}}}
\newrobustcmd{\widehatMuz}[2][]{\ensuremath{\subp{\widehat{M}}{}{#2}{}{#1}}}
\newrobustcmd{\checkMuz}[2][]{\ensuremath{\subp{\check{M}}{}{#2}{}{#1}}}
\newrobustcmd{\tildeMuz}[2][]{\ensuremath{\subp{\tilde{M}}{}{#2}{}{#1}}}
\newrobustcmd{\widetildeMuz}[2][]{\ensuremath{\subp{\widetilde{M}}{}{#2}{}{#1}}}
\newrobustcmd{\acuteMuz}[2][]{\ensuremath{\subp{\acute{M}}{}{#2}{}{#1}}}
\newrobustcmd{\graveMuz}[2][]{\ensuremath{\subp{\grave{M}}{}{#2}{}{#1}}}
\newrobustcmd{\dotMuz}[2][]{\ensuremath{\subp{\dot{M}}{}{#2}{}{#1}}}
\newrobustcmd{\ddotMuz}[2][]{\ensuremath{\subp{\ddot{M}}{}{#2}{}{#1}}}
\newrobustcmd{\breveMuz}[2][]{\ensuremath{\subp{\breve{M}}{}{#2}{}{#1}}}
\newrobustcmd{\barMuz}[2][]{\ensuremath{\subp{\bar{M}}{}{#2}{}{#1}}}
\newrobustcmd{\vecMuz}[2][]{\ensuremath{\subp{\vec{M}}{}{#2}{}{#1}}}
\newrobustcmd{\bmMuz}[2][]{\ensuremath{\subp{\bm{M}}{}{#2}{}{#1}}}
\newrobustcmd{\hatbmMuz}[2][]{\ensuremath{\subp{\hat{\bm{M}}}{}{#2}{}{#1}}}
\newrobustcmd{\widehatbmMuz}[2][]{\ensuremath{\subp{\widehat{\bm{M}}}{}{#2}{}{#1}}}
\newrobustcmd{\checkbmMuz}[2][]{\ensuremath{\subp{\check{\bm{M}}}{}{#2}{}{#1}}}
\newrobustcmd{\tildebmMuz}[2][]{\ensuremath{\subp{\tilde{\bm{M}}}{}{#2}{}{#1}}}
\newrobustcmd{\widetildebmMuz}[2][]{\ensuremath{\subp{\widetilde{\bm{M}}}{}{#2}{}{#1}}}
\newrobustcmd{\acutebmMuz}[2][]{\ensuremath{\subp{\acute{\bm{M}}}{}{#2}{}{#1}}}
\newrobustcmd{\gravebmMuz}[2][]{\ensuremath{\subp{\grave{\bm{M}}}{}{#2}{}{#1}}}
\newrobustcmd{\dotbmMuz}[2][]{\ensuremath{\subp{\dot{\bm{M}}}{}{#2}{}{#1}}}
\newrobustcmd{\ddotbmMuz}[2][]{\ensuremath{\subp{\ddot{\bm{M}}}{}{#2}{}{#1}}}
\newrobustcmd{\brevebmMuz}[2][]{\ensuremath{\subp{\breve{\bm{M}}}{}{#2}{}{#1}}}
\newrobustcmd{\barbmMuz}[2][]{\ensuremath{\subp{\bar{\bm{M}}}{}{#2}{}{#1}}}
\newrobustcmd{\vecbmMuz}[2][]{\ensuremath{\subp{\vec{\bm{M}}}{}{#2}{}{#1}}}
\newrobustcmd{\Nuz}[2][]{\ensuremath{\subp{N}{}{#2}{}{#1}}}
\newrobustcmd{\hatNuz}[2][]{\ensuremath{\subp{\hat{N}}{}{#2}{}{#1}}}
\newrobustcmd{\widehatNuz}[2][]{\ensuremath{\subp{\widehat{N}}{}{#2}{}{#1}}}
\newrobustcmd{\checkNuz}[2][]{\ensuremath{\subp{\check{N}}{}{#2}{}{#1}}}
\newrobustcmd{\tildeNuz}[2][]{\ensuremath{\subp{\tilde{N}}{}{#2}{}{#1}}}
\newrobustcmd{\widetildeNuz}[2][]{\ensuremath{\subp{\widetilde{N}}{}{#2}{}{#1}}}
\newrobustcmd{\acuteNuz}[2][]{\ensuremath{\subp{\acute{N}}{}{#2}{}{#1}}}
\newrobustcmd{\graveNuz}[2][]{\ensuremath{\subp{\grave{N}}{}{#2}{}{#1}}}
\newrobustcmd{\dotNuz}[2][]{\ensuremath{\subp{\dot{N}}{}{#2}{}{#1}}}
\newrobustcmd{\ddotNuz}[2][]{\ensuremath{\subp{\ddot{N}}{}{#2}{}{#1}}}
\newrobustcmd{\breveNuz}[2][]{\ensuremath{\subp{\breve{N}}{}{#2}{}{#1}}}
\newrobustcmd{\barNuz}[2][]{\ensuremath{\subp{\bar{N}}{}{#2}{}{#1}}}
\newrobustcmd{\vecNuz}[2][]{\ensuremath{\subp{\vec{N}}{}{#2}{}{#1}}}
\newrobustcmd{\bmNuz}[2][]{\ensuremath{\subp{\bm{N}}{}{#2}{}{#1}}}
\newrobustcmd{\hatbmNuz}[2][]{\ensuremath{\subp{\hat{\bm{N}}}{}{#2}{}{#1}}}
\newrobustcmd{\widehatbmNuz}[2][]{\ensuremath{\subp{\widehat{\bm{N}}}{}{#2}{}{#1}}}
\newrobustcmd{\checkbmNuz}[2][]{\ensuremath{\subp{\check{\bm{N}}}{}{#2}{}{#1}}}
\newrobustcmd{\tildebmNuz}[2][]{\ensuremath{\subp{\tilde{\bm{N}}}{}{#2}{}{#1}}}
\newrobustcmd{\widetildebmNuz}[2][]{\ensuremath{\subp{\widetilde{\bm{N}}}{}{#2}{}{#1}}}
\newrobustcmd{\acutebmNuz}[2][]{\ensuremath{\subp{\acute{\bm{N}}}{}{#2}{}{#1}}}
\newrobustcmd{\gravebmNuz}[2][]{\ensuremath{\subp{\grave{\bm{N}}}{}{#2}{}{#1}}}
\newrobustcmd{\dotbmNuz}[2][]{\ensuremath{\subp{\dot{\bm{N}}}{}{#2}{}{#1}}}
\newrobustcmd{\ddotbmNuz}[2][]{\ensuremath{\subp{\ddot{\bm{N}}}{}{#2}{}{#1}}}
\newrobustcmd{\brevebmNuz}[2][]{\ensuremath{\subp{\breve{\bm{N}}}{}{#2}{}{#1}}}
\newrobustcmd{\barbmNuz}[2][]{\ensuremath{\subp{\bar{\bm{N}}}{}{#2}{}{#1}}}
\newrobustcmd{\vecbmNuz}[2][]{\ensuremath{\subp{\vec{\bm{N}}}{}{#2}{}{#1}}}
\newrobustcmd{\Xiz}[2][]{\ensuremath{\subp{\Xi}{}{#2}{}{#1}}}
\newrobustcmd{\hatXiz}[2][]{\ensuremath{\subp{\hat{\Xi}}{}{#2}{}{#1}}}
\newrobustcmd{\widehatXiz}[2][]{\ensuremath{\subp{\widehat{\Xi}}{}{#2}{}{#1}}}
\newrobustcmd{\checkXiz}[2][]{\ensuremath{\subp{\check{\Xi}}{}{#2}{}{#1}}}
\newrobustcmd{\tildeXiz}[2][]{\ensuremath{\subp{\tilde{\Xi}}{}{#2}{}{#1}}}
\newrobustcmd{\widetildeXiz}[2][]{\ensuremath{\subp{\widetilde{\Xi}}{}{#2}{}{#1}}}
\newrobustcmd{\acuteXiz}[2][]{\ensuremath{\subp{\acute{\Xi}}{}{#2}{}{#1}}}
\newrobustcmd{\graveXiz}[2][]{\ensuremath{\subp{\grave{\Xi}}{}{#2}{}{#1}}}
\newrobustcmd{\dotXiz}[2][]{\ensuremath{\subp{\dot{\Xi}}{}{#2}{}{#1}}}
\newrobustcmd{\ddotXiz}[2][]{\ensuremath{\subp{\ddot{\Xi}}{}{#2}{}{#1}}}
\newrobustcmd{\breveXiz}[2][]{\ensuremath{\subp{\breve{\Xi}}{}{#2}{}{#1}}}
\newrobustcmd{\barXiz}[2][]{\ensuremath{\subp{\bar{\Xi}}{}{#2}{}{#1}}}
\newrobustcmd{\vecXiz}[2][]{\ensuremath{\subp{\vec{\Xi}}{}{#2}{}{#1}}}
\newrobustcmd{\bmXiz}[2][]{\ensuremath{\subp{\bm{\Xi}}{}{#2}{}{#1}}}
\newrobustcmd{\hatbmXiz}[2][]{\ensuremath{\subp{\hat{\bm{\Xi}}}{}{#2}{}{#1}}}
\newrobustcmd{\widehatbmXiz}[2][]{\ensuremath{\subp{\widehat{\bm{\Xi}}}{}{#2}{}{#1}}}
\newrobustcmd{\checkbmXiz}[2][]{\ensuremath{\subp{\check{\bm{\Xi}}}{}{#2}{}{#1}}}
\newrobustcmd{\tildebmXiz}[2][]{\ensuremath{\subp{\tilde{\bm{\Xi}}}{}{#2}{}{#1}}}
\newrobustcmd{\widetildebmXiz}[2][]{\ensuremath{\subp{\widetilde{\bm{\Xi}}}{}{#2}{}{#1}}}
\newrobustcmd{\acutebmXiz}[2][]{\ensuremath{\subp{\acute{\bm{\Xi}}}{}{#2}{}{#1}}}
\newrobustcmd{\gravebmXiz}[2][]{\ensuremath{\subp{\grave{\bm{\Xi}}}{}{#2}{}{#1}}}
\newrobustcmd{\dotbmXiz}[2][]{\ensuremath{\subp{\dot{\bm{\Xi}}}{}{#2}{}{#1}}}
\newrobustcmd{\ddotbmXiz}[2][]{\ensuremath{\subp{\ddot{\bm{\Xi}}}{}{#2}{}{#1}}}
\newrobustcmd{\brevebmXiz}[2][]{\ensuremath{\subp{\breve{\bm{\Xi}}}{}{#2}{}{#1}}}
\newrobustcmd{\barbmXiz}[2][]{\ensuremath{\subp{\bar{\bm{\Xi}}}{}{#2}{}{#1}}}
\newrobustcmd{\vecbmXiz}[2][]{\ensuremath{\subp{\vec{\bm{\Xi}}}{}{#2}{}{#1}}}
\newrobustcmd{\Piz}[2][]{\ensuremath{\subp{\Pi}{}{#2}{}{#1}}}
\newrobustcmd{\hatPiz}[2][]{\ensuremath{\subp{\hat{\Pi}}{}{#2}{}{#1}}}
\newrobustcmd{\widehatPiz}[2][]{\ensuremath{\subp{\widehat{\Pi}}{}{#2}{}{#1}}}
\newrobustcmd{\checkPiz}[2][]{\ensuremath{\subp{\check{\Pi}}{}{#2}{}{#1}}}
\newrobustcmd{\tildePiz}[2][]{\ensuremath{\subp{\tilde{\Pi}}{}{#2}{}{#1}}}
\newrobustcmd{\widetildePiz}[2][]{\ensuremath{\subp{\widetilde{\Pi}}{}{#2}{}{#1}}}
\newrobustcmd{\acutePiz}[2][]{\ensuremath{\subp{\acute{\Pi}}{}{#2}{}{#1}}}
\newrobustcmd{\gravePiz}[2][]{\ensuremath{\subp{\grave{\Pi}}{}{#2}{}{#1}}}
\newrobustcmd{\dotPiz}[2][]{\ensuremath{\subp{\dot{\Pi}}{}{#2}{}{#1}}}
\newrobustcmd{\ddotPiz}[2][]{\ensuremath{\subp{\ddot{\Pi}}{}{#2}{}{#1}}}
\newrobustcmd{\brevePiz}[2][]{\ensuremath{\subp{\breve{\Pi}}{}{#2}{}{#1}}}
\newrobustcmd{\barPiz}[2][]{\ensuremath{\subp{\bar{\Pi}}{}{#2}{}{#1}}}
\newrobustcmd{\vecPiz}[2][]{\ensuremath{\subp{\vec{\Pi}}{}{#2}{}{#1}}}
\newrobustcmd{\bmPiz}[2][]{\ensuremath{\subp{\bm{\Pi}}{}{#2}{}{#1}}}
\newrobustcmd{\hatbmPiz}[2][]{\ensuremath{\subp{\hat{\bm{\Pi}}}{}{#2}{}{#1}}}
\newrobustcmd{\widehatbmPiz}[2][]{\ensuremath{\subp{\widehat{\bm{\Pi}}}{}{#2}{}{#1}}}
\newrobustcmd{\checkbmPiz}[2][]{\ensuremath{\subp{\check{\bm{\Pi}}}{}{#2}{}{#1}}}
\newrobustcmd{\tildebmPiz}[2][]{\ensuremath{\subp{\tilde{\bm{\Pi}}}{}{#2}{}{#1}}}
\newrobustcmd{\widetildebmPiz}[2][]{\ensuremath{\subp{\widetilde{\bm{\Pi}}}{}{#2}{}{#1}}}
\newrobustcmd{\acutebmPiz}[2][]{\ensuremath{\subp{\acute{\bm{\Pi}}}{}{#2}{}{#1}}}
\newrobustcmd{\gravebmPiz}[2][]{\ensuremath{\subp{\grave{\bm{\Pi}}}{}{#2}{}{#1}}}
\newrobustcmd{\dotbmPiz}[2][]{\ensuremath{\subp{\dot{\bm{\Pi}}}{}{#2}{}{#1}}}
\newrobustcmd{\ddotbmPiz}[2][]{\ensuremath{\subp{\ddot{\bm{\Pi}}}{}{#2}{}{#1}}}
\newrobustcmd{\brevebmPiz}[2][]{\ensuremath{\subp{\breve{\bm{\Pi}}}{}{#2}{}{#1}}}
\newrobustcmd{\barbmPiz}[2][]{\ensuremath{\subp{\bar{\bm{\Pi}}}{}{#2}{}{#1}}}
\newrobustcmd{\vecbmPiz}[2][]{\ensuremath{\subp{\vec{\bm{\Pi}}}{}{#2}{}{#1}}}
\newrobustcmd{\Rhoz}[2][]{\ensuremath{\subp{R}{}{#2}{}{#1}}}
\newrobustcmd{\hatRhoz}[2][]{\ensuremath{\subp{\hat{R}}{}{#2}{}{#1}}}
\newrobustcmd{\widehatRhoz}[2][]{\ensuremath{\subp{\widehat{R}}{}{#2}{}{#1}}}
\newrobustcmd{\checkRhoz}[2][]{\ensuremath{\subp{\check{R}}{}{#2}{}{#1}}}
\newrobustcmd{\tildeRhoz}[2][]{\ensuremath{\subp{\tilde{R}}{}{#2}{}{#1}}}
\newrobustcmd{\widetildeRhoz}[2][]{\ensuremath{\subp{\widetilde{R}}{}{#2}{}{#1}}}
\newrobustcmd{\acuteRhoz}[2][]{\ensuremath{\subp{\acute{R}}{}{#2}{}{#1}}}
\newrobustcmd{\graveRhoz}[2][]{\ensuremath{\subp{\grave{R}}{}{#2}{}{#1}}}
\newrobustcmd{\dotRhoz}[2][]{\ensuremath{\subp{\dot{R}}{}{#2}{}{#1}}}
\newrobustcmd{\ddotRhoz}[2][]{\ensuremath{\subp{\ddot{R}}{}{#2}{}{#1}}}
\newrobustcmd{\breveRhoz}[2][]{\ensuremath{\subp{\breve{R}}{}{#2}{}{#1}}}
\newrobustcmd{\barRhoz}[2][]{\ensuremath{\subp{\bar{R}}{}{#2}{}{#1}}}
\newrobustcmd{\vecRhoz}[2][]{\ensuremath{\subp{\vec{R}}{}{#2}{}{#1}}}
\newrobustcmd{\bmRhoz}[2][]{\ensuremath{\subp{\bm{R}}{}{#2}{}{#1}}}
\newrobustcmd{\hatbmRhoz}[2][]{\ensuremath{\subp{\hat{\bm{R}}}{}{#2}{}{#1}}}
\newrobustcmd{\widehatbmRhoz}[2][]{\ensuremath{\subp{\widehat{\bm{R}}}{}{#2}{}{#1}}}
\newrobustcmd{\checkbmRhoz}[2][]{\ensuremath{\subp{\check{\bm{R}}}{}{#2}{}{#1}}}
\newrobustcmd{\tildebmRhoz}[2][]{\ensuremath{\subp{\tilde{\bm{R}}}{}{#2}{}{#1}}}
\newrobustcmd{\widetildebmRhoz}[2][]{\ensuremath{\subp{\widetilde{\bm{R}}}{}{#2}{}{#1}}}
\newrobustcmd{\acutebmRhoz}[2][]{\ensuremath{\subp{\acute{\bm{R}}}{}{#2}{}{#1}}}
\newrobustcmd{\gravebmRhoz}[2][]{\ensuremath{\subp{\grave{\bm{R}}}{}{#2}{}{#1}}}
\newrobustcmd{\dotbmRhoz}[2][]{\ensuremath{\subp{\dot{\bm{R}}}{}{#2}{}{#1}}}
\newrobustcmd{\ddotbmRhoz}[2][]{\ensuremath{\subp{\ddot{\bm{R}}}{}{#2}{}{#1}}}
\newrobustcmd{\brevebmRhoz}[2][]{\ensuremath{\subp{\breve{\bm{R}}}{}{#2}{}{#1}}}
\newrobustcmd{\barbmRhoz}[2][]{\ensuremath{\subp{\bar{\bm{R}}}{}{#2}{}{#1}}}
\newrobustcmd{\vecbmRhoz}[2][]{\ensuremath{\subp{\vec{\bm{R}}}{}{#2}{}{#1}}}
\newrobustcmd{\Sigmaz}[2][]{\ensuremath{\subp{\Sigma}{}{#2}{}{#1}}}
\newrobustcmd{\hatSigmaz}[2][]{\ensuremath{\subp{\hat{\Sigma}}{}{#2}{}{#1}}}
\newrobustcmd{\widehatSigmaz}[2][]{\ensuremath{\subp{\widehat{\Sigma}}{}{#2}{}{#1}}}
\newrobustcmd{\checkSigmaz}[2][]{\ensuremath{\subp{\check{\Sigma}}{}{#2}{}{#1}}}
\newrobustcmd{\tildeSigmaz}[2][]{\ensuremath{\subp{\tilde{\Sigma}}{}{#2}{}{#1}}}
\newrobustcmd{\widetildeSigmaz}[2][]{\ensuremath{\subp{\widetilde{\Sigma}}{}{#2}{}{#1}}}
\newrobustcmd{\acuteSigmaz}[2][]{\ensuremath{\subp{\acute{\Sigma}}{}{#2}{}{#1}}}
\newrobustcmd{\graveSigmaz}[2][]{\ensuremath{\subp{\grave{\Sigma}}{}{#2}{}{#1}}}
\newrobustcmd{\dotSigmaz}[2][]{\ensuremath{\subp{\dot{\Sigma}}{}{#2}{}{#1}}}
\newrobustcmd{\ddotSigmaz}[2][]{\ensuremath{\subp{\ddot{\Sigma}}{}{#2}{}{#1}}}
\newrobustcmd{\breveSigmaz}[2][]{\ensuremath{\subp{\breve{\Sigma}}{}{#2}{}{#1}}}
\newrobustcmd{\barSigmaz}[2][]{\ensuremath{\subp{\bar{\Sigma}}{}{#2}{}{#1}}}
\newrobustcmd{\vecSigmaz}[2][]{\ensuremath{\subp{\vec{\Sigma}}{}{#2}{}{#1}}}
\newrobustcmd{\bmSigmaz}[2][]{\ensuremath{\subp{\bm{\Sigma}}{}{#2}{}{#1}}}
\newrobustcmd{\hatbmSigmaz}[2][]{\ensuremath{\subp{\hat{\bm{\Sigma}}}{}{#2}{}{#1}}}
\newrobustcmd{\widehatbmSigmaz}[2][]{\ensuremath{\subp{\widehat{\bm{\Sigma}}}{}{#2}{}{#1}}}
\newrobustcmd{\checkbmSigmaz}[2][]{\ensuremath{\subp{\check{\bm{\Sigma}}}{}{#2}{}{#1}}}
\newrobustcmd{\tildebmSigmaz}[2][]{\ensuremath{\subp{\tilde{\bm{\Sigma}}}{}{#2}{}{#1}}}
\newrobustcmd{\widetildebmSigmaz}[2][]{\ensuremath{\subp{\widetilde{\bm{\Sigma}}}{}{#2}{}{#1}}}
\newrobustcmd{\acutebmSigmaz}[2][]{\ensuremath{\subp{\acute{\bm{\Sigma}}}{}{#2}{}{#1}}}
\newrobustcmd{\gravebmSigmaz}[2][]{\ensuremath{\subp{\grave{\bm{\Sigma}}}{}{#2}{}{#1}}}
\newrobustcmd{\dotbmSigmaz}[2][]{\ensuremath{\subp{\dot{\bm{\Sigma}}}{}{#2}{}{#1}}}
\newrobustcmd{\ddotbmSigmaz}[2][]{\ensuremath{\subp{\ddot{\bm{\Sigma}}}{}{#2}{}{#1}}}
\newrobustcmd{\brevebmSigmaz}[2][]{\ensuremath{\subp{\breve{\bm{\Sigma}}}{}{#2}{}{#1}}}
\newrobustcmd{\barbmSigmaz}[2][]{\ensuremath{\subp{\bar{\bm{\Sigma}}}{}{#2}{}{#1}}}
\newrobustcmd{\vecbmSigmaz}[2][]{\ensuremath{\subp{\vec{\bm{\Sigma}}}{}{#2}{}{#1}}}
\newrobustcmd{\Tauz}[2][]{\ensuremath{\subp{T}{}{#2}{}{#1}}}
\newrobustcmd{\hatTauz}[2][]{\ensuremath{\subp{\hat{T}}{}{#2}{}{#1}}}
\newrobustcmd{\widehatTauz}[2][]{\ensuremath{\subp{\widehat{T}}{}{#2}{}{#1}}}
\newrobustcmd{\checkTauz}[2][]{\ensuremath{\subp{\check{T}}{}{#2}{}{#1}}}
\newrobustcmd{\tildeTauz}[2][]{\ensuremath{\subp{\tilde{T}}{}{#2}{}{#1}}}
\newrobustcmd{\widetildeTauz}[2][]{\ensuremath{\subp{\widetilde{T}}{}{#2}{}{#1}}}
\newrobustcmd{\acuteTauz}[2][]{\ensuremath{\subp{\acute{T}}{}{#2}{}{#1}}}
\newrobustcmd{\graveTauz}[2][]{\ensuremath{\subp{\grave{T}}{}{#2}{}{#1}}}
\newrobustcmd{\dotTauz}[2][]{\ensuremath{\subp{\dot{T}}{}{#2}{}{#1}}}
\newrobustcmd{\ddotTauz}[2][]{\ensuremath{\subp{\ddot{T}}{}{#2}{}{#1}}}
\newrobustcmd{\breveTauz}[2][]{\ensuremath{\subp{\breve{T}}{}{#2}{}{#1}}}
\newrobustcmd{\barTauz}[2][]{\ensuremath{\subp{\bar{T}}{}{#2}{}{#1}}}
\newrobustcmd{\vecTauz}[2][]{\ensuremath{\subp{\vec{T}}{}{#2}{}{#1}}}
\newrobustcmd{\bmTauz}[2][]{\ensuremath{\subp{\bm{T}}{}{#2}{}{#1}}}
\newrobustcmd{\hatbmTauz}[2][]{\ensuremath{\subp{\hat{\bm{T}}}{}{#2}{}{#1}}}
\newrobustcmd{\widehatbmTauz}[2][]{\ensuremath{\subp{\widehat{\bm{T}}}{}{#2}{}{#1}}}
\newrobustcmd{\checkbmTauz}[2][]{\ensuremath{\subp{\check{\bm{T}}}{}{#2}{}{#1}}}
\newrobustcmd{\tildebmTauz}[2][]{\ensuremath{\subp{\tilde{\bm{T}}}{}{#2}{}{#1}}}
\newrobustcmd{\widetildebmTauz}[2][]{\ensuremath{\subp{\widetilde{\bm{T}}}{}{#2}{}{#1}}}
\newrobustcmd{\acutebmTauz}[2][]{\ensuremath{\subp{\acute{\bm{T}}}{}{#2}{}{#1}}}
\newrobustcmd{\gravebmTauz}[2][]{\ensuremath{\subp{\grave{\bm{T}}}{}{#2}{}{#1}}}
\newrobustcmd{\dotbmTauz}[2][]{\ensuremath{\subp{\dot{\bm{T}}}{}{#2}{}{#1}}}
\newrobustcmd{\ddotbmTauz}[2][]{\ensuremath{\subp{\ddot{\bm{T}}}{}{#2}{}{#1}}}
\newrobustcmd{\brevebmTauz}[2][]{\ensuremath{\subp{\breve{\bm{T}}}{}{#2}{}{#1}}}
\newrobustcmd{\barbmTauz}[2][]{\ensuremath{\subp{\bar{\bm{T}}}{}{#2}{}{#1}}}
\newrobustcmd{\vecbmTauz}[2][]{\ensuremath{\subp{\vec{\bm{T}}}{}{#2}{}{#1}}}
\newrobustcmd{\Upsilonz}[2][]{\ensuremath{\subp{\Upsilon}{}{#2}{}{#1}}}
\newrobustcmd{\hatUpsilonz}[2][]{\ensuremath{\subp{\hat{\Upsilon}}{}{#2}{}{#1}}}
\newrobustcmd{\widehatUpsilonz}[2][]{\ensuremath{\subp{\widehat{\Upsilon}}{}{#2}{}{#1}}}
\newrobustcmd{\checkUpsilonz}[2][]{\ensuremath{\subp{\check{\Upsilon}}{}{#2}{}{#1}}}
\newrobustcmd{\tildeUpsilonz}[2][]{\ensuremath{\subp{\tilde{\Upsilon}}{}{#2}{}{#1}}}
\newrobustcmd{\widetildeUpsilonz}[2][]{\ensuremath{\subp{\widetilde{\Upsilon}}{}{#2}{}{#1}}}
\newrobustcmd{\acuteUpsilonz}[2][]{\ensuremath{\subp{\acute{\Upsilon}}{}{#2}{}{#1}}}
\newrobustcmd{\graveUpsilonz}[2][]{\ensuremath{\subp{\grave{\Upsilon}}{}{#2}{}{#1}}}
\newrobustcmd{\dotUpsilonz}[2][]{\ensuremath{\subp{\dot{\Upsilon}}{}{#2}{}{#1}}}
\newrobustcmd{\ddotUpsilonz}[2][]{\ensuremath{\subp{\ddot{\Upsilon}}{}{#2}{}{#1}}}
\newrobustcmd{\breveUpsilonz}[2][]{\ensuremath{\subp{\breve{\Upsilon}}{}{#2}{}{#1}}}
\newrobustcmd{\barUpsilonz}[2][]{\ensuremath{\subp{\bar{\Upsilon}}{}{#2}{}{#1}}}
\newrobustcmd{\vecUpsilonz}[2][]{\ensuremath{\subp{\vec{\Upsilon}}{}{#2}{}{#1}}}
\newrobustcmd{\bmUpsilonz}[2][]{\ensuremath{\subp{\bm{\Upsilon}}{}{#2}{}{#1}}}
\newrobustcmd{\hatbmUpsilonz}[2][]{\ensuremath{\subp{\hat{\bm{\Upsilon}}}{}{#2}{}{#1}}}
\newrobustcmd{\widehatbmUpsilonz}[2][]{\ensuremath{\subp{\widehat{\bm{\Upsilon}}}{}{#2}{}{#1}}}
\newrobustcmd{\checkbmUpsilonz}[2][]{\ensuremath{\subp{\check{\bm{\Upsilon}}}{}{#2}{}{#1}}}
\newrobustcmd{\tildebmUpsilonz}[2][]{\ensuremath{\subp{\tilde{\bm{\Upsilon}}}{}{#2}{}{#1}}}
\newrobustcmd{\widetildebmUpsilonz}[2][]{\ensuremath{\subp{\widetilde{\bm{\Upsilon}}}{}{#2}{}{#1}}}
\newrobustcmd{\acutebmUpsilonz}[2][]{\ensuremath{\subp{\acute{\bm{\Upsilon}}}{}{#2}{}{#1}}}
\newrobustcmd{\gravebmUpsilonz}[2][]{\ensuremath{\subp{\grave{\bm{\Upsilon}}}{}{#2}{}{#1}}}
\newrobustcmd{\dotbmUpsilonz}[2][]{\ensuremath{\subp{\dot{\bm{\Upsilon}}}{}{#2}{}{#1}}}
\newrobustcmd{\ddotbmUpsilonz}[2][]{\ensuremath{\subp{\ddot{\bm{\Upsilon}}}{}{#2}{}{#1}}}
\newrobustcmd{\brevebmUpsilonz}[2][]{\ensuremath{\subp{\breve{\bm{\Upsilon}}}{}{#2}{}{#1}}}
\newrobustcmd{\barbmUpsilonz}[2][]{\ensuremath{\subp{\bar{\bm{\Upsilon}}}{}{#2}{}{#1}}}
\newrobustcmd{\vecbmUpsilonz}[2][]{\ensuremath{\subp{\vec{\bm{\Upsilon}}}{}{#2}{}{#1}}}
\newrobustcmd{\Phiz}[2][]{\ensuremath{\subp{\Phi}{}{#2}{}{#1}}}
\newrobustcmd{\hatPhiz}[2][]{\ensuremath{\subp{\hat{\Phi}}{}{#2}{}{#1}}}
\newrobustcmd{\widehatPhiz}[2][]{\ensuremath{\subp{\widehat{\Phi}}{}{#2}{}{#1}}}
\newrobustcmd{\checkPhiz}[2][]{\ensuremath{\subp{\check{\Phi}}{}{#2}{}{#1}}}
\newrobustcmd{\tildePhiz}[2][]{\ensuremath{\subp{\tilde{\Phi}}{}{#2}{}{#1}}}
\newrobustcmd{\widetildePhiz}[2][]{\ensuremath{\subp{\widetilde{\Phi}}{}{#2}{}{#1}}}
\newrobustcmd{\acutePhiz}[2][]{\ensuremath{\subp{\acute{\Phi}}{}{#2}{}{#1}}}
\newrobustcmd{\gravePhiz}[2][]{\ensuremath{\subp{\grave{\Phi}}{}{#2}{}{#1}}}
\newrobustcmd{\dotPhiz}[2][]{\ensuremath{\subp{\dot{\Phi}}{}{#2}{}{#1}}}
\newrobustcmd{\ddotPhiz}[2][]{\ensuremath{\subp{\ddot{\Phi}}{}{#2}{}{#1}}}
\newrobustcmd{\brevePhiz}[2][]{\ensuremath{\subp{\breve{\Phi}}{}{#2}{}{#1}}}
\newrobustcmd{\barPhiz}[2][]{\ensuremath{\subp{\bar{\Phi}}{}{#2}{}{#1}}}
\newrobustcmd{\vecPhiz}[2][]{\ensuremath{\subp{\vec{\Phi}}{}{#2}{}{#1}}}
\newrobustcmd{\bmPhiz}[2][]{\ensuremath{\subp{\bm{\Phi}}{}{#2}{}{#1}}}
\newrobustcmd{\hatbmPhiz}[2][]{\ensuremath{\subp{\hat{\bm{\Phi}}}{}{#2}{}{#1}}}
\newrobustcmd{\widehatbmPhiz}[2][]{\ensuremath{\subp{\widehat{\bm{\Phi}}}{}{#2}{}{#1}}}
\newrobustcmd{\checkbmPhiz}[2][]{\ensuremath{\subp{\check{\bm{\Phi}}}{}{#2}{}{#1}}}
\newrobustcmd{\tildebmPhiz}[2][]{\ensuremath{\subp{\tilde{\bm{\Phi}}}{}{#2}{}{#1}}}
\newrobustcmd{\widetildebmPhiz}[2][]{\ensuremath{\subp{\widetilde{\bm{\Phi}}}{}{#2}{}{#1}}}
\newrobustcmd{\acutebmPhiz}[2][]{\ensuremath{\subp{\acute{\bm{\Phi}}}{}{#2}{}{#1}}}
\newrobustcmd{\gravebmPhiz}[2][]{\ensuremath{\subp{\grave{\bm{\Phi}}}{}{#2}{}{#1}}}
\newrobustcmd{\dotbmPhiz}[2][]{\ensuremath{\subp{\dot{\bm{\Phi}}}{}{#2}{}{#1}}}
\newrobustcmd{\ddotbmPhiz}[2][]{\ensuremath{\subp{\ddot{\bm{\Phi}}}{}{#2}{}{#1}}}
\newrobustcmd{\brevebmPhiz}[2][]{\ensuremath{\subp{\breve{\bm{\Phi}}}{}{#2}{}{#1}}}
\newrobustcmd{\barbmPhiz}[2][]{\ensuremath{\subp{\bar{\bm{\Phi}}}{}{#2}{}{#1}}}
\newrobustcmd{\vecbmPhiz}[2][]{\ensuremath{\subp{\vec{\bm{\Phi}}}{}{#2}{}{#1}}}
\newrobustcmd{\Chiz}[2][]{\ensuremath{\subp{\Chi}{}{#2}{}{#1}}}
\newrobustcmd{\hatChiz}[2][]{\ensuremath{\subp{\hat{\Chi}}{}{#2}{}{#1}}}
\newrobustcmd{\widehatChiz}[2][]{\ensuremath{\subp{\widehat{\Chi}}{}{#2}{}{#1}}}
\newrobustcmd{\checkChiz}[2][]{\ensuremath{\subp{\check{\Chi}}{}{#2}{}{#1}}}
\newrobustcmd{\tildeChiz}[2][]{\ensuremath{\subp{\tilde{\Chi}}{}{#2}{}{#1}}}
\newrobustcmd{\widetildeChiz}[2][]{\ensuremath{\subp{\widetilde{\Chi}}{}{#2}{}{#1}}}
\newrobustcmd{\acuteChiz}[2][]{\ensuremath{\subp{\acute{\Chi}}{}{#2}{}{#1}}}
\newrobustcmd{\graveChiz}[2][]{\ensuremath{\subp{\grave{\Chi}}{}{#2}{}{#1}}}
\newrobustcmd{\dotChiz}[2][]{\ensuremath{\subp{\dot{\Chi}}{}{#2}{}{#1}}}
\newrobustcmd{\ddotChiz}[2][]{\ensuremath{\subp{\ddot{\Chi}}{}{#2}{}{#1}}}
\newrobustcmd{\breveChiz}[2][]{\ensuremath{\subp{\breve{\Chi}}{}{#2}{}{#1}}}
\newrobustcmd{\barChiz}[2][]{\ensuremath{\subp{\bar{\Chi}}{}{#2}{}{#1}}}
\newrobustcmd{\vecChiz}[2][]{\ensuremath{\subp{\vec{\Chi}}{}{#2}{}{#1}}}
\newrobustcmd{\bmChiz}[2][]{\ensuremath{\subp{\bm{\Chi}}{}{#2}{}{#1}}}
\newrobustcmd{\hatbmChiz}[2][]{\ensuremath{\subp{\hat{\bm{\Chi}}}{}{#2}{}{#1}}}
\newrobustcmd{\widehatbmChiz}[2][]{\ensuremath{\subp{\widehat{\bm{\Chi}}}{}{#2}{}{#1}}}
\newrobustcmd{\checkbmChiz}[2][]{\ensuremath{\subp{\check{\bm{\Chi}}}{}{#2}{}{#1}}}
\newrobustcmd{\tildebmChiz}[2][]{\ensuremath{\subp{\tilde{\bm{\Chi}}}{}{#2}{}{#1}}}
\newrobustcmd{\widetildebmChiz}[2][]{\ensuremath{\subp{\widetilde{\bm{\Chi}}}{}{#2}{}{#1}}}
\newrobustcmd{\acutebmChiz}[2][]{\ensuremath{\subp{\acute{\bm{\Chi}}}{}{#2}{}{#1}}}
\newrobustcmd{\gravebmChiz}[2][]{\ensuremath{\subp{\grave{\bm{\Chi}}}{}{#2}{}{#1}}}
\newrobustcmd{\dotbmChiz}[2][]{\ensuremath{\subp{\dot{\bm{\Chi}}}{}{#2}{}{#1}}}
\newrobustcmd{\ddotbmChiz}[2][]{\ensuremath{\subp{\ddot{\bm{\Chi}}}{}{#2}{}{#1}}}
\newrobustcmd{\brevebmChiz}[2][]{\ensuremath{\subp{\breve{\bm{\Chi}}}{}{#2}{}{#1}}}
\newrobustcmd{\barbmChiz}[2][]{\ensuremath{\subp{\bar{\bm{\Chi}}}{}{#2}{}{#1}}}
\newrobustcmd{\vecbmChiz}[2][]{\ensuremath{\subp{\vec{\bm{\Chi}}}{}{#2}{}{#1}}}
\newrobustcmd{\Psiz}[2][]{\ensuremath{\subp{\Psi}{}{#2}{}{#1}}}
\newrobustcmd{\hatPsiz}[2][]{\ensuremath{\subp{\hat{\Psi}}{}{#2}{}{#1}}}
\newrobustcmd{\widehatPsiz}[2][]{\ensuremath{\subp{\widehat{\Psi}}{}{#2}{}{#1}}}
\newrobustcmd{\checkPsiz}[2][]{\ensuremath{\subp{\check{\Psi}}{}{#2}{}{#1}}}
\newrobustcmd{\tildePsiz}[2][]{\ensuremath{\subp{\tilde{\Psi}}{}{#2}{}{#1}}}
\newrobustcmd{\widetildePsiz}[2][]{\ensuremath{\subp{\widetilde{\Psi}}{}{#2}{}{#1}}}
\newrobustcmd{\acutePsiz}[2][]{\ensuremath{\subp{\acute{\Psi}}{}{#2}{}{#1}}}
\newrobustcmd{\gravePsiz}[2][]{\ensuremath{\subp{\grave{\Psi}}{}{#2}{}{#1}}}
\newrobustcmd{\dotPsiz}[2][]{\ensuremath{\subp{\dot{\Psi}}{}{#2}{}{#1}}}
\newrobustcmd{\ddotPsiz}[2][]{\ensuremath{\subp{\ddot{\Psi}}{}{#2}{}{#1}}}
\newrobustcmd{\brevePsiz}[2][]{\ensuremath{\subp{\breve{\Psi}}{}{#2}{}{#1}}}
\newrobustcmd{\barPsiz}[2][]{\ensuremath{\subp{\bar{\Psi}}{}{#2}{}{#1}}}
\newrobustcmd{\vecPsiz}[2][]{\ensuremath{\subp{\vec{\Psi}}{}{#2}{}{#1}}}
\newrobustcmd{\bmPsiz}[2][]{\ensuremath{\subp{\bm{\Psi}}{}{#2}{}{#1}}}
\newrobustcmd{\hatbmPsiz}[2][]{\ensuremath{\subp{\hat{\bm{\Psi}}}{}{#2}{}{#1}}}
\newrobustcmd{\widehatbmPsiz}[2][]{\ensuremath{\subp{\widehat{\bm{\Psi}}}{}{#2}{}{#1}}}
\newrobustcmd{\checkbmPsiz}[2][]{\ensuremath{\subp{\check{\bm{\Psi}}}{}{#2}{}{#1}}}
\newrobustcmd{\tildebmPsiz}[2][]{\ensuremath{\subp{\tilde{\bm{\Psi}}}{}{#2}{}{#1}}}
\newrobustcmd{\widetildebmPsiz}[2][]{\ensuremath{\subp{\widetilde{\bm{\Psi}}}{}{#2}{}{#1}}}
\newrobustcmd{\acutebmPsiz}[2][]{\ensuremath{\subp{\acute{\bm{\Psi}}}{}{#2}{}{#1}}}
\newrobustcmd{\gravebmPsiz}[2][]{\ensuremath{\subp{\grave{\bm{\Psi}}}{}{#2}{}{#1}}}
\newrobustcmd{\dotbmPsiz}[2][]{\ensuremath{\subp{\dot{\bm{\Psi}}}{}{#2}{}{#1}}}
\newrobustcmd{\ddotbmPsiz}[2][]{\ensuremath{\subp{\ddot{\bm{\Psi}}}{}{#2}{}{#1}}}
\newrobustcmd{\brevebmPsiz}[2][]{\ensuremath{\subp{\breve{\bm{\Psi}}}{}{#2}{}{#1}}}
\newrobustcmd{\barbmPsiz}[2][]{\ensuremath{\subp{\bar{\bm{\Psi}}}{}{#2}{}{#1}}}
\newrobustcmd{\vecbmPsiz}[2][]{\ensuremath{\subp{\vec{\bm{\Psi}}}{}{#2}{}{#1}}}
\newrobustcmd{\Omegaz}[2][]{\ensuremath{\subp{\Omega}{}{#2}{}{#1}}}
\newrobustcmd{\hatOmegaz}[2][]{\ensuremath{\subp{\hat{\Omega}}{}{#2}{}{#1}}}
\newrobustcmd{\widehatOmegaz}[2][]{\ensuremath{\subp{\widehat{\Omega}}{}{#2}{}{#1}}}
\newrobustcmd{\checkOmegaz}[2][]{\ensuremath{\subp{\check{\Omega}}{}{#2}{}{#1}}}
\newrobustcmd{\tildeOmegaz}[2][]{\ensuremath{\subp{\tilde{\Omega}}{}{#2}{}{#1}}}
\newrobustcmd{\widetildeOmegaz}[2][]{\ensuremath{\subp{\widetilde{\Omega}}{}{#2}{}{#1}}}
\newrobustcmd{\acuteOmegaz}[2][]{\ensuremath{\subp{\acute{\Omega}}{}{#2}{}{#1}}}
\newrobustcmd{\graveOmegaz}[2][]{\ensuremath{\subp{\grave{\Omega}}{}{#2}{}{#1}}}
\newrobustcmd{\dotOmegaz}[2][]{\ensuremath{\subp{\dot{\Omega}}{}{#2}{}{#1}}}
\newrobustcmd{\ddotOmegaz}[2][]{\ensuremath{\subp{\ddot{\Omega}}{}{#2}{}{#1}}}
\newrobustcmd{\breveOmegaz}[2][]{\ensuremath{\subp{\breve{\Omega}}{}{#2}{}{#1}}}
\newrobustcmd{\barOmegaz}[2][]{\ensuremath{\subp{\bar{\Omega}}{}{#2}{}{#1}}}
\newrobustcmd{\vecOmegaz}[2][]{\ensuremath{\subp{\vec{\Omega}}{}{#2}{}{#1}}}
\newrobustcmd{\bmOmegaz}[2][]{\ensuremath{\subp{\bm{\Omega}}{}{#2}{}{#1}}}
\newrobustcmd{\hatbmOmegaz}[2][]{\ensuremath{\subp{\hat{\bm{\Omega}}}{}{#2}{}{#1}}}
\newrobustcmd{\widehatbmOmegaz}[2][]{\ensuremath{\subp{\widehat{\bm{\Omega}}}{}{#2}{}{#1}}}
\newrobustcmd{\checkbmOmegaz}[2][]{\ensuremath{\subp{\check{\bm{\Omega}}}{}{#2}{}{#1}}}
\newrobustcmd{\tildebmOmegaz}[2][]{\ensuremath{\subp{\tilde{\bm{\Omega}}}{}{#2}{}{#1}}}
\newrobustcmd{\widetildebmOmegaz}[2][]{\ensuremath{\subp{\widetilde{\bm{\Omega}}}{}{#2}{}{#1}}}
\newrobustcmd{\acutebmOmegaz}[2][]{\ensuremath{\subp{\acute{\bm{\Omega}}}{}{#2}{}{#1}}}
\newrobustcmd{\gravebmOmegaz}[2][]{\ensuremath{\subp{\grave{\bm{\Omega}}}{}{#2}{}{#1}}}
\newrobustcmd{\dotbmOmegaz}[2][]{\ensuremath{\subp{\dot{\bm{\Omega}}}{}{#2}{}{#1}}}
\newrobustcmd{\ddotbmOmegaz}[2][]{\ensuremath{\subp{\ddot{\bm{\Omega}}}{}{#2}{}{#1}}}
\newrobustcmd{\brevebmOmegaz}[2][]{\ensuremath{\subp{\breve{\bm{\Omega}}}{}{#2}{}{#1}}}
\newrobustcmd{\barbmOmegaz}[2][]{\ensuremath{\subp{\bar{\bm{\Omega}}}{}{#2}{}{#1}}}
\newrobustcmd{\vecbmOmegaz}[2][]{\ensuremath{\subp{\vec{\bm{\Omega}}}{}{#2}{}{#1}}}

\newrobustcmd{\mathfrakaz}[2][]{\ensuremath{\subp{\mathfrak{a}}{}{#2}{}{#1}}}
\newrobustcmd{\hatmathfrakaz}[2][]{\ensuremath{\subp{\hat{\mathfrak{a}}}{}{#2}{}{#1}}}
\newrobustcmd{\widehatmathfrakaz}[2][]{\ensuremath{\subp{\widehat{\mathfrak{a}}}{}{#2}{}{#1}}}
\newrobustcmd{\checkmathfrakaz}[2][]{\ensuremath{\subp{\check{\mathfrak{a}}}{}{#2}{}{#1}}}
\newrobustcmd{\tildemathfrakaz}[2][]{\ensuremath{\subp{\tilde{\mathfrak{a}}}{}{#2}{}{#1}}}
\newrobustcmd{\widetildemathfrakaz}[2][]{\ensuremath{\subp{\widetilde{\mathfrak{a}}}{}{#2}{}{#1}}}
\newrobustcmd{\acutemathfrakaz}[2][]{\ensuremath{\subp{\acute{\mathfrak{a}}}{}{#2}{}{#1}}}
\newrobustcmd{\gravemathfrakaz}[2][]{\ensuremath{\subp{\grave{\mathfrak{a}}}{}{#2}{}{#1}}}
\newrobustcmd{\dotmathfrakaz}[2][]{\ensuremath{\subp{\dot{\mathfrak{a}}}{}{#2}{}{#1}}}
\newrobustcmd{\ddotmathfrakaz}[2][]{\ensuremath{\subp{\ddot{\mathfrak{a}}}{}{#2}{}{#1}}}
\newrobustcmd{\brevemathfrakaz}[2][]{\ensuremath{\subp{\breve{\mathfrak{a}}}{}{#2}{}{#1}}}
\newrobustcmd{\barmathfrakaz}[2][]{\ensuremath{\subp{\bar{\mathfrak{a}}}{}{#2}{}{#1}}}
\newrobustcmd{\vecmathfrakaz}[2][]{\ensuremath{\subp{\vec{\mathfrak{a}}}{}{#2}{}{#1}}}
\newrobustcmd{\bmmathfrakaz}[2][]{\ensuremath{\subp{\bm{\mathfrak{a}}}{}{#2}{}{#1}}}
\newrobustcmd{\hatbmmathfrakaz}[2][]{\ensuremath{\subp{\hat{\bm{\mathfrak{a}}}}{}{#2}{}{#1}}}
\newrobustcmd{\widehatbmmathfrakaz}[2][]{\ensuremath{\subp{\widehat{\bm{\mathfrak{a}}}}{}{#2}{}{#1}}}
\newrobustcmd{\checkbmmathfrakaz}[2][]{\ensuremath{\subp{\check{\bm{\mathfrak{a}}}}{}{#2}{}{#1}}}
\newrobustcmd{\tildebmmathfrakaz}[2][]{\ensuremath{\subp{\tilde{\bm{\mathfrak{a}}}}{}{#2}{}{#1}}}
\newrobustcmd{\widetildebmmathfrakaz}[2][]{\ensuremath{\subp{\widetilde{\bm{\mathfrak{a}}}}{}{#2}{}{#1}}}
\newrobustcmd{\acutebmmathfrakaz}[2][]{\ensuremath{\subp{\acute{\bm{\mathfrak{a}}}}{}{#2}{}{#1}}}
\newrobustcmd{\gravebmmathfrakaz}[2][]{\ensuremath{\subp{\grave{\bm{\mathfrak{a}}}}{}{#2}{}{#1}}}
\newrobustcmd{\dotbmmathfrakaz}[2][]{\ensuremath{\subp{\dot{\bm{\mathfrak{a}}}}{}{#2}{}{#1}}}
\newrobustcmd{\ddotbmmathfrakaz}[2][]{\ensuremath{\subp{\ddot{\bm{\mathfrak{a}}}}{}{#2}{}{#1}}}
\newrobustcmd{\brevebmmathfrakaz}[2][]{\ensuremath{\subp{\breve{\bm{\mathfrak{a}}}}{}{#2}{}{#1}}}
\newrobustcmd{\barbmmathfrakaz}[2][]{\ensuremath{\subp{\bar{\bm{\mathfrak{a}}}}{}{#2}{}{#1}}}
\newrobustcmd{\vecbmmathfrakaz}[2][]{\ensuremath{\subp{\vec{\bm{\mathfrak{a}}}}{}{#2}{}{#1}}}
\newrobustcmd{\mathfrakbz}[2][]{\ensuremath{\subp{\mathfrak{b}}{}{#2}{}{#1}}}
\newrobustcmd{\hatmathfrakbz}[2][]{\ensuremath{\subp{\hat{\mathfrak{b}}}{}{#2}{}{#1}}}
\newrobustcmd{\widehatmathfrakbz}[2][]{\ensuremath{\subp{\widehat{\mathfrak{b}}}{}{#2}{}{#1}}}
\newrobustcmd{\checkmathfrakbz}[2][]{\ensuremath{\subp{\check{\mathfrak{b}}}{}{#2}{}{#1}}}
\newrobustcmd{\tildemathfrakbz}[2][]{\ensuremath{\subp{\tilde{\mathfrak{b}}}{}{#2}{}{#1}}}
\newrobustcmd{\widetildemathfrakbz}[2][]{\ensuremath{\subp{\widetilde{\mathfrak{b}}}{}{#2}{}{#1}}}
\newrobustcmd{\acutemathfrakbz}[2][]{\ensuremath{\subp{\acute{\mathfrak{b}}}{}{#2}{}{#1}}}
\newrobustcmd{\gravemathfrakbz}[2][]{\ensuremath{\subp{\grave{\mathfrak{b}}}{}{#2}{}{#1}}}
\newrobustcmd{\dotmathfrakbz}[2][]{\ensuremath{\subp{\dot{\mathfrak{b}}}{}{#2}{}{#1}}}
\newrobustcmd{\ddotmathfrakbz}[2][]{\ensuremath{\subp{\ddot{\mathfrak{b}}}{}{#2}{}{#1}}}
\newrobustcmd{\brevemathfrakbz}[2][]{\ensuremath{\subp{\breve{\mathfrak{b}}}{}{#2}{}{#1}}}
\newrobustcmd{\barmathfrakbz}[2][]{\ensuremath{\subp{\bar{\mathfrak{b}}}{}{#2}{}{#1}}}
\newrobustcmd{\vecmathfrakbz}[2][]{\ensuremath{\subp{\vec{\mathfrak{b}}}{}{#2}{}{#1}}}
\newrobustcmd{\bmmathfrakbz}[2][]{\ensuremath{\subp{\bm{\mathfrak{b}}}{}{#2}{}{#1}}}
\newrobustcmd{\hatbmmathfrakbz}[2][]{\ensuremath{\subp{\hat{\bm{\mathfrak{b}}}}{}{#2}{}{#1}}}
\newrobustcmd{\widehatbmmathfrakbz}[2][]{\ensuremath{\subp{\widehat{\bm{\mathfrak{b}}}}{}{#2}{}{#1}}}
\newrobustcmd{\checkbmmathfrakbz}[2][]{\ensuremath{\subp{\check{\bm{\mathfrak{b}}}}{}{#2}{}{#1}}}
\newrobustcmd{\tildebmmathfrakbz}[2][]{\ensuremath{\subp{\tilde{\bm{\mathfrak{b}}}}{}{#2}{}{#1}}}
\newrobustcmd{\widetildebmmathfrakbz}[2][]{\ensuremath{\subp{\widetilde{\bm{\mathfrak{b}}}}{}{#2}{}{#1}}}
\newrobustcmd{\acutebmmathfrakbz}[2][]{\ensuremath{\subp{\acute{\bm{\mathfrak{b}}}}{}{#2}{}{#1}}}
\newrobustcmd{\gravebmmathfrakbz}[2][]{\ensuremath{\subp{\grave{\bm{\mathfrak{b}}}}{}{#2}{}{#1}}}
\newrobustcmd{\dotbmmathfrakbz}[2][]{\ensuremath{\subp{\dot{\bm{\mathfrak{b}}}}{}{#2}{}{#1}}}
\newrobustcmd{\ddotbmmathfrakbz}[2][]{\ensuremath{\subp{\ddot{\bm{\mathfrak{b}}}}{}{#2}{}{#1}}}
\newrobustcmd{\brevebmmathfrakbz}[2][]{\ensuremath{\subp{\breve{\bm{\mathfrak{b}}}}{}{#2}{}{#1}}}
\newrobustcmd{\barbmmathfrakbz}[2][]{\ensuremath{\subp{\bar{\bm{\mathfrak{b}}}}{}{#2}{}{#1}}}
\newrobustcmd{\vecbmmathfrakbz}[2][]{\ensuremath{\subp{\vec{\bm{\mathfrak{b}}}}{}{#2}{}{#1}}}
\newrobustcmd{\mathfrakcz}[2][]{\ensuremath{\subp{\mathfrak{c}}{}{#2}{}{#1}}}
\newrobustcmd{\hatmathfrakcz}[2][]{\ensuremath{\subp{\hat{\mathfrak{c}}}{}{#2}{}{#1}}}
\newrobustcmd{\widehatmathfrakcz}[2][]{\ensuremath{\subp{\widehat{\mathfrak{c}}}{}{#2}{}{#1}}}
\newrobustcmd{\checkmathfrakcz}[2][]{\ensuremath{\subp{\check{\mathfrak{c}}}{}{#2}{}{#1}}}
\newrobustcmd{\tildemathfrakcz}[2][]{\ensuremath{\subp{\tilde{\mathfrak{c}}}{}{#2}{}{#1}}}
\newrobustcmd{\widetildemathfrakcz}[2][]{\ensuremath{\subp{\widetilde{\mathfrak{c}}}{}{#2}{}{#1}}}
\newrobustcmd{\acutemathfrakcz}[2][]{\ensuremath{\subp{\acute{\mathfrak{c}}}{}{#2}{}{#1}}}
\newrobustcmd{\gravemathfrakcz}[2][]{\ensuremath{\subp{\grave{\mathfrak{c}}}{}{#2}{}{#1}}}
\newrobustcmd{\dotmathfrakcz}[2][]{\ensuremath{\subp{\dot{\mathfrak{c}}}{}{#2}{}{#1}}}
\newrobustcmd{\ddotmathfrakcz}[2][]{\ensuremath{\subp{\ddot{\mathfrak{c}}}{}{#2}{}{#1}}}
\newrobustcmd{\brevemathfrakcz}[2][]{\ensuremath{\subp{\breve{\mathfrak{c}}}{}{#2}{}{#1}}}
\newrobustcmd{\barmathfrakcz}[2][]{\ensuremath{\subp{\bar{\mathfrak{c}}}{}{#2}{}{#1}}}
\newrobustcmd{\vecmathfrakcz}[2][]{\ensuremath{\subp{\vec{\mathfrak{c}}}{}{#2}{}{#1}}}
\newrobustcmd{\bmmathfrakcz}[2][]{\ensuremath{\subp{\bm{\mathfrak{c}}}{}{#2}{}{#1}}}
\newrobustcmd{\hatbmmathfrakcz}[2][]{\ensuremath{\subp{\hat{\bm{\mathfrak{c}}}}{}{#2}{}{#1}}}
\newrobustcmd{\widehatbmmathfrakcz}[2][]{\ensuremath{\subp{\widehat{\bm{\mathfrak{c}}}}{}{#2}{}{#1}}}
\newrobustcmd{\checkbmmathfrakcz}[2][]{\ensuremath{\subp{\check{\bm{\mathfrak{c}}}}{}{#2}{}{#1}}}
\newrobustcmd{\tildebmmathfrakcz}[2][]{\ensuremath{\subp{\tilde{\bm{\mathfrak{c}}}}{}{#2}{}{#1}}}
\newrobustcmd{\widetildebmmathfrakcz}[2][]{\ensuremath{\subp{\widetilde{\bm{\mathfrak{c}}}}{}{#2}{}{#1}}}
\newrobustcmd{\acutebmmathfrakcz}[2][]{\ensuremath{\subp{\acute{\bm{\mathfrak{c}}}}{}{#2}{}{#1}}}
\newrobustcmd{\gravebmmathfrakcz}[2][]{\ensuremath{\subp{\grave{\bm{\mathfrak{c}}}}{}{#2}{}{#1}}}
\newrobustcmd{\dotbmmathfrakcz}[2][]{\ensuremath{\subp{\dot{\bm{\mathfrak{c}}}}{}{#2}{}{#1}}}
\newrobustcmd{\ddotbmmathfrakcz}[2][]{\ensuremath{\subp{\ddot{\bm{\mathfrak{c}}}}{}{#2}{}{#1}}}
\newrobustcmd{\brevebmmathfrakcz}[2][]{\ensuremath{\subp{\breve{\bm{\mathfrak{c}}}}{}{#2}{}{#1}}}
\newrobustcmd{\barbmmathfrakcz}[2][]{\ensuremath{\subp{\bar{\bm{\mathfrak{c}}}}{}{#2}{}{#1}}}
\newrobustcmd{\vecbmmathfrakcz}[2][]{\ensuremath{\subp{\vec{\bm{\mathfrak{c}}}}{}{#2}{}{#1}}}
\newrobustcmd{\mathfrakdz}[2][]{\ensuremath{\subp{\mathfrak{d}}{}{#2}{}{#1}}}
\newrobustcmd{\hatmathfrakdz}[2][]{\ensuremath{\subp{\hat{\mathfrak{d}}}{}{#2}{}{#1}}}
\newrobustcmd{\widehatmathfrakdz}[2][]{\ensuremath{\subp{\widehat{\mathfrak{d}}}{}{#2}{}{#1}}}
\newrobustcmd{\checkmathfrakdz}[2][]{\ensuremath{\subp{\check{\mathfrak{d}}}{}{#2}{}{#1}}}
\newrobustcmd{\tildemathfrakdz}[2][]{\ensuremath{\subp{\tilde{\mathfrak{d}}}{}{#2}{}{#1}}}
\newrobustcmd{\widetildemathfrakdz}[2][]{\ensuremath{\subp{\widetilde{\mathfrak{d}}}{}{#2}{}{#1}}}
\newrobustcmd{\acutemathfrakdz}[2][]{\ensuremath{\subp{\acute{\mathfrak{d}}}{}{#2}{}{#1}}}
\newrobustcmd{\gravemathfrakdz}[2][]{\ensuremath{\subp{\grave{\mathfrak{d}}}{}{#2}{}{#1}}}
\newrobustcmd{\dotmathfrakdz}[2][]{\ensuremath{\subp{\dot{\mathfrak{d}}}{}{#2}{}{#1}}}
\newrobustcmd{\ddotmathfrakdz}[2][]{\ensuremath{\subp{\ddot{\mathfrak{d}}}{}{#2}{}{#1}}}
\newrobustcmd{\brevemathfrakdz}[2][]{\ensuremath{\subp{\breve{\mathfrak{d}}}{}{#2}{}{#1}}}
\newrobustcmd{\barmathfrakdz}[2][]{\ensuremath{\subp{\bar{\mathfrak{d}}}{}{#2}{}{#1}}}
\newrobustcmd{\vecmathfrakdz}[2][]{\ensuremath{\subp{\vec{\mathfrak{d}}}{}{#2}{}{#1}}}
\newrobustcmd{\bmmathfrakdz}[2][]{\ensuremath{\subp{\bm{\mathfrak{d}}}{}{#2}{}{#1}}}
\newrobustcmd{\hatbmmathfrakdz}[2][]{\ensuremath{\subp{\hat{\bm{\mathfrak{d}}}}{}{#2}{}{#1}}}
\newrobustcmd{\widehatbmmathfrakdz}[2][]{\ensuremath{\subp{\widehat{\bm{\mathfrak{d}}}}{}{#2}{}{#1}}}
\newrobustcmd{\checkbmmathfrakdz}[2][]{\ensuremath{\subp{\check{\bm{\mathfrak{d}}}}{}{#2}{}{#1}}}
\newrobustcmd{\tildebmmathfrakdz}[2][]{\ensuremath{\subp{\tilde{\bm{\mathfrak{d}}}}{}{#2}{}{#1}}}
\newrobustcmd{\widetildebmmathfrakdz}[2][]{\ensuremath{\subp{\widetilde{\bm{\mathfrak{d}}}}{}{#2}{}{#1}}}
\newrobustcmd{\acutebmmathfrakdz}[2][]{\ensuremath{\subp{\acute{\bm{\mathfrak{d}}}}{}{#2}{}{#1}}}
\newrobustcmd{\gravebmmathfrakdz}[2][]{\ensuremath{\subp{\grave{\bm{\mathfrak{d}}}}{}{#2}{}{#1}}}
\newrobustcmd{\dotbmmathfrakdz}[2][]{\ensuremath{\subp{\dot{\bm{\mathfrak{d}}}}{}{#2}{}{#1}}}
\newrobustcmd{\ddotbmmathfrakdz}[2][]{\ensuremath{\subp{\ddot{\bm{\mathfrak{d}}}}{}{#2}{}{#1}}}
\newrobustcmd{\brevebmmathfrakdz}[2][]{\ensuremath{\subp{\breve{\bm{\mathfrak{d}}}}{}{#2}{}{#1}}}
\newrobustcmd{\barbmmathfrakdz}[2][]{\ensuremath{\subp{\bar{\bm{\mathfrak{d}}}}{}{#2}{}{#1}}}
\newrobustcmd{\vecbmmathfrakdz}[2][]{\ensuremath{\subp{\vec{\bm{\mathfrak{d}}}}{}{#2}{}{#1}}}
\newrobustcmd{\mathfrakez}[2][]{\ensuremath{\subp{\mathfrak{e}}{}{#2}{}{#1}}}
\newrobustcmd{\hatmathfrakez}[2][]{\ensuremath{\subp{\hat{\mathfrak{e}}}{}{#2}{}{#1}}}
\newrobustcmd{\widehatmathfrakez}[2][]{\ensuremath{\subp{\widehat{\mathfrak{e}}}{}{#2}{}{#1}}}
\newrobustcmd{\checkmathfrakez}[2][]{\ensuremath{\subp{\check{\mathfrak{e}}}{}{#2}{}{#1}}}
\newrobustcmd{\tildemathfrakez}[2][]{\ensuremath{\subp{\tilde{\mathfrak{e}}}{}{#2}{}{#1}}}
\newrobustcmd{\widetildemathfrakez}[2][]{\ensuremath{\subp{\widetilde{\mathfrak{e}}}{}{#2}{}{#1}}}
\newrobustcmd{\acutemathfrakez}[2][]{\ensuremath{\subp{\acute{\mathfrak{e}}}{}{#2}{}{#1}}}
\newrobustcmd{\gravemathfrakez}[2][]{\ensuremath{\subp{\grave{\mathfrak{e}}}{}{#2}{}{#1}}}
\newrobustcmd{\dotmathfrakez}[2][]{\ensuremath{\subp{\dot{\mathfrak{e}}}{}{#2}{}{#1}}}
\newrobustcmd{\ddotmathfrakez}[2][]{\ensuremath{\subp{\ddot{\mathfrak{e}}}{}{#2}{}{#1}}}
\newrobustcmd{\brevemathfrakez}[2][]{\ensuremath{\subp{\breve{\mathfrak{e}}}{}{#2}{}{#1}}}
\newrobustcmd{\barmathfrakez}[2][]{\ensuremath{\subp{\bar{\mathfrak{e}}}{}{#2}{}{#1}}}
\newrobustcmd{\vecmathfrakez}[2][]{\ensuremath{\subp{\vec{\mathfrak{e}}}{}{#2}{}{#1}}}
\newrobustcmd{\bmmathfrakez}[2][]{\ensuremath{\subp{\bm{\mathfrak{e}}}{}{#2}{}{#1}}}
\newrobustcmd{\hatbmmathfrakez}[2][]{\ensuremath{\subp{\hat{\bm{\mathfrak{e}}}}{}{#2}{}{#1}}}
\newrobustcmd{\widehatbmmathfrakez}[2][]{\ensuremath{\subp{\widehat{\bm{\mathfrak{e}}}}{}{#2}{}{#1}}}
\newrobustcmd{\checkbmmathfrakez}[2][]{\ensuremath{\subp{\check{\bm{\mathfrak{e}}}}{}{#2}{}{#1}}}
\newrobustcmd{\tildebmmathfrakez}[2][]{\ensuremath{\subp{\tilde{\bm{\mathfrak{e}}}}{}{#2}{}{#1}}}
\newrobustcmd{\widetildebmmathfrakez}[2][]{\ensuremath{\subp{\widetilde{\bm{\mathfrak{e}}}}{}{#2}{}{#1}}}
\newrobustcmd{\acutebmmathfrakez}[2][]{\ensuremath{\subp{\acute{\bm{\mathfrak{e}}}}{}{#2}{}{#1}}}
\newrobustcmd{\gravebmmathfrakez}[2][]{\ensuremath{\subp{\grave{\bm{\mathfrak{e}}}}{}{#2}{}{#1}}}
\newrobustcmd{\dotbmmathfrakez}[2][]{\ensuremath{\subp{\dot{\bm{\mathfrak{e}}}}{}{#2}{}{#1}}}
\newrobustcmd{\ddotbmmathfrakez}[2][]{\ensuremath{\subp{\ddot{\bm{\mathfrak{e}}}}{}{#2}{}{#1}}}
\newrobustcmd{\brevebmmathfrakez}[2][]{\ensuremath{\subp{\breve{\bm{\mathfrak{e}}}}{}{#2}{}{#1}}}
\newrobustcmd{\barbmmathfrakez}[2][]{\ensuremath{\subp{\bar{\bm{\mathfrak{e}}}}{}{#2}{}{#1}}}
\newrobustcmd{\vecbmmathfrakez}[2][]{\ensuremath{\subp{\vec{\bm{\mathfrak{e}}}}{}{#2}{}{#1}}}
\newrobustcmd{\mathfrakfz}[2][]{\ensuremath{\subp{\mathfrak{f}}{}{#2}{}{#1}}}
\newrobustcmd{\hatmathfrakfz}[2][]{\ensuremath{\subp{\hat{\mathfrak{f}}}{}{#2}{}{#1}}}
\newrobustcmd{\widehatmathfrakfz}[2][]{\ensuremath{\subp{\widehat{\mathfrak{f}}}{}{#2}{}{#1}}}
\newrobustcmd{\checkmathfrakfz}[2][]{\ensuremath{\subp{\check{\mathfrak{f}}}{}{#2}{}{#1}}}
\newrobustcmd{\tildemathfrakfz}[2][]{\ensuremath{\subp{\tilde{\mathfrak{f}}}{}{#2}{}{#1}}}
\newrobustcmd{\widetildemathfrakfz}[2][]{\ensuremath{\subp{\widetilde{\mathfrak{f}}}{}{#2}{}{#1}}}
\newrobustcmd{\acutemathfrakfz}[2][]{\ensuremath{\subp{\acute{\mathfrak{f}}}{}{#2}{}{#1}}}
\newrobustcmd{\gravemathfrakfz}[2][]{\ensuremath{\subp{\grave{\mathfrak{f}}}{}{#2}{}{#1}}}
\newrobustcmd{\dotmathfrakfz}[2][]{\ensuremath{\subp{\dot{\mathfrak{f}}}{}{#2}{}{#1}}}
\newrobustcmd{\ddotmathfrakfz}[2][]{\ensuremath{\subp{\ddot{\mathfrak{f}}}{}{#2}{}{#1}}}
\newrobustcmd{\brevemathfrakfz}[2][]{\ensuremath{\subp{\breve{\mathfrak{f}}}{}{#2}{}{#1}}}
\newrobustcmd{\barmathfrakfz}[2][]{\ensuremath{\subp{\bar{\mathfrak{f}}}{}{#2}{}{#1}}}
\newrobustcmd{\vecmathfrakfz}[2][]{\ensuremath{\subp{\vec{\mathfrak{f}}}{}{#2}{}{#1}}}
\newrobustcmd{\bmmathfrakfz}[2][]{\ensuremath{\subp{\bm{\mathfrak{f}}}{}{#2}{}{#1}}}
\newrobustcmd{\hatbmmathfrakfz}[2][]{\ensuremath{\subp{\hat{\bm{\mathfrak{f}}}}{}{#2}{}{#1}}}
\newrobustcmd{\widehatbmmathfrakfz}[2][]{\ensuremath{\subp{\widehat{\bm{\mathfrak{f}}}}{}{#2}{}{#1}}}
\newrobustcmd{\checkbmmathfrakfz}[2][]{\ensuremath{\subp{\check{\bm{\mathfrak{f}}}}{}{#2}{}{#1}}}
\newrobustcmd{\tildebmmathfrakfz}[2][]{\ensuremath{\subp{\tilde{\bm{\mathfrak{f}}}}{}{#2}{}{#1}}}
\newrobustcmd{\widetildebmmathfrakfz}[2][]{\ensuremath{\subp{\widetilde{\bm{\mathfrak{f}}}}{}{#2}{}{#1}}}
\newrobustcmd{\acutebmmathfrakfz}[2][]{\ensuremath{\subp{\acute{\bm{\mathfrak{f}}}}{}{#2}{}{#1}}}
\newrobustcmd{\gravebmmathfrakfz}[2][]{\ensuremath{\subp{\grave{\bm{\mathfrak{f}}}}{}{#2}{}{#1}}}
\newrobustcmd{\dotbmmathfrakfz}[2][]{\ensuremath{\subp{\dot{\bm{\mathfrak{f}}}}{}{#2}{}{#1}}}
\newrobustcmd{\ddotbmmathfrakfz}[2][]{\ensuremath{\subp{\ddot{\bm{\mathfrak{f}}}}{}{#2}{}{#1}}}
\newrobustcmd{\brevebmmathfrakfz}[2][]{\ensuremath{\subp{\breve{\bm{\mathfrak{f}}}}{}{#2}{}{#1}}}
\newrobustcmd{\barbmmathfrakfz}[2][]{\ensuremath{\subp{\bar{\bm{\mathfrak{f}}}}{}{#2}{}{#1}}}
\newrobustcmd{\vecbmmathfrakfz}[2][]{\ensuremath{\subp{\vec{\bm{\mathfrak{f}}}}{}{#2}{}{#1}}}
\newrobustcmd{\mathfrakgz}[2][]{\ensuremath{\subp{\mathfrak{g}}{}{#2}{}{#1}}}
\newrobustcmd{\hatmathfrakgz}[2][]{\ensuremath{\subp{\hat{\mathfrak{g}}}{}{#2}{}{#1}}}
\newrobustcmd{\widehatmathfrakgz}[2][]{\ensuremath{\subp{\widehat{\mathfrak{g}}}{}{#2}{}{#1}}}
\newrobustcmd{\checkmathfrakgz}[2][]{\ensuremath{\subp{\check{\mathfrak{g}}}{}{#2}{}{#1}}}
\newrobustcmd{\tildemathfrakgz}[2][]{\ensuremath{\subp{\tilde{\mathfrak{g}}}{}{#2}{}{#1}}}
\newrobustcmd{\widetildemathfrakgz}[2][]{\ensuremath{\subp{\widetilde{\mathfrak{g}}}{}{#2}{}{#1}}}
\newrobustcmd{\acutemathfrakgz}[2][]{\ensuremath{\subp{\acute{\mathfrak{g}}}{}{#2}{}{#1}}}
\newrobustcmd{\gravemathfrakgz}[2][]{\ensuremath{\subp{\grave{\mathfrak{g}}}{}{#2}{}{#1}}}
\newrobustcmd{\dotmathfrakgz}[2][]{\ensuremath{\subp{\dot{\mathfrak{g}}}{}{#2}{}{#1}}}
\newrobustcmd{\ddotmathfrakgz}[2][]{\ensuremath{\subp{\ddot{\mathfrak{g}}}{}{#2}{}{#1}}}
\newrobustcmd{\brevemathfrakgz}[2][]{\ensuremath{\subp{\breve{\mathfrak{g}}}{}{#2}{}{#1}}}
\newrobustcmd{\barmathfrakgz}[2][]{\ensuremath{\subp{\bar{\mathfrak{g}}}{}{#2}{}{#1}}}
\newrobustcmd{\vecmathfrakgz}[2][]{\ensuremath{\subp{\vec{\mathfrak{g}}}{}{#2}{}{#1}}}
\newrobustcmd{\bmmathfrakgz}[2][]{\ensuremath{\subp{\bm{\mathfrak{g}}}{}{#2}{}{#1}}}
\newrobustcmd{\hatbmmathfrakgz}[2][]{\ensuremath{\subp{\hat{\bm{\mathfrak{g}}}}{}{#2}{}{#1}}}
\newrobustcmd{\widehatbmmathfrakgz}[2][]{\ensuremath{\subp{\widehat{\bm{\mathfrak{g}}}}{}{#2}{}{#1}}}
\newrobustcmd{\checkbmmathfrakgz}[2][]{\ensuremath{\subp{\check{\bm{\mathfrak{g}}}}{}{#2}{}{#1}}}
\newrobustcmd{\tildebmmathfrakgz}[2][]{\ensuremath{\subp{\tilde{\bm{\mathfrak{g}}}}{}{#2}{}{#1}}}
\newrobustcmd{\widetildebmmathfrakgz}[2][]{\ensuremath{\subp{\widetilde{\bm{\mathfrak{g}}}}{}{#2}{}{#1}}}
\newrobustcmd{\acutebmmathfrakgz}[2][]{\ensuremath{\subp{\acute{\bm{\mathfrak{g}}}}{}{#2}{}{#1}}}
\newrobustcmd{\gravebmmathfrakgz}[2][]{\ensuremath{\subp{\grave{\bm{\mathfrak{g}}}}{}{#2}{}{#1}}}
\newrobustcmd{\dotbmmathfrakgz}[2][]{\ensuremath{\subp{\dot{\bm{\mathfrak{g}}}}{}{#2}{}{#1}}}
\newrobustcmd{\ddotbmmathfrakgz}[2][]{\ensuremath{\subp{\ddot{\bm{\mathfrak{g}}}}{}{#2}{}{#1}}}
\newrobustcmd{\brevebmmathfrakgz}[2][]{\ensuremath{\subp{\breve{\bm{\mathfrak{g}}}}{}{#2}{}{#1}}}
\newrobustcmd{\barbmmathfrakgz}[2][]{\ensuremath{\subp{\bar{\bm{\mathfrak{g}}}}{}{#2}{}{#1}}}
\newrobustcmd{\vecbmmathfrakgz}[2][]{\ensuremath{\subp{\vec{\bm{\mathfrak{g}}}}{}{#2}{}{#1}}}
\newrobustcmd{\mathfrakhz}[2][]{\ensuremath{\subp{\mathfrak{h}}{}{#2}{}{#1}}}
\newrobustcmd{\hatmathfrakhz}[2][]{\ensuremath{\subp{\hat{\mathfrak{h}}}{}{#2}{}{#1}}}
\newrobustcmd{\widehatmathfrakhz}[2][]{\ensuremath{\subp{\widehat{\mathfrak{h}}}{}{#2}{}{#1}}}
\newrobustcmd{\checkmathfrakhz}[2][]{\ensuremath{\subp{\check{\mathfrak{h}}}{}{#2}{}{#1}}}
\newrobustcmd{\tildemathfrakhz}[2][]{\ensuremath{\subp{\tilde{\mathfrak{h}}}{}{#2}{}{#1}}}
\newrobustcmd{\widetildemathfrakhz}[2][]{\ensuremath{\subp{\widetilde{\mathfrak{h}}}{}{#2}{}{#1}}}
\newrobustcmd{\acutemathfrakhz}[2][]{\ensuremath{\subp{\acute{\mathfrak{h}}}{}{#2}{}{#1}}}
\newrobustcmd{\gravemathfrakhz}[2][]{\ensuremath{\subp{\grave{\mathfrak{h}}}{}{#2}{}{#1}}}
\newrobustcmd{\dotmathfrakhz}[2][]{\ensuremath{\subp{\dot{\mathfrak{h}}}{}{#2}{}{#1}}}
\newrobustcmd{\ddotmathfrakhz}[2][]{\ensuremath{\subp{\ddot{\mathfrak{h}}}{}{#2}{}{#1}}}
\newrobustcmd{\brevemathfrakhz}[2][]{\ensuremath{\subp{\breve{\mathfrak{h}}}{}{#2}{}{#1}}}
\newrobustcmd{\barmathfrakhz}[2][]{\ensuremath{\subp{\bar{\mathfrak{h}}}{}{#2}{}{#1}}}
\newrobustcmd{\vecmathfrakhz}[2][]{\ensuremath{\subp{\vec{\mathfrak{h}}}{}{#2}{}{#1}}}
\newrobustcmd{\bmmathfrakhz}[2][]{\ensuremath{\subp{\bm{\mathfrak{h}}}{}{#2}{}{#1}}}
\newrobustcmd{\hatbmmathfrakhz}[2][]{\ensuremath{\subp{\hat{\bm{\mathfrak{h}}}}{}{#2}{}{#1}}}
\newrobustcmd{\widehatbmmathfrakhz}[2][]{\ensuremath{\subp{\widehat{\bm{\mathfrak{h}}}}{}{#2}{}{#1}}}
\newrobustcmd{\checkbmmathfrakhz}[2][]{\ensuremath{\subp{\check{\bm{\mathfrak{h}}}}{}{#2}{}{#1}}}
\newrobustcmd{\tildebmmathfrakhz}[2][]{\ensuremath{\subp{\tilde{\bm{\mathfrak{h}}}}{}{#2}{}{#1}}}
\newrobustcmd{\widetildebmmathfrakhz}[2][]{\ensuremath{\subp{\widetilde{\bm{\mathfrak{h}}}}{}{#2}{}{#1}}}
\newrobustcmd{\acutebmmathfrakhz}[2][]{\ensuremath{\subp{\acute{\bm{\mathfrak{h}}}}{}{#2}{}{#1}}}
\newrobustcmd{\gravebmmathfrakhz}[2][]{\ensuremath{\subp{\grave{\bm{\mathfrak{h}}}}{}{#2}{}{#1}}}
\newrobustcmd{\dotbmmathfrakhz}[2][]{\ensuremath{\subp{\dot{\bm{\mathfrak{h}}}}{}{#2}{}{#1}}}
\newrobustcmd{\ddotbmmathfrakhz}[2][]{\ensuremath{\subp{\ddot{\bm{\mathfrak{h}}}}{}{#2}{}{#1}}}
\newrobustcmd{\brevebmmathfrakhz}[2][]{\ensuremath{\subp{\breve{\bm{\mathfrak{h}}}}{}{#2}{}{#1}}}
\newrobustcmd{\barbmmathfrakhz}[2][]{\ensuremath{\subp{\bar{\bm{\mathfrak{h}}}}{}{#2}{}{#1}}}
\newrobustcmd{\vecbmmathfrakhz}[2][]{\ensuremath{\subp{\vec{\bm{\mathfrak{h}}}}{}{#2}{}{#1}}}
\newrobustcmd{\mathfrakiz}[2][]{\ensuremath{\subp{\mathfrak{i}}{}{#2}{}{#1}}}
\newrobustcmd{\hatmathfrakiz}[2][]{\ensuremath{\subp{\hat{\mathfrak{i}}}{}{#2}{}{#1}}}
\newrobustcmd{\widehatmathfrakiz}[2][]{\ensuremath{\subp{\widehat{\mathfrak{i}}}{}{#2}{}{#1}}}
\newrobustcmd{\checkmathfrakiz}[2][]{\ensuremath{\subp{\check{\mathfrak{i}}}{}{#2}{}{#1}}}
\newrobustcmd{\tildemathfrakiz}[2][]{\ensuremath{\subp{\tilde{\mathfrak{i}}}{}{#2}{}{#1}}}
\newrobustcmd{\widetildemathfrakiz}[2][]{\ensuremath{\subp{\widetilde{\mathfrak{i}}}{}{#2}{}{#1}}}
\newrobustcmd{\acutemathfrakiz}[2][]{\ensuremath{\subp{\acute{\mathfrak{i}}}{}{#2}{}{#1}}}
\newrobustcmd{\gravemathfrakiz}[2][]{\ensuremath{\subp{\grave{\mathfrak{i}}}{}{#2}{}{#1}}}
\newrobustcmd{\dotmathfrakiz}[2][]{\ensuremath{\subp{\dot{\mathfrak{i}}}{}{#2}{}{#1}}}
\newrobustcmd{\ddotmathfrakiz}[2][]{\ensuremath{\subp{\ddot{\mathfrak{i}}}{}{#2}{}{#1}}}
\newrobustcmd{\brevemathfrakiz}[2][]{\ensuremath{\subp{\breve{\mathfrak{i}}}{}{#2}{}{#1}}}
\newrobustcmd{\barmathfrakiz}[2][]{\ensuremath{\subp{\bar{\mathfrak{i}}}{}{#2}{}{#1}}}
\newrobustcmd{\vecmathfrakiz}[2][]{\ensuremath{\subp{\vec{\mathfrak{i}}}{}{#2}{}{#1}}}
\newrobustcmd{\bmmathfrakiz}[2][]{\ensuremath{\subp{\bm{\mathfrak{i}}}{}{#2}{}{#1}}}
\newrobustcmd{\hatbmmathfrakiz}[2][]{\ensuremath{\subp{\hat{\bm{\mathfrak{i}}}}{}{#2}{}{#1}}}
\newrobustcmd{\widehatbmmathfrakiz}[2][]{\ensuremath{\subp{\widehat{\bm{\mathfrak{i}}}}{}{#2}{}{#1}}}
\newrobustcmd{\checkbmmathfrakiz}[2][]{\ensuremath{\subp{\check{\bm{\mathfrak{i}}}}{}{#2}{}{#1}}}
\newrobustcmd{\tildebmmathfrakiz}[2][]{\ensuremath{\subp{\tilde{\bm{\mathfrak{i}}}}{}{#2}{}{#1}}}
\newrobustcmd{\widetildebmmathfrakiz}[2][]{\ensuremath{\subp{\widetilde{\bm{\mathfrak{i}}}}{}{#2}{}{#1}}}
\newrobustcmd{\acutebmmathfrakiz}[2][]{\ensuremath{\subp{\acute{\bm{\mathfrak{i}}}}{}{#2}{}{#1}}}
\newrobustcmd{\gravebmmathfrakiz}[2][]{\ensuremath{\subp{\grave{\bm{\mathfrak{i}}}}{}{#2}{}{#1}}}
\newrobustcmd{\dotbmmathfrakiz}[2][]{\ensuremath{\subp{\dot{\bm{\mathfrak{i}}}}{}{#2}{}{#1}}}
\newrobustcmd{\ddotbmmathfrakiz}[2][]{\ensuremath{\subp{\ddot{\bm{\mathfrak{i}}}}{}{#2}{}{#1}}}
\newrobustcmd{\brevebmmathfrakiz}[2][]{\ensuremath{\subp{\breve{\bm{\mathfrak{i}}}}{}{#2}{}{#1}}}
\newrobustcmd{\barbmmathfrakiz}[2][]{\ensuremath{\subp{\bar{\bm{\mathfrak{i}}}}{}{#2}{}{#1}}}
\newrobustcmd{\vecbmmathfrakiz}[2][]{\ensuremath{\subp{\vec{\bm{\mathfrak{i}}}}{}{#2}{}{#1}}}
\newrobustcmd{\mathfrakjz}[2][]{\ensuremath{\subp{\mathfrak{j}}{}{#2}{}{#1}}}
\newrobustcmd{\hatmathfrakjz}[2][]{\ensuremath{\subp{\hat{\mathfrak{j}}}{}{#2}{}{#1}}}
\newrobustcmd{\widehatmathfrakjz}[2][]{\ensuremath{\subp{\widehat{\mathfrak{j}}}{}{#2}{}{#1}}}
\newrobustcmd{\checkmathfrakjz}[2][]{\ensuremath{\subp{\check{\mathfrak{j}}}{}{#2}{}{#1}}}
\newrobustcmd{\tildemathfrakjz}[2][]{\ensuremath{\subp{\tilde{\mathfrak{j}}}{}{#2}{}{#1}}}
\newrobustcmd{\widetildemathfrakjz}[2][]{\ensuremath{\subp{\widetilde{\mathfrak{j}}}{}{#2}{}{#1}}}
\newrobustcmd{\acutemathfrakjz}[2][]{\ensuremath{\subp{\acute{\mathfrak{j}}}{}{#2}{}{#1}}}
\newrobustcmd{\gravemathfrakjz}[2][]{\ensuremath{\subp{\grave{\mathfrak{j}}}{}{#2}{}{#1}}}
\newrobustcmd{\dotmathfrakjz}[2][]{\ensuremath{\subp{\dot{\mathfrak{j}}}{}{#2}{}{#1}}}
\newrobustcmd{\ddotmathfrakjz}[2][]{\ensuremath{\subp{\ddot{\mathfrak{j}}}{}{#2}{}{#1}}}
\newrobustcmd{\brevemathfrakjz}[2][]{\ensuremath{\subp{\breve{\mathfrak{j}}}{}{#2}{}{#1}}}
\newrobustcmd{\barmathfrakjz}[2][]{\ensuremath{\subp{\bar{\mathfrak{j}}}{}{#2}{}{#1}}}
\newrobustcmd{\vecmathfrakjz}[2][]{\ensuremath{\subp{\vec{\mathfrak{j}}}{}{#2}{}{#1}}}
\newrobustcmd{\bmmathfrakjz}[2][]{\ensuremath{\subp{\bm{\mathfrak{j}}}{}{#2}{}{#1}}}
\newrobustcmd{\hatbmmathfrakjz}[2][]{\ensuremath{\subp{\hat{\bm{\mathfrak{j}}}}{}{#2}{}{#1}}}
\newrobustcmd{\widehatbmmathfrakjz}[2][]{\ensuremath{\subp{\widehat{\bm{\mathfrak{j}}}}{}{#2}{}{#1}}}
\newrobustcmd{\checkbmmathfrakjz}[2][]{\ensuremath{\subp{\check{\bm{\mathfrak{j}}}}{}{#2}{}{#1}}}
\newrobustcmd{\tildebmmathfrakjz}[2][]{\ensuremath{\subp{\tilde{\bm{\mathfrak{j}}}}{}{#2}{}{#1}}}
\newrobustcmd{\widetildebmmathfrakjz}[2][]{\ensuremath{\subp{\widetilde{\bm{\mathfrak{j}}}}{}{#2}{}{#1}}}
\newrobustcmd{\acutebmmathfrakjz}[2][]{\ensuremath{\subp{\acute{\bm{\mathfrak{j}}}}{}{#2}{}{#1}}}
\newrobustcmd{\gravebmmathfrakjz}[2][]{\ensuremath{\subp{\grave{\bm{\mathfrak{j}}}}{}{#2}{}{#1}}}
\newrobustcmd{\dotbmmathfrakjz}[2][]{\ensuremath{\subp{\dot{\bm{\mathfrak{j}}}}{}{#2}{}{#1}}}
\newrobustcmd{\ddotbmmathfrakjz}[2][]{\ensuremath{\subp{\ddot{\bm{\mathfrak{j}}}}{}{#2}{}{#1}}}
\newrobustcmd{\brevebmmathfrakjz}[2][]{\ensuremath{\subp{\breve{\bm{\mathfrak{j}}}}{}{#2}{}{#1}}}
\newrobustcmd{\barbmmathfrakjz}[2][]{\ensuremath{\subp{\bar{\bm{\mathfrak{j}}}}{}{#2}{}{#1}}}
\newrobustcmd{\vecbmmathfrakjz}[2][]{\ensuremath{\subp{\vec{\bm{\mathfrak{j}}}}{}{#2}{}{#1}}}
\newrobustcmd{\mathfrakkz}[2][]{\ensuremath{\subp{\mathfrak{k}}{}{#2}{}{#1}}}
\newrobustcmd{\hatmathfrakkz}[2][]{\ensuremath{\subp{\hat{\mathfrak{k}}}{}{#2}{}{#1}}}
\newrobustcmd{\widehatmathfrakkz}[2][]{\ensuremath{\subp{\widehat{\mathfrak{k}}}{}{#2}{}{#1}}}
\newrobustcmd{\checkmathfrakkz}[2][]{\ensuremath{\subp{\check{\mathfrak{k}}}{}{#2}{}{#1}}}
\newrobustcmd{\tildemathfrakkz}[2][]{\ensuremath{\subp{\tilde{\mathfrak{k}}}{}{#2}{}{#1}}}
\newrobustcmd{\widetildemathfrakkz}[2][]{\ensuremath{\subp{\widetilde{\mathfrak{k}}}{}{#2}{}{#1}}}
\newrobustcmd{\acutemathfrakkz}[2][]{\ensuremath{\subp{\acute{\mathfrak{k}}}{}{#2}{}{#1}}}
\newrobustcmd{\gravemathfrakkz}[2][]{\ensuremath{\subp{\grave{\mathfrak{k}}}{}{#2}{}{#1}}}
\newrobustcmd{\dotmathfrakkz}[2][]{\ensuremath{\subp{\dot{\mathfrak{k}}}{}{#2}{}{#1}}}
\newrobustcmd{\ddotmathfrakkz}[2][]{\ensuremath{\subp{\ddot{\mathfrak{k}}}{}{#2}{}{#1}}}
\newrobustcmd{\brevemathfrakkz}[2][]{\ensuremath{\subp{\breve{\mathfrak{k}}}{}{#2}{}{#1}}}
\newrobustcmd{\barmathfrakkz}[2][]{\ensuremath{\subp{\bar{\mathfrak{k}}}{}{#2}{}{#1}}}
\newrobustcmd{\vecmathfrakkz}[2][]{\ensuremath{\subp{\vec{\mathfrak{k}}}{}{#2}{}{#1}}}
\newrobustcmd{\bmmathfrakkz}[2][]{\ensuremath{\subp{\bm{\mathfrak{k}}}{}{#2}{}{#1}}}
\newrobustcmd{\hatbmmathfrakkz}[2][]{\ensuremath{\subp{\hat{\bm{\mathfrak{k}}}}{}{#2}{}{#1}}}
\newrobustcmd{\widehatbmmathfrakkz}[2][]{\ensuremath{\subp{\widehat{\bm{\mathfrak{k}}}}{}{#2}{}{#1}}}
\newrobustcmd{\checkbmmathfrakkz}[2][]{\ensuremath{\subp{\check{\bm{\mathfrak{k}}}}{}{#2}{}{#1}}}
\newrobustcmd{\tildebmmathfrakkz}[2][]{\ensuremath{\subp{\tilde{\bm{\mathfrak{k}}}}{}{#2}{}{#1}}}
\newrobustcmd{\widetildebmmathfrakkz}[2][]{\ensuremath{\subp{\widetilde{\bm{\mathfrak{k}}}}{}{#2}{}{#1}}}
\newrobustcmd{\acutebmmathfrakkz}[2][]{\ensuremath{\subp{\acute{\bm{\mathfrak{k}}}}{}{#2}{}{#1}}}
\newrobustcmd{\gravebmmathfrakkz}[2][]{\ensuremath{\subp{\grave{\bm{\mathfrak{k}}}}{}{#2}{}{#1}}}
\newrobustcmd{\dotbmmathfrakkz}[2][]{\ensuremath{\subp{\dot{\bm{\mathfrak{k}}}}{}{#2}{}{#1}}}
\newrobustcmd{\ddotbmmathfrakkz}[2][]{\ensuremath{\subp{\ddot{\bm{\mathfrak{k}}}}{}{#2}{}{#1}}}
\newrobustcmd{\brevebmmathfrakkz}[2][]{\ensuremath{\subp{\breve{\bm{\mathfrak{k}}}}{}{#2}{}{#1}}}
\newrobustcmd{\barbmmathfrakkz}[2][]{\ensuremath{\subp{\bar{\bm{\mathfrak{k}}}}{}{#2}{}{#1}}}
\newrobustcmd{\vecbmmathfrakkz}[2][]{\ensuremath{\subp{\vec{\bm{\mathfrak{k}}}}{}{#2}{}{#1}}}
\newrobustcmd{\mathfraklz}[2][]{\ensuremath{\subp{\mathfrak{l}}{}{#2}{}{#1}}}
\newrobustcmd{\hatmathfraklz}[2][]{\ensuremath{\subp{\hat{\mathfrak{l}}}{}{#2}{}{#1}}}
\newrobustcmd{\widehatmathfraklz}[2][]{\ensuremath{\subp{\widehat{\mathfrak{l}}}{}{#2}{}{#1}}}
\newrobustcmd{\checkmathfraklz}[2][]{\ensuremath{\subp{\check{\mathfrak{l}}}{}{#2}{}{#1}}}
\newrobustcmd{\tildemathfraklz}[2][]{\ensuremath{\subp{\tilde{\mathfrak{l}}}{}{#2}{}{#1}}}
\newrobustcmd{\widetildemathfraklz}[2][]{\ensuremath{\subp{\widetilde{\mathfrak{l}}}{}{#2}{}{#1}}}
\newrobustcmd{\acutemathfraklz}[2][]{\ensuremath{\subp{\acute{\mathfrak{l}}}{}{#2}{}{#1}}}
\newrobustcmd{\gravemathfraklz}[2][]{\ensuremath{\subp{\grave{\mathfrak{l}}}{}{#2}{}{#1}}}
\newrobustcmd{\dotmathfraklz}[2][]{\ensuremath{\subp{\dot{\mathfrak{l}}}{}{#2}{}{#1}}}
\newrobustcmd{\ddotmathfraklz}[2][]{\ensuremath{\subp{\ddot{\mathfrak{l}}}{}{#2}{}{#1}}}
\newrobustcmd{\brevemathfraklz}[2][]{\ensuremath{\subp{\breve{\mathfrak{l}}}{}{#2}{}{#1}}}
\newrobustcmd{\barmathfraklz}[2][]{\ensuremath{\subp{\bar{\mathfrak{l}}}{}{#2}{}{#1}}}
\newrobustcmd{\vecmathfraklz}[2][]{\ensuremath{\subp{\vec{\mathfrak{l}}}{}{#2}{}{#1}}}
\newrobustcmd{\bmmathfraklz}[2][]{\ensuremath{\subp{\bm{\mathfrak{l}}}{}{#2}{}{#1}}}
\newrobustcmd{\hatbmmathfraklz}[2][]{\ensuremath{\subp{\hat{\bm{\mathfrak{l}}}}{}{#2}{}{#1}}}
\newrobustcmd{\widehatbmmathfraklz}[2][]{\ensuremath{\subp{\widehat{\bm{\mathfrak{l}}}}{}{#2}{}{#1}}}
\newrobustcmd{\checkbmmathfraklz}[2][]{\ensuremath{\subp{\check{\bm{\mathfrak{l}}}}{}{#2}{}{#1}}}
\newrobustcmd{\tildebmmathfraklz}[2][]{\ensuremath{\subp{\tilde{\bm{\mathfrak{l}}}}{}{#2}{}{#1}}}
\newrobustcmd{\widetildebmmathfraklz}[2][]{\ensuremath{\subp{\widetilde{\bm{\mathfrak{l}}}}{}{#2}{}{#1}}}
\newrobustcmd{\acutebmmathfraklz}[2][]{\ensuremath{\subp{\acute{\bm{\mathfrak{l}}}}{}{#2}{}{#1}}}
\newrobustcmd{\gravebmmathfraklz}[2][]{\ensuremath{\subp{\grave{\bm{\mathfrak{l}}}}{}{#2}{}{#1}}}
\newrobustcmd{\dotbmmathfraklz}[2][]{\ensuremath{\subp{\dot{\bm{\mathfrak{l}}}}{}{#2}{}{#1}}}
\newrobustcmd{\ddotbmmathfraklz}[2][]{\ensuremath{\subp{\ddot{\bm{\mathfrak{l}}}}{}{#2}{}{#1}}}
\newrobustcmd{\brevebmmathfraklz}[2][]{\ensuremath{\subp{\breve{\bm{\mathfrak{l}}}}{}{#2}{}{#1}}}
\newrobustcmd{\barbmmathfraklz}[2][]{\ensuremath{\subp{\bar{\bm{\mathfrak{l}}}}{}{#2}{}{#1}}}
\newrobustcmd{\vecbmmathfraklz}[2][]{\ensuremath{\subp{\vec{\bm{\mathfrak{l}}}}{}{#2}{}{#1}}}
\newrobustcmd{\mathfrakmz}[2][]{\ensuremath{\subp{\mathfrak{m}}{}{#2}{}{#1}}}
\newrobustcmd{\hatmathfrakmz}[2][]{\ensuremath{\subp{\hat{\mathfrak{m}}}{}{#2}{}{#1}}}
\newrobustcmd{\widehatmathfrakmz}[2][]{\ensuremath{\subp{\widehat{\mathfrak{m}}}{}{#2}{}{#1}}}
\newrobustcmd{\checkmathfrakmz}[2][]{\ensuremath{\subp{\check{\mathfrak{m}}}{}{#2}{}{#1}}}
\newrobustcmd{\tildemathfrakmz}[2][]{\ensuremath{\subp{\tilde{\mathfrak{m}}}{}{#2}{}{#1}}}
\newrobustcmd{\widetildemathfrakmz}[2][]{\ensuremath{\subp{\widetilde{\mathfrak{m}}}{}{#2}{}{#1}}}
\newrobustcmd{\acutemathfrakmz}[2][]{\ensuremath{\subp{\acute{\mathfrak{m}}}{}{#2}{}{#1}}}
\newrobustcmd{\gravemathfrakmz}[2][]{\ensuremath{\subp{\grave{\mathfrak{m}}}{}{#2}{}{#1}}}
\newrobustcmd{\dotmathfrakmz}[2][]{\ensuremath{\subp{\dot{\mathfrak{m}}}{}{#2}{}{#1}}}
\newrobustcmd{\ddotmathfrakmz}[2][]{\ensuremath{\subp{\ddot{\mathfrak{m}}}{}{#2}{}{#1}}}
\newrobustcmd{\brevemathfrakmz}[2][]{\ensuremath{\subp{\breve{\mathfrak{m}}}{}{#2}{}{#1}}}
\newrobustcmd{\barmathfrakmz}[2][]{\ensuremath{\subp{\bar{\mathfrak{m}}}{}{#2}{}{#1}}}
\newrobustcmd{\vecmathfrakmz}[2][]{\ensuremath{\subp{\vec{\mathfrak{m}}}{}{#2}{}{#1}}}
\newrobustcmd{\bmmathfrakmz}[2][]{\ensuremath{\subp{\bm{\mathfrak{m}}}{}{#2}{}{#1}}}
\newrobustcmd{\hatbmmathfrakmz}[2][]{\ensuremath{\subp{\hat{\bm{\mathfrak{m}}}}{}{#2}{}{#1}}}
\newrobustcmd{\widehatbmmathfrakmz}[2][]{\ensuremath{\subp{\widehat{\bm{\mathfrak{m}}}}{}{#2}{}{#1}}}
\newrobustcmd{\checkbmmathfrakmz}[2][]{\ensuremath{\subp{\check{\bm{\mathfrak{m}}}}{}{#2}{}{#1}}}
\newrobustcmd{\tildebmmathfrakmz}[2][]{\ensuremath{\subp{\tilde{\bm{\mathfrak{m}}}}{}{#2}{}{#1}}}
\newrobustcmd{\widetildebmmathfrakmz}[2][]{\ensuremath{\subp{\widetilde{\bm{\mathfrak{m}}}}{}{#2}{}{#1}}}
\newrobustcmd{\acutebmmathfrakmz}[2][]{\ensuremath{\subp{\acute{\bm{\mathfrak{m}}}}{}{#2}{}{#1}}}
\newrobustcmd{\gravebmmathfrakmz}[2][]{\ensuremath{\subp{\grave{\bm{\mathfrak{m}}}}{}{#2}{}{#1}}}
\newrobustcmd{\dotbmmathfrakmz}[2][]{\ensuremath{\subp{\dot{\bm{\mathfrak{m}}}}{}{#2}{}{#1}}}
\newrobustcmd{\ddotbmmathfrakmz}[2][]{\ensuremath{\subp{\ddot{\bm{\mathfrak{m}}}}{}{#2}{}{#1}}}
\newrobustcmd{\brevebmmathfrakmz}[2][]{\ensuremath{\subp{\breve{\bm{\mathfrak{m}}}}{}{#2}{}{#1}}}
\newrobustcmd{\barbmmathfrakmz}[2][]{\ensuremath{\subp{\bar{\bm{\mathfrak{m}}}}{}{#2}{}{#1}}}
\newrobustcmd{\vecbmmathfrakmz}[2][]{\ensuremath{\subp{\vec{\bm{\mathfrak{m}}}}{}{#2}{}{#1}}}
\newrobustcmd{\mathfraknz}[2][]{\ensuremath{\subp{\mathfrak{n}}{}{#2}{}{#1}}}
\newrobustcmd{\hatmathfraknz}[2][]{\ensuremath{\subp{\hat{\mathfrak{n}}}{}{#2}{}{#1}}}
\newrobustcmd{\widehatmathfraknz}[2][]{\ensuremath{\subp{\widehat{\mathfrak{n}}}{}{#2}{}{#1}}}
\newrobustcmd{\checkmathfraknz}[2][]{\ensuremath{\subp{\check{\mathfrak{n}}}{}{#2}{}{#1}}}
\newrobustcmd{\tildemathfraknz}[2][]{\ensuremath{\subp{\tilde{\mathfrak{n}}}{}{#2}{}{#1}}}
\newrobustcmd{\widetildemathfraknz}[2][]{\ensuremath{\subp{\widetilde{\mathfrak{n}}}{}{#2}{}{#1}}}
\newrobustcmd{\acutemathfraknz}[2][]{\ensuremath{\subp{\acute{\mathfrak{n}}}{}{#2}{}{#1}}}
\newrobustcmd{\gravemathfraknz}[2][]{\ensuremath{\subp{\grave{\mathfrak{n}}}{}{#2}{}{#1}}}
\newrobustcmd{\dotmathfraknz}[2][]{\ensuremath{\subp{\dot{\mathfrak{n}}}{}{#2}{}{#1}}}
\newrobustcmd{\ddotmathfraknz}[2][]{\ensuremath{\subp{\ddot{\mathfrak{n}}}{}{#2}{}{#1}}}
\newrobustcmd{\brevemathfraknz}[2][]{\ensuremath{\subp{\breve{\mathfrak{n}}}{}{#2}{}{#1}}}
\newrobustcmd{\barmathfraknz}[2][]{\ensuremath{\subp{\bar{\mathfrak{n}}}{}{#2}{}{#1}}}
\newrobustcmd{\vecmathfraknz}[2][]{\ensuremath{\subp{\vec{\mathfrak{n}}}{}{#2}{}{#1}}}
\newrobustcmd{\bmmathfraknz}[2][]{\ensuremath{\subp{\bm{\mathfrak{n}}}{}{#2}{}{#1}}}
\newrobustcmd{\hatbmmathfraknz}[2][]{\ensuremath{\subp{\hat{\bm{\mathfrak{n}}}}{}{#2}{}{#1}}}
\newrobustcmd{\widehatbmmathfraknz}[2][]{\ensuremath{\subp{\widehat{\bm{\mathfrak{n}}}}{}{#2}{}{#1}}}
\newrobustcmd{\checkbmmathfraknz}[2][]{\ensuremath{\subp{\check{\bm{\mathfrak{n}}}}{}{#2}{}{#1}}}
\newrobustcmd{\tildebmmathfraknz}[2][]{\ensuremath{\subp{\tilde{\bm{\mathfrak{n}}}}{}{#2}{}{#1}}}
\newrobustcmd{\widetildebmmathfraknz}[2][]{\ensuremath{\subp{\widetilde{\bm{\mathfrak{n}}}}{}{#2}{}{#1}}}
\newrobustcmd{\acutebmmathfraknz}[2][]{\ensuremath{\subp{\acute{\bm{\mathfrak{n}}}}{}{#2}{}{#1}}}
\newrobustcmd{\gravebmmathfraknz}[2][]{\ensuremath{\subp{\grave{\bm{\mathfrak{n}}}}{}{#2}{}{#1}}}
\newrobustcmd{\dotbmmathfraknz}[2][]{\ensuremath{\subp{\dot{\bm{\mathfrak{n}}}}{}{#2}{}{#1}}}
\newrobustcmd{\ddotbmmathfraknz}[2][]{\ensuremath{\subp{\ddot{\bm{\mathfrak{n}}}}{}{#2}{}{#1}}}
\newrobustcmd{\brevebmmathfraknz}[2][]{\ensuremath{\subp{\breve{\bm{\mathfrak{n}}}}{}{#2}{}{#1}}}
\newrobustcmd{\barbmmathfraknz}[2][]{\ensuremath{\subp{\bar{\bm{\mathfrak{n}}}}{}{#2}{}{#1}}}
\newrobustcmd{\vecbmmathfraknz}[2][]{\ensuremath{\subp{\vec{\bm{\mathfrak{n}}}}{}{#2}{}{#1}}}
\newrobustcmd{\mathfrakoz}[2][]{\ensuremath{\subp{\mathfrak{o}}{}{#2}{}{#1}}}
\newrobustcmd{\hatmathfrakoz}[2][]{\ensuremath{\subp{\hat{\mathfrak{o}}}{}{#2}{}{#1}}}
\newrobustcmd{\widehatmathfrakoz}[2][]{\ensuremath{\subp{\widehat{\mathfrak{o}}}{}{#2}{}{#1}}}
\newrobustcmd{\checkmathfrakoz}[2][]{\ensuremath{\subp{\check{\mathfrak{o}}}{}{#2}{}{#1}}}
\newrobustcmd{\tildemathfrakoz}[2][]{\ensuremath{\subp{\tilde{\mathfrak{o}}}{}{#2}{}{#1}}}
\newrobustcmd{\widetildemathfrakoz}[2][]{\ensuremath{\subp{\widetilde{\mathfrak{o}}}{}{#2}{}{#1}}}
\newrobustcmd{\acutemathfrakoz}[2][]{\ensuremath{\subp{\acute{\mathfrak{o}}}{}{#2}{}{#1}}}
\newrobustcmd{\gravemathfrakoz}[2][]{\ensuremath{\subp{\grave{\mathfrak{o}}}{}{#2}{}{#1}}}
\newrobustcmd{\dotmathfrakoz}[2][]{\ensuremath{\subp{\dot{\mathfrak{o}}}{}{#2}{}{#1}}}
\newrobustcmd{\ddotmathfrakoz}[2][]{\ensuremath{\subp{\ddot{\mathfrak{o}}}{}{#2}{}{#1}}}
\newrobustcmd{\brevemathfrakoz}[2][]{\ensuremath{\subp{\breve{\mathfrak{o}}}{}{#2}{}{#1}}}
\newrobustcmd{\barmathfrakoz}[2][]{\ensuremath{\subp{\bar{\mathfrak{o}}}{}{#2}{}{#1}}}
\newrobustcmd{\vecmathfrakoz}[2][]{\ensuremath{\subp{\vec{\mathfrak{o}}}{}{#2}{}{#1}}}
\newrobustcmd{\bmmathfrakoz}[2][]{\ensuremath{\subp{\bm{\mathfrak{o}}}{}{#2}{}{#1}}}
\newrobustcmd{\hatbmmathfrakoz}[2][]{\ensuremath{\subp{\hat{\bm{\mathfrak{o}}}}{}{#2}{}{#1}}}
\newrobustcmd{\widehatbmmathfrakoz}[2][]{\ensuremath{\subp{\widehat{\bm{\mathfrak{o}}}}{}{#2}{}{#1}}}
\newrobustcmd{\checkbmmathfrakoz}[2][]{\ensuremath{\subp{\check{\bm{\mathfrak{o}}}}{}{#2}{}{#1}}}
\newrobustcmd{\tildebmmathfrakoz}[2][]{\ensuremath{\subp{\tilde{\bm{\mathfrak{o}}}}{}{#2}{}{#1}}}
\newrobustcmd{\widetildebmmathfrakoz}[2][]{\ensuremath{\subp{\widetilde{\bm{\mathfrak{o}}}}{}{#2}{}{#1}}}
\newrobustcmd{\acutebmmathfrakoz}[2][]{\ensuremath{\subp{\acute{\bm{\mathfrak{o}}}}{}{#2}{}{#1}}}
\newrobustcmd{\gravebmmathfrakoz}[2][]{\ensuremath{\subp{\grave{\bm{\mathfrak{o}}}}{}{#2}{}{#1}}}
\newrobustcmd{\dotbmmathfrakoz}[2][]{\ensuremath{\subp{\dot{\bm{\mathfrak{o}}}}{}{#2}{}{#1}}}
\newrobustcmd{\ddotbmmathfrakoz}[2][]{\ensuremath{\subp{\ddot{\bm{\mathfrak{o}}}}{}{#2}{}{#1}}}
\newrobustcmd{\brevebmmathfrakoz}[2][]{\ensuremath{\subp{\breve{\bm{\mathfrak{o}}}}{}{#2}{}{#1}}}
\newrobustcmd{\barbmmathfrakoz}[2][]{\ensuremath{\subp{\bar{\bm{\mathfrak{o}}}}{}{#2}{}{#1}}}
\newrobustcmd{\vecbmmathfrakoz}[2][]{\ensuremath{\subp{\vec{\bm{\mathfrak{o}}}}{}{#2}{}{#1}}}
\newrobustcmd{\mathfrakpz}[2][]{\ensuremath{\subp{\mathfrak{p}}{}{#2}{}{#1}}}
\newrobustcmd{\hatmathfrakpz}[2][]{\ensuremath{\subp{\hat{\mathfrak{p}}}{}{#2}{}{#1}}}
\newrobustcmd{\widehatmathfrakpz}[2][]{\ensuremath{\subp{\widehat{\mathfrak{p}}}{}{#2}{}{#1}}}
\newrobustcmd{\checkmathfrakpz}[2][]{\ensuremath{\subp{\check{\mathfrak{p}}}{}{#2}{}{#1}}}
\newrobustcmd{\tildemathfrakpz}[2][]{\ensuremath{\subp{\tilde{\mathfrak{p}}}{}{#2}{}{#1}}}
\newrobustcmd{\widetildemathfrakpz}[2][]{\ensuremath{\subp{\widetilde{\mathfrak{p}}}{}{#2}{}{#1}}}
\newrobustcmd{\acutemathfrakpz}[2][]{\ensuremath{\subp{\acute{\mathfrak{p}}}{}{#2}{}{#1}}}
\newrobustcmd{\gravemathfrakpz}[2][]{\ensuremath{\subp{\grave{\mathfrak{p}}}{}{#2}{}{#1}}}
\newrobustcmd{\dotmathfrakpz}[2][]{\ensuremath{\subp{\dot{\mathfrak{p}}}{}{#2}{}{#1}}}
\newrobustcmd{\ddotmathfrakpz}[2][]{\ensuremath{\subp{\ddot{\mathfrak{p}}}{}{#2}{}{#1}}}
\newrobustcmd{\brevemathfrakpz}[2][]{\ensuremath{\subp{\breve{\mathfrak{p}}}{}{#2}{}{#1}}}
\newrobustcmd{\barmathfrakpz}[2][]{\ensuremath{\subp{\bar{\mathfrak{p}}}{}{#2}{}{#1}}}
\newrobustcmd{\vecmathfrakpz}[2][]{\ensuremath{\subp{\vec{\mathfrak{p}}}{}{#2}{}{#1}}}
\newrobustcmd{\bmmathfrakpz}[2][]{\ensuremath{\subp{\bm{\mathfrak{p}}}{}{#2}{}{#1}}}
\newrobustcmd{\hatbmmathfrakpz}[2][]{\ensuremath{\subp{\hat{\bm{\mathfrak{p}}}}{}{#2}{}{#1}}}
\newrobustcmd{\widehatbmmathfrakpz}[2][]{\ensuremath{\subp{\widehat{\bm{\mathfrak{p}}}}{}{#2}{}{#1}}}
\newrobustcmd{\checkbmmathfrakpz}[2][]{\ensuremath{\subp{\check{\bm{\mathfrak{p}}}}{}{#2}{}{#1}}}
\newrobustcmd{\tildebmmathfrakpz}[2][]{\ensuremath{\subp{\tilde{\bm{\mathfrak{p}}}}{}{#2}{}{#1}}}
\newrobustcmd{\widetildebmmathfrakpz}[2][]{\ensuremath{\subp{\widetilde{\bm{\mathfrak{p}}}}{}{#2}{}{#1}}}
\newrobustcmd{\acutebmmathfrakpz}[2][]{\ensuremath{\subp{\acute{\bm{\mathfrak{p}}}}{}{#2}{}{#1}}}
\newrobustcmd{\gravebmmathfrakpz}[2][]{\ensuremath{\subp{\grave{\bm{\mathfrak{p}}}}{}{#2}{}{#1}}}
\newrobustcmd{\dotbmmathfrakpz}[2][]{\ensuremath{\subp{\dot{\bm{\mathfrak{p}}}}{}{#2}{}{#1}}}
\newrobustcmd{\ddotbmmathfrakpz}[2][]{\ensuremath{\subp{\ddot{\bm{\mathfrak{p}}}}{}{#2}{}{#1}}}
\newrobustcmd{\brevebmmathfrakpz}[2][]{\ensuremath{\subp{\breve{\bm{\mathfrak{p}}}}{}{#2}{}{#1}}}
\newrobustcmd{\barbmmathfrakpz}[2][]{\ensuremath{\subp{\bar{\bm{\mathfrak{p}}}}{}{#2}{}{#1}}}
\newrobustcmd{\vecbmmathfrakpz}[2][]{\ensuremath{\subp{\vec{\bm{\mathfrak{p}}}}{}{#2}{}{#1}}}
\newrobustcmd{\mathfrakqz}[2][]{\ensuremath{\subp{\mathfrak{q}}{}{#2}{}{#1}}}
\newrobustcmd{\hatmathfrakqz}[2][]{\ensuremath{\subp{\hat{\mathfrak{q}}}{}{#2}{}{#1}}}
\newrobustcmd{\widehatmathfrakqz}[2][]{\ensuremath{\subp{\widehat{\mathfrak{q}}}{}{#2}{}{#1}}}
\newrobustcmd{\checkmathfrakqz}[2][]{\ensuremath{\subp{\check{\mathfrak{q}}}{}{#2}{}{#1}}}
\newrobustcmd{\tildemathfrakqz}[2][]{\ensuremath{\subp{\tilde{\mathfrak{q}}}{}{#2}{}{#1}}}
\newrobustcmd{\widetildemathfrakqz}[2][]{\ensuremath{\subp{\widetilde{\mathfrak{q}}}{}{#2}{}{#1}}}
\newrobustcmd{\acutemathfrakqz}[2][]{\ensuremath{\subp{\acute{\mathfrak{q}}}{}{#2}{}{#1}}}
\newrobustcmd{\gravemathfrakqz}[2][]{\ensuremath{\subp{\grave{\mathfrak{q}}}{}{#2}{}{#1}}}
\newrobustcmd{\dotmathfrakqz}[2][]{\ensuremath{\subp{\dot{\mathfrak{q}}}{}{#2}{}{#1}}}
\newrobustcmd{\ddotmathfrakqz}[2][]{\ensuremath{\subp{\ddot{\mathfrak{q}}}{}{#2}{}{#1}}}
\newrobustcmd{\brevemathfrakqz}[2][]{\ensuremath{\subp{\breve{\mathfrak{q}}}{}{#2}{}{#1}}}
\newrobustcmd{\barmathfrakqz}[2][]{\ensuremath{\subp{\bar{\mathfrak{q}}}{}{#2}{}{#1}}}
\newrobustcmd{\vecmathfrakqz}[2][]{\ensuremath{\subp{\vec{\mathfrak{q}}}{}{#2}{}{#1}}}
\newrobustcmd{\bmmathfrakqz}[2][]{\ensuremath{\subp{\bm{\mathfrak{q}}}{}{#2}{}{#1}}}
\newrobustcmd{\hatbmmathfrakqz}[2][]{\ensuremath{\subp{\hat{\bm{\mathfrak{q}}}}{}{#2}{}{#1}}}
\newrobustcmd{\widehatbmmathfrakqz}[2][]{\ensuremath{\subp{\widehat{\bm{\mathfrak{q}}}}{}{#2}{}{#1}}}
\newrobustcmd{\checkbmmathfrakqz}[2][]{\ensuremath{\subp{\check{\bm{\mathfrak{q}}}}{}{#2}{}{#1}}}
\newrobustcmd{\tildebmmathfrakqz}[2][]{\ensuremath{\subp{\tilde{\bm{\mathfrak{q}}}}{}{#2}{}{#1}}}
\newrobustcmd{\widetildebmmathfrakqz}[2][]{\ensuremath{\subp{\widetilde{\bm{\mathfrak{q}}}}{}{#2}{}{#1}}}
\newrobustcmd{\acutebmmathfrakqz}[2][]{\ensuremath{\subp{\acute{\bm{\mathfrak{q}}}}{}{#2}{}{#1}}}
\newrobustcmd{\gravebmmathfrakqz}[2][]{\ensuremath{\subp{\grave{\bm{\mathfrak{q}}}}{}{#2}{}{#1}}}
\newrobustcmd{\dotbmmathfrakqz}[2][]{\ensuremath{\subp{\dot{\bm{\mathfrak{q}}}}{}{#2}{}{#1}}}
\newrobustcmd{\ddotbmmathfrakqz}[2][]{\ensuremath{\subp{\ddot{\bm{\mathfrak{q}}}}{}{#2}{}{#1}}}
\newrobustcmd{\brevebmmathfrakqz}[2][]{\ensuremath{\subp{\breve{\bm{\mathfrak{q}}}}{}{#2}{}{#1}}}
\newrobustcmd{\barbmmathfrakqz}[2][]{\ensuremath{\subp{\bar{\bm{\mathfrak{q}}}}{}{#2}{}{#1}}}
\newrobustcmd{\vecbmmathfrakqz}[2][]{\ensuremath{\subp{\vec{\bm{\mathfrak{q}}}}{}{#2}{}{#1}}}
\newrobustcmd{\mathfrakrz}[2][]{\ensuremath{\subp{\mathfrak{r}}{}{#2}{}{#1}}}
\newrobustcmd{\hatmathfrakrz}[2][]{\ensuremath{\subp{\hat{\mathfrak{r}}}{}{#2}{}{#1}}}
\newrobustcmd{\widehatmathfrakrz}[2][]{\ensuremath{\subp{\widehat{\mathfrak{r}}}{}{#2}{}{#1}}}
\newrobustcmd{\checkmathfrakrz}[2][]{\ensuremath{\subp{\check{\mathfrak{r}}}{}{#2}{}{#1}}}
\newrobustcmd{\tildemathfrakrz}[2][]{\ensuremath{\subp{\tilde{\mathfrak{r}}}{}{#2}{}{#1}}}
\newrobustcmd{\widetildemathfrakrz}[2][]{\ensuremath{\subp{\widetilde{\mathfrak{r}}}{}{#2}{}{#1}}}
\newrobustcmd{\acutemathfrakrz}[2][]{\ensuremath{\subp{\acute{\mathfrak{r}}}{}{#2}{}{#1}}}
\newrobustcmd{\gravemathfrakrz}[2][]{\ensuremath{\subp{\grave{\mathfrak{r}}}{}{#2}{}{#1}}}
\newrobustcmd{\dotmathfrakrz}[2][]{\ensuremath{\subp{\dot{\mathfrak{r}}}{}{#2}{}{#1}}}
\newrobustcmd{\ddotmathfrakrz}[2][]{\ensuremath{\subp{\ddot{\mathfrak{r}}}{}{#2}{}{#1}}}
\newrobustcmd{\brevemathfrakrz}[2][]{\ensuremath{\subp{\breve{\mathfrak{r}}}{}{#2}{}{#1}}}
\newrobustcmd{\barmathfrakrz}[2][]{\ensuremath{\subp{\bar{\mathfrak{r}}}{}{#2}{}{#1}}}
\newrobustcmd{\vecmathfrakrz}[2][]{\ensuremath{\subp{\vec{\mathfrak{r}}}{}{#2}{}{#1}}}
\newrobustcmd{\bmmathfrakrz}[2][]{\ensuremath{\subp{\bm{\mathfrak{r}}}{}{#2}{}{#1}}}
\newrobustcmd{\hatbmmathfrakrz}[2][]{\ensuremath{\subp{\hat{\bm{\mathfrak{r}}}}{}{#2}{}{#1}}}
\newrobustcmd{\widehatbmmathfrakrz}[2][]{\ensuremath{\subp{\widehat{\bm{\mathfrak{r}}}}{}{#2}{}{#1}}}
\newrobustcmd{\checkbmmathfrakrz}[2][]{\ensuremath{\subp{\check{\bm{\mathfrak{r}}}}{}{#2}{}{#1}}}
\newrobustcmd{\tildebmmathfrakrz}[2][]{\ensuremath{\subp{\tilde{\bm{\mathfrak{r}}}}{}{#2}{}{#1}}}
\newrobustcmd{\widetildebmmathfrakrz}[2][]{\ensuremath{\subp{\widetilde{\bm{\mathfrak{r}}}}{}{#2}{}{#1}}}
\newrobustcmd{\acutebmmathfrakrz}[2][]{\ensuremath{\subp{\acute{\bm{\mathfrak{r}}}}{}{#2}{}{#1}}}
\newrobustcmd{\gravebmmathfrakrz}[2][]{\ensuremath{\subp{\grave{\bm{\mathfrak{r}}}}{}{#2}{}{#1}}}
\newrobustcmd{\dotbmmathfrakrz}[2][]{\ensuremath{\subp{\dot{\bm{\mathfrak{r}}}}{}{#2}{}{#1}}}
\newrobustcmd{\ddotbmmathfrakrz}[2][]{\ensuremath{\subp{\ddot{\bm{\mathfrak{r}}}}{}{#2}{}{#1}}}
\newrobustcmd{\brevebmmathfrakrz}[2][]{\ensuremath{\subp{\breve{\bm{\mathfrak{r}}}}{}{#2}{}{#1}}}
\newrobustcmd{\barbmmathfrakrz}[2][]{\ensuremath{\subp{\bar{\bm{\mathfrak{r}}}}{}{#2}{}{#1}}}
\newrobustcmd{\vecbmmathfrakrz}[2][]{\ensuremath{\subp{\vec{\bm{\mathfrak{r}}}}{}{#2}{}{#1}}}
\newrobustcmd{\mathfraksz}[2][]{\ensuremath{\subp{\mathfrak{s}}{}{#2}{}{#1}}}
\newrobustcmd{\hatmathfraksz}[2][]{\ensuremath{\subp{\hat{\mathfrak{s}}}{}{#2}{}{#1}}}
\newrobustcmd{\widehatmathfraksz}[2][]{\ensuremath{\subp{\widehat{\mathfrak{s}}}{}{#2}{}{#1}}}
\newrobustcmd{\checkmathfraksz}[2][]{\ensuremath{\subp{\check{\mathfrak{s}}}{}{#2}{}{#1}}}
\newrobustcmd{\tildemathfraksz}[2][]{\ensuremath{\subp{\tilde{\mathfrak{s}}}{}{#2}{}{#1}}}
\newrobustcmd{\widetildemathfraksz}[2][]{\ensuremath{\subp{\widetilde{\mathfrak{s}}}{}{#2}{}{#1}}}
\newrobustcmd{\acutemathfraksz}[2][]{\ensuremath{\subp{\acute{\mathfrak{s}}}{}{#2}{}{#1}}}
\newrobustcmd{\gravemathfraksz}[2][]{\ensuremath{\subp{\grave{\mathfrak{s}}}{}{#2}{}{#1}}}
\newrobustcmd{\dotmathfraksz}[2][]{\ensuremath{\subp{\dot{\mathfrak{s}}}{}{#2}{}{#1}}}
\newrobustcmd{\ddotmathfraksz}[2][]{\ensuremath{\subp{\ddot{\mathfrak{s}}}{}{#2}{}{#1}}}
\newrobustcmd{\brevemathfraksz}[2][]{\ensuremath{\subp{\breve{\mathfrak{s}}}{}{#2}{}{#1}}}
\newrobustcmd{\barmathfraksz}[2][]{\ensuremath{\subp{\bar{\mathfrak{s}}}{}{#2}{}{#1}}}
\newrobustcmd{\vecmathfraksz}[2][]{\ensuremath{\subp{\vec{\mathfrak{s}}}{}{#2}{}{#1}}}
\newrobustcmd{\bmmathfraksz}[2][]{\ensuremath{\subp{\bm{\mathfrak{s}}}{}{#2}{}{#1}}}
\newrobustcmd{\hatbmmathfraksz}[2][]{\ensuremath{\subp{\hat{\bm{\mathfrak{s}}}}{}{#2}{}{#1}}}
\newrobustcmd{\widehatbmmathfraksz}[2][]{\ensuremath{\subp{\widehat{\bm{\mathfrak{s}}}}{}{#2}{}{#1}}}
\newrobustcmd{\checkbmmathfraksz}[2][]{\ensuremath{\subp{\check{\bm{\mathfrak{s}}}}{}{#2}{}{#1}}}
\newrobustcmd{\tildebmmathfraksz}[2][]{\ensuremath{\subp{\tilde{\bm{\mathfrak{s}}}}{}{#2}{}{#1}}}
\newrobustcmd{\widetildebmmathfraksz}[2][]{\ensuremath{\subp{\widetilde{\bm{\mathfrak{s}}}}{}{#2}{}{#1}}}
\newrobustcmd{\acutebmmathfraksz}[2][]{\ensuremath{\subp{\acute{\bm{\mathfrak{s}}}}{}{#2}{}{#1}}}
\newrobustcmd{\gravebmmathfraksz}[2][]{\ensuremath{\subp{\grave{\bm{\mathfrak{s}}}}{}{#2}{}{#1}}}
\newrobustcmd{\dotbmmathfraksz}[2][]{\ensuremath{\subp{\dot{\bm{\mathfrak{s}}}}{}{#2}{}{#1}}}
\newrobustcmd{\ddotbmmathfraksz}[2][]{\ensuremath{\subp{\ddot{\bm{\mathfrak{s}}}}{}{#2}{}{#1}}}
\newrobustcmd{\brevebmmathfraksz}[2][]{\ensuremath{\subp{\breve{\bm{\mathfrak{s}}}}{}{#2}{}{#1}}}
\newrobustcmd{\barbmmathfraksz}[2][]{\ensuremath{\subp{\bar{\bm{\mathfrak{s}}}}{}{#2}{}{#1}}}
\newrobustcmd{\vecbmmathfraksz}[2][]{\ensuremath{\subp{\vec{\bm{\mathfrak{s}}}}{}{#2}{}{#1}}}
\newrobustcmd{\mathfraktz}[2][]{\ensuremath{\subp{\mathfrak{t}}{}{#2}{}{#1}}}
\newrobustcmd{\hatmathfraktz}[2][]{\ensuremath{\subp{\hat{\mathfrak{t}}}{}{#2}{}{#1}}}
\newrobustcmd{\widehatmathfraktz}[2][]{\ensuremath{\subp{\widehat{\mathfrak{t}}}{}{#2}{}{#1}}}
\newrobustcmd{\checkmathfraktz}[2][]{\ensuremath{\subp{\check{\mathfrak{t}}}{}{#2}{}{#1}}}
\newrobustcmd{\tildemathfraktz}[2][]{\ensuremath{\subp{\tilde{\mathfrak{t}}}{}{#2}{}{#1}}}
\newrobustcmd{\widetildemathfraktz}[2][]{\ensuremath{\subp{\widetilde{\mathfrak{t}}}{}{#2}{}{#1}}}
\newrobustcmd{\acutemathfraktz}[2][]{\ensuremath{\subp{\acute{\mathfrak{t}}}{}{#2}{}{#1}}}
\newrobustcmd{\gravemathfraktz}[2][]{\ensuremath{\subp{\grave{\mathfrak{t}}}{}{#2}{}{#1}}}
\newrobustcmd{\dotmathfraktz}[2][]{\ensuremath{\subp{\dot{\mathfrak{t}}}{}{#2}{}{#1}}}
\newrobustcmd{\ddotmathfraktz}[2][]{\ensuremath{\subp{\ddot{\mathfrak{t}}}{}{#2}{}{#1}}}
\newrobustcmd{\brevemathfraktz}[2][]{\ensuremath{\subp{\breve{\mathfrak{t}}}{}{#2}{}{#1}}}
\newrobustcmd{\barmathfraktz}[2][]{\ensuremath{\subp{\bar{\mathfrak{t}}}{}{#2}{}{#1}}}
\newrobustcmd{\vecmathfraktz}[2][]{\ensuremath{\subp{\vec{\mathfrak{t}}}{}{#2}{}{#1}}}
\newrobustcmd{\bmmathfraktz}[2][]{\ensuremath{\subp{\bm{\mathfrak{t}}}{}{#2}{}{#1}}}
\newrobustcmd{\hatbmmathfraktz}[2][]{\ensuremath{\subp{\hat{\bm{\mathfrak{t}}}}{}{#2}{}{#1}}}
\newrobustcmd{\widehatbmmathfraktz}[2][]{\ensuremath{\subp{\widehat{\bm{\mathfrak{t}}}}{}{#2}{}{#1}}}
\newrobustcmd{\checkbmmathfraktz}[2][]{\ensuremath{\subp{\check{\bm{\mathfrak{t}}}}{}{#2}{}{#1}}}
\newrobustcmd{\tildebmmathfraktz}[2][]{\ensuremath{\subp{\tilde{\bm{\mathfrak{t}}}}{}{#2}{}{#1}}}
\newrobustcmd{\widetildebmmathfraktz}[2][]{\ensuremath{\subp{\widetilde{\bm{\mathfrak{t}}}}{}{#2}{}{#1}}}
\newrobustcmd{\acutebmmathfraktz}[2][]{\ensuremath{\subp{\acute{\bm{\mathfrak{t}}}}{}{#2}{}{#1}}}
\newrobustcmd{\gravebmmathfraktz}[2][]{\ensuremath{\subp{\grave{\bm{\mathfrak{t}}}}{}{#2}{}{#1}}}
\newrobustcmd{\dotbmmathfraktz}[2][]{\ensuremath{\subp{\dot{\bm{\mathfrak{t}}}}{}{#2}{}{#1}}}
\newrobustcmd{\ddotbmmathfraktz}[2][]{\ensuremath{\subp{\ddot{\bm{\mathfrak{t}}}}{}{#2}{}{#1}}}
\newrobustcmd{\brevebmmathfraktz}[2][]{\ensuremath{\subp{\breve{\bm{\mathfrak{t}}}}{}{#2}{}{#1}}}
\newrobustcmd{\barbmmathfraktz}[2][]{\ensuremath{\subp{\bar{\bm{\mathfrak{t}}}}{}{#2}{}{#1}}}
\newrobustcmd{\vecbmmathfraktz}[2][]{\ensuremath{\subp{\vec{\bm{\mathfrak{t}}}}{}{#2}{}{#1}}}
\newrobustcmd{\mathfrakuz}[2][]{\ensuremath{\subp{\mathfrak{u}}{}{#2}{}{#1}}}
\newrobustcmd{\hatmathfrakuz}[2][]{\ensuremath{\subp{\hat{\mathfrak{u}}}{}{#2}{}{#1}}}
\newrobustcmd{\widehatmathfrakuz}[2][]{\ensuremath{\subp{\widehat{\mathfrak{u}}}{}{#2}{}{#1}}}
\newrobustcmd{\checkmathfrakuz}[2][]{\ensuremath{\subp{\check{\mathfrak{u}}}{}{#2}{}{#1}}}
\newrobustcmd{\tildemathfrakuz}[2][]{\ensuremath{\subp{\tilde{\mathfrak{u}}}{}{#2}{}{#1}}}
\newrobustcmd{\widetildemathfrakuz}[2][]{\ensuremath{\subp{\widetilde{\mathfrak{u}}}{}{#2}{}{#1}}}
\newrobustcmd{\acutemathfrakuz}[2][]{\ensuremath{\subp{\acute{\mathfrak{u}}}{}{#2}{}{#1}}}
\newrobustcmd{\gravemathfrakuz}[2][]{\ensuremath{\subp{\grave{\mathfrak{u}}}{}{#2}{}{#1}}}
\newrobustcmd{\dotmathfrakuz}[2][]{\ensuremath{\subp{\dot{\mathfrak{u}}}{}{#2}{}{#1}}}
\newrobustcmd{\ddotmathfrakuz}[2][]{\ensuremath{\subp{\ddot{\mathfrak{u}}}{}{#2}{}{#1}}}
\newrobustcmd{\brevemathfrakuz}[2][]{\ensuremath{\subp{\breve{\mathfrak{u}}}{}{#2}{}{#1}}}
\newrobustcmd{\barmathfrakuz}[2][]{\ensuremath{\subp{\bar{\mathfrak{u}}}{}{#2}{}{#1}}}
\newrobustcmd{\vecmathfrakuz}[2][]{\ensuremath{\subp{\vec{\mathfrak{u}}}{}{#2}{}{#1}}}
\newrobustcmd{\bmmathfrakuz}[2][]{\ensuremath{\subp{\bm{\mathfrak{u}}}{}{#2}{}{#1}}}
\newrobustcmd{\hatbmmathfrakuz}[2][]{\ensuremath{\subp{\hat{\bm{\mathfrak{u}}}}{}{#2}{}{#1}}}
\newrobustcmd{\widehatbmmathfrakuz}[2][]{\ensuremath{\subp{\widehat{\bm{\mathfrak{u}}}}{}{#2}{}{#1}}}
\newrobustcmd{\checkbmmathfrakuz}[2][]{\ensuremath{\subp{\check{\bm{\mathfrak{u}}}}{}{#2}{}{#1}}}
\newrobustcmd{\tildebmmathfrakuz}[2][]{\ensuremath{\subp{\tilde{\bm{\mathfrak{u}}}}{}{#2}{}{#1}}}
\newrobustcmd{\widetildebmmathfrakuz}[2][]{\ensuremath{\subp{\widetilde{\bm{\mathfrak{u}}}}{}{#2}{}{#1}}}
\newrobustcmd{\acutebmmathfrakuz}[2][]{\ensuremath{\subp{\acute{\bm{\mathfrak{u}}}}{}{#2}{}{#1}}}
\newrobustcmd{\gravebmmathfrakuz}[2][]{\ensuremath{\subp{\grave{\bm{\mathfrak{u}}}}{}{#2}{}{#1}}}
\newrobustcmd{\dotbmmathfrakuz}[2][]{\ensuremath{\subp{\dot{\bm{\mathfrak{u}}}}{}{#2}{}{#1}}}
\newrobustcmd{\ddotbmmathfrakuz}[2][]{\ensuremath{\subp{\ddot{\bm{\mathfrak{u}}}}{}{#2}{}{#1}}}
\newrobustcmd{\brevebmmathfrakuz}[2][]{\ensuremath{\subp{\breve{\bm{\mathfrak{u}}}}{}{#2}{}{#1}}}
\newrobustcmd{\barbmmathfrakuz}[2][]{\ensuremath{\subp{\bar{\bm{\mathfrak{u}}}}{}{#2}{}{#1}}}
\newrobustcmd{\vecbmmathfrakuz}[2][]{\ensuremath{\subp{\vec{\bm{\mathfrak{u}}}}{}{#2}{}{#1}}}
\newrobustcmd{\mathfrakvz}[2][]{\ensuremath{\subp{\mathfrak{v}}{}{#2}{}{#1}}}
\newrobustcmd{\hatmathfrakvz}[2][]{\ensuremath{\subp{\hat{\mathfrak{v}}}{}{#2}{}{#1}}}
\newrobustcmd{\widehatmathfrakvz}[2][]{\ensuremath{\subp{\widehat{\mathfrak{v}}}{}{#2}{}{#1}}}
\newrobustcmd{\checkmathfrakvz}[2][]{\ensuremath{\subp{\check{\mathfrak{v}}}{}{#2}{}{#1}}}
\newrobustcmd{\tildemathfrakvz}[2][]{\ensuremath{\subp{\tilde{\mathfrak{v}}}{}{#2}{}{#1}}}
\newrobustcmd{\widetildemathfrakvz}[2][]{\ensuremath{\subp{\widetilde{\mathfrak{v}}}{}{#2}{}{#1}}}
\newrobustcmd{\acutemathfrakvz}[2][]{\ensuremath{\subp{\acute{\mathfrak{v}}}{}{#2}{}{#1}}}
\newrobustcmd{\gravemathfrakvz}[2][]{\ensuremath{\subp{\grave{\mathfrak{v}}}{}{#2}{}{#1}}}
\newrobustcmd{\dotmathfrakvz}[2][]{\ensuremath{\subp{\dot{\mathfrak{v}}}{}{#2}{}{#1}}}
\newrobustcmd{\ddotmathfrakvz}[2][]{\ensuremath{\subp{\ddot{\mathfrak{v}}}{}{#2}{}{#1}}}
\newrobustcmd{\brevemathfrakvz}[2][]{\ensuremath{\subp{\breve{\mathfrak{v}}}{}{#2}{}{#1}}}
\newrobustcmd{\barmathfrakvz}[2][]{\ensuremath{\subp{\bar{\mathfrak{v}}}{}{#2}{}{#1}}}
\newrobustcmd{\vecmathfrakvz}[2][]{\ensuremath{\subp{\vec{\mathfrak{v}}}{}{#2}{}{#1}}}
\newrobustcmd{\bmmathfrakvz}[2][]{\ensuremath{\subp{\bm{\mathfrak{v}}}{}{#2}{}{#1}}}
\newrobustcmd{\hatbmmathfrakvz}[2][]{\ensuremath{\subp{\hat{\bm{\mathfrak{v}}}}{}{#2}{}{#1}}}
\newrobustcmd{\widehatbmmathfrakvz}[2][]{\ensuremath{\subp{\widehat{\bm{\mathfrak{v}}}}{}{#2}{}{#1}}}
\newrobustcmd{\checkbmmathfrakvz}[2][]{\ensuremath{\subp{\check{\bm{\mathfrak{v}}}}{}{#2}{}{#1}}}
\newrobustcmd{\tildebmmathfrakvz}[2][]{\ensuremath{\subp{\tilde{\bm{\mathfrak{v}}}}{}{#2}{}{#1}}}
\newrobustcmd{\widetildebmmathfrakvz}[2][]{\ensuremath{\subp{\widetilde{\bm{\mathfrak{v}}}}{}{#2}{}{#1}}}
\newrobustcmd{\acutebmmathfrakvz}[2][]{\ensuremath{\subp{\acute{\bm{\mathfrak{v}}}}{}{#2}{}{#1}}}
\newrobustcmd{\gravebmmathfrakvz}[2][]{\ensuremath{\subp{\grave{\bm{\mathfrak{v}}}}{}{#2}{}{#1}}}
\newrobustcmd{\dotbmmathfrakvz}[2][]{\ensuremath{\subp{\dot{\bm{\mathfrak{v}}}}{}{#2}{}{#1}}}
\newrobustcmd{\ddotbmmathfrakvz}[2][]{\ensuremath{\subp{\ddot{\bm{\mathfrak{v}}}}{}{#2}{}{#1}}}
\newrobustcmd{\brevebmmathfrakvz}[2][]{\ensuremath{\subp{\breve{\bm{\mathfrak{v}}}}{}{#2}{}{#1}}}
\newrobustcmd{\barbmmathfrakvz}[2][]{\ensuremath{\subp{\bar{\bm{\mathfrak{v}}}}{}{#2}{}{#1}}}
\newrobustcmd{\vecbmmathfrakvz}[2][]{\ensuremath{\subp{\vec{\bm{\mathfrak{v}}}}{}{#2}{}{#1}}}
\newrobustcmd{\mathfrakwz}[2][]{\ensuremath{\subp{\mathfrak{w}}{}{#2}{}{#1}}}
\newrobustcmd{\hatmathfrakwz}[2][]{\ensuremath{\subp{\hat{\mathfrak{w}}}{}{#2}{}{#1}}}
\newrobustcmd{\widehatmathfrakwz}[2][]{\ensuremath{\subp{\widehat{\mathfrak{w}}}{}{#2}{}{#1}}}
\newrobustcmd{\checkmathfrakwz}[2][]{\ensuremath{\subp{\check{\mathfrak{w}}}{}{#2}{}{#1}}}
\newrobustcmd{\tildemathfrakwz}[2][]{\ensuremath{\subp{\tilde{\mathfrak{w}}}{}{#2}{}{#1}}}
\newrobustcmd{\widetildemathfrakwz}[2][]{\ensuremath{\subp{\widetilde{\mathfrak{w}}}{}{#2}{}{#1}}}
\newrobustcmd{\acutemathfrakwz}[2][]{\ensuremath{\subp{\acute{\mathfrak{w}}}{}{#2}{}{#1}}}
\newrobustcmd{\gravemathfrakwz}[2][]{\ensuremath{\subp{\grave{\mathfrak{w}}}{}{#2}{}{#1}}}
\newrobustcmd{\dotmathfrakwz}[2][]{\ensuremath{\subp{\dot{\mathfrak{w}}}{}{#2}{}{#1}}}
\newrobustcmd{\ddotmathfrakwz}[2][]{\ensuremath{\subp{\ddot{\mathfrak{w}}}{}{#2}{}{#1}}}
\newrobustcmd{\brevemathfrakwz}[2][]{\ensuremath{\subp{\breve{\mathfrak{w}}}{}{#2}{}{#1}}}
\newrobustcmd{\barmathfrakwz}[2][]{\ensuremath{\subp{\bar{\mathfrak{w}}}{}{#2}{}{#1}}}
\newrobustcmd{\vecmathfrakwz}[2][]{\ensuremath{\subp{\vec{\mathfrak{w}}}{}{#2}{}{#1}}}
\newrobustcmd{\bmmathfrakwz}[2][]{\ensuremath{\subp{\bm{\mathfrak{w}}}{}{#2}{}{#1}}}
\newrobustcmd{\hatbmmathfrakwz}[2][]{\ensuremath{\subp{\hat{\bm{\mathfrak{w}}}}{}{#2}{}{#1}}}
\newrobustcmd{\widehatbmmathfrakwz}[2][]{\ensuremath{\subp{\widehat{\bm{\mathfrak{w}}}}{}{#2}{}{#1}}}
\newrobustcmd{\checkbmmathfrakwz}[2][]{\ensuremath{\subp{\check{\bm{\mathfrak{w}}}}{}{#2}{}{#1}}}
\newrobustcmd{\tildebmmathfrakwz}[2][]{\ensuremath{\subp{\tilde{\bm{\mathfrak{w}}}}{}{#2}{}{#1}}}
\newrobustcmd{\widetildebmmathfrakwz}[2][]{\ensuremath{\subp{\widetilde{\bm{\mathfrak{w}}}}{}{#2}{}{#1}}}
\newrobustcmd{\acutebmmathfrakwz}[2][]{\ensuremath{\subp{\acute{\bm{\mathfrak{w}}}}{}{#2}{}{#1}}}
\newrobustcmd{\gravebmmathfrakwz}[2][]{\ensuremath{\subp{\grave{\bm{\mathfrak{w}}}}{}{#2}{}{#1}}}
\newrobustcmd{\dotbmmathfrakwz}[2][]{\ensuremath{\subp{\dot{\bm{\mathfrak{w}}}}{}{#2}{}{#1}}}
\newrobustcmd{\ddotbmmathfrakwz}[2][]{\ensuremath{\subp{\ddot{\bm{\mathfrak{w}}}}{}{#2}{}{#1}}}
\newrobustcmd{\brevebmmathfrakwz}[2][]{\ensuremath{\subp{\breve{\bm{\mathfrak{w}}}}{}{#2}{}{#1}}}
\newrobustcmd{\barbmmathfrakwz}[2][]{\ensuremath{\subp{\bar{\bm{\mathfrak{w}}}}{}{#2}{}{#1}}}
\newrobustcmd{\vecbmmathfrakwz}[2][]{\ensuremath{\subp{\vec{\bm{\mathfrak{w}}}}{}{#2}{}{#1}}}
\newrobustcmd{\mathfrakxz}[2][]{\ensuremath{\subp{\mathfrak{x}}{}{#2}{}{#1}}}
\newrobustcmd{\hatmathfrakxz}[2][]{\ensuremath{\subp{\hat{\mathfrak{x}}}{}{#2}{}{#1}}}
\newrobustcmd{\widehatmathfrakxz}[2][]{\ensuremath{\subp{\widehat{\mathfrak{x}}}{}{#2}{}{#1}}}
\newrobustcmd{\checkmathfrakxz}[2][]{\ensuremath{\subp{\check{\mathfrak{x}}}{}{#2}{}{#1}}}
\newrobustcmd{\tildemathfrakxz}[2][]{\ensuremath{\subp{\tilde{\mathfrak{x}}}{}{#2}{}{#1}}}
\newrobustcmd{\widetildemathfrakxz}[2][]{\ensuremath{\subp{\widetilde{\mathfrak{x}}}{}{#2}{}{#1}}}
\newrobustcmd{\acutemathfrakxz}[2][]{\ensuremath{\subp{\acute{\mathfrak{x}}}{}{#2}{}{#1}}}
\newrobustcmd{\gravemathfrakxz}[2][]{\ensuremath{\subp{\grave{\mathfrak{x}}}{}{#2}{}{#1}}}
\newrobustcmd{\dotmathfrakxz}[2][]{\ensuremath{\subp{\dot{\mathfrak{x}}}{}{#2}{}{#1}}}
\newrobustcmd{\ddotmathfrakxz}[2][]{\ensuremath{\subp{\ddot{\mathfrak{x}}}{}{#2}{}{#1}}}
\newrobustcmd{\brevemathfrakxz}[2][]{\ensuremath{\subp{\breve{\mathfrak{x}}}{}{#2}{}{#1}}}
\newrobustcmd{\barmathfrakxz}[2][]{\ensuremath{\subp{\bar{\mathfrak{x}}}{}{#2}{}{#1}}}
\newrobustcmd{\vecmathfrakxz}[2][]{\ensuremath{\subp{\vec{\mathfrak{x}}}{}{#2}{}{#1}}}
\newrobustcmd{\bmmathfrakxz}[2][]{\ensuremath{\subp{\bm{\mathfrak{x}}}{}{#2}{}{#1}}}
\newrobustcmd{\hatbmmathfrakxz}[2][]{\ensuremath{\subp{\hat{\bm{\mathfrak{x}}}}{}{#2}{}{#1}}}
\newrobustcmd{\widehatbmmathfrakxz}[2][]{\ensuremath{\subp{\widehat{\bm{\mathfrak{x}}}}{}{#2}{}{#1}}}
\newrobustcmd{\checkbmmathfrakxz}[2][]{\ensuremath{\subp{\check{\bm{\mathfrak{x}}}}{}{#2}{}{#1}}}
\newrobustcmd{\tildebmmathfrakxz}[2][]{\ensuremath{\subp{\tilde{\bm{\mathfrak{x}}}}{}{#2}{}{#1}}}
\newrobustcmd{\widetildebmmathfrakxz}[2][]{\ensuremath{\subp{\widetilde{\bm{\mathfrak{x}}}}{}{#2}{}{#1}}}
\newrobustcmd{\acutebmmathfrakxz}[2][]{\ensuremath{\subp{\acute{\bm{\mathfrak{x}}}}{}{#2}{}{#1}}}
\newrobustcmd{\gravebmmathfrakxz}[2][]{\ensuremath{\subp{\grave{\bm{\mathfrak{x}}}}{}{#2}{}{#1}}}
\newrobustcmd{\dotbmmathfrakxz}[2][]{\ensuremath{\subp{\dot{\bm{\mathfrak{x}}}}{}{#2}{}{#1}}}
\newrobustcmd{\ddotbmmathfrakxz}[2][]{\ensuremath{\subp{\ddot{\bm{\mathfrak{x}}}}{}{#2}{}{#1}}}
\newrobustcmd{\brevebmmathfrakxz}[2][]{\ensuremath{\subp{\breve{\bm{\mathfrak{x}}}}{}{#2}{}{#1}}}
\newrobustcmd{\barbmmathfrakxz}[2][]{\ensuremath{\subp{\bar{\bm{\mathfrak{x}}}}{}{#2}{}{#1}}}
\newrobustcmd{\vecbmmathfrakxz}[2][]{\ensuremath{\subp{\vec{\bm{\mathfrak{x}}}}{}{#2}{}{#1}}}
\newrobustcmd{\mathfrakyz}[2][]{\ensuremath{\subp{\mathfrak{y}}{}{#2}{}{#1}}}
\newrobustcmd{\hatmathfrakyz}[2][]{\ensuremath{\subp{\hat{\mathfrak{y}}}{}{#2}{}{#1}}}
\newrobustcmd{\widehatmathfrakyz}[2][]{\ensuremath{\subp{\widehat{\mathfrak{y}}}{}{#2}{}{#1}}}
\newrobustcmd{\checkmathfrakyz}[2][]{\ensuremath{\subp{\check{\mathfrak{y}}}{}{#2}{}{#1}}}
\newrobustcmd{\tildemathfrakyz}[2][]{\ensuremath{\subp{\tilde{\mathfrak{y}}}{}{#2}{}{#1}}}
\newrobustcmd{\widetildemathfrakyz}[2][]{\ensuremath{\subp{\widetilde{\mathfrak{y}}}{}{#2}{}{#1}}}
\newrobustcmd{\acutemathfrakyz}[2][]{\ensuremath{\subp{\acute{\mathfrak{y}}}{}{#2}{}{#1}}}
\newrobustcmd{\gravemathfrakyz}[2][]{\ensuremath{\subp{\grave{\mathfrak{y}}}{}{#2}{}{#1}}}
\newrobustcmd{\dotmathfrakyz}[2][]{\ensuremath{\subp{\dot{\mathfrak{y}}}{}{#2}{}{#1}}}
\newrobustcmd{\ddotmathfrakyz}[2][]{\ensuremath{\subp{\ddot{\mathfrak{y}}}{}{#2}{}{#1}}}
\newrobustcmd{\brevemathfrakyz}[2][]{\ensuremath{\subp{\breve{\mathfrak{y}}}{}{#2}{}{#1}}}
\newrobustcmd{\barmathfrakyz}[2][]{\ensuremath{\subp{\bar{\mathfrak{y}}}{}{#2}{}{#1}}}
\newrobustcmd{\vecmathfrakyz}[2][]{\ensuremath{\subp{\vec{\mathfrak{y}}}{}{#2}{}{#1}}}
\newrobustcmd{\bmmathfrakyz}[2][]{\ensuremath{\subp{\bm{\mathfrak{y}}}{}{#2}{}{#1}}}
\newrobustcmd{\hatbmmathfrakyz}[2][]{\ensuremath{\subp{\hat{\bm{\mathfrak{y}}}}{}{#2}{}{#1}}}
\newrobustcmd{\widehatbmmathfrakyz}[2][]{\ensuremath{\subp{\widehat{\bm{\mathfrak{y}}}}{}{#2}{}{#1}}}
\newrobustcmd{\checkbmmathfrakyz}[2][]{\ensuremath{\subp{\check{\bm{\mathfrak{y}}}}{}{#2}{}{#1}}}
\newrobustcmd{\tildebmmathfrakyz}[2][]{\ensuremath{\subp{\tilde{\bm{\mathfrak{y}}}}{}{#2}{}{#1}}}
\newrobustcmd{\widetildebmmathfrakyz}[2][]{\ensuremath{\subp{\widetilde{\bm{\mathfrak{y}}}}{}{#2}{}{#1}}}
\newrobustcmd{\acutebmmathfrakyz}[2][]{\ensuremath{\subp{\acute{\bm{\mathfrak{y}}}}{}{#2}{}{#1}}}
\newrobustcmd{\gravebmmathfrakyz}[2][]{\ensuremath{\subp{\grave{\bm{\mathfrak{y}}}}{}{#2}{}{#1}}}
\newrobustcmd{\dotbmmathfrakyz}[2][]{\ensuremath{\subp{\dot{\bm{\mathfrak{y}}}}{}{#2}{}{#1}}}
\newrobustcmd{\ddotbmmathfrakyz}[2][]{\ensuremath{\subp{\ddot{\bm{\mathfrak{y}}}}{}{#2}{}{#1}}}
\newrobustcmd{\brevebmmathfrakyz}[2][]{\ensuremath{\subp{\breve{\bm{\mathfrak{y}}}}{}{#2}{}{#1}}}
\newrobustcmd{\barbmmathfrakyz}[2][]{\ensuremath{\subp{\bar{\bm{\mathfrak{y}}}}{}{#2}{}{#1}}}
\newrobustcmd{\vecbmmathfrakyz}[2][]{\ensuremath{\subp{\vec{\bm{\mathfrak{y}}}}{}{#2}{}{#1}}}
\newrobustcmd{\mathfrakzz}[2][]{\ensuremath{\subp{\mathfrak{z}}{}{#2}{}{#1}}}
\newrobustcmd{\hatmathfrakzz}[2][]{\ensuremath{\subp{\hat{\mathfrak{z}}}{}{#2}{}{#1}}}
\newrobustcmd{\widehatmathfrakzz}[2][]{\ensuremath{\subp{\widehat{\mathfrak{z}}}{}{#2}{}{#1}}}
\newrobustcmd{\checkmathfrakzz}[2][]{\ensuremath{\subp{\check{\mathfrak{z}}}{}{#2}{}{#1}}}
\newrobustcmd{\tildemathfrakzz}[2][]{\ensuremath{\subp{\tilde{\mathfrak{z}}}{}{#2}{}{#1}}}
\newrobustcmd{\widetildemathfrakzz}[2][]{\ensuremath{\subp{\widetilde{\mathfrak{z}}}{}{#2}{}{#1}}}
\newrobustcmd{\acutemathfrakzz}[2][]{\ensuremath{\subp{\acute{\mathfrak{z}}}{}{#2}{}{#1}}}
\newrobustcmd{\gravemathfrakzz}[2][]{\ensuremath{\subp{\grave{\mathfrak{z}}}{}{#2}{}{#1}}}
\newrobustcmd{\dotmathfrakzz}[2][]{\ensuremath{\subp{\dot{\mathfrak{z}}}{}{#2}{}{#1}}}
\newrobustcmd{\ddotmathfrakzz}[2][]{\ensuremath{\subp{\ddot{\mathfrak{z}}}{}{#2}{}{#1}}}
\newrobustcmd{\brevemathfrakzz}[2][]{\ensuremath{\subp{\breve{\mathfrak{z}}}{}{#2}{}{#1}}}
\newrobustcmd{\barmathfrakzz}[2][]{\ensuremath{\subp{\bar{\mathfrak{z}}}{}{#2}{}{#1}}}
\newrobustcmd{\vecmathfrakzz}[2][]{\ensuremath{\subp{\vec{\mathfrak{z}}}{}{#2}{}{#1}}}
\newrobustcmd{\bmmathfrakzz}[2][]{\ensuremath{\subp{\bm{\mathfrak{z}}}{}{#2}{}{#1}}}
\newrobustcmd{\hatbmmathfrakzz}[2][]{\ensuremath{\subp{\hat{\bm{\mathfrak{z}}}}{}{#2}{}{#1}}}
\newrobustcmd{\widehatbmmathfrakzz}[2][]{\ensuremath{\subp{\widehat{\bm{\mathfrak{z}}}}{}{#2}{}{#1}}}
\newrobustcmd{\checkbmmathfrakzz}[2][]{\ensuremath{\subp{\check{\bm{\mathfrak{z}}}}{}{#2}{}{#1}}}
\newrobustcmd{\tildebmmathfrakzz}[2][]{\ensuremath{\subp{\tilde{\bm{\mathfrak{z}}}}{}{#2}{}{#1}}}
\newrobustcmd{\widetildebmmathfrakzz}[2][]{\ensuremath{\subp{\widetilde{\bm{\mathfrak{z}}}}{}{#2}{}{#1}}}
\newrobustcmd{\acutebmmathfrakzz}[2][]{\ensuremath{\subp{\acute{\bm{\mathfrak{z}}}}{}{#2}{}{#1}}}
\newrobustcmd{\gravebmmathfrakzz}[2][]{\ensuremath{\subp{\grave{\bm{\mathfrak{z}}}}{}{#2}{}{#1}}}
\newrobustcmd{\dotbmmathfrakzz}[2][]{\ensuremath{\subp{\dot{\bm{\mathfrak{z}}}}{}{#2}{}{#1}}}
\newrobustcmd{\ddotbmmathfrakzz}[2][]{\ensuremath{\subp{\ddot{\bm{\mathfrak{z}}}}{}{#2}{}{#1}}}
\newrobustcmd{\brevebmmathfrakzz}[2][]{\ensuremath{\subp{\breve{\bm{\mathfrak{z}}}}{}{#2}{}{#1}}}
\newrobustcmd{\barbmmathfrakzz}[2][]{\ensuremath{\subp{\bar{\bm{\mathfrak{z}}}}{}{#2}{}{#1}}}
\newrobustcmd{\vecbmmathfrakzz}[2][]{\ensuremath{\subp{\vec{\bm{\mathfrak{z}}}}{}{#2}{}{#1}}}

\newrobustcmd{\mathfrakAz}[2][]{\ensuremath{\subp{\mathfrak{A}}{}{#2}{}{#1}}}
\newrobustcmd{\hatmathfrakAz}[2][]{\ensuremath{\subp{\hat{\mathfrak{A}}}{}{#2}{}{#1}}}
\newrobustcmd{\widehatmathfrakAz}[2][]{\ensuremath{\subp{\widehat{\mathfrak{A}}}{}{#2}{}{#1}}}
\newrobustcmd{\checkmathfrakAz}[2][]{\ensuremath{\subp{\check{\mathfrak{A}}}{}{#2}{}{#1}}}
\newrobustcmd{\tildemathfrakAz}[2][]{\ensuremath{\subp{\tilde{\mathfrak{A}}}{}{#2}{}{#1}}}
\newrobustcmd{\widetildemathfrakAz}[2][]{\ensuremath{\subp{\widetilde{\mathfrak{A}}}{}{#2}{}{#1}}}
\newrobustcmd{\acutemathfrakAz}[2][]{\ensuremath{\subp{\acute{\mathfrak{A}}}{}{#2}{}{#1}}}
\newrobustcmd{\gravemathfrakAz}[2][]{\ensuremath{\subp{\grave{\mathfrak{A}}}{}{#2}{}{#1}}}
\newrobustcmd{\dotmathfrakAz}[2][]{\ensuremath{\subp{\dot{\mathfrak{A}}}{}{#2}{}{#1}}}
\newrobustcmd{\ddotmathfrakAz}[2][]{\ensuremath{\subp{\ddot{\mathfrak{A}}}{}{#2}{}{#1}}}
\newrobustcmd{\brevemathfrakAz}[2][]{\ensuremath{\subp{\breve{\mathfrak{A}}}{}{#2}{}{#1}}}
\newrobustcmd{\barmathfrakAz}[2][]{\ensuremath{\subp{\bar{\mathfrak{A}}}{}{#2}{}{#1}}}
\newrobustcmd{\vecmathfrakAz}[2][]{\ensuremath{\subp{\vec{\mathfrak{A}}}{}{#2}{}{#1}}}
\newrobustcmd{\bmmathfrakAz}[2][]{\ensuremath{\subp{\bm{\mathfrak{A}}}{}{#2}{}{#1}}}
\newrobustcmd{\hatbmmathfrakAz}[2][]{\ensuremath{\subp{\hat{\bm{\mathfrak{A}}}}{}{#2}{}{#1}}}
\newrobustcmd{\widehatbmmathfrakAz}[2][]{\ensuremath{\subp{\widehat{\bm{\mathfrak{A}}}}{}{#2}{}{#1}}}
\newrobustcmd{\checkbmmathfrakAz}[2][]{\ensuremath{\subp{\check{\bm{\mathfrak{A}}}}{}{#2}{}{#1}}}
\newrobustcmd{\tildebmmathfrakAz}[2][]{\ensuremath{\subp{\tilde{\bm{\mathfrak{A}}}}{}{#2}{}{#1}}}
\newrobustcmd{\widetildebmmathfrakAz}[2][]{\ensuremath{\subp{\widetilde{\bm{\mathfrak{A}}}}{}{#2}{}{#1}}}
\newrobustcmd{\acutebmmathfrakAz}[2][]{\ensuremath{\subp{\acute{\bm{\mathfrak{A}}}}{}{#2}{}{#1}}}
\newrobustcmd{\gravebmmathfrakAz}[2][]{\ensuremath{\subp{\grave{\bm{\mathfrak{A}}}}{}{#2}{}{#1}}}
\newrobustcmd{\dotbmmathfrakAz}[2][]{\ensuremath{\subp{\dot{\bm{\mathfrak{A}}}}{}{#2}{}{#1}}}
\newrobustcmd{\ddotbmmathfrakAz}[2][]{\ensuremath{\subp{\ddot{\bm{\mathfrak{A}}}}{}{#2}{}{#1}}}
\newrobustcmd{\brevebmmathfrakAz}[2][]{\ensuremath{\subp{\breve{\bm{\mathfrak{A}}}}{}{#2}{}{#1}}}
\newrobustcmd{\barbmmathfrakAz}[2][]{\ensuremath{\subp{\bar{\bm{\mathfrak{A}}}}{}{#2}{}{#1}}}
\newrobustcmd{\vecbmmathfrakAz}[2][]{\ensuremath{\subp{\vec{\bm{\mathfrak{A}}}}{}{#2}{}{#1}}}
\newrobustcmd{\mathfrakBz}[2][]{\ensuremath{\subp{\mathfrak{B}}{}{#2}{}{#1}}}
\newrobustcmd{\hatmathfrakBz}[2][]{\ensuremath{\subp{\hat{\mathfrak{B}}}{}{#2}{}{#1}}}
\newrobustcmd{\widehatmathfrakBz}[2][]{\ensuremath{\subp{\widehat{\mathfrak{B}}}{}{#2}{}{#1}}}
\newrobustcmd{\checkmathfrakBz}[2][]{\ensuremath{\subp{\check{\mathfrak{B}}}{}{#2}{}{#1}}}
\newrobustcmd{\tildemathfrakBz}[2][]{\ensuremath{\subp{\tilde{\mathfrak{B}}}{}{#2}{}{#1}}}
\newrobustcmd{\widetildemathfrakBz}[2][]{\ensuremath{\subp{\widetilde{\mathfrak{B}}}{}{#2}{}{#1}}}
\newrobustcmd{\acutemathfrakBz}[2][]{\ensuremath{\subp{\acute{\mathfrak{B}}}{}{#2}{}{#1}}}
\newrobustcmd{\gravemathfrakBz}[2][]{\ensuremath{\subp{\grave{\mathfrak{B}}}{}{#2}{}{#1}}}
\newrobustcmd{\dotmathfrakBz}[2][]{\ensuremath{\subp{\dot{\mathfrak{B}}}{}{#2}{}{#1}}}
\newrobustcmd{\ddotmathfrakBz}[2][]{\ensuremath{\subp{\ddot{\mathfrak{B}}}{}{#2}{}{#1}}}
\newrobustcmd{\brevemathfrakBz}[2][]{\ensuremath{\subp{\breve{\mathfrak{B}}}{}{#2}{}{#1}}}
\newrobustcmd{\barmathfrakBz}[2][]{\ensuremath{\subp{\bar{\mathfrak{B}}}{}{#2}{}{#1}}}
\newrobustcmd{\vecmathfrakBz}[2][]{\ensuremath{\subp{\vec{\mathfrak{B}}}{}{#2}{}{#1}}}
\newrobustcmd{\bmmathfrakBz}[2][]{\ensuremath{\subp{\bm{\mathfrak{B}}}{}{#2}{}{#1}}}
\newrobustcmd{\hatbmmathfrakBz}[2][]{\ensuremath{\subp{\hat{\bm{\mathfrak{B}}}}{}{#2}{}{#1}}}
\newrobustcmd{\widehatbmmathfrakBz}[2][]{\ensuremath{\subp{\widehat{\bm{\mathfrak{B}}}}{}{#2}{}{#1}}}
\newrobustcmd{\checkbmmathfrakBz}[2][]{\ensuremath{\subp{\check{\bm{\mathfrak{B}}}}{}{#2}{}{#1}}}
\newrobustcmd{\tildebmmathfrakBz}[2][]{\ensuremath{\subp{\tilde{\bm{\mathfrak{B}}}}{}{#2}{}{#1}}}
\newrobustcmd{\widetildebmmathfrakBz}[2][]{\ensuremath{\subp{\widetilde{\bm{\mathfrak{B}}}}{}{#2}{}{#1}}}
\newrobustcmd{\acutebmmathfrakBz}[2][]{\ensuremath{\subp{\acute{\bm{\mathfrak{B}}}}{}{#2}{}{#1}}}
\newrobustcmd{\gravebmmathfrakBz}[2][]{\ensuremath{\subp{\grave{\bm{\mathfrak{B}}}}{}{#2}{}{#1}}}
\newrobustcmd{\dotbmmathfrakBz}[2][]{\ensuremath{\subp{\dot{\bm{\mathfrak{B}}}}{}{#2}{}{#1}}}
\newrobustcmd{\ddotbmmathfrakBz}[2][]{\ensuremath{\subp{\ddot{\bm{\mathfrak{B}}}}{}{#2}{}{#1}}}
\newrobustcmd{\brevebmmathfrakBz}[2][]{\ensuremath{\subp{\breve{\bm{\mathfrak{B}}}}{}{#2}{}{#1}}}
\newrobustcmd{\barbmmathfrakBz}[2][]{\ensuremath{\subp{\bar{\bm{\mathfrak{B}}}}{}{#2}{}{#1}}}
\newrobustcmd{\vecbmmathfrakBz}[2][]{\ensuremath{\subp{\vec{\bm{\mathfrak{B}}}}{}{#2}{}{#1}}}
\newrobustcmd{\mathfrakCz}[2][]{\ensuremath{\subp{\mathfrak{C}}{}{#2}{}{#1}}}
\newrobustcmd{\hatmathfrakCz}[2][]{\ensuremath{\subp{\hat{\mathfrak{C}}}{}{#2}{}{#1}}}
\newrobustcmd{\widehatmathfrakCz}[2][]{\ensuremath{\subp{\widehat{\mathfrak{C}}}{}{#2}{}{#1}}}
\newrobustcmd{\checkmathfrakCz}[2][]{\ensuremath{\subp{\check{\mathfrak{C}}}{}{#2}{}{#1}}}
\newrobustcmd{\tildemathfrakCz}[2][]{\ensuremath{\subp{\tilde{\mathfrak{C}}}{}{#2}{}{#1}}}
\newrobustcmd{\widetildemathfrakCz}[2][]{\ensuremath{\subp{\widetilde{\mathfrak{C}}}{}{#2}{}{#1}}}
\newrobustcmd{\acutemathfrakCz}[2][]{\ensuremath{\subp{\acute{\mathfrak{C}}}{}{#2}{}{#1}}}
\newrobustcmd{\gravemathfrakCz}[2][]{\ensuremath{\subp{\grave{\mathfrak{C}}}{}{#2}{}{#1}}}
\newrobustcmd{\dotmathfrakCz}[2][]{\ensuremath{\subp{\dot{\mathfrak{C}}}{}{#2}{}{#1}}}
\newrobustcmd{\ddotmathfrakCz}[2][]{\ensuremath{\subp{\ddot{\mathfrak{C}}}{}{#2}{}{#1}}}
\newrobustcmd{\brevemathfrakCz}[2][]{\ensuremath{\subp{\breve{\mathfrak{C}}}{}{#2}{}{#1}}}
\newrobustcmd{\barmathfrakCz}[2][]{\ensuremath{\subp{\bar{\mathfrak{C}}}{}{#2}{}{#1}}}
\newrobustcmd{\vecmathfrakCz}[2][]{\ensuremath{\subp{\vec{\mathfrak{C}}}{}{#2}{}{#1}}}
\newrobustcmd{\bmmathfrakCz}[2][]{\ensuremath{\subp{\bm{\mathfrak{C}}}{}{#2}{}{#1}}}
\newrobustcmd{\hatbmmathfrakCz}[2][]{\ensuremath{\subp{\hat{\bm{\mathfrak{C}}}}{}{#2}{}{#1}}}
\newrobustcmd{\widehatbmmathfrakCz}[2][]{\ensuremath{\subp{\widehat{\bm{\mathfrak{C}}}}{}{#2}{}{#1}}}
\newrobustcmd{\checkbmmathfrakCz}[2][]{\ensuremath{\subp{\check{\bm{\mathfrak{C}}}}{}{#2}{}{#1}}}
\newrobustcmd{\tildebmmathfrakCz}[2][]{\ensuremath{\subp{\tilde{\bm{\mathfrak{C}}}}{}{#2}{}{#1}}}
\newrobustcmd{\widetildebmmathfrakCz}[2][]{\ensuremath{\subp{\widetilde{\bm{\mathfrak{C}}}}{}{#2}{}{#1}}}
\newrobustcmd{\acutebmmathfrakCz}[2][]{\ensuremath{\subp{\acute{\bm{\mathfrak{C}}}}{}{#2}{}{#1}}}
\newrobustcmd{\gravebmmathfrakCz}[2][]{\ensuremath{\subp{\grave{\bm{\mathfrak{C}}}}{}{#2}{}{#1}}}
\newrobustcmd{\dotbmmathfrakCz}[2][]{\ensuremath{\subp{\dot{\bm{\mathfrak{C}}}}{}{#2}{}{#1}}}
\newrobustcmd{\ddotbmmathfrakCz}[2][]{\ensuremath{\subp{\ddot{\bm{\mathfrak{C}}}}{}{#2}{}{#1}}}
\newrobustcmd{\brevebmmathfrakCz}[2][]{\ensuremath{\subp{\breve{\bm{\mathfrak{C}}}}{}{#2}{}{#1}}}
\newrobustcmd{\barbmmathfrakCz}[2][]{\ensuremath{\subp{\bar{\bm{\mathfrak{C}}}}{}{#2}{}{#1}}}
\newrobustcmd{\vecbmmathfrakCz}[2][]{\ensuremath{\subp{\vec{\bm{\mathfrak{C}}}}{}{#2}{}{#1}}}
\newrobustcmd{\mathfrakDz}[2][]{\ensuremath{\subp{\mathfrak{D}}{}{#2}{}{#1}}}
\newrobustcmd{\hatmathfrakDz}[2][]{\ensuremath{\subp{\hat{\mathfrak{D}}}{}{#2}{}{#1}}}
\newrobustcmd{\widehatmathfrakDz}[2][]{\ensuremath{\subp{\widehat{\mathfrak{D}}}{}{#2}{}{#1}}}
\newrobustcmd{\checkmathfrakDz}[2][]{\ensuremath{\subp{\check{\mathfrak{D}}}{}{#2}{}{#1}}}
\newrobustcmd{\tildemathfrakDz}[2][]{\ensuremath{\subp{\tilde{\mathfrak{D}}}{}{#2}{}{#1}}}
\newrobustcmd{\widetildemathfrakDz}[2][]{\ensuremath{\subp{\widetilde{\mathfrak{D}}}{}{#2}{}{#1}}}
\newrobustcmd{\acutemathfrakDz}[2][]{\ensuremath{\subp{\acute{\mathfrak{D}}}{}{#2}{}{#1}}}
\newrobustcmd{\gravemathfrakDz}[2][]{\ensuremath{\subp{\grave{\mathfrak{D}}}{}{#2}{}{#1}}}
\newrobustcmd{\dotmathfrakDz}[2][]{\ensuremath{\subp{\dot{\mathfrak{D}}}{}{#2}{}{#1}}}
\newrobustcmd{\ddotmathfrakDz}[2][]{\ensuremath{\subp{\ddot{\mathfrak{D}}}{}{#2}{}{#1}}}
\newrobustcmd{\brevemathfrakDz}[2][]{\ensuremath{\subp{\breve{\mathfrak{D}}}{}{#2}{}{#1}}}
\newrobustcmd{\barmathfrakDz}[2][]{\ensuremath{\subp{\bar{\mathfrak{D}}}{}{#2}{}{#1}}}
\newrobustcmd{\vecmathfrakDz}[2][]{\ensuremath{\subp{\vec{\mathfrak{D}}}{}{#2}{}{#1}}}
\newrobustcmd{\bmmathfrakDz}[2][]{\ensuremath{\subp{\bm{\mathfrak{D}}}{}{#2}{}{#1}}}
\newrobustcmd{\hatbmmathfrakDz}[2][]{\ensuremath{\subp{\hat{\bm{\mathfrak{D}}}}{}{#2}{}{#1}}}
\newrobustcmd{\widehatbmmathfrakDz}[2][]{\ensuremath{\subp{\widehat{\bm{\mathfrak{D}}}}{}{#2}{}{#1}}}
\newrobustcmd{\checkbmmathfrakDz}[2][]{\ensuremath{\subp{\check{\bm{\mathfrak{D}}}}{}{#2}{}{#1}}}
\newrobustcmd{\tildebmmathfrakDz}[2][]{\ensuremath{\subp{\tilde{\bm{\mathfrak{D}}}}{}{#2}{}{#1}}}
\newrobustcmd{\widetildebmmathfrakDz}[2][]{\ensuremath{\subp{\widetilde{\bm{\mathfrak{D}}}}{}{#2}{}{#1}}}
\newrobustcmd{\acutebmmathfrakDz}[2][]{\ensuremath{\subp{\acute{\bm{\mathfrak{D}}}}{}{#2}{}{#1}}}
\newrobustcmd{\gravebmmathfrakDz}[2][]{\ensuremath{\subp{\grave{\bm{\mathfrak{D}}}}{}{#2}{}{#1}}}
\newrobustcmd{\dotbmmathfrakDz}[2][]{\ensuremath{\subp{\dot{\bm{\mathfrak{D}}}}{}{#2}{}{#1}}}
\newrobustcmd{\ddotbmmathfrakDz}[2][]{\ensuremath{\subp{\ddot{\bm{\mathfrak{D}}}}{}{#2}{}{#1}}}
\newrobustcmd{\brevebmmathfrakDz}[2][]{\ensuremath{\subp{\breve{\bm{\mathfrak{D}}}}{}{#2}{}{#1}}}
\newrobustcmd{\barbmmathfrakDz}[2][]{\ensuremath{\subp{\bar{\bm{\mathfrak{D}}}}{}{#2}{}{#1}}}
\newrobustcmd{\vecbmmathfrakDz}[2][]{\ensuremath{\subp{\vec{\bm{\mathfrak{D}}}}{}{#2}{}{#1}}}
\newrobustcmd{\mathfrakEz}[2][]{\ensuremath{\subp{\mathfrak{E}}{}{#2}{}{#1}}}
\newrobustcmd{\hatmathfrakEz}[2][]{\ensuremath{\subp{\hat{\mathfrak{E}}}{}{#2}{}{#1}}}
\newrobustcmd{\widehatmathfrakEz}[2][]{\ensuremath{\subp{\widehat{\mathfrak{E}}}{}{#2}{}{#1}}}
\newrobustcmd{\checkmathfrakEz}[2][]{\ensuremath{\subp{\check{\mathfrak{E}}}{}{#2}{}{#1}}}
\newrobustcmd{\tildemathfrakEz}[2][]{\ensuremath{\subp{\tilde{\mathfrak{E}}}{}{#2}{}{#1}}}
\newrobustcmd{\widetildemathfrakEz}[2][]{\ensuremath{\subp{\widetilde{\mathfrak{E}}}{}{#2}{}{#1}}}
\newrobustcmd{\acutemathfrakEz}[2][]{\ensuremath{\subp{\acute{\mathfrak{E}}}{}{#2}{}{#1}}}
\newrobustcmd{\gravemathfrakEz}[2][]{\ensuremath{\subp{\grave{\mathfrak{E}}}{}{#2}{}{#1}}}
\newrobustcmd{\dotmathfrakEz}[2][]{\ensuremath{\subp{\dot{\mathfrak{E}}}{}{#2}{}{#1}}}
\newrobustcmd{\ddotmathfrakEz}[2][]{\ensuremath{\subp{\ddot{\mathfrak{E}}}{}{#2}{}{#1}}}
\newrobustcmd{\brevemathfrakEz}[2][]{\ensuremath{\subp{\breve{\mathfrak{E}}}{}{#2}{}{#1}}}
\newrobustcmd{\barmathfrakEz}[2][]{\ensuremath{\subp{\bar{\mathfrak{E}}}{}{#2}{}{#1}}}
\newrobustcmd{\vecmathfrakEz}[2][]{\ensuremath{\subp{\vec{\mathfrak{E}}}{}{#2}{}{#1}}}
\newrobustcmd{\bmmathfrakEz}[2][]{\ensuremath{\subp{\bm{\mathfrak{E}}}{}{#2}{}{#1}}}
\newrobustcmd{\hatbmmathfrakEz}[2][]{\ensuremath{\subp{\hat{\bm{\mathfrak{E}}}}{}{#2}{}{#1}}}
\newrobustcmd{\widehatbmmathfrakEz}[2][]{\ensuremath{\subp{\widehat{\bm{\mathfrak{E}}}}{}{#2}{}{#1}}}
\newrobustcmd{\checkbmmathfrakEz}[2][]{\ensuremath{\subp{\check{\bm{\mathfrak{E}}}}{}{#2}{}{#1}}}
\newrobustcmd{\tildebmmathfrakEz}[2][]{\ensuremath{\subp{\tilde{\bm{\mathfrak{E}}}}{}{#2}{}{#1}}}
\newrobustcmd{\widetildebmmathfrakEz}[2][]{\ensuremath{\subp{\widetilde{\bm{\mathfrak{E}}}}{}{#2}{}{#1}}}
\newrobustcmd{\acutebmmathfrakEz}[2][]{\ensuremath{\subp{\acute{\bm{\mathfrak{E}}}}{}{#2}{}{#1}}}
\newrobustcmd{\gravebmmathfrakEz}[2][]{\ensuremath{\subp{\grave{\bm{\mathfrak{E}}}}{}{#2}{}{#1}}}
\newrobustcmd{\dotbmmathfrakEz}[2][]{\ensuremath{\subp{\dot{\bm{\mathfrak{E}}}}{}{#2}{}{#1}}}
\newrobustcmd{\ddotbmmathfrakEz}[2][]{\ensuremath{\subp{\ddot{\bm{\mathfrak{E}}}}{}{#2}{}{#1}}}
\newrobustcmd{\brevebmmathfrakEz}[2][]{\ensuremath{\subp{\breve{\bm{\mathfrak{E}}}}{}{#2}{}{#1}}}
\newrobustcmd{\barbmmathfrakEz}[2][]{\ensuremath{\subp{\bar{\bm{\mathfrak{E}}}}{}{#2}{}{#1}}}
\newrobustcmd{\vecbmmathfrakEz}[2][]{\ensuremath{\subp{\vec{\bm{\mathfrak{E}}}}{}{#2}{}{#1}}}
\newrobustcmd{\mathfrakFz}[2][]{\ensuremath{\subp{\mathfrak{F}}{}{#2}{}{#1}}}
\newrobustcmd{\hatmathfrakFz}[2][]{\ensuremath{\subp{\hat{\mathfrak{F}}}{}{#2}{}{#1}}}
\newrobustcmd{\widehatmathfrakFz}[2][]{\ensuremath{\subp{\widehat{\mathfrak{F}}}{}{#2}{}{#1}}}
\newrobustcmd{\checkmathfrakFz}[2][]{\ensuremath{\subp{\check{\mathfrak{F}}}{}{#2}{}{#1}}}
\newrobustcmd{\tildemathfrakFz}[2][]{\ensuremath{\subp{\tilde{\mathfrak{F}}}{}{#2}{}{#1}}}
\newrobustcmd{\widetildemathfrakFz}[2][]{\ensuremath{\subp{\widetilde{\mathfrak{F}}}{}{#2}{}{#1}}}
\newrobustcmd{\acutemathfrakFz}[2][]{\ensuremath{\subp{\acute{\mathfrak{F}}}{}{#2}{}{#1}}}
\newrobustcmd{\gravemathfrakFz}[2][]{\ensuremath{\subp{\grave{\mathfrak{F}}}{}{#2}{}{#1}}}
\newrobustcmd{\dotmathfrakFz}[2][]{\ensuremath{\subp{\dot{\mathfrak{F}}}{}{#2}{}{#1}}}
\newrobustcmd{\ddotmathfrakFz}[2][]{\ensuremath{\subp{\ddot{\mathfrak{F}}}{}{#2}{}{#1}}}
\newrobustcmd{\brevemathfrakFz}[2][]{\ensuremath{\subp{\breve{\mathfrak{F}}}{}{#2}{}{#1}}}
\newrobustcmd{\barmathfrakFz}[2][]{\ensuremath{\subp{\bar{\mathfrak{F}}}{}{#2}{}{#1}}}
\newrobustcmd{\vecmathfrakFz}[2][]{\ensuremath{\subp{\vec{\mathfrak{F}}}{}{#2}{}{#1}}}
\newrobustcmd{\bmmathfrakFz}[2][]{\ensuremath{\subp{\bm{\mathfrak{F}}}{}{#2}{}{#1}}}
\newrobustcmd{\hatbmmathfrakFz}[2][]{\ensuremath{\subp{\hat{\bm{\mathfrak{F}}}}{}{#2}{}{#1}}}
\newrobustcmd{\widehatbmmathfrakFz}[2][]{\ensuremath{\subp{\widehat{\bm{\mathfrak{F}}}}{}{#2}{}{#1}}}
\newrobustcmd{\checkbmmathfrakFz}[2][]{\ensuremath{\subp{\check{\bm{\mathfrak{F}}}}{}{#2}{}{#1}}}
\newrobustcmd{\tildebmmathfrakFz}[2][]{\ensuremath{\subp{\tilde{\bm{\mathfrak{F}}}}{}{#2}{}{#1}}}
\newrobustcmd{\widetildebmmathfrakFz}[2][]{\ensuremath{\subp{\widetilde{\bm{\mathfrak{F}}}}{}{#2}{}{#1}}}
\newrobustcmd{\acutebmmathfrakFz}[2][]{\ensuremath{\subp{\acute{\bm{\mathfrak{F}}}}{}{#2}{}{#1}}}
\newrobustcmd{\gravebmmathfrakFz}[2][]{\ensuremath{\subp{\grave{\bm{\mathfrak{F}}}}{}{#2}{}{#1}}}
\newrobustcmd{\dotbmmathfrakFz}[2][]{\ensuremath{\subp{\dot{\bm{\mathfrak{F}}}}{}{#2}{}{#1}}}
\newrobustcmd{\ddotbmmathfrakFz}[2][]{\ensuremath{\subp{\ddot{\bm{\mathfrak{F}}}}{}{#2}{}{#1}}}
\newrobustcmd{\brevebmmathfrakFz}[2][]{\ensuremath{\subp{\breve{\bm{\mathfrak{F}}}}{}{#2}{}{#1}}}
\newrobustcmd{\barbmmathfrakFz}[2][]{\ensuremath{\subp{\bar{\bm{\mathfrak{F}}}}{}{#2}{}{#1}}}
\newrobustcmd{\vecbmmathfrakFz}[2][]{\ensuremath{\subp{\vec{\bm{\mathfrak{F}}}}{}{#2}{}{#1}}}
\newrobustcmd{\mathfrakGz}[2][]{\ensuremath{\subp{\mathfrak{G}}{}{#2}{}{#1}}}
\newrobustcmd{\hatmathfrakGz}[2][]{\ensuremath{\subp{\hat{\mathfrak{G}}}{}{#2}{}{#1}}}
\newrobustcmd{\widehatmathfrakGz}[2][]{\ensuremath{\subp{\widehat{\mathfrak{G}}}{}{#2}{}{#1}}}
\newrobustcmd{\checkmathfrakGz}[2][]{\ensuremath{\subp{\check{\mathfrak{G}}}{}{#2}{}{#1}}}
\newrobustcmd{\tildemathfrakGz}[2][]{\ensuremath{\subp{\tilde{\mathfrak{G}}}{}{#2}{}{#1}}}
\newrobustcmd{\widetildemathfrakGz}[2][]{\ensuremath{\subp{\widetilde{\mathfrak{G}}}{}{#2}{}{#1}}}
\newrobustcmd{\acutemathfrakGz}[2][]{\ensuremath{\subp{\acute{\mathfrak{G}}}{}{#2}{}{#1}}}
\newrobustcmd{\gravemathfrakGz}[2][]{\ensuremath{\subp{\grave{\mathfrak{G}}}{}{#2}{}{#1}}}
\newrobustcmd{\dotmathfrakGz}[2][]{\ensuremath{\subp{\dot{\mathfrak{G}}}{}{#2}{}{#1}}}
\newrobustcmd{\ddotmathfrakGz}[2][]{\ensuremath{\subp{\ddot{\mathfrak{G}}}{}{#2}{}{#1}}}
\newrobustcmd{\brevemathfrakGz}[2][]{\ensuremath{\subp{\breve{\mathfrak{G}}}{}{#2}{}{#1}}}
\newrobustcmd{\barmathfrakGz}[2][]{\ensuremath{\subp{\bar{\mathfrak{G}}}{}{#2}{}{#1}}}
\newrobustcmd{\vecmathfrakGz}[2][]{\ensuremath{\subp{\vec{\mathfrak{G}}}{}{#2}{}{#1}}}
\newrobustcmd{\bmmathfrakGz}[2][]{\ensuremath{\subp{\bm{\mathfrak{G}}}{}{#2}{}{#1}}}
\newrobustcmd{\hatbmmathfrakGz}[2][]{\ensuremath{\subp{\hat{\bm{\mathfrak{G}}}}{}{#2}{}{#1}}}
\newrobustcmd{\widehatbmmathfrakGz}[2][]{\ensuremath{\subp{\widehat{\bm{\mathfrak{G}}}}{}{#2}{}{#1}}}
\newrobustcmd{\checkbmmathfrakGz}[2][]{\ensuremath{\subp{\check{\bm{\mathfrak{G}}}}{}{#2}{}{#1}}}
\newrobustcmd{\tildebmmathfrakGz}[2][]{\ensuremath{\subp{\tilde{\bm{\mathfrak{G}}}}{}{#2}{}{#1}}}
\newrobustcmd{\widetildebmmathfrakGz}[2][]{\ensuremath{\subp{\widetilde{\bm{\mathfrak{G}}}}{}{#2}{}{#1}}}
\newrobustcmd{\acutebmmathfrakGz}[2][]{\ensuremath{\subp{\acute{\bm{\mathfrak{G}}}}{}{#2}{}{#1}}}
\newrobustcmd{\gravebmmathfrakGz}[2][]{\ensuremath{\subp{\grave{\bm{\mathfrak{G}}}}{}{#2}{}{#1}}}
\newrobustcmd{\dotbmmathfrakGz}[2][]{\ensuremath{\subp{\dot{\bm{\mathfrak{G}}}}{}{#2}{}{#1}}}
\newrobustcmd{\ddotbmmathfrakGz}[2][]{\ensuremath{\subp{\ddot{\bm{\mathfrak{G}}}}{}{#2}{}{#1}}}
\newrobustcmd{\brevebmmathfrakGz}[2][]{\ensuremath{\subp{\breve{\bm{\mathfrak{G}}}}{}{#2}{}{#1}}}
\newrobustcmd{\barbmmathfrakGz}[2][]{\ensuremath{\subp{\bar{\bm{\mathfrak{G}}}}{}{#2}{}{#1}}}
\newrobustcmd{\vecbmmathfrakGz}[2][]{\ensuremath{\subp{\vec{\bm{\mathfrak{G}}}}{}{#2}{}{#1}}}
\newrobustcmd{\mathfrakHz}[2][]{\ensuremath{\subp{\mathfrak{H}}{}{#2}{}{#1}}}
\newrobustcmd{\hatmathfrakHz}[2][]{\ensuremath{\subp{\hat{\mathfrak{H}}}{}{#2}{}{#1}}}
\newrobustcmd{\widehatmathfrakHz}[2][]{\ensuremath{\subp{\widehat{\mathfrak{H}}}{}{#2}{}{#1}}}
\newrobustcmd{\checkmathfrakHz}[2][]{\ensuremath{\subp{\check{\mathfrak{H}}}{}{#2}{}{#1}}}
\newrobustcmd{\tildemathfrakHz}[2][]{\ensuremath{\subp{\tilde{\mathfrak{H}}}{}{#2}{}{#1}}}
\newrobustcmd{\widetildemathfrakHz}[2][]{\ensuremath{\subp{\widetilde{\mathfrak{H}}}{}{#2}{}{#1}}}
\newrobustcmd{\acutemathfrakHz}[2][]{\ensuremath{\subp{\acute{\mathfrak{H}}}{}{#2}{}{#1}}}
\newrobustcmd{\gravemathfrakHz}[2][]{\ensuremath{\subp{\grave{\mathfrak{H}}}{}{#2}{}{#1}}}
\newrobustcmd{\dotmathfrakHz}[2][]{\ensuremath{\subp{\dot{\mathfrak{H}}}{}{#2}{}{#1}}}
\newrobustcmd{\ddotmathfrakHz}[2][]{\ensuremath{\subp{\ddot{\mathfrak{H}}}{}{#2}{}{#1}}}
\newrobustcmd{\brevemathfrakHz}[2][]{\ensuremath{\subp{\breve{\mathfrak{H}}}{}{#2}{}{#1}}}
\newrobustcmd{\barmathfrakHz}[2][]{\ensuremath{\subp{\bar{\mathfrak{H}}}{}{#2}{}{#1}}}
\newrobustcmd{\vecmathfrakHz}[2][]{\ensuremath{\subp{\vec{\mathfrak{H}}}{}{#2}{}{#1}}}
\newrobustcmd{\bmmathfrakHz}[2][]{\ensuremath{\subp{\bm{\mathfrak{H}}}{}{#2}{}{#1}}}
\newrobustcmd{\hatbmmathfrakHz}[2][]{\ensuremath{\subp{\hat{\bm{\mathfrak{H}}}}{}{#2}{}{#1}}}
\newrobustcmd{\widehatbmmathfrakHz}[2][]{\ensuremath{\subp{\widehat{\bm{\mathfrak{H}}}}{}{#2}{}{#1}}}
\newrobustcmd{\checkbmmathfrakHz}[2][]{\ensuremath{\subp{\check{\bm{\mathfrak{H}}}}{}{#2}{}{#1}}}
\newrobustcmd{\tildebmmathfrakHz}[2][]{\ensuremath{\subp{\tilde{\bm{\mathfrak{H}}}}{}{#2}{}{#1}}}
\newrobustcmd{\widetildebmmathfrakHz}[2][]{\ensuremath{\subp{\widetilde{\bm{\mathfrak{H}}}}{}{#2}{}{#1}}}
\newrobustcmd{\acutebmmathfrakHz}[2][]{\ensuremath{\subp{\acute{\bm{\mathfrak{H}}}}{}{#2}{}{#1}}}
\newrobustcmd{\gravebmmathfrakHz}[2][]{\ensuremath{\subp{\grave{\bm{\mathfrak{H}}}}{}{#2}{}{#1}}}
\newrobustcmd{\dotbmmathfrakHz}[2][]{\ensuremath{\subp{\dot{\bm{\mathfrak{H}}}}{}{#2}{}{#1}}}
\newrobustcmd{\ddotbmmathfrakHz}[2][]{\ensuremath{\subp{\ddot{\bm{\mathfrak{H}}}}{}{#2}{}{#1}}}
\newrobustcmd{\brevebmmathfrakHz}[2][]{\ensuremath{\subp{\breve{\bm{\mathfrak{H}}}}{}{#2}{}{#1}}}
\newrobustcmd{\barbmmathfrakHz}[2][]{\ensuremath{\subp{\bar{\bm{\mathfrak{H}}}}{}{#2}{}{#1}}}
\newrobustcmd{\vecbmmathfrakHz}[2][]{\ensuremath{\subp{\vec{\bm{\mathfrak{H}}}}{}{#2}{}{#1}}}
\newrobustcmd{\mathfrakIz}[2][]{\ensuremath{\subp{\mathfrak{I}}{}{#2}{}{#1}}}
\newrobustcmd{\hatmathfrakIz}[2][]{\ensuremath{\subp{\hat{\mathfrak{I}}}{}{#2}{}{#1}}}
\newrobustcmd{\widehatmathfrakIz}[2][]{\ensuremath{\subp{\widehat{\mathfrak{I}}}{}{#2}{}{#1}}}
\newrobustcmd{\checkmathfrakIz}[2][]{\ensuremath{\subp{\check{\mathfrak{I}}}{}{#2}{}{#1}}}
\newrobustcmd{\tildemathfrakIz}[2][]{\ensuremath{\subp{\tilde{\mathfrak{I}}}{}{#2}{}{#1}}}
\newrobustcmd{\widetildemathfrakIz}[2][]{\ensuremath{\subp{\widetilde{\mathfrak{I}}}{}{#2}{}{#1}}}
\newrobustcmd{\acutemathfrakIz}[2][]{\ensuremath{\subp{\acute{\mathfrak{I}}}{}{#2}{}{#1}}}
\newrobustcmd{\gravemathfrakIz}[2][]{\ensuremath{\subp{\grave{\mathfrak{I}}}{}{#2}{}{#1}}}
\newrobustcmd{\dotmathfrakIz}[2][]{\ensuremath{\subp{\dot{\mathfrak{I}}}{}{#2}{}{#1}}}
\newrobustcmd{\ddotmathfrakIz}[2][]{\ensuremath{\subp{\ddot{\mathfrak{I}}}{}{#2}{}{#1}}}
\newrobustcmd{\brevemathfrakIz}[2][]{\ensuremath{\subp{\breve{\mathfrak{I}}}{}{#2}{}{#1}}}
\newrobustcmd{\barmathfrakIz}[2][]{\ensuremath{\subp{\bar{\mathfrak{I}}}{}{#2}{}{#1}}}
\newrobustcmd{\vecmathfrakIz}[2][]{\ensuremath{\subp{\vec{\mathfrak{I}}}{}{#2}{}{#1}}}
\newrobustcmd{\bmmathfrakIz}[2][]{\ensuremath{\subp{\bm{\mathfrak{I}}}{}{#2}{}{#1}}}
\newrobustcmd{\hatbmmathfrakIz}[2][]{\ensuremath{\subp{\hat{\bm{\mathfrak{I}}}}{}{#2}{}{#1}}}
\newrobustcmd{\widehatbmmathfrakIz}[2][]{\ensuremath{\subp{\widehat{\bm{\mathfrak{I}}}}{}{#2}{}{#1}}}
\newrobustcmd{\checkbmmathfrakIz}[2][]{\ensuremath{\subp{\check{\bm{\mathfrak{I}}}}{}{#2}{}{#1}}}
\newrobustcmd{\tildebmmathfrakIz}[2][]{\ensuremath{\subp{\tilde{\bm{\mathfrak{I}}}}{}{#2}{}{#1}}}
\newrobustcmd{\widetildebmmathfrakIz}[2][]{\ensuremath{\subp{\widetilde{\bm{\mathfrak{I}}}}{}{#2}{}{#1}}}
\newrobustcmd{\acutebmmathfrakIz}[2][]{\ensuremath{\subp{\acute{\bm{\mathfrak{I}}}}{}{#2}{}{#1}}}
\newrobustcmd{\gravebmmathfrakIz}[2][]{\ensuremath{\subp{\grave{\bm{\mathfrak{I}}}}{}{#2}{}{#1}}}
\newrobustcmd{\dotbmmathfrakIz}[2][]{\ensuremath{\subp{\dot{\bm{\mathfrak{I}}}}{}{#2}{}{#1}}}
\newrobustcmd{\ddotbmmathfrakIz}[2][]{\ensuremath{\subp{\ddot{\bm{\mathfrak{I}}}}{}{#2}{}{#1}}}
\newrobustcmd{\brevebmmathfrakIz}[2][]{\ensuremath{\subp{\breve{\bm{\mathfrak{I}}}}{}{#2}{}{#1}}}
\newrobustcmd{\barbmmathfrakIz}[2][]{\ensuremath{\subp{\bar{\bm{\mathfrak{I}}}}{}{#2}{}{#1}}}
\newrobustcmd{\vecbmmathfrakIz}[2][]{\ensuremath{\subp{\vec{\bm{\mathfrak{I}}}}{}{#2}{}{#1}}}
\newrobustcmd{\mathfrakJz}[2][]{\ensuremath{\subp{\mathfrak{J}}{}{#2}{}{#1}}}
\newrobustcmd{\hatmathfrakJz}[2][]{\ensuremath{\subp{\hat{\mathfrak{J}}}{}{#2}{}{#1}}}
\newrobustcmd{\widehatmathfrakJz}[2][]{\ensuremath{\subp{\widehat{\mathfrak{J}}}{}{#2}{}{#1}}}
\newrobustcmd{\checkmathfrakJz}[2][]{\ensuremath{\subp{\check{\mathfrak{J}}}{}{#2}{}{#1}}}
\newrobustcmd{\tildemathfrakJz}[2][]{\ensuremath{\subp{\tilde{\mathfrak{J}}}{}{#2}{}{#1}}}
\newrobustcmd{\widetildemathfrakJz}[2][]{\ensuremath{\subp{\widetilde{\mathfrak{J}}}{}{#2}{}{#1}}}
\newrobustcmd{\acutemathfrakJz}[2][]{\ensuremath{\subp{\acute{\mathfrak{J}}}{}{#2}{}{#1}}}
\newrobustcmd{\gravemathfrakJz}[2][]{\ensuremath{\subp{\grave{\mathfrak{J}}}{}{#2}{}{#1}}}
\newrobustcmd{\dotmathfrakJz}[2][]{\ensuremath{\subp{\dot{\mathfrak{J}}}{}{#2}{}{#1}}}
\newrobustcmd{\ddotmathfrakJz}[2][]{\ensuremath{\subp{\ddot{\mathfrak{J}}}{}{#2}{}{#1}}}
\newrobustcmd{\brevemathfrakJz}[2][]{\ensuremath{\subp{\breve{\mathfrak{J}}}{}{#2}{}{#1}}}
\newrobustcmd{\barmathfrakJz}[2][]{\ensuremath{\subp{\bar{\mathfrak{J}}}{}{#2}{}{#1}}}
\newrobustcmd{\vecmathfrakJz}[2][]{\ensuremath{\subp{\vec{\mathfrak{J}}}{}{#2}{}{#1}}}
\newrobustcmd{\bmmathfrakJz}[2][]{\ensuremath{\subp{\bm{\mathfrak{J}}}{}{#2}{}{#1}}}
\newrobustcmd{\hatbmmathfrakJz}[2][]{\ensuremath{\subp{\hat{\bm{\mathfrak{J}}}}{}{#2}{}{#1}}}
\newrobustcmd{\widehatbmmathfrakJz}[2][]{\ensuremath{\subp{\widehat{\bm{\mathfrak{J}}}}{}{#2}{}{#1}}}
\newrobustcmd{\checkbmmathfrakJz}[2][]{\ensuremath{\subp{\check{\bm{\mathfrak{J}}}}{}{#2}{}{#1}}}
\newrobustcmd{\tildebmmathfrakJz}[2][]{\ensuremath{\subp{\tilde{\bm{\mathfrak{J}}}}{}{#2}{}{#1}}}
\newrobustcmd{\widetildebmmathfrakJz}[2][]{\ensuremath{\subp{\widetilde{\bm{\mathfrak{J}}}}{}{#2}{}{#1}}}
\newrobustcmd{\acutebmmathfrakJz}[2][]{\ensuremath{\subp{\acute{\bm{\mathfrak{J}}}}{}{#2}{}{#1}}}
\newrobustcmd{\gravebmmathfrakJz}[2][]{\ensuremath{\subp{\grave{\bm{\mathfrak{J}}}}{}{#2}{}{#1}}}
\newrobustcmd{\dotbmmathfrakJz}[2][]{\ensuremath{\subp{\dot{\bm{\mathfrak{J}}}}{}{#2}{}{#1}}}
\newrobustcmd{\ddotbmmathfrakJz}[2][]{\ensuremath{\subp{\ddot{\bm{\mathfrak{J}}}}{}{#2}{}{#1}}}
\newrobustcmd{\brevebmmathfrakJz}[2][]{\ensuremath{\subp{\breve{\bm{\mathfrak{J}}}}{}{#2}{}{#1}}}
\newrobustcmd{\barbmmathfrakJz}[2][]{\ensuremath{\subp{\bar{\bm{\mathfrak{J}}}}{}{#2}{}{#1}}}
\newrobustcmd{\vecbmmathfrakJz}[2][]{\ensuremath{\subp{\vec{\bm{\mathfrak{J}}}}{}{#2}{}{#1}}}
\newrobustcmd{\mathfrakKz}[2][]{\ensuremath{\subp{\mathfrak{K}}{}{#2}{}{#1}}}
\newrobustcmd{\hatmathfrakKz}[2][]{\ensuremath{\subp{\hat{\mathfrak{K}}}{}{#2}{}{#1}}}
\newrobustcmd{\widehatmathfrakKz}[2][]{\ensuremath{\subp{\widehat{\mathfrak{K}}}{}{#2}{}{#1}}}
\newrobustcmd{\checkmathfrakKz}[2][]{\ensuremath{\subp{\check{\mathfrak{K}}}{}{#2}{}{#1}}}
\newrobustcmd{\tildemathfrakKz}[2][]{\ensuremath{\subp{\tilde{\mathfrak{K}}}{}{#2}{}{#1}}}
\newrobustcmd{\widetildemathfrakKz}[2][]{\ensuremath{\subp{\widetilde{\mathfrak{K}}}{}{#2}{}{#1}}}
\newrobustcmd{\acutemathfrakKz}[2][]{\ensuremath{\subp{\acute{\mathfrak{K}}}{}{#2}{}{#1}}}
\newrobustcmd{\gravemathfrakKz}[2][]{\ensuremath{\subp{\grave{\mathfrak{K}}}{}{#2}{}{#1}}}
\newrobustcmd{\dotmathfrakKz}[2][]{\ensuremath{\subp{\dot{\mathfrak{K}}}{}{#2}{}{#1}}}
\newrobustcmd{\ddotmathfrakKz}[2][]{\ensuremath{\subp{\ddot{\mathfrak{K}}}{}{#2}{}{#1}}}
\newrobustcmd{\brevemathfrakKz}[2][]{\ensuremath{\subp{\breve{\mathfrak{K}}}{}{#2}{}{#1}}}
\newrobustcmd{\barmathfrakKz}[2][]{\ensuremath{\subp{\bar{\mathfrak{K}}}{}{#2}{}{#1}}}
\newrobustcmd{\vecmathfrakKz}[2][]{\ensuremath{\subp{\vec{\mathfrak{K}}}{}{#2}{}{#1}}}
\newrobustcmd{\bmmathfrakKz}[2][]{\ensuremath{\subp{\bm{\mathfrak{K}}}{}{#2}{}{#1}}}
\newrobustcmd{\hatbmmathfrakKz}[2][]{\ensuremath{\subp{\hat{\bm{\mathfrak{K}}}}{}{#2}{}{#1}}}
\newrobustcmd{\widehatbmmathfrakKz}[2][]{\ensuremath{\subp{\widehat{\bm{\mathfrak{K}}}}{}{#2}{}{#1}}}
\newrobustcmd{\checkbmmathfrakKz}[2][]{\ensuremath{\subp{\check{\bm{\mathfrak{K}}}}{}{#2}{}{#1}}}
\newrobustcmd{\tildebmmathfrakKz}[2][]{\ensuremath{\subp{\tilde{\bm{\mathfrak{K}}}}{}{#2}{}{#1}}}
\newrobustcmd{\widetildebmmathfrakKz}[2][]{\ensuremath{\subp{\widetilde{\bm{\mathfrak{K}}}}{}{#2}{}{#1}}}
\newrobustcmd{\acutebmmathfrakKz}[2][]{\ensuremath{\subp{\acute{\bm{\mathfrak{K}}}}{}{#2}{}{#1}}}
\newrobustcmd{\gravebmmathfrakKz}[2][]{\ensuremath{\subp{\grave{\bm{\mathfrak{K}}}}{}{#2}{}{#1}}}
\newrobustcmd{\dotbmmathfrakKz}[2][]{\ensuremath{\subp{\dot{\bm{\mathfrak{K}}}}{}{#2}{}{#1}}}
\newrobustcmd{\ddotbmmathfrakKz}[2][]{\ensuremath{\subp{\ddot{\bm{\mathfrak{K}}}}{}{#2}{}{#1}}}
\newrobustcmd{\brevebmmathfrakKz}[2][]{\ensuremath{\subp{\breve{\bm{\mathfrak{K}}}}{}{#2}{}{#1}}}
\newrobustcmd{\barbmmathfrakKz}[2][]{\ensuremath{\subp{\bar{\bm{\mathfrak{K}}}}{}{#2}{}{#1}}}
\newrobustcmd{\vecbmmathfrakKz}[2][]{\ensuremath{\subp{\vec{\bm{\mathfrak{K}}}}{}{#2}{}{#1}}}
\newrobustcmd{\mathfrakLz}[2][]{\ensuremath{\subp{\mathfrak{L}}{}{#2}{}{#1}}}
\newrobustcmd{\hatmathfrakLz}[2][]{\ensuremath{\subp{\hat{\mathfrak{L}}}{}{#2}{}{#1}}}
\newrobustcmd{\widehatmathfrakLz}[2][]{\ensuremath{\subp{\widehat{\mathfrak{L}}}{}{#2}{}{#1}}}
\newrobustcmd{\checkmathfrakLz}[2][]{\ensuremath{\subp{\check{\mathfrak{L}}}{}{#2}{}{#1}}}
\newrobustcmd{\tildemathfrakLz}[2][]{\ensuremath{\subp{\tilde{\mathfrak{L}}}{}{#2}{}{#1}}}
\newrobustcmd{\widetildemathfrakLz}[2][]{\ensuremath{\subp{\widetilde{\mathfrak{L}}}{}{#2}{}{#1}}}
\newrobustcmd{\acutemathfrakLz}[2][]{\ensuremath{\subp{\acute{\mathfrak{L}}}{}{#2}{}{#1}}}
\newrobustcmd{\gravemathfrakLz}[2][]{\ensuremath{\subp{\grave{\mathfrak{L}}}{}{#2}{}{#1}}}
\newrobustcmd{\dotmathfrakLz}[2][]{\ensuremath{\subp{\dot{\mathfrak{L}}}{}{#2}{}{#1}}}
\newrobustcmd{\ddotmathfrakLz}[2][]{\ensuremath{\subp{\ddot{\mathfrak{L}}}{}{#2}{}{#1}}}
\newrobustcmd{\brevemathfrakLz}[2][]{\ensuremath{\subp{\breve{\mathfrak{L}}}{}{#2}{}{#1}}}
\newrobustcmd{\barmathfrakLz}[2][]{\ensuremath{\subp{\bar{\mathfrak{L}}}{}{#2}{}{#1}}}
\newrobustcmd{\vecmathfrakLz}[2][]{\ensuremath{\subp{\vec{\mathfrak{L}}}{}{#2}{}{#1}}}
\newrobustcmd{\bmmathfrakLz}[2][]{\ensuremath{\subp{\bm{\mathfrak{L}}}{}{#2}{}{#1}}}
\newrobustcmd{\hatbmmathfrakLz}[2][]{\ensuremath{\subp{\hat{\bm{\mathfrak{L}}}}{}{#2}{}{#1}}}
\newrobustcmd{\widehatbmmathfrakLz}[2][]{\ensuremath{\subp{\widehat{\bm{\mathfrak{L}}}}{}{#2}{}{#1}}}
\newrobustcmd{\checkbmmathfrakLz}[2][]{\ensuremath{\subp{\check{\bm{\mathfrak{L}}}}{}{#2}{}{#1}}}
\newrobustcmd{\tildebmmathfrakLz}[2][]{\ensuremath{\subp{\tilde{\bm{\mathfrak{L}}}}{}{#2}{}{#1}}}
\newrobustcmd{\widetildebmmathfrakLz}[2][]{\ensuremath{\subp{\widetilde{\bm{\mathfrak{L}}}}{}{#2}{}{#1}}}
\newrobustcmd{\acutebmmathfrakLz}[2][]{\ensuremath{\subp{\acute{\bm{\mathfrak{L}}}}{}{#2}{}{#1}}}
\newrobustcmd{\gravebmmathfrakLz}[2][]{\ensuremath{\subp{\grave{\bm{\mathfrak{L}}}}{}{#2}{}{#1}}}
\newrobustcmd{\dotbmmathfrakLz}[2][]{\ensuremath{\subp{\dot{\bm{\mathfrak{L}}}}{}{#2}{}{#1}}}
\newrobustcmd{\ddotbmmathfrakLz}[2][]{\ensuremath{\subp{\ddot{\bm{\mathfrak{L}}}}{}{#2}{}{#1}}}
\newrobustcmd{\brevebmmathfrakLz}[2][]{\ensuremath{\subp{\breve{\bm{\mathfrak{L}}}}{}{#2}{}{#1}}}
\newrobustcmd{\barbmmathfrakLz}[2][]{\ensuremath{\subp{\bar{\bm{\mathfrak{L}}}}{}{#2}{}{#1}}}
\newrobustcmd{\vecbmmathfrakLz}[2][]{\ensuremath{\subp{\vec{\bm{\mathfrak{L}}}}{}{#2}{}{#1}}}
\newrobustcmd{\mathfrakMz}[2][]{\ensuremath{\subp{\mathfrak{M}}{}{#2}{}{#1}}}
\newrobustcmd{\hatmathfrakMz}[2][]{\ensuremath{\subp{\hat{\mathfrak{M}}}{}{#2}{}{#1}}}
\newrobustcmd{\widehatmathfrakMz}[2][]{\ensuremath{\subp{\widehat{\mathfrak{M}}}{}{#2}{}{#1}}}
\newrobustcmd{\checkmathfrakMz}[2][]{\ensuremath{\subp{\check{\mathfrak{M}}}{}{#2}{}{#1}}}
\newrobustcmd{\tildemathfrakMz}[2][]{\ensuremath{\subp{\tilde{\mathfrak{M}}}{}{#2}{}{#1}}}
\newrobustcmd{\widetildemathfrakMz}[2][]{\ensuremath{\subp{\widetilde{\mathfrak{M}}}{}{#2}{}{#1}}}
\newrobustcmd{\acutemathfrakMz}[2][]{\ensuremath{\subp{\acute{\mathfrak{M}}}{}{#2}{}{#1}}}
\newrobustcmd{\gravemathfrakMz}[2][]{\ensuremath{\subp{\grave{\mathfrak{M}}}{}{#2}{}{#1}}}
\newrobustcmd{\dotmathfrakMz}[2][]{\ensuremath{\subp{\dot{\mathfrak{M}}}{}{#2}{}{#1}}}
\newrobustcmd{\ddotmathfrakMz}[2][]{\ensuremath{\subp{\ddot{\mathfrak{M}}}{}{#2}{}{#1}}}
\newrobustcmd{\brevemathfrakMz}[2][]{\ensuremath{\subp{\breve{\mathfrak{M}}}{}{#2}{}{#1}}}
\newrobustcmd{\barmathfrakMz}[2][]{\ensuremath{\subp{\bar{\mathfrak{M}}}{}{#2}{}{#1}}}
\newrobustcmd{\vecmathfrakMz}[2][]{\ensuremath{\subp{\vec{\mathfrak{M}}}{}{#2}{}{#1}}}
\newrobustcmd{\bmmathfrakMz}[2][]{\ensuremath{\subp{\bm{\mathfrak{M}}}{}{#2}{}{#1}}}
\newrobustcmd{\hatbmmathfrakMz}[2][]{\ensuremath{\subp{\hat{\bm{\mathfrak{M}}}}{}{#2}{}{#1}}}
\newrobustcmd{\widehatbmmathfrakMz}[2][]{\ensuremath{\subp{\widehat{\bm{\mathfrak{M}}}}{}{#2}{}{#1}}}
\newrobustcmd{\checkbmmathfrakMz}[2][]{\ensuremath{\subp{\check{\bm{\mathfrak{M}}}}{}{#2}{}{#1}}}
\newrobustcmd{\tildebmmathfrakMz}[2][]{\ensuremath{\subp{\tilde{\bm{\mathfrak{M}}}}{}{#2}{}{#1}}}
\newrobustcmd{\widetildebmmathfrakMz}[2][]{\ensuremath{\subp{\widetilde{\bm{\mathfrak{M}}}}{}{#2}{}{#1}}}
\newrobustcmd{\acutebmmathfrakMz}[2][]{\ensuremath{\subp{\acute{\bm{\mathfrak{M}}}}{}{#2}{}{#1}}}
\newrobustcmd{\gravebmmathfrakMz}[2][]{\ensuremath{\subp{\grave{\bm{\mathfrak{M}}}}{}{#2}{}{#1}}}
\newrobustcmd{\dotbmmathfrakMz}[2][]{\ensuremath{\subp{\dot{\bm{\mathfrak{M}}}}{}{#2}{}{#1}}}
\newrobustcmd{\ddotbmmathfrakMz}[2][]{\ensuremath{\subp{\ddot{\bm{\mathfrak{M}}}}{}{#2}{}{#1}}}
\newrobustcmd{\brevebmmathfrakMz}[2][]{\ensuremath{\subp{\breve{\bm{\mathfrak{M}}}}{}{#2}{}{#1}}}
\newrobustcmd{\barbmmathfrakMz}[2][]{\ensuremath{\subp{\bar{\bm{\mathfrak{M}}}}{}{#2}{}{#1}}}
\newrobustcmd{\vecbmmathfrakMz}[2][]{\ensuremath{\subp{\vec{\bm{\mathfrak{M}}}}{}{#2}{}{#1}}}
\newrobustcmd{\mathfrakNz}[2][]{\ensuremath{\subp{\mathfrak{N}}{}{#2}{}{#1}}}
\newrobustcmd{\hatmathfrakNz}[2][]{\ensuremath{\subp{\hat{\mathfrak{N}}}{}{#2}{}{#1}}}
\newrobustcmd{\widehatmathfrakNz}[2][]{\ensuremath{\subp{\widehat{\mathfrak{N}}}{}{#2}{}{#1}}}
\newrobustcmd{\checkmathfrakNz}[2][]{\ensuremath{\subp{\check{\mathfrak{N}}}{}{#2}{}{#1}}}
\newrobustcmd{\tildemathfrakNz}[2][]{\ensuremath{\subp{\tilde{\mathfrak{N}}}{}{#2}{}{#1}}}
\newrobustcmd{\widetildemathfrakNz}[2][]{\ensuremath{\subp{\widetilde{\mathfrak{N}}}{}{#2}{}{#1}}}
\newrobustcmd{\acutemathfrakNz}[2][]{\ensuremath{\subp{\acute{\mathfrak{N}}}{}{#2}{}{#1}}}
\newrobustcmd{\gravemathfrakNz}[2][]{\ensuremath{\subp{\grave{\mathfrak{N}}}{}{#2}{}{#1}}}
\newrobustcmd{\dotmathfrakNz}[2][]{\ensuremath{\subp{\dot{\mathfrak{N}}}{}{#2}{}{#1}}}
\newrobustcmd{\ddotmathfrakNz}[2][]{\ensuremath{\subp{\ddot{\mathfrak{N}}}{}{#2}{}{#1}}}
\newrobustcmd{\brevemathfrakNz}[2][]{\ensuremath{\subp{\breve{\mathfrak{N}}}{}{#2}{}{#1}}}
\newrobustcmd{\barmathfrakNz}[2][]{\ensuremath{\subp{\bar{\mathfrak{N}}}{}{#2}{}{#1}}}
\newrobustcmd{\vecmathfrakNz}[2][]{\ensuremath{\subp{\vec{\mathfrak{N}}}{}{#2}{}{#1}}}
\newrobustcmd{\bmmathfrakNz}[2][]{\ensuremath{\subp{\bm{\mathfrak{N}}}{}{#2}{}{#1}}}
\newrobustcmd{\hatbmmathfrakNz}[2][]{\ensuremath{\subp{\hat{\bm{\mathfrak{N}}}}{}{#2}{}{#1}}}
\newrobustcmd{\widehatbmmathfrakNz}[2][]{\ensuremath{\subp{\widehat{\bm{\mathfrak{N}}}}{}{#2}{}{#1}}}
\newrobustcmd{\checkbmmathfrakNz}[2][]{\ensuremath{\subp{\check{\bm{\mathfrak{N}}}}{}{#2}{}{#1}}}
\newrobustcmd{\tildebmmathfrakNz}[2][]{\ensuremath{\subp{\tilde{\bm{\mathfrak{N}}}}{}{#2}{}{#1}}}
\newrobustcmd{\widetildebmmathfrakNz}[2][]{\ensuremath{\subp{\widetilde{\bm{\mathfrak{N}}}}{}{#2}{}{#1}}}
\newrobustcmd{\acutebmmathfrakNz}[2][]{\ensuremath{\subp{\acute{\bm{\mathfrak{N}}}}{}{#2}{}{#1}}}
\newrobustcmd{\gravebmmathfrakNz}[2][]{\ensuremath{\subp{\grave{\bm{\mathfrak{N}}}}{}{#2}{}{#1}}}
\newrobustcmd{\dotbmmathfrakNz}[2][]{\ensuremath{\subp{\dot{\bm{\mathfrak{N}}}}{}{#2}{}{#1}}}
\newrobustcmd{\ddotbmmathfrakNz}[2][]{\ensuremath{\subp{\ddot{\bm{\mathfrak{N}}}}{}{#2}{}{#1}}}
\newrobustcmd{\brevebmmathfrakNz}[2][]{\ensuremath{\subp{\breve{\bm{\mathfrak{N}}}}{}{#2}{}{#1}}}
\newrobustcmd{\barbmmathfrakNz}[2][]{\ensuremath{\subp{\bar{\bm{\mathfrak{N}}}}{}{#2}{}{#1}}}
\newrobustcmd{\vecbmmathfrakNz}[2][]{\ensuremath{\subp{\vec{\bm{\mathfrak{N}}}}{}{#2}{}{#1}}}
\newrobustcmd{\mathfrakOz}[2][]{\ensuremath{\subp{\mathfrak{O}}{}{#2}{}{#1}}}
\newrobustcmd{\hatmathfrakOz}[2][]{\ensuremath{\subp{\hat{\mathfrak{O}}}{}{#2}{}{#1}}}
\newrobustcmd{\widehatmathfrakOz}[2][]{\ensuremath{\subp{\widehat{\mathfrak{O}}}{}{#2}{}{#1}}}
\newrobustcmd{\checkmathfrakOz}[2][]{\ensuremath{\subp{\check{\mathfrak{O}}}{}{#2}{}{#1}}}
\newrobustcmd{\tildemathfrakOz}[2][]{\ensuremath{\subp{\tilde{\mathfrak{O}}}{}{#2}{}{#1}}}
\newrobustcmd{\widetildemathfrakOz}[2][]{\ensuremath{\subp{\widetilde{\mathfrak{O}}}{}{#2}{}{#1}}}
\newrobustcmd{\acutemathfrakOz}[2][]{\ensuremath{\subp{\acute{\mathfrak{O}}}{}{#2}{}{#1}}}
\newrobustcmd{\gravemathfrakOz}[2][]{\ensuremath{\subp{\grave{\mathfrak{O}}}{}{#2}{}{#1}}}
\newrobustcmd{\dotmathfrakOz}[2][]{\ensuremath{\subp{\dot{\mathfrak{O}}}{}{#2}{}{#1}}}
\newrobustcmd{\ddotmathfrakOz}[2][]{\ensuremath{\subp{\ddot{\mathfrak{O}}}{}{#2}{}{#1}}}
\newrobustcmd{\brevemathfrakOz}[2][]{\ensuremath{\subp{\breve{\mathfrak{O}}}{}{#2}{}{#1}}}
\newrobustcmd{\barmathfrakOz}[2][]{\ensuremath{\subp{\bar{\mathfrak{O}}}{}{#2}{}{#1}}}
\newrobustcmd{\vecmathfrakOz}[2][]{\ensuremath{\subp{\vec{\mathfrak{O}}}{}{#2}{}{#1}}}
\newrobustcmd{\bmmathfrakOz}[2][]{\ensuremath{\subp{\bm{\mathfrak{O}}}{}{#2}{}{#1}}}
\newrobustcmd{\hatbmmathfrakOz}[2][]{\ensuremath{\subp{\hat{\bm{\mathfrak{O}}}}{}{#2}{}{#1}}}
\newrobustcmd{\widehatbmmathfrakOz}[2][]{\ensuremath{\subp{\widehat{\bm{\mathfrak{O}}}}{}{#2}{}{#1}}}
\newrobustcmd{\checkbmmathfrakOz}[2][]{\ensuremath{\subp{\check{\bm{\mathfrak{O}}}}{}{#2}{}{#1}}}
\newrobustcmd{\tildebmmathfrakOz}[2][]{\ensuremath{\subp{\tilde{\bm{\mathfrak{O}}}}{}{#2}{}{#1}}}
\newrobustcmd{\widetildebmmathfrakOz}[2][]{\ensuremath{\subp{\widetilde{\bm{\mathfrak{O}}}}{}{#2}{}{#1}}}
\newrobustcmd{\acutebmmathfrakOz}[2][]{\ensuremath{\subp{\acute{\bm{\mathfrak{O}}}}{}{#2}{}{#1}}}
\newrobustcmd{\gravebmmathfrakOz}[2][]{\ensuremath{\subp{\grave{\bm{\mathfrak{O}}}}{}{#2}{}{#1}}}
\newrobustcmd{\dotbmmathfrakOz}[2][]{\ensuremath{\subp{\dot{\bm{\mathfrak{O}}}}{}{#2}{}{#1}}}
\newrobustcmd{\ddotbmmathfrakOz}[2][]{\ensuremath{\subp{\ddot{\bm{\mathfrak{O}}}}{}{#2}{}{#1}}}
\newrobustcmd{\brevebmmathfrakOz}[2][]{\ensuremath{\subp{\breve{\bm{\mathfrak{O}}}}{}{#2}{}{#1}}}
\newrobustcmd{\barbmmathfrakOz}[2][]{\ensuremath{\subp{\bar{\bm{\mathfrak{O}}}}{}{#2}{}{#1}}}
\newrobustcmd{\vecbmmathfrakOz}[2][]{\ensuremath{\subp{\vec{\bm{\mathfrak{O}}}}{}{#2}{}{#1}}}
\newrobustcmd{\mathfrakPz}[2][]{\ensuremath{\subp{\mathfrak{P}}{}{#2}{}{#1}}}
\newrobustcmd{\hatmathfrakPz}[2][]{\ensuremath{\subp{\hat{\mathfrak{P}}}{}{#2}{}{#1}}}
\newrobustcmd{\widehatmathfrakPz}[2][]{\ensuremath{\subp{\widehat{\mathfrak{P}}}{}{#2}{}{#1}}}
\newrobustcmd{\checkmathfrakPz}[2][]{\ensuremath{\subp{\check{\mathfrak{P}}}{}{#2}{}{#1}}}
\newrobustcmd{\tildemathfrakPz}[2][]{\ensuremath{\subp{\tilde{\mathfrak{P}}}{}{#2}{}{#1}}}
\newrobustcmd{\widetildemathfrakPz}[2][]{\ensuremath{\subp{\widetilde{\mathfrak{P}}}{}{#2}{}{#1}}}
\newrobustcmd{\acutemathfrakPz}[2][]{\ensuremath{\subp{\acute{\mathfrak{P}}}{}{#2}{}{#1}}}
\newrobustcmd{\gravemathfrakPz}[2][]{\ensuremath{\subp{\grave{\mathfrak{P}}}{}{#2}{}{#1}}}
\newrobustcmd{\dotmathfrakPz}[2][]{\ensuremath{\subp{\dot{\mathfrak{P}}}{}{#2}{}{#1}}}
\newrobustcmd{\ddotmathfrakPz}[2][]{\ensuremath{\subp{\ddot{\mathfrak{P}}}{}{#2}{}{#1}}}
\newrobustcmd{\brevemathfrakPz}[2][]{\ensuremath{\subp{\breve{\mathfrak{P}}}{}{#2}{}{#1}}}
\newrobustcmd{\barmathfrakPz}[2][]{\ensuremath{\subp{\bar{\mathfrak{P}}}{}{#2}{}{#1}}}
\newrobustcmd{\vecmathfrakPz}[2][]{\ensuremath{\subp{\vec{\mathfrak{P}}}{}{#2}{}{#1}}}
\newrobustcmd{\bmmathfrakPz}[2][]{\ensuremath{\subp{\bm{\mathfrak{P}}}{}{#2}{}{#1}}}
\newrobustcmd{\hatbmmathfrakPz}[2][]{\ensuremath{\subp{\hat{\bm{\mathfrak{P}}}}{}{#2}{}{#1}}}
\newrobustcmd{\widehatbmmathfrakPz}[2][]{\ensuremath{\subp{\widehat{\bm{\mathfrak{P}}}}{}{#2}{}{#1}}}
\newrobustcmd{\checkbmmathfrakPz}[2][]{\ensuremath{\subp{\check{\bm{\mathfrak{P}}}}{}{#2}{}{#1}}}
\newrobustcmd{\tildebmmathfrakPz}[2][]{\ensuremath{\subp{\tilde{\bm{\mathfrak{P}}}}{}{#2}{}{#1}}}
\newrobustcmd{\widetildebmmathfrakPz}[2][]{\ensuremath{\subp{\widetilde{\bm{\mathfrak{P}}}}{}{#2}{}{#1}}}
\newrobustcmd{\acutebmmathfrakPz}[2][]{\ensuremath{\subp{\acute{\bm{\mathfrak{P}}}}{}{#2}{}{#1}}}
\newrobustcmd{\gravebmmathfrakPz}[2][]{\ensuremath{\subp{\grave{\bm{\mathfrak{P}}}}{}{#2}{}{#1}}}
\newrobustcmd{\dotbmmathfrakPz}[2][]{\ensuremath{\subp{\dot{\bm{\mathfrak{P}}}}{}{#2}{}{#1}}}
\newrobustcmd{\ddotbmmathfrakPz}[2][]{\ensuremath{\subp{\ddot{\bm{\mathfrak{P}}}}{}{#2}{}{#1}}}
\newrobustcmd{\brevebmmathfrakPz}[2][]{\ensuremath{\subp{\breve{\bm{\mathfrak{P}}}}{}{#2}{}{#1}}}
\newrobustcmd{\barbmmathfrakPz}[2][]{\ensuremath{\subp{\bar{\bm{\mathfrak{P}}}}{}{#2}{}{#1}}}
\newrobustcmd{\vecbmmathfrakPz}[2][]{\ensuremath{\subp{\vec{\bm{\mathfrak{P}}}}{}{#2}{}{#1}}}
\newrobustcmd{\mathfrakQz}[2][]{\ensuremath{\subp{\mathfrak{Q}}{}{#2}{}{#1}}}
\newrobustcmd{\hatmathfrakQz}[2][]{\ensuremath{\subp{\hat{\mathfrak{Q}}}{}{#2}{}{#1}}}
\newrobustcmd{\widehatmathfrakQz}[2][]{\ensuremath{\subp{\widehat{\mathfrak{Q}}}{}{#2}{}{#1}}}
\newrobustcmd{\checkmathfrakQz}[2][]{\ensuremath{\subp{\check{\mathfrak{Q}}}{}{#2}{}{#1}}}
\newrobustcmd{\tildemathfrakQz}[2][]{\ensuremath{\subp{\tilde{\mathfrak{Q}}}{}{#2}{}{#1}}}
\newrobustcmd{\widetildemathfrakQz}[2][]{\ensuremath{\subp{\widetilde{\mathfrak{Q}}}{}{#2}{}{#1}}}
\newrobustcmd{\acutemathfrakQz}[2][]{\ensuremath{\subp{\acute{\mathfrak{Q}}}{}{#2}{}{#1}}}
\newrobustcmd{\gravemathfrakQz}[2][]{\ensuremath{\subp{\grave{\mathfrak{Q}}}{}{#2}{}{#1}}}
\newrobustcmd{\dotmathfrakQz}[2][]{\ensuremath{\subp{\dot{\mathfrak{Q}}}{}{#2}{}{#1}}}
\newrobustcmd{\ddotmathfrakQz}[2][]{\ensuremath{\subp{\ddot{\mathfrak{Q}}}{}{#2}{}{#1}}}
\newrobustcmd{\brevemathfrakQz}[2][]{\ensuremath{\subp{\breve{\mathfrak{Q}}}{}{#2}{}{#1}}}
\newrobustcmd{\barmathfrakQz}[2][]{\ensuremath{\subp{\bar{\mathfrak{Q}}}{}{#2}{}{#1}}}
\newrobustcmd{\vecmathfrakQz}[2][]{\ensuremath{\subp{\vec{\mathfrak{Q}}}{}{#2}{}{#1}}}
\newrobustcmd{\bmmathfrakQz}[2][]{\ensuremath{\subp{\bm{\mathfrak{Q}}}{}{#2}{}{#1}}}
\newrobustcmd{\hatbmmathfrakQz}[2][]{\ensuremath{\subp{\hat{\bm{\mathfrak{Q}}}}{}{#2}{}{#1}}}
\newrobustcmd{\widehatbmmathfrakQz}[2][]{\ensuremath{\subp{\widehat{\bm{\mathfrak{Q}}}}{}{#2}{}{#1}}}
\newrobustcmd{\checkbmmathfrakQz}[2][]{\ensuremath{\subp{\check{\bm{\mathfrak{Q}}}}{}{#2}{}{#1}}}
\newrobustcmd{\tildebmmathfrakQz}[2][]{\ensuremath{\subp{\tilde{\bm{\mathfrak{Q}}}}{}{#2}{}{#1}}}
\newrobustcmd{\widetildebmmathfrakQz}[2][]{\ensuremath{\subp{\widetilde{\bm{\mathfrak{Q}}}}{}{#2}{}{#1}}}
\newrobustcmd{\acutebmmathfrakQz}[2][]{\ensuremath{\subp{\acute{\bm{\mathfrak{Q}}}}{}{#2}{}{#1}}}
\newrobustcmd{\gravebmmathfrakQz}[2][]{\ensuremath{\subp{\grave{\bm{\mathfrak{Q}}}}{}{#2}{}{#1}}}
\newrobustcmd{\dotbmmathfrakQz}[2][]{\ensuremath{\subp{\dot{\bm{\mathfrak{Q}}}}{}{#2}{}{#1}}}
\newrobustcmd{\ddotbmmathfrakQz}[2][]{\ensuremath{\subp{\ddot{\bm{\mathfrak{Q}}}}{}{#2}{}{#1}}}
\newrobustcmd{\brevebmmathfrakQz}[2][]{\ensuremath{\subp{\breve{\bm{\mathfrak{Q}}}}{}{#2}{}{#1}}}
\newrobustcmd{\barbmmathfrakQz}[2][]{\ensuremath{\subp{\bar{\bm{\mathfrak{Q}}}}{}{#2}{}{#1}}}
\newrobustcmd{\vecbmmathfrakQz}[2][]{\ensuremath{\subp{\vec{\bm{\mathfrak{Q}}}}{}{#2}{}{#1}}}
\newrobustcmd{\mathfrakRz}[2][]{\ensuremath{\subp{\mathfrak{R}}{}{#2}{}{#1}}}
\newrobustcmd{\hatmathfrakRz}[2][]{\ensuremath{\subp{\hat{\mathfrak{R}}}{}{#2}{}{#1}}}
\newrobustcmd{\widehatmathfrakRz}[2][]{\ensuremath{\subp{\widehat{\mathfrak{R}}}{}{#2}{}{#1}}}
\newrobustcmd{\checkmathfrakRz}[2][]{\ensuremath{\subp{\check{\mathfrak{R}}}{}{#2}{}{#1}}}
\newrobustcmd{\tildemathfrakRz}[2][]{\ensuremath{\subp{\tilde{\mathfrak{R}}}{}{#2}{}{#1}}}
\newrobustcmd{\widetildemathfrakRz}[2][]{\ensuremath{\subp{\widetilde{\mathfrak{R}}}{}{#2}{}{#1}}}
\newrobustcmd{\acutemathfrakRz}[2][]{\ensuremath{\subp{\acute{\mathfrak{R}}}{}{#2}{}{#1}}}
\newrobustcmd{\gravemathfrakRz}[2][]{\ensuremath{\subp{\grave{\mathfrak{R}}}{}{#2}{}{#1}}}
\newrobustcmd{\dotmathfrakRz}[2][]{\ensuremath{\subp{\dot{\mathfrak{R}}}{}{#2}{}{#1}}}
\newrobustcmd{\ddotmathfrakRz}[2][]{\ensuremath{\subp{\ddot{\mathfrak{R}}}{}{#2}{}{#1}}}
\newrobustcmd{\brevemathfrakRz}[2][]{\ensuremath{\subp{\breve{\mathfrak{R}}}{}{#2}{}{#1}}}
\newrobustcmd{\barmathfrakRz}[2][]{\ensuremath{\subp{\bar{\mathfrak{R}}}{}{#2}{}{#1}}}
\newrobustcmd{\vecmathfrakRz}[2][]{\ensuremath{\subp{\vec{\mathfrak{R}}}{}{#2}{}{#1}}}
\newrobustcmd{\bmmathfrakRz}[2][]{\ensuremath{\subp{\bm{\mathfrak{R}}}{}{#2}{}{#1}}}
\newrobustcmd{\hatbmmathfrakRz}[2][]{\ensuremath{\subp{\hat{\bm{\mathfrak{R}}}}{}{#2}{}{#1}}}
\newrobustcmd{\widehatbmmathfrakRz}[2][]{\ensuremath{\subp{\widehat{\bm{\mathfrak{R}}}}{}{#2}{}{#1}}}
\newrobustcmd{\checkbmmathfrakRz}[2][]{\ensuremath{\subp{\check{\bm{\mathfrak{R}}}}{}{#2}{}{#1}}}
\newrobustcmd{\tildebmmathfrakRz}[2][]{\ensuremath{\subp{\tilde{\bm{\mathfrak{R}}}}{}{#2}{}{#1}}}
\newrobustcmd{\widetildebmmathfrakRz}[2][]{\ensuremath{\subp{\widetilde{\bm{\mathfrak{R}}}}{}{#2}{}{#1}}}
\newrobustcmd{\acutebmmathfrakRz}[2][]{\ensuremath{\subp{\acute{\bm{\mathfrak{R}}}}{}{#2}{}{#1}}}
\newrobustcmd{\gravebmmathfrakRz}[2][]{\ensuremath{\subp{\grave{\bm{\mathfrak{R}}}}{}{#2}{}{#1}}}
\newrobustcmd{\dotbmmathfrakRz}[2][]{\ensuremath{\subp{\dot{\bm{\mathfrak{R}}}}{}{#2}{}{#1}}}
\newrobustcmd{\ddotbmmathfrakRz}[2][]{\ensuremath{\subp{\ddot{\bm{\mathfrak{R}}}}{}{#2}{}{#1}}}
\newrobustcmd{\brevebmmathfrakRz}[2][]{\ensuremath{\subp{\breve{\bm{\mathfrak{R}}}}{}{#2}{}{#1}}}
\newrobustcmd{\barbmmathfrakRz}[2][]{\ensuremath{\subp{\bar{\bm{\mathfrak{R}}}}{}{#2}{}{#1}}}
\newrobustcmd{\vecbmmathfrakRz}[2][]{\ensuremath{\subp{\vec{\bm{\mathfrak{R}}}}{}{#2}{}{#1}}}
\newrobustcmd{\mathfrakSz}[2][]{\ensuremath{\subp{\mathfrak{S}}{}{#2}{}{#1}}}
\newrobustcmd{\hatmathfrakSz}[2][]{\ensuremath{\subp{\hat{\mathfrak{S}}}{}{#2}{}{#1}}}
\newrobustcmd{\widehatmathfrakSz}[2][]{\ensuremath{\subp{\widehat{\mathfrak{S}}}{}{#2}{}{#1}}}
\newrobustcmd{\checkmathfrakSz}[2][]{\ensuremath{\subp{\check{\mathfrak{S}}}{}{#2}{}{#1}}}
\newrobustcmd{\tildemathfrakSz}[2][]{\ensuremath{\subp{\tilde{\mathfrak{S}}}{}{#2}{}{#1}}}
\newrobustcmd{\widetildemathfrakSz}[2][]{\ensuremath{\subp{\widetilde{\mathfrak{S}}}{}{#2}{}{#1}}}
\newrobustcmd{\acutemathfrakSz}[2][]{\ensuremath{\subp{\acute{\mathfrak{S}}}{}{#2}{}{#1}}}
\newrobustcmd{\gravemathfrakSz}[2][]{\ensuremath{\subp{\grave{\mathfrak{S}}}{}{#2}{}{#1}}}
\newrobustcmd{\dotmathfrakSz}[2][]{\ensuremath{\subp{\dot{\mathfrak{S}}}{}{#2}{}{#1}}}
\newrobustcmd{\ddotmathfrakSz}[2][]{\ensuremath{\subp{\ddot{\mathfrak{S}}}{}{#2}{}{#1}}}
\newrobustcmd{\brevemathfrakSz}[2][]{\ensuremath{\subp{\breve{\mathfrak{S}}}{}{#2}{}{#1}}}
\newrobustcmd{\barmathfrakSz}[2][]{\ensuremath{\subp{\bar{\mathfrak{S}}}{}{#2}{}{#1}}}
\newrobustcmd{\vecmathfrakSz}[2][]{\ensuremath{\subp{\vec{\mathfrak{S}}}{}{#2}{}{#1}}}
\newrobustcmd{\bmmathfrakSz}[2][]{\ensuremath{\subp{\bm{\mathfrak{S}}}{}{#2}{}{#1}}}
\newrobustcmd{\hatbmmathfrakSz}[2][]{\ensuremath{\subp{\hat{\bm{\mathfrak{S}}}}{}{#2}{}{#1}}}
\newrobustcmd{\widehatbmmathfrakSz}[2][]{\ensuremath{\subp{\widehat{\bm{\mathfrak{S}}}}{}{#2}{}{#1}}}
\newrobustcmd{\checkbmmathfrakSz}[2][]{\ensuremath{\subp{\check{\bm{\mathfrak{S}}}}{}{#2}{}{#1}}}
\newrobustcmd{\tildebmmathfrakSz}[2][]{\ensuremath{\subp{\tilde{\bm{\mathfrak{S}}}}{}{#2}{}{#1}}}
\newrobustcmd{\widetildebmmathfrakSz}[2][]{\ensuremath{\subp{\widetilde{\bm{\mathfrak{S}}}}{}{#2}{}{#1}}}
\newrobustcmd{\acutebmmathfrakSz}[2][]{\ensuremath{\subp{\acute{\bm{\mathfrak{S}}}}{}{#2}{}{#1}}}
\newrobustcmd{\gravebmmathfrakSz}[2][]{\ensuremath{\subp{\grave{\bm{\mathfrak{S}}}}{}{#2}{}{#1}}}
\newrobustcmd{\dotbmmathfrakSz}[2][]{\ensuremath{\subp{\dot{\bm{\mathfrak{S}}}}{}{#2}{}{#1}}}
\newrobustcmd{\ddotbmmathfrakSz}[2][]{\ensuremath{\subp{\ddot{\bm{\mathfrak{S}}}}{}{#2}{}{#1}}}
\newrobustcmd{\brevebmmathfrakSz}[2][]{\ensuremath{\subp{\breve{\bm{\mathfrak{S}}}}{}{#2}{}{#1}}}
\newrobustcmd{\barbmmathfrakSz}[2][]{\ensuremath{\subp{\bar{\bm{\mathfrak{S}}}}{}{#2}{}{#1}}}
\newrobustcmd{\vecbmmathfrakSz}[2][]{\ensuremath{\subp{\vec{\bm{\mathfrak{S}}}}{}{#2}{}{#1}}}
\newrobustcmd{\mathfrakTz}[2][]{\ensuremath{\subp{\mathfrak{T}}{}{#2}{}{#1}}}
\newrobustcmd{\hatmathfrakTz}[2][]{\ensuremath{\subp{\hat{\mathfrak{T}}}{}{#2}{}{#1}}}
\newrobustcmd{\widehatmathfrakTz}[2][]{\ensuremath{\subp{\widehat{\mathfrak{T}}}{}{#2}{}{#1}}}
\newrobustcmd{\checkmathfrakTz}[2][]{\ensuremath{\subp{\check{\mathfrak{T}}}{}{#2}{}{#1}}}
\newrobustcmd{\tildemathfrakTz}[2][]{\ensuremath{\subp{\tilde{\mathfrak{T}}}{}{#2}{}{#1}}}
\newrobustcmd{\widetildemathfrakTz}[2][]{\ensuremath{\subp{\widetilde{\mathfrak{T}}}{}{#2}{}{#1}}}
\newrobustcmd{\acutemathfrakTz}[2][]{\ensuremath{\subp{\acute{\mathfrak{T}}}{}{#2}{}{#1}}}
\newrobustcmd{\gravemathfrakTz}[2][]{\ensuremath{\subp{\grave{\mathfrak{T}}}{}{#2}{}{#1}}}
\newrobustcmd{\dotmathfrakTz}[2][]{\ensuremath{\subp{\dot{\mathfrak{T}}}{}{#2}{}{#1}}}
\newrobustcmd{\ddotmathfrakTz}[2][]{\ensuremath{\subp{\ddot{\mathfrak{T}}}{}{#2}{}{#1}}}
\newrobustcmd{\brevemathfrakTz}[2][]{\ensuremath{\subp{\breve{\mathfrak{T}}}{}{#2}{}{#1}}}
\newrobustcmd{\barmathfrakTz}[2][]{\ensuremath{\subp{\bar{\mathfrak{T}}}{}{#2}{}{#1}}}
\newrobustcmd{\vecmathfrakTz}[2][]{\ensuremath{\subp{\vec{\mathfrak{T}}}{}{#2}{}{#1}}}
\newrobustcmd{\bmmathfrakTz}[2][]{\ensuremath{\subp{\bm{\mathfrak{T}}}{}{#2}{}{#1}}}
\newrobustcmd{\hatbmmathfrakTz}[2][]{\ensuremath{\subp{\hat{\bm{\mathfrak{T}}}}{}{#2}{}{#1}}}
\newrobustcmd{\widehatbmmathfrakTz}[2][]{\ensuremath{\subp{\widehat{\bm{\mathfrak{T}}}}{}{#2}{}{#1}}}
\newrobustcmd{\checkbmmathfrakTz}[2][]{\ensuremath{\subp{\check{\bm{\mathfrak{T}}}}{}{#2}{}{#1}}}
\newrobustcmd{\tildebmmathfrakTz}[2][]{\ensuremath{\subp{\tilde{\bm{\mathfrak{T}}}}{}{#2}{}{#1}}}
\newrobustcmd{\widetildebmmathfrakTz}[2][]{\ensuremath{\subp{\widetilde{\bm{\mathfrak{T}}}}{}{#2}{}{#1}}}
\newrobustcmd{\acutebmmathfrakTz}[2][]{\ensuremath{\subp{\acute{\bm{\mathfrak{T}}}}{}{#2}{}{#1}}}
\newrobustcmd{\gravebmmathfrakTz}[2][]{\ensuremath{\subp{\grave{\bm{\mathfrak{T}}}}{}{#2}{}{#1}}}
\newrobustcmd{\dotbmmathfrakTz}[2][]{\ensuremath{\subp{\dot{\bm{\mathfrak{T}}}}{}{#2}{}{#1}}}
\newrobustcmd{\ddotbmmathfrakTz}[2][]{\ensuremath{\subp{\ddot{\bm{\mathfrak{T}}}}{}{#2}{}{#1}}}
\newrobustcmd{\brevebmmathfrakTz}[2][]{\ensuremath{\subp{\breve{\bm{\mathfrak{T}}}}{}{#2}{}{#1}}}
\newrobustcmd{\barbmmathfrakTz}[2][]{\ensuremath{\subp{\bar{\bm{\mathfrak{T}}}}{}{#2}{}{#1}}}
\newrobustcmd{\vecbmmathfrakTz}[2][]{\ensuremath{\subp{\vec{\bm{\mathfrak{T}}}}{}{#2}{}{#1}}}
\newrobustcmd{\mathfrakUz}[2][]{\ensuremath{\subp{\mathfrak{U}}{}{#2}{}{#1}}}
\newrobustcmd{\hatmathfrakUz}[2][]{\ensuremath{\subp{\hat{\mathfrak{U}}}{}{#2}{}{#1}}}
\newrobustcmd{\widehatmathfrakUz}[2][]{\ensuremath{\subp{\widehat{\mathfrak{U}}}{}{#2}{}{#1}}}
\newrobustcmd{\checkmathfrakUz}[2][]{\ensuremath{\subp{\check{\mathfrak{U}}}{}{#2}{}{#1}}}
\newrobustcmd{\tildemathfrakUz}[2][]{\ensuremath{\subp{\tilde{\mathfrak{U}}}{}{#2}{}{#1}}}
\newrobustcmd{\widetildemathfrakUz}[2][]{\ensuremath{\subp{\widetilde{\mathfrak{U}}}{}{#2}{}{#1}}}
\newrobustcmd{\acutemathfrakUz}[2][]{\ensuremath{\subp{\acute{\mathfrak{U}}}{}{#2}{}{#1}}}
\newrobustcmd{\gravemathfrakUz}[2][]{\ensuremath{\subp{\grave{\mathfrak{U}}}{}{#2}{}{#1}}}
\newrobustcmd{\dotmathfrakUz}[2][]{\ensuremath{\subp{\dot{\mathfrak{U}}}{}{#2}{}{#1}}}
\newrobustcmd{\ddotmathfrakUz}[2][]{\ensuremath{\subp{\ddot{\mathfrak{U}}}{}{#2}{}{#1}}}
\newrobustcmd{\brevemathfrakUz}[2][]{\ensuremath{\subp{\breve{\mathfrak{U}}}{}{#2}{}{#1}}}
\newrobustcmd{\barmathfrakUz}[2][]{\ensuremath{\subp{\bar{\mathfrak{U}}}{}{#2}{}{#1}}}
\newrobustcmd{\vecmathfrakUz}[2][]{\ensuremath{\subp{\vec{\mathfrak{U}}}{}{#2}{}{#1}}}
\newrobustcmd{\bmmathfrakUz}[2][]{\ensuremath{\subp{\bm{\mathfrak{U}}}{}{#2}{}{#1}}}
\newrobustcmd{\hatbmmathfrakUz}[2][]{\ensuremath{\subp{\hat{\bm{\mathfrak{U}}}}{}{#2}{}{#1}}}
\newrobustcmd{\widehatbmmathfrakUz}[2][]{\ensuremath{\subp{\widehat{\bm{\mathfrak{U}}}}{}{#2}{}{#1}}}
\newrobustcmd{\checkbmmathfrakUz}[2][]{\ensuremath{\subp{\check{\bm{\mathfrak{U}}}}{}{#2}{}{#1}}}
\newrobustcmd{\tildebmmathfrakUz}[2][]{\ensuremath{\subp{\tilde{\bm{\mathfrak{U}}}}{}{#2}{}{#1}}}
\newrobustcmd{\widetildebmmathfrakUz}[2][]{\ensuremath{\subp{\widetilde{\bm{\mathfrak{U}}}}{}{#2}{}{#1}}}
\newrobustcmd{\acutebmmathfrakUz}[2][]{\ensuremath{\subp{\acute{\bm{\mathfrak{U}}}}{}{#2}{}{#1}}}
\newrobustcmd{\gravebmmathfrakUz}[2][]{\ensuremath{\subp{\grave{\bm{\mathfrak{U}}}}{}{#2}{}{#1}}}
\newrobustcmd{\dotbmmathfrakUz}[2][]{\ensuremath{\subp{\dot{\bm{\mathfrak{U}}}}{}{#2}{}{#1}}}
\newrobustcmd{\ddotbmmathfrakUz}[2][]{\ensuremath{\subp{\ddot{\bm{\mathfrak{U}}}}{}{#2}{}{#1}}}
\newrobustcmd{\brevebmmathfrakUz}[2][]{\ensuremath{\subp{\breve{\bm{\mathfrak{U}}}}{}{#2}{}{#1}}}
\newrobustcmd{\barbmmathfrakUz}[2][]{\ensuremath{\subp{\bar{\bm{\mathfrak{U}}}}{}{#2}{}{#1}}}
\newrobustcmd{\vecbmmathfrakUz}[2][]{\ensuremath{\subp{\vec{\bm{\mathfrak{U}}}}{}{#2}{}{#1}}}
\newrobustcmd{\mathfrakVz}[2][]{\ensuremath{\subp{\mathfrak{V}}{}{#2}{}{#1}}}
\newrobustcmd{\hatmathfrakVz}[2][]{\ensuremath{\subp{\hat{\mathfrak{V}}}{}{#2}{}{#1}}}
\newrobustcmd{\widehatmathfrakVz}[2][]{\ensuremath{\subp{\widehat{\mathfrak{V}}}{}{#2}{}{#1}}}
\newrobustcmd{\checkmathfrakVz}[2][]{\ensuremath{\subp{\check{\mathfrak{V}}}{}{#2}{}{#1}}}
\newrobustcmd{\tildemathfrakVz}[2][]{\ensuremath{\subp{\tilde{\mathfrak{V}}}{}{#2}{}{#1}}}
\newrobustcmd{\widetildemathfrakVz}[2][]{\ensuremath{\subp{\widetilde{\mathfrak{V}}}{}{#2}{}{#1}}}
\newrobustcmd{\acutemathfrakVz}[2][]{\ensuremath{\subp{\acute{\mathfrak{V}}}{}{#2}{}{#1}}}
\newrobustcmd{\gravemathfrakVz}[2][]{\ensuremath{\subp{\grave{\mathfrak{V}}}{}{#2}{}{#1}}}
\newrobustcmd{\dotmathfrakVz}[2][]{\ensuremath{\subp{\dot{\mathfrak{V}}}{}{#2}{}{#1}}}
\newrobustcmd{\ddotmathfrakVz}[2][]{\ensuremath{\subp{\ddot{\mathfrak{V}}}{}{#2}{}{#1}}}
\newrobustcmd{\brevemathfrakVz}[2][]{\ensuremath{\subp{\breve{\mathfrak{V}}}{}{#2}{}{#1}}}
\newrobustcmd{\barmathfrakVz}[2][]{\ensuremath{\subp{\bar{\mathfrak{V}}}{}{#2}{}{#1}}}
\newrobustcmd{\vecmathfrakVz}[2][]{\ensuremath{\subp{\vec{\mathfrak{V}}}{}{#2}{}{#1}}}
\newrobustcmd{\bmmathfrakVz}[2][]{\ensuremath{\subp{\bm{\mathfrak{V}}}{}{#2}{}{#1}}}
\newrobustcmd{\hatbmmathfrakVz}[2][]{\ensuremath{\subp{\hat{\bm{\mathfrak{V}}}}{}{#2}{}{#1}}}
\newrobustcmd{\widehatbmmathfrakVz}[2][]{\ensuremath{\subp{\widehat{\bm{\mathfrak{V}}}}{}{#2}{}{#1}}}
\newrobustcmd{\checkbmmathfrakVz}[2][]{\ensuremath{\subp{\check{\bm{\mathfrak{V}}}}{}{#2}{}{#1}}}
\newrobustcmd{\tildebmmathfrakVz}[2][]{\ensuremath{\subp{\tilde{\bm{\mathfrak{V}}}}{}{#2}{}{#1}}}
\newrobustcmd{\widetildebmmathfrakVz}[2][]{\ensuremath{\subp{\widetilde{\bm{\mathfrak{V}}}}{}{#2}{}{#1}}}
\newrobustcmd{\acutebmmathfrakVz}[2][]{\ensuremath{\subp{\acute{\bm{\mathfrak{V}}}}{}{#2}{}{#1}}}
\newrobustcmd{\gravebmmathfrakVz}[2][]{\ensuremath{\subp{\grave{\bm{\mathfrak{V}}}}{}{#2}{}{#1}}}
\newrobustcmd{\dotbmmathfrakVz}[2][]{\ensuremath{\subp{\dot{\bm{\mathfrak{V}}}}{}{#2}{}{#1}}}
\newrobustcmd{\ddotbmmathfrakVz}[2][]{\ensuremath{\subp{\ddot{\bm{\mathfrak{V}}}}{}{#2}{}{#1}}}
\newrobustcmd{\brevebmmathfrakVz}[2][]{\ensuremath{\subp{\breve{\bm{\mathfrak{V}}}}{}{#2}{}{#1}}}
\newrobustcmd{\barbmmathfrakVz}[2][]{\ensuremath{\subp{\bar{\bm{\mathfrak{V}}}}{}{#2}{}{#1}}}
\newrobustcmd{\vecbmmathfrakVz}[2][]{\ensuremath{\subp{\vec{\bm{\mathfrak{V}}}}{}{#2}{}{#1}}}
\newrobustcmd{\mathfrakWz}[2][]{\ensuremath{\subp{\mathfrak{W}}{}{#2}{}{#1}}}
\newrobustcmd{\hatmathfrakWz}[2][]{\ensuremath{\subp{\hat{\mathfrak{W}}}{}{#2}{}{#1}}}
\newrobustcmd{\widehatmathfrakWz}[2][]{\ensuremath{\subp{\widehat{\mathfrak{W}}}{}{#2}{}{#1}}}
\newrobustcmd{\checkmathfrakWz}[2][]{\ensuremath{\subp{\check{\mathfrak{W}}}{}{#2}{}{#1}}}
\newrobustcmd{\tildemathfrakWz}[2][]{\ensuremath{\subp{\tilde{\mathfrak{W}}}{}{#2}{}{#1}}}
\newrobustcmd{\widetildemathfrakWz}[2][]{\ensuremath{\subp{\widetilde{\mathfrak{W}}}{}{#2}{}{#1}}}
\newrobustcmd{\acutemathfrakWz}[2][]{\ensuremath{\subp{\acute{\mathfrak{W}}}{}{#2}{}{#1}}}
\newrobustcmd{\gravemathfrakWz}[2][]{\ensuremath{\subp{\grave{\mathfrak{W}}}{}{#2}{}{#1}}}
\newrobustcmd{\dotmathfrakWz}[2][]{\ensuremath{\subp{\dot{\mathfrak{W}}}{}{#2}{}{#1}}}
\newrobustcmd{\ddotmathfrakWz}[2][]{\ensuremath{\subp{\ddot{\mathfrak{W}}}{}{#2}{}{#1}}}
\newrobustcmd{\brevemathfrakWz}[2][]{\ensuremath{\subp{\breve{\mathfrak{W}}}{}{#2}{}{#1}}}
\newrobustcmd{\barmathfrakWz}[2][]{\ensuremath{\subp{\bar{\mathfrak{W}}}{}{#2}{}{#1}}}
\newrobustcmd{\vecmathfrakWz}[2][]{\ensuremath{\subp{\vec{\mathfrak{W}}}{}{#2}{}{#1}}}
\newrobustcmd{\bmmathfrakWz}[2][]{\ensuremath{\subp{\bm{\mathfrak{W}}}{}{#2}{}{#1}}}
\newrobustcmd{\hatbmmathfrakWz}[2][]{\ensuremath{\subp{\hat{\bm{\mathfrak{W}}}}{}{#2}{}{#1}}}
\newrobustcmd{\widehatbmmathfrakWz}[2][]{\ensuremath{\subp{\widehat{\bm{\mathfrak{W}}}}{}{#2}{}{#1}}}
\newrobustcmd{\checkbmmathfrakWz}[2][]{\ensuremath{\subp{\check{\bm{\mathfrak{W}}}}{}{#2}{}{#1}}}
\newrobustcmd{\tildebmmathfrakWz}[2][]{\ensuremath{\subp{\tilde{\bm{\mathfrak{W}}}}{}{#2}{}{#1}}}
\newrobustcmd{\widetildebmmathfrakWz}[2][]{\ensuremath{\subp{\widetilde{\bm{\mathfrak{W}}}}{}{#2}{}{#1}}}
\newrobustcmd{\acutebmmathfrakWz}[2][]{\ensuremath{\subp{\acute{\bm{\mathfrak{W}}}}{}{#2}{}{#1}}}
\newrobustcmd{\gravebmmathfrakWz}[2][]{\ensuremath{\subp{\grave{\bm{\mathfrak{W}}}}{}{#2}{}{#1}}}
\newrobustcmd{\dotbmmathfrakWz}[2][]{\ensuremath{\subp{\dot{\bm{\mathfrak{W}}}}{}{#2}{}{#1}}}
\newrobustcmd{\ddotbmmathfrakWz}[2][]{\ensuremath{\subp{\ddot{\bm{\mathfrak{W}}}}{}{#2}{}{#1}}}
\newrobustcmd{\brevebmmathfrakWz}[2][]{\ensuremath{\subp{\breve{\bm{\mathfrak{W}}}}{}{#2}{}{#1}}}
\newrobustcmd{\barbmmathfrakWz}[2][]{\ensuremath{\subp{\bar{\bm{\mathfrak{W}}}}{}{#2}{}{#1}}}
\newrobustcmd{\vecbmmathfrakWz}[2][]{\ensuremath{\subp{\vec{\bm{\mathfrak{W}}}}{}{#2}{}{#1}}}
\newrobustcmd{\mathfrakXz}[2][]{\ensuremath{\subp{\mathfrak{X}}{}{#2}{}{#1}}}
\newrobustcmd{\hatmathfrakXz}[2][]{\ensuremath{\subp{\hat{\mathfrak{X}}}{}{#2}{}{#1}}}
\newrobustcmd{\widehatmathfrakXz}[2][]{\ensuremath{\subp{\widehat{\mathfrak{X}}}{}{#2}{}{#1}}}
\newrobustcmd{\checkmathfrakXz}[2][]{\ensuremath{\subp{\check{\mathfrak{X}}}{}{#2}{}{#1}}}
\newrobustcmd{\tildemathfrakXz}[2][]{\ensuremath{\subp{\tilde{\mathfrak{X}}}{}{#2}{}{#1}}}
\newrobustcmd{\widetildemathfrakXz}[2][]{\ensuremath{\subp{\widetilde{\mathfrak{X}}}{}{#2}{}{#1}}}
\newrobustcmd{\acutemathfrakXz}[2][]{\ensuremath{\subp{\acute{\mathfrak{X}}}{}{#2}{}{#1}}}
\newrobustcmd{\gravemathfrakXz}[2][]{\ensuremath{\subp{\grave{\mathfrak{X}}}{}{#2}{}{#1}}}
\newrobustcmd{\dotmathfrakXz}[2][]{\ensuremath{\subp{\dot{\mathfrak{X}}}{}{#2}{}{#1}}}
\newrobustcmd{\ddotmathfrakXz}[2][]{\ensuremath{\subp{\ddot{\mathfrak{X}}}{}{#2}{}{#1}}}
\newrobustcmd{\brevemathfrakXz}[2][]{\ensuremath{\subp{\breve{\mathfrak{X}}}{}{#2}{}{#1}}}
\newrobustcmd{\barmathfrakXz}[2][]{\ensuremath{\subp{\bar{\mathfrak{X}}}{}{#2}{}{#1}}}
\newrobustcmd{\vecmathfrakXz}[2][]{\ensuremath{\subp{\vec{\mathfrak{X}}}{}{#2}{}{#1}}}
\newrobustcmd{\bmmathfrakXz}[2][]{\ensuremath{\subp{\bm{\mathfrak{X}}}{}{#2}{}{#1}}}
\newrobustcmd{\hatbmmathfrakXz}[2][]{\ensuremath{\subp{\hat{\bm{\mathfrak{X}}}}{}{#2}{}{#1}}}
\newrobustcmd{\widehatbmmathfrakXz}[2][]{\ensuremath{\subp{\widehat{\bm{\mathfrak{X}}}}{}{#2}{}{#1}}}
\newrobustcmd{\checkbmmathfrakXz}[2][]{\ensuremath{\subp{\check{\bm{\mathfrak{X}}}}{}{#2}{}{#1}}}
\newrobustcmd{\tildebmmathfrakXz}[2][]{\ensuremath{\subp{\tilde{\bm{\mathfrak{X}}}}{}{#2}{}{#1}}}
\newrobustcmd{\widetildebmmathfrakXz}[2][]{\ensuremath{\subp{\widetilde{\bm{\mathfrak{X}}}}{}{#2}{}{#1}}}
\newrobustcmd{\acutebmmathfrakXz}[2][]{\ensuremath{\subp{\acute{\bm{\mathfrak{X}}}}{}{#2}{}{#1}}}
\newrobustcmd{\gravebmmathfrakXz}[2][]{\ensuremath{\subp{\grave{\bm{\mathfrak{X}}}}{}{#2}{}{#1}}}
\newrobustcmd{\dotbmmathfrakXz}[2][]{\ensuremath{\subp{\dot{\bm{\mathfrak{X}}}}{}{#2}{}{#1}}}
\newrobustcmd{\ddotbmmathfrakXz}[2][]{\ensuremath{\subp{\ddot{\bm{\mathfrak{X}}}}{}{#2}{}{#1}}}
\newrobustcmd{\brevebmmathfrakXz}[2][]{\ensuremath{\subp{\breve{\bm{\mathfrak{X}}}}{}{#2}{}{#1}}}
\newrobustcmd{\barbmmathfrakXz}[2][]{\ensuremath{\subp{\bar{\bm{\mathfrak{X}}}}{}{#2}{}{#1}}}
\newrobustcmd{\vecbmmathfrakXz}[2][]{\ensuremath{\subp{\vec{\bm{\mathfrak{X}}}}{}{#2}{}{#1}}}
\newrobustcmd{\mathfrakYz}[2][]{\ensuremath{\subp{\mathfrak{Y}}{}{#2}{}{#1}}}
\newrobustcmd{\hatmathfrakYz}[2][]{\ensuremath{\subp{\hat{\mathfrak{Y}}}{}{#2}{}{#1}}}
\newrobustcmd{\widehatmathfrakYz}[2][]{\ensuremath{\subp{\widehat{\mathfrak{Y}}}{}{#2}{}{#1}}}
\newrobustcmd{\checkmathfrakYz}[2][]{\ensuremath{\subp{\check{\mathfrak{Y}}}{}{#2}{}{#1}}}
\newrobustcmd{\tildemathfrakYz}[2][]{\ensuremath{\subp{\tilde{\mathfrak{Y}}}{}{#2}{}{#1}}}
\newrobustcmd{\widetildemathfrakYz}[2][]{\ensuremath{\subp{\widetilde{\mathfrak{Y}}}{}{#2}{}{#1}}}
\newrobustcmd{\acutemathfrakYz}[2][]{\ensuremath{\subp{\acute{\mathfrak{Y}}}{}{#2}{}{#1}}}
\newrobustcmd{\gravemathfrakYz}[2][]{\ensuremath{\subp{\grave{\mathfrak{Y}}}{}{#2}{}{#1}}}
\newrobustcmd{\dotmathfrakYz}[2][]{\ensuremath{\subp{\dot{\mathfrak{Y}}}{}{#2}{}{#1}}}
\newrobustcmd{\ddotmathfrakYz}[2][]{\ensuremath{\subp{\ddot{\mathfrak{Y}}}{}{#2}{}{#1}}}
\newrobustcmd{\brevemathfrakYz}[2][]{\ensuremath{\subp{\breve{\mathfrak{Y}}}{}{#2}{}{#1}}}
\newrobustcmd{\barmathfrakYz}[2][]{\ensuremath{\subp{\bar{\mathfrak{Y}}}{}{#2}{}{#1}}}
\newrobustcmd{\vecmathfrakYz}[2][]{\ensuremath{\subp{\vec{\mathfrak{Y}}}{}{#2}{}{#1}}}
\newrobustcmd{\bmmathfrakYz}[2][]{\ensuremath{\subp{\bm{\mathfrak{Y}}}{}{#2}{}{#1}}}
\newrobustcmd{\hatbmmathfrakYz}[2][]{\ensuremath{\subp{\hat{\bm{\mathfrak{Y}}}}{}{#2}{}{#1}}}
\newrobustcmd{\widehatbmmathfrakYz}[2][]{\ensuremath{\subp{\widehat{\bm{\mathfrak{Y}}}}{}{#2}{}{#1}}}
\newrobustcmd{\checkbmmathfrakYz}[2][]{\ensuremath{\subp{\check{\bm{\mathfrak{Y}}}}{}{#2}{}{#1}}}
\newrobustcmd{\tildebmmathfrakYz}[2][]{\ensuremath{\subp{\tilde{\bm{\mathfrak{Y}}}}{}{#2}{}{#1}}}
\newrobustcmd{\widetildebmmathfrakYz}[2][]{\ensuremath{\subp{\widetilde{\bm{\mathfrak{Y}}}}{}{#2}{}{#1}}}
\newrobustcmd{\acutebmmathfrakYz}[2][]{\ensuremath{\subp{\acute{\bm{\mathfrak{Y}}}}{}{#2}{}{#1}}}
\newrobustcmd{\gravebmmathfrakYz}[2][]{\ensuremath{\subp{\grave{\bm{\mathfrak{Y}}}}{}{#2}{}{#1}}}
\newrobustcmd{\dotbmmathfrakYz}[2][]{\ensuremath{\subp{\dot{\bm{\mathfrak{Y}}}}{}{#2}{}{#1}}}
\newrobustcmd{\ddotbmmathfrakYz}[2][]{\ensuremath{\subp{\ddot{\bm{\mathfrak{Y}}}}{}{#2}{}{#1}}}
\newrobustcmd{\brevebmmathfrakYz}[2][]{\ensuremath{\subp{\breve{\bm{\mathfrak{Y}}}}{}{#2}{}{#1}}}
\newrobustcmd{\barbmmathfrakYz}[2][]{\ensuremath{\subp{\bar{\bm{\mathfrak{Y}}}}{}{#2}{}{#1}}}
\newrobustcmd{\vecbmmathfrakYz}[2][]{\ensuremath{\subp{\vec{\bm{\mathfrak{Y}}}}{}{#2}{}{#1}}}
\newrobustcmd{\mathfrakZz}[2][]{\ensuremath{\subp{\mathfrak{Z}}{}{#2}{}{#1}}}
\newrobustcmd{\hatmathfrakZz}[2][]{\ensuremath{\subp{\hat{\mathfrak{Z}}}{}{#2}{}{#1}}}
\newrobustcmd{\widehatmathfrakZz}[2][]{\ensuremath{\subp{\widehat{\mathfrak{Z}}}{}{#2}{}{#1}}}
\newrobustcmd{\checkmathfrakZz}[2][]{\ensuremath{\subp{\check{\mathfrak{Z}}}{}{#2}{}{#1}}}
\newrobustcmd{\tildemathfrakZz}[2][]{\ensuremath{\subp{\tilde{\mathfrak{Z}}}{}{#2}{}{#1}}}
\newrobustcmd{\widetildemathfrakZz}[2][]{\ensuremath{\subp{\widetilde{\mathfrak{Z}}}{}{#2}{}{#1}}}
\newrobustcmd{\acutemathfrakZz}[2][]{\ensuremath{\subp{\acute{\mathfrak{Z}}}{}{#2}{}{#1}}}
\newrobustcmd{\gravemathfrakZz}[2][]{\ensuremath{\subp{\grave{\mathfrak{Z}}}{}{#2}{}{#1}}}
\newrobustcmd{\dotmathfrakZz}[2][]{\ensuremath{\subp{\dot{\mathfrak{Z}}}{}{#2}{}{#1}}}
\newrobustcmd{\ddotmathfrakZz}[2][]{\ensuremath{\subp{\ddot{\mathfrak{Z}}}{}{#2}{}{#1}}}
\newrobustcmd{\brevemathfrakZz}[2][]{\ensuremath{\subp{\breve{\mathfrak{Z}}}{}{#2}{}{#1}}}
\newrobustcmd{\barmathfrakZz}[2][]{\ensuremath{\subp{\bar{\mathfrak{Z}}}{}{#2}{}{#1}}}
\newrobustcmd{\vecmathfrakZz}[2][]{\ensuremath{\subp{\vec{\mathfrak{Z}}}{}{#2}{}{#1}}}
\newrobustcmd{\bmmathfrakZz}[2][]{\ensuremath{\subp{\bm{\mathfrak{Z}}}{}{#2}{}{#1}}}
\newrobustcmd{\hatbmmathfrakZz}[2][]{\ensuremath{\subp{\hat{\bm{\mathfrak{Z}}}}{}{#2}{}{#1}}}
\newrobustcmd{\widehatbmmathfrakZz}[2][]{\ensuremath{\subp{\widehat{\bm{\mathfrak{Z}}}}{}{#2}{}{#1}}}
\newrobustcmd{\checkbmmathfrakZz}[2][]{\ensuremath{\subp{\check{\bm{\mathfrak{Z}}}}{}{#2}{}{#1}}}
\newrobustcmd{\tildebmmathfrakZz}[2][]{\ensuremath{\subp{\tilde{\bm{\mathfrak{Z}}}}{}{#2}{}{#1}}}
\newrobustcmd{\widetildebmmathfrakZz}[2][]{\ensuremath{\subp{\widetilde{\bm{\mathfrak{Z}}}}{}{#2}{}{#1}}}
\newrobustcmd{\acutebmmathfrakZz}[2][]{\ensuremath{\subp{\acute{\bm{\mathfrak{Z}}}}{}{#2}{}{#1}}}
\newrobustcmd{\gravebmmathfrakZz}[2][]{\ensuremath{\subp{\grave{\bm{\mathfrak{Z}}}}{}{#2}{}{#1}}}
\newrobustcmd{\dotbmmathfrakZz}[2][]{\ensuremath{\subp{\dot{\bm{\mathfrak{Z}}}}{}{#2}{}{#1}}}
\newrobustcmd{\ddotbmmathfrakZz}[2][]{\ensuremath{\subp{\ddot{\bm{\mathfrak{Z}}}}{}{#2}{}{#1}}}
\newrobustcmd{\brevebmmathfrakZz}[2][]{\ensuremath{\subp{\breve{\bm{\mathfrak{Z}}}}{}{#2}{}{#1}}}
\newrobustcmd{\barbmmathfrakZz}[2][]{\ensuremath{\subp{\bar{\bm{\mathfrak{Z}}}}{}{#2}{}{#1}}}
\newrobustcmd{\vecbmmathfrakZz}[2][]{\ensuremath{\subp{\vec{\bm{\mathfrak{Z}}}}{}{#2}{}{#1}}}

\newrobustcmd{\mathcalAz}[2][]{\ensuremath{\subp{\mathcal{A}}{}{#2}{}{#1}}}
\newrobustcmd{\hatmathcalAz}[2][]{\ensuremath{\subp{\hat{\mathcal{A}}}{}{#2}{}{#1}}}
\newrobustcmd{\widehatmathcalAz}[2][]{\ensuremath{\subp{\widehat{\mathcal{A}}}{}{#2}{}{#1}}}
\newrobustcmd{\checkmathcalAz}[2][]{\ensuremath{\subp{\check{\mathcal{A}}}{}{#2}{}{#1}}}
\newrobustcmd{\tildemathcalAz}[2][]{\ensuremath{\subp{\tilde{\mathcal{A}}}{}{#2}{}{#1}}}
\newrobustcmd{\widetildemathcalAz}[2][]{\ensuremath{\subp{\widetilde{\mathcal{A}}}{}{#2}{}{#1}}}
\newrobustcmd{\acutemathcalAz}[2][]{\ensuremath{\subp{\acute{\mathcal{A}}}{}{#2}{}{#1}}}
\newrobustcmd{\gravemathcalAz}[2][]{\ensuremath{\subp{\grave{\mathcal{A}}}{}{#2}{}{#1}}}
\newrobustcmd{\dotmathcalAz}[2][]{\ensuremath{\subp{\dot{\mathcal{A}}}{}{#2}{}{#1}}}
\newrobustcmd{\ddotmathcalAz}[2][]{\ensuremath{\subp{\ddot{\mathcal{A}}}{}{#2}{}{#1}}}
\newrobustcmd{\brevemathcalAz}[2][]{\ensuremath{\subp{\breve{\mathcal{A}}}{}{#2}{}{#1}}}
\newrobustcmd{\barmathcalAz}[2][]{\ensuremath{\subp{\bar{\mathcal{A}}}{}{#2}{}{#1}}}
\newrobustcmd{\vecmathcalAz}[2][]{\ensuremath{\subp{\vec{\mathcal{A}}}{}{#2}{}{#1}}}
\newrobustcmd{\bmmathcalAz}[2][]{\ensuremath{\subp{\bm{\mathcal{A}}}{}{#2}{}{#1}}}
\newrobustcmd{\hatbmmathcalAz}[2][]{\ensuremath{\subp{\hat{\bm{\mathcal{A}}}}{}{#2}{}{#1}}}
\newrobustcmd{\widehatbmmathcalAz}[2][]{\ensuremath{\subp{\widehat{\bm{\mathcal{A}}}}{}{#2}{}{#1}}}
\newrobustcmd{\checkbmmathcalAz}[2][]{\ensuremath{\subp{\check{\bm{\mathcal{A}}}}{}{#2}{}{#1}}}
\newrobustcmd{\tildebmmathcalAz}[2][]{\ensuremath{\subp{\tilde{\bm{\mathcal{A}}}}{}{#2}{}{#1}}}
\newrobustcmd{\widetildebmmathcalAz}[2][]{\ensuremath{\subp{\widetilde{\bm{\mathcal{A}}}}{}{#2}{}{#1}}}
\newrobustcmd{\acutebmmathcalAz}[2][]{\ensuremath{\subp{\acute{\bm{\mathcal{A}}}}{}{#2}{}{#1}}}
\newrobustcmd{\gravebmmathcalAz}[2][]{\ensuremath{\subp{\grave{\bm{\mathcal{A}}}}{}{#2}{}{#1}}}
\newrobustcmd{\dotbmmathcalAz}[2][]{\ensuremath{\subp{\dot{\bm{\mathcal{A}}}}{}{#2}{}{#1}}}
\newrobustcmd{\ddotbmmathcalAz}[2][]{\ensuremath{\subp{\ddot{\bm{\mathcal{A}}}}{}{#2}{}{#1}}}
\newrobustcmd{\brevebmmathcalAz}[2][]{\ensuremath{\subp{\breve{\bm{\mathcal{A}}}}{}{#2}{}{#1}}}
\newrobustcmd{\barbmmathcalAz}[2][]{\ensuremath{\subp{\bar{\bm{\mathcal{A}}}}{}{#2}{}{#1}}}
\newrobustcmd{\vecbmmathcalAz}[2][]{\ensuremath{\subp{\vec{\bm{\mathcal{A}}}}{}{#2}{}{#1}}}
\newrobustcmd{\mathcalBz}[2][]{\ensuremath{\subp{\mathcal{B}}{}{#2}{}{#1}}}
\newrobustcmd{\hatmathcalBz}[2][]{\ensuremath{\subp{\hat{\mathcal{B}}}{}{#2}{}{#1}}}
\newrobustcmd{\widehatmathcalBz}[2][]{\ensuremath{\subp{\widehat{\mathcal{B}}}{}{#2}{}{#1}}}
\newrobustcmd{\checkmathcalBz}[2][]{\ensuremath{\subp{\check{\mathcal{B}}}{}{#2}{}{#1}}}
\newrobustcmd{\tildemathcalBz}[2][]{\ensuremath{\subp{\tilde{\mathcal{B}}}{}{#2}{}{#1}}}
\newrobustcmd{\widetildemathcalBz}[2][]{\ensuremath{\subp{\widetilde{\mathcal{B}}}{}{#2}{}{#1}}}
\newrobustcmd{\acutemathcalBz}[2][]{\ensuremath{\subp{\acute{\mathcal{B}}}{}{#2}{}{#1}}}
\newrobustcmd{\gravemathcalBz}[2][]{\ensuremath{\subp{\grave{\mathcal{B}}}{}{#2}{}{#1}}}
\newrobustcmd{\dotmathcalBz}[2][]{\ensuremath{\subp{\dot{\mathcal{B}}}{}{#2}{}{#1}}}
\newrobustcmd{\ddotmathcalBz}[2][]{\ensuremath{\subp{\ddot{\mathcal{B}}}{}{#2}{}{#1}}}
\newrobustcmd{\brevemathcalBz}[2][]{\ensuremath{\subp{\breve{\mathcal{B}}}{}{#2}{}{#1}}}
\newrobustcmd{\barmathcalBz}[2][]{\ensuremath{\subp{\bar{\mathcal{B}}}{}{#2}{}{#1}}}
\newrobustcmd{\vecmathcalBz}[2][]{\ensuremath{\subp{\vec{\mathcal{B}}}{}{#2}{}{#1}}}
\newrobustcmd{\bmmathcalBz}[2][]{\ensuremath{\subp{\bm{\mathcal{B}}}{}{#2}{}{#1}}}
\newrobustcmd{\hatbmmathcalBz}[2][]{\ensuremath{\subp{\hat{\bm{\mathcal{B}}}}{}{#2}{}{#1}}}
\newrobustcmd{\widehatbmmathcalBz}[2][]{\ensuremath{\subp{\widehat{\bm{\mathcal{B}}}}{}{#2}{}{#1}}}
\newrobustcmd{\checkbmmathcalBz}[2][]{\ensuremath{\subp{\check{\bm{\mathcal{B}}}}{}{#2}{}{#1}}}
\newrobustcmd{\tildebmmathcalBz}[2][]{\ensuremath{\subp{\tilde{\bm{\mathcal{B}}}}{}{#2}{}{#1}}}
\newrobustcmd{\widetildebmmathcalBz}[2][]{\ensuremath{\subp{\widetilde{\bm{\mathcal{B}}}}{}{#2}{}{#1}}}
\newrobustcmd{\acutebmmathcalBz}[2][]{\ensuremath{\subp{\acute{\bm{\mathcal{B}}}}{}{#2}{}{#1}}}
\newrobustcmd{\gravebmmathcalBz}[2][]{\ensuremath{\subp{\grave{\bm{\mathcal{B}}}}{}{#2}{}{#1}}}
\newrobustcmd{\dotbmmathcalBz}[2][]{\ensuremath{\subp{\dot{\bm{\mathcal{B}}}}{}{#2}{}{#1}}}
\newrobustcmd{\ddotbmmathcalBz}[2][]{\ensuremath{\subp{\ddot{\bm{\mathcal{B}}}}{}{#2}{}{#1}}}
\newrobustcmd{\brevebmmathcalBz}[2][]{\ensuremath{\subp{\breve{\bm{\mathcal{B}}}}{}{#2}{}{#1}}}
\newrobustcmd{\barbmmathcalBz}[2][]{\ensuremath{\subp{\bar{\bm{\mathcal{B}}}}{}{#2}{}{#1}}}
\newrobustcmd{\vecbmmathcalBz}[2][]{\ensuremath{\subp{\vec{\bm{\mathcal{B}}}}{}{#2}{}{#1}}}
\newrobustcmd{\mathcalCz}[2][]{\ensuremath{\subp{\mathcal{C}}{}{#2}{}{#1}}}
\newrobustcmd{\hatmathcalCz}[2][]{\ensuremath{\subp{\hat{\mathcal{C}}}{}{#2}{}{#1}}}
\newrobustcmd{\widehatmathcalCz}[2][]{\ensuremath{\subp{\widehat{\mathcal{C}}}{}{#2}{}{#1}}}
\newrobustcmd{\checkmathcalCz}[2][]{\ensuremath{\subp{\check{\mathcal{C}}}{}{#2}{}{#1}}}
\newrobustcmd{\tildemathcalCz}[2][]{\ensuremath{\subp{\tilde{\mathcal{C}}}{}{#2}{}{#1}}}
\newrobustcmd{\widetildemathcalCz}[2][]{\ensuremath{\subp{\widetilde{\mathcal{C}}}{}{#2}{}{#1}}}
\newrobustcmd{\acutemathcalCz}[2][]{\ensuremath{\subp{\acute{\mathcal{C}}}{}{#2}{}{#1}}}
\newrobustcmd{\gravemathcalCz}[2][]{\ensuremath{\subp{\grave{\mathcal{C}}}{}{#2}{}{#1}}}
\newrobustcmd{\dotmathcalCz}[2][]{\ensuremath{\subp{\dot{\mathcal{C}}}{}{#2}{}{#1}}}
\newrobustcmd{\ddotmathcalCz}[2][]{\ensuremath{\subp{\ddot{\mathcal{C}}}{}{#2}{}{#1}}}
\newrobustcmd{\brevemathcalCz}[2][]{\ensuremath{\subp{\breve{\mathcal{C}}}{}{#2}{}{#1}}}
\newrobustcmd{\barmathcalCz}[2][]{\ensuremath{\subp{\bar{\mathcal{C}}}{}{#2}{}{#1}}}
\newrobustcmd{\vecmathcalCz}[2][]{\ensuremath{\subp{\vec{\mathcal{C}}}{}{#2}{}{#1}}}
\newrobustcmd{\bmmathcalCz}[2][]{\ensuremath{\subp{\bm{\mathcal{C}}}{}{#2}{}{#1}}}
\newrobustcmd{\hatbmmathcalCz}[2][]{\ensuremath{\subp{\hat{\bm{\mathcal{C}}}}{}{#2}{}{#1}}}
\newrobustcmd{\widehatbmmathcalCz}[2][]{\ensuremath{\subp{\widehat{\bm{\mathcal{C}}}}{}{#2}{}{#1}}}
\newrobustcmd{\checkbmmathcalCz}[2][]{\ensuremath{\subp{\check{\bm{\mathcal{C}}}}{}{#2}{}{#1}}}
\newrobustcmd{\tildebmmathcalCz}[2][]{\ensuremath{\subp{\tilde{\bm{\mathcal{C}}}}{}{#2}{}{#1}}}
\newrobustcmd{\widetildebmmathcalCz}[2][]{\ensuremath{\subp{\widetilde{\bm{\mathcal{C}}}}{}{#2}{}{#1}}}
\newrobustcmd{\acutebmmathcalCz}[2][]{\ensuremath{\subp{\acute{\bm{\mathcal{C}}}}{}{#2}{}{#1}}}
\newrobustcmd{\gravebmmathcalCz}[2][]{\ensuremath{\subp{\grave{\bm{\mathcal{C}}}}{}{#2}{}{#1}}}
\newrobustcmd{\dotbmmathcalCz}[2][]{\ensuremath{\subp{\dot{\bm{\mathcal{C}}}}{}{#2}{}{#1}}}
\newrobustcmd{\ddotbmmathcalCz}[2][]{\ensuremath{\subp{\ddot{\bm{\mathcal{C}}}}{}{#2}{}{#1}}}
\newrobustcmd{\brevebmmathcalCz}[2][]{\ensuremath{\subp{\breve{\bm{\mathcal{C}}}}{}{#2}{}{#1}}}
\newrobustcmd{\barbmmathcalCz}[2][]{\ensuremath{\subp{\bar{\bm{\mathcal{C}}}}{}{#2}{}{#1}}}
\newrobustcmd{\vecbmmathcalCz}[2][]{\ensuremath{\subp{\vec{\bm{\mathcal{C}}}}{}{#2}{}{#1}}}
\newrobustcmd{\mathcalDz}[2][]{\ensuremath{\subp{\mathcal{D}}{}{#2}{}{#1}}}
\newrobustcmd{\hatmathcalDz}[2][]{\ensuremath{\subp{\hat{\mathcal{D}}}{}{#2}{}{#1}}}
\newrobustcmd{\widehatmathcalDz}[2][]{\ensuremath{\subp{\widehat{\mathcal{D}}}{}{#2}{}{#1}}}
\newrobustcmd{\checkmathcalDz}[2][]{\ensuremath{\subp{\check{\mathcal{D}}}{}{#2}{}{#1}}}
\newrobustcmd{\tildemathcalDz}[2][]{\ensuremath{\subp{\tilde{\mathcal{D}}}{}{#2}{}{#1}}}
\newrobustcmd{\widetildemathcalDz}[2][]{\ensuremath{\subp{\widetilde{\mathcal{D}}}{}{#2}{}{#1}}}
\newrobustcmd{\acutemathcalDz}[2][]{\ensuremath{\subp{\acute{\mathcal{D}}}{}{#2}{}{#1}}}
\newrobustcmd{\gravemathcalDz}[2][]{\ensuremath{\subp{\grave{\mathcal{D}}}{}{#2}{}{#1}}}
\newrobustcmd{\dotmathcalDz}[2][]{\ensuremath{\subp{\dot{\mathcal{D}}}{}{#2}{}{#1}}}
\newrobustcmd{\ddotmathcalDz}[2][]{\ensuremath{\subp{\ddot{\mathcal{D}}}{}{#2}{}{#1}}}
\newrobustcmd{\brevemathcalDz}[2][]{\ensuremath{\subp{\breve{\mathcal{D}}}{}{#2}{}{#1}}}
\newrobustcmd{\barmathcalDz}[2][]{\ensuremath{\subp{\bar{\mathcal{D}}}{}{#2}{}{#1}}}
\newrobustcmd{\vecmathcalDz}[2][]{\ensuremath{\subp{\vec{\mathcal{D}}}{}{#2}{}{#1}}}
\newrobustcmd{\bmmathcalDz}[2][]{\ensuremath{\subp{\bm{\mathcal{D}}}{}{#2}{}{#1}}}
\newrobustcmd{\hatbmmathcalDz}[2][]{\ensuremath{\subp{\hat{\bm{\mathcal{D}}}}{}{#2}{}{#1}}}
\newrobustcmd{\widehatbmmathcalDz}[2][]{\ensuremath{\subp{\widehat{\bm{\mathcal{D}}}}{}{#2}{}{#1}}}
\newrobustcmd{\checkbmmathcalDz}[2][]{\ensuremath{\subp{\check{\bm{\mathcal{D}}}}{}{#2}{}{#1}}}
\newrobustcmd{\tildebmmathcalDz}[2][]{\ensuremath{\subp{\tilde{\bm{\mathcal{D}}}}{}{#2}{}{#1}}}
\newrobustcmd{\widetildebmmathcalDz}[2][]{\ensuremath{\subp{\widetilde{\bm{\mathcal{D}}}}{}{#2}{}{#1}}}
\newrobustcmd{\acutebmmathcalDz}[2][]{\ensuremath{\subp{\acute{\bm{\mathcal{D}}}}{}{#2}{}{#1}}}
\newrobustcmd{\gravebmmathcalDz}[2][]{\ensuremath{\subp{\grave{\bm{\mathcal{D}}}}{}{#2}{}{#1}}}
\newrobustcmd{\dotbmmathcalDz}[2][]{\ensuremath{\subp{\dot{\bm{\mathcal{D}}}}{}{#2}{}{#1}}}
\newrobustcmd{\ddotbmmathcalDz}[2][]{\ensuremath{\subp{\ddot{\bm{\mathcal{D}}}}{}{#2}{}{#1}}}
\newrobustcmd{\brevebmmathcalDz}[2][]{\ensuremath{\subp{\breve{\bm{\mathcal{D}}}}{}{#2}{}{#1}}}
\newrobustcmd{\barbmmathcalDz}[2][]{\ensuremath{\subp{\bar{\bm{\mathcal{D}}}}{}{#2}{}{#1}}}
\newrobustcmd{\vecbmmathcalDz}[2][]{\ensuremath{\subp{\vec{\bm{\mathcal{D}}}}{}{#2}{}{#1}}}
\newrobustcmd{\mathcalEz}[2][]{\ensuremath{\subp{\mathcal{E}}{}{#2}{}{#1}}}
\newrobustcmd{\hatmathcalEz}[2][]{\ensuremath{\subp{\hat{\mathcal{E}}}{}{#2}{}{#1}}}
\newrobustcmd{\widehatmathcalEz}[2][]{\ensuremath{\subp{\widehat{\mathcal{E}}}{}{#2}{}{#1}}}
\newrobustcmd{\checkmathcalEz}[2][]{\ensuremath{\subp{\check{\mathcal{E}}}{}{#2}{}{#1}}}
\newrobustcmd{\tildemathcalEz}[2][]{\ensuremath{\subp{\tilde{\mathcal{E}}}{}{#2}{}{#1}}}
\newrobustcmd{\widetildemathcalEz}[2][]{\ensuremath{\subp{\widetilde{\mathcal{E}}}{}{#2}{}{#1}}}
\newrobustcmd{\acutemathcalEz}[2][]{\ensuremath{\subp{\acute{\mathcal{E}}}{}{#2}{}{#1}}}
\newrobustcmd{\gravemathcalEz}[2][]{\ensuremath{\subp{\grave{\mathcal{E}}}{}{#2}{}{#1}}}
\newrobustcmd{\dotmathcalEz}[2][]{\ensuremath{\subp{\dot{\mathcal{E}}}{}{#2}{}{#1}}}
\newrobustcmd{\ddotmathcalEz}[2][]{\ensuremath{\subp{\ddot{\mathcal{E}}}{}{#2}{}{#1}}}
\newrobustcmd{\brevemathcalEz}[2][]{\ensuremath{\subp{\breve{\mathcal{E}}}{}{#2}{}{#1}}}
\newrobustcmd{\barmathcalEz}[2][]{\ensuremath{\subp{\bar{\mathcal{E}}}{}{#2}{}{#1}}}
\newrobustcmd{\vecmathcalEz}[2][]{\ensuremath{\subp{\vec{\mathcal{E}}}{}{#2}{}{#1}}}
\newrobustcmd{\bmmathcalEz}[2][]{\ensuremath{\subp{\bm{\mathcal{E}}}{}{#2}{}{#1}}}
\newrobustcmd{\hatbmmathcalEz}[2][]{\ensuremath{\subp{\hat{\bm{\mathcal{E}}}}{}{#2}{}{#1}}}
\newrobustcmd{\widehatbmmathcalEz}[2][]{\ensuremath{\subp{\widehat{\bm{\mathcal{E}}}}{}{#2}{}{#1}}}
\newrobustcmd{\checkbmmathcalEz}[2][]{\ensuremath{\subp{\check{\bm{\mathcal{E}}}}{}{#2}{}{#1}}}
\newrobustcmd{\tildebmmathcalEz}[2][]{\ensuremath{\subp{\tilde{\bm{\mathcal{E}}}}{}{#2}{}{#1}}}
\newrobustcmd{\widetildebmmathcalEz}[2][]{\ensuremath{\subp{\widetilde{\bm{\mathcal{E}}}}{}{#2}{}{#1}}}
\newrobustcmd{\acutebmmathcalEz}[2][]{\ensuremath{\subp{\acute{\bm{\mathcal{E}}}}{}{#2}{}{#1}}}
\newrobustcmd{\gravebmmathcalEz}[2][]{\ensuremath{\subp{\grave{\bm{\mathcal{E}}}}{}{#2}{}{#1}}}
\newrobustcmd{\dotbmmathcalEz}[2][]{\ensuremath{\subp{\dot{\bm{\mathcal{E}}}}{}{#2}{}{#1}}}
\newrobustcmd{\ddotbmmathcalEz}[2][]{\ensuremath{\subp{\ddot{\bm{\mathcal{E}}}}{}{#2}{}{#1}}}
\newrobustcmd{\brevebmmathcalEz}[2][]{\ensuremath{\subp{\breve{\bm{\mathcal{E}}}}{}{#2}{}{#1}}}
\newrobustcmd{\barbmmathcalEz}[2][]{\ensuremath{\subp{\bar{\bm{\mathcal{E}}}}{}{#2}{}{#1}}}
\newrobustcmd{\vecbmmathcalEz}[2][]{\ensuremath{\subp{\vec{\bm{\mathcal{E}}}}{}{#2}{}{#1}}}
\newrobustcmd{\mathcalFz}[2][]{\ensuremath{\subp{\mathcal{F}}{}{#2}{}{#1}}}
\newrobustcmd{\hatmathcalFz}[2][]{\ensuremath{\subp{\hat{\mathcal{F}}}{}{#2}{}{#1}}}
\newrobustcmd{\widehatmathcalFz}[2][]{\ensuremath{\subp{\widehat{\mathcal{F}}}{}{#2}{}{#1}}}
\newrobustcmd{\checkmathcalFz}[2][]{\ensuremath{\subp{\check{\mathcal{F}}}{}{#2}{}{#1}}}
\newrobustcmd{\tildemathcalFz}[2][]{\ensuremath{\subp{\tilde{\mathcal{F}}}{}{#2}{}{#1}}}
\newrobustcmd{\widetildemathcalFz}[2][]{\ensuremath{\subp{\widetilde{\mathcal{F}}}{}{#2}{}{#1}}}
\newrobustcmd{\acutemathcalFz}[2][]{\ensuremath{\subp{\acute{\mathcal{F}}}{}{#2}{}{#1}}}
\newrobustcmd{\gravemathcalFz}[2][]{\ensuremath{\subp{\grave{\mathcal{F}}}{}{#2}{}{#1}}}
\newrobustcmd{\dotmathcalFz}[2][]{\ensuremath{\subp{\dot{\mathcal{F}}}{}{#2}{}{#1}}}
\newrobustcmd{\ddotmathcalFz}[2][]{\ensuremath{\subp{\ddot{\mathcal{F}}}{}{#2}{}{#1}}}
\newrobustcmd{\brevemathcalFz}[2][]{\ensuremath{\subp{\breve{\mathcal{F}}}{}{#2}{}{#1}}}
\newrobustcmd{\barmathcalFz}[2][]{\ensuremath{\subp{\bar{\mathcal{F}}}{}{#2}{}{#1}}}
\newrobustcmd{\vecmathcalFz}[2][]{\ensuremath{\subp{\vec{\mathcal{F}}}{}{#2}{}{#1}}}
\newrobustcmd{\bmmathcalFz}[2][]{\ensuremath{\subp{\bm{\mathcal{F}}}{}{#2}{}{#1}}}
\newrobustcmd{\hatbmmathcalFz}[2][]{\ensuremath{\subp{\hat{\bm{\mathcal{F}}}}{}{#2}{}{#1}}}
\newrobustcmd{\widehatbmmathcalFz}[2][]{\ensuremath{\subp{\widehat{\bm{\mathcal{F}}}}{}{#2}{}{#1}}}
\newrobustcmd{\checkbmmathcalFz}[2][]{\ensuremath{\subp{\check{\bm{\mathcal{F}}}}{}{#2}{}{#1}}}
\newrobustcmd{\tildebmmathcalFz}[2][]{\ensuremath{\subp{\tilde{\bm{\mathcal{F}}}}{}{#2}{}{#1}}}
\newrobustcmd{\widetildebmmathcalFz}[2][]{\ensuremath{\subp{\widetilde{\bm{\mathcal{F}}}}{}{#2}{}{#1}}}
\newrobustcmd{\acutebmmathcalFz}[2][]{\ensuremath{\subp{\acute{\bm{\mathcal{F}}}}{}{#2}{}{#1}}}
\newrobustcmd{\gravebmmathcalFz}[2][]{\ensuremath{\subp{\grave{\bm{\mathcal{F}}}}{}{#2}{}{#1}}}
\newrobustcmd{\dotbmmathcalFz}[2][]{\ensuremath{\subp{\dot{\bm{\mathcal{F}}}}{}{#2}{}{#1}}}
\newrobustcmd{\ddotbmmathcalFz}[2][]{\ensuremath{\subp{\ddot{\bm{\mathcal{F}}}}{}{#2}{}{#1}}}
\newrobustcmd{\brevebmmathcalFz}[2][]{\ensuremath{\subp{\breve{\bm{\mathcal{F}}}}{}{#2}{}{#1}}}
\newrobustcmd{\barbmmathcalFz}[2][]{\ensuremath{\subp{\bar{\bm{\mathcal{F}}}}{}{#2}{}{#1}}}
\newrobustcmd{\vecbmmathcalFz}[2][]{\ensuremath{\subp{\vec{\bm{\mathcal{F}}}}{}{#2}{}{#1}}}
\newrobustcmd{\mathcalGz}[2][]{\ensuremath{\subp{\mathcal{G}}{}{#2}{}{#1}}}
\newrobustcmd{\hatmathcalGz}[2][]{\ensuremath{\subp{\hat{\mathcal{G}}}{}{#2}{}{#1}}}
\newrobustcmd{\widehatmathcalGz}[2][]{\ensuremath{\subp{\widehat{\mathcal{G}}}{}{#2}{}{#1}}}
\newrobustcmd{\checkmathcalGz}[2][]{\ensuremath{\subp{\check{\mathcal{G}}}{}{#2}{}{#1}}}
\newrobustcmd{\tildemathcalGz}[2][]{\ensuremath{\subp{\tilde{\mathcal{G}}}{}{#2}{}{#1}}}
\newrobustcmd{\widetildemathcalGz}[2][]{\ensuremath{\subp{\widetilde{\mathcal{G}}}{}{#2}{}{#1}}}
\newrobustcmd{\acutemathcalGz}[2][]{\ensuremath{\subp{\acute{\mathcal{G}}}{}{#2}{}{#1}}}
\newrobustcmd{\gravemathcalGz}[2][]{\ensuremath{\subp{\grave{\mathcal{G}}}{}{#2}{}{#1}}}
\newrobustcmd{\dotmathcalGz}[2][]{\ensuremath{\subp{\dot{\mathcal{G}}}{}{#2}{}{#1}}}
\newrobustcmd{\ddotmathcalGz}[2][]{\ensuremath{\subp{\ddot{\mathcal{G}}}{}{#2}{}{#1}}}
\newrobustcmd{\brevemathcalGz}[2][]{\ensuremath{\subp{\breve{\mathcal{G}}}{}{#2}{}{#1}}}
\newrobustcmd{\barmathcalGz}[2][]{\ensuremath{\subp{\bar{\mathcal{G}}}{}{#2}{}{#1}}}
\newrobustcmd{\vecmathcalGz}[2][]{\ensuremath{\subp{\vec{\mathcal{G}}}{}{#2}{}{#1}}}
\newrobustcmd{\bmmathcalGz}[2][]{\ensuremath{\subp{\bm{\mathcal{G}}}{}{#2}{}{#1}}}
\newrobustcmd{\hatbmmathcalGz}[2][]{\ensuremath{\subp{\hat{\bm{\mathcal{G}}}}{}{#2}{}{#1}}}
\newrobustcmd{\widehatbmmathcalGz}[2][]{\ensuremath{\subp{\widehat{\bm{\mathcal{G}}}}{}{#2}{}{#1}}}
\newrobustcmd{\checkbmmathcalGz}[2][]{\ensuremath{\subp{\check{\bm{\mathcal{G}}}}{}{#2}{}{#1}}}
\newrobustcmd{\tildebmmathcalGz}[2][]{\ensuremath{\subp{\tilde{\bm{\mathcal{G}}}}{}{#2}{}{#1}}}
\newrobustcmd{\widetildebmmathcalGz}[2][]{\ensuremath{\subp{\widetilde{\bm{\mathcal{G}}}}{}{#2}{}{#1}}}
\newrobustcmd{\acutebmmathcalGz}[2][]{\ensuremath{\subp{\acute{\bm{\mathcal{G}}}}{}{#2}{}{#1}}}
\newrobustcmd{\gravebmmathcalGz}[2][]{\ensuremath{\subp{\grave{\bm{\mathcal{G}}}}{}{#2}{}{#1}}}
\newrobustcmd{\dotbmmathcalGz}[2][]{\ensuremath{\subp{\dot{\bm{\mathcal{G}}}}{}{#2}{}{#1}}}
\newrobustcmd{\ddotbmmathcalGz}[2][]{\ensuremath{\subp{\ddot{\bm{\mathcal{G}}}}{}{#2}{}{#1}}}
\newrobustcmd{\brevebmmathcalGz}[2][]{\ensuremath{\subp{\breve{\bm{\mathcal{G}}}}{}{#2}{}{#1}}}
\newrobustcmd{\barbmmathcalGz}[2][]{\ensuremath{\subp{\bar{\bm{\mathcal{G}}}}{}{#2}{}{#1}}}
\newrobustcmd{\vecbmmathcalGz}[2][]{\ensuremath{\subp{\vec{\bm{\mathcal{G}}}}{}{#2}{}{#1}}}
\newrobustcmd{\mathcalHz}[2][]{\ensuremath{\subp{\mathcal{H}}{}{#2}{}{#1}}}
\newrobustcmd{\hatmathcalHz}[2][]{\ensuremath{\subp{\hat{\mathcal{H}}}{}{#2}{}{#1}}}
\newrobustcmd{\widehatmathcalHz}[2][]{\ensuremath{\subp{\widehat{\mathcal{H}}}{}{#2}{}{#1}}}
\newrobustcmd{\checkmathcalHz}[2][]{\ensuremath{\subp{\check{\mathcal{H}}}{}{#2}{}{#1}}}
\newrobustcmd{\tildemathcalHz}[2][]{\ensuremath{\subp{\tilde{\mathcal{H}}}{}{#2}{}{#1}}}
\newrobustcmd{\widetildemathcalHz}[2][]{\ensuremath{\subp{\widetilde{\mathcal{H}}}{}{#2}{}{#1}}}
\newrobustcmd{\acutemathcalHz}[2][]{\ensuremath{\subp{\acute{\mathcal{H}}}{}{#2}{}{#1}}}
\newrobustcmd{\gravemathcalHz}[2][]{\ensuremath{\subp{\grave{\mathcal{H}}}{}{#2}{}{#1}}}
\newrobustcmd{\dotmathcalHz}[2][]{\ensuremath{\subp{\dot{\mathcal{H}}}{}{#2}{}{#1}}}
\newrobustcmd{\ddotmathcalHz}[2][]{\ensuremath{\subp{\ddot{\mathcal{H}}}{}{#2}{}{#1}}}
\newrobustcmd{\brevemathcalHz}[2][]{\ensuremath{\subp{\breve{\mathcal{H}}}{}{#2}{}{#1}}}
\newrobustcmd{\barmathcalHz}[2][]{\ensuremath{\subp{\bar{\mathcal{H}}}{}{#2}{}{#1}}}
\newrobustcmd{\vecmathcalHz}[2][]{\ensuremath{\subp{\vec{\mathcal{H}}}{}{#2}{}{#1}}}
\newrobustcmd{\bmmathcalHz}[2][]{\ensuremath{\subp{\bm{\mathcal{H}}}{}{#2}{}{#1}}}
\newrobustcmd{\hatbmmathcalHz}[2][]{\ensuremath{\subp{\hat{\bm{\mathcal{H}}}}{}{#2}{}{#1}}}
\newrobustcmd{\widehatbmmathcalHz}[2][]{\ensuremath{\subp{\widehat{\bm{\mathcal{H}}}}{}{#2}{}{#1}}}
\newrobustcmd{\checkbmmathcalHz}[2][]{\ensuremath{\subp{\check{\bm{\mathcal{H}}}}{}{#2}{}{#1}}}
\newrobustcmd{\tildebmmathcalHz}[2][]{\ensuremath{\subp{\tilde{\bm{\mathcal{H}}}}{}{#2}{}{#1}}}
\newrobustcmd{\widetildebmmathcalHz}[2][]{\ensuremath{\subp{\widetilde{\bm{\mathcal{H}}}}{}{#2}{}{#1}}}
\newrobustcmd{\acutebmmathcalHz}[2][]{\ensuremath{\subp{\acute{\bm{\mathcal{H}}}}{}{#2}{}{#1}}}
\newrobustcmd{\gravebmmathcalHz}[2][]{\ensuremath{\subp{\grave{\bm{\mathcal{H}}}}{}{#2}{}{#1}}}
\newrobustcmd{\dotbmmathcalHz}[2][]{\ensuremath{\subp{\dot{\bm{\mathcal{H}}}}{}{#2}{}{#1}}}
\newrobustcmd{\ddotbmmathcalHz}[2][]{\ensuremath{\subp{\ddot{\bm{\mathcal{H}}}}{}{#2}{}{#1}}}
\newrobustcmd{\brevebmmathcalHz}[2][]{\ensuremath{\subp{\breve{\bm{\mathcal{H}}}}{}{#2}{}{#1}}}
\newrobustcmd{\barbmmathcalHz}[2][]{\ensuremath{\subp{\bar{\bm{\mathcal{H}}}}{}{#2}{}{#1}}}
\newrobustcmd{\vecbmmathcalHz}[2][]{\ensuremath{\subp{\vec{\bm{\mathcal{H}}}}{}{#2}{}{#1}}}
\newrobustcmd{\mathcalIz}[2][]{\ensuremath{\subp{\mathcal{I}}{}{#2}{}{#1}}}
\newrobustcmd{\hatmathcalIz}[2][]{\ensuremath{\subp{\hat{\mathcal{I}}}{}{#2}{}{#1}}}
\newrobustcmd{\widehatmathcalIz}[2][]{\ensuremath{\subp{\widehat{\mathcal{I}}}{}{#2}{}{#1}}}
\newrobustcmd{\checkmathcalIz}[2][]{\ensuremath{\subp{\check{\mathcal{I}}}{}{#2}{}{#1}}}
\newrobustcmd{\tildemathcalIz}[2][]{\ensuremath{\subp{\tilde{\mathcal{I}}}{}{#2}{}{#1}}}
\newrobustcmd{\widetildemathcalIz}[2][]{\ensuremath{\subp{\widetilde{\mathcal{I}}}{}{#2}{}{#1}}}
\newrobustcmd{\acutemathcalIz}[2][]{\ensuremath{\subp{\acute{\mathcal{I}}}{}{#2}{}{#1}}}
\newrobustcmd{\gravemathcalIz}[2][]{\ensuremath{\subp{\grave{\mathcal{I}}}{}{#2}{}{#1}}}
\newrobustcmd{\dotmathcalIz}[2][]{\ensuremath{\subp{\dot{\mathcal{I}}}{}{#2}{}{#1}}}
\newrobustcmd{\ddotmathcalIz}[2][]{\ensuremath{\subp{\ddot{\mathcal{I}}}{}{#2}{}{#1}}}
\newrobustcmd{\brevemathcalIz}[2][]{\ensuremath{\subp{\breve{\mathcal{I}}}{}{#2}{}{#1}}}
\newrobustcmd{\barmathcalIz}[2][]{\ensuremath{\subp{\bar{\mathcal{I}}}{}{#2}{}{#1}}}
\newrobustcmd{\vecmathcalIz}[2][]{\ensuremath{\subp{\vec{\mathcal{I}}}{}{#2}{}{#1}}}
\newrobustcmd{\bmmathcalIz}[2][]{\ensuremath{\subp{\bm{\mathcal{I}}}{}{#2}{}{#1}}}
\newrobustcmd{\hatbmmathcalIz}[2][]{\ensuremath{\subp{\hat{\bm{\mathcal{I}}}}{}{#2}{}{#1}}}
\newrobustcmd{\widehatbmmathcalIz}[2][]{\ensuremath{\subp{\widehat{\bm{\mathcal{I}}}}{}{#2}{}{#1}}}
\newrobustcmd{\checkbmmathcalIz}[2][]{\ensuremath{\subp{\check{\bm{\mathcal{I}}}}{}{#2}{}{#1}}}
\newrobustcmd{\tildebmmathcalIz}[2][]{\ensuremath{\subp{\tilde{\bm{\mathcal{I}}}}{}{#2}{}{#1}}}
\newrobustcmd{\widetildebmmathcalIz}[2][]{\ensuremath{\subp{\widetilde{\bm{\mathcal{I}}}}{}{#2}{}{#1}}}
\newrobustcmd{\acutebmmathcalIz}[2][]{\ensuremath{\subp{\acute{\bm{\mathcal{I}}}}{}{#2}{}{#1}}}
\newrobustcmd{\gravebmmathcalIz}[2][]{\ensuremath{\subp{\grave{\bm{\mathcal{I}}}}{}{#2}{}{#1}}}
\newrobustcmd{\dotbmmathcalIz}[2][]{\ensuremath{\subp{\dot{\bm{\mathcal{I}}}}{}{#2}{}{#1}}}
\newrobustcmd{\ddotbmmathcalIz}[2][]{\ensuremath{\subp{\ddot{\bm{\mathcal{I}}}}{}{#2}{}{#1}}}
\newrobustcmd{\brevebmmathcalIz}[2][]{\ensuremath{\subp{\breve{\bm{\mathcal{I}}}}{}{#2}{}{#1}}}
\newrobustcmd{\barbmmathcalIz}[2][]{\ensuremath{\subp{\bar{\bm{\mathcal{I}}}}{}{#2}{}{#1}}}
\newrobustcmd{\vecbmmathcalIz}[2][]{\ensuremath{\subp{\vec{\bm{\mathcal{I}}}}{}{#2}{}{#1}}}
\newrobustcmd{\mathcalJz}[2][]{\ensuremath{\subp{\mathcal{J}}{}{#2}{}{#1}}}
\newrobustcmd{\hatmathcalJz}[2][]{\ensuremath{\subp{\hat{\mathcal{J}}}{}{#2}{}{#1}}}
\newrobustcmd{\widehatmathcalJz}[2][]{\ensuremath{\subp{\widehat{\mathcal{J}}}{}{#2}{}{#1}}}
\newrobustcmd{\checkmathcalJz}[2][]{\ensuremath{\subp{\check{\mathcal{J}}}{}{#2}{}{#1}}}
\newrobustcmd{\tildemathcalJz}[2][]{\ensuremath{\subp{\tilde{\mathcal{J}}}{}{#2}{}{#1}}}
\newrobustcmd{\widetildemathcalJz}[2][]{\ensuremath{\subp{\widetilde{\mathcal{J}}}{}{#2}{}{#1}}}
\newrobustcmd{\acutemathcalJz}[2][]{\ensuremath{\subp{\acute{\mathcal{J}}}{}{#2}{}{#1}}}
\newrobustcmd{\gravemathcalJz}[2][]{\ensuremath{\subp{\grave{\mathcal{J}}}{}{#2}{}{#1}}}
\newrobustcmd{\dotmathcalJz}[2][]{\ensuremath{\subp{\dot{\mathcal{J}}}{}{#2}{}{#1}}}
\newrobustcmd{\ddotmathcalJz}[2][]{\ensuremath{\subp{\ddot{\mathcal{J}}}{}{#2}{}{#1}}}
\newrobustcmd{\brevemathcalJz}[2][]{\ensuremath{\subp{\breve{\mathcal{J}}}{}{#2}{}{#1}}}
\newrobustcmd{\barmathcalJz}[2][]{\ensuremath{\subp{\bar{\mathcal{J}}}{}{#2}{}{#1}}}
\newrobustcmd{\vecmathcalJz}[2][]{\ensuremath{\subp{\vec{\mathcal{J}}}{}{#2}{}{#1}}}
\newrobustcmd{\bmmathcalJz}[2][]{\ensuremath{\subp{\bm{\mathcal{J}}}{}{#2}{}{#1}}}
\newrobustcmd{\hatbmmathcalJz}[2][]{\ensuremath{\subp{\hat{\bm{\mathcal{J}}}}{}{#2}{}{#1}}}
\newrobustcmd{\widehatbmmathcalJz}[2][]{\ensuremath{\subp{\widehat{\bm{\mathcal{J}}}}{}{#2}{}{#1}}}
\newrobustcmd{\checkbmmathcalJz}[2][]{\ensuremath{\subp{\check{\bm{\mathcal{J}}}}{}{#2}{}{#1}}}
\newrobustcmd{\tildebmmathcalJz}[2][]{\ensuremath{\subp{\tilde{\bm{\mathcal{J}}}}{}{#2}{}{#1}}}
\newrobustcmd{\widetildebmmathcalJz}[2][]{\ensuremath{\subp{\widetilde{\bm{\mathcal{J}}}}{}{#2}{}{#1}}}
\newrobustcmd{\acutebmmathcalJz}[2][]{\ensuremath{\subp{\acute{\bm{\mathcal{J}}}}{}{#2}{}{#1}}}
\newrobustcmd{\gravebmmathcalJz}[2][]{\ensuremath{\subp{\grave{\bm{\mathcal{J}}}}{}{#2}{}{#1}}}
\newrobustcmd{\dotbmmathcalJz}[2][]{\ensuremath{\subp{\dot{\bm{\mathcal{J}}}}{}{#2}{}{#1}}}
\newrobustcmd{\ddotbmmathcalJz}[2][]{\ensuremath{\subp{\ddot{\bm{\mathcal{J}}}}{}{#2}{}{#1}}}
\newrobustcmd{\brevebmmathcalJz}[2][]{\ensuremath{\subp{\breve{\bm{\mathcal{J}}}}{}{#2}{}{#1}}}
\newrobustcmd{\barbmmathcalJz}[2][]{\ensuremath{\subp{\bar{\bm{\mathcal{J}}}}{}{#2}{}{#1}}}
\newrobustcmd{\vecbmmathcalJz}[2][]{\ensuremath{\subp{\vec{\bm{\mathcal{J}}}}{}{#2}{}{#1}}}
\newrobustcmd{\mathcalKz}[2][]{\ensuremath{\subp{\mathcal{K}}{}{#2}{}{#1}}}
\newrobustcmd{\hatmathcalKz}[2][]{\ensuremath{\subp{\hat{\mathcal{K}}}{}{#2}{}{#1}}}
\newrobustcmd{\widehatmathcalKz}[2][]{\ensuremath{\subp{\widehat{\mathcal{K}}}{}{#2}{}{#1}}}
\newrobustcmd{\checkmathcalKz}[2][]{\ensuremath{\subp{\check{\mathcal{K}}}{}{#2}{}{#1}}}
\newrobustcmd{\tildemathcalKz}[2][]{\ensuremath{\subp{\tilde{\mathcal{K}}}{}{#2}{}{#1}}}
\newrobustcmd{\widetildemathcalKz}[2][]{\ensuremath{\subp{\widetilde{\mathcal{K}}}{}{#2}{}{#1}}}
\newrobustcmd{\acutemathcalKz}[2][]{\ensuremath{\subp{\acute{\mathcal{K}}}{}{#2}{}{#1}}}
\newrobustcmd{\gravemathcalKz}[2][]{\ensuremath{\subp{\grave{\mathcal{K}}}{}{#2}{}{#1}}}
\newrobustcmd{\dotmathcalKz}[2][]{\ensuremath{\subp{\dot{\mathcal{K}}}{}{#2}{}{#1}}}
\newrobustcmd{\ddotmathcalKz}[2][]{\ensuremath{\subp{\ddot{\mathcal{K}}}{}{#2}{}{#1}}}
\newrobustcmd{\brevemathcalKz}[2][]{\ensuremath{\subp{\breve{\mathcal{K}}}{}{#2}{}{#1}}}
\newrobustcmd{\barmathcalKz}[2][]{\ensuremath{\subp{\bar{\mathcal{K}}}{}{#2}{}{#1}}}
\newrobustcmd{\vecmathcalKz}[2][]{\ensuremath{\subp{\vec{\mathcal{K}}}{}{#2}{}{#1}}}
\newrobustcmd{\bmmathcalKz}[2][]{\ensuremath{\subp{\bm{\mathcal{K}}}{}{#2}{}{#1}}}
\newrobustcmd{\hatbmmathcalKz}[2][]{\ensuremath{\subp{\hat{\bm{\mathcal{K}}}}{}{#2}{}{#1}}}
\newrobustcmd{\widehatbmmathcalKz}[2][]{\ensuremath{\subp{\widehat{\bm{\mathcal{K}}}}{}{#2}{}{#1}}}
\newrobustcmd{\checkbmmathcalKz}[2][]{\ensuremath{\subp{\check{\bm{\mathcal{K}}}}{}{#2}{}{#1}}}
\newrobustcmd{\tildebmmathcalKz}[2][]{\ensuremath{\subp{\tilde{\bm{\mathcal{K}}}}{}{#2}{}{#1}}}
\newrobustcmd{\widetildebmmathcalKz}[2][]{\ensuremath{\subp{\widetilde{\bm{\mathcal{K}}}}{}{#2}{}{#1}}}
\newrobustcmd{\acutebmmathcalKz}[2][]{\ensuremath{\subp{\acute{\bm{\mathcal{K}}}}{}{#2}{}{#1}}}
\newrobustcmd{\gravebmmathcalKz}[2][]{\ensuremath{\subp{\grave{\bm{\mathcal{K}}}}{}{#2}{}{#1}}}
\newrobustcmd{\dotbmmathcalKz}[2][]{\ensuremath{\subp{\dot{\bm{\mathcal{K}}}}{}{#2}{}{#1}}}
\newrobustcmd{\ddotbmmathcalKz}[2][]{\ensuremath{\subp{\ddot{\bm{\mathcal{K}}}}{}{#2}{}{#1}}}
\newrobustcmd{\brevebmmathcalKz}[2][]{\ensuremath{\subp{\breve{\bm{\mathcal{K}}}}{}{#2}{}{#1}}}
\newrobustcmd{\barbmmathcalKz}[2][]{\ensuremath{\subp{\bar{\bm{\mathcal{K}}}}{}{#2}{}{#1}}}
\newrobustcmd{\vecbmmathcalKz}[2][]{\ensuremath{\subp{\vec{\bm{\mathcal{K}}}}{}{#2}{}{#1}}}
\newrobustcmd{\mathcalLz}[2][]{\ensuremath{\subp{\mathcal{L}}{}{#2}{}{#1}}}
\newrobustcmd{\hatmathcalLz}[2][]{\ensuremath{\subp{\hat{\mathcal{L}}}{}{#2}{}{#1}}}
\newrobustcmd{\widehatmathcalLz}[2][]{\ensuremath{\subp{\widehat{\mathcal{L}}}{}{#2}{}{#1}}}
\newrobustcmd{\checkmathcalLz}[2][]{\ensuremath{\subp{\check{\mathcal{L}}}{}{#2}{}{#1}}}
\newrobustcmd{\tildemathcalLz}[2][]{\ensuremath{\subp{\tilde{\mathcal{L}}}{}{#2}{}{#1}}}
\newrobustcmd{\widetildemathcalLz}[2][]{\ensuremath{\subp{\widetilde{\mathcal{L}}}{}{#2}{}{#1}}}
\newrobustcmd{\acutemathcalLz}[2][]{\ensuremath{\subp{\acute{\mathcal{L}}}{}{#2}{}{#1}}}
\newrobustcmd{\gravemathcalLz}[2][]{\ensuremath{\subp{\grave{\mathcal{L}}}{}{#2}{}{#1}}}
\newrobustcmd{\dotmathcalLz}[2][]{\ensuremath{\subp{\dot{\mathcal{L}}}{}{#2}{}{#1}}}
\newrobustcmd{\ddotmathcalLz}[2][]{\ensuremath{\subp{\ddot{\mathcal{L}}}{}{#2}{}{#1}}}
\newrobustcmd{\brevemathcalLz}[2][]{\ensuremath{\subp{\breve{\mathcal{L}}}{}{#2}{}{#1}}}
\newrobustcmd{\barmathcalLz}[2][]{\ensuremath{\subp{\bar{\mathcal{L}}}{}{#2}{}{#1}}}
\newrobustcmd{\vecmathcalLz}[2][]{\ensuremath{\subp{\vec{\mathcal{L}}}{}{#2}{}{#1}}}
\newrobustcmd{\bmmathcalLz}[2][]{\ensuremath{\subp{\bm{\mathcal{L}}}{}{#2}{}{#1}}}
\newrobustcmd{\hatbmmathcalLz}[2][]{\ensuremath{\subp{\hat{\bm{\mathcal{L}}}}{}{#2}{}{#1}}}
\newrobustcmd{\widehatbmmathcalLz}[2][]{\ensuremath{\subp{\widehat{\bm{\mathcal{L}}}}{}{#2}{}{#1}}}
\newrobustcmd{\checkbmmathcalLz}[2][]{\ensuremath{\subp{\check{\bm{\mathcal{L}}}}{}{#2}{}{#1}}}
\newrobustcmd{\tildebmmathcalLz}[2][]{\ensuremath{\subp{\tilde{\bm{\mathcal{L}}}}{}{#2}{}{#1}}}
\newrobustcmd{\widetildebmmathcalLz}[2][]{\ensuremath{\subp{\widetilde{\bm{\mathcal{L}}}}{}{#2}{}{#1}}}
\newrobustcmd{\acutebmmathcalLz}[2][]{\ensuremath{\subp{\acute{\bm{\mathcal{L}}}}{}{#2}{}{#1}}}
\newrobustcmd{\gravebmmathcalLz}[2][]{\ensuremath{\subp{\grave{\bm{\mathcal{L}}}}{}{#2}{}{#1}}}
\newrobustcmd{\dotbmmathcalLz}[2][]{\ensuremath{\subp{\dot{\bm{\mathcal{L}}}}{}{#2}{}{#1}}}
\newrobustcmd{\ddotbmmathcalLz}[2][]{\ensuremath{\subp{\ddot{\bm{\mathcal{L}}}}{}{#2}{}{#1}}}
\newrobustcmd{\brevebmmathcalLz}[2][]{\ensuremath{\subp{\breve{\bm{\mathcal{L}}}}{}{#2}{}{#1}}}
\newrobustcmd{\barbmmathcalLz}[2][]{\ensuremath{\subp{\bar{\bm{\mathcal{L}}}}{}{#2}{}{#1}}}
\newrobustcmd{\vecbmmathcalLz}[2][]{\ensuremath{\subp{\vec{\bm{\mathcal{L}}}}{}{#2}{}{#1}}}
\newrobustcmd{\mathcalMz}[2][]{\ensuremath{\subp{\mathcal{M}}{}{#2}{}{#1}}}
\newrobustcmd{\hatmathcalMz}[2][]{\ensuremath{\subp{\hat{\mathcal{M}}}{}{#2}{}{#1}}}
\newrobustcmd{\widehatmathcalMz}[2][]{\ensuremath{\subp{\widehat{\mathcal{M}}}{}{#2}{}{#1}}}
\newrobustcmd{\checkmathcalMz}[2][]{\ensuremath{\subp{\check{\mathcal{M}}}{}{#2}{}{#1}}}
\newrobustcmd{\tildemathcalMz}[2][]{\ensuremath{\subp{\tilde{\mathcal{M}}}{}{#2}{}{#1}}}
\newrobustcmd{\widetildemathcalMz}[2][]{\ensuremath{\subp{\widetilde{\mathcal{M}}}{}{#2}{}{#1}}}
\newrobustcmd{\acutemathcalMz}[2][]{\ensuremath{\subp{\acute{\mathcal{M}}}{}{#2}{}{#1}}}
\newrobustcmd{\gravemathcalMz}[2][]{\ensuremath{\subp{\grave{\mathcal{M}}}{}{#2}{}{#1}}}
\newrobustcmd{\dotmathcalMz}[2][]{\ensuremath{\subp{\dot{\mathcal{M}}}{}{#2}{}{#1}}}
\newrobustcmd{\ddotmathcalMz}[2][]{\ensuremath{\subp{\ddot{\mathcal{M}}}{}{#2}{}{#1}}}
\newrobustcmd{\brevemathcalMz}[2][]{\ensuremath{\subp{\breve{\mathcal{M}}}{}{#2}{}{#1}}}
\newrobustcmd{\barmathcalMz}[2][]{\ensuremath{\subp{\bar{\mathcal{M}}}{}{#2}{}{#1}}}
\newrobustcmd{\vecmathcalMz}[2][]{\ensuremath{\subp{\vec{\mathcal{M}}}{}{#2}{}{#1}}}
\newrobustcmd{\bmmathcalMz}[2][]{\ensuremath{\subp{\bm{\mathcal{M}}}{}{#2}{}{#1}}}
\newrobustcmd{\hatbmmathcalMz}[2][]{\ensuremath{\subp{\hat{\bm{\mathcal{M}}}}{}{#2}{}{#1}}}
\newrobustcmd{\widehatbmmathcalMz}[2][]{\ensuremath{\subp{\widehat{\bm{\mathcal{M}}}}{}{#2}{}{#1}}}
\newrobustcmd{\checkbmmathcalMz}[2][]{\ensuremath{\subp{\check{\bm{\mathcal{M}}}}{}{#2}{}{#1}}}
\newrobustcmd{\tildebmmathcalMz}[2][]{\ensuremath{\subp{\tilde{\bm{\mathcal{M}}}}{}{#2}{}{#1}}}
\newrobustcmd{\widetildebmmathcalMz}[2][]{\ensuremath{\subp{\widetilde{\bm{\mathcal{M}}}}{}{#2}{}{#1}}}
\newrobustcmd{\acutebmmathcalMz}[2][]{\ensuremath{\subp{\acute{\bm{\mathcal{M}}}}{}{#2}{}{#1}}}
\newrobustcmd{\gravebmmathcalMz}[2][]{\ensuremath{\subp{\grave{\bm{\mathcal{M}}}}{}{#2}{}{#1}}}
\newrobustcmd{\dotbmmathcalMz}[2][]{\ensuremath{\subp{\dot{\bm{\mathcal{M}}}}{}{#2}{}{#1}}}
\newrobustcmd{\ddotbmmathcalMz}[2][]{\ensuremath{\subp{\ddot{\bm{\mathcal{M}}}}{}{#2}{}{#1}}}
\newrobustcmd{\brevebmmathcalMz}[2][]{\ensuremath{\subp{\breve{\bm{\mathcal{M}}}}{}{#2}{}{#1}}}
\newrobustcmd{\barbmmathcalMz}[2][]{\ensuremath{\subp{\bar{\bm{\mathcal{M}}}}{}{#2}{}{#1}}}
\newrobustcmd{\vecbmmathcalMz}[2][]{\ensuremath{\subp{\vec{\bm{\mathcal{M}}}}{}{#2}{}{#1}}}
\newrobustcmd{\mathcalNz}[2][]{\ensuremath{\subp{\mathcal{N}}{}{#2}{}{#1}}}
\newrobustcmd{\hatmathcalNz}[2][]{\ensuremath{\subp{\hat{\mathcal{N}}}{}{#2}{}{#1}}}
\newrobustcmd{\widehatmathcalNz}[2][]{\ensuremath{\subp{\widehat{\mathcal{N}}}{}{#2}{}{#1}}}
\newrobustcmd{\checkmathcalNz}[2][]{\ensuremath{\subp{\check{\mathcal{N}}}{}{#2}{}{#1}}}
\newrobustcmd{\tildemathcalNz}[2][]{\ensuremath{\subp{\tilde{\mathcal{N}}}{}{#2}{}{#1}}}
\newrobustcmd{\widetildemathcalNz}[2][]{\ensuremath{\subp{\widetilde{\mathcal{N}}}{}{#2}{}{#1}}}
\newrobustcmd{\acutemathcalNz}[2][]{\ensuremath{\subp{\acute{\mathcal{N}}}{}{#2}{}{#1}}}
\newrobustcmd{\gravemathcalNz}[2][]{\ensuremath{\subp{\grave{\mathcal{N}}}{}{#2}{}{#1}}}
\newrobustcmd{\dotmathcalNz}[2][]{\ensuremath{\subp{\dot{\mathcal{N}}}{}{#2}{}{#1}}}
\newrobustcmd{\ddotmathcalNz}[2][]{\ensuremath{\subp{\ddot{\mathcal{N}}}{}{#2}{}{#1}}}
\newrobustcmd{\brevemathcalNz}[2][]{\ensuremath{\subp{\breve{\mathcal{N}}}{}{#2}{}{#1}}}
\newrobustcmd{\barmathcalNz}[2][]{\ensuremath{\subp{\bar{\mathcal{N}}}{}{#2}{}{#1}}}
\newrobustcmd{\vecmathcalNz}[2][]{\ensuremath{\subp{\vec{\mathcal{N}}}{}{#2}{}{#1}}}
\newrobustcmd{\bmmathcalNz}[2][]{\ensuremath{\subp{\bm{\mathcal{N}}}{}{#2}{}{#1}}}
\newrobustcmd{\hatbmmathcalNz}[2][]{\ensuremath{\subp{\hat{\bm{\mathcal{N}}}}{}{#2}{}{#1}}}
\newrobustcmd{\widehatbmmathcalNz}[2][]{\ensuremath{\subp{\widehat{\bm{\mathcal{N}}}}{}{#2}{}{#1}}}
\newrobustcmd{\checkbmmathcalNz}[2][]{\ensuremath{\subp{\check{\bm{\mathcal{N}}}}{}{#2}{}{#1}}}
\newrobustcmd{\tildebmmathcalNz}[2][]{\ensuremath{\subp{\tilde{\bm{\mathcal{N}}}}{}{#2}{}{#1}}}
\newrobustcmd{\widetildebmmathcalNz}[2][]{\ensuremath{\subp{\widetilde{\bm{\mathcal{N}}}}{}{#2}{}{#1}}}
\newrobustcmd{\acutebmmathcalNz}[2][]{\ensuremath{\subp{\acute{\bm{\mathcal{N}}}}{}{#2}{}{#1}}}
\newrobustcmd{\gravebmmathcalNz}[2][]{\ensuremath{\subp{\grave{\bm{\mathcal{N}}}}{}{#2}{}{#1}}}
\newrobustcmd{\dotbmmathcalNz}[2][]{\ensuremath{\subp{\dot{\bm{\mathcal{N}}}}{}{#2}{}{#1}}}
\newrobustcmd{\ddotbmmathcalNz}[2][]{\ensuremath{\subp{\ddot{\bm{\mathcal{N}}}}{}{#2}{}{#1}}}
\newrobustcmd{\brevebmmathcalNz}[2][]{\ensuremath{\subp{\breve{\bm{\mathcal{N}}}}{}{#2}{}{#1}}}
\newrobustcmd{\barbmmathcalNz}[2][]{\ensuremath{\subp{\bar{\bm{\mathcal{N}}}}{}{#2}{}{#1}}}
\newrobustcmd{\vecbmmathcalNz}[2][]{\ensuremath{\subp{\vec{\bm{\mathcal{N}}}}{}{#2}{}{#1}}}
\newrobustcmd{\mathcalOz}[2][]{\ensuremath{\subp{\mathcal{O}}{}{#2}{}{#1}}}
\newrobustcmd{\hatmathcalOz}[2][]{\ensuremath{\subp{\hat{\mathcal{O}}}{}{#2}{}{#1}}}
\newrobustcmd{\widehatmathcalOz}[2][]{\ensuremath{\subp{\widehat{\mathcal{O}}}{}{#2}{}{#1}}}
\newrobustcmd{\checkmathcalOz}[2][]{\ensuremath{\subp{\check{\mathcal{O}}}{}{#2}{}{#1}}}
\newrobustcmd{\tildemathcalOz}[2][]{\ensuremath{\subp{\tilde{\mathcal{O}}}{}{#2}{}{#1}}}
\newrobustcmd{\widetildemathcalOz}[2][]{\ensuremath{\subp{\widetilde{\mathcal{O}}}{}{#2}{}{#1}}}
\newrobustcmd{\acutemathcalOz}[2][]{\ensuremath{\subp{\acute{\mathcal{O}}}{}{#2}{}{#1}}}
\newrobustcmd{\gravemathcalOz}[2][]{\ensuremath{\subp{\grave{\mathcal{O}}}{}{#2}{}{#1}}}
\newrobustcmd{\dotmathcalOz}[2][]{\ensuremath{\subp{\dot{\mathcal{O}}}{}{#2}{}{#1}}}
\newrobustcmd{\ddotmathcalOz}[2][]{\ensuremath{\subp{\ddot{\mathcal{O}}}{}{#2}{}{#1}}}
\newrobustcmd{\brevemathcalOz}[2][]{\ensuremath{\subp{\breve{\mathcal{O}}}{}{#2}{}{#1}}}
\newrobustcmd{\barmathcalOz}[2][]{\ensuremath{\subp{\bar{\mathcal{O}}}{}{#2}{}{#1}}}
\newrobustcmd{\vecmathcalOz}[2][]{\ensuremath{\subp{\vec{\mathcal{O}}}{}{#2}{}{#1}}}
\newrobustcmd{\bmmathcalOz}[2][]{\ensuremath{\subp{\bm{\mathcal{O}}}{}{#2}{}{#1}}}
\newrobustcmd{\hatbmmathcalOz}[2][]{\ensuremath{\subp{\hat{\bm{\mathcal{O}}}}{}{#2}{}{#1}}}
\newrobustcmd{\widehatbmmathcalOz}[2][]{\ensuremath{\subp{\widehat{\bm{\mathcal{O}}}}{}{#2}{}{#1}}}
\newrobustcmd{\checkbmmathcalOz}[2][]{\ensuremath{\subp{\check{\bm{\mathcal{O}}}}{}{#2}{}{#1}}}
\newrobustcmd{\tildebmmathcalOz}[2][]{\ensuremath{\subp{\tilde{\bm{\mathcal{O}}}}{}{#2}{}{#1}}}
\newrobustcmd{\widetildebmmathcalOz}[2][]{\ensuremath{\subp{\widetilde{\bm{\mathcal{O}}}}{}{#2}{}{#1}}}
\newrobustcmd{\acutebmmathcalOz}[2][]{\ensuremath{\subp{\acute{\bm{\mathcal{O}}}}{}{#2}{}{#1}}}
\newrobustcmd{\gravebmmathcalOz}[2][]{\ensuremath{\subp{\grave{\bm{\mathcal{O}}}}{}{#2}{}{#1}}}
\newrobustcmd{\dotbmmathcalOz}[2][]{\ensuremath{\subp{\dot{\bm{\mathcal{O}}}}{}{#2}{}{#1}}}
\newrobustcmd{\ddotbmmathcalOz}[2][]{\ensuremath{\subp{\ddot{\bm{\mathcal{O}}}}{}{#2}{}{#1}}}
\newrobustcmd{\brevebmmathcalOz}[2][]{\ensuremath{\subp{\breve{\bm{\mathcal{O}}}}{}{#2}{}{#1}}}
\newrobustcmd{\barbmmathcalOz}[2][]{\ensuremath{\subp{\bar{\bm{\mathcal{O}}}}{}{#2}{}{#1}}}
\newrobustcmd{\vecbmmathcalOz}[2][]{\ensuremath{\subp{\vec{\bm{\mathcal{O}}}}{}{#2}{}{#1}}}
\newrobustcmd{\mathcalPz}[2][]{\ensuremath{\subp{\mathcal{P}}{}{#2}{}{#1}}}
\newrobustcmd{\hatmathcalPz}[2][]{\ensuremath{\subp{\hat{\mathcal{P}}}{}{#2}{}{#1}}}
\newrobustcmd{\widehatmathcalPz}[2][]{\ensuremath{\subp{\widehat{\mathcal{P}}}{}{#2}{}{#1}}}
\newrobustcmd{\checkmathcalPz}[2][]{\ensuremath{\subp{\check{\mathcal{P}}}{}{#2}{}{#1}}}
\newrobustcmd{\tildemathcalPz}[2][]{\ensuremath{\subp{\tilde{\mathcal{P}}}{}{#2}{}{#1}}}
\newrobustcmd{\widetildemathcalPz}[2][]{\ensuremath{\subp{\widetilde{\mathcal{P}}}{}{#2}{}{#1}}}
\newrobustcmd{\acutemathcalPz}[2][]{\ensuremath{\subp{\acute{\mathcal{P}}}{}{#2}{}{#1}}}
\newrobustcmd{\gravemathcalPz}[2][]{\ensuremath{\subp{\grave{\mathcal{P}}}{}{#2}{}{#1}}}
\newrobustcmd{\dotmathcalPz}[2][]{\ensuremath{\subp{\dot{\mathcal{P}}}{}{#2}{}{#1}}}
\newrobustcmd{\ddotmathcalPz}[2][]{\ensuremath{\subp{\ddot{\mathcal{P}}}{}{#2}{}{#1}}}
\newrobustcmd{\brevemathcalPz}[2][]{\ensuremath{\subp{\breve{\mathcal{P}}}{}{#2}{}{#1}}}
\newrobustcmd{\barmathcalPz}[2][]{\ensuremath{\subp{\bar{\mathcal{P}}}{}{#2}{}{#1}}}
\newrobustcmd{\vecmathcalPz}[2][]{\ensuremath{\subp{\vec{\mathcal{P}}}{}{#2}{}{#1}}}
\newrobustcmd{\bmmathcalPz}[2][]{\ensuremath{\subp{\bm{\mathcal{P}}}{}{#2}{}{#1}}}
\newrobustcmd{\hatbmmathcalPz}[2][]{\ensuremath{\subp{\hat{\bm{\mathcal{P}}}}{}{#2}{}{#1}}}
\newrobustcmd{\widehatbmmathcalPz}[2][]{\ensuremath{\subp{\widehat{\bm{\mathcal{P}}}}{}{#2}{}{#1}}}
\newrobustcmd{\checkbmmathcalPz}[2][]{\ensuremath{\subp{\check{\bm{\mathcal{P}}}}{}{#2}{}{#1}}}
\newrobustcmd{\tildebmmathcalPz}[2][]{\ensuremath{\subp{\tilde{\bm{\mathcal{P}}}}{}{#2}{}{#1}}}
\newrobustcmd{\widetildebmmathcalPz}[2][]{\ensuremath{\subp{\widetilde{\bm{\mathcal{P}}}}{}{#2}{}{#1}}}
\newrobustcmd{\acutebmmathcalPz}[2][]{\ensuremath{\subp{\acute{\bm{\mathcal{P}}}}{}{#2}{}{#1}}}
\newrobustcmd{\gravebmmathcalPz}[2][]{\ensuremath{\subp{\grave{\bm{\mathcal{P}}}}{}{#2}{}{#1}}}
\newrobustcmd{\dotbmmathcalPz}[2][]{\ensuremath{\subp{\dot{\bm{\mathcal{P}}}}{}{#2}{}{#1}}}
\newrobustcmd{\ddotbmmathcalPz}[2][]{\ensuremath{\subp{\ddot{\bm{\mathcal{P}}}}{}{#2}{}{#1}}}
\newrobustcmd{\brevebmmathcalPz}[2][]{\ensuremath{\subp{\breve{\bm{\mathcal{P}}}}{}{#2}{}{#1}}}
\newrobustcmd{\barbmmathcalPz}[2][]{\ensuremath{\subp{\bar{\bm{\mathcal{P}}}}{}{#2}{}{#1}}}
\newrobustcmd{\vecbmmathcalPz}[2][]{\ensuremath{\subp{\vec{\bm{\mathcal{P}}}}{}{#2}{}{#1}}}
\newrobustcmd{\mathcalQz}[2][]{\ensuremath{\subp{\mathcal{Q}}{}{#2}{}{#1}}}
\newrobustcmd{\hatmathcalQz}[2][]{\ensuremath{\subp{\hat{\mathcal{Q}}}{}{#2}{}{#1}}}
\newrobustcmd{\widehatmathcalQz}[2][]{\ensuremath{\subp{\widehat{\mathcal{Q}}}{}{#2}{}{#1}}}
\newrobustcmd{\checkmathcalQz}[2][]{\ensuremath{\subp{\check{\mathcal{Q}}}{}{#2}{}{#1}}}
\newrobustcmd{\tildemathcalQz}[2][]{\ensuremath{\subp{\tilde{\mathcal{Q}}}{}{#2}{}{#1}}}
\newrobustcmd{\widetildemathcalQz}[2][]{\ensuremath{\subp{\widetilde{\mathcal{Q}}}{}{#2}{}{#1}}}
\newrobustcmd{\acutemathcalQz}[2][]{\ensuremath{\subp{\acute{\mathcal{Q}}}{}{#2}{}{#1}}}
\newrobustcmd{\gravemathcalQz}[2][]{\ensuremath{\subp{\grave{\mathcal{Q}}}{}{#2}{}{#1}}}
\newrobustcmd{\dotmathcalQz}[2][]{\ensuremath{\subp{\dot{\mathcal{Q}}}{}{#2}{}{#1}}}
\newrobustcmd{\ddotmathcalQz}[2][]{\ensuremath{\subp{\ddot{\mathcal{Q}}}{}{#2}{}{#1}}}
\newrobustcmd{\brevemathcalQz}[2][]{\ensuremath{\subp{\breve{\mathcal{Q}}}{}{#2}{}{#1}}}
\newrobustcmd{\barmathcalQz}[2][]{\ensuremath{\subp{\bar{\mathcal{Q}}}{}{#2}{}{#1}}}
\newrobustcmd{\vecmathcalQz}[2][]{\ensuremath{\subp{\vec{\mathcal{Q}}}{}{#2}{}{#1}}}
\newrobustcmd{\bmmathcalQz}[2][]{\ensuremath{\subp{\bm{\mathcal{Q}}}{}{#2}{}{#1}}}
\newrobustcmd{\hatbmmathcalQz}[2][]{\ensuremath{\subp{\hat{\bm{\mathcal{Q}}}}{}{#2}{}{#1}}}
\newrobustcmd{\widehatbmmathcalQz}[2][]{\ensuremath{\subp{\widehat{\bm{\mathcal{Q}}}}{}{#2}{}{#1}}}
\newrobustcmd{\checkbmmathcalQz}[2][]{\ensuremath{\subp{\check{\bm{\mathcal{Q}}}}{}{#2}{}{#1}}}
\newrobustcmd{\tildebmmathcalQz}[2][]{\ensuremath{\subp{\tilde{\bm{\mathcal{Q}}}}{}{#2}{}{#1}}}
\newrobustcmd{\widetildebmmathcalQz}[2][]{\ensuremath{\subp{\widetilde{\bm{\mathcal{Q}}}}{}{#2}{}{#1}}}
\newrobustcmd{\acutebmmathcalQz}[2][]{\ensuremath{\subp{\acute{\bm{\mathcal{Q}}}}{}{#2}{}{#1}}}
\newrobustcmd{\gravebmmathcalQz}[2][]{\ensuremath{\subp{\grave{\bm{\mathcal{Q}}}}{}{#2}{}{#1}}}
\newrobustcmd{\dotbmmathcalQz}[2][]{\ensuremath{\subp{\dot{\bm{\mathcal{Q}}}}{}{#2}{}{#1}}}
\newrobustcmd{\ddotbmmathcalQz}[2][]{\ensuremath{\subp{\ddot{\bm{\mathcal{Q}}}}{}{#2}{}{#1}}}
\newrobustcmd{\brevebmmathcalQz}[2][]{\ensuremath{\subp{\breve{\bm{\mathcal{Q}}}}{}{#2}{}{#1}}}
\newrobustcmd{\barbmmathcalQz}[2][]{\ensuremath{\subp{\bar{\bm{\mathcal{Q}}}}{}{#2}{}{#1}}}
\newrobustcmd{\vecbmmathcalQz}[2][]{\ensuremath{\subp{\vec{\bm{\mathcal{Q}}}}{}{#2}{}{#1}}}
\newrobustcmd{\mathcalRz}[2][]{\ensuremath{\subp{\mathcal{R}}{}{#2}{}{#1}}}
\newrobustcmd{\hatmathcalRz}[2][]{\ensuremath{\subp{\hat{\mathcal{R}}}{}{#2}{}{#1}}}
\newrobustcmd{\widehatmathcalRz}[2][]{\ensuremath{\subp{\widehat{\mathcal{R}}}{}{#2}{}{#1}}}
\newrobustcmd{\checkmathcalRz}[2][]{\ensuremath{\subp{\check{\mathcal{R}}}{}{#2}{}{#1}}}
\newrobustcmd{\tildemathcalRz}[2][]{\ensuremath{\subp{\tilde{\mathcal{R}}}{}{#2}{}{#1}}}
\newrobustcmd{\widetildemathcalRz}[2][]{\ensuremath{\subp{\widetilde{\mathcal{R}}}{}{#2}{}{#1}}}
\newrobustcmd{\acutemathcalRz}[2][]{\ensuremath{\subp{\acute{\mathcal{R}}}{}{#2}{}{#1}}}
\newrobustcmd{\gravemathcalRz}[2][]{\ensuremath{\subp{\grave{\mathcal{R}}}{}{#2}{}{#1}}}
\newrobustcmd{\dotmathcalRz}[2][]{\ensuremath{\subp{\dot{\mathcal{R}}}{}{#2}{}{#1}}}
\newrobustcmd{\ddotmathcalRz}[2][]{\ensuremath{\subp{\ddot{\mathcal{R}}}{}{#2}{}{#1}}}
\newrobustcmd{\brevemathcalRz}[2][]{\ensuremath{\subp{\breve{\mathcal{R}}}{}{#2}{}{#1}}}
\newrobustcmd{\barmathcalRz}[2][]{\ensuremath{\subp{\bar{\mathcal{R}}}{}{#2}{}{#1}}}
\newrobustcmd{\vecmathcalRz}[2][]{\ensuremath{\subp{\vec{\mathcal{R}}}{}{#2}{}{#1}}}
\newrobustcmd{\bmmathcalRz}[2][]{\ensuremath{\subp{\bm{\mathcal{R}}}{}{#2}{}{#1}}}
\newrobustcmd{\hatbmmathcalRz}[2][]{\ensuremath{\subp{\hat{\bm{\mathcal{R}}}}{}{#2}{}{#1}}}
\newrobustcmd{\widehatbmmathcalRz}[2][]{\ensuremath{\subp{\widehat{\bm{\mathcal{R}}}}{}{#2}{}{#1}}}
\newrobustcmd{\checkbmmathcalRz}[2][]{\ensuremath{\subp{\check{\bm{\mathcal{R}}}}{}{#2}{}{#1}}}
\newrobustcmd{\tildebmmathcalRz}[2][]{\ensuremath{\subp{\tilde{\bm{\mathcal{R}}}}{}{#2}{}{#1}}}
\newrobustcmd{\widetildebmmathcalRz}[2][]{\ensuremath{\subp{\widetilde{\bm{\mathcal{R}}}}{}{#2}{}{#1}}}
\newrobustcmd{\acutebmmathcalRz}[2][]{\ensuremath{\subp{\acute{\bm{\mathcal{R}}}}{}{#2}{}{#1}}}
\newrobustcmd{\gravebmmathcalRz}[2][]{\ensuremath{\subp{\grave{\bm{\mathcal{R}}}}{}{#2}{}{#1}}}
\newrobustcmd{\dotbmmathcalRz}[2][]{\ensuremath{\subp{\dot{\bm{\mathcal{R}}}}{}{#2}{}{#1}}}
\newrobustcmd{\ddotbmmathcalRz}[2][]{\ensuremath{\subp{\ddot{\bm{\mathcal{R}}}}{}{#2}{}{#1}}}
\newrobustcmd{\brevebmmathcalRz}[2][]{\ensuremath{\subp{\breve{\bm{\mathcal{R}}}}{}{#2}{}{#1}}}
\newrobustcmd{\barbmmathcalRz}[2][]{\ensuremath{\subp{\bar{\bm{\mathcal{R}}}}{}{#2}{}{#1}}}
\newrobustcmd{\vecbmmathcalRz}[2][]{\ensuremath{\subp{\vec{\bm{\mathcal{R}}}}{}{#2}{}{#1}}}
\newrobustcmd{\mathcalSz}[2][]{\ensuremath{\subp{\mathcal{S}}{}{#2}{}{#1}}}
\newrobustcmd{\hatmathcalSz}[2][]{\ensuremath{\subp{\hat{\mathcal{S}}}{}{#2}{}{#1}}}
\newrobustcmd{\widehatmathcalSz}[2][]{\ensuremath{\subp{\widehat{\mathcal{S}}}{}{#2}{}{#1}}}
\newrobustcmd{\checkmathcalSz}[2][]{\ensuremath{\subp{\check{\mathcal{S}}}{}{#2}{}{#1}}}
\newrobustcmd{\tildemathcalSz}[2][]{\ensuremath{\subp{\tilde{\mathcal{S}}}{}{#2}{}{#1}}}
\newrobustcmd{\widetildemathcalSz}[2][]{\ensuremath{\subp{\widetilde{\mathcal{S}}}{}{#2}{}{#1}}}
\newrobustcmd{\acutemathcalSz}[2][]{\ensuremath{\subp{\acute{\mathcal{S}}}{}{#2}{}{#1}}}
\newrobustcmd{\gravemathcalSz}[2][]{\ensuremath{\subp{\grave{\mathcal{S}}}{}{#2}{}{#1}}}
\newrobustcmd{\dotmathcalSz}[2][]{\ensuremath{\subp{\dot{\mathcal{S}}}{}{#2}{}{#1}}}
\newrobustcmd{\ddotmathcalSz}[2][]{\ensuremath{\subp{\ddot{\mathcal{S}}}{}{#2}{}{#1}}}
\newrobustcmd{\brevemathcalSz}[2][]{\ensuremath{\subp{\breve{\mathcal{S}}}{}{#2}{}{#1}}}
\newrobustcmd{\barmathcalSz}[2][]{\ensuremath{\subp{\bar{\mathcal{S}}}{}{#2}{}{#1}}}
\newrobustcmd{\vecmathcalSz}[2][]{\ensuremath{\subp{\vec{\mathcal{S}}}{}{#2}{}{#1}}}
\newrobustcmd{\bmmathcalSz}[2][]{\ensuremath{\subp{\bm{\mathcal{S}}}{}{#2}{}{#1}}}
\newrobustcmd{\hatbmmathcalSz}[2][]{\ensuremath{\subp{\hat{\bm{\mathcal{S}}}}{}{#2}{}{#1}}}
\newrobustcmd{\widehatbmmathcalSz}[2][]{\ensuremath{\subp{\widehat{\bm{\mathcal{S}}}}{}{#2}{}{#1}}}
\newrobustcmd{\checkbmmathcalSz}[2][]{\ensuremath{\subp{\check{\bm{\mathcal{S}}}}{}{#2}{}{#1}}}
\newrobustcmd{\tildebmmathcalSz}[2][]{\ensuremath{\subp{\tilde{\bm{\mathcal{S}}}}{}{#2}{}{#1}}}
\newrobustcmd{\widetildebmmathcalSz}[2][]{\ensuremath{\subp{\widetilde{\bm{\mathcal{S}}}}{}{#2}{}{#1}}}
\newrobustcmd{\acutebmmathcalSz}[2][]{\ensuremath{\subp{\acute{\bm{\mathcal{S}}}}{}{#2}{}{#1}}}
\newrobustcmd{\gravebmmathcalSz}[2][]{\ensuremath{\subp{\grave{\bm{\mathcal{S}}}}{}{#2}{}{#1}}}
\newrobustcmd{\dotbmmathcalSz}[2][]{\ensuremath{\subp{\dot{\bm{\mathcal{S}}}}{}{#2}{}{#1}}}
\newrobustcmd{\ddotbmmathcalSz}[2][]{\ensuremath{\subp{\ddot{\bm{\mathcal{S}}}}{}{#2}{}{#1}}}
\newrobustcmd{\brevebmmathcalSz}[2][]{\ensuremath{\subp{\breve{\bm{\mathcal{S}}}}{}{#2}{}{#1}}}
\newrobustcmd{\barbmmathcalSz}[2][]{\ensuremath{\subp{\bar{\bm{\mathcal{S}}}}{}{#2}{}{#1}}}
\newrobustcmd{\vecbmmathcalSz}[2][]{\ensuremath{\subp{\vec{\bm{\mathcal{S}}}}{}{#2}{}{#1}}}
\newrobustcmd{\mathcalTz}[2][]{\ensuremath{\subp{\mathcal{T}}{}{#2}{}{#1}}}
\newrobustcmd{\hatmathcalTz}[2][]{\ensuremath{\subp{\hat{\mathcal{T}}}{}{#2}{}{#1}}}
\newrobustcmd{\widehatmathcalTz}[2][]{\ensuremath{\subp{\widehat{\mathcal{T}}}{}{#2}{}{#1}}}
\newrobustcmd{\checkmathcalTz}[2][]{\ensuremath{\subp{\check{\mathcal{T}}}{}{#2}{}{#1}}}
\newrobustcmd{\tildemathcalTz}[2][]{\ensuremath{\subp{\tilde{\mathcal{T}}}{}{#2}{}{#1}}}
\newrobustcmd{\widetildemathcalTz}[2][]{\ensuremath{\subp{\widetilde{\mathcal{T}}}{}{#2}{}{#1}}}
\newrobustcmd{\acutemathcalTz}[2][]{\ensuremath{\subp{\acute{\mathcal{T}}}{}{#2}{}{#1}}}
\newrobustcmd{\gravemathcalTz}[2][]{\ensuremath{\subp{\grave{\mathcal{T}}}{}{#2}{}{#1}}}
\newrobustcmd{\dotmathcalTz}[2][]{\ensuremath{\subp{\dot{\mathcal{T}}}{}{#2}{}{#1}}}
\newrobustcmd{\ddotmathcalTz}[2][]{\ensuremath{\subp{\ddot{\mathcal{T}}}{}{#2}{}{#1}}}
\newrobustcmd{\brevemathcalTz}[2][]{\ensuremath{\subp{\breve{\mathcal{T}}}{}{#2}{}{#1}}}
\newrobustcmd{\barmathcalTz}[2][]{\ensuremath{\subp{\bar{\mathcal{T}}}{}{#2}{}{#1}}}
\newrobustcmd{\vecmathcalTz}[2][]{\ensuremath{\subp{\vec{\mathcal{T}}}{}{#2}{}{#1}}}
\newrobustcmd{\bmmathcalTz}[2][]{\ensuremath{\subp{\bm{\mathcal{T}}}{}{#2}{}{#1}}}
\newrobustcmd{\hatbmmathcalTz}[2][]{\ensuremath{\subp{\hat{\bm{\mathcal{T}}}}{}{#2}{}{#1}}}
\newrobustcmd{\widehatbmmathcalTz}[2][]{\ensuremath{\subp{\widehat{\bm{\mathcal{T}}}}{}{#2}{}{#1}}}
\newrobustcmd{\checkbmmathcalTz}[2][]{\ensuremath{\subp{\check{\bm{\mathcal{T}}}}{}{#2}{}{#1}}}
\newrobustcmd{\tildebmmathcalTz}[2][]{\ensuremath{\subp{\tilde{\bm{\mathcal{T}}}}{}{#2}{}{#1}}}
\newrobustcmd{\widetildebmmathcalTz}[2][]{\ensuremath{\subp{\widetilde{\bm{\mathcal{T}}}}{}{#2}{}{#1}}}
\newrobustcmd{\acutebmmathcalTz}[2][]{\ensuremath{\subp{\acute{\bm{\mathcal{T}}}}{}{#2}{}{#1}}}
\newrobustcmd{\gravebmmathcalTz}[2][]{\ensuremath{\subp{\grave{\bm{\mathcal{T}}}}{}{#2}{}{#1}}}
\newrobustcmd{\dotbmmathcalTz}[2][]{\ensuremath{\subp{\dot{\bm{\mathcal{T}}}}{}{#2}{}{#1}}}
\newrobustcmd{\ddotbmmathcalTz}[2][]{\ensuremath{\subp{\ddot{\bm{\mathcal{T}}}}{}{#2}{}{#1}}}
\newrobustcmd{\brevebmmathcalTz}[2][]{\ensuremath{\subp{\breve{\bm{\mathcal{T}}}}{}{#2}{}{#1}}}
\newrobustcmd{\barbmmathcalTz}[2][]{\ensuremath{\subp{\bar{\bm{\mathcal{T}}}}{}{#2}{}{#1}}}
\newrobustcmd{\vecbmmathcalTz}[2][]{\ensuremath{\subp{\vec{\bm{\mathcal{T}}}}{}{#2}{}{#1}}}
\newrobustcmd{\mathcalUz}[2][]{\ensuremath{\subp{\mathcal{U}}{}{#2}{}{#1}}}
\newrobustcmd{\hatmathcalUz}[2][]{\ensuremath{\subp{\hat{\mathcal{U}}}{}{#2}{}{#1}}}
\newrobustcmd{\widehatmathcalUz}[2][]{\ensuremath{\subp{\widehat{\mathcal{U}}}{}{#2}{}{#1}}}
\newrobustcmd{\checkmathcalUz}[2][]{\ensuremath{\subp{\check{\mathcal{U}}}{}{#2}{}{#1}}}
\newrobustcmd{\tildemathcalUz}[2][]{\ensuremath{\subp{\tilde{\mathcal{U}}}{}{#2}{}{#1}}}
\newrobustcmd{\widetildemathcalUz}[2][]{\ensuremath{\subp{\widetilde{\mathcal{U}}}{}{#2}{}{#1}}}
\newrobustcmd{\acutemathcalUz}[2][]{\ensuremath{\subp{\acute{\mathcal{U}}}{}{#2}{}{#1}}}
\newrobustcmd{\gravemathcalUz}[2][]{\ensuremath{\subp{\grave{\mathcal{U}}}{}{#2}{}{#1}}}
\newrobustcmd{\dotmathcalUz}[2][]{\ensuremath{\subp{\dot{\mathcal{U}}}{}{#2}{}{#1}}}
\newrobustcmd{\ddotmathcalUz}[2][]{\ensuremath{\subp{\ddot{\mathcal{U}}}{}{#2}{}{#1}}}
\newrobustcmd{\brevemathcalUz}[2][]{\ensuremath{\subp{\breve{\mathcal{U}}}{}{#2}{}{#1}}}
\newrobustcmd{\barmathcalUz}[2][]{\ensuremath{\subp{\bar{\mathcal{U}}}{}{#2}{}{#1}}}
\newrobustcmd{\vecmathcalUz}[2][]{\ensuremath{\subp{\vec{\mathcal{U}}}{}{#2}{}{#1}}}
\newrobustcmd{\bmmathcalUz}[2][]{\ensuremath{\subp{\bm{\mathcal{U}}}{}{#2}{}{#1}}}
\newrobustcmd{\hatbmmathcalUz}[2][]{\ensuremath{\subp{\hat{\bm{\mathcal{U}}}}{}{#2}{}{#1}}}
\newrobustcmd{\widehatbmmathcalUz}[2][]{\ensuremath{\subp{\widehat{\bm{\mathcal{U}}}}{}{#2}{}{#1}}}
\newrobustcmd{\checkbmmathcalUz}[2][]{\ensuremath{\subp{\check{\bm{\mathcal{U}}}}{}{#2}{}{#1}}}
\newrobustcmd{\tildebmmathcalUz}[2][]{\ensuremath{\subp{\tilde{\bm{\mathcal{U}}}}{}{#2}{}{#1}}}
\newrobustcmd{\widetildebmmathcalUz}[2][]{\ensuremath{\subp{\widetilde{\bm{\mathcal{U}}}}{}{#2}{}{#1}}}
\newrobustcmd{\acutebmmathcalUz}[2][]{\ensuremath{\subp{\acute{\bm{\mathcal{U}}}}{}{#2}{}{#1}}}
\newrobustcmd{\gravebmmathcalUz}[2][]{\ensuremath{\subp{\grave{\bm{\mathcal{U}}}}{}{#2}{}{#1}}}
\newrobustcmd{\dotbmmathcalUz}[2][]{\ensuremath{\subp{\dot{\bm{\mathcal{U}}}}{}{#2}{}{#1}}}
\newrobustcmd{\ddotbmmathcalUz}[2][]{\ensuremath{\subp{\ddot{\bm{\mathcal{U}}}}{}{#2}{}{#1}}}
\newrobustcmd{\brevebmmathcalUz}[2][]{\ensuremath{\subp{\breve{\bm{\mathcal{U}}}}{}{#2}{}{#1}}}
\newrobustcmd{\barbmmathcalUz}[2][]{\ensuremath{\subp{\bar{\bm{\mathcal{U}}}}{}{#2}{}{#1}}}
\newrobustcmd{\vecbmmathcalUz}[2][]{\ensuremath{\subp{\vec{\bm{\mathcal{U}}}}{}{#2}{}{#1}}}
\newrobustcmd{\mathcalVz}[2][]{\ensuremath{\subp{\mathcal{V}}{}{#2}{}{#1}}}
\newrobustcmd{\hatmathcalVz}[2][]{\ensuremath{\subp{\hat{\mathcal{V}}}{}{#2}{}{#1}}}
\newrobustcmd{\widehatmathcalVz}[2][]{\ensuremath{\subp{\widehat{\mathcal{V}}}{}{#2}{}{#1}}}
\newrobustcmd{\checkmathcalVz}[2][]{\ensuremath{\subp{\check{\mathcal{V}}}{}{#2}{}{#1}}}
\newrobustcmd{\tildemathcalVz}[2][]{\ensuremath{\subp{\tilde{\mathcal{V}}}{}{#2}{}{#1}}}
\newrobustcmd{\widetildemathcalVz}[2][]{\ensuremath{\subp{\widetilde{\mathcal{V}}}{}{#2}{}{#1}}}
\newrobustcmd{\acutemathcalVz}[2][]{\ensuremath{\subp{\acute{\mathcal{V}}}{}{#2}{}{#1}}}
\newrobustcmd{\gravemathcalVz}[2][]{\ensuremath{\subp{\grave{\mathcal{V}}}{}{#2}{}{#1}}}
\newrobustcmd{\dotmathcalVz}[2][]{\ensuremath{\subp{\dot{\mathcal{V}}}{}{#2}{}{#1}}}
\newrobustcmd{\ddotmathcalVz}[2][]{\ensuremath{\subp{\ddot{\mathcal{V}}}{}{#2}{}{#1}}}
\newrobustcmd{\brevemathcalVz}[2][]{\ensuremath{\subp{\breve{\mathcal{V}}}{}{#2}{}{#1}}}
\newrobustcmd{\barmathcalVz}[2][]{\ensuremath{\subp{\bar{\mathcal{V}}}{}{#2}{}{#1}}}
\newrobustcmd{\vecmathcalVz}[2][]{\ensuremath{\subp{\vec{\mathcal{V}}}{}{#2}{}{#1}}}
\newrobustcmd{\bmmathcalVz}[2][]{\ensuremath{\subp{\bm{\mathcal{V}}}{}{#2}{}{#1}}}
\newrobustcmd{\hatbmmathcalVz}[2][]{\ensuremath{\subp{\hat{\bm{\mathcal{V}}}}{}{#2}{}{#1}}}
\newrobustcmd{\widehatbmmathcalVz}[2][]{\ensuremath{\subp{\widehat{\bm{\mathcal{V}}}}{}{#2}{}{#1}}}
\newrobustcmd{\checkbmmathcalVz}[2][]{\ensuremath{\subp{\check{\bm{\mathcal{V}}}}{}{#2}{}{#1}}}
\newrobustcmd{\tildebmmathcalVz}[2][]{\ensuremath{\subp{\tilde{\bm{\mathcal{V}}}}{}{#2}{}{#1}}}
\newrobustcmd{\widetildebmmathcalVz}[2][]{\ensuremath{\subp{\widetilde{\bm{\mathcal{V}}}}{}{#2}{}{#1}}}
\newrobustcmd{\acutebmmathcalVz}[2][]{\ensuremath{\subp{\acute{\bm{\mathcal{V}}}}{}{#2}{}{#1}}}
\newrobustcmd{\gravebmmathcalVz}[2][]{\ensuremath{\subp{\grave{\bm{\mathcal{V}}}}{}{#2}{}{#1}}}
\newrobustcmd{\dotbmmathcalVz}[2][]{\ensuremath{\subp{\dot{\bm{\mathcal{V}}}}{}{#2}{}{#1}}}
\newrobustcmd{\ddotbmmathcalVz}[2][]{\ensuremath{\subp{\ddot{\bm{\mathcal{V}}}}{}{#2}{}{#1}}}
\newrobustcmd{\brevebmmathcalVz}[2][]{\ensuremath{\subp{\breve{\bm{\mathcal{V}}}}{}{#2}{}{#1}}}
\newrobustcmd{\barbmmathcalVz}[2][]{\ensuremath{\subp{\bar{\bm{\mathcal{V}}}}{}{#2}{}{#1}}}
\newrobustcmd{\vecbmmathcalVz}[2][]{\ensuremath{\subp{\vec{\bm{\mathcal{V}}}}{}{#2}{}{#1}}}
\newrobustcmd{\mathcalWz}[2][]{\ensuremath{\subp{\mathcal{W}}{}{#2}{}{#1}}}
\newrobustcmd{\hatmathcalWz}[2][]{\ensuremath{\subp{\hat{\mathcal{W}}}{}{#2}{}{#1}}}
\newrobustcmd{\widehatmathcalWz}[2][]{\ensuremath{\subp{\widehat{\mathcal{W}}}{}{#2}{}{#1}}}
\newrobustcmd{\checkmathcalWz}[2][]{\ensuremath{\subp{\check{\mathcal{W}}}{}{#2}{}{#1}}}
\newrobustcmd{\tildemathcalWz}[2][]{\ensuremath{\subp{\tilde{\mathcal{W}}}{}{#2}{}{#1}}}
\newrobustcmd{\widetildemathcalWz}[2][]{\ensuremath{\subp{\widetilde{\mathcal{W}}}{}{#2}{}{#1}}}
\newrobustcmd{\acutemathcalWz}[2][]{\ensuremath{\subp{\acute{\mathcal{W}}}{}{#2}{}{#1}}}
\newrobustcmd{\gravemathcalWz}[2][]{\ensuremath{\subp{\grave{\mathcal{W}}}{}{#2}{}{#1}}}
\newrobustcmd{\dotmathcalWz}[2][]{\ensuremath{\subp{\dot{\mathcal{W}}}{}{#2}{}{#1}}}
\newrobustcmd{\ddotmathcalWz}[2][]{\ensuremath{\subp{\ddot{\mathcal{W}}}{}{#2}{}{#1}}}
\newrobustcmd{\brevemathcalWz}[2][]{\ensuremath{\subp{\breve{\mathcal{W}}}{}{#2}{}{#1}}}
\newrobustcmd{\barmathcalWz}[2][]{\ensuremath{\subp{\bar{\mathcal{W}}}{}{#2}{}{#1}}}
\newrobustcmd{\vecmathcalWz}[2][]{\ensuremath{\subp{\vec{\mathcal{W}}}{}{#2}{}{#1}}}
\newrobustcmd{\bmmathcalWz}[2][]{\ensuremath{\subp{\bm{\mathcal{W}}}{}{#2}{}{#1}}}
\newrobustcmd{\hatbmmathcalWz}[2][]{\ensuremath{\subp{\hat{\bm{\mathcal{W}}}}{}{#2}{}{#1}}}
\newrobustcmd{\widehatbmmathcalWz}[2][]{\ensuremath{\subp{\widehat{\bm{\mathcal{W}}}}{}{#2}{}{#1}}}
\newrobustcmd{\checkbmmathcalWz}[2][]{\ensuremath{\subp{\check{\bm{\mathcal{W}}}}{}{#2}{}{#1}}}
\newrobustcmd{\tildebmmathcalWz}[2][]{\ensuremath{\subp{\tilde{\bm{\mathcal{W}}}}{}{#2}{}{#1}}}
\newrobustcmd{\widetildebmmathcalWz}[2][]{\ensuremath{\subp{\widetilde{\bm{\mathcal{W}}}}{}{#2}{}{#1}}}
\newrobustcmd{\acutebmmathcalWz}[2][]{\ensuremath{\subp{\acute{\bm{\mathcal{W}}}}{}{#2}{}{#1}}}
\newrobustcmd{\gravebmmathcalWz}[2][]{\ensuremath{\subp{\grave{\bm{\mathcal{W}}}}{}{#2}{}{#1}}}
\newrobustcmd{\dotbmmathcalWz}[2][]{\ensuremath{\subp{\dot{\bm{\mathcal{W}}}}{}{#2}{}{#1}}}
\newrobustcmd{\ddotbmmathcalWz}[2][]{\ensuremath{\subp{\ddot{\bm{\mathcal{W}}}}{}{#2}{}{#1}}}
\newrobustcmd{\brevebmmathcalWz}[2][]{\ensuremath{\subp{\breve{\bm{\mathcal{W}}}}{}{#2}{}{#1}}}
\newrobustcmd{\barbmmathcalWz}[2][]{\ensuremath{\subp{\bar{\bm{\mathcal{W}}}}{}{#2}{}{#1}}}
\newrobustcmd{\vecbmmathcalWz}[2][]{\ensuremath{\subp{\vec{\bm{\mathcal{W}}}}{}{#2}{}{#1}}}
\newrobustcmd{\mathcalXz}[2][]{\ensuremath{\subp{\mathcal{X}}{}{#2}{}{#1}}}
\newrobustcmd{\hatmathcalXz}[2][]{\ensuremath{\subp{\hat{\mathcal{X}}}{}{#2}{}{#1}}}
\newrobustcmd{\widehatmathcalXz}[2][]{\ensuremath{\subp{\widehat{\mathcal{X}}}{}{#2}{}{#1}}}
\newrobustcmd{\checkmathcalXz}[2][]{\ensuremath{\subp{\check{\mathcal{X}}}{}{#2}{}{#1}}}
\newrobustcmd{\tildemathcalXz}[2][]{\ensuremath{\subp{\tilde{\mathcal{X}}}{}{#2}{}{#1}}}
\newrobustcmd{\widetildemathcalXz}[2][]{\ensuremath{\subp{\widetilde{\mathcal{X}}}{}{#2}{}{#1}}}
\newrobustcmd{\acutemathcalXz}[2][]{\ensuremath{\subp{\acute{\mathcal{X}}}{}{#2}{}{#1}}}
\newrobustcmd{\gravemathcalXz}[2][]{\ensuremath{\subp{\grave{\mathcal{X}}}{}{#2}{}{#1}}}
\newrobustcmd{\dotmathcalXz}[2][]{\ensuremath{\subp{\dot{\mathcal{X}}}{}{#2}{}{#1}}}
\newrobustcmd{\ddotmathcalXz}[2][]{\ensuremath{\subp{\ddot{\mathcal{X}}}{}{#2}{}{#1}}}
\newrobustcmd{\brevemathcalXz}[2][]{\ensuremath{\subp{\breve{\mathcal{X}}}{}{#2}{}{#1}}}
\newrobustcmd{\barmathcalXz}[2][]{\ensuremath{\subp{\bar{\mathcal{X}}}{}{#2}{}{#1}}}
\newrobustcmd{\vecmathcalXz}[2][]{\ensuremath{\subp{\vec{\mathcal{X}}}{}{#2}{}{#1}}}
\newrobustcmd{\bmmathcalXz}[2][]{\ensuremath{\subp{\bm{\mathcal{X}}}{}{#2}{}{#1}}}
\newrobustcmd{\hatbmmathcalXz}[2][]{\ensuremath{\subp{\hat{\bm{\mathcal{X}}}}{}{#2}{}{#1}}}
\newrobustcmd{\widehatbmmathcalXz}[2][]{\ensuremath{\subp{\widehat{\bm{\mathcal{X}}}}{}{#2}{}{#1}}}
\newrobustcmd{\checkbmmathcalXz}[2][]{\ensuremath{\subp{\check{\bm{\mathcal{X}}}}{}{#2}{}{#1}}}
\newrobustcmd{\tildebmmathcalXz}[2][]{\ensuremath{\subp{\tilde{\bm{\mathcal{X}}}}{}{#2}{}{#1}}}
\newrobustcmd{\widetildebmmathcalXz}[2][]{\ensuremath{\subp{\widetilde{\bm{\mathcal{X}}}}{}{#2}{}{#1}}}
\newrobustcmd{\acutebmmathcalXz}[2][]{\ensuremath{\subp{\acute{\bm{\mathcal{X}}}}{}{#2}{}{#1}}}
\newrobustcmd{\gravebmmathcalXz}[2][]{\ensuremath{\subp{\grave{\bm{\mathcal{X}}}}{}{#2}{}{#1}}}
\newrobustcmd{\dotbmmathcalXz}[2][]{\ensuremath{\subp{\dot{\bm{\mathcal{X}}}}{}{#2}{}{#1}}}
\newrobustcmd{\ddotbmmathcalXz}[2][]{\ensuremath{\subp{\ddot{\bm{\mathcal{X}}}}{}{#2}{}{#1}}}
\newrobustcmd{\brevebmmathcalXz}[2][]{\ensuremath{\subp{\breve{\bm{\mathcal{X}}}}{}{#2}{}{#1}}}
\newrobustcmd{\barbmmathcalXz}[2][]{\ensuremath{\subp{\bar{\bm{\mathcal{X}}}}{}{#2}{}{#1}}}
\newrobustcmd{\vecbmmathcalXz}[2][]{\ensuremath{\subp{\vec{\bm{\mathcal{X}}}}{}{#2}{}{#1}}}
\newrobustcmd{\mathcalYz}[2][]{\ensuremath{\subp{\mathcal{Y}}{}{#2}{}{#1}}}
\newrobustcmd{\hatmathcalYz}[2][]{\ensuremath{\subp{\hat{\mathcal{Y}}}{}{#2}{}{#1}}}
\newrobustcmd{\widehatmathcalYz}[2][]{\ensuremath{\subp{\widehat{\mathcal{Y}}}{}{#2}{}{#1}}}
\newrobustcmd{\checkmathcalYz}[2][]{\ensuremath{\subp{\check{\mathcal{Y}}}{}{#2}{}{#1}}}
\newrobustcmd{\tildemathcalYz}[2][]{\ensuremath{\subp{\tilde{\mathcal{Y}}}{}{#2}{}{#1}}}
\newrobustcmd{\widetildemathcalYz}[2][]{\ensuremath{\subp{\widetilde{\mathcal{Y}}}{}{#2}{}{#1}}}
\newrobustcmd{\acutemathcalYz}[2][]{\ensuremath{\subp{\acute{\mathcal{Y}}}{}{#2}{}{#1}}}
\newrobustcmd{\gravemathcalYz}[2][]{\ensuremath{\subp{\grave{\mathcal{Y}}}{}{#2}{}{#1}}}
\newrobustcmd{\dotmathcalYz}[2][]{\ensuremath{\subp{\dot{\mathcal{Y}}}{}{#2}{}{#1}}}
\newrobustcmd{\ddotmathcalYz}[2][]{\ensuremath{\subp{\ddot{\mathcal{Y}}}{}{#2}{}{#1}}}
\newrobustcmd{\brevemathcalYz}[2][]{\ensuremath{\subp{\breve{\mathcal{Y}}}{}{#2}{}{#1}}}
\newrobustcmd{\barmathcalYz}[2][]{\ensuremath{\subp{\bar{\mathcal{Y}}}{}{#2}{}{#1}}}
\newrobustcmd{\vecmathcalYz}[2][]{\ensuremath{\subp{\vec{\mathcal{Y}}}{}{#2}{}{#1}}}
\newrobustcmd{\bmmathcalYz}[2][]{\ensuremath{\subp{\bm{\mathcal{Y}}}{}{#2}{}{#1}}}
\newrobustcmd{\hatbmmathcalYz}[2][]{\ensuremath{\subp{\hat{\bm{\mathcal{Y}}}}{}{#2}{}{#1}}}
\newrobustcmd{\widehatbmmathcalYz}[2][]{\ensuremath{\subp{\widehat{\bm{\mathcal{Y}}}}{}{#2}{}{#1}}}
\newrobustcmd{\checkbmmathcalYz}[2][]{\ensuremath{\subp{\check{\bm{\mathcal{Y}}}}{}{#2}{}{#1}}}
\newrobustcmd{\tildebmmathcalYz}[2][]{\ensuremath{\subp{\tilde{\bm{\mathcal{Y}}}}{}{#2}{}{#1}}}
\newrobustcmd{\widetildebmmathcalYz}[2][]{\ensuremath{\subp{\widetilde{\bm{\mathcal{Y}}}}{}{#2}{}{#1}}}
\newrobustcmd{\acutebmmathcalYz}[2][]{\ensuremath{\subp{\acute{\bm{\mathcal{Y}}}}{}{#2}{}{#1}}}
\newrobustcmd{\gravebmmathcalYz}[2][]{\ensuremath{\subp{\grave{\bm{\mathcal{Y}}}}{}{#2}{}{#1}}}
\newrobustcmd{\dotbmmathcalYz}[2][]{\ensuremath{\subp{\dot{\bm{\mathcal{Y}}}}{}{#2}{}{#1}}}
\newrobustcmd{\ddotbmmathcalYz}[2][]{\ensuremath{\subp{\ddot{\bm{\mathcal{Y}}}}{}{#2}{}{#1}}}
\newrobustcmd{\brevebmmathcalYz}[2][]{\ensuremath{\subp{\breve{\bm{\mathcal{Y}}}}{}{#2}{}{#1}}}
\newrobustcmd{\barbmmathcalYz}[2][]{\ensuremath{\subp{\bar{\bm{\mathcal{Y}}}}{}{#2}{}{#1}}}
\newrobustcmd{\vecbmmathcalYz}[2][]{\ensuremath{\subp{\vec{\bm{\mathcal{Y}}}}{}{#2}{}{#1}}}
\newrobustcmd{\mathcalZz}[2][]{\ensuremath{\subp{\mathcal{Z}}{}{#2}{}{#1}}}
\newrobustcmd{\hatmathcalZz}[2][]{\ensuremath{\subp{\hat{\mathcal{Z}}}{}{#2}{}{#1}}}
\newrobustcmd{\widehatmathcalZz}[2][]{\ensuremath{\subp{\widehat{\mathcal{Z}}}{}{#2}{}{#1}}}
\newrobustcmd{\checkmathcalZz}[2][]{\ensuremath{\subp{\check{\mathcal{Z}}}{}{#2}{}{#1}}}
\newrobustcmd{\tildemathcalZz}[2][]{\ensuremath{\subp{\tilde{\mathcal{Z}}}{}{#2}{}{#1}}}
\newrobustcmd{\widetildemathcalZz}[2][]{\ensuremath{\subp{\widetilde{\mathcal{Z}}}{}{#2}{}{#1}}}
\newrobustcmd{\acutemathcalZz}[2][]{\ensuremath{\subp{\acute{\mathcal{Z}}}{}{#2}{}{#1}}}
\newrobustcmd{\gravemathcalZz}[2][]{\ensuremath{\subp{\grave{\mathcal{Z}}}{}{#2}{}{#1}}}
\newrobustcmd{\dotmathcalZz}[2][]{\ensuremath{\subp{\dot{\mathcal{Z}}}{}{#2}{}{#1}}}
\newrobustcmd{\ddotmathcalZz}[2][]{\ensuremath{\subp{\ddot{\mathcal{Z}}}{}{#2}{}{#1}}}
\newrobustcmd{\brevemathcalZz}[2][]{\ensuremath{\subp{\breve{\mathcal{Z}}}{}{#2}{}{#1}}}
\newrobustcmd{\barmathcalZz}[2][]{\ensuremath{\subp{\bar{\mathcal{Z}}}{}{#2}{}{#1}}}
\newrobustcmd{\vecmathcalZz}[2][]{\ensuremath{\subp{\vec{\mathcal{Z}}}{}{#2}{}{#1}}}
\newrobustcmd{\bmmathcalZz}[2][]{\ensuremath{\subp{\bm{\mathcal{Z}}}{}{#2}{}{#1}}}
\newrobustcmd{\hatbmmathcalZz}[2][]{\ensuremath{\subp{\hat{\bm{\mathcal{Z}}}}{}{#2}{}{#1}}}
\newrobustcmd{\widehatbmmathcalZz}[2][]{\ensuremath{\subp{\widehat{\bm{\mathcal{Z}}}}{}{#2}{}{#1}}}
\newrobustcmd{\checkbmmathcalZz}[2][]{\ensuremath{\subp{\check{\bm{\mathcal{Z}}}}{}{#2}{}{#1}}}
\newrobustcmd{\tildebmmathcalZz}[2][]{\ensuremath{\subp{\tilde{\bm{\mathcal{Z}}}}{}{#2}{}{#1}}}
\newrobustcmd{\widetildebmmathcalZz}[2][]{\ensuremath{\subp{\widetilde{\bm{\mathcal{Z}}}}{}{#2}{}{#1}}}
\newrobustcmd{\acutebmmathcalZz}[2][]{\ensuremath{\subp{\acute{\bm{\mathcal{Z}}}}{}{#2}{}{#1}}}
\newrobustcmd{\gravebmmathcalZz}[2][]{\ensuremath{\subp{\grave{\bm{\mathcal{Z}}}}{}{#2}{}{#1}}}
\newrobustcmd{\dotbmmathcalZz}[2][]{\ensuremath{\subp{\dot{\bm{\mathcal{Z}}}}{}{#2}{}{#1}}}
\newrobustcmd{\ddotbmmathcalZz}[2][]{\ensuremath{\subp{\ddot{\bm{\mathcal{Z}}}}{}{#2}{}{#1}}}
\newrobustcmd{\brevebmmathcalZz}[2][]{\ensuremath{\subp{\breve{\bm{\mathcal{Z}}}}{}{#2}{}{#1}}}
\newrobustcmd{\barbmmathcalZz}[2][]{\ensuremath{\subp{\bar{\bm{\mathcal{Z}}}}{}{#2}{}{#1}}}
\newrobustcmd{\vecbmmathcalZz}[2][]{\ensuremath{\subp{\vec{\bm{\mathcal{Z}}}}{}{#2}{}{#1}}}

\newrobustcmd{\ellz}[2][]{\ensuremath{\subp{\ell}{}{#2}{}{#1}}}

\newcommand{\TSR}[3]{   \mbox{\ensuremath{\subp{   \left\{ #1 \right\}}{}{#2}{}{#3}}}
}

\newcommand{\lgcor}[2][]{
  \rhoz{#2
    \IfStrEq{#1}{\empty}{\!}{|#1}}}

\newcommand{\hatlgcor}[2][]{
  \widehatrhoz{#2
    \IfStrEq{#1}{\empty}{\!}{|#1}}}

\newcommand{\lgacr}[3][]{
  \rhoz{#2
    \IfStrEq{#1}{\empty}{\!}{|#1}}(#3)}

\newcommand{\lgccr}[4][]{
  \rhoz{#2:#3
    \IfStrEq{#1}{\empty}{\!}{|#1}}(#4)}

\newcommand{\hatlgacr}[3][]{
  \widehatrhoz{#2
    \IfStrEq{#1}{\empty}{\!}{|#1}}(#3)}
\newcommand{\hatlgccr}[4][]{
  \widehatrhoz{#2:#3
    \IfStrEq{#1}{\empty}{\!}{|#1}}(#4)}

\newcommand{\hatlgacrb}[4][]{
  \widehatrhoz{\!#2
    \IfStrEq{#1}{\empty}{\!}{|#1}
  }(#3|\scalebox{.7}{$#4$})}

\newcommand{\hatlgccrb}[5][]{
  \widehatrhoz{#2:#3
    \IfStrEq{#1}{\empty}{\!}{|#1}
  }(#4|\scalebox{.7}{$#5$})}

\newcommand{\hatlgthetab}[4][]{
  \widehatbmthetaz{\!#2
    \IfStrEq{#1}{\empty}{\!}{|#1}
  }(#3|\scalebox{.7}{$#4$})}

\newcommand{\lgsd}[3][]{
  \fz{\!#2
    \IfStrEq{#1}{\empty}{\!}{|#1}}(#3)}

\newcommand{\lgcsd}[4][]{
  \fz{#2:#3
    \IfStrEq{#1}{\empty}{\!}{|#1}}(#4)}

\newcommand{\lgcsdCo}[4][]{
  \cz{#2:#3
    \IfStrEq{#1}{\empty}{\!}{|#1}}(#4)}
\newcommand{\lgcsdQuad}[4][]{
  \qz{#2:#3
    \IfStrEq{#1}{\empty}{\!}{|#1}}(#4)}
\newcommand{\lgcsdAmplitude}[4][]{
  \alphaz{#2:#3
    \IfStrEq{#1}{\empty}{\!}{|#1}}(#4)}
\newcommand{\lgcsdPhase}[4][]{
  \phiz{#2:#3
    \IfStrEq{#1}{\empty}{\!}{|#1}}(#4)}

\newcommand{\lgcsdCoM}[5][]{
  \cz[#5]{#2:#3
    \IfStrEq{#1}{\empty}{\!}{|#1}}(#4)}
\newcommand{\lgcsdQuadM}[5][]{
  \qz[#5]{#2:#3
    \IfStrEq{#1}{\empty}{\!}{|#1}}(#4)}
\newcommand{\lgcsdAmplitudeM}[5][]{
  \alphaz[#5]{#2:#3
    \IfStrEq{#1}{\empty}{\!}{|#1}}(#4)}
\newcommand{\lgcsdPhaseM}[5][]{
  \phiz[#5]{#2:#3
    \IfStrEq{#1}{\empty}{\!}{|#1}}(#4)}

\newcommand{\hatlgcsdCoM}[5][]{
  \widehatcz[\ #5]{#2:#3
    \IfStrEq{#1}{\empty}{\!}{|#1}}(#4)}
\newcommand{\hatlgcsdQuadM}[5][]{
  \widehatqz[\ #5]{#2:#3
    \IfStrEq{#1}{\empty}{\!}{|#1}}(#4)}
\newcommand{\hatlgcsdAmplitudeM}[5][]{
  \widehatalphaz[\ #5]{#2:#3
    \IfStrEq{#1}{\empty}{\!}{|#1}}(#4)}
\newcommand{\hatlgcsdPhaseM}[5][]{
  \widehatphiz[\ #5]{#2:#3
    \IfStrEq{#1}{\empty}{\!}{|#1}}(#4)}

\newcommand{\lgcsdSQ}[4][]{
  \mathcalKz[asc]{#2:#3
    \IfStrEq{#1}{\empty}{\!}{|#1}}(#4)}

\newcommand{\lgsdM}[4][]{
  \fz[#4]{\!#2
    \IfStrEq{#1}{\empty}{\!}{|#1}}(#3)}
\newcommand{\lgcsdM}[5][]{
  \fz[#5]{\!#2:#3
    \IfStrEq{#1}{\empty}{\!}{|#1}}(#4)}

\newcommand{\hatlgsd}[3][]{
  \widehatfz{\!#2
    \IfStrEq{#1}{\empty}{\!}{|#1}}(#3)}

\newcommand{\hatlgcsd}[4][]{
  \widehatfz{\!#2:#3
    \IfStrEq{#1}{\empty}{\!}{|#1}}(#4)}

\newcommand{\hatlgsdM}[4][]{
  \widehatfz[#4]{\!#2
    \IfStrEq{#1}{\empty}{\!}{|#1}}(#3)}
\newcommand{\hatlgcsdM}[5][]{
  \widehatfz[#5]{\!#2:#3
    \IfStrEq{#1}{\empty}{\!}{|#1}}(#4)}

\newcommand{\lgsdRE}[3][]{
  \cz{#2
    \IfStrEq{#1}{\empty}{\!}{|#1}}(#3)}
\newcommand{\hatlgsdREM}[4][]{
  \widehatcz[\ #4]{#2
    \IfStrEq{#1}{\empty}{\!}{|#1}}(#3)}

\newcommand{\lgsdIM}[3][]{
  \qz{#2
    \IfStrEq{#1}{\empty}{\!}{|#1}}(#3)}
\newcommand{\hatlgsdIMM}[4][]{
  \widehatqz[\ #4]{#2
    \IfStrEq{#1}{\empty}{\!}{|#1}}(#3)}

\newcommand{\sigmaYhtRange}[2]{
  \ensuremath{\subp{\mathcal{F}}{}{#1}{}{#2}} }

\newcommand{\overbar}[1]{\mkern 2.5mu\overline{\mkern-2.5mu#1\mkern-2.5mu}\mkern 2.5mu}

\newcommand{\Ubar}[1]{\mkern 2.5mu\underline{\mkern-2.5mu#1\mkern-2.5mu}\mkern 2.5mu}

\newcommand{\OUbar}[1]{\ensuremath{\overbar{\Ubar{#1}}}}

\newcommand{\Mblock}[1]{
  \overbar{#1}}

\newcommand{\MMAS}[8]{
  \ensuremath{\subp{#1}{}{                 #2
                \ifstrempty{#4}{#3}{\Mblock{#3}}
                        #5#6}{}{}{}
            \IfStrEq{#8}{}{}{#7\left(#8\right)} 
          }}

\newcommand{\LGpsymbol}{v}
\newcommand{\LGp}{\ensuremath{\bm{\LGpsymbol}}}
\newcommand{\LGpd}{\ensuremath{\breve{\LGp}}}
\newcommand{\LGpi}[2][]{
  \subp{\LGpsymbol}{}{#2}{}{#1}
}
\newcommand{\LGpoint}{
  \ensuremath{\parenR{\LGpi{1},\LGpi{2}}}}
\newcommand{\LGpointd}{
  \ensuremath{\parenR{\LGpi{2},\LGpi{1}}}}
\newcommand{\LGpdiagonal}{
  \ensuremath{\LGpi{1} = \LGpi{2}}}
\newcommand{\LGpnotdiagonal}{
  \ensuremath{\LGpi{1} \neq \LGpi{2}}}

\newcommand{\LGpz}[2][]{
  \subp{\LGp}{}{#2}{}{#1}
}

\newrobustcmd{\Yht}[2]{
  \MMAS{\bm{Y}}{}{#1}{}{:}{#2}{}{} }
\newrobustcmd{\Yhtc}[3]{
  \MMAS{\bm{Y}}{}{#1:#2}{}{:}{#3}{}{} }
\newrobustcmd{\YMt}[2]{
  \MMAS{\bm{Y}}{}{#1}{-}{:}{#2}{}{} }
\newrobustcmd{\yh}[1]{
  \MMAS{\bm{y}}{}{#1}{}{}{}{}{} }
\newrobustcmd{\yht}[2]{
  \MMAS{\bm{y}}{}{#1}{}{:}{#2}{}{} }
\newrobustcmd{\yM}[1]{
  \MMAS{\bm{y}}{}{#1}{-}{}{}{}{} }
\newrobustcmd{\yMt}[2]{
  \MMAS{\bm{y}}{}{#1}{-}{:}{#2}{}{} }

\newrobustcmd{\yhc}[2]{
  \MMAS{\bm{y}}{}{#1:#2}{}{}{}{}{} }

\newrobustcmd{\dyh}[1]{
  \MMAS{\d{\bm{y}}}{}{#1}{}{}{}{}{} }
\newrobustcmd{\dyhc}[2]{
  \MMAS{\d{\bm{y}}}{}{#1:#2}{}{}{}{}{} }
\newrobustcmd{\dyM}[1]{
  \MMAS{\d{\bm{y}}}{}{#1}{-}{}{}{}{} }
\newrobustcmd{\dyMh}[2]{
  \MMAS{\d{\bm{y}}}{}{#1}{-}{/}{#2}{}{} }
\newrobustcmd{\gh}[2][]{
  \MMAS{g}{}{#2}{}{}{}{\!}{#1} }
\newrobustcmd{\ght}[3][]{
  \MMAS{g}{}{#2}{}{:}{#3}{\!}{#1} }
\newrobustcmd{\gM}[2][]{
  \MMAS{g}{}{#2}{-}{}{}{\!}{#1} }
\newrobustcmd{\gMt}[3][]{
  \MMAS{g}{}{#2}{-}{:}{#3}{\!}{#1} }

\newrobustcmd{\ghc}[3][]{
  \MMAS{g}{}{#2:#3}{}{}{}{\!}{#1} }

\newrobustcmd{\Xhit}[3]{
  \MMAS{X}{}{#1#2}{}{:}{#3}{}{} }
\newrobustcmd{\Xat}[2][]{
  \MMAS{X}{}{#2}{}{}{}{\!}{\bm{#1}} }
\newrobustcmd{\Xt}[1]{
  \MMAS{\bm{X}}{}{#1}{}{}{}{\!}{} }
\newrobustcmd{\Xht}[2]{
  \MMAS{\bm{X}}{}{#1}{}{:}{#2}{\!}{} }

\newcommand{\Lp}[2][]{\ensuremath{L_{#1}^{{}^{#2}}}}
\newcommand{\LpM}[2]{\ensuremath{L_{\overset{\overbar{#1}}{}}^{{}^{#2}}}}
\newcommand{\Lph}[2]{\ensuremath{L_{\overset{#1}{}}^{{}^{#2}}}}

\newcommand{\absp}[2][]{  \ensuremath{\subp{\left|#2\right|}{}{}{}{#1}} }

\newcommand{\etahi}[3][]{
  \MMAS{\eta}{}{#2#3}{}{}{}{\!}{#1} }
\newcommand{\tildeetahi}[3][]{
  \MMAS{\widetilde{\eta}}{}{#2#3}{}{}{}{\!}{#1} }

\newcommand{\zetah}[2][]{
  \MMAS{\bm{\zeta}}{}{#2}{}{}{}{\!}{#1} }

\newcommand{\zetahi}[3][]{
  \MMAS{\zeta}{}{#2#3}{}{}{}{\!}{#1} }
\newcommand{\tildezetahi}[3][]{
  \MMAS{\widetilde{\zeta}}{}{#2#3}{}{}{}{\!}{#1} }

\newcommand{\vechcomp}[2]{\ensuremath{#2}_{\overset{#1}{}}}

\newcommand{\thetahcomp}[3][]{
  \thetaz{
    \IfStrEq{#1}{\empty}{}{\!#1|}
    #2:#3}}

\newcommand{\thetah}[2][]{
  {\bmthetaz{
    \IfStrEq{#1}{\empty}{}{\!#1|}
    #2}}}

\newcommand{\thetahc}[4]{
  {\bmthetaz{
      \IfStrEq{#2}{\empty}{}{\!#2|}
      #3:#4
      \IfStrEq{#1}{\empty}{}{|#1}
}}}

\newcommand{\hatthetahc}[4]{
  {\widehatbmthetaz{
      \IfStrEq{#2}{\empty}{}{\!#2|}
      #3:#4
      \IfStrEq{#1}{\empty}{}{|#1}
}}
}

\newcommand{\thetaM}[2][]{
  \MMAS{\bm{\theta}}{}{#2}{-}{}{}{\!}{#1} }
\newcommand{\thetahb}[3][]{
  \MMAS{\bm{\theta}}{}{#2}{}{:}{\bm{#3}}{\!}{#1} }
\newcommand{\thetahbc}[5]{
  \MMAS{\bm{\theta}}{}{#2|#3:#4}{}{:}{\bm{#5}
    \IfStrEq{#1}{\empty}{}{|#1}
  }{}{} }

\newcommand{\thetaMb}[3][]{
  \MMAS{\bm{\theta}}{}{#2}{-}{:}{\bm{#3}}{\!}{#1} }
\newcommand{\hatthetahN}[3][]{
  \MMAS{\widehat{\bm{\theta}}}{}{#2}{}{:}{#3}{\!}{#1} }
\newcommand{\hatthetahNc}[6][]{
    \MMAS{\widehat{\bm{\theta}}}{}{#3|#4:#5}{}{:}{#6
      \IfStrEq{#2}{\empty}{}{|#2}
    }{\!}{#1} }
\newcommand{\hatthetaMN}[3][]{
  \MMAS{\widehat{\bm{\theta}}}{}{#2}{-}
  {\IfStrEq{#3}{}{}{:}}
  {\IfStrEq{#3}{}{}{#3}}
  {\!}{#1} }

\newcommand{\nablah}[2][]{
  \ensuremath{\subp{\bm{\nabla}}{\!\!}{#2}{}{#1}}}
\newcommand{\nablahc}[4][]{
  \ensuremath{\subp{\bm{\nabla}}{\!\!}{#3:#4
      \IfStrEq{#2}{\empty}{}{|#2}}{}{#1}}}

\newcommand{\nablaM}[2][]{
  \ensuremath{\subp{\bm{\nabla}}{\!\!}{\overbar{#2}}{}{#1}}}
\newcommand{\nablaMc}[4][]{
  \ensuremath{\subp{\bm{\nabla}}{\!\!}{#3:\OUbar{#4}
      \IfStrEq{#2}{\empty}{}{|#2}}{}{#1}}}

\newcommand{\uh}[2][]{
  \MMAS{\bm{u}}{}{#2}{}{}{}{\!}{#1} }
\newcommand{\uhi}[3][]{
  \MMAS{u}{}{#2#3}{}{}{}{\!}{#1} }
\newcommand{\Uh}[2][]{
  \MMAS{u}{}{#2}{}{}{}{\!}{#1} }
\newrobustcmd{\UhL}[3][]{
  \MMAS{u}{}{#2}{}{}{}{\!}{#1}^{{}_{#3}} }

\newcommand{\Kh}[2][]{\ensuremath{\Kz{\!#2}
    \IfStrEq{#1}{}{}{\!\left(#1\right)} }}
\newcommand{\Khb}[3][]{\ensuremath{K_{\overset{#2:#3}{}}
    \IfStrEq{#1}{}{}{\!\left(#1\right)} }}
\newcommand{\Khbc}[4][]{\ensuremath{K_{\overset{#2:#3:#4}{}}
    \IfStrEq{#1}{}{}{\!\left(#1\right)} }}
\newcommand{\KhN}[3][]{\ensuremath{K_{\overset{#2:#3}{}}
    \IfStrEq{#1}{}{}{\!\left(#1\right)} }}
\newcommand{\Khbdef}{\ensuremath{\Khb[\yh{h}-\LGp]{h}{\bm{b}}}}
\newcommand{\KhbDEF}{\ensuremath{\Khb[\Yht{h}{t}-\LGp]{h}{\bm{b}}}}

\newcommand{\qh}[2]{
  \MMAS{q}{}{#1}{}{:}{\bm{#2}}{}{} }
\newcommand{\qhc}[5]{
  \MMAS{q}{}{#2|#3:#4}{}{:}{\bm{#5}
    \IfStrEq{#1}{}{}{|#1}}{}{} }
\newcommand{\LhN}[3][]{
  \MMAS{L}{}{#2}{}{:}{#3}{\!}{#1} }
\newcommand{\LhNc}[6][]{
  \MMAS{L}{}{#3|#4:#5}{}{:}{#6
    \IfStrEq{#2}{\empty}{}{|#2}
  }{\!}{#1} }
\newcommand{\QhN}[3][]{
  \MMAS{Q}{}{#2}{}{:}{#3}{\!}{#1} }
\newcommand{\QhNc}[6][]{
  \MMAS{Q}{}{#3|#4:#5}{}{:}{#6
    \IfStrEq{#2}{\empty}{}{|#2}
  }{\!}{#1} }
\newcommand{\tildeQhN}[3][]{
  \MMAS{\widetilde{Q}}{}{#2}{}{:}{#3}{\!}{#1} }
\newcommand{\tildeQhNc}[6][]{
  \MMAS{\widetilde{Q}}{}{#3|#4:#5}{}{:}{#6
    \IfStrEq{#2}{\empty}{}{|#2}
  }{\!}{#1} }
\newcommand{\QMN}[3][]{
  \MMAS{Q}{}{#2}{-}{:}{#3}{\!}{#1} }

\newcommand{\QMNc}[7][]{
  \MMAS{Q}{}{#3:#4|#5:\OUbar{#6}}{}{:}{#7
    \IfStrEq{#2}{\empty}{}{|#2}
  }{\!}{#1} }

\newcommand{\ThN}[3][]{
    \ensuremath{\Tz{#2
      \IfStrEq{#3}{}{}{:#3}}
    \IfStrEq{#1}{}{}{\!\left(#1\right)}}
}
\newcommand{\ThNc}[6][]{
  \ensuremath{\Tz{#3|#4:#5
      \IfStrEq{#6}{}{}{:#6
        \IfStrEq{#2}{\empty}{}{|#2}
    }}
    \IfStrEq{#1}{}{}{\!\left(#1\right)}}
}

\newcommand{\TMbN}[4][]{
  \ensuremath{\Tz{\overbar{#2}
      \IfStrEq{#3}{}{}{|#3}
      \IfStrEq{#4}{}{}{:#4}}
    \IfStrEq{#1}{}{}{\!\left(#1\right)}}
}

\newcommand{\WhN}[4][]{
  \ensuremath{\Wz[
    \IfStrEq{#2}{}{}{#2}]{\!#3
            \IfStrEq{#4}{}{}{:#4}}
    \IfStrEq{#1}{}{}{\!\left(#1\right)}}
}
\newcommand{\WhNc}[7][]{
  \ensuremath{\Wz[
      \IfStrEq{#5}{}{}{#5}]{\!#3|#4:#6
      \IfStrEq{#7}{}{}{:#7}
      \IfStrEq{#2}{\empty}{}{|#2}
    }
    \IfStrEq{#1}{}{}{\!\left(#1\right)}}
}

\newcommand{\WMbN}[4][]{
  \ensuremath{\Wz{\!\overbar{#2}
      \IfStrEq{#3}{}{}{|#3}
      \IfStrEq{#4}{}{}{:#4}}
    \IfStrEq{#1}{}{}{\!\left(#1\right)}}
}

\newcommand{\VhN}[4][]{
  \ensuremath{\Vz[
    \IfStrEq{#2}{}{}{#2}]{\!#3
            \IfStrEq{#4}{}{}{:#4}}
    \IfStrEq{#1}{}{}{\!\left(#1\right)}}
}
\newcommand{\VhNc}[7][]{
  \ensuremath{\Vz[
      \IfStrEq{#5}{}{}{#5}]{\!#3|#4:#6
      \IfStrEq{#7}{}{}{:#7}
      \IfStrEq{#2}{\empty}{}{|#2}
    }
    \IfStrEq{#1}{}{}{\!\left(#1\right)}}
}

\newcommand{\VMbN}[4][]{
  \ensuremath{\Vz{\overbar{#2}
      \IfStrEq{#3}{}{}{|#3}
      \IfStrEq{#4}{}{}{:#4}}
    \IfStrEq{#1}{}{}{\!\left(#1\right)}}
}

\newcommand{\RRn}[1]{\ensuremath{\subp{\RR}{}{}{}{#1}}}

{
  \theoremstyle{definition}
  \newtheorem{assumption}{Assumption}[section]
}

\crefname{assumption}{assumption}{assumptions}

{
  \theoremstyle{remark}
  
}
\crefname{remark}{remark}{remarks}

\renewcommand{\uh}[2][]{
  \MMAS{\bm{u}}{}{#2}{}{}{}{\!}{#1} }
\newrobustcmd{\uhc}[4][]{
  \MMAS{\bm{u}}{}{#3:#4
    \IfStrEq{#2}{\empty}{}{|#2}
  }{}{}{}{\!}{#1} }
\newrobustcmd{\bNi}[2][N]{
  \MMAS{b}{}{#1}{}{}{#2}{}{} }
\newrobustcmd{\aNi}[2][]{
  \MMAS{a}{}{#1}{}{}{#2}{}{} }
\newrobustcmd{\cNi}[2][]{
  \MMAS{c}{}{#1}{}{}{#2}{}{} }
\newrobustcmd{\rNi}[2][]{
  \MMAS{r}{}{#1}{}{}{#2}{}{} }
\newrobustcmd{\sNi}[2][]{
  \MMAS{s}{}{#1}{}{}{#2}{}{} }
\newrobustcmd{\ellNi}[2][]{
  \MMAS{\ell}{}{#1}{}{}{#2}{}{} }
\newrobustcmd{\XNht}[4][]{\Xz[#2\IfStrEq{#1}{\empty}{}{|#1}]{#3:#4}} 
\newrobustcmd{\tildeXNht}[4][]{\widetildeXz[#2\IfStrEq{#1}{\empty}{}{|#1}]{#3:#4}} 
\newrobustcmd{\XNMt}[4][]{\Xz[#2\IfStrEq{#1}{\empty}{}{|#1}]{\overbar{#3}:#4}}

\newrobustcmd{\ZNht}[4][]{\Zz[#2\IfStrEq{#1}{\empty}{}{|#1}]{#3:#4}} 
\newrobustcmd{\ZNhtvec}[3][]{\bmZz[#2\IfStrEq{#1}{\empty}{}{|#1}]{#3}} 
\newrobustcmd{\ZNMt}[4][]{\Zz[#2\IfStrEq{#1}{\empty}{}{|#1}]{\overbar{#3}:#4}}

\newrobustcmd{\QNh}[3][]{\subp{\reflectbox{$Q$}}{}{#3}{}{#2\IfStrEq{#1}{\empty}{}{|#1}}} 
\newrobustcmd{\QNhvec}[3][]{\subp{\reflectbox{$\bm{Q}$}}{}{#3}{}{#2\IfStrEq{#1}{\empty}{}{|#1}}} 
\newrobustcmd{\QNM}[3][]{\subp{\reflectbox{$Q$}}{}{\overbar{#3}}{}{#2\IfStrEq{#1}{\empty}{}{|#1}}}

\newrobustcmd{\muNh}[3][]{
  \MMAS{\mu}{}{#3}{}{}{}{}{}^{{}_{#2\IfStrEq{#1}{\empty}{}{|#1}}} }
\newrobustcmd{\muNM}[3][]{
  \MMAS{\mu}{}{#3}{-}{}{}{}{}^{{}_{#2\IfStrEq{#1}{\empty}{}{|#1}}} }

\newrobustcmd{\zetaNj}[3][*]{
  \MMAS{\zeta}{}{}{}{}{#3}{}{}^{{}_{#2|#1}} }

\newrobustcmd{\pnt}[2]{
  \ensuremath{p\!\left(
      \IfStrEq{#1}{#2}{#2}{#2|#1}\right)}}
\newrobustcmd{\qnt}[2]{
  \ensuremath{q\!\left(
      \IfStrEq{#1}{#2}{#2}{#2|#1}\right)}}

\newrobustcmd{\tildeqnt}[2]{
  \ensuremath{\widetilde{q}\!\left(
      \IfStrEq{#1}{#2}{#2}{#2|#1}\right)}}

\newrobustcmd{\innt}[2]{
  \ensuremath{i\!\left(
      \IfStrEq{#1}{#2}{#2}{#2|#1}\right)}}
  
\newrobustcmd{\tildeinnt}[2]{
  \ensuremath{\tilde{\imath}\!\left(
      \IfStrEq{#1}{#2}{#2}{#2|#1}\right)}}

\newrobustcmd{\anp}[3][]{
  \ensuremath{\subp{a}{
      \IfStrEq{#2}{}{}{#2:}
      }{#3}{}{#1}}}

\newrobustcmd{\vnp}[3][]{
  \ensuremath{\subp{v}{#2:}{#3}{}{#1}}}

\newrobustcmd{\tp}[1]{
  \ensuremath{\subp{t}{}{}{}{#1}}}

\newrobustcmd{\lambdaiM}[2]{
  \ensuremath{\subp{\lambda}{}{#1}{}{}\!\left(#2 \right)}}

\newrobustcmd{\lambdazM}[3][]{
  \lambdaz[#1]{#3\!}(#2)}

\newrobustcmd{\bmbzh}[1]{\bmbz{#1}}

\newrobustcmd{\rni}[3][]{
  \ensuremath{\subp{r}{
      \IfStrEq{#2}{}{}{#2:}}{#3}{}{}        
    \IfStrEq{#1}{}
    {}{\!\left(#1\right) }}}

\newrobustcmd{\unit}[2][]{
  \ensuremath{\subp{\bm{e}}{}{#2}{}{#1}}}

\newrobustcmd{\baM}[2][]{
  \ensuremath{\subp{\bm{a}}{}{\overbar{#2}}{}{#1}}}

\newrobustcmd{\mathfrakUhb}[3][]{
  \ensuremath{\subp{\mathfrak{U}}{#2}{
      \IfStrEq{#3}{}{}{:\bm{#3}}
    }{}{#1}}}

\newrobustcmd{\mathcalCnki}[3]{
  \ensuremath{\subp{\mathcal{C}}{
      \IfStrEq{#1}{}{}{#1:}
    }{
      \IfStrEq{#2}{}{}{#2|}
      #3}{}{}}}

\newrobustcmd{\mathcalVnki}[3]{
  \ensuremath{\subp{\mathcal{V}}{#1:}{#2|#3}{}{}}}

\newrobustcmd{\mathcalWnki}[3]{
  \ensuremath{\subp{\mathcal{W}}{#1:}{#2|#3}{}{}}}

\newrobustcmd{\QMx}[2]{
  \ensuremath{\subp{Q}{}{#1}{}{}\!\!\left(#2\right)}}

\newrobustcmd{\sigmah}[2]{
  \MMAS{\sigma}{}{#1}{}{|}{#2}{}{}^{\,{}_{2}} }

\newrobustcmd{\sigmaM}[3][]{
  \MMAS{\sigma}{}{#2}{-}{|}{#3}{}{}^{\,{}_{2}}\!\!\left(#1\right) }

\newrobustcmd{\sigmaMlimit}[2][]{
  \MMAS{\sigma}{}{\bullet}{}{|}{#2}{}{}^{\,{}_{2}}\!\!\left(#1\right) }

\newrobustcmd{\SigmaM}[2]{
  \MMAS{\bm{\Sigma}}{}{#1}{-}{|}{#2}{}{} }

\newrobustcmd{\INhjl}[5][]{\Iz[#2\IfStrEq{#1}{\empty}{}{|#1}]{#3#4:#5}}
\newrobustcmd{\INMl}[4][]{\Iz[#2\IfStrEq{#1}{\empty}{}{|#1}]{\overbar{#3}:#4}}

\newcommand{\Oh}[2][]{   \ensuremath{\subp{O}{}{#1}{}{}\!\left(#2\right)} }

\newcommand{\oh}[2][]{   \ensuremath{\subp{o}{}{#1}{}{}\!\left(#2\right)} }

\newcommand{\SNp}[2]{   \ensuremath{\subp{S}{}{#1}{}{(#2)}} }

\newcommand{\blimit}{
  \ensuremath{\bm{b}\rightarrow\subp{\bm{0}}{}{}{}{+}}
}

\newcommand{\hlimit}{
  \ensuremath{h\rightarrow\infty}
}

\newcommand{\mlimit}{
  \ensuremath{m\rightarrow\infty}
}

\newcommand{\nlimit}{
  \ensuremath{n\rightarrow\infty}
}

\newcommand{\Norm}[3][]{
  \ensuremath{\subp{\operatorname{N}}{\!}{#1}{}{} \!     \left(#2, #3 \right) }
}

\newcommand{\operatorss}[3]{  \relax   \ifmmode   #1_{#2}^{#3}   \else   \ensuremath{\subp{#1}{}{#2}{}{#3}}   \fi }

\newcommand{\cupss}[2]{
  \operatorss{\bigcup}{#1}{#2}
}

\newcommand{\intss}[2]{
  \operatorss{\int}{#1}{#2}
}

\newcommand{\oplusss}[2]{
  \operatorss{\bigoplus}{#1}{#2}
}

\newcommand{\prodss}[2]{
  \operatorss{\prod}{#1}{#2}
}

\newcommand{\sumss}[2]{
  \operatorss{\sum}{#1}{#2}
}

\newcommand{\ith}[2][]{
  \ensuremath{#2^{#1 \operatorname{th}}}
}

\expandarg

\newrobustcmd{\parenType}[2][L]{
    \StrLeft{#2}{1}[\parSelect]    \IfStrEq{#1}{L}{
    \def\parDirection{\left}
         \IfStrEq{\parSelect}{.}{\def\parType{.}}{}
    \IfStrEq{\parSelect}{|}{\def\parType{|}}{}
    \IfStrEq{\parSelect}{1}{\def\parType{\|}}{}
    \IfStrEq{\parSelect}{u}{\def\parType{\uparrow}}{}
    \IfStrEq{\parSelect}{d}{\def\parType{\downarrow}}{}
    \IfStrEq{\parSelect}{U}{\def\parType{\Uparrow}}{}
    \IfStrEq{\parSelect}{D}{\def\parType{\Downarrow}}{}
        \IfStrEq{\parSelect}{2}{\def\parType{(}}{}
    \IfStrEq{\parSelect}{3}{\def\parType{)}}{}
        \IfStrEq{\parSelect}{4}{\def\parType{[}}{}
    \IfStrEq{\parSelect}{5}{\def\parType{]}}{}
        \IfStrEq{\parSelect}{6}{\def\parType{\{}}{}
    \IfStrEq{\parSelect}{7}{\def\parType{\}}}{}
        \IfStrEq{\parSelect}{8}{\def\parType{<}}{}
    \IfStrEq{\parSelect}{9}{\def\parType{>}}{}
                    \IfStrEq{\parSelect}{b}{\def\parType{\lfloor}}{}
    \IfStrEq{\parSelect}{B}{\def\parType{\rfloor}}{}
        \IfStrEq{\parSelect}{b}{\def\parType{\lfloor}}{}
    \IfStrEq{\parSelect}{B}{\def\parType{\rfloor}}{}
        \IfStrEq{\parSelect}{t}{\def\parType{\lceil}}{}
    \IfStrEq{\parSelect}{T}{\def\parType{\rceil}}{}
  }{
    \def\parDirection{\right}
         \IfStrEq{\parSelect}{.}{\def\parType{.}}{}
    \IfStrEq{\parSelect}{|}{\def\parType{|}}{}
    \IfStrEq{\parSelect}{1}{\def\parType{\|}}{}
    \IfStrEq{\parSelect}{u}{\def\parType{\uparrow}}{}
    \IfStrEq{\parSelect}{d}{\def\parType{\downarrow}}{}
    \IfStrEq{\parSelect}{U}{\def\parType{\Uparrow}}{}
    \IfStrEq{\parSelect}{D}{\def\parType{\Downarrow}}{}
        \IfStrEq{\parSelect}{2}{\def\parType{)}}{}
    \IfStrEq{\parSelect}{3}{\def\parType{(}}{}
        \IfStrEq{\parSelect}{4}{\def\parType{]}}{}
    \IfStrEq{\parSelect}{5}{\def\parType{[}}{}
        \IfStrEq{\parSelect}{6}{\def\parType{\}}}{}
    \IfStrEq{\parSelect}{7}{\def\parType{\{}}{}
        \IfStrEq{\parSelect}{8}{\def\parType{>}}{}
    \IfStrEq{\parSelect}{9}{\def\parType{<}}{}
                    \IfStrEq{\parSelect}{b}{\def\parType{\rfloor}}{}
    \IfStrEq{\parSelect}{B}{\def\parType{\lfloor}}{}
        \IfStrEq{\parSelect}{T}{\def\parType{\rceil}}{}
    \IfStrEq{\parSelect}{t}{\Def\parType{\lceil}}{}
  }
    \IfStrEq{#2}{}{}{\parDirection\parType}
}

\newrobustcmd{\parenthesisGenerator}[5][]{
  \begingroup
        \IfStrEq{#1}{}{\def\firstArg{#2}}{\def\firstArg{#1}}
    \IfStrEq{#2}{}{\def\firstArg{#2}}{}
    \subp{
    \IfStrEq{#2}{}{}{\parenType[L]{#2}}
    \IfStrEq{#2}{}{
      \mbox{\ensuremath{#3}}}{#3}
    \IfStrEq{\firstArg}{}{}{\parenType[R]{\firstArg}}}{}{#4}{}{#5}
  \endgroup
}

\newrobustcmd{\parenR}[1]{
  \parenthesisGenerator{2}{#1}{}{}
}
\newrobustcmd{\parenRz}[3][]{
  \parenthesisGenerator{2}{#3}{\!#2}{\!#1}
}
\newrobustcmd{\parenA}[1]{
  \parenthesisGenerator{8}{#1}{}{}
}
\newrobustcmd{\parenAz}[3][]{
  \parenthesisGenerator{8}{#3}{\!#2}{\!#1}
}
\newrobustcmd{\parenC}[1]{
  \parenthesisGenerator{6}{#1}{}{}
} 
\newrobustcmd{\parenCz}[3][]{
  \parenthesisGenerator{6}{#3}{#2}{#1}
}
\newrobustcmd{\parenS}[1]{
  \parenthesisGenerator{4}{#1}{}{}
} 
\newrobustcmd{\parenSz}[3][]{
  \parenthesisGenerator{4}{#3}{\!#2}{\!#1}
}
\newrobustcmd{\parenAbs}[1]{
  \parenthesisGenerator{|}{#1}{}{}
} 
\newrobustcmd{\parenAbsz}[3][]{
  \parenthesisGenerator{|}{#3}{#2}{#1}
}
\newrobustcmd{\parenABS}[1]{
  \parenthesisGenerator{1}{#1}{}{}
} 
\newrobustcmd{\parenABSz}[3][]{
  \parenthesisGenerator{1}{#3}{#2}{#1}
}

\newrobustcmd{\power}[3][]{
  \parenthesisGenerator{}{#3}{#2}{#1}
}

\newrobustcmd{\Vector}[2][]{
  \parenSz[\IfStrEq{#1}{}{}{\,#1}]{}{#2}
}

\newrobustcmd{\Co}{\texttt{Co}\xspace}
\newrobustcmd{\Quad}{\texttt{Quad}\xspace}
\newrobustcmd{\Phase}{\texttt{Phase}\xspace}
\newrobustcmd{\Amplitude}{\texttt{Amplitude}\xspace}

\newrobustcmd{\eval}[3][]{
  \subp{\left.#2\right|}{}{#3}{}{}
}

\makeatletter{}

\renewcommand{\lgcor}[2][]{
  \rhoz{#2
    \IfStrEq{#1}{\empty}{\!}{}}}

\renewcommand{\hatlgcor}[2][]{
  \widehatrhoz{#2
    \IfStrEq{#1}{\empty}{\!}{}}}

\renewcommand{\lgacr}[3][]{
  \rhoz{#2
    \IfStrEq{#1}{\empty}{\!}{}}(#3)}

\renewcommand{\lgccr}[4][]{
  \rhoz{#2:#3
    \IfStrEq{#1}{\empty}{\!}{}}(#4)}

\renewcommand{\hatlgacr}[3][]{
  \widehatrhoz{#2
    \IfStrEq{#1}{\empty}{\!}{}}(#3)}
\renewcommand{\hatlgccr}[4][]{
  \widehatrhoz{#2:#3
    \IfStrEq{#1}{\empty}{\!}{}}(#4)}

\renewcommand{\hatlgacrb}[4][]{
  \widehatrhoz{\!#2
    \IfStrEq{#1}{\empty}{\!}{}
  }(#3|\scalebox{.7}{$#4$})}

\renewcommand{\hatlgccrb}[5][]{
  \widehatrhoz{#2:#3
    \IfStrEq{#1}{\empty}{\!}{}
  }(#4|\scalebox{.7}{$#5$})}

\renewcommand{\hatlgthetab}[4][]{
  \widehatbmthetaz{\!#2
    \IfStrEq{#1}{\empty}{\!}{}
  }(#3|\scalebox{.7}{$#4$})}

\renewcommand{\lgsd}[3][]{
  \fz{\!#2
    \IfStrEq{#1}{\empty}{\!}{}}(#3)}

\renewcommand{\lgcsd}[4][]{
  \fz{#2:#3
    \IfStrEq{#1}{\empty}{\!}{}}(#4)}

\renewcommand{\lgcsdCo}[4][]{
  \cz{#2:#3
    \IfStrEq{#1}{\empty}{\!}{}}(#4)}
\renewcommand{\lgcsdQuad}[4][]{
  \qz{#2:#3
    \IfStrEq{#1}{\empty}{\!}{}}(#4)}
\renewcommand{\lgcsdAmplitude}[4][]{
  \alphaz{#2:#3
    \IfStrEq{#1}{\empty}{\!}{}}(#4)}
\renewcommand{\lgcsdPhase}[4][]{
  \phiz{#2:#3
    \IfStrEq{#1}{\empty}{\!}{}}(#4)}

\renewcommand{\lgcsdCoM}[5][]{
  \cz[#5]{#2:#3
    \IfStrEq{#1}{\empty}{\!}{}}(#4)}
\renewcommand{\lgcsdQuadM}[5][]{
  \qz[#5]{#2:#3
    \IfStrEq{#1}{\empty}{\!}{}}(#4)}
\renewcommand{\lgcsdAmplitudeM}[5][]{
  \alphaz[#5]{#2:#3
    \IfStrEq{#1}{\empty}{\!}{}}(#4)}
\renewcommand{\lgcsdPhaseM}[5][]{
  \phiz[#5]{#2:#3
    \IfStrEq{#1}{\empty}{\!}{}}(#4)}

\renewcommand{\hatlgcsdCoM}[5][]{
  \widehatcz[\ #5]{#2:#3
    \IfStrEq{#1}{\empty}{\!}{}}(#4)}
\renewcommand{\hatlgcsdQuadM}[5][]{
  \widehatqz[\ #5]{#2:#3
    \IfStrEq{#1}{\empty}{\!}{}}(#4)}
\renewcommand{\hatlgcsdAmplitudeM}[5][]{
  \widehatalphaz[\ #5]{#2:#3
    \IfStrEq{#1}{\empty}{\!}{}}(#4)}
\renewcommand{\hatlgcsdPhaseM}[5][]{
  \widehatphiz[\ #5]{#2:#3
    \IfStrEq{#1}{\empty}{\!}{}}(#4)}

\renewcommand{\lgcsdSQ}[4][]{
  \mathcalKz[asc]{#2:#3
    \IfStrEq{#1}{\empty}{\!}{}}(#4)}

\renewcommand{\lgsdM}[4][]{
  \fz[#4]{\!#2
    \IfStrEq{#1}{\empty}{\!}{}}(#3)}
\renewcommand{\lgcsdM}[5][]{
  \fz[#5]{\!#2:#3
    \IfStrEq{#1}{\empty}{\!}{}}(#4)}

\renewcommand{\hatlgsd}[3][]{
  \widehatfz{\!#2
    \IfStrEq{#1}{\empty}{\!}{}}(#3)}

\renewcommand{\hatlgcsd}[4][]{
  \widehatfz{\!#2:#3
    \IfStrEq{#1}{\empty}{\!}{}}(#4)}

\renewcommand{\hatlgsdM}[4][]{
  \widehatfz[#4]{\!#2
    \IfStrEq{#1}{\empty}{\!}{}}(#3)}
\renewcommand{\hatlgcsdM}[5][]{
  \widehatfz[#5]{\!#2:#3
    \IfStrEq{#1}{\empty}{\!}{}}(#4)}

\renewcommand{\lgsdRE}[3][]{
  \cz{#2
    \IfStrEq{#1}{\empty}{\!}{}}(#3)}
\renewcommand{\hatlgsdREM}[4][]{
  \widehatcz[\ #4]{#2
    \IfStrEq{#1}{\empty}{\!}{}}(#3)}

\renewcommand{\lgsdIM}[3][]{
  \qz{#2
    \IfStrEq{#1}{\empty}{\!}{}}(#3)}
\renewcommand{\hatlgsdIMM}[4][]{
  \widehatqz[\ #4]{#2
    \IfStrEq{#1}{\empty}{\!}{}}(#3)}

\clearpage{}

\newcommand{\myblind}[2]{
  \if1\blind
  #1
  \fi
  \if0\blind
  \texttt{::#2::}\xspace
  \fi
}

\usepackage{marginnote}

\usepackage[
linewidth=1pt,
middlelinecolor= black,
middlelinewidth=0.4pt,
roundcorner=1pt,
topline = false,
rightline = false,
bottomline = false,
rightmargin=0pt,
skipabove=0pt,
skipbelow=0pt,
leftmargin=-.3cm,
innerleftmargin=.3cm,
innerrightmargin=0pt,
innertopmargin=0pt,
innerbottommargin=0pt,
]{mdframed}

\usepackage{bibunits}
\defaultbibliography{../Bibliography}
\defaultbibliographystyle{elsarticle-harv}

\begin{document}

\def\spacingset#1{\renewcommand{\baselinestretch}
  {#1}\small\normalsize} \spacingset{1}

\begin{bibunit}

  \if1\blind
  {
  \title{\bf   {Nonlinear spectral analysis:\\ A local Gaussian
      approach}}
  
  \author{      Lars Arne Jordanger\thanks{     Western Norway University of Applied Sciences, Faculty of
      Engineering and Science, P.B 7030, 5020 Bergen, Norway
      E-mail: \textrm{Lars.Arne.Jordanger@hvl.no}}  \and       Dag Tj{\o}stheim\thanks{     University of Bergen, Department of Mathematics, P.B. 7803, 5020
      Bergen, Norway} }
  \date{\vspace{-5ex}}
  
  \maketitle
  } \fi

  \if0\blind
  {
    \bigskip
    \bigskip
    \bigskip
    \begin{center}
      {\LARGE\bf {Nonlinear spectral analysis:\\ A local Gaussian
      approach}}
    \end{center}
    \medskip
  } \fi

  \bigskip
  \begin{abstract}

    The spectral distribution $f(\omega)$ of a stationary time series
    $\TSR{\Yz{t}}{t\in\ZZ}{}$ can be used to investigate whether or not
    periodic structures are present in $\TSR{\Yz{t}}{t\in\ZZ}{}$, but
    $f(\omega)$ has some limitations due to its dependence on the
    autocovariances $\gamma(h)$.  For example, $f(\omega)$ can not
    distinguish white i.i.d.\ noise from GARCH-type models (whose terms
    are dependent, but uncorrelated), which implies that $f(\omega)$ can
    be an inadequate tool when $\TSR{\Yz{t}}{t\in\ZZ}{}$ contains
    asymmetries and nonlinear dependencies.

    Asymmetries between the upper and lower tails of a time series can
    be investigated by means of the \textit{local Gaussian
      autocorrelations} introduced in \citet{Tjostheim201333}, and these
    \textit{local measures of dependence} can be used to construct the
    \textit{local Gaussian spectral density} presented in this paper.  A
    key feature of the new local spectral density is that it coincides
    with $f(\omega)$ for Gaussian time series, which implies that it can
    be used to detect non-Gaussian traits in the time series under
    investigation.  In particular, if $f(\omega)$ is flat, then peaks
    and troughs of the new local spectral density can indicate nonlinear
    traits, which potentially might discover \textit{local periodic
      phenomena} that remain undetected in an ordinary spectral analysis.

  \end{abstract}




  \makeatletter{}


\section{Introduction}
\label{sec:Introduction}

Spectral analysis is an important tool in time series analysis.  In
its classical form, assuming $\sum|\gamma(h)|<\infty$, the spectral
density function of a stationary times series
$\TSR{\Yz{t}}{t\in\ZZ}{}$ is the Fourier transform of the
autocovariances $\TSR{\gamma(h)=\Cov{\Yz{t+h}}{\Yz{t}}}{h\in\ZZ}{}$.
Furthermore, since $\gamma(h)=\Var{\Yz{t}}\cdot\rho(h)$, with
$\rho(h)$ the autocorrelations, this can be expressed as:
\begin{equation}
  \label{eq:classical_spectral_density}
  f(\omega) \defeq \sum_{h\in\ZZ} \gamma(h)\cdot \ez[-2\pi i\omega h]{} =
  \Var{\Yz{t}} \cdot \sum_{h\in\ZZ} \rho(h)\cdot \ez[-2\pi i\omega h]{}.
\end{equation}

The connection $\Var{\Yz{t}} = \intss{-1/2}{1/2} f(\omega) \d{\omega}$
follows from the inverse Fourier transformation, and this reveals how
$f(\omega)$ gives a decomposition of the variance over different
frequencies.  In particular, the spectral density function $f(\omega)$
captures the components of periodic linear structure decomposed over
frequency for $\TSR{\Yz{t}}{t\in\ZZ}{}$, and the peaks and troughs of
$f(\omega)$ can thus reveal important features of the time series
under investigation.

Nonlinear dependencies between the terms of a time series
$\TSR{\Yz{t}}{t\in\ZZ}{}$ will however not be reflected in the
spectral density $f(\omega)$, since it is the linear dependencies that
are detected by the autocovariance functions $\gamma(h)$.  The most
obvious example is the GARCH model from \citet{bollerslev86:_gener}.
The GARCH model is much used in econometrics, and it is well known
that this model in general 
exhibits dependence over many lags (long range dependence).  But this
dependence is not captured by the autocovariance function, since
$\gamma(h)$ is zero for lags $|h|\geq1$.  This again implies that the
spectral density is flat for a GARCH model.

{
  \label{pAE1}
  An estimate of $f(\omega)$ based on samples from e.g.\ a
  GARCH(1,1)-model will then, as seen in the left panel of
  \cref{fig:GARCH_intro_example_short}, not reveal any information at
  all.  An investigation based on the method presented in this paper
  can however detect the nonlinear structure --- as seen in the right
  panel of \cref{fig:GARCH_intro_example_short}, where a point in the
  lower tail has been inspected.  }

\begin{figure}[h]
  {\centering \includegraphics[width=1\linewidth]{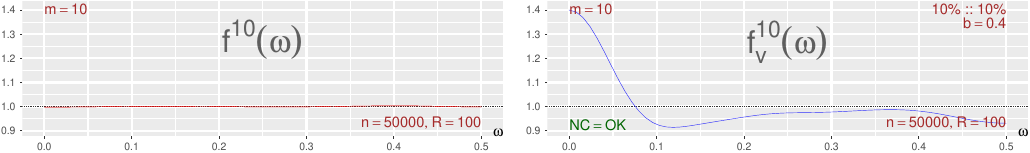}
  }
  \caption[]{Left: Estimated ordinary (variance-rescaled) spectral
    density based on a GARCH(1,1)-example.  Right: Estimated local
    Gaussian spectral density at a point in the lower tail.  See
    \cref{app:fig:GARCH_intro_example_short} for details regarding the
    underlying data.}\label{fig:GARCH_intro_example_short}
\end{figure}

One may ask whether there exist classes of processes for which the
spectral density gives complete information about the probabilistic
dependence structure.  The answer is simple: If
$\TSR{\Yz{t}}{t\in\ZZ}{}$ is a stationary Gaussian process, then its
complete distributional dependence structure (assuming a zero mean
process) can be set up in terms of its spectral density.  This is in
fact a starting point for the Whittle-type likelihood in time series
analysis.

This paper is concerned with finding a generalisation of
\cref{eq:classical_spectral_density} that enables the investigation of
nonlinear structures in general non-Gaussian stationary processes.
This will be based on a local approach using Gaussian approximations,
which ensures the desirable property that the ordinary spectral
density is returned for a Gaussian process.

A number of attempts have been made in the literature to extend the
standard spectral density $f(\omega)$, and these can roughly be
divided into three categories.

Perhaps the best known, and probably the procedure going furthest back
in time, is represented by the higher order spectra; see
\citet{Brillinger:1984:CWJa,brillinger1991some,Tukey:1959:IMS}.  The
formula for the ordinary spectral density $f(\omega)$ from
\cref{eq:classical_spectral_density} is then supplemented by
considering the Fourier transformations of the higher order moments
(or cumulants), such as $\E{\Yz{r}\Yz{s}\Yz{t}}$ resulting in the
bispectrum depending on a double set of frequencies and
$\E{\Yz{r}\Yz{s}\Yz{t}\Yz{u}}$ producing the trispectrum dependent on
a triple of frequencies.  {\label{R1mc3_1}
  These cumulant-based higher order spectra are identical to zero for
  Gaussian processes.} 
The multi-frequency dependence of the bispectrum and trispectrum are
not always easy to interpret, and one may also question the existence
of higher order moments; in econometrics thick tails often makes this
into an issue.

Another approach is to replace $\gamma(h)$ in
\cref{eq:classical_spectral_density} by another measure of dependence
as a function of $h$.  Recently there has been much activity in
constructing an alternative to \cref{eq:classical_spectral_density} by
considering covariances of a stationary process obtained by describing
quantile crossings, see \citet{hagemann2011robust} for a well-written
introduction and many references.  This is a local spectrum in the
sense that it varies with the chosen quantile.  It is not always
possible to give a local periodic frequency interpretation as in
\cref{eq:classical_spectral_density}, but \citet{li2012quantile}
emphasises a local sinusoidal construction by analogy with quantile
regression models.  See also
\citet{Linton2007250,Han2016251,li08:_laplac_period_time_series_analy,Li20102133,li2014quantile,li12:_detec,li10:_nonlin_method_robus_spect_analy,li10:_robus,li2012robust}.
{\label{R1mc3_2}
  These approaches does usually not recover the ordinary spectrum for
  the Gaussian processes.}  
This loss of recovery is also the case if a local spectrum is
constructed on the basis of the so-called conditional correlation
function (\citet{silvapulle01:_large}).  Still another viewpoint would
be obtained in a spectral analysis of the distance Brownian covariance
function \citet{szekely2009}.

A third alternative is constituted by Hong's generalised spectrum, see
\citet{hong1999hypothesis,hong2000generalized}, which is obtained by
replacing the covariance function $\gamma(h)$ in
\cref{eq:classical_spectral_density} by the bivariate covariance
function $\sigmaz{h}(u,v)$ constructed by taking covariances between
the characteristic function expressions $\exp\parenR{iu\Yz{t+h}}$ and
$\exp\parenR{iv\Yz{h}}$.  Again, this gives a complete distributional
characterisation of dependence properties, but so far not much
attention has been given to concrete data analytic interpretation of
this frequency representation.  Rather, it has been used to test for
independence, conditional independence and predictability
\citet{li2016generalized,wang2017characteristic}.

The new approach presented in this paper follows the strategy where
the $\gamma(h)$ of \cref{eq:classical_spectral_density} is replaced by
another dependence measure, i.e.\ the \textit{local Gaussian
  autocorrelation} introduced in \citet{Tjostheim201333}, see
\citet{lacal2017local,lacal:_estim} for a number of recent
references.  
\label{R1mc4} 
The definition of the local Gaussian autocorrelation is given in
\cref{sec:Definitions}, but the gist of it can be described as
follows: The joint distribution of $\parenR{\Yz{t+h},\Yz{t}}$ is
approximated locally at a point $\LGp$, say, by a Gaussian bivariate
distribution --- and the correlation parameter from this approximating
Gaussian distribution is then taken as the local Gaussian
autocorrelation $\lgacr{\LGp}{h}$ at the point $\LGp$.  If
$\sum |\lgacr[p]{\LGp}{h}| < \infty$, the \textit{local Gaussian
  spectral density} at the point $\LGp$ can be defined in the
following manner,
\begin{align}
  \label{eq:lgsd_definition_intro}
  \lgsd[p]{\LGp}{\omega} \defeq \sumss{h=-\infty}{\infty}
  \lgacr[p]{\LGp}{h} \cdot \ez[-2\pi i\omega h]{}.
\end{align}

This enables a local frequency decomposition with different frequency
representations at different points $\LGp$, e.g.\ different
oscillatory behaviour at extremes (cf.\ also the extremogram of
\citet{davis2009}) as compared to oscillatory behaviour in the center
of the process.  The point $\LGp$ will naturally correspond to a pair
of quantiles, but this concept is distinctly different from the
quantile spectra referred to above in that it considers a
neighbourhood of $\LGp$ and not $\LGp$ as a threshold.  {
  \label{R1mc3_3}
  Moreover, this approach returns a scaled version of the ordinary
  spectrum when a Gaussian process is investigated, with equality when
  $\Var{\Yz{t}}=1$.}

Due to issues related to numerical convergence, the estimates
presented in this paper will be based on an initial normalisation of
$\TSR{\Yz{t}}{t\in\ZZ}{}$, and for the normalised processes the
correlation $\rho(h)$ will always equal the covariance $\gamma(h)$.
All references to $f(\omega)$ will henceforth refer to the spectral
density of a normalised process, i.e.\ $f(\omega)$ will now refer to
the following rescaled version instead of the one given in
\cref{eq:classical_spectral_density},
\begin{align}
  \label{eq:basic_spectral_density_rho}
  f(\omega) \defeq  \sumss{h\in\ZZ}{} \rho(h)\cdot \ez[-2\pi i\omega ]{}.
\end{align}
For the normalised processes, 
$f(\omega)$ and $\lgsd[p]{\LGp}{\omega}$ will by construction be
identical for Gaussian time series, and a comparison of the ordinary
spectrum $f(\omega)$ and the local Gaussian spectrum
$\lgsd[p]{\LGp}{\omega}$ can thus be used to investigate at a local
level how a non-Gaussian time series deviates from being Gaussian.

Much more details of this framework is given in \cref{sec:LGSD}.  This
section also contains the asymptotic theory with detailed proofs in
the Supplementary Material.  The real and simulated examples of
\cref{sec:Examples} show that local spectral estimates can detect
local periodic phenomena and detect nonlinearities in non-Gaussian
white noise.  Note that the scripts needed for the reproduction of
these examples are contained in the \Rpackage
\lgsdRpackage,\footnote{
  \label{LGSD_footnote:install_package} Use
  \devtoolgithublgsdRpackage\ to install the package.  See the
  documentation of the function \Rref{LG\_extract\_scripts} for
  further details.  See also \cref{app:data_details}. } 
where it in addition is possible to use an interactive tool to see how
adjustments of the input parameters (used in the estimation
algorithms) influence the estimates of~$\lgsd{\LGp}{\omega}$.

The theory developed in this paper can be extended to the multivariate
case, see \myblind{\citet{jordanger17:_lgcsd}}{blinded reference}.

\section{Local Gaussian spectral densities}
\label{sec:LGSD}

The local Gaussian correlation (LGC) was introduced in
\citet{Tjostheim201333}, with theory that showed how it could be used
to estimate the local Gaussian autocorrelations for a time series.  It
has been further developed in a number of papers, primarily
\citet{lacal2017local,lacal:_estim}, but see also
\citet{otneim2017locally,otneim2018conditional,Berentsen2014:departure_from_independence,Berentsen2014:Copula,Berentsen:2014:ILR,berentsen16:_some_proper_local_gauss_correl,stove2014using,stoeve14:_measur}
for related issues.  In \citet{Tjostheim201333} the possibility
of developing a local Gaussian spectral analysis was briefly
mentioned, and this is the topic of the present paper.

This section gives a brief summary of the local Gaussian
autocorrelations, and use them to define the local Gaussian spectral
density for strictly\footnote{
  Strict stationarity is necessary in order for the machinery of the
  local Gaussian approximations to be feasible, since Gaussian pdfs
  will be used to locally approximate the pdfs corresponding to the
  bivariate pairs~\mbox{$\left(\Yz{t+h},\Yz{t}\right)$}.
} stationary
univariate time series $\TSR{\Yz{t}}{t\in\ZZ}{}$, and give estimators
with a corresponding asymptotic theory.

\subsection{The local Gaussian correlations}
\label{sec:Definitions}

Details related to the estimation regime, and asymptotic properties,
can be found in \cref{app:bivariate_penalty_functions} in the
Supplementary Material.  Note that other approaches to the concept of
local Gaussian correlation also have been investigated, cf.\
\citet{berentsen16:_some_proper_local_gauss_correl} for~details.

\subsubsection{Local Gaussian correlation, general version}
\label{sec:LGC}
Consider a bivariate random variable
\mbox{$\bm{W} = \left(\Wz{1},\Wz{2}\right)$} with joint
cdf~$G(\bm{w})$ and joint pdf~$g(\bm{w})$.  For a specified
point~\mbox{$\LGp \defeq \LGpoint$}, the main idea is to find the 
bivariate Gaussian distribution whose density function best
approximates~$g(\bm{w})$ in a neighbourhood of the point of interest.
The LGC will then be defined to be the correlation of this local
Gaussian~approximation.

For the purpose of this investigation, the vector containing the five
local parameters $\muz{1}$, $\muz{2}$, $\sigmaz{1}$, $\sigmaz{2}$
and~$\rho$ will be denoted
by~$\bm{\theta} = \bm{\theta}(\LGp)$,\footnote{The vector
  $\bm{\theta}$ is a function of the point~$\LGp$, but this will
  henceforth be suppressed in the~notation.} and the approximating
bivariate Gaussian density function at the point $\LGp$ will be
denoted~$\psi(\bm{w};\bm{\theta})$, i.e.\
\begin{equation}
  \label{eq:psi_density}
  \psi(\bm{w}; \bm{\theta}) \defeq \tfrac{1}{2\pi \cdot     \sigmaz{1} \sigmaz{2} \sqrt{1 -
      \rhoz[2]{}}}   \exp \parenC{- \tfrac{      \sigmaz[2]{1} \parenRz[2]{}{\wz{1} - \muz{1}}      - 2\sigmaz{1}\sigmaz{2}\rho       \left(\wz{1} - \muz{1} \right)      \left(\wz{2} - \muz{2} \right)      + \sigmaz[2]{2} \parenRz[2]{}{\wz{2} - \muz{2}}    }{      2\sigmaz[2]{1} \sigmaz[2]{2} \left( 1 -
        \rhoz[2]{} \right) } }.
\end{equation}

In order for $\psi(\bm{w};\bmthetaz{})$ to be considered a good
approximation of $g(\bm{w})$ in a neighbourhood of the point $\LGp$,
it should at least coincide with $g(\bm{w})$ at $\LGp$, and it
furthermore seems natural to require that the tangent planes should
coincide too, i.e.\
\begin{subequations}
  \label{eq:equality_g_and_psi_intro}
\begin{align}
  \label{eq:equality_g_and_psi_intro_0-part}
  &g(\LGp) = {\psi(\LGp;\bmthetaz{})}, \\
  \label{eq:equality_g_and_psi_intro_1-part}
  &\frac{\partial}{\partial \wz{1}} g(\LGp) =   {\frac{\partial}{\partial \wz{1}} 
    \psi(\LGp;\bmthetaz{})} \text{\ and\ }   \frac{\partial}{\partial \wz{2}} g(\LGp) =   {\frac{\partial}{\partial \wz{2}}
    \psi(\LGp;\bmthetaz{})}.
\end{align}
\end{subequations}

It is easy to verify analytically that a solution $\bmthetaz{}$ can be
found for any point $\LGp$ where $g(\bm{w})$ is smooth~--- but these
solutions are not unique: $\psi(\bm{w};\bmthetaz{})$ and
$\psi(\bm{w};\bmthetaz[']{})$ can have the same first order
linearisation around the point~$\LGp$, without $\bmthetaz{}$ being
identical to~$\bmthetaz[']{}$.  It is tempting to extend
\cref{eq:equality_g_and_psi_intro} to also include similar
requirements for the second order partial derivatives, but the system
of equations will then in general have no solution.

This shows that it, in order to find the local Gaussian parameters in
$\bmthetaz{}$, is insufficient to only consider requirements
at~$\LGp$, it is necessary to apply an argument that also takes into
account a neighbourhood around~$\LGp$.  Applying the approach used
when estimating densities in \citet{hjort96:_local}, one can consider
a \mbox{$\blimit$} limit of parameters
$\thetah{\bm{b}} = \thetah{\bm{b}}(\LGp)$ that minimise the penalty
function
\begin{align}
  \label{eq:locally_weighted_Kullback_Leibler_intro}
  \qz{\bm{b}} = \int \Kh[\bm{w}-\LGp]{\bm{b}}   \left[\psi(\bm{w};\bm{\theta}) -   g(\bm{w})\log\left(\psi(\bm{w};\bm{\theta})\right)   \right]\d{\bm{w}},
\end{align}
where~$\Kh[\bm{w}-\LGp]{\bm{b}}$ is a kernel function with
bandwidth~$\bm{b}$.  As explained in
\citet[Section~2.1]{hjort96:_local}, this can be interpreted as a
locally weighted Kullback-Leibler distance between the targeted
density~$g(\bm{w})$ and the approximating
density~$\psi(\bm{w};\bm{\theta})$.
An optimal parameter configuration~$\thetah{\bm{b}}$ for
\cref{eq:locally_weighted_Kullback_Leibler_intro} should solve the
vector~equation
\begin{align}
  \label{eq:locally_weighted_Kullback_Leibler_equation_intro}
  \int \Kh[\bm{w}-\LGp]{\bm{b}}\bm{u}(\bm{w};\bm{\theta})  \left[\psi(\bm{w};\bm{\theta}) - g(\bm{w})\right]\d{\bm{w}} =
  \bm{0},
\end{align}
where~\mbox{$\bm{u}(\bm{w};\bm{\theta}) \defeq
  \tfrac{\partial}{\partial \bm{\theta}}
  \log\left(\psi(\bm{w};\bm{\theta}) \right)$} is the score function
of the approximating density~$\psi(\bm{w};\bm{\theta})$.  There will,
under suitable assumptions \citet{hjort96:_local,Tjostheim201333}, be
a unique limiting solution of
\cref{eq:locally_weighted_Kullback_Leibler_equation_intro}, i.e.
\begin{equation}
  \label{eq:theta0_limit}
  \bmthetaz{0} = \bmthetaz{0}(\LGp) =
  \lim_{\blimit} \thetah{\bm{b}}(\LGp)
\end{equation}
will be well-defined,\footnote{The solution $\bmthetaz{0}$ will always
  satisfy \cref{eq:equality_g_and_psi_intro_0-part}, but it will in
  general not satisfy \cref{eq:equality_g_and_psi_intro_1-part}.}  and
the $\rho$-part of the $\bmthetaz{0}$-vector can be used to define a
LGC at the point~$\LGp$.

For the special case where $g(\bm{w})$ is a bivariate normal
distribution, i.e.\ when
\begin{equation}
  \label{eq:W_bivariate_Gaussian}
  \bm{W} \sim \operatorname{N}\!\left(      \left[        \begin{matrix}
        \muz{1} \\
        \muz{1}
      \end{matrix}
    \right], 
    \left[         \begin{matrix}
        \sigmaz[2]{1}           & \sigmaz{1}\sigmaz{2}\rho \\
        \sigmaz{1}\sigmaz{2}\rho           & \sigmaz[2]{2}
      \end{matrix}
    \right] 
  \right), 
\end{equation}
then, for any point~$\LGp$ and any bandwidth~$\bm{b}$, the parameters
$\thetah{\bm{b}}$ that gives the optimal solution of
\cref{eq:locally_weighted_Kullback_Leibler_equation_intro} will be the
parameters given in~\cref{eq:W_bivariate_Gaussian}.  The
limit~$\bmthetaz{0}$ in \cref{eq:theta0_limit} will thus of course
also be these parameters, which implies that the LGC coincides with
the global parameter~$\rho$ at all points in the Gaussian case.  The
interested reader should consult \citet[p.~33]{Tjostheim201333} for
further details/remarks that motivates the use of the~LGC.
  
An estimate of the local Gaussian parameters $\bmthetaz{0}(\LGp)$ in
\cref{eq:theta0_limit} can, for a given bivariate sample
$\TSR{\bmWz{t}}{t=1}{n}$ and some reasonable bandwidth $\bm{b}$, be
found as the parameter-vector $\widehatbmthetaz{\bm{b}}(\LGp)$
that maximises the local log-likelihood\footnote{Confer
  \cref{app:bivariate_penalty_functions} in the supplementary material
  for a detailed exposition.}
\begin{align}
  \label{eq:LN_intro}
  \Lz{n}(\bmthetaz{})
  &\defeq
    n\inv \sumss{t=1}{n}
    \Kz{\bm{b}}\!\parenR{\bmWz{t}-\LGp}\log\psi\!\parenR{\bmWz{t};\bmthetaz{}} 
    - \intss{\RRn{2}}{} \Kz{\bm{b}}\!\parenR{\bm{w}-\LGp}\psi\!\parenR{\bm{w};\bmthetaz{}}  \d{\bm{w}}.
\end{align}

The asymptotic behaviour of $\widehatbmthetaz{\bm{b}}(\LGp)$ (as
$\nlimit$ and $\blimit$) is in \citet{Tjostheim201333} investigated by
entities derived from a local penalty function $\Qz{n}(\bmthetaz{})$
defined as $-n\cdot\Lz{n}(\bmthetaz{})$, i.e.\
\begin{align}
  \label{eq:QhN_paper}
  \Qz{n}(\bmthetaz{}) 
  &= - \sumss{t=1}{n}
    \Kz{\bm{b}}\!\parenR{\bmWz{t}-\LGp}\log\psi\!\parenR{\bmWz{t};\bmthetaz{}} 
    + n \intss{\RRn{2}}{} \Kz{\bm{b}}\!\parenR{\bm{w}-\LGp}\psi\!\parenR{\bm{w};\bmthetaz{}}  \d{\bm{w}}.
\end{align}
The key ingredient in the analysis is the corresponding vector of
partial derivatives, 
\begin{align}
  \label{eq:QhN_derivatives_paper}
  \nabla\Qz{n}(\bmthetaz{})
  &=  - \sumss{t=1}{n}
    \left[\Kz{\bm{b}}\!\parenR{\bmWz{t}-\LGp}\bm{u}\!\parenR{\bmWz{t};\bmthetaz{}}
    - \intss{\RRn{2}}{}
    \Kz{\bm{b}}\!\parenR{\bm{w}-\LGp}\bm{u}\!\parenR{\bm{w};\bmthetaz{}}\psi\!\parenR{\bm{w};\bmthetaz{}}
    \d{\bm{w}} \right],
\end{align}
and, as will be seen later on, the asymptotic investigation of the
local Gaussian spectral density $\lgsd{\LGp}{\omega}$ introduced in
this paper does also build on this entity.

Notice that the bias-variance balance of the estimate
$\widehatbmthetaz{\bm{b}}(\LGp)$ depends on the bandwidth-vector
$\bm{b}$, and an estimate based on a $\bm{b}$ too close to
$\bmzeroz{}$ might thus be dubious.  However, it can still be of
interest (for a given sample) to compare estimates
$\widehatbmthetaz{\bm{b}}(\LGp)$ for different scales of $\bm{b}$ in
order to see how they behave.

Since the goal is to estimate $\bmthetaz{0}(\LGp)$, it is of course
important to find $\widehatbmthetaz{\bm{b}}(\LGp)$ for not too large
bandwidth-vectors $\bm{b}$~--- but it might still be of interest to
point out how \cref{eq:LN_intro} behaves in the \enquote{global limit
  $\bm{b}\rightarrow \bm{\infty}=(\infty,\infty)$}.  In this case the
second term goes to zero, and the parameter-vector
$\widehatbmthetaz{\bm{\infty}}(\LGp)$ that maximises the first term
becomes the ordinary (global) least squares estimates of a global
parameter vector $\bmthetaz{}$ which contains the ordinary means,
variances and correlation.

\subsubsection{Local Gaussian correlation, normalised version}
\label{sec:LGC_revised}

The algorithm that estimates the LGC (see
\citet{Berentsen2014:departure_from_independence} for an
\R-implementation) can run into problems if the data under
investigation contains outliers --- i.e.\ the numerical convergence
might not succeed for points~$\LGp$ in the periphery of the data.  It
is possible to counter this problem by removing the most extreme
outliers, but an alternative strategy based on normalisation will be
applied instead.

The key observation is that the numerical estimation problem does not
occur when the marginal distributions are standard normal - which
motivates an adjusted strategy similar to the copula-concept
from~\citet{sklar1959:_copula}.  Sklar's theorem gives the existence
of a copula~$C(\uz{1},\uz{2})$ such that the joint cdf $G(\bm{w})$ can
be expressed as
\mbox{$C\!\left(\Gz{1}\!\left(\wz{1}\right)\!,
    \Gz{2}\!\left(\wz{2}\right)\right)$}, with
$\Gz{i}\!\left(\wz{i}\right)$ the marginal cdf corresponding to
$\Wz{i}$.  This copula~$C$ contains all the interdependence
information between the two marginal random variables $\Wz{1}$ and
$\Wz{2}$, it will be unique when the two margins are continuous, and
it will then be invariant under strictly increasing transformations of
the margins.\footnote{For a proof of this statement, see e.g.\
  \citet[Theorem~2.4.3]{Nelsen2006}.} Under this continuity
assumption, the random variable
\mbox{$\bm{W} = \left(\Wz{1},\Wz{2}\right)$} will have the same copula
as the transformed random variable
\mbox{$\bm{Z} \defeq \left(\Phiz[-1]{}\!\left(
      \Gz{1}\!\left(\Wz{1}\right)\right), \Phiz[-1]{}\!\left(
      \Gz{2}\!\left(\Wz{2}\right)\right)\right)$}, where $\Phi$ is the
cdf of the standard normal distribution~--- whose corresponding pdf as
usual will be denoted by~$\phi$.\footnote{ See
  \citet{Berentsen2014:Copula} for an approach where this is used to
  construct a \textit{canonical local Gaussian correlation} for the
  copula~$C$.}  This transformed version of $\bm{W}$ has standard
normal margins, so the LGC-estimation algorithm will not run into
numerical problems --- which motivates the following alternative
approach to the definition of~LGC: Instead of finding a Gaussian
approximating of the pdf $g(\bm{w)}$ (of the original random
variable~$\bm{W}$) at a point~$\LGp$, find a Gaussian approximation of
the pdf $\gz{\bm{Z}}(\bm{z})$ of the transformed random
variable~$\bm{Z}$ at a transformed point~$\LGpz{\bm{Z}}$.  Expressed
relative to the pdf~$c$ of the copula~$C$, this means that the setup
in \cref{eq:LGC_normalised_after} below will be used instead of the
setup in \cref{eq:LGC_normalised_before}.
\begin{subequations}
  \label{eq:LGC_normalised}
  \begin{align}
    \label{eq:LGC_normalised_before}
    g(\bm{w)}     & = c\!\left(\Gz{1}\!\left(\wz{1}\right),       \Gz{2}\!\left(\wz{2}\right)\right)       \gz{1}\!\left(\wz{1}\right)      \gz{2}\!\left(\wz{2}\right) 
    & \text{approximate at }     &  \LGp=\LGpoint, \\
    \label{eq:LGC_normalised_after}
    \gz{\bm{Z}}(\bm{z})     & = c\!\left(      \Phi\!\left(\zz{1}\right),       \Phi\!\left(\zz{2}\right)\right)       \phi\!\left(\zz{1}\right)      \phi\!\left(\zz{2}\right)     & \text{approximate at }     &  \LGpz{\bm{Z}}\defeq
    \left(      \Phiz[-1]{}\!\left(      \Gz{1}\!\left(\LGpi{1}\right)\right),       \Phiz[-1]{}\!\left(      \Gz{2}\!\left(\LGpi{2}\right)\right)
    \right).
  \end{align}
\end{subequations}

The normalised version of the LGC will return values that differ from
those obtained from the general LGC-version introduced in
\cref{sec:LGC}, but the two versions coincide when the random variable
$\bm{W}$ is bivariate Gaussian.  The transformed random
variable~$\bm{Z}$ corresponding to the~$\bm{W}$
from~\cref{eq:W_bivariate_Gaussian} will then be
\mbox{$\bm{Z} = \left(    \left(\Wz{1} - \muz{1} \right)/\sigmaz{1},     \left(\Wz{2} - \muz{2} \right)/\sigmaz{2} \right)$}, which~implies
\begin{equation}
  \label{eq:Y_bivariate_Gaussian}
  \bm{Z} \sim \operatorname{N}\!\left(    \left[    \begin{matrix}
      0 \\
      0
    \end{matrix}
     \right], 
     \left[        \begin{matrix}
         1          & \rho \\
         \rho          & 1
       \end{matrix}
     \right] 
   \right), 
\end{equation}
so the normalised LGC will thus also coincide with the
global parameter~$\rho$ at all~points.

The convergence rate for the estimates is rather slow for the LGC
cases discussed above (it is
$\sqrt{n\!\subp{\parenR{\bz{1}\bz{2}}}{}{}{}{3}}$), and that is due to
the kernel function~$\Kz{\bm{b}}$ in
\cref{eq:locally_weighted_Kullback_Leibler_intro}.  Briefly
summarised, the \mbox{$5\times5$} covariance matrix of the estimate
$\widehatbmthetaz{\bm{b}}$ will have the form
\mbox{$\Vz[-1]{\bm{b}}\Wz{\bm{b}}\Vz[-1]{\bm{b}}$}, the presence of
the kernel~$\Kz{\bm{b}}$ means that the matrices $\Vz{\bm{b}}$ and
$\Wz{\bm{b}}$ have rank one in the limit $\blimit$, and this slows
down the convergence rate, cf.\ \citet[Th.~3]{Tjostheim201333} for the
details.

The property that the limiting matrices have rank one does not pose a
problem if only one parameter is estimated,\footnote{The matrices
  then becomes \mbox{$1\times1$}, so the singularity problems does not
  occur.} and the convergence rate would then be much faster (i.e.\
$\sqrt{n\bz{1}\bz{2}}$).  Inspired by the fact that the transformed
random variable~$\bm{Z}$ have standard normal margins, it has been
introduced a simplified normalised version of the LGC where only the
$\rho$-parameter should be estimated when using the approximation
approach from \cref{eq:LGC_normalised_after}, i.e.\ the values of
$\muz{1}$, $\muz{2}$ are taken to be~0, whereas $\sigmaz[2]{1}$ and
$\sigmaz[2]{2}$ are taken to be~1.  This simplified approach has been
applied successfully with regard to density estimation\footnote{Note
  that it is \textit{not} the local Gaussian correlation that is the
  target of interest when this simplified approach is used for density
  estimation.}
in \citet{otneim2017locally,otneim2018conditional}, but for the local
spectrum analysis considered in this paper it gave inferior results
--- and this paper will thus not include any plots based on the
normalised one-parameter version.\footnote{
  The theory for the normalised one-free-parameter version of LGC is
  avaialbe in the first authors PhD-thesis,
  \myblind{\url{https://bora.uib.no/handle/1956/16950}.}{blinded url.}
  This also contains a discussion with regard to why an approach based
  on the normalised one-free-parameter approach fails to produce
  decent results.}

\makeatletter{}

\subsection{The local Gaussian spectral densities}
\label{sec:LGSD_definition}

{
  \label{pR2bp3}
  An extension of the spectral density $f(\omega)$ from
  \cref{eq:basic_spectral_density_rho} can be based on any of the
  three LGC-versions mentioned in \cref{sec:LGC,sec:LGC_revised}.  The
  one presented below is based on the normalised five-parameter local
  Gaussian autocorrelation, since that ensures that the estimation
  algorithm avoids the aforementioned numerical convergence problems
  --- but the theory developed in the Supplementary Material does also
  cover the general situation.
}

\begin{definition}
  \label{def:LGSD_def_main_part}
  The local Gaussian spectral density (LGSD), at a
  point,~\mbox{$\LGp=\LGpoint$}, for a strictly stationary univariate
  time series $\TSR{\Yz{t}}{t\in\ZZ}{}$ is constructed in the
  following manner.
  \begin{enumerate}[label=(\alph*)]
  \item
    \label{def:LGSD_def_main_part_normalisation}
    With $G$ the univariate \textit{marginal} cumulative distribution
    of $\TSR{\Yz{t}}{t\in\ZZ}{}$, and $\Phi$ the cumulative
    distribution of the standard normal distribution, define a
    normalised version $\TSR{\Zz{t}}{t\in\ZZ}{}$ of
    $\TSR{\Yz{t}}{t\in\ZZ}{}$ by
    \begin{align}
      \label{eq:LGSD_def_main_part_normalisation}
      \TSR{\Zz{t} \defeq
      \Phiz[-1]{}\!\parenR{G\!\parenR{\Yz{t}}}}{t\in\ZZ}{}.
    \end{align}
  \item     For a given point \mbox{$\LGp=\LGpoint$} and for each
    \textit{bivariate pair}
    \mbox{$\bmZz{h:t}\defeq \parenR{\Zz{t+h},\Zz{t}}$},
        a \textit{local Gaussian
      autocorrelation} $\lgacr[p]{\LGp}{h}$ can be computed.  The convention
    \mbox{$\lgacr[p]{\LGp}{0}\equiv 1$} is used when \mbox{$h=0$}.
  \item     \label{def:LGSD_requirement}
        When
    \mbox{$\sumss{h\in\ZZ}{} \parenAbs{\lgacr[p]{\LGp}{h}} < \infty$},
    the \textit{local Gaussian spectral density} at the point~$\LGp$
    is defined~as
    \begin{align}
      \label{eq:lgsd_definition_main_document}
      \lgsd[p]{\LGp}{\omega} \defeq \sumss{h=-\infty}{\infty}
      \lgacr[p]{\LGp}{h} \cdot \ez[-2\pi i\omega h]{}.
    \end{align}
  \end{enumerate}
\end{definition}

Notice that the requirement
\mbox{$\sumss{h\in\ZZ}{} \parenAbs{\lgacr[p]{\LGp}{h}} < \infty$} in
\myref{def:LGSD_def_main_part}{def:LGSD_requirement} implies that the
concept of local Gaussian spectral density in general might not be
well defined for all stationary time series $\TSR{\Yz{t}}{t\in\ZZ}{}$
and all points \mbox{$\LGp\in\RRn{2}$}.

{
  \label{pR2bp1_marginal}
  The normalisation in \cref{eq:LGSD_def_main_part_normalisation}
  preserves the copula-structure of the original time series, but a
  standard normal marginal will be used instead of its original
  marginal distribution.  This implies that the transformed time
  series will have all moments, even though that might not be the case
  for a tick tailed original time series.  A local Gaussian
  investigation of the normalised time series can detect non-Gaussian
  dependency structures in the original time series, but keep in mind
  that an investigation of the original marginal might also be of
  interest in many situations, e.g.\ with regard to discriminant
  analysis.
  
  Finally, note that the normalisation in
  \label{pR1MC2}
  \cref{eq:LGSD_def_main_part_normalisation} can be compared to, but
  is very different from, the normalization in
  \citet{kl94:_some_limit_theor_self_normal}.  }

The following definition of time reversible time series, from
\citet[def.~4.6]{tong1990non}, is needed
in~\myref{th:lgsd_properties}{th:lgsd_real_when_reversible_time_series}.

\begin{definition}
  \label{def:Yt_reversible}
  A stationary time series $\TSR{\Yz{t}}{t\in\ZZ}{}$ is time
  reversible if for every positive integer~$n$ and every
  \mbox{$\tz{1},\tz{2},\dotsc,\tz{n} \in \ZZ$}, the vectors
  \mbox{$\parenR{\Yz{\tz{1}},\Yz{\tz{2}}\dotsc,\Yz{\tz{n}}}$} and
  \mbox{$\parenR{\Yz{-\tz{1}},\Yz{-\tz{2}}\dotsc,\Yz{-\tz{n}}}$} have
  the same joint distributions.
\end{definition}

{\begin{lemma}
  \label{th:lgsd_properties}
  The following properties holds for $\lgsd[p]{\LGp}{\omega}$.
  \begin{enumerate}[label=(\alph*)]
  \item \label{th:lgsd_equal_to_osd_when_Gaussian}
    $\lgsd[p]{\LGp}{\omega}$ coincides with $f(\omega)$ for all
    \mbox{$\LGp\in\RRn{2}$} when $\TSR{\Yz{t}}{t\in\ZZ}{}$ is a
    Gaussian time series, or when $\TSR{\Yz{t}}{t\in\ZZ}{}$ consists
    of i.i.d.\ observations.
  \item     \label{th:lgsd_reflection_property}
        The following holds when \mbox{$\LGpd\defeq\LGpointd$} is the
    diagonal reflection of \mbox{$\LGp=\LGpoint$};
    \begin{subequations}
      \label{eq:th:lgsd_properties}
      \begin{align}
        \label{eq:th:lgsd_folding_property}
        \lgsd[p]{\LGp}{\omega}         &= 1 + \sumss{h=1}{\infty}
          \lgacr[p]{\LGpd}{h} \cdot \ez[+2\pi i\omega h]{} 
          + \sumss{h=1}{\infty}
          \lgacr[p]{\LGp}{h}  \cdot \ez[-2\pi i\omega h]{}, \\
                \label{eq:th:lgsd_conjugate_property}
        \lgsd[p]{\LGp}{\omega}         &= \overline{\lgsd[p]{\LGpd}{\omega}}.
      \end{align}
    \end{subequations}
  \item     \label{th:lgsd_real_when_reversible_time_series}
        When $\TSR{\Yz{t}}{t\in\ZZ}{}$ is time reversible, then
    $\lgsd[p]{\LGp}{\omega}$ is real valued for all
    \mbox{$\LGp\in\RRn{2}$}, i.e.\
    \begin{align}
      \label{eq:th:lgsd_reflection_property}
      \lgsd[p]{\LGp}{\omega}       &= 1 + 2\cdot \sumss{h=1}{\infty}
        \lgacr[p]{\LGp}{h} \cdot \cos(2\pi\omega h).
    \end{align}
  \item     \label{th:lgsd_real_on_diagonal}
        $\lgsd[p]{\LGp}{\omega}$ will in general be complex-valued, but it
    will always be real valued when the point $\LGp$ lies on the
    diagonal, i.e.\ when \mbox{$\LGpdiagonal$}.
    \Cref{eq:th:lgsd_reflection_property} will hold in this diagonal
    case too.
  \end{enumerate}
\end{lemma}}

\begin{proof}
  \Cref{th:lgsd_equal_to_osd_when_Gaussian} follows for the Gaussian
  case since the local Gaussian autocorrelations $\lgacr[p]{\LGp}{h}$
  by construction coincides with the ordinary (global)
  autocorrelations $\rho(h)$ in the Gaussian case.  Similarly, when
  $\TSR{\Yz{t}}{t\in\ZZ}{}$ consists of i.i.d.\ observations, then
  both local and global autocorrelations will be~0 when
  \mbox{$h\neq0$}, and the local and global spectra both become the
  constant function~1.
  \Cref{th:lgsd_reflection_property,th:lgsd_real_when_reversible_time_series,th:lgsd_real_on_diagonal}
  are trivial consequences of the diagonal folding property from
  \cref{th:diagonal_folding_property}, i.e.\
  \mbox{$\lgacr[p]{\LGp}{-h} = \lgacr[p]{\LGpd}{h}$}, and the
  definition of time reversibility, see
  \cref{app:_diagonal_argument,app:_time_reversible} for details.
\end{proof}

For general points \mbox{$\LGp=\LGpoint$}, the complex valued result
of $\lgsd[p]{\LGp}{\omega}$ might be hard to investigate and
interpret~--- but, due to
\myref{th:lgsd_properties}{th:lgsd_real_on_diagonal}, the
investigation becomes simpler for points on the diagonal.
This might also be the situation of most practical interest, since it
corresponds to estimating the local spectrum at (or around) a given
value of $\TSR{\Yz{t}}{t\in\ZZ}{}$ --- such as a certain quantile for
the distribution of $\Yz{t}$.
The real valued results $\lgsd[p]{\LGp}{\omega}$ for
$\LGp$ along the diagonal can be compared with the result of the
ordinary (global) spectral density $f(\omega)$, as given in
\cref{eq:basic_spectral_density_rho}, and this might detect cases
where the times series $\TSR{\Yz{t}}{t\in\ZZ}{}$ deviates from
\textit{being Gaussian}.  Furthermore, if the global spectrum
$f(\omega)$ is flat, then any peaks and troughs of
$\lgsd[p]{\LGp}{\omega}$ might be interpreted as indicators of e.g.\
\textit{periodicities at a local level}.  This implies that estimates
of $\lgsd[p]{\LGp}{\omega}$ might be useful as an exploratory tool, an
idea that will be pursued in~\cref{sec:Examples}.

Note that the collection of local Gaussian autocorrelations
$\TSR{\lgacr[p]{\LGp}{h}}{h\in\ZZ}{}$ might not be non-negative
definite, which implies that both the theoretical and estimated local
Gaussian spectral densities might therefore become negative.  However,
as the artificial process investigated in \cref{fig:trigonometric}
(page \pageref{fig:trigonometric}) shows, the peaks of
$\lgsd[p]{\LGp}{\omega}$ still occur at the expected frequencies for
the investigated points --- which implies that the lack of
non-negativity does not prevent this tool from detecting nonlinear
structures in non-Gaussian white noise.

The following definition is needed when the discussion later on
refers to $m$-truncated versions of the different spectra.

\begin{definition}
  \label{def:m_truncated_lgsd_and_sd}
  The \mbox{$m$-truncated} versions $\lgsdM[p]{\LGp}{\omega}{m}$ and
  $\fz[m]{}(\omega)$ of $\lgsd[p]{\LGp}{\omega}$ and $f(\omega)$, for
  some lag-window function $\lambdaz{m}(h)$, is defined by means of
  \begin{subequations}
    \label{eq:def:m_truncated_lgsd_and_sd}
    \begin{align}
      \label{eq:def:m_truncated_lgsd}
      \lgsdM[p]{\LGp}{\omega}{m} 
      &\defeq 1 +         \sumss{h=1}{m}
        \lambdazM{h}{m}\cdot \lgacr[p]{\LGpd}{h}
        \cdot \ez[+2\pi i\omega h]{} +         \sumss{h=1}{m}
        \lambdazM{h}{m}\cdot \lgacr[p]{\LGp}{h} \cdot
        \ez[-2\pi i\omega h]{}, \\
      \label{eq:def:m_truncated_sd}
      \fz[m]{}(\omega)
      &\defeq \sumss{h=-m}{m} \lambdazM{h}{m}\cdot \rho(h) \cdot
        \ez[-2\pi i\omega h]{}.
    \end{align}
  \end{subequations}
\end{definition}

\subsection{Estimation}
\label{sec:Estimation}

Theoretical and numerical estimates of the ordinary spectral density
$f(\omega)$ is typically investigated by means of the fast Fourier
transform (FFT) and techniques related to the periodogram.  The
FFT-approach can not be used in the local case since there is no
natural factorisation of terms making up a local estimated covariance,
but there does exist a pre-FFT approach for the estimation of
$f(\omega)$,
where a Fourier transform is taken of the estimated autocorrelations
after they have been smoothed and truncated by means of some
lag-window function~--- and the pre-FFT approach can be adapted to deal
with the estimates of the local Gaussian spectral~densities.

\newtheorem{algorithm}[theorem]{Algorithm}

\begin{algorithm}
  \label{def:lgsd_estimator}
  For a sample $\TSR{\yz{t}}{t=1}{n}$ of size~$n$, an
  \mbox{$m$-truncated} estimate $\hatlgsdM[p]{\LGp}{\omega}{m}$ of
  $\lgsd[p]{\LGp}{\omega}$ is constructed by means of the following
  procedure.
  \begin{enumerate}[label=(\alph*)]
  \item     \label{def:lgsd_estimator_ecdf}
    Find an estimate $\widehatGz{\!n}$ of the marginal cumulative
    distribution function, and compute the \textit{pseudo-normalised
      observations}
    $\TSR{\widehatzz{t} \defeq
      \Phiz[-1]{}\!\parenR{\widehatGz{\!n}\!\parenR{\yz{t}}}}{t=1}{n}$
    that corresponds to $\TSR{\yz{t}}{t=1}{n}$.
  \item     \label{def:lgsd_estimator_pseudo_normalisation}
        Create the lag~$h$ pseudo-normalised pairs
    $\TSR{\parenR{\widehatzz{t+h},\widehatzz{t}}}{t=1}{n-h}$ for
    \mbox{$h=1,\dotsc,m$}, and estimate, both for the point
    \mbox{$\LGp=\LGpoint$} and its diagonal reflection
    \mbox{$\LGpd=\LGpointd$}, the local Gaussian autocorrelations
    $\TSR{\hatlgacrb[p]{\LGp}{h}{\bmbzh{h}}}{h=1}{m}$ and
    $\TSR{\hatlgacrb[p]{\LGpd}{h}{\bmbzh{h}}}{h=1}{m}$, where the
    $\TSR{\bmbzh{h}}{h=1}{m}$ is the bandwidths used during the
    estimation of the local Gaussian autocorrelation for the different
    lags.  
          \item     \label{def:lgsd_esitimator_folded}
        Adjust \cref{eq:th:lgsd_folding_property} from
    \myref{th:lgsd_properties}{th:lgsd_reflection_property} with some
    lag-window function $\lambdazM{h}{m}$ to get the estimate
    \begin{align}
      \label{eq:App:hatlgsd_definition_main_document}
      \hatlgsdM[p]{\LGp}{\omega}{m} \defeq 1 +       \sumss{h=1}{m}
      \lambdazM{h}{m}\cdot \hatlgacrb[p]{\LGpd}{h}{\bmbzh{h}}
      \cdot \ez[+2\pi i\omega h]{} +       \sumss{h=1}{m}
      \lambdazM{h}{m}\cdot \hatlgacrb[p]{\LGp}{h}{\bmbzh{h}} \cdot
      \ez[-2\pi i\omega h]{}.
    \end{align}
  \end{enumerate}
\end{algorithm}

{
  \label{pR1MC3}
  The presence of the kernel $\Kh[\bm{w}-\LGp]{\bm{b}}$ in
  \cref{eq:locally_weighted_Kullback_Leibler_intro} implies that
  {small sample effects} can occur when the local Gaussian spectrum
  $\lgsd[p]{\LGp}{\omega}$ is estimated for some combinations of
  points $\LGp$ and bandwidths $\bm{b}$ --- and this can in particular
  be an issue if the points lie in the low density regions
  corresponding to the tails of our distribution.  Roughly speaking:
  When the bandwidth $\bm{b}$ becomes \enquote{too small}, then the
  estimated local Gaussian autocorrelations will have a tendency to
  approach either \enquote{$-1$} or \enquote{$+1$}, cf.\
  \cref{app:Bandwidth_sensitivity} --- and these estimates will then
  in general only reflect the random configuration of those lag-$h$
  pairs that happened to lie closest to the point $\LGp$.
  \Cref{sec:Examples_the_parameters} presents strategies that can be
  used in order to detect/avoid this issue, and additional details are
  presented in the Supplementary Material.

} 

The following result is an analogue to
\cref{eq:th:lgsd_reflection_property}
of \myref{th:lgsd_properties}{th:lgsd_real_when_reversible_time_series}
\begin{lemma}
  \label{th:lgsd_estimate_real_valued}
  When it is assumed that the sample $\TSR{\yz{t}}{t=1}{n}$ comes from
  a time reversible stochastic process $\TSR{\Yz{t}}{t\in\ZZ}{}$, the
  $m$-truncated estimate $\hatlgsdM[p]{\LGp}{\omega}{m}$ can for all
  points $\LGp\in\RRn{2}$ be written as
  \begin{align}
    \label{eq:hatlgsd_definition_simplified_main_document}
    \hatlgsdM[p]{\LGp}{\omega}{m} 
    &= 1 + 2\cdot \sumss{h=1}{m}
      \lambdazM{h}{m}\cdot \hatlgacrb[p]{\LGp}{h}{\bmbzh{h}}
      \cdot \cos(2\pi\omega h).
  \end{align}
  Moreover, \cref{eq:hatlgsd_definition_simplified_main_document} will
  always hold when the point $\LGp$ lies on the diagonal,
  i.e.\ \mbox{$\LGpdiagonal$}.
\end{lemma}

\begin{proof}
  This follows from
  \cref{th:lgsd_real_when_reversible_time_series,th:lgsd_real_on_diagonal}
  of \cref{th:lgsd_properties}.
      \end{proof}

  \label{note:pseudo_normalisation}
  The estimated $\widehatGz{\!n}$ in
  \myref{def:lgsd_estimator}{def:lgsd_estimator_pseudo_normalisation}
  can e.g.\ be the rescaled empirical cumulative distribution function
  created from the sample $\TSR{\yz{t}}{t=1}{n}$
  (\label{pR2bp1_empirical_marginal}
  which transforms original data into ranks divided by $n+1$), or it
  could be based on some logspline technique like the one implemented
  in \citet{otneim2017locally}.

  \label{note:bandwidths_asymptotically_equivalent}
  The bandwidths \mbox{$\bmbzh{h} = \parenR{\bz{h1},\bz{h2}}$} in
  \myref{def:lgsd_estimator}{def:lgsd_estimator_pseudo_normalisation}
  does not need to be equal for all the lags~$h$ when an estimate
  $\hatlgsdM[p]{\LGp}{\omega}{m}$ is computed.  For the asymptotic
  investigation it is sufficient to require that $\bz{h1}$ and
  $\bz{h2}$ approach zero at the same rate, i.e.\ that there exists
  \mbox{$\bm{b}=\parenR{\bz{1},\bz{2}}$} such that
  \mbox{$\bz{hi}\asymp\bz{i}$} for \mbox{$i=1,2$} and
  for all~$h$ (that is to say, \mbox{$\lim \bz{hi}/\bz{i} = 1$}).

  \label{note:asympotitic_effect_of_pseudo_normalisation}
  The asymptotic theory for $\hatlgacrb[p]{\LGp}{h}{\bmbzh{h}}$, given
  that the required regularity conditions are satisfied, follows when
  the original argument from \citet{Tjostheim201333} is combined with
  the argument in \citet{otneim2017locally}.  The analysis in
  \citet{Tjostheim201333} considered the general case where the
  original observations $\TSR{\yz{t}}{t=1}{n}$ were used instead of
  the normalised observations
  $\TSR{\zz{t} \defeq \Phiz[-1]{}\parenR{G(\yz{t})}}{t=1}{n}$.  Since
  the cumulative density function $G$ in general will be unknown, the
  present asymptotic analysis must work with the pseudo-normalised
  observations $\TSR{\widehatzz{t}}{t=1}{n}$, which makes it necessary
  to take into account the difference between the true normalised
  values $\zz{t}$ and the estimated pseudo-normalised values
  $\widehatzz{t}$.  The analysis in \citet{otneim2017locally} implies
  that $\widehatGz{\!n}\!\parenR{\yz{t}}$ approaches
  $G\!\parenR{\yz{t}}$ at a faster rate than the rate of convergence
  for the estimated local Gaussian correlation, so (under some
  regularity conditions) the convergence rate of
  $\hatlgacrb[p]{\LGp}{h}{\bmbzh{h}}$ will thus not be affected by the
  distinction between $\zz{t}$ and $\widehatzz{t}$.  The present
  analysis will not duplicate the arguments related to this
  distinction, and the interested reader should consult
  \citet[Section~3]{otneim2017locally} for the details.

  The bias-variance balance for the estimates
  $\hatlgsdM[p]{\LGp}{\omega}{m}$ must consider the size of $m$
  relative to both $n$ and the bandwidths $\TSR{\bmbzh{h}}{h=1}{m}$,
  i.e.\ the kernel function reduces the number of observations that
  effectively contributes to the computations of the estimates~--- and
  that number of effective contributors can also depend on the
  location of the point $\LGp$, i.e.\ whether the point $\LGp$ lies at
  the center or in the periphery of the pseudo-normalised observations
  $\TSR{\parenR{\widehatzz{t+h},\widehatzz{t}}}{t=1}{n-h}$.  Confer
  \cref{sec:estimation_aspect_for_the_given_parameter_configuration}
  for further details.

\makeatletter{}

\Cref{fig:dmbp_original_normalised} shows the effect of the
pseudo-normalisation on the \texttt{dmbp}
example\footnote{\label{footnote:dmbp} This is the
  Deutschemark/British pound Exchange Rate (\texttt{dmbp}) data from
  \citet{bollerslev96:_period}, which is a common benchmark data set
  for GARCH-type models, and as such models are among the motivating
  factors for the study of the local Gaussian spectral density, it
  seems natural to test the method on \texttt{dmbp}.  The data plotted
  here was found in the \Rpackage \Rref{rugarch}, see
  \citet{ghalanos15_rugarch}, where the following description was
  given: \enquote{The daily percentage nominal returns computed as
    \mbox{$100\left[\ln\left(P_t\right) -
        \ln\left(P_t-1\right)\right]$}, where $P_t$ is the bilateral
    Deutschemark/British pound rate constructed from the corresponding
    U.S.\ dollar rates.}  } that will be discussed in
\cref{sec:Real_data}.  The uppermost part shows the original
\texttt{dmbp}-series (of length~1974) whereas the lowermost part shows
the pseudo-normalised transformation of it, and it is clear that the
shape of the pseudo-normalised version resembles the shape of the
original version.

\begin{knitrout}\footnotesize
\definecolor{shadecolor}{rgb}{0.969, 0.969, 0.969}\color{fgcolor}\begin{figure}[h]

{\centering \includegraphics[width=1\linewidth]{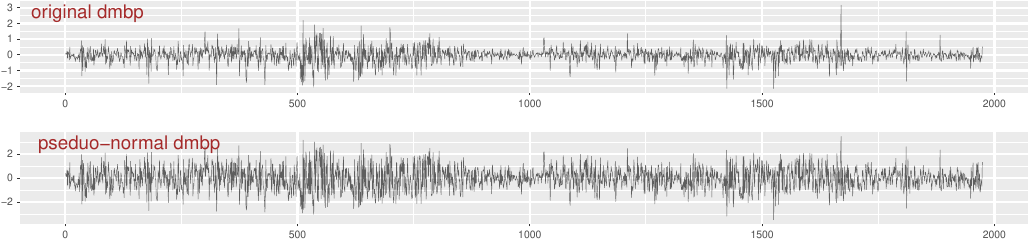} 

}

\caption[\texttt{dmbp}, original version and pseudo-normalised version]{\texttt{dmbp}, original version and pseudo-normalised version.}\label{fig:dmbp_original_normalised}
\end{figure}

\end{knitrout}

\makeatletter{}
\subsection{Asymptotic theory for $\hatlgsdM[p]{\LGp}{\omega}{m}$}
\label{sec:Asymptotic_theory}

This section presents asymptotic results for the cases where
$\hatlgsdM[p]{\LGp}{\omega}{m}$ are real-valued functions.  Note that
both assumptions and results are stated relative to the original
observations instead of the pseudo-normalised observations.  This
simplification does not affect the final convergence rates (see
earlier remarks, page
\pageref{note:asympotitic_effect_of_pseudo_normalisation},
for details)
and it makes the analysis easier.  The requirement that the LGSD
should be defined relative to the normalised observations is due to
computational issues, and the theoretical investigation shows that it
could just as well have been phrased in terms of the
original~observations.

\subsubsection{A definition and an assumption for $\Yz{t}$}
\label{app:assumptions_upon_Yt}

The assumption to be imposed on the univariate time series
$\TSR{\Yz{t}}{t\in\ZZ}{}$ is given in terms of components related to
the bivariate lag-$h$-pairs that can be constructed from it.  The
theoretical analysis of $\hatlgsdM[p]{\LGp}{\omega}{m}$ also requires
that \mbox{$(m+1)$-variate} pairs are considered.  Note that
\cref{def:lgsd_esitimator_folded} of \cref{def:lgsd_estimator} implies
that it is sufficient to only consider positive values
for~$h$.

\begin{definition}
  \label{def:Yht_and_YMt}
  For a strictly stationary univariate time series
  $\TSR{\Yz{t}}{t\in\ZZ}{}$, with \mbox{$h\geq1$} and \mbox{$m\geq2$},
  define bivariate and \mbox{$(m+1)$-variate} time~series as~follows,
  \begin{align}
    \label{eq:Yht_and_YMt}
    \Yht{h}{t} \defeq \Vector[\prime]{\Yz{t+h}, \Yz{t}},  \qquad
    \YMt{m}{t} \defeq \Vector[\prime]{\Yz{t+m}, \dotsc, \Yz{t}},
  \end{align}
  and let $\gh[\yh{h}]{h}$ and $\gM[\yM{m}]{m}$ denote the respective
  probability density functions.
\end{definition}

  \label{note:details_for_Yt_definition}
  The bivariate densities $\gh{h}$ can all be obtained from the
  \mbox{$(m+1)$} variate density $\gM{m}$ by integrating out the $m-1$
  redundant marginals, which in particular implies that if an
  \mbox{$(m+1)$-variate} function
  \mbox{$\tildeetahi[\yM{m}]{h}{} :\RRn{m+1}\rightarrow\RRn{1}$} is
  the \textit{obvious extension}\footnote{
    Consider the function to be a constant with respect to all the new
    variables that are introduced.}  of a bivariate function
  \mbox{$\etahi[\yh{h}]{h}{} :\RRn{2}\rightarrow\RRn{1}$}, then
  \begin{align}
    \label{eq:expectation_hierarchy}
    \E{\etahi[\Yht{h}{t}]{h}{}} =
    \E{\tildeetahi[\YMt{m}{t}]{h}{}},     \qquad \text{for } h \in \parenC{1,\dotsc,m}.
  \end{align}

  With the notation from \cref{def:Yht_and_YMt} the following
  \cref{assumption_Yt} can now be imposed on $\Yz{t}$.  Note that
  \cref{assumption:gh_bh0,assumption:gh_b0_infimum_of_bh0,assumption:h_x:y_x:y:z_finite_expectations}
  of \cref{assumption_Yt} contain references to definitions that
  first are given explicitly in \cref{sec:Technical_Results} in the
  Supplementary Material; these definitions are related to an
  \mbox{$(m+1)$}-variate penalty function for the time series
  $\YMt{m}{t}$\! ~--- and they are quite technical so it would impede
  the flow of the paper to include all the details here.  For the
  present section, it is sufficient to know that the new
  \mbox{$(m+1)$}-variate function can be expressed as a sum of $m$
  bivariate penalty-functions of the form given in
  \cref{eq:QhN_paper}.

  The key idea is that $\bmWz{t}$ and $g(\bm{w})$ in
  \cref{eq:locally_weighted_Kullback_Leibler_intro,eq:locally_weighted_Kullback_Leibler_equation_intro,eq:theta0_limit,eq:W_bivariate_Gaussian,eq:LN_intro,eq:QhN_paper,eq:QhN_derivatives_paper}
  are replaced with $\Yht{h}{t}$ and $\gh[\yh{h}]{h}$, which implies
  that an additional index $h$ must be added in order to keep track of
  the bookkeeping.  In particular, an inspection of
  \cref{eq:QhN_derivatives_paper} motivates the introduction of a
  random variable vector
  $\bmXz{h:t} =
  \Kz{\bm{b}}\!\parenR{\Yht{h}{t}-\LGp}\bm{u}\!\parenR{\Yht{h}{t};\bmthetaz{}}$,
  and the random variables $\XNht{n}{hq}{i}$ that occur in
  \myref{assumption_Yt}{assumption:h_x:y_x:y:z_finite_expectations}
  are the components of $\sqrt{\bz{1}\bz{2}}\bmXz{h:t}$.  Furthermore,
  notice that different combinations of the indices $h,i,j$ and $k$ in
  the product $\XNht{n}{hq}{i}\cdot\XNht{n}{jr}{k}$ implies that it
  can contain from two to four different terms of the time series
  $\TSR{\Yz{t}}{t\in\ZZ}{}$, so the corresponding density function can
  thus either be bi-, tri- or tetravariate.  The indices
  $q,r=1,\dotsc,5$ keep track of the appropriate derivatives of the
  5-dimensional parameter vector $\bm{\theta}$.  See
  \cref{def:Uh,def:Xht_*Lc} for details.

\begin{assumption}
  \label{assumption_Yt}
      The univariate process $\TSR{\Yz{t}}{t\in\ZZ}{}$ will be assumed to
  satisfy the following properties, with \mbox{$\LGp=\LGpoint$} in
  \cref{assumption:gh_differentiable_at_LGp} the point at which the
  estimate $\hatlgsdM[p]{\LGp}{\omega}{m}$ of $\lgsd[p]{\LGp}{\omega}$
  is to be computed.
  \begin{enumerate}[label=(\alph*)]
  \item     \label{assumption_Yt_strictly_stationary}
        $\TSR{\Yz{t}}{t\in\ZZ}{}$ is strictly stationary.
      \item     \label{assumption_Yt_strong_mixing}
        $\TSR{\Yz{t}}{t\in\ZZ}{}$ is strongly mixing, with
    mixing coefficient $\alpha(j)$ satisfying
    \begin{align}
      \label{eq:alpha_requirement}
      \sumss{j=1}{\infty} \jz[a]{}       \parenSz[1-2/\nu]{}{\alpha(j)} < \infty \qquad       \text{for some $\nu>2$ and $a > 1 - 2/\nu$}.
    \end{align}
      \item     \label{assumption_Yt_finite_variance}
        $\Var{\Yz{t}} < \infty$.
      \end{enumerate}
      
              The bivariate density functions $\gh[\yh{h}]{h}$ of 
  the lag~$h$ pairs $\Yht{h}{t}$ of the univariate time series
  $\TSR{\Yz{t}}{t\in\ZZ}{}$, must satisfy the following requirements
  for a given point \mbox{$\LGp=\LGpoint$}.
    \begin{enumerate}[label=(\alph*),resume]
  \item     \label{assumption:gh_differentiable_at_LGp}
        $\gh[\yh{h}]{h}$ is differentiable at $\LGp$, such that 
    Taylor's theorem can be used to write $\gh[\yh{h}]{h}$~as
    \begin{align}
      \label{eq_assumption:gh_differentiable_at_LGp}
      \gh[\yh{h}]{h}       & =         \gh[\LGp]{h} +         \power[\prime]{}{\bmmathfrakgz{h}(\LGp)}         \left[\yh{h}-\LGp \right] +         \power[\prime]{}{\bmmathfrakRz{h}(\yh{h})}         \left[\yh{h}-\LGp \right], \\
      \text{where } \bmmathfrakgz{h}(\LGp)       \nonumber
      &=         \Vector[\prime]{        \left.\tfrac{\partial}{\partial
        \yz{h}} \gh[\yh{h}]{h}\right|_{\yh{h} =\, \LGp},
        \left.\tfrac{\partial}{\partial \yz{0}} \gh[\yh{h}]{h}\right|_{\yh{h} =\, \LGp}}         \text{ and }         \lim_{\yh{h}\longrightarrow\, \LGp}         \bmmathfrakRz{h}(\yh{h})         = \bm{0},
                                                            \end{align}
    and the same requirement must also hold for the diagonally
    reflected point \mbox{$\LGpd=\parenR{\LGpi{2},\LGpi{1}}$}.
      \item     \label{assumption:gh_bh0}
        There exists a bandwidth $\bmbz{h0}$ such that there
    for every \mbox{$\bm{0} < \bm{b} < \bmbz{h0}$} is a
    unique minimiser $\thetahb{h}{b}$ of the penalty
    function~$\qh{h}{b}$ defined in \cref{eq:penalty_qh},
    which is obtained from
    \cref{eq:locally_weighted_Kullback_Leibler_intro} by putting
    $\bm{w}= \yh{h}$.
  \item     \label{assumption:gh_b0_infimum_of_bh0}
    The collection of bandwidths
    $\TSR{\bmbz{h0}}{h\in\ZZ}{}$ has a positive
    infimum, i.e.\ there exists a $\bmbz{0}$ such that
    $\bm{0} < \bmbz{0}       \defeq \inf_{h\in\ZZ} \bmbz{h0}$,
    which implies that this $\bmbz{0}$ can be used
    simultaneously for all the lags.
  \item     \label{assumption:h_x:y_x:y:z_finite_expectations}
    For $\XNht{n}{hq}{i}$ from \cref{def:Xht_*Lc},
    the bivariate, trivariate and tetravariate
    density functions must be such that the expectations
    $\E{\XNht{n}{hq}{i}}$, $\E{\absp[\nu]{\XNht{n}{hq}{i}}}$ and
    $\E{\XNht{n}{hq}{i}\cdot\XNht{n}{jr}{k}}$ all are
    finite.
  \end{enumerate}
\end{assumption}

These assumptions on $\Yz{t}$
are extensions of those used for the LGC-case in
\citet{Tjostheim201333}.
\Myref{assumption_Yt}{assumption_Yt_strong_mixing} is a bit more
general than the one used in \citet{Tjostheim201333}, but that is not
a problem since the arguments given there trivially extends to
the present case.

The $\alpha$-mixing requirement in
\cref{assumption_Yt_strong_mixing} ensures that $\Yz{t+h}$
and~$\Yz{t}$ will be asymptotically independent as $\hlimit$, i.e.\
the bivariate density functions $\gh[\yh{h}]{h}$ will for large
lags~$h$ approach the product of the marginal densities, and the
situation will thus stabilise when $h$ is large enough.  This is in
particular of importance for \cref{assumption:gh_b0_infimum_of_bh0},
since it implies that it will be possible to find a nonzero
$\bmbz{0}$ that works for all~$h$.

{
  {\label{PR1MC1}
    We do not consider the $\alpha$-mixing condition to be very
    strong.  In particular, note that GARCH type models, which are
    frequently used in econometrics, and also in the present paper,
    cf.\ \cref{sec:GARCH_model}, are $\beta$-mixing under weak
    conditions, see e.g.\ \citet{carrasco_chen_2002}; and $\beta$-mixing implies
    $\alpha$-mixing.}
}

The finiteness requirements in
\myref{assumption_Yt}{assumption:h_x:y_x:y:z_finite_expectations}
will be trivially satisfied if the densities are bounded, i.e.\ they
will then be consequences of properties of the kernel function
$\Kz{\bm{b}}$ and the score function of the bivariate Gaussian
distribution, see \cref{th:integrals_kernel_and_score_components}
for details.

\subsubsection{An assumption for $\Yz{t}$ and the
  score function $\bm{u}(\bm{w};\bm{\theta})$ of
  $\psi(\bm{w};\bm{\theta})$}
\label{sec:assumption_score_function}

The score function in
\cref{eq:locally_weighted_Kullback_Leibler_equation_intro}, i.e.\
\mbox{$\bm{u}(\bm{w};\bm{\theta}) \defeq \tfrac{\partial}{\partial
    \bm{\theta}} \log\left(\psi(\bm{w};\bm{\theta}) \right)$}, plays a
central role in the local density-estimation approach of
\citet{hjort96:_local}, and it 
also plays a pivotal role in the local Gaussian correlation theory
developed in \citet{Tjostheim201333}.

In particular, the convergence rate that in \citet{Tjostheim201333} is
given for ${\widehatbmthetaz{\LGp}-\bmthetaz{\LGp}}$ does implicitly
require that $\bm{u}(\LGp;\bmthetaz{\LGp})\neq\bm{0}$ in order for the
corresponding asymptotic covariance matrix to be well defined.  The
investigation of
$\parenR{\hatlgsdM[p]{\LGp}{\omega}{m} - \lgsd[p]{\LGp}{\omega}}$ in
this paper builds on the asymptotic results from
\citet{Tjostheim201333}, and the following assumption must 
be satisfied in order for the given convergence rates and asymptotic
variances to be valid.

\begin{assumption}
  \label{assumption_score_function}
  The collection of local Gaussian parameters
  $\TSR{\bmthetaz{\LGp}(h)}{}{}$ at the point $\LGp$ for the
  bivariate probability density functions $\gh[\yh{h}]{h}$, must all
  be such that
  \begin{enumerate}[label=(\alph*)]
  \item     \label{assumption_score_function_finite_h}
        $\bmuz{}(\LGp;\bmthetaz{\LGp}(h))\neq\bm{0}$ for
    all finite $h$.
  \item     \label{assumption_score_function_limit_h}
        $\lim \bmuz{}(\LGp;\bmthetaz{\LGp}(h))\neq\bm{0}$.
  \end{enumerate}
\end{assumption}

  It is, for a given time series $\Yz{t}$ and a given point $\LGp$,
  possible to inspect the 5 equations in
  \mbox{$\bmuz{}(\bm{w};\bmthetaz{})=\bm{0}$} in order to see when
  \cref{assumption_score_function_finite_h,assumption_score_function_limit_h}
  of \cref{assumption_score_function} might fail.
    For the case of the asymptotic requirement in
  \cref{assumption_score_function_limit_h}, the key observation is
  that the strong mixing requirement from
  \myref{assumption_Yt}{assumption_Yt_strong_mixing} implies that
  $\Yz{t+h}$ and $\Yz{t}$ will become independent when $\hlimit$.
  Together with the assumption of normalised marginals, this implies
  that the limit of $\bmthetaz{\LGp}(h)$ always becomes
  $\Vector[\prime]{\muz{1},\muz{2},\sigmaz{1},\sigmaz{2},\rho} =
  \Vector[\prime]{0,0,1,1,0}$, which means that
  \myref{assumption_score_function}{assumption_score_function_limit_h}
  will fail for any point $\LGp$ that solves
  $\bmuz{}(\LGp;\Vector[\prime]{0,0,1,1,0})=0$.

\subsubsection{Assumptions for $n$, $m$ and $\bm{b}$}
\label{app:assumptions_upon_mNb}

For simplicity, the present analysis will use the
\mbox{$\bm{b}=\parenR{\bz{1},\bz{2}}$} introduced in
the second paragraph after \cref{th:lgsd_estimate_real_valued},
i.e.\ it
will be assumed that the individual bandwidths $\bmbzh{h}$ for the
different lags~$h$ approach zero at the same rate --- and that it for
the asymptotic investigation thus can be assumed that the same
bandwidth is used for all the lags.
\begin{assumption}
  \label{assumption_Nmb}
  Let \mbox{$m \defeq \mz{n} \rightarrow \infty$} be a sequence of
  integers denoting the number of lags to include, and let
  \mbox{$\bm{b} \defeq \bmbz{n} \rightarrow \bmzeroz[+]{}$} be the
  bandwidths used when estimating the local Gaussian correlations for
  the lags \mbox{$h = 1,\dotsc,m$} (based on $n$ observations).
      Let $\bz{1}$ and $\bz{2}$ refer to the two components of
  $\bm{b}$, and let $\alpha$, $\nu$ and~$a$ be as introduced in
  \myref{assumption_Yt}{assumption_Yt_strong_mixing}.  Let
  \mbox{$s\defeq \sz{n} \rightarrow\infty$} be a sequence of integers
  such that \mbox{$s=\oh{\sqrt{n\bz{1}\bz{2}/m}}$}, and let $\tau$ be
  a positive constant.
    The following requirements must be satisfied for these
  entities.\footnote{Notational convention: \enquote{$\vee$} denotes the maximum of two
    numbers, whereas \enquote{$\wedge$} denotes the minimum.} 
  \begin{enumerate}[label=(\alph*)]
  \item     \label{eq:assumption_N_and_b1b2}
        $\log n / n\! \parenRz[5]{}{\bz{1}\bz{2}} \longrightarrow 0$.,
  \item     \label{eq:assumption_Nb1b2/m}
        $n\bz{1}\bz{2}/m \longrightarrow \infty$.
  \item     \label{eq:assumption_m(b1 join b2)}
        $\mz[\delta]{}\!\left(\bz{1}\vee\bz{2}\right) \longrightarrow 0,     \text{ where }     \delta = 2 \vee \tfrac{\nu(a+1)}{\nu(a-1)-2}$.
                                          \item     \label{eq:assumption_alpha(s)_and_(N/b1b2)^1/2}
        $\sqrt{nm/\bz{1}\bz{2}}\cdot     \sz[\tau]{}\!\cdot\alpha(s-m+1) \longrightarrow \infty$.
  \item     \label{eq:assumption_m=o((Nb1b2)^tau/(2+5tau)-lambda)}
        $m = \oh{\parenRz[\tau/(2+5\tau)-\lambda]{}{n\bz{1}\bz{2}}},    \text{ for some } \lambda \in \left(0, \tau/(2+5\tau)\right)$.
  \item     \label{eq:assumption_m=o(s)}
        $m = \oh{s}$.
  \end{enumerate}
\end{assumption}
\Myref{assumption_Nmb}{eq:assumption_N_and_b1b2} is needed in order
for the asymptotic theory from \citet{Tjostheim201333} to be valid for
the estimates $\hatlgacr[5]{\LGp}{h}$.  See
\cref{th:block_sizes_for_main_result} for a verification of the
internal consistency of the requirements given in
\cref{assumption_Nmb}.  
\label{pR2bp2_number_of_points_in_the_window_main} 
The expected number of observations near $\LGp$ will for large $n$
and small $\bz{1}$ and $\bz{2}$ be of order
$n\bz{1}\bz{2}\cdot\gh[\LGp]{h}$ --- and this will, when
$\gh[\LGp]{h}>0$, go to infinity when $\nlimit$ and $\blimit$.  See
the end of \cref{pR2bp2_number_of_points_in_the_window_SM} for
further details.

\subsection{Convergence theorems for $\hatlgsdM[p]{\LGp}{\omega}{m}$}
\label{sec:lgsd_convergence_theorems}

\begin{theorem}[$\LGp$ on diagonal, i.e.\ $\LGpdiagonal$]
  \label{th:asymptotics_for_hatlgsd}
  The local Gaussian spectral density $\lgsd[p]{\LGp}{\omega}$ is a
  real valued function when the point $\LGp$ lies on the
  diagonal.
    Furthermore; when the univariate time series $\Yz{t}$ satisfies
  \cref{assumption_Yt,assumption_score_function}, and $n$, $m$ and
  \mbox{$\bm{b}=\parenR{\bz{1},\bz{2}}$} are as given in
  \cref{assumption_Nmb}, then the following asymptotic results holds
  for the \mbox{$m$-truncated} estimate
  $\hatlgsdM[p]{\LGp}{\omega}{m}$,
  \begin{align}
    \label{eq:th:asymptotics_for_hatlgsd}
    \sqrt{n \!\parenRz[3]{}{\bz{1}\bz{2}} \!/ m} \cdot     \left(\hatlgsdM[p]{\LGp}{\omega}{m} - \lgsd[p]{\LGp}{\omega} \right)
      \stackrel{\scriptscriptstyle d}{\longrightarrow}
    \UVN{\bm{0}}{\sigmaz[2]{\!\LGp}(\omega)},
  \end{align}
                        where the formula
  \begin{align}
    \label{eq:th:asymptotics_for_hatlgsd_variance}
    \sigmaz[2]{\!\LGp}(\omega) =     4 \lim_{\mlimit} \frac{1}{m} \sumss{h=1}{m}     \lambdazM[2]{h}{m} \cdot     \subp{\cos}{}{}{}{2} (2\pi\omega h)     \cdot \tildesigmaz[2]{\!\LGp}(h)
  \end{align}
  relates the variance $\sigmaz[2]{\!\LGp}(\omega)$ to the
  asymptotic variances $\tildesigmaz[2]{\!\LGp}(h)$ of
  $\sqrt{n \!\parenRz[3]{}{\bz{1}\bz{2}}}\cdot
  \parenR{\hatlgacrb[p]{\LGp}{h}{\bmbzh{h}} - \lgacr[p]{\LGp}{h}}$.
\end{theorem}

\begin{proof}
  The proof is given in \cref{sec:proof_of_th:asymptotics_for_hatlgsd}.
\end{proof}

{
  \label{pR2bp2_1}
  \label{pR1MC6_2} 
  The variance $\sigmaz[2]{\!\LGp}(\omega)$ depends on all the
  bivariate density functions through the variances
  $\tildesigmaz[2]{\!\LGp}(h)$.  Moreover, it is clear from
  \cref{eq:th:asymptotics_for_hatlgsd_variance} that
  $\sigmaz[2]{\!\LGp}(\omega)$ as a function of the frequency $\omega$
  is symmetric around $\omega=\tfrac{1}{4}$, with its highest values
  when $\omega \in \parenC{0,\tfrac{1}{2}}$.  The same symmetry is not
  present for the variance of the $m$-truncated spectra
  $\hatlgsdM[p]{\LGp}{\omega}{m}$, and the variance of
  $\hatlgsdM[p]{\LGp}{\omega}{m}$ will have its highest value when
  $\omega=0$, cf.\ \cref{app:finite_sample_and_variance_LGSD} for
  details.  }

A similar result to \cref{th:asymptotics_for_hatlgsd} can be stated
for time reversible stochastic processes.
\begin{theorem}[$\Yz{t}$ time reversible]
  \label{th:asymptotics_for_hatlgsd_reversible}
  The local Gaussian spectral density $\lgsd[p]{\LGp}{\omega}$ is a
  real valued function for all points $\LGp$ when $\Yz{t}$ is time
  reversible (see \cref{def:Yt_reversible}).
  Furthermore under
  \cref{assumption_Yt,assumption_score_function,assumption_Nmb}, the
  same asymptotic results as stated in
  \cref{th:asymptotics_for_hatlgsd} holds for the \mbox{$m$-truncated}
  estimate $\hatlgsdM[p]{\LGp}{\omega}{m}$.
\end{theorem}

\begin{proof}
  \Myref{th:lgsd_properties}{th:lgsd_real_when_reversible_time_series}
  states that $\lgsd[p]{\LGp}{\omega}$ is a real-valued function, and
  the proof of \cref{th:asymptotics_for_hatlgsd} (see
  \cref{sec:proof_of_th:asymptotics_for_hatlgsd}) can then be repeated
  without any modifications.
\end{proof}

The asymptotic normality results in
\cref{th:asymptotics_for_hatlgsd,th:asymptotics_for_hatlgsd_reversible}
do not easily enable a computation of pointwise confidence intervals
for the estimated LGSD.  Thus, the pointwise confidence intervals
later on will either be estimated based on suitable quantiles obtained
by repeated sampling from a known distribution, or they will be based
on bootstrapping techniques for those cases where real data have been
investigated.  Confer \citet[ch.~7.2.5 and
7.2.6]{terasvirta2010modelling} for further details with regard to the
need for bootstrapping in such situations.
See also \citet{lacal2017local,lacal:_estim} for analytic results on the bootstrap and block
bootstrap in the case of estimation of the local Gaussian auto- and
cross-correlation functions.

The asymptotic result for $\hatlgsdM[p]{\LGp}{\omega}{m}$
complex-valued is given in
\cref{sec:lgsd_estimate_complex_valued_case}, where it can be seen
that
\mbox{$\sqrt{n \!\parenRz[3]{}{\bz{1}\bz{2}} \!/ m} \cdot     \left(\hatlgsdM[p]{\LGp}{\omega}{m} - \lgsd[p]{\LGp}{\omega}
  \right)$} then asymptotically approaches a complex-valued normal
distribution.

\section{Visualisations and interpretations}
\label{sec:Examples}

\label{pAE2_navigation} 
\label{pAE2_comparison} 
This section will show how different visualisations of the
\mbox{$m$-truncated} estimates $\hatlgsdM[p]{\LGp}{\omega}{m}$ can be
used to detect 
nonlinear dependency structures in a time series.  
Similar graphical methods can also be found in
\citet{li19:_quant_frequen_analy_spect_diver,BIRR2019122}, and the
heatmap-plot presented in this section is in particular inspired by
the one encountered in \citet{li19:_quant_frequen_analy_spect_diver}.

Technical details, and the description of the selected tuning
parameters of $\hatlgsdM[p]{\LGp}{\omega}{m}$, are given in
\cref{sec:Examples_the_parameters}.
\Cref{sec:estimation_aspect_for_the_given_parameter_configuration}
uses the aforementioned \texttt{dmbp}-data (see page
\pageref{footnote:dmbp}) to highlight how the different tuning
parameters of the estimation algorithm are interconnected.

A sanity test of the implemented estimation algorithm is presented in
\cref{sec:Some_simulations}, and it is there seen that
$\hatlgsdM[p]{\LGp}{\omega}{m}$ can detect local periodic structures
in an example where a heuristic argument enables the prediction of the
anticipated result.  
\Cref{sec:Real_data} applies the local Gaussian machinery to the
\texttt{dmbp}-data, and it also contains the results from a GARCH-type
model fitted to the \texttt{dmbp}-data.  
A comparison of the results from the original data and the fitted
model can reveal to what extent the internal dependency structure of
the fitted model actually reflects the dependency structure of the
original sample, and this might be of interest with regard to model
selection.

A few extreme examples have been included in the Supplementary
Material in order to investigate the limitations of this method.
\Cref{app:fig:trigonometric.C1.component} examine the detection of a
periodic component located far out in the tail of a large sample,
and \cref{sec:lgch:beware_of_global_structures} consider a situation
based on a deterministic function perturbed by very low random
fluctuations.

\subsection{The input parameters and some other technical details}
\label{sec:Examples_the_parameters}

Several tuning parameters must be selected in order to compute the
$m$-truncated local Gaussian spectral density estimates
$\hatlgsdM[p]{\LGp}{\omega}{m}$, and the values used for the
plots in this section are given below.  
Note that these parameters have been selected in order to provide a
\textit{proof of concept} for the fact that nonlinear dependency
structures 
can be detected by this approach, and the quest for \enquote{optimal
  parameters} is a topic for further work.  
The interested reader can consult \cref{app:sensitivity_analysis} in
the Supplementary Material for a sensitivity analysis of the different
tuning parameters.

\textbf{The pseudo-normalisation:} 
The initial step of the computation of $\hatlgsdM[p]{\LGp}{\omega}{m}$
is to replace the observations $\TSR{\yz{t}}{t=1}{n}$ with the
corresponding pseudo-normalised observations
$\TSR{\widehatzz{t}}{t=1}{n}$, cf.\ \cref{def:lgsd_estimator}, i.e.\
an estimate of the marginal cumulative density function $G$ is needed.
The present analysis has used the rescaled empirical cumulative
density function $\widehatGz{\!n}$ for this purpose, but the
computations could also have been based on a logspline-estimate
of~$G$.
A preliminary test revealed that the two normalisation procedures
created strikingly similar estimates of
$\hatlgsdM[p]{\LGp}{\omega}{m}$, so the computationally faster
approach based on the rescaled empirical cumulative density-function
has thus been applied for the present investigation.

\textbf{The length $n$ of the samples:} 
All samples have the same length as the  \texttt{dmbp}-data, i.e.\
$n=1974$.  
The estimation machinery produces similar results for shorter samples,
but it is important to keep in mind that too short samples might not
reveal the dependency structure of interest --- which in particular
might be an issue for the tails of the distribution.

\textbf{The points $\LGp$ of investigation:} 
Three diagonal points, with coordinates corresponding to the 10\%,
50\% and 90\% percentiles of the standard normal
distribution,\footnote{
  The corresponding coordinates are $(-1.28,-1.28)$, $(0,0)$ and
  $(1.28,1.28)$.} 
will be used in the basic plots in this section.  These points will
often be referred to as \textit{lower tail}, \textit{center} and
\textit{upper tail} when discussed in the text.  Confer
\cref{app:Bandwidth_sensitivity} for further details related to the
selection of $\LGp$, and see
\cref{fig:dmt_heatmap_and__levels_vs_norm} for a heatmap-based plot.

\textbf{The lag-window function $\lambdazM{h}{m}$:} 
The smoothing of the estimated local Gaussian autocorrelations, cf.\
\myref{def:lgsd_estimator}{def:lgsd_esitimator_folded}, was done by
the Tukey-Hanning lag-window kernel:
$\lambdazM{h}{m} = \tfrac{1}{2} \cdot \left(1 + \cos\left(\pi\cdot
    \tfrac{h}{m}\right) \right)$ for $|h| \leq m $,
$\lambdazM{h}{m} =0$ for $|h|>m$.

\textbf{The bandwidth $\bm{b}$:} 
The estimation of the local Gaussian autocorrelations requires the
selection of a bandwidth-vector $\bm{b} = \parenR{\bz{1},\bz{2}}$, and
the majority of the plots in this section have used
\mbox{$\bm{b}=(.5,.5)$}.  Note that it is natural to require
\mbox{$\bz{1}=\bz{2}$} since both of the components in the lag~$h$
pseudo-normalised pairs comes from the same univariate time series.
Further discussion of choice of bandwidth is given in
\cref{How.to.select.the.tuning.parameters?}

\textbf{The truncation level $m$:} 
The value \mbox{$m=10$} was used for the truncation level, since it
was possible to 
detect nonlinear dependency structures even for that low truncation
level.

\textbf{The number of replicates $R$:} 
The estimated values (means and 90\% pointwise confidence intervals)
have been based on $R=100$ replicates.  Simulations were used for the
cases with known parametric models, whereas a bootstrap based
resampling strategy were used for the real data example (cf.\
\cref{app:regarding_resampling} for the technical details).

\textbf{Numerical convergence:} 
The \Rpackage \Rref{localgauss}, see \citet{Berentsen:2014:ILR},
estimates the local Gaussian autocorrelations $\lgacr{\LGp}{h}$ and
returns them together with an attribute 
that reveals whether or not the estimation algorithm converged
numerically.  
The $m$-truncated estimates $\hatlgsdM[p]{\LGp}{\omega}{m}$ inherits
the convergence-attributes from the estimates
$\TSR{\hatlgacr{\LGp}{h}}{h=-m}{m}$, and either \enquote{\texttt{NC =
    OK}} or \enquote{\texttt{NC = FAIL}} will be added to the plot
depending on the convergence status.  Note that convergence-problems
hardly occurs when the computations are based on pseudo-normalised
observations.

\textbf{Reproducibility and interactive investigations:} 
All the examples in this paper can be reproduced by the scripts (see
\cref{app:data_details}) that are contained in the \Rpackage
\lgsdRpackage.  Note that the computations of
$\hatlgsdM[p]{\LGp}{\omega}{m}$ can be performed for a wide range of
tuning parameters, which allows an integrated interactive
investigation of the results by means of a
\Rref{shiny}-application.\footnote{\label{fotnote:shiny_details} See
  \citet{chang16_shiny} for details about \Rref{shiny}.}

\subsection{Estimation aspects for the given parameter configuration}
\label{sec:estimation_aspect_for_the_given_parameter_configuration}

The estimation of $\hatlgsdM[5]{\LGp}{\omega}{m}$ for a point
\mbox{$\LGp = \LGpoint$} that lies on the diagonal, i.e.\
\mbox{$\LGpdiagonal$}, will be based on the estimates of
$\hatlgacr[5]{\LGp}{h}$ for \mbox{$h=1,\dotsc,m$}, and it is thus of
interest to first investigate how
these estimates depend on the configuration of the tuning parameters
given in \cref{sec:Examples_the_parameters}.  This is most easily done
in terms of an example, and the pseudo-normalised \texttt{dmbp}-data
(of length~1974) will be used for this purpose.

First of all, note that the combination of point $\LGp$ and bandwidth
$\bm{b}$ influences how many of the $h$-lagged pairs that effectively
contribute to the computation of $\hatlgacr[5]{\LGp}{h}$.  This is
shown in \cref{fig:dmbp_pseudo_normalised_with_rectangles} for the
pseudo-normalised \texttt{dmbp}-data.
In the plot of the pseudo-normalised time series (top panel), the
three horizontal dashed lines represent the \textit{levels} which
corresponds to the coordinates of the three points~$\LGp$, whereas the
horizontal strips centered at those lines show which observations that
lie within a distance of \mbox{$b=0.5$} from the respective lines.
The three plots at the bottom shows the corresponding $1$-lagged
pairs, each with a \textit{bandwidth-square} (of width~$2b$) centered
at one of the selected points~$\LGp$.

\begin{knitrout}\footnotesize
\definecolor{shadecolor}{rgb}{0.969, 0.969, 0.969}\color{fgcolor}\begin{figure}[h]

{\centering \includegraphics[width=1\linewidth]{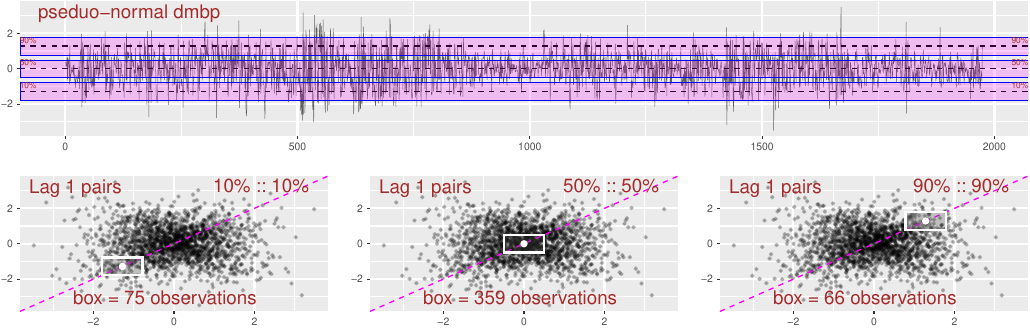} 

}

\caption[\texttt{dmbp} (pseudo-normalised version), \textit{levels}
and \textit{bandwidth-bands} (top) and \textit{lag~1
  bandwidth-squares} (bottom)]{\texttt{dmbp} (pseudo-normalised
  version), \textit{levels} and \textit{bandwidth-bands} (top) and
  \textit{lag~1 bandwidth-squares} (bottom).  Further details in the
  main text.}\label{fig:dmbp_pseudo_normalised_with_rectangles}
\end{figure}

\end{knitrout}

The estimates of $\lgacr{\LGp}{1}$ are based on the 1-lagged pairs
seen in the lower part of
\cref{fig:dmbp_pseudo_normalised_with_rectangles}, and these and
similar estimates for lags up to 200 (based on
\mbox{$\bm{b}=\parenR{0.5,0.5}$}) are shown in \cref{fig:dmbp_lag}.
An investigation of \cref{fig:dmbp_lag} shows how
$\hatlgacr[5]{\LGp}{h}$ varies for the three points of interest, and
there is a clear distinction between the center and the two tails.
Note that the bias-variance balance of the estimates
$\hatlgacr[5]{\LGp}{h}$ depends on the number of $h$-lagged pairs
that effectively contribute during the computation, and it is thus
clear that the variance will increase for points $\LGp$ that lie
farther out in the tails.  The selection of which tail-points to
investigate must thus take into account the number of available
observations for the lags to be included.

\begin{knitrout}\footnotesize
  \definecolor{shadecolor}{rgb}{0.969, 0.969, 0.969}\color{fgcolor}\begin{figure}[h]

    {\centering \includegraphics[width=1\linewidth]{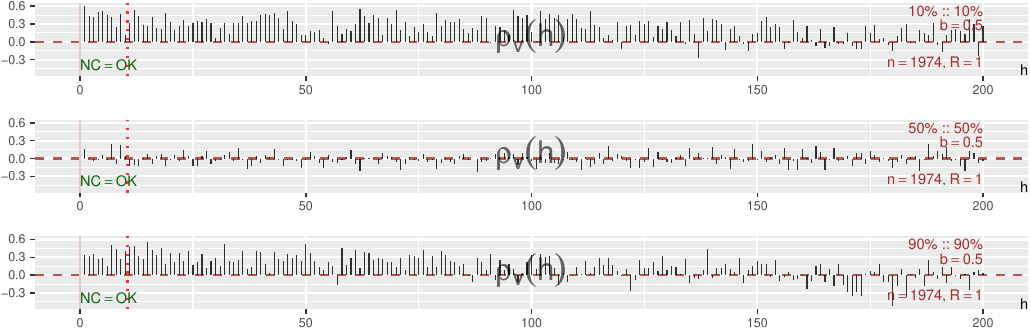} 

    }

    \caption[]{\texttt{dmbp}-data, $\hatlgacr[5]{\LGp}{h}$ for
      \mbox{$h=1,\dotsc,200$} (for the three points of interest).  The
      estimates for \mbox{$h=1,\dotsc,10$} will be used for
      $\hatlgsdM[5]{\LGp}{\omega}{m}$, cf.\ \cref{fig:dmbp}.  
    }\label{fig:dmbp_lag}
  \end{figure}

\end{knitrout}

The $\hatlgacr[5]{\LGp}{h}$ tends to fluctuate around~0 at the center,
which implies that the corresponding estimated spectral density
$\hatlgsdM[5]{\LGp}{\omega}{m}$ most likely will be rather flat and
close to~1.  For the two tails, it seems natural to assumme that some
long-range dependency must be present, and one might also suspect that
there is an asymmetry between the two tails.\footnote{A further
  investigation of this is easy when the \Rref{shiny}-application in
  the \Rpackage \lgsdRpackage\ is used, since it then is possible to
  immediately switch to an investigation of the corresponding
  spectra.}

Based on the impression from \cref{fig:dmbp_lag}, it might be a
connection between the global long-range dependence in the
\texttt{dmbp}-data and the local dependency structure in the tails
--- but note that the estimates in \cref{fig:dmbp_lag} are based on
the pseudo-normalised data, so the information from the marginal
distribution is not present here.  However, the same kind of
behaviour has been observed for pseudo-normalised samples from
different GARCH-type models, so the dependency structure of the
tails could be a significant contributor to the global long-range
dependency seen in time series models like ARCH and GARCH.

\subsection{Sanity testing the implemented estimation algorithm}
\label{sec:Some_simulations}

The purpose of this section is to check whether or not the
implemented estimation algorithm returns reasonable results for some
simulated examples.  It is only for the Gaussian case that the true
value of the local Gaussian spectral densities $\lgsd{\LGp}{\omega}$
are known, and it is thus important to specifically construct an
example where heuristic arguments enable the prediction of the
anticipated results.

The strategy used to create the plots for the simulated data works
as follows: First draw a given number of independent replicates from
the specified model, and compute $\hatlgsdM[5]{\LGp}{\omega}{m}$ and
$\widehatfz[m]{}(\omega)$ for each of the replicates.  Then extract
the mean of these estimates to get estimates of the true values of
$\lgsdM[5]{\LGp}{\omega}{m}$ and $\fz[m]{}(\omega)$, and select
suitable upper and lower percentiles of the estimates to produce an
estimate of the pointwise confidence intervals.

Note that the plots have been annotated with the following
information: The numerical convergence status {NC} in the lower left
corner; the truncation level $m$ in the upper left corner; the
percentiles of the point $\LGp$ of investigation, and the bandwidth
$\bm{b}$ in the upper right corner; the length $n$ and the number of
replicates $R$ in the lower right corner.

\subsubsection{Gaussian white noise}
\label{sec:Gaussian_white_noise}

The sanity testing of the implemented estimation algorithm starts with
the trivial case.
\Cref{fig:Gaussian_WN} shows the result when the estimation procedure
is used on 100 independent samples of length 1974 from a standard
normal distribution~$\UVN{0}{1}$.  The computations are based on the
bandwidth \mbox{$\bm{b}=\parenR{0.5,0.5}$}, and the points (on the
diagonal) corresponds to the $0.1$, $0.5$ and~$0.9$ quantiles of the
standard normal distribution.  The top left panel shows the
pseudo-normalised version of the first time series that was sampled
from the model, with dashed
lines at the levels that corresponds to the above mentioned points.
The three other panels contains information about the
\mbox{$m$-truncated} ordinary spectral density $\fz[m]{}(\omega)$ (red
part,\footnote{If you have a black and white copy of this paper, then
  read \enquote{red} as \enquote{dark} and \enquote{blue} as
  \enquote{light}.} the same for all the plots) and the
\mbox{$m$-truncated} local Gaussian spectral densities
$\lgsdM[5]{\LGp}{\omega}{m}$ for the three points under investigation
(blue part).

\begin{knitrout}\footnotesize
\definecolor{shadecolor}{rgb}{0.969, 0.969, 0.969}\color{fgcolor}\begin{figure}[h]

{\centering \includegraphics[width=1\linewidth]{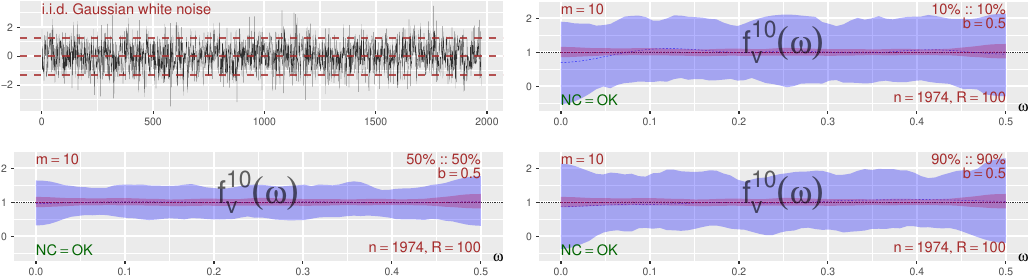} 

}

\caption[i.i.d]{i.i.d.\ Gaussian white noise, with global and local
  spectra for three points.}\label{fig:Gaussian_WN}
\end{figure}

\end{knitrout}

It can be seen from \cref{fig:Gaussian_WN} that the means
of the estimates
(the dashed lines at the center of the regions)
are good estimates of
$\fz[m]{}(\omega)$ and $\lgsdM[5]{\LGp}{\omega}{m}$,
which in this case in fact coincides with $f(\omega)$ and
$\lgsd[5]{\LGp}{\omega}$, i.e.\ it is known that the true values are
identical to~1 both for the local and global case.  Observe that the
estimated 90\% pointwise confidence intervals are wider for the local
Gaussian spectral densities, which is as expected since the bandwidth
used in the estimation of the local Gaussian autocorrelations reduces
the number of observations that effectively contributes to the
estimated values, and thus makes the estimates more prone to
small-sample variation.  Note also that the pointwise confidence
intervals are wider in the tails, which is a natural consequence of
the reduced number of points in those regions, cf.\ the discussion
related to \cref{fig:dmbp_pseudo_normalised_with_rectangles}.  The
width of these pointwise confidence intervals will decrease when the
bandwidth increases, cf.\ the discussion related to
\cref{fig:trigonometric_one_cosine}.

The estimation procedure gave good estimates of the true values
$f(\omega)$ and $\lgsd[5]{\LGp}{\omega}$ in the simple example of
\cref{fig:Gaussian_WN}, but it is important to keep in mind that these
plots actually shows estimates of $\fz[m]{}(\omega)$ and
$\lgsdM[5]{\LGp}{\omega}{m}$.  It might be necessary to apply a (much)
higher truncation level $m$ before $\fz[m]{}(\omega)$ and
$\lgsdM[5]{\LGp}{\omega}{m}$ gives decent approximations of the true
values $f(\omega)$ and $\lgsd[5]{\LGp}{\omega}$.  
However, for the task of interest in \cref{sec:Examples} it is
actually not a problem if the selected truncation level does not
give \enquote{optimal estimates} of $f(\omega)$ and
$\lgsd[5]{\LGp}{\omega}$ --- since the detection of nonlinear
dependency structures can be seen for a wide range of different
truncation levels.  The recommended approach is
to estimate $\hatlgsdM[5]{\LGp}{\omega}{m}$ for a range of possible
truncation levels~$m$, and then check if the shape of the estimates
for different truncations share the same properties with regard to the
position of any peaks and troughs.  The \Rpackage \lgsdRpackage\ is
designed in such a way that this is trivial to do.

\subsubsection{Some trigonometric examples}
\label{sec:Deterministic_trigonometric_models}

Beyond the realm of Gaussian time series, it is not known what the
true value for the local Gaussian spectral density actually should
be.  The sanity of the implemented estimation algorithm will thus be
tested by the means of an artificially constructed \textit{local
  trigonometric} time series, for which it at least can be reasonably
argued what the expected outcome should be for some specially
designated points~$\LGp$ (given a suitable bandwidth~$\bm{b}$).  These
artificial time series will not satisfy the requirements needed for
the asymptotic theory to hold true (as is also the case for standard
global spectral analysis), but they can still be used to show how an
exploratory tool based on the local Gaussian spectral density can
detect local periodic properties that the ordinary spectral density
fails to~detect.

As a prerequisite (and a reference) for the investigation of the local
trigonometric time series, it is prudent to first investigate the
result based on independent samples from a time series of the form
\mbox{$\Yz{t}=\cos\!\parenR{2\pi\alpha t + \varphi} + \wz{t}$}, where
$\wz{t}$ is Gaussian white noise with mean zero and standard
deviation~$\sigma$, and where it in addition is such that~$\alpha$ is
fixed for all the replicates, whereas the phase-adjustment $\varphi$
is randomly generated for each individual replicate.  A realisation
with \mbox{$\alpha=0.302$} and \mbox{$\sigma = 0.75$} is shown in
\cref{fig:trigonometric_one_cosine}, where the frequency $\alpha$ has
been indicated with a vertical line in order to show that both the
local and global approaches in this case have a peak at the expected
position.  The plots are based on 100 samples of length 1974, and
shows 90\% pointwise confidence intervals.  Some useful remarks can be
based on \cref{fig:trigonometric_one_cosine}, before \textit{the local
  trigonometric} case is defined and investigated.

\begin{knitrout}\footnotesize
\definecolor{shadecolor}{rgb}{0.969, 0.969, 0.969}\color{fgcolor}\begin{figure}[h]

{\centering \includegraphics[width=1\linewidth]{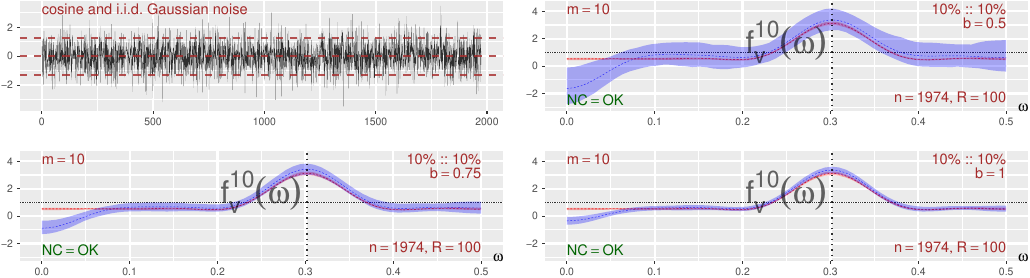} 

}

\caption[Single cosine and i.i.d]{Single cosine and i.i.d.\ white
  noise, same point, bandwidths 0.5, 0.75
  and~1.}\label{fig:trigonometric_one_cosine}
\end{figure}
\end{knitrout}

All the plots in \cref{fig:trigonometric_one_cosine} show the same
point (corresponding to the 10\% quantile) in the lower tail, but they
differ with regard to the bandwidths that have been used.  In
particular, the upper right plot is based on the bandwidth
\mbox{$\bm{b}=(.5,.5)$} (the bandwidth used in all the other
examples), whereas the two plots at the bottom shows the situation for
the bandwidths \mbox{$\bm{b}=(.75,.75)$} and \mbox{$\bm{b}=(1,1)$},
respectively at the left and right.  In this case, the widths of the
pointwise confidence intervals are influenced by the selected
bandwidths, but the overall shape is similar and close to the global
estimate shown in red.  This feature is also present for the other
examples that have been~investigated.

Note that the cosine is recovered using just a neighbourhood of the
10\% quantile.  Furthermore, the portion of the local Gaussian
spectral density that is negative decreases with increasing bandwidth,
which is in accordance with the remark at the end of \cref{sec:LGC}.
Using the notation from \cref{def:lgsd_estimator}, this can for the
estimates of the local Gaussian autocorrelations be stated as
$\hatlgacrb[p]{\LGp}{h}{\bm{b}}\rightarrow\widehatrhoz{}\,(h)$ when
$\bm{b}\rightarrow \bm{\infty}$, which implies that the
estimate $\hatlgsdM[5]{\LGp}{\omega}{m}$ converges towards the global
non-negative estimate $\widehatfz[m]{}(\omega)$.  It is thus possible
to reduce the amount of negative values for the estimates
$\hatlgsdM[5]{\LGp}{\omega}{m}$ by increasing the bandwidth $\bm{b}$,
but keep in mind that it is the limits $\blimit$ and $\mlimit$ that
should be taken in order to actually estimate the local Gaussian
spectral density $\lgsd[5]{\LGp}{\omega}$.

The truncation level used in \cref{fig:trigonometric_one_cosine} is
rather low, i.e.\ $m=10$, 
but it can be seen that the peak is observed at the correct
frequency.  The peak will grow taller and narrower when a higher
truncation level is used, but it will stay at the same frequency.
This indicates that these plots (even for low truncation values) can
detect properties of the underlying structure.  Again, this feature
is shared with the other examples that have been investigated.

The local Gaussian spectral densities in
\cref{fig:trigonometric_one_cosine} goes below zero for low
frequencies, a feature that is not entirely unexpected as
$\TSR{\lgacr[5]{\LGp}{h}}{h\in\ZZ}{}$, the collection of local
Gaussian autocorrelations, may not be a non-negative definite
function.  In fact, based on the observation that the estimates of
$\hatlgsdM[5]{\LGp}{\omega}{m}$ have peaks that are taller and wider
than those of $\widehatfz[m]{}(\omega)$, it is as expected that these
estimates might need to have negative values somewhere.  The reason
for this is that all the spectral densities (global, local and
\mbox{$m$-truncated}) by construction necessarily must integrate to
one over the interval \mbox{$(-\tfrac{1}{2},\tfrac{1}{2}]$}.  The
higher and wider peaks of the estimates for
$\hatlgsdM[5]{\LGp}{\omega}{m}$ thus requires that it has to lie below
the estimates of $\widehatfz[m]{}(\omega)$ in some other region, and
if necessary it must attain negative values somewhere.  The
interesting details in the plots are thus the position of the peaks of
$\hatlgsdM[5]{\LGp}{\omega}{m}$, and regions with negative values
should not in general be considered a too troublesome feature.

Note that, under certain circumstances,
$\hatlgsdM[5]{\LGp}{\omega}{m}$ might contain spurious artefacts
when it is computed for time series having a non-flat ordinary
spectrum, c.f.\ \cref{sec:lgch:beware_of_global_structures} for a
discussion related to a case based on a deterministic function with
small noise.

\textbf{The local trigonometric case:} 
The key idea in this example is that an artificial time series
$\TSR{\Yz{t}}{t\in\ZZ}{}$ can be constructed by the following
scheme:
\begin{enumerate}
\item 
  Select $r$ time series
  $\TSR{\Cz{i}(t)}{i=1}{r}$.  
\item 
  Select a random variable $I$ with values in the set
  $\parenC{1,\dotsc,r}$, and use this to sample a collection of
  indices $\TSR{\Iz{t}}{t\in\ZZ}{}$ (i.e.\ for each $t$ an
  independent realisation of $I$ is taken).  Let
  $\pz{i}\defeq\Prob{\Iz{i}=i}$ denote the probabilities for the
  different outcomes.
\item 
  Define $\Yz{t}$ by
  means of the equation
  \begin{align}
    \label{eq:local_trigonometric_sum}
    \Yz{t} \defeq \sumss{i=1}{r} \Ind{\Iz{t} = i}\cdot\Cz{i}(t).
  \end{align}
  The indicator function $\Ind{\cdot}$ ensures that only one of the
  $\Cz{i}(t)$ contribute for a given value~$t$, i.e.\ it is also
  possible to write $\Yz{t} = \Cz{\Iz{t}}(t)$.
\end{enumerate}

The \textit{local trigonometric} time series (needed for the sanity
testing of the implemented estimation algorithm) are constructed by
selecting $r$ cosine-functions that oscillates around different
horizontal base-lines $\Lz{i}$, i.e.\ 
\begin{align}
  \label{eq:local_trigonometric_components}
  \Cz{i}(t) = \Lz{i} + \Az{i}(t) \cdot   \cos \left(2\pi\alphaz{i} t + \varphiz{i} \right), \qquad i = 1,\dotsc,r,
\end{align}
where $\alphaz{i}$ and $\varphiz{i}$ respectively represent the
frequency and phase-adjustment occurring in the cosine-function, and
where the amplitudes $\Az{i}(t)$ are uniformly distributed in some
interval $\parenS{\az{i},\bz{i}}$.  Note that it is assumed that the
phases $\varphiz{i}$ are uniformly drawn (one time for each
realisation) from the interval between~$0$ and~$2\pi$, and it is
moreover also assumed that the stochastic processes $\varphiz{i}$,
$\Az{i}(t)$ and $\Iz{t}$ are independent of each other.

The autocorrelation $\rho(h)$ of the time series
$\TSR{\Yz{t}}{t\in\ZZ}{}$, with $\Cz{i}(t)$ as given in
\cref{eq:local_trigonometric_components}, has been computed in the
Supplementary Material, cf.\
\cref{app:eq:correlation.of.the.artificial.time.series} in
\cref{app:fig:trigonometric}.  For the purpose of the present
section, it is sufficient to know that it is possible to find
parameter-configurations for which the global spectrum is rather
flat (at least when truncated at $m=10$), which 
implies that it cannot detect the frequencies $\alphaz{i}$ of the
underlying structure.

Strictly speaking, neither $f(\omega)$ nor $\lgsd{\LGp}{\omega}$ are well defined for
the \textit{local trigonometric} times series, but this is not
important since it still is possible to {predict} (cf.\
\cref{app:fig:trigonometric} for details) that the $m$-truncated
estimates $\hatlgsdM{\LGp}{\omega}{m}$ for some points $\LGp$ should
resemble \cref{fig:trigonometric_one_cosine} --- and this can be
used, cf.\ \cref{fig:trigonometric}, to test the sanity of the
implemented estimation algorithm.

The explicit expression for the \textit{local trigonometric} example
studied in \cref{fig:trigonometric} is given by $r=4$ components
$\Cz{i}(t)$ of the form given in
\cref{eq:local_trigonometric_components}, where the probabilities
$\pz{i}$ 
are given by $(0.05, 1/3-0.05, 1/3, 1/3)$, the frequencies
$\alphaz{i}$ are given by $(0.267, 0.091, 0.431, 0.270)$, the
base-lines $\Lz{i}$ are given by the values $(-2, -1, 0, 1)$, and
the lower and upper ranges for the uniforms sampling of the
amplitudes $\Az{i}(t)$ are respectively given by
{$(0.5, 0.2, 0.2, 0.5)$} and {$(1.0, 0.5, 0.3, 0.6)$}.  
Note that $\Lz{i}$ and $\Az{i}(t)$ should be selected in order to
give a minimal amount of overlap between the different components,
cf.\ \cref{app:fig:trigonometric} for further details.

\Cref{fig:trigonometric} shows $\widehatfz[m]{}(\omega)$ and
$\hatlgsdM{\LGp}{\omega}{m}$ for the \textit{local trigonometric}
example.  The ordinary spectrum does not detect the frequencies
$\alphaz{i}$ (indicated by vertical lines), whereas the local
Gaussian spectra does have clear peaks at the frequencies from
respectively $\Cz{2}(t)$, $\Cz{3}(t)$ and $\Cz{4}(t)$.  Moreover, a
comparison with \cref{fig:trigonometric_one_cosine} shows that
$\hatlgsdM{\LGp}{\omega}{m}$ indeed does look like predicted, which
verifies the sanity of the implemented estimation algorithm.

\begin{knitrout}\footnotesize
  \definecolor{shadecolor}{rgb}{0.969, 0.969, 0.969}\color{fgcolor}\begin{figure}[h]

    {\centering \includegraphics[width=1\linewidth]{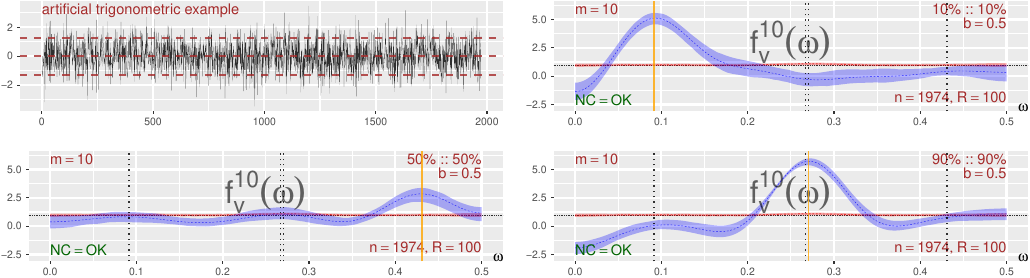} 

    }

    \caption[Artifical example, \textit{hidden trigonometric
      components}]{Artifical example, \textit{local trigonometric
        components}.  Global and local spectra for the three points $\LGp$
      on the diagonal, i.e.\ lower tail, center and upper
      tail.}\label{fig:trigonometric}
  \end{figure}

\end{knitrout}

The selected percentiles $\TSR{\pz{i}}{i=1}{4}$ implies that
observations from the $\Cz{3}(t)$ component after
pseudo-normalisation should lie between $\Phiz[-1]{}(1/3)=-0.43$ and
$\Phiz[-1]{}(2/3)=.43$.  The estimation of $\lgsdM{\LGp}{\omega}{m}$
is based on the bandwidth $\bm{b}=(0.5,0.5)$, which implies that the
estimate at the center will be \enquote{contaminated} by
observations from the neighbouring components --- and this explains
the lower amplitude seen for this point.

The three points $\LGp$ in \cref{fig:trigonometric} correspond
roughly to the base-lines $\Lz{2}, \Lz{3}$ and $\Lz{4}$, and the
corresponding frequencies $\alphaz{2}, \alphaz{3}$ and $\alphaz{4}$
are here detected by $\hatlgsdM{\LGp}{\omega}{m}$.  But what about
the base-line $\Lz{1}$ and the $\alphaz{1}$-frequency?

The low probability at which the $\Cz{1}(t)$ component is selected
implies that the point $\LGp$ corresponding to the base-line
$\Lz{1}$ must lie far out in the lower tail, and for the present
sample size (of $n=1974$) the scarcity of observations in this
region implies that it is not possible to obtain decent estimates of
the required local Gaussian autocorrelations $\lgacr{\LGp}{h}$.  A
countermeasure to this problem would be to use a larger bandwidth
$\bm{b}$, but the result would then be \enquote{contaminated} by the
observations from the $\Cz{2}(t)$ component --- and the peak of
$\hatlgsdM{\LGp}{\omega}{m}$ would then be at the frequency
$\alphaz{2}$ instead of $\alphaz{1}$.  This implies that misleading
results can occur when the bandwidth $\bm{b}$ is to large.  

However, note that for a large enough sample it is possible to
detect the frequency $\alphaz{1}$ that belongs to the
$\Cz{1}(t)$-component, cf.\
\cref{app:fig:trigonometric.C1.component} for further details.

The $\Cz{1}(t)$ component was included in this example in order to
emphasise that extra care is needed when investigating the outer
tails of a sample.  This of course begs the question: For a given
sample $\TSR{\Yz{t}}{t=1}{n}$, how can an investigator figure out
whether or not the estimate of $\lgsdM{\LGp}{\omega}{m}$, for a
given combination of point $\LGp$ and bandwidth $\bm{b}$, seems
trustworthy or not?  Another important question for an investigator
is to decide if some points $\LGp$ might be more interesting than
others.  Both of these questions can be investigated by means of the
two plots seen in \cref{fig:dmt_heatmap_and__levels_vs_norm}, which
(for a single sample from the aforementioned \textit{local
  trigonometric} construction) investigates the $m=10$ truncated
local Gaussian spectra $\lgsdM{\LGp}{\omega}{m}$ for points along
the diagonal.  Note that the points $\LGp$ are represented by their
respective percentiles, and the range goes from the 5\% percentile
to the 95\% percentile.

\begin{figure}[h]
  {\centering \includegraphics[width=1\textwidth]{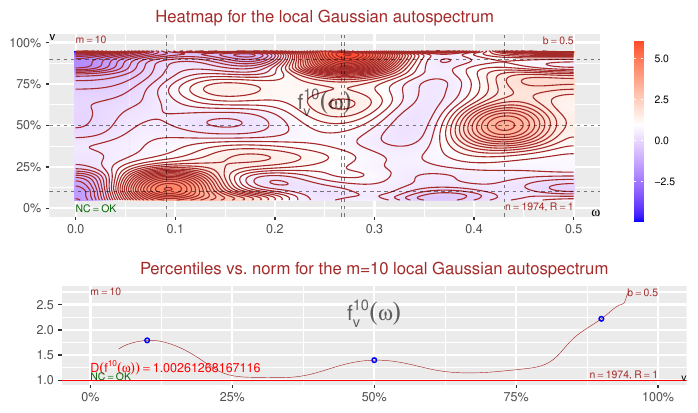}
  }
  \caption[]{Heatmap-plot with corresponding distance-plot, based on
    the local trigonometric case, showing how
    $\hatlgsdM{\LGp}{\omega}{10}$ varies with the percentiles for
    the diagonal-points $\LGp$.  The percentiles and frequencies
    used in \cref{fig:trigonometric} have been indicated with
    lines/points.}\label{fig:dmt_heatmap_and__levels_vs_norm}
\end{figure}

The upper part of \cref{fig:dmt_heatmap_and__levels_vs_norm} is a
heatmap-plot for $\hatlgsdM{\LGp}{\omega}{m}$ (inspired by plots in
\citet{li19:_quant_frequen_analy_spect_diver}), which in this case
is based on one sample of length $n=1974$.  The contour-lines in
this plot clearly reveals that the highest peaks occur approximately
at the points investigated in \cref{fig:trigonometric}.  In fact,
looking at the heatmap, the peak at the 90\% percentile of
\cref{fig:dmt_heatmap_and__levels_vs_norm}, may have its maximum
closer to the 95\% percentile, but one has to be a little careful
here since the estimates of $\lgacr{\LGp}{h}$ might degenerate
towards $+1$ (or $-1$) in the outer part of the tail.

The lower part of \cref{fig:dmt_heatmap_and__levels_vs_norm} shows
the corresponding distance-plot
$D\!\parenR{\lgsdM{\LGp}{\omega}{m}}$, where the norms of the
$m$-truncated spectra (realised as elements of the complex Hilbert
space of Fourier series, cf.\
\cref{app:method_for_sensitivity_analysis} for details) are plotted
against the diagonal points.  Note that distance-based plots do not
contain any information about the frequencies, and completely
different spectral densities can have the same distance-value.  It
is thus important to always combine a distance-based plot with a
plot that reveals the frequency-component.

The horizontal line at the bottom of the distance-plot gives the
norm of the ordinary spectrum, and it can be seen that this line is
very close to the white-noise value which is~1.  It is interesting
and reassuring that it picks up the peaks at the 10\% and 50\%
percentiles.  It does however not indicate a peak close to the 95\%
percentile, but this is also the least clear peak of the heatmap.

This discussion shows that it is important to include a wide range
of points when performing an investigation based on local Gaussian
spectral densities, since it is necessary to check how 
$\hatlgsdM[5]{\LGp}{\omega}{m}$ changes as the diagonal point $\LGp$
varies from the lower tail to the upper tail.  The \Rpackage
\lgsdRpackage\ is designed for such investigations, and it includes
an interactive interface that can switch between different
visualisations.  Note that \lgsdRpackage\ also can deal with points
$\LGp$ that lies outside of the diagonal, and it can in addition
also digest multivariate time series.

\subsection{Real data and a fitted GARCH-type model}
\label{sec:Real_data}

The local Gaussian machinery will now be used on the
\texttt{dmbp}-data.  It will here be seen that local properties of the
nonlinear dependency structure 
indeed can be obtained by comparing $\widehatfz[m]{}(\omega)$ and
$\hatlgsdM{\LGp}{\omega}{m}$, and this works even for low values of
the truncation level $m$.

Another topic that it is natural to consider is the comparison of
$\hatlgsdM{\LGp}{\omega}{m}$ based on the 
data and $\hatlgsdM{\LGp}{\omega}{m}$ based on simulations from a
model fitted to the data --- and this will in particular be
investigated for a GARCH-type model that was fitted to the
\texttt{dmbp}-data by the \Rpackage \Rref{rugarch},
\citet{ghalanos15_rugarch}.

\subsubsection{The real data example }
\label{sec:Real_data_dmbp}

The \texttt{dmbp}-data (length~1974), whose original and
pseudo-normalised versions can be seen in
\cref{fig:dmbp_original_normalised}, will now be investigated by the
$m$-truncated local Gaussian spectral densities
$\lgsdM{\LGp}{\omega}{m}$.  These estimates will be based on the
bandwidth \mbox{$\bm{b}=\parenR{0.5,0.5}$}, and they will be
computed for the three diagonal points corresponding to the 10\%,
50\% and 90\% percentiles of the standard normal distribution.  The
estimated local Gaussian autocorrelations $\hatlgacr{\LGp}{h}$ that
is used in the computation of $\hatlgsdM{\LGp}{\omega}{m}$ can be
seen in \cref{fig:dmbp_lag}, and the estimated values of
$\widehatfz[m]{}(\omega)$ and $\hatlgsdM{\LGp}{\omega}{m}$ (for the
$m=10$ case) are shown as the red and blue solid lines\footnote{
  Solid lines are always used by the \Rpackage \lgsdRpackage\ when
  $\hatlgsdM{\LGp}{\omega}{m}$ is based on real data.} in
\cref{fig:dmbp}.  
The pointwise confidence intervals are based on the resampling
strategy discussed on page \pageref{page.the.resampling.algorithm}.

\begin{knitrout}\footnotesize
  \definecolor{shadecolor}{rgb}{0.969, 0.969, 0.969}\color{fgcolor}\begin{figure}[h]

    {\centering \includegraphics[width=1\linewidth]{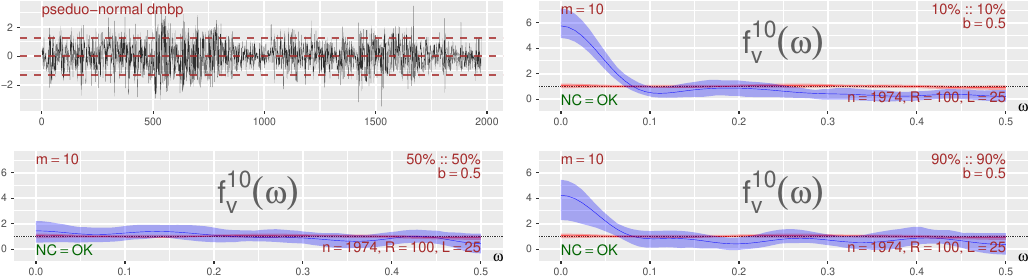} 

    }

    \caption[\texttt{dmbp}-data, bootstrapped based confidence
    intervals]{\texttt{dmbp}-data, bootstrapped based confidence
      intervals.  Global and local spectra for the three diagonal
      points.}\label{fig:dmbp}
  \end{figure}

\end{knitrout}

The global spectrum $\widehatfz[m]{}(\omega)$ is flat, which is in
agreement with the knowledge that the \texttt{dmbp}-data resembles
white noise.  The local Gaussian spectrum
$\hatlgsdM{\LGp}{\omega}{m}$ at the center is also rather flat,
which is no surprise given the values $\hatlgacr{\LGp}{h}$ seen in
the middle panel of \cref{fig:dmbp_lag}.  The estimates
$\hatlgsdM{\LGp}{\omega}{m}$ in the tails are obviously not flat,
and the clear peaks at the frequency $\omega=0$ are again in
agreement with the corresponding values $\hatlgacr{\LGp}{h}$ from
\cref{fig:dmbp_lag}.

The difference between the (solid lines in the) lower and upper tail
could indicate the presence of an asymmetry, i.e.\ the peak are more
prominent for the lower tail.  It would be premature to draw a firm
conclusion regarding asymmetry based one a single plot using the low
truncation level $m=10$, but the asymmetry can also be seen for
higher truncation levels (investigated up to $m=200$), with an
increasing difference between the height of these peaks.
Such an asymmetry, with a higher peak at the lower tail, would be in
agreement with the asymmetry between a \textit{bear market} (going
down) and a \textit{bull market} (going up).

A comparison solely based on the solid lines in \cref{fig:dmbp} is
not sufficient, since an observed difference could be due to the
variability of the estimator used to find
$\hatlgsdM{\LGp}{\omega}{m}$.  It is thus necessary to decide on a
reasonable resampling strategy (described below) that can provide
pointwise confidence intervals like those shown in \cref{fig:dmbp}.
Based on the pointwise confidence intervals, it is clear that
the 
truncated local and global spectra indeed do show that the
\texttt{dmbp}-data contains local non-linear dependency structures
in the tails.  
Note that the width of the pointwise confidence interval is a
function of the frequency, cf.\
\cref{app:finite_sample_and_variance_LGSD}, and this can in some
cases give it a wide \enquote{trumpet shape} near $\omega=0$, as
seen in the lower and upper tails in \cref{fig:dmbp} (and which is
even more prominent in \cref{fig:GARCH})

The pointwise confidence intervals in \cref{fig:dmbp} requires a
resampling strategy that takes into account that the local Gaussian
autocorrelations ${\lgacr{\LGp}{1},\dotsc,\lgacr{\LGp}{m}}$
are estimated by a local likelihood approach. The asymptotic
properties of these estimates were developed in the present paper
using the procedure from \citet{klimko1978}, cf.\ 
\cref{App:local_penalty_function_Klimko_Nelson_approach}.

The block bootstrap can be used for a variety of estimators, and it
can in particular, cf.\ \citet[Example 2.4, p.\
1219-20]{kuensch89:_jackk_boots_gener_station_obser}, be applied for
estimators based on the Klimko-Nelson procedure.  
The block bootstrap was thus used as the resampling strategy in an
earlier draft of this paper, and the results were similar to
\cref{fig:dmbp} when a block length of {$L=100$ \label{page:L=100}} was used.  The
selected block length $L$ seemed reasonable based on the
$\hatlgacr{\LGp}{h}$-values seen in \cref{fig:dmbp_lag}.
See \cref{app:What.about.the.ordinary.block.bootstrap?} for further
details.

Some comments related to the block bootstrap were received during
the review-process, and those motivated the investigation 
presented in \cref{app:regarding_resampling}, which 
lead to the adjusted resampling strategy given in
\label{page.the.resampling.algorithm}
\cref{def:index_based_block_bootstrap_for_tuples}.  The adjusted
resampling method uses a two step procedure, where the first step
uses the block bootstrap on the \textit{indices of the
  observations}, and the next step uses those resampled indices to
identify the $h$-lagged pairs $\parenR{\Yz{t+h},\Yz{t}}$ that should
be used when estimating $\lgacr{\LGp}{h}$ for the resampled data.

The adjusted resampling approach reduce the edge-effect noise that
occurs when the components of a resampled pair belong to different
blocks, and this implies that it works well with lower block lengths
than those needed for the ordinary block bootstrap.  A sensitivity
analysis related to the selection of the block length $L$ is
presented in \cref{app:Block_length_sensitivity}.

\subsubsection{A heatmap/distance plot for the \texttt{dmbp}-data} 
\label{sec:A.heatmap/distance.plot.for.the.dmbp-data}

It is of interest to know how $\hatlgsdM{\LGp}{\omega}{m}$ behaves
for other diagonal points, 
and this can be seen in \cref{fig:dmbp_heatmap_and__levels_vs_norm}
which is constructed in the same manner as
\cref{fig:dmt_heatmap_and__levels_vs_norm}. Keep in mind that these
plots are based on pseudo-normalised data,
i.e.\ the information in the marginal distribution is not present,
and \cref{fig:dmbp_heatmap_and__levels_vs_norm} thus primarily
reveals information about the 
copula-structure of the time series under investigation, cf.\
\cref{app:Point_sensitivity} for further details.

\begin{figure}[h]
  {\centering \includegraphics[width=1\textwidth]{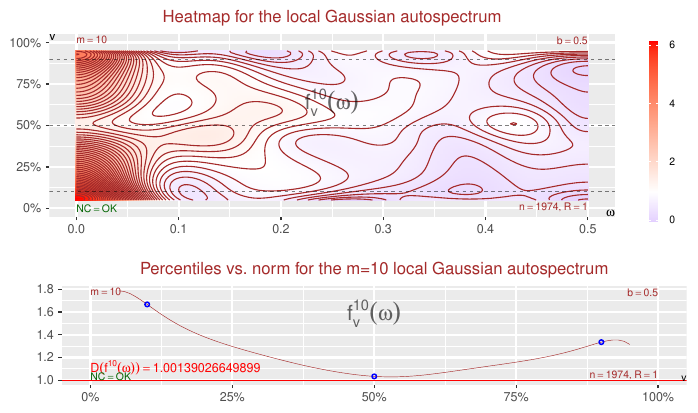}
  }
  \caption[]{Heatmap and corresponding distance-based plots based on
    the \texttt{dmbp-data}, showing how $\hatlgsdM{\LGp}{\omega}{10}$
    varies with the percentiles for the diagonal-points $\LGp$.  The
    percentiles used in \cref{fig:dmbp}, i.e.\ 10\%, 50\% and 90\%,
    have been
    highlighted with lines/points.}\label{fig:dmbp_heatmap_and__levels_vs_norm}
\end{figure}

\Cref{fig:dmbp_heatmap_and__levels_vs_norm} supports the impression
that there is an asymmetry between the lower tail and the upper
tail, and it can also be seen that the local dependency structure is
weak near the center.  Note that these plots go from the 5\% to 95\%
percentile, in order to show that it might be perilous to go too far
out in the tail for the present sample size ($n=1974$).  This is
discussed in more detail in \cref{app:Point_sensitivity}, where
heatmap based plots of the estimated underlying local Gaussian
autocorrelations can be found, cf.\
\cref{fig:dmbp_v_heatmap_lgacr,fig:dmt_v_heatmap_lgacr,fig:apARCH_v_heatmap_lgacr}.

\subsubsection{A GARCH-type model}
\label{sec:GARCH_model}

This section will consider 
an \textit{asymmetric power ARCH-model} (apARCH) of order $(2,3)$,
with parameters based on a fitting to the
\texttt{dmbp}-data.\footnote{
  The \Rpackage \Rref{rugarch}, \citet{ghalanos15_rugarch} was used
  to find the parameters of a multitude of GARCH-type models, and
  the asymmetric power ARCH model with the best fit was then
  selected.} 
Technical details about this model, and comments regarding the
script needed for the reproduction of this example, can be found in
\cref{app:fig:GARCH} in the Supplementary Material.

For a comparison with the results 
based on the \texttt{dmbp}-data, it is natural to consider $R=100$
samples of length $n=1974$ from the fitted apARCH$(2,3)$ model ---
and the estimates of $\lgsdM{\LGp}{\omega}{m}$ should be computed
for the same points $\LGp$ and with the same tuning parameters
$\bm{b}$ and $m$.  The result from such an investigation can be seen
in \cref{fig:GARCH}.

\begin{knitrout}\footnotesize
  \definecolor{shadecolor}{rgb}{0.969, 0.969, 0.969}\color{fgcolor}\begin{figure}[h]

    {\centering \includegraphics[width=1\linewidth]{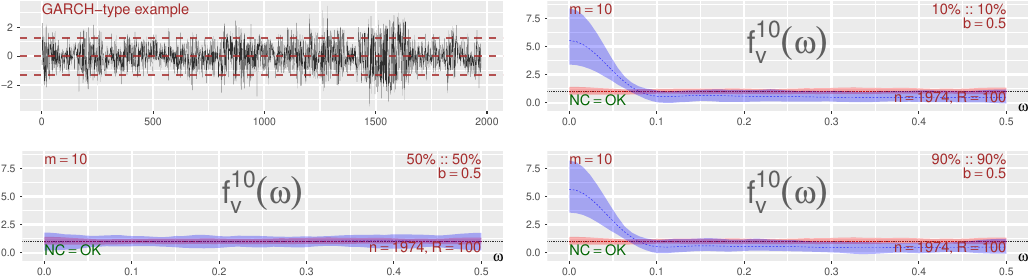} 

    }

    \caption[GARCH-type model, based on \texttt{dmbp}]{GARCH-type model,
      based on \texttt{dmbp}.  Global and local spectra for three
      points.}\label{fig:GARCH}
  \end{figure}

\end{knitrout}

It is clear from \cref{fig:GARCH} that the estimate of the
$m$-truncated global spectrum is flat, and this is in agreement with
the knowledge that $f(\omega) = 1$ for a GARCH-type model (since
\mbox{$\rho(h)=0$} when \mbox{$h\neq0$}).  It can also be seen that
the esimates $\hatlgsdM{\LGp}{\omega}{m}$  based on the fitted model
have the same overall structure as those in \cref{fig:dmbp}.  In
particular, there is a flat spectrum at the center, and the tails
show the presence of nonlinear structures with peaks at $\omega=0$.
\Cref{fig:GARCH} does however not pick up the apparent and
intuitively reasonable asymmetry seen in the solid lines in
\cref{fig:dmbp}, which also are supported by the plots in
\cref{fig:dmbp_heatmap_and__levels_vs_norm}.

\subsubsection{Local testing of fitted models}
\label{sec:local.testing.of.fitted.models}

A comparison of plots like those in \cref{fig:dmbp,fig:GARCH} can be
used to perform a \enquote{local sanity check} of whether or not the
dependency structure of the fitted model properly matches the
dependency structure of the data --- and it is also possible to
perform \enquote{local comparisons} of different models that has
been fitted to the same data.  The interested reader can find
similar local investigations of data and fitted models in e.g.\
\citet{li19:_quant_frequen_analy_spect_diver,BIRR2019122}.

Note that it for such comparisons also is of interest to include
points $\LGp$ outside the diagonal.  The plots needed for
off-diagonal points must take into account that
$\hatlgsdM{\LGp}{\omega}{m}$ will be complex-valued outside the
diagonal, but this has already been taken care of in the \Rpackage
\lgsdRpackage, where the implemented solution simply mimics the
co-spectra, quadrature-spectra, phase-spectra and amplitude spectra
that is used for the ordinary complex-valued cross-spectra.

An alternative strategy to the comparison of two sets of plots, like
those in \cref{fig:dmbp,fig:GARCH}, is to superimpose the
$\hatlgsdM{\LGp}{\omega}{m}$ from the \texttt{dmbp}-data on the top
of the corresponding plots based on the fitted model.  A plot based
on this superposition principle (inspired by a similar plot from
\citet{BIRR2019122}) is given in \cref{fig:aparch_dmbp_comparison}
in the Supplementary Material, cf.\
\cref{app:resampling_parametric_bootstrap}.  Note that this plot also
contains visualisations of complex-valued spectra.

\section{Conclusion}
\label{sec:Conclusions}

The local Gaussian spectral density $\lgsd[p]{\LGp}{\omega}$ has in
this paper been introduced as a new tool for the study of nonlinear
time-series.  The examples show that even for low truncation levels
it is possible to detect nonlinear periodicities missed by the
ordinary spectral density.  Further, one can detect the presence of
general nonlinear dependency structures by a comparison of the
$m$-truncated versions of the ordinary spectrum and the local
Gaussian spectra.

The $m$-truncated spectra $\lgsdM{\LGp}{\omega}{m}$ can also be of
interest with regard to 
\textit{local comparisons} of models fitted to a given sample, as
discussed at the end of \cref{sec:Real_data}.

The \Rpackage \lgsdRpackage\ can estimate $\lgsdM{\LGp}{\omega}{m}$
for a large number of combinations of points $\LGp$, truncation
levels $m$, and block lengths $\bm{b}$ --- and it does also have an
integrated \Rref{shiny}-application that enables an easy interactive
investigation of the results.  The Supplementary Material contains a
sensitivity analysis that shows how $\hatlgsdM{\LGp}{\omega}{m}$
reacts to adjustments of $\LGp$, $m$ and $\bm{b}$ --- and it is
there also seen that adjustments of the block length $L$, within
wide intervals,
have a minimal impact on the pointwise confidence intervals.

\section*{Acknowledgements}

The authors are most grateful for the valuable comments and
suggestions from the referees and the associate editor.

\section*{Supplementary Material}
The online Supplementary Material contains the appendices.  The
scripts needed for the reproduction of the examples in this paper is
contained in the \Rpackage \lgsdRpackage, cf.\
\cref{app:data_details} for further details.

\clearpage{}

  \putbib     
\end{bibunit}

\begin{bibunit}[elsarticle-harv]
  \spacingset{1.5} 

  {  \clearpage{}\clearpage

\appendix

\renewcommand\thefigure{\thesection.\arabic{figure}}
\setcounter{figure}{0}

\begin{center}
  {\Large SUPPLEMENTARY MATERIAL}
  \setcounter{page}{1}
\end{center}

This part contains the supplementary material to the paper
\textit{Nonlinear spectral analysis: A local Gaussian approach}.
The {asymptotic results for $\hatlgsdM[p]{\LGp}{\omega}{m}$} are
presented in \cref{app:asymptotics_for_hatlgsd},
\cref{sec:Technical_Results} contains the underlying {asymptotic
  results for the parameters
  $\widehatbmthetaz{\LGp|\overbar{m}|\bm{b}}$}, and a collection of
technical details is given in \cref{app:Tecnhical_details}.

A sensitivity analysis of the tuning parameters is given in
\cref{app:sensitivity_analysis}, and some comments related to the
selection of the tuning parameters are given in
\cref{How.to.select.the.tuning.parameters?}.
\Cref{app:regarding_resampling} discusses issues related to sampling
and resampling, including a sensitivity analysis of the block length
$L$ for the slightly adjusted block bootstrap that is used in this
paper.

Finally, \cref{app:data_details} contains some additional
information about the examples used in the main document, and it
does also include comments related to the reproducibility scripts
that are contained in the \Rpackage \lgsdRpackage.

\section{Asymptotic results for $\hatlgsdM[p]{\LGp}{\omega}{m}$}
\label{app:asymptotics_for_hatlgsd}

This appendix presents the asymptotic properties of
$\hatlgsdM[p]{\LGp}{\omega}{m}$, the \textit{$m$-truncated estimate of
  the local Gaussian spectral density}, i.e.\ the proof of
\cref{th:asymptotics_for_hatlgsd} is given here together with a
theorem that covers the case when $\hatlgsdM[p]{\LGp}{\omega}{m}$ is
complex-valued.  The technical details needed for the proofs are
covered in \cref{sec:Technical_Results,app:Tecnhical_details}.
Note that the theory is given for the general situation, i.e.\ it is
not required that the time series under investigation should have been
replaced with a pseudo-normalised version.

\subsection{The proof of \cref{th:asymptotics_for_hatlgsd}}
\label{sec:proof_of_th:asymptotics_for_hatlgsd}
\begin{proof}
  The property that $\lgsd[p]{\LGp}{\omega}$ is a real-valued function
  when $\LGp$ lies on the diagonal was proved in
  \myref{th:lgsd_properties}{th:lgsd_real_on_diagonal}.  The
  expression for $\hatlgsdM[p]{\LGp}{\omega}{m}$ from
  \cref{th:lgsd_estimate_real_valued} can by vectors be written~as
  \begin{align}
    \hatlgsdM[p]{\LGp}{\omega}{m} 
    & = 1 + 2 \cdot \bmLambdaz[\!\prime]{m}(\omega)\cdot \widehatbmPz{\LGp|m|\bm{b}},
  \end{align}
                                                                          i.e.\ the sum can be expressed as the inner product of the two
  vectors
  \begin{subequations}
    \begin{align}
      \label{eq:App_Lambda_vec}
      \bmLambdaz[\prime]{m}(\omega) 
      &\defeq \Vector{        \lambdazM{1}{m} \cdot \cos \left(2\pi\omega \cdot 1\right),         \dotsc,         \lambdazM{m}{m} \cdot \cos \left(2\pi\omega \cdot m\right)}, \\
      \label{eq:App_P_vec}
      \widehatbmPz{\LGp|m|\bm{b}}
      &\defeq \Vector[\prime]{
                                \hatlgacrb[p]{\LGp}{1}{\bmbzh{1}},         \dotsc,         \hatlgacrb[p]{\LGp}{m}{\bmbzh{m}}}.
    \end{align}
  \end{subequations}
  Since $\hatlgacrb[p]{\LGp}{h}{\bmbzh{h}}$ is one of the 5 estimated
  parameters $\hatlgthetab[p]{\LGp}{h}{\bmbzh{h}}$ from the local
  Gaussian approximation (of the lag~$h$ pairs) at the
  point~$\LGp$,\footnote{ The properties of
    $\hatlgthetab[p]{\LGp}{h}{\bmbzh{h}}$ was investigated in
    \citet{Tjostheim201333}.  A brief summary, with notation adjusted
    to fit the multivariate framework of the present paper, is given
    \cref{app:bivariate_penalty_functions}.}  it is clear that it is
  possible to write
  \mbox{$\hatlgacrb[p]{\LGp}{h}{\bmbzh{h}} = \bmez[\prime]{5}\cdot
    \hatlgthetab[p]{\LGp}{h}{\bmbzh{h}}$}, where $\bmez[\prime]{5}$ is
  the unit vector that picks out $\hatlgacrb[p]{\LGp}{h}{\bmbzh{h}}$
  from~$\hatlgthetab[p]{\LGp}{h}{\bmbzh{h}}$.  The vectors
  $\TSR{\hatlgthetab[p]{\LGp}{h}{\bmbzh{h}}}{h=1}{m}$ can be stacked
  on top of each other to give a joint parameter vector
  $\widehatbmthetaz{\LGp|\overbar{m}|\bm{b}}$, and it follows that the
  vector $\widehatbmPz{\LGp|m|\bm{b}}$ can be expressed as
  \mbox{$\widehatbmPz{\LGp|m|\bm{b}} = \bmEz[\prime]{m}\cdot
    \widehatbmthetaz{\LGp|\overbar{m}|\bm{b}}$}, where
  $\bmEz[\prime]{m}$ is the matrix that picks out the relevant
  components from $\widehatbmthetaz{\LGp|\overbar{m}|\bm{b}}$.
  It follows from this, and \citet[Proposition~6.4.2,
  p.~211]{Brockwell:1986:TST:17326}, that an asymptotic normality
  result for $\widehatbmthetaz{\LGp|\overbar{m}|\bm{b}}$ will give an
  asymptotic normality result for $\hatlgsdM[p]{\LGp}{\omega}{m}$.  In
  particular, if a suitable scaling factor\footnote{ $\cz{n|m|\bm{b}}$
    must be a function of $n$, $m$ and $\TSR{\bmbzh{h}}{h=1}{m}$, such
    that \mbox{$\cz{n|m|\bm{b}}\rightarrow\infty$} when $\nlimit$,
    $\mlimit$ and \mbox{$\bmbzh{h}\rightarrow\bmzeroz[+]{}$}.}
  $\cz{n|m|\bm{b}}$ gives a \mbox{$5m$-variate} asymptotic normality
  result for~$\widehatbmthetaz{\LGp|\overbar{m}|\bm{b}}$,
  \begin{align}
    \label{eq:asymptotic_limit_for_theta}
    \cz{n|m|\bm{b}} \cdot
    \left(\widehatbmthetaz{\LGp|\overbar{m}|\bm{b}} -
    \bmthetaz{\LGp|\overbar{m}}
    \right) \stackrel{\scriptscriptstyle
    d}{\longrightarrow} \UVN{\bm{0}}{\Sigmaz{\LGp|\overbar{m}}},
  \end{align}
  then a scaling factor $\cz[\,\prime]{n|m|\bm{b}}$ can be found
  that gives a univariate asymptotic normality result
  for~$\hatlgsdM[p]{\LGp}{\omega}{m}$,
  \begin{align}
    \label{eq:asymptotic_limit_for_hatlgsd}
    \cz[\,\prime]{n|m|\bm{b}} \cdot
    \left(\hatlgsdM[p]{\LGp}{\omega}{m} -
    \lgsd[p]{\LGp}{\omega} 
    \right) \stackrel{\scriptscriptstyle
    d}{\longrightarrow} \UVN{\bm{0}}{\sigmaz[2]{\LGp}(\omega)}, 
  \end{align}
  where the variance $\sigmaz[2]{\LGp}(\omega)$ is a suitably scaled
  version of the limit of
  \begin{align}
    \nonumber
    \Var{\hatlgsdM[p]{\LGp}{\omega}{m}}
    &= 4\cdot \Var{\bmLambdaz[\!\prime]{m}(\omega)\cdot
      \bmEz[\prime]{m}\cdot\widehatbmthetaz{\LGp|\overbar{m}|\bm{b}}} \\
    \label{eq:lgsde_var_m}
    &= 4\cdot \bmLambdaz[\prime]{m}(\omega)\cdot \bmEz[\prime]{m}\cdot
    \Var{\widehatbmthetaz{\LGp|\overbar{m}|\bm{b}}}\cdot
    \bmEz{m}\cdot \bmLambdaz{m}(\omega).
  \end{align}

  The asymptotic normality required in
  \cref{eq:asymptotic_limit_for_theta} follows from
  \cref{th:asymptotics_for_hatlgtheta} (page
  \pageref{th:asymptotics_for_hatlgtheta}), i.e.\ the scaling factor
  $\cz{n|m|\bm{b}}$ will be
  $\sqrt{n \!\parenRz[3]{}{\bz{1}\bz{2}}}$, whereas the
  asymptotic covariance matrix $\Sigmaz{\LGp|\overbar{m}}$ can be
  written as the direct sum of the covariance matrices for
  \mbox{$\sqrt{n \!\parenRz[3]{}{\bz{1}\bz{2}}}\cdot
    \hatlgthetab[p]{\LGp}{h}{\bmbzh{h}}$}, i.e.\
  \begin{align}
    \Var{\sqrt{n \!\parenRz[3]{}{\bz{1}\bz{2}}}\cdot
    \widehatbmthetaz{\LGp|\overbar{m}|\bm{b}}}  
    = \oplusss{h=1}{m} \Var{\sqrt{n
    \!\parenRz[3]{}{\bz{1}\bz{2}}}\cdot
    \hatlgthetab[p]{\LGp}{h}{\bmbzh{h}}}, 
  \end{align}
  from which a simple calculation gives
  \begin{align}
    \label{eq:lgsde_var_m_computed}
    \Var{\sqrt{n \!\parenRz[3]{}{\bz{1}\bz{2}}}\cdot
    \hatlgsdM[p]{\LGp}{\omega}{m}} =     4\cdot \sumss{h=1}{m}     \lambdazM[2]{h}{m} \cdot     \subp{\cos}{}{}{}{2} (2\pi\omega h)     \cdot \Var{\sqrt{n \!\parenRz[3]{}{\bz{1}\bz{2}}}\cdot
    \hatlgacrb[p]{\LGp}{h}{\bmbzh{h}}}. 
  \end{align}

  From this it is clear that the scaling factor $\cz{n|m|\bm{b}}$
  requires an additional scaling with $\sqrt{1/m}$ in order to include
  the averaging factor $1/m$ for the sum in
  \cref{eq:lgsde_var_m_computed}.  Thus,
  \mbox{$\cz[\,\prime]{n|m|\bm{b}} = \sqrt{n
      \!\parenRz[3]{}{\bz{1}\bz{2}} \!/ m}$}, which completes
  the proof.
                                                \end{proof}

  Some care must be taken formally with regard to the limiting
  \mbox{$5m$-variate} normal distribution in
  \cref{eq:asymptotic_limit_for_theta}, since it has to be interpreted
  as something that is approximately valid for large (but finite)
  values of the truncation point~$m$.  The univariate normal
  distribution in \cref{eq:asymptotic_limit_for_hatlgsd} is the one of
  interest, and this will under the required assumptions
    be well defined in the limit.

\subsection{The complex-valued case}
\label{sec:lgsd_estimate_complex_valued_case}

\begin{theorem}[Complex-valued case]
    \label{th:asymptotics_for_hatlgsd_complex_valued}
                            If the local Gaussian spectral density $\lgsd[p]{\LGp}{\omega}$ is a
  complex valued function for a point
  \mbox{$\LGp=\parenR{\LGpi{1},\LGpi{2}}$}, i.e.\
  \mbox{$\lgsd[p]{\LGp}{\omega} = \lgsdRE[p]{\LGp}{\omega} -
    i\lgsdIM[p]{\LGp}{\omega}$}, with
  \mbox{$\lgsdIM[p]{\LGp}{\omega}\not\equiv0$}, then, under
  \cref{assumption_Yt,assumption_score_function,assumption_Nmb}, the components
  $\hatlgsdREM[p]{\LGp}{\omega}{m}$ and
  $\hatlgsdIMM[p]{\LGp}{\omega}{m}$ of the \mbox{$m$-truncated}
  estimate $\hatlgsdM[p]{\LGp}{\omega}{m}$ will, when
  $\omega\not\in\tfrac{1}{2}\cdot\ZZ\defeq\parenC{\dotsc,-1,-\tfrac{1}{2},0,\tfrac{1}{2},1,\dotsc}$,
  be jointly asymptotically normally distributed as given below.
    \begin{align}
    \label{eq:th:asymptotics_for_hatlgsd_off_diagonal_Yt_not_reversible}
    \sqrt{n \!\parenRz[3]{}{\bz{1}\bz{2}} \!/ m} \cdot     \parenR{
    \begin{bmatrix}
      \hatlgsdREM[p]{\LGp}{\omega}{m} \\
      \hatlgsdIMM[p]{\LGp}{\omega}{m}
    \end{bmatrix}
    -     \begin{bmatrix}
      \lgsdRE[p]{\LGp}{\omega} \\
      \lgsdIM[p]{\LGp}{\omega}
    \end{bmatrix} }
    \stackrel{\scriptscriptstyle d}{\longrightarrow}
    \UVN{
    \begin{bmatrix}
      0 \\
      0
    \end{bmatrix}}{
    \begin{bmatrix}
      \sigmaz[2]{\!c:\LGp}(\omega) & 0 \\
      0 &\sigmaz[2]{\!q:\LGp}(\omega) 
    \end{bmatrix}},
  \end{align}
  where the variances $\sigmaz[2]{\!c:\LGp}(\omega)$ and
  $\sigmaz[2]{\!q:\LGp}(\omega)$ are given by
  \begin{subequations}
    \label{eq:th:asymptotics_for_hatlgsd_variance_RE_and_IM}
    \begin{align}
    \label{eq:th:asymptotics_for_hatlgsd_variance_RE}
      \sigmaz[2]{\!c:\LGp}(\omega)
      &=         \lim_{\mlimit} \frac{1}{m} \sumss{h=1}{m}       \lambdazM[2]{h}{m} \cdot       \subp{\cos}{}{}{}{2} (2\pi\omega h)       \cdot \parenC{
        \tildesigmaz[2]{\!\LGp}(h) +   \tildesigmaz[2]{\!\LGpd}(h) 
        } \\
    \label{eq:th:asymptotics_for_hatlgsd_variance_IM}
      \sigmaz[2]{\!q:\LGp}(\omega)
      &=         \lim_{\mlimit} \frac{1}{m} \sumss{h=1}{m}       \lambdazM[2]{h}{m} \cdot       \subp{\sin}{}{}{}{2} (2\pi\omega h)       \cdot \parenC{
        \tildesigmaz[2]{\!\LGp}(h) +   \tildesigmaz[2]{\!\LGpd}(h) 
        },
    \end{align}
  \end{subequations}
                with $\tildesigmaz[2]{\!\LGp}(h)$ and
  $\tildesigmaz[2]{\!\LGpd}(h)$ related to respectively
  $\hatlgacrb[p]{\LGp}{h}{\bmbzh{h}}$ and
  $\hatlgacrb[p]{\LGpd}{h}{\bmbzh{h}}$ as given in
  \cref{th:asymptotics_for_hatlgsd}.

  The component $\hatlgsdIMM[p]{\LGp}{\omega}{m}$ is identical to~0
  when $\omega\in\tfrac{1}{2}\cdot\ZZ$, and for these frequencies the
  following asymptotic result holds under the given assumptions
  \begin{align}
    \label{eq:th:asymptotics_for_hatlgsd_off_diagonal_Yt_not_reversible_corner_cases}
    \sqrt{n \!\parenRz[3]{}{\bz{1}\bz{2}} \!/ m} \cdot     \parenR{
      \hatlgsdM[p]{\LGp}{\omega}{m} - 
      \lgsd[p]{\LGp}{\omega}}
    \stackrel{\scriptscriptstyle d}{\longrightarrow}
    \UVN{0}{\sigmaz[2]{\!c:\LGp}(\omega)}.
  \end{align}
  
\end{theorem}

\begin{proof}
  The case $\omega\in\tfrac{1}{2}\cdot\ZZ$ can be proved by the exact
  same argument that was used in the proof of
  \cref{th:asymptotics_for_hatlgsd}, whereas the general case requires
  a bivariate extension of that proof.  In particular, when the proof
  of \cref{th:asymptotics_for_hatlgsd} is used on
  $\hatlgsdREM[p]{\LGp}{\omega}{m}$ and
  $\hatlgsdIMM[p]{\LGp}{\omega}{m}$, it follows that they can be
  written as
    \begin{subequations}
    \label{eq:M_truncated_lgsdRE_and_lgsdIM_vectors}
    \begin{align}
      \label{eq:M_truncated_lgsdRE_vector}
      \hatlgsdREM[p]{\LGp}{\omega}{m} 
      &= 1 +         \bmLambdaz[\!\prime]{c|m}(\omega)\cdot \widehatbmPz{\LGp|m|\bm{b}} +         \bmLambdaz[\!\prime]{c|m}(\omega)\cdot \widehatbmPz{\LGpd|m|\bm{b}}         = 1 + \bmLambdaz[\!\prime]{c|\OUbar{m}}(\omega)\cdot
        \widehatbmPz{\LGp|\OUbar{m}|\bm{b}} \\
            \label{eq:M_truncated_lgsdIM_vector}
      \hatlgsdIMM[p]{\LGp}{\omega}{m} 
      &= 0 +         \bmLambdaz[\!\prime]{q|m}(\omega)\cdot \widehatbmPz{\LGp|m|\bm{b}} -         \bmLambdaz[\!\prime]{q|m}(\omega)\cdot \widehatbmPz{\LGpd|m|\bm{b}}         = 0 + \bmLambdaz[\!\prime]{q|\OUbar{m}}(\omega)\cdot
        \widehatbmPz{\LGp|\OUbar{m}|\bm{b}},
    \end{align}
  \end{subequations}
    where $\bmLambdaz[\!\prime]{c|m}(\omega)$ and
  $\bmLambdaz[\!\prime]{q|m}(\omega)$ are the coefficient vectors
  containing respectively the cosines and sines, where
  $\widehatbmPz{\LGp|m|\bm{b}}$ and $\widehatbmPz{\LGpd|m|\bm{b}}$
  contains the estimated correlations corresponding to $\LGp$ and
  $\LGpd$ for the lags under consideration, and where the length~$2m$
  vectors $\bmLambdaz[\!\prime]{c|\OUbar{m}}(\omega)$,
  $\bmLambdaz[\!\prime]{q|\OUbar{m}}(\omega)$ and
  $\widehatbmPz{\LGp|\OUbar{m}|\bm{b}}$ are defined in the obvious
  manner in order to get a more compact notation.
              Following the same line of argument as in the proof of
  \cref{th:asymptotics_for_hatlgsd}, it follows that
  \mbox{$\widehatbmPz{\LGp|\OUbar{m}|\bm{b}} = \parenR{
      \bmEz[\prime]{m} \oplus \bmEz[\prime]{m}}\cdot
    \widehatbmThetaz{\OUbar{m}|\bm{b}}\!\parenR{\LGp,\LGpd}$}, where
  $\widehatbmThetaz{\OUbar{m}|\bm{b}}\!\parenR{\LGp,\LGpd}$ is the
  full set of estimated parameters from the local Gaussian
  approximations at $\LGp$ and $\LGpd$ for the lags under
  consideration,\footnote{    The vector
    $\widehatbmThetaz{\OUbar{m}|\bm{b}}\!\parenR{\LGp,\LGpd}$ can be
    expressed as a combination of
    $\widehatbmthetaz{\LGp|\overbar{m}|\bm{b}}$ and
    $\widehatbmthetaz{\LGpd|\overbar{m}|\bm{b}}$, where
    $\widehatbmthetaz{\LGp|\overbar{m}|\bm{b}}$ is the parameter
    vector from the proof of \cref{th:asymptotics_for_hatlgsd}.}   and where
  \mbox{$\parenR{ \bmEz[\prime]{m} \oplus \bmEz[\prime]{m}}$} is
  the matrix that picks out the relevant autocorrelations.

  Based upon this, it follows that the target of interest can be
  written as
  \begin{align}
    \begin{bmatrix}
      \hatlgsdREM[p]{\LGp}{\omega}{m} \\
      \hatlgsdIMM[p]{\LGp}{\omega}{m}
    \end{bmatrix} =     \begin{bmatrix}
      1 \\
      0
    \end{bmatrix} +     \begin{bmatrix}
      \bmLambdaz[\!\prime]{c|\OUbar{m}}(\omega) \\
      \bmLambdaz[\!\prime]{q|\OUbar{m}}(\omega)
    \end{bmatrix}\cdot
    \parenR{ \bmEz[\prime]{m} \oplus \bmEz[\prime]{m}}\cdot
    \widehatbmThetaz{\OUbar{m}|\bm{b}}\!\parenR{\LGp,\LGpd},
  \end{align}
  which together with the asymptotic normality result from
  \cref{th:asymptotics_for_hatlgtheta_LGp_and_LGpd}, i.e.\
  \begin{align}
        \sqrt{n \!\parenRz[3]{}{\bz{1}\bz{2}}} \cdot     \left(\widehatbmThetaz{\OUbar{m}|\bm{b}}\!\parenR{\LGp,\LGpd} -
    \bmThetaz{\OUbar{m}}\!\parenR{\LGp,\LGpd} \right)     \stackrel{\scriptscriptstyle d}{\longrightarrow}
    \UVN{\bm{0}}{      \Sigmaz{\LGp|\overbar{m}} \oplus \Sigmaz{\LGpd|\overbar{m}}
    },
  \end{align}
  gives the result when the arguments in the proof of
  \cref{th:asymptotics_for_hatlgsd} are applied to the present setup.
  Note that the requirement $\omega\not\in\tfrac{1}{2}\cdot\ZZ$ is
  needed in order to ensure that the variance
  $\sigmaz[2]{\!q:\LGp}(\omega)$ is different from~0, which is
  needed in order for \citet[Proposition~6.4.2,
  p.~211]{Brockwell:1986:TST:17326} to be valid in this~case.
\end{proof}

\subsection{The finite sample case and the variance of $\hatlgsdM{\LGp}{\omega}{m}$}
\label{app:finite_sample_and_variance_LGSD}

The variance of the estimated local Gaussian spectral density
$\hatlgsdM{\LGp}{\omega}{m}$, as seen in \cref{eq:lgsde_var_m}, is a
function of both the point $\LGp$ and the frequency $\omega$.  It is
with regard to this of interest to note that the variance
$\sigmaz[2]{\!\LGp}(\omega)$ is symmetric around
$\omega=\tfrac{1}{4}$, and it attains its highest values when
$\omega \in \parenC{0,\tfrac{1}{2}}$.  This symmetry is a consequence
of the fact that all the correlation terms are asymptotically
negligible.

The correlation-terms are however still present in the $m$-truncated
case, and this changes the situation a bit.  To clarify: The
correlation terms will depend on the frequency $\omega$ trough the
functions $\cos(2\pi\omega k)\cdot\cos(2\pi\omega \ell)$, and these
functions are in general not symmetrical around $\omega=\tfrac{1}{4}$.
For $\omega=0$ all these products are equal to~1, whereas the value
for $\omega=\tfrac{1}{2}$ will be given by
$\cos(2\pi\omega k)\cdot\cos(2\pi\omega \ell) = (-1)^{k+\ell}$.  The
consequence of this is that the highest value of this variance is
obtained at $\omega=0$ --- which in particular was evident in the
plots related to the apARCH-model and the \texttt{dmpb}-data, cf.\
\cref{fig:GARCH,fig:dmbp} on pages \pageref{fig:GARCH} and
\pageref{fig:dmbp} in the main document, where a \enquote{trumpet
  shape} could be seen for the pointwise confidence intervals near
$\omega=0$.

\section{Asymptotic results for $\widehatbmthetaz{\LGp|\overbar{m}|\bm{b}}$}
\label{sec:Technical_Results}
\setcounter{figure}{0}

This section will investigate the asymptotic properties of the
parameter vector $\widehatbmthetaz{\LGp|\overbar{m}|\bm{b}}$, that
is used in the proof of \cref{th:asymptotics_for_hatlgsd}.  The proof
is similar in spirit to the one used in \citet{Tjostheim201333} for
the asymptotic investigation of the parameter vectors
$\hatlgthetab[p]{\LGp}{h}{\bmbzh{h}}$, i.e.\ the Klimko-Nelson penalty
function approach will be used to derive the desired result.

\Cref{App:local_penalty_function_Klimko_Nelson_approach} explains the
Klimko-Nelson approach and shows how a local penalty function for the
present case can be constructed based on the local penalty function
encountered in \citet{Tjostheim201333}.
\Cref{App:A4_requirement_general_case} verifies the fourth of the
requirements needed for the Klimko-Nelson approach, and the asymptotic
results for $\widehatbmthetaz{\LGp|\overbar{m}|\bm{b}}$ are
collected in \cref{app:asymptotics_for_theta}.  

  The asymptotic investigation requires several indices in order to
  keep track of the different components, and to simplify references
  to $\LGp$ and $\bm{b}$ will whenever possible be suppressed
  from the~notation.

\makeatletter{}

\subsection{Local penalty functions and the Klimko-Nelson approach}
\label{App:local_penalty_function_Klimko_Nelson_approach}

\citet{Tjostheim201333} used a local penalty function to define the
\textit{local Gaussian correlation}~$\lgcor[5]{\LGp}$ as a new
\textit{local measure of dependence} at a point~$\LGp$, and then used
the approach formalised
in \citet{klimko1978}, to investigate the asymptotic properties
of~$\hatlgcor[5]{\LGp}$.
The \textit{local Gaussian spectral density}~$\lgsd[p]{\LGp}{\omega$} is
based on the local Gaussian autocorrelations~$\lgacr[p]{\LGp}{h}$, and
the asymptotic properties of the estimates $\hatlgsdM[p]{\LGp}{\omega}{m}$
are thus closely connected to the asymptotic properties
of~$\hatlgacr[p]{\LGp}{h}$.

The Klimko-Nelson approach shows how the asymptotic properties of
\textit{an estimate of the parameters of a penalty function $Q$} can
be expressed relative to the asymptotic properties of (entities
related to) the penalty function itself.
This result plays a pivotal role in the present analysis, and it has
thus been included in \cref{app:Klimko-Nelson-approach}.

\Cref{app:bivariate_penalty_functions} presents the bivariate
definitions and results from~\citet{Tjostheim201333}, with the
notational modifications that are needed in order to make it fit into
the multivariate approach in the present~paper.  The bivariate penalty
functions~$\QhN{h}{n}$ from \citet{Tjostheim201333} will be used as
building blocks for the new penalty function.

\makeatletter{}
\subsubsection{The Klimko-Nelson approach}
\label{app:Klimko-Nelson-approach}

The following presentation is based on
\citet[Th.~3.2.23]{taniguchi00:_asymp_theor_statis_infer_time_series}.

Let $\TSR{\bmXz{t}}{t\in\ZZ}{}$ be an \mbox{$m$-variate} strictly
stationary and ergodic process that satisfies the requirement
\mbox{$\E{\parenABSz[2]{}{\bmXz{t}}}<\infty$}.  Consider a general
real valued penalty function
{$\Qz{n}=\Qz{n}(\bm{\theta}) =
  \Qz{n}\!\left(\bmXz{1},\dotsc,\bmXz{n};\bm{\theta}\right)$}, which
should depend upon $n$ observations $\TSR{\bmXz{t}}{i=1}{n}$ and a
parameter vector $\bm{\theta}$ that lies in an open set
\mbox{$\bm{\Theta}\in\RRn{p}$}, and let the true value of the
parameter be denoted by $\bmthetaz[\circ]{}$.  Add the requirement
that $\Qz{n}$ must be twice continuously differentiable with respect
to $\bm{\theta}$ a.e.\ in a neighbourhood $\mathcal{N}$ of
$\bmthetaz[\circ]{}$, such that the following Taylor expansion is
valid (in the neighbourhood $\mathcal{N}$) for
\mbox{$\parenABS{\bm{\theta} - \bmthetaz[\circ]{}} < \delta$},
\begin{subequations}
  \label{eq:Klimko_Nelson_Taylor_expansion}
  \begin{align}
    \nonumber
    \Qz{n}(\bm{\theta})      &= \Qz{n}\!\left(\bmthetaz[\circ]{}\right) +
      \parenRz[\prime]{}{\bm{\theta} - \bmthetaz[\circ]{}}
      \frac{\partial}{\partial
      \bm{\theta}}\Qz{n}\!\left(\bmthetaz[\circ]{}\right) +
      \frac{1}{2} \parenRz[\prime]{}{\bm{\theta} - \bmthetaz[\circ]{}}
      \frac{\partialz[2]{}}{\partial\bm{\theta}\partial
      \bmthetaz[\prime]{}} 
      \Qz{n}\!\left(\bmthetaz[\circ]{}\right) 
      \parenRz{}{\bm{\theta} - \bmthetaz[\circ]{}} \\
    \label{eq:Klimko_Nelson_Taylor_expansion_before_matrices}
    &\phantom{=\ } +
      \frac{1}{2} \parenRz[\prime]{}{\bm{\theta} - \bmthetaz[\circ]{}}
      \parenC{\frac{\partialz[2]{}}{\partial\bm{\theta}\partial
      \bmthetaz[\prime]{}} 
      \Qz{n}\!\left(\bmthetaz[*]{}\right) -
      \frac{\partialz[2]{}}{\partial\bm{\theta}\partial 
      \bmthetaz[\prime]{}} 
      \Qz{n}\!\left(\bmthetaz[\circ]{}\right) }
      \parenRz{}{\bm{\theta} - \bmthetaz[\circ]{}} \\
      \nonumber
    &= \Qz{n}\!\left(\bmthetaz[\circ]{}\right) +
      \parenRz[\prime]{}{\bm{\theta} - \bmthetaz[\circ]{}}
      \frac{\partial}{\partial
      \bm{\theta}}\Qz{n}\!\left(\bmthetaz[\circ]{}\right) +
      \frac{1}{2} \parenRz[\prime]{}{\bm{\theta} - \bmthetaz[\circ]{}}
      \Vz{n}       \parenRz{}{\bm{\theta} - \bmthetaz[\circ]{}} \\
    \label{eq:Klimko_Nelson_Taylor_expansion_with_matrices}
    &\phantom{=\ } +
      \frac{1}{2} \parenRz[\prime]{}{\bm{\theta} - \bmthetaz[\circ]{}}
      \Tz{n}\!\left(\bmthetaz[*]{}\right) 
      \parenRz{}{\bm{\theta} - \bmthetaz[\circ]{}} 
  \end{align}
\end{subequations}
where $\Vz{n}$ and $\Tz{n}\!\left(\bmthetaz[*]{}\right)$ are defined
in the obvious manner, with
\mbox{$\bmthetaz[*]{} =
  \bmthetaz[*]{}\!\left(\bmXz{1},\dotsc,\bmXz{n};\bm{\theta}\right)$}
an intermediate point between $\bm{\theta}$ and
$\bmthetaz[\circ]{}$ (determined by the mean value theorem).

\begin{theorem}[Klimko-Nelson, \citet{klimko1978}]   \label{th:Klimko_Nelson_1978}
  Assume that $\TSR{\bmXz{t}}{t\in\ZZ}{}$ and $\Qz{n}$ are such that as
  $\nlimit$
  \begin{enumerate}[label=\emph{(A\arabic*)}]
  \item     \label{th:Klimko_Nelson_1978_A1}
    $\nz[-1]{}(\partial/\partial\bm{\theta}) \Qz{n}\!\left(
      \bmthetaz[\circ]{}\right)
    \stackrel{\scriptscriptstyle a.s.\ }{\longrightarrow} \bm{0}$,
  \item     \label{th:Klimko_Nelson_1978_A2}
    $\nz[-1]{} \Vz{n} \stackrel{\scriptscriptstyle
      a.s.\ }{\longrightarrow} V$, where $V$ is a \mbox{$p\times p$}
    positive definite matrix, and
  \item     \label{th:Klimko_Nelson_1978_A3}
  for $j,k = 1,\dotsc,p$
  \begin{align}
    \adjustlimits \lim_{\nlimit} \sup_{\delta\rightarrow 0}
    \parenRz[-1]{}{n\delta}
    \parenAbs{\power{jk}{\Tz{n}\!\parenC{\bmthetaz[*]{}}}} <
    \infty\qquad \text{a.s.}  
  \end{align}
  where
  $\power{jk}{\Tz{n}\!\parenC{\bmthetaz[*]{}}}$
  is the \mbox{$(j,k)$}th component of
  $\Tz{n}\!\parenC{\bmthetaz[*]{}}$.
  \end{enumerate}
  Then there exists a sequence of estimators
  \mbox{$\widehatbmthetaz{n} =
    \parenRz[\prime]{}{\widehatthetaz{1},\dotsc,\widehatthetaz{p}}$},
  such that
  \mbox{$\widehatbmthetaz{n} \stackrel{\scriptscriptstyle a.s.\
    }{\longrightarrow} \bmthetaz[\circ]{}$}, and for any
  \mbox{$\epsilon>0$}, there exists an event $E$ with
  \mbox{$P(E)>1-\epsilon$} and an $\nz[\circ]{}$ such that on $E$, for
  \mbox{$n>\nz[\circ]{}$},
  \mbox{$(\partial/\partial\bm{\theta})\Qz{n}(\widehatbmthetaz{n}) =
    \bm{0}$} and $\Qz{n}$ attains a relative minimum at
  $\widehatbmthetaz{n}$.  Furthermore, if
  \begin{enumerate}[label=\emph{(A\arabic*)}, resume]
  \item     \label{th:Klimko_Nelson_1978_A4}
    $\nz[-1/2]{}(\partial/\partial\bm{\theta})\Qz{n}\!\left(
      \bmthetaz[\circ]{}\right) 
    \stackrel{\scriptscriptstyle d}{\longrightarrow} \UVN{\bm{0}}{W}$
  \end{enumerate}
  then
  \begin{align}
    \label{eq:Klimko_Nelson_original_theorem}
    \nz[1/2]{}(\widehatbmthetaz{n} -
      \bmthetaz[\circ]{})
    \stackrel{\scriptscriptstyle d}{\longrightarrow} 
    \UVN{\bm{0}}{\Vz[-1]{}\Wz{}\Vz[-1]{}}.
  \end{align}
\end{theorem}

\makeatletter{}
\subsubsection{The bivariate penalty functions}
\label{app:bivariate_penalty_functions}

This section will translate the bivariate results from
\citet{Tjostheim201333} into the present multivariate framework, and
these bivariate components will then be used to define a new penalty
function in \cref{app:new_penalty_function}.

The main idea from \citet{Tjostheim201333} is to use bivariate
Gaussian densities $\psi\!\left(\yh{h}; \thetah[\LGp]{h}\right)$ to
approximate the bivariate densities $\gh{h}\!\left(\yh{h}\right)$
at a point~$\LGp$, where
\mbox{$\thetah[\LGp]{h} =
  \Vector[\prime]{\thetahcomp[\LGp]{h}{1},\dotsc,\thetahcomp[\LGp]{h}{5}}$}
is the five dimensional parameter-vector of the bivariate Gaussian
distribution.  The point $\LGp$ will be fixed for the remainder
of this discussion, and it will henceforth be dropped from the
notation for the parameters, i.e.\ $\thetah{h}$ should always be
understood as~$\thetah[\LGp]{h}$.

The local investigation requires a bandwidth vector
\mbox{$\bm{b}=\left(\bz{1},\bz{2}\right)$} and a kernel
function~$K(\bm{w})$, which is used to define
$\Khbdef\defeq \tfrac{1}{\bz{1}\bz{2}} K\!\left(  \tfrac{\yz{h} - \LGpi{1}}{\bz{1}},   \tfrac{\yz{0} - \LGpi{2}}{\bz{2}}\right)$, which in turn is used in the following local approximation
around~$\LGp$,
\begin{align}
  \label{eq:penalty_qh}
  \qh{h}{b}
  &\defeq \intss{\RRn{2}}{} \Khbdef \left[     \psi\!\left(\yh{h}; \thetah{h}\right)     - \gh{h}\!\left(\yh{h}\right)    \log\psi\!\left(\yh{h}; \thetah{h}\right)   \right] \dyh{h},
\end{align}
a minimiser of which should satisfy the vector equation
\begin{align}
  \label{eq:def_of_theta_hb} 
  \intss{\RRn{2}}{} \Khbdef \uh[\yh{h}; \thetah{h}]{h} \left[   \psi\!\left(\yh{h}; \thetah{h}\right) -  \gh{h}\!\left(\yh{h}\right)   \right] \dyh{h} &= \bm{0},
\end{align}
where
\mbox{$\uh[\yh{h}; \thetah{h}]{h} \defeq \nablah{h} \log
  \psi\!\left(\yh{h}; \thetah{h}\right)$} is the score function of
$\psi\!\left(\yh{h}; \thetah{h}\right)$ (with \mbox{$\nablah{h}   \defeq \partial/\partial \thetah{h}$}).
Under the assumption that there is a bandwidth $\bmbz{0}$ such that
there exists a minimiser $\thetahb{h}{b}$ of \cref{eq:penalty_qh}
which satisfies \cref{eq:def_of_theta_hb} for any $\bm{b}$ with
\mbox{$\bm{0} < \bm{b} < \bmbz{0}$},\footnote{      Inequalities involving vectors are to be interpreted in a
  component-wise manner.} this $\thetahb{h}{b}$ will be referred to as the population~value for
the given bandwidth~$\bm{b}$.

\Cref{eq:penalty_qh} is a special case of a tool that
\citet{hjort96:_local} introduced in order to perform \textit{locally
  parametric nonparametric density estimation}, but (as was done in
\citet{Tjostheim201333}) it can also be used to define and estimate
local Gaussian parameters --- whose asymptotic properties can be
investigated by means of a local penalty function
$\QhN[\thetah{h}]{h}{n}$, to be described below, and the Klimko-Nelson
approach.

For a sample of size~$n$ from $\TSR{\Yht{h}{t}}{t\in\ZZ}{}$, the
following \mbox{$M$-estimator}\footnote{  The entity $\LhN[\thetah{h}]{h}{n}$ can for independent observations
  be thought of as a \textit{local log-likelihood} or a \textit{local
    kernel-smoothed log-likelihood}, see
  \citet[Section~2-3]{hjort96:_local} for details.  In the realm of
  time series, where the observations are dependent, it is according
  to \citet[page~36]{Tjostheim201333} better to interpret it as an
  $M$-estimation penalty function} will be used, which (due to the ergodicity implied by
\myref{assumption_Yt}{assumption_Yt_strictly_stationary}) will
converge towards the penalty function~$\qh{h}{b}$,
\begin{align}
  \nonumber   \LhN[\thetah{h}]{h}{n}   &\defeq     \LhN[\Yht{h}{1},\dotsc,\Yht{h}{n}; \thetah{h}]{h}{n} \\
  &\defeq   \label{eq:LhN}
    n\inv \sumss{t=1}{n} \KhbDEF \log \psi\!\left( \Yht{h}{t}; \thetah{h}
    \right) 
    - \intss{\RRn{2}}{} \Khbdef \psi\!\left(\yh{h};
    \thetah{h}\right) \dyh{h}.
\end{align}

The local penalty function from \citet{Tjostheim201333} can be
described as
\begin{align}
  \nonumber   \QhN[\thetah{h}]{h}{n} 
  &\defeq     \QhN[\Yht{h}{1},\dotsc,\Yht{h}{n}; \thetah{h}]{h}{n}
    \defeq     -n \LhN[\thetah{h}]{h}{n} \\
  \label{eq:QhN}
  &= - \sumss{t=1}{n} \KhbDEF \log \psi\!\left( \Yht{h}{t}; \thetah{h}
  \right) 
  + n \intss{\RRn{2}}{}  \Khbdef \psi\!\left(\yh{h};
    \thetah{h}\right) \dyh{h},
\end{align}
and it remains to write out how the different components in
\cref{app:Klimko-Nelson-approach} looks like for this particular
penalty function.  A central component is the vector of partial
derivatives, which by the score function $\uh[\yh{h}; \thetah{h}]{h}$
can be given as,
\begin{align}
  \label{eq:QhN_derivatives}
  \nablah{h}\QhN[\thetah{h}]{h}{n} &=   - \sumss{t=1}{n} \left[\KhbDEF \uh[\Yht{h}{t}; \thetah{h}]{h}
    - \intss{\RRn{2}}{} \Khbdef \uh[\yh{h}; \thetah{h}]{h} \psi\!\left(\yh{h};
    \thetah{h}\right) \dyh{h} \right].
\end{align}
Note that the expectation of the bracketed expression in the sum gives
the left hand side of \cref{eq:def_of_theta_hb}, which implies that
the expectation will be~$\bm{0}$ when
$\nablah{h}\QhN[\thetah{h}]{h}{n}$ is evaluated at the population
value~$\thetahb{h}{b}$.

Given a bandwidth $\bm{b}$ which is small enough to ensure a unique
solution $\thetahb{h}{b}$, the next part of interest is the Taylor expansion of order two in a
neighbourhood
$\mathcalNz{h} \defeq \parenC{\thetah{h}: \absp{\thetah{h} -
    \thetahb{h}{b}} < \delta}$ of $\thetahb{h}{b}$,~i.e.\
\begin{subequations}
  \label{eq:Taylor_expansion_QhN}
  \begin{align}
  \nonumber   \QhN[\thetah{h}]{h}{n} 
  &= \QhN[\thetahb{h}{b}]{h}{n} +     \Vector[\prime]{\thetah{h} - \thetahb{h}{b}} \nablah{h}
    \QhN[\thetahb{h}{b}]{h}{n} 
    + \frac{1}{2}   \Vector[\prime]{\thetah{h} - \thetahb{h}{b}}     \VhN{}{h:\bm{b}}{n}     \left[\thetah{h} - \thetahb{h}{b}\right] \\
  \label{eq:Taylor_expansion_QhN_formula}
  &\phantom{=\ } + \frac{1}{2} \Vector[\prime]{\thetah{h} -
    \thetahb{h}{b}}
    \ThN{h:\bm{b}}{n}     \Vector{\thetah{h} - \thetahb{h}{b}},
        \intertext{where}
        \label{eq:matrix_WhN}
    \VhN{}{h:\bm{b}}{n} 
    &\defeq \VhN[\thetahb{h}{b}]{}{h:\bm{b}}{n} \defeq       \nablah{h}\nablah[\prime]{h} \QhN[\thetahb{h}{b}]{h}{n}, \\
    \label{eq:matrix_ThN}
    \ThN{h:\bm{b}}{n} 
    &\defeq \ThN[{\bmthetaz[*]{h}}, \thetahb{h}{b}]{h:\bm{b}}{n} \defeq       \nablah{h}\nablah[\prime]{h} \QhN[{\bmthetaz[*]{h}}]{h}{n} -
      \nablah{h}\nablah[\prime]{h} \QhN[\thetahb{h}{b}]{h}{n},
  \end{align}
\end{subequations}
with $\bmthetaz[*]{h}$ an intermediate point between $\thetah{h}$ and
$\thetahb{h}{b}$, again determined by the mean value theorem.

With the preceding definitions, \citet[theorem~1]{Tjostheim201333}
investigated the case where the bandwidth $\bm{b}$ was fixed as
$\nlimit$, i.e.\
\cref{th:Klimko_Nelson_1978_A1,th:Klimko_Nelson_1978_A2,th:Klimko_Nelson_1978_A3,th:Klimko_Nelson_1978_A4}
of \cref{th:Klimko_Nelson_1978} was verified in order to obtain the
following result for the estimated local Gaussian
parameters~$\hatthetahN{h}{n}$; for every \mbox{$\epsilon > 0$} there
exists an event $\Az{h}$ (possibly depending on the point~$\LGp$) with
\mbox{$\Prob{\Az[c]{h}} < \epsilon$}, such that there exists a
sequence of estimators $\hatthetahN{h}{n}$ that converges almost
surely to $\thetahb{h}{b}$ (the minimiser of $\qh{h}{b}$
from~\cref{eq:penalty_qh}).  And, moreover, the following asymptotic
behaviour is observed
\begin{align}
  \label{eq:CLT_bivariate}
  \parenRz[1/2]{}{n\bz{1}\bz{2}}   \left(\hatthetahN{h}{n} - \thetahb{h}{b}\right)   \stackrel{\scriptscriptstyle d}{\longrightarrow}   \UVN{\bm{0}}{\Sigmaz{h:\bm{b}}},
\end{align}
where
\mbox{$\Sigmaz{h:\bm{b}}\defeq
  \VhN{-1}{h:\bm{b}}{}\WhN{}{h:\bm{b}}{}\VhN{-1}{h:\bm{b}}{}$} with
$\WhN{}{h:\bm{b}}{}$ the matrix occurring in
\cref{th:Klimko_Nelson_1978_A4} of \cref{th:Klimko_Nelson_1978}.

The situation when $\blimit$ as $\nlimit$ requires some extra care
since the presence of the kernel function $\Khb[\bm{w}]{h}{\bm{b}}$ in
$\QhN[\thetah{h}]{h}{n}$, see \cref{eq:QhN}, gives limiting matrices
of $\VhN{}{h:\bm{b}}{}$ and $\WhN{}{h:\bm{b}}{}$ of rank one.
The details are covered in theorems 2 and~3 in
\citet[p.~39-40]{Tjostheim201333}, which ends out with the following
adjusted version of \cref{eq:CLT_bivariate}, where $n$ and
\mbox{$\bm{b}=\left(\bz{1},\bz{2}\right)$} are such that
\mbox{$\log n /n\!\parenRz[5]{}{\bz{1}\bz{2}} \rightarrow 0$},
\begin{align}
  \label{eq:CLT_bivariate_general_case}
  \parenRz[1/2]{}{n\parenRz[3]{}{\bz{1}\bz{2}}}   \left(\hatthetahN{h}{n} - \bmthetaz[\circ]{h}\right)   \stackrel{\scriptscriptstyle d}{\longrightarrow}   \UVN{\bm{0}}{\Sigmaz[\circ]{h}},
\end{align}
where $\bmthetaz[\circ]{h}$ is the $\blimit$ value of $\thetahb{h}{b}$
and where the limiting matrix $\Sigmaz[\circ]{h}$ is a
$\parenRz[2]{}{\bz{1}\bz{2}}$-rescaled version of
matrices \textit{related to} the matrices $\VhN{}{h:\bm{b}}{}$ and
$\WhN{}{h:\bm{b}}{}$, see the discussion in \citet{Tjostheim201333}
for details.

\makeatletter{}

\subsubsection{A new penalty function}
\label{app:new_penalty_function}

The proof of \cref{th:asymptotics_for_hatlgsd} requires an asymptotic
result for the parameter vector
$\widehatbmthetaz{n|\overbar{m}|\bm{b}}$, which was obtained by
combining $m$ parameter vectors corresponding to the bivariate lag~$h$
pairs \mbox{$\parenR{\Yz{t+h},\Yz{t}}$} for \mbox{$h=1,\dotsc,m$}.
This section will show
how a penalty function for $\widehatbmthetaz{n|\overbar{m}|\bm{b}}$
can be constructed based on the bivariate penalty
functions~$\QhN{h}{n}$ defined in
\cref{app:bivariate_penalty_functions}.  The indices $n$ and $\bm{b}$
will for notational simplicity be suppressed from the notation, and
only $\thetaM{m}$ will henceforth be~used.

An analysis akin to the one in Theorem~1 of \citet{Tjostheim201333}
will be performed in this section, i.e.\ the asymptotic situation will
be investigated for the simple case where the truncation~$m$ and the
bandwidth~$\bm{b}$ both are fixed as $\nlimit$.  
The proof that the new penalty function satisfies the four
requirements
\cref{th:Klimko_Nelson_1978_A1,th:Klimko_Nelson_1978_A2,th:Klimko_Nelson_1978_A3,th:Klimko_Nelson_1978_A4}
of \cref{th:Klimko_Nelson_1978} can then be based upon corresponding
components of the proof of Theorem~1 from \citet{Tjostheim201333}.

The general case, where $\mlimit$ and $\blimit$ when $\nlimit$, can
recycle the arguments given here for the requirements in
\cref{th:Klimko_Nelson_1978_A1,th:Klimko_Nelson_1978_A2,th:Klimko_Nelson_1978_A3},
but extra work is needed for the requirement given in
\cref{th:Klimko_Nelson_1978_A4}.  The details needed for
\cref{th:Klimko_Nelson_1978_A4} will be covered in
\cref{App:A4_requirement_general_case}.

With regard to the construction of the new penalty function, the main
observation of interest is that the~$\QhN[\thetah{h}]{h}{n}$ from
\cref{app:bivariate_penalty_functions} was defined for bivariate time
series~$\TSR{\Yht{h}{t}}{t\in\ZZ}{}$, whereas the new penalty function
will be defined for the \mbox{$(m+1)$-variate} time
series~$\TSR{\YMt{m}{t}}{t\in\ZZ}{}$.
The first step is to extend the penalty functions $\QhN{h}{n}$,
\mbox{$h=1,\dotsc,m$} from expression based on $\Yht{h}{t}$ to
expressions based on $\YMt{m}{t}$, but this is trivial since the
bivariate functions occurring in the definition of
$\QhN[\thetah{h}]{h}{n}$ can be extended in a natural manner to
\mbox{$(m+1)$-variate} functions, as mentioned in
\cref{note:details_for_Yt_definition}, which gives the desired
functions $\tildeQhN[\thetah{h}]{h}{n}$.

\begin{definition}
  Let the new penalty function $\QMN[\thetaM{m}]{m}{n}$ be given as
  follows,
                                \begin{subequations}
    \label{eq:QMN_definition}
    \begin{align}
      \QMN[\thetaM{m}]{m}{n} 
      &\defeq 
        \QMN[\YMt{m}{1},\dotsc,\YMt{m}{n}; \thetaM{m}]{m}{n} \defeq 
        \sumss{h=1}{m} \tildeQhN[\thetah{h}]{h}{n}, \\
      \intertext{where $\thetaM{m}$ is the column vector obtained by stacking all the
    individual $\thetah{h}$ on top of each other,~i.e.}
        \thetaM{m} 
      &\defeq \Vector[\prime]{\bmthetaz[\prime]{1},\dotsc,\bmthetaz[\prime]{m}}.
  \end{align}
  \end{subequations}
\end{definition}

The $m$ components $\tildeQhN[\thetah{h}]{h}{n}$ in the sum that
defines $\QMN[\thetaM{m}]{m}{n}$ have no common parameters, which
implies that the optimisation of the parameters for the different
summands can be performed~independently.
For a given sample from~$\TSR{\YMt{m}{t}}{t\in\ZZ}{}$ and for a given
bandwidth~$\bm{b}$, the optimal parameter vector~$\hatthetaMN{m}{n}$
for $\QMN[\thetaM{m}]{m}{n}$ can thus be constructed by stacking on
top of each other the parameter vectors that optimise the individual
summands in~\cref{eq:QMN_definition} --- and these are the parameter
vectors~$\hatthetahN{h}{n}$ that shows up for the~$m$ bivariate cases
in \cref{eq:CLT_bivariate}.  Since each~$\hatthetahN{h}{n}$ converge
almost surely to~$\thetahb{h}{b}$, it is clear
that~$\hatthetaMN{m}{n}$ will converge almost surely
to~$\thetaMb{m}{b}$, the vector obtained by stacking the~$m$
vectors~$\thetahb{h}{b}$ on top of each other.  

The desired asymptotic result for the fixed~$\bm{b}$ and fixed~$m$
estimates $\hatlgsdM{\LGp}{\omega}{m}$ can be obtained directly from
the preceding observation and Theorem~1 in \citet{Tjostheim201333},
but that would not reveal how $m$ and $\bm{b}$ must behave in the
general situation.
The rest of this section will thus be used to verify
\cref{th:Klimko_Nelson_1978_A1,th:Klimko_Nelson_1978_A2,th:Klimko_Nelson_1978_A3,th:Klimko_Nelson_1978_A4}
from \cref{th:Klimko_Nelson_1978}, which in essence only requires a
minor adjustment of the bivariate discussion from
\cref{app:bivariate_penalty_functions}, i.e.\ the discussion can start
with the following Taylor-expansion of $\QMN[\thetaM{m}]{m}{n}$,
\begin{align}
  \nonumber   \QMN[\thetaM{m}]{m}{n} 
  &= \QMN[\thetaMb{m}{b}]{m}{n} +     \Vector[\prime]{\thetaM{m}-\thetaMb{m}{b}} \nablaM{m}
    \QMN[\thetaMb{m}{b}]{m}{n} 
    + \frac{1}{2}   \Vector[\prime]{\thetaM{m}-\thetaMb{m}{b}}     \VMbN{m}{\bm{b}}{n}     \Vector{\thetaM{m} - \thetaMb{m}{b}} \\
  \label{eq:Taylor_expansion_QMN}
  &\phantom{=\ } + \frac{1}{2}   \Vector[\prime]{\thetaM{m} - 
    \thetaMb{m}{b}}   \TMbN{m}{\bm{b}}{n}   \Vector{\thetaM{m} - \thetaMb{m}{b}},
\end{align}
where $\thetaMb{m}{b}$ represents the vector obtained by stacking on
top of each other the $m$ individual population parameters
$\thetahb{h}{b}$, where
\mbox{$\nablaM{m} \defeq
  \Vector[\prime]{\nablah[\prime]{1},\dotsc,\nablah[\prime]{m}}$}, and
where the matrices $\VMbN{m}{\bm{b}}{n}$ and $\TMbN{m}{\bm{b}}{n}$
corresponds to the matrices $\VhN{}{h:\bm{b}}{n}$ and
$\ThN{h:\bm{b}}{n}$ from \cref{eq:Taylor_expansion_QhN}.

  The following matrix-observations gives the foundation for the
  extension from the bivariate case to the multivariate case.
  \begin{enumerate}
  \item     Keeping in mind how $\nablaM{m}$ is defined relative to $\nablah{h}$,
    and how $\QMN{m}{n}$ is defined relative to~$\QhN{h}{n}$, it is clear
    that $\nablaM{m}\QMN[\thetaMb{m}{b}]{m}{n}$ is the vector obtained
    by stacking the $m$ vectors $\nablah{h}\QhN[\thetahb{h}{b}]{h}{n}$
    on top of each~other.
  \item     The operator $\nablaM{m}\nablaM[\prime]{m}$ can be viewed as an
    \mbox{$m\times m$} block-matrix, consisting of the \mbox{$5\times5$}
    matrices $\nablah{j}\nablah[\prime]{k}$, \mbox{$j,k = 1,\dotsc,m$}.  Due to
    the definition of $\QMN{m}{n}$, it is clear that the only operators
    $\nablah{j}\nablah[\prime]{k}$ that will return a nonzero result are those
    having~\mbox{$j=k$}.
  \item     The preceding observation implies that
    \mbox{$\VMbN{m}{\bm{b}}{n} = \oplusss{h=1}{m}
      \VhN{}{h:\bm{b}}{n}$}, i.e.\ $\VMbN{m}{\bm{b}}{n}$ is the direct
    sum of the matrices~$\VhN{}{h:\bm{b}}{n}$ (the block diagonal matrix
    where the diagonal blocks equals~$\VhN{}{h:\bm{b}}{n}$, and all
    other blocks are zero, cf.\ e.g.\
    \citet[p.30]{Horn:2012:MA:2422911} for further details).
  \item     The same observation implies that
    \mbox{$\TMbN{m}{\bm{b}}{n} = \oplusss{h=1}{m}
      \ThN{h:\bm{b}}{n}$}
  \end{enumerate}

With these observations, and the details from the proof of Theorem~1
in \citet{Tjostheim201333}, it is straightforward to verify
\cref{th:Klimko_Nelson_1978_A1,th:Klimko_Nelson_1978_A2,th:Klimko_Nelson_1978_A3}
of \cref{th:Klimko_Nelson_1978}, whereas
\cref{th:Klimko_Nelson_1978_A4} requires some more work.

\begin{lemma}[\Cref{th:Klimko_Nelson_1978_A1} of
  \cref{th:Klimko_Nelson_1978}.]  {\ }\\ 
  \label{Res:QMN_A1_requirement}
  $\nz[-1]{}\nablaM{m}\QMN[\thetaMb{m}{b}]{m}{n}
  \stackrel{\scriptscriptstyle a.s.}{\longrightarrow} \bm{0}$
\end{lemma}
\begin{proof}
  Since $\nablaM{m}\QMN[\thetaMb{m}{b}]{m}{n}$ is the vector obtained
  by stacking the $m$ vectors $\nablah{h}\QhN[\thetahb{h}{b}]{h}{n}$
  on top of each other, and the proof of Theorem~1 in
  \citet{Tjostheim201333} shows that
  $\nz[-1]{}\nablah{h}\QhN[\thetahb{h}{b}]{h}{n}$ converges almost surely
  to~$\bm{0}$, the same must necessarily be true for the combined
  vector $\nz[-1]{}\nablaM{m}\QMN[\thetaMb{m}{b}]{m}{n}$~too.
\end{proof}

\begin{lemma}[\Cref{th:Klimko_Nelson_1978_A2} of
  \cref{th:Klimko_Nelson_1978}.]   {\ }\\
  \label{Res:QMN_A2_requirement}
  $\nz[-1]{}\VMbN{m}{\bm{b}}{n} \stackrel{\scriptscriptstyle
    a.s.}{\longrightarrow} \VMbN{m}{\bm{b}}{}$, where
  $\VMbN{m}{\bm{b}}{}$ is a \mbox{$5m\times 5m$} positive definite
  matrix.
\end{lemma}
\begin{proof}
  Since $\VMbN{m}{\bm{b}}{n}$ is the direct sum of the $m$ matrices
  $\VhN{}{h:\bm{b}}{n}$, the behaviour of those will describe the
  behaviour of $\VMbN{m}{\bm{b}}{n}$.  The proof of Theorem~1 in
  \citet{Tjostheim201333} shows that the matrices
  $\nz[-1]{}\VhN{}{h:\bm{b}}{n}$ converges almost surely to
  positive definite matrices~$\VhN{}{h:\bm{b}}{}$, and this implies
  that $\nz[-1]{}\VMbN{m}{\bm{b}}{n}$ will converge almost
  surely to a block diagonal matrix $\VMbN{m}{\bm{b}}{}$, defined as
  the direct sum of the matrices~$\VhN{}{h:\bm{b}}{}$.  Since the set
  of eigenvalues for a direct sum of matrices equals the union of the
  eigenvalues for its components, see
  \citet[p.30]{Horn:2012:MA:2422911} for details, if follows that
  $\VMbN{m}{\bm{b}}{n}$ is positive definite since all the
  $\VhN{}{h:\bm{b}}{n}$ are positive definite.
\end{proof}

\begin{lemma}[\Cref{th:Klimko_Nelson_1978_A3} of
  \cref{th:Klimko_Nelson_1978}.] {\ }\\
  \label{Res:QMN_A3_requirement}
  For \mbox{$j,k = 1,\dotsc, 5m$},
  \begin{align}
    \label{eq:TMN_requirement}
    \adjustlimits \lim_{n\rightarrow\infty} \sup_{\delta\rightarrow 0}     (n\delta)\inv \absp{\power{jk}{\TMbN{m}{\bm{b}}{n}}} < \infty\     \quad \text{a.s.},
  \end{align}
  where $\power{jk}{\TMbN{m}{\bm{b}}{n}}$ is the
  $\ith{(j,k)}$ component of $\TMbN{m}{\bm{b}}{n}$.
\end{lemma}
\begin{proof}
  $\TMbN{m}{\bm{b}}{n}$ is the direct sum of the $m$ matrices
  $\ThN{h:\bm{b}}{n}$, so the required inequality is trivially
  satisfied for all entries $j$ and $k$ that gives an element outside
  of the diagonal-blocks.  The proof of Theorem~1 in
  \citet{Tjostheim201333} shows that the inequality is satisfied
  almost surely on each of the $m$ blocks $\ThN{h:\bm{b}}{n}$, which
  implies that it holds for $\TMbN{m}{\bm{b}}{n}$~too.
\end{proof}

\begin{lemma}[\Cref{th:Klimko_Nelson_1978_A4} of
  \cref{th:Klimko_Nelson_1978}.] {\ }\\
  \label{Res:QMN_A4_requirement}
      $\nz[-1/2]{} \nablaM{m} \QMN[\thetaMb{m}{b}]{m}{n}
  \stackrel{\scriptscriptstyle d}{\longrightarrow} \Norm{\bm{0}}{\WMbN{m}{\bm{b}}{}}$
\end{lemma}
\begin{proof}
  As done in the proof of Theorem~1 in \citet{Tjostheim201333}, the
  idea is to first prove asymptotic normality of each individual
  component of $\nablaM{m} \QMN[\thetaMb{m}{b}]{m}{n}$ by 
  Theorem~2.20(i) and Theorem~2.21(i) from
  \citet[p.~74-75]{fan03:_nonlin_time_series}.  Then the
  Cram\'{e}r-Wold Theorem (see e.g.\ Theorem~29.4 in
  \citet{Billingsley12:_probab_measur}) will be used to conclude that
  the joint distribution of $\nablaM{m} \QMN[\thetaMb{m}{b}]{m}{n}$
  will be the joint distribution of these limiting components, and
  finally a simple observation based on moment-generating functions
  tells us that this limiting joint distribution is
  asymptotically~normal.

  Since
  \mbox{$\nablaM{m} \QMN[\thetaMb{m}{b}]{m}{n} = \Vector[\prime]{
      \power[\prime]{}{\nablah{1} \QhN[\thetahb{h}{b}]{1}{n}},       \dotsc,       \power[\prime]{}{\nablah{m} \QhN[\thetahb{h}{b}]{m}{n}}}$},
  its components can be indexed by pairs $[h,i]$,
  \mbox{$h = 1,\dotsc,m$} and \mbox{$i= 1, \dotsc, 5$}.  From
  \cref{eq:QhN_derivatives} it is clear that the
  \mbox{$[h,i]$-component} of the vector can be written~as
  \begin{align}
    \label{eq:QMN_hi_component}
    \parenRz{[h,i]}{\nablaM{m} \QMN[\thetaMb{m}{b}]{m}{n}} 
    &= - \sumss{t=1}{n} \Xhit{h}{i}{t},
  \end{align}
  where the random variable $\Xhit{h}{i}{t}$ is defined as
    \begin{align}
    \label{eq:Xhit}
    \Xhit{h}{i}{t} &\defeq     \KhbDEF \uhi[\Yht{h}{t}; \thetahb{h}{b}]{h}{i} - \intss{\RRn{2}}{}
    \Khbdef \uhi[\yh{h}; \thetahb{h}{b}]{h}{i} \psi\!\left(\yh{h};
      \thetah{h}\right) \dyh{h},
  \end{align}
  and where $\uhi{h}{i}$ refers to the $i^{\operatorname{th}}$
  component of the $\ith{h}$ score function $\uh{h}$.
  
  The required \mbox{$\alpha$-mixing} property (and thus ergodicity)
  are inherited from the original univariate time series $\Yz{t}$ to
  $\Xhit{h}{i}{t}$ (see \cref{eq:alpha_requirement_derived_from_Yt}
  for details), 
    and the connection with $\Lp{\nu}$-theory observed in
  \cref{eq:expectation_and_Lph} gives
  \mbox{$\E{\absp[\nu]{\Xhit{h}{i}{t}}} < \infty$}.
    Finally, since $\thetahb{h}{b}$ is the population value parameter
  that minimise \cref{eq:def_of_theta_hb}, it follows that
  \mbox{$\E{\Xhit{h}{i}{t}} = 0$}.
    These observations show that $\Xhit{h}{i}{t}$ satisfies the
  requirements needed in order to apply Theorem~2.20(i) and
  Theorem~2.21(i) from \citet[p.~74-75]{fan03:_nonlin_time_series},
  i.e.\ for \mbox{$\Sz{hi|n} \defeq \sumss{t=1}{n} \Xhit{h}{i}{t}$},
  Theorem~2.20(i) gives the asymptotic result
    \begin{align}
    \label{eq:Th2.20(i)}
    \nz[-1]{} \Sz{hi|n} \longrightarrow \sigmaz[2]{} \defeq \gammaz{0} +
    2 \sumss{\ell\,\geq1}{} \gammaz{\ell},
  \end{align}
  with $\gammaz{\ell}$ being the $\ith[\,]{\ell}$ autocovariance of
  the series $\TSR{\Xhit{h}{i}{t}}{t\in\ZZ}{}$.
    From Theorem~2.21(i) it now follows that there is a component-wise
  asymptotic normality, i.e.\
  \begin{align}
    \label{eq:Th2.21(i)}
    \nz[-1/2]{} \Sz{hi|n} \stackrel{\scriptscriptstyle
    d}{\longrightarrow}
    \UVN{0}{\sigmaz[2]{}}.
  \end{align}
  
  In order to apply the Cram\'{e}r-Wold device, all possible linear
  combinations of the components in
  $\nablaM{m} \QMN[\thetaMb{m}{b}]{m}{n}$ must be considered.  Such
  general sums can be represented as
  \mbox{$\Sz{n}(\bm{a}) \defeq \bmaz[\prime]{}\,\nablaM{m}
    \QMN[\thetaMb{m}{b}]{m}{n}$}, where
  \mbox{$\bm{a}\in\RRn{5\times m}$}.
    This can be rewritten, by \enquote{taking the sum outside of the
    vector $\nablaM{m} \QMN[\thetaMb{m}{b}]{m}{n}$}, as
  \begin{align}
    \Sz{n}(\bm{a}) &= \sumss{t=1}{n} \Xat[a]{t},
  \end{align}
  where \mbox{$\Xat[a]{t} = \bmaz[\prime]{}\Xt{t}$}, with the vector
  $\Xt{t}$ obtained by stacking all the components $\Xhit{h}{i}{t}$ on
  top of each other, i.e.\
  \mbox{$\Xt{t} = \Vector[\prime]{\Xhit{1}{1}{t}, \dotsc,
      \Xhit{m}{5}{t}}$}.
  
  By construction, \mbox{$\E{\Xat[a]{t}} = 0$}, the required
  $\alpha$-mixing are inherited from the original time series
  $\parenC{\Yz{t}}$ (see \cref{eq:alpha_requirement_derived_from_Yt}),
  and \cref{th:Lp_expectation} ensures that the property
  \mbox{$\E{\absp[\nu]{\Xat[a]{t}}} < \infty$} holds true.  That is,
  $\Xat[a]{t}$ does also satisfy the requirements stated in
  Theorem~2.20(i) and Theorem~2.21(i), which gives the following
  asymptotic results;
  \begin{align}
    \nz[-1]{} \Sz{n}(\bm{a}) 
    &\longrightarrow \sigmaz[2]{}(\bm{a}) \defeq
      \gammaz{0}(\bm{a}) +
      2 \sumss{\ell\,\geq1}{} \gammaz{\ell}(\bm{a}) \\
          \nz[-1/2]{} \Sz{n}(\bm{a}) 
    &\stackrel{\scriptscriptstyle d}{\longrightarrow} 
            \UVN{0}{\sigmaz[2]{}(\bm{a})},
  \end{align}
  where the autocovariances $\gammaz{\ell}(\bm{a})$ now are with respect to
  the time series \mbox{$\Xat[a]{t} = \bmaz[\prime]{}\Xt{t}$}.
  
  Since
  {$\gammaz{0}(\bm{a}) = \Var{\bmaz[\prime]{}\Xt{t}}
    = \bmaz[\prime]{}\Var{\Xt{t}}\bm{a}$} and
  {$\gammaz{\ell}(\bm{a}) =
    \Cov{\bmaz[\prime]{}\Xt{t+\ell}}{\bmaz[\prime]{}\Xt{t}}
    = \bmaz[\prime]{}\Cov{\Xt{t+\ell}}{\Xt{t}}\bm{a}$},   it follows that we can write
  {$\sigmaz[2]{}(\bm{a}) =
    \bmaz[\prime]{}\WMbN{m}{\bm{b}}{}\bm{a}$}, with
  $\WMbN{m}{\bm{b}}{}$ being the matrix obtained in the obvious manner
  by factorising out $\bmaz[\prime]{}$ and $\bm{a}$ from
  the sum of autocovariances, i.e.\
  \begin{align}
    \label{eq:SMb_definition}
    \WMbN{m}{\bm{b}}{} &\defeq \Var{\Xt{t}} + 2 \sumss{\ell\,\geq1}{}
                 \Cov{\Xt{t+\ell}}{\Xt{t}} \\
    \label{eq:SMb_definition_simplified}
               &= \E{\Xt{t}\bmXz[\prime]{t}} +2 \sumss{\ell\,\geq1}{}                  \E{\Xt{t+\ell}\bmXz[\prime]{t}},
  \end{align}
  where the second equality follows since \mbox{$\E{\Xt{t}}=\bm{0}$}.

  The Cram\'{e}r-Wold device now gives the required conclusion,
  \mbox{$\nz[-1/2]{} \nablaM{m} \QMN[\thetaMb{m}{b}]{m}{n}
    \stackrel{\scriptscriptstyle d}{\longrightarrow}
    \Norm{\bm{0}}{\WMbN{m}{\bm{b}}{}}$}.
\end{proof}

\Cref{Res:QMN_A1_requirement,Res:QMN_A2_requirement,Res:QMN_A3_requirement,Res:QMN_A4_requirement}
shows that the penalty function $\QMN[\thetaM{m}]{m}{n}$ (for fixed
$m$ and fixed $\bm{b}$) satisfies the four requirements given in
\cref{th:Klimko_Nelson_1978_A1,th:Klimko_Nelson_1978_A2,th:Klimko_Nelson_1978_A3,th:Klimko_Nelson_1978_A4}
of \cref{th:Klimko_Nelson_1978}, and this implies that the following
asymptotic results holds in this particular case
\begin{align}
  \sqrt{n}\left(\hatthetaMN{m}{n} - \thetaMb{m}{b}\right)
  \stackrel{\scriptscriptstyle d}{\longrightarrow}   \UVN{\bm{0}}{  \Vz[-1]{\overbar{m}|\bm{b}} \WMbN{m}{\bm{b}}{} \Vz[-1]{\overbar{m}|\bm{b}}}.
\end{align}

The hard task to deal with in the general situation, when $\mlimit$
and $\blimit$ as $\nlimit$, is the asymptotic behaviour of
\mbox{$\nz[-1/2]{} \nablaM{m} \QMN[\thetaMb{m}{b}]{m}{n}$}.
This will be treated in \cref{App:A4_requirement_general_case}.

\subsection{The A4-requirement in the general case}
\label{App:A4_requirement_general_case}

The verification of the three first requirements of the Klimko-Nelson
approach does work as before when \enquote{$\mlimit$ and $\blimit$
  when $\nlimit$}, whereas the asymptotic normality in the fourth
requirement demands a more detailed investigation.
\Cref{app:final_building_blocks} will introduce some new building
blocks to be used in the investigation of the asymptotic properties,
which will be developed in
\cref{app:the_asymptotic_results,app:the_asymptotic_results_advanced}.
Some technical details that only depend upon the kernel function and
the score functions have been collected in
\cref{app:integrals_kernel_score}.

\makeatletter{}
\subsubsection{The final building blocks}
\label{app:final_building_blocks}

The bivariate processes~$\Yht{h}{t}$ from \cref{def:Yht_and_YMt} will
now be used to construct new random variables, that culminates in a
random variable $\QNhvec{n}{\overbar{m}}$ which has the same limiting
distribution\footnote{Due to the presence of the kernel function
  $\Khb[\bm{w}]{h}{\bm{b}}$, the fourth requirement of the
  Klimko-Nelson approach will (when $\blimit$) require that the
  scaling factor $\nz[-1/2]{}$ is adjusted with
  $\parenRz[1/2]{}{\bz{1}\bz{2}}$, and this scaling must thus also be
  included in the discussion in the present approach.} $\sqrt{\bz{1}\bz{2}}\nablaM{m}\QMN[\thetaMb{m}{b}]{m}{n}$.
Looking upon \cref{eq:QhN_derivatives}, it is clear that everything
depends upon the three functions
\mbox{$\psi\!\left(\yh{h}; \thetah{h}\right)$},
\mbox{$\uh[\yh{h}; \thetah{h}]{h}$} and $\Khbdef$.

\begin{definition}
  \label{def:Uh}
  For \mbox{$\psi\!\left(\yh{h}; \thetah{h}\right)$} the local
  Gaussian density used when approximating~$\gh[\yh{h}]{h}$ at the
  point \mbox{$\LGp=\LGpoint$},
      define for
  all \mbox{$h \in \NN$} and   \mbox{$q \in \parenC{1,\dotsc,5}$} 
  \begin{enumerate}[label=(\alph*)]
  \item     \label{item_def:Uh}
        With $\thetahb{h}{b}$ the population value that minimises the
    penalty function~$\qh{h}{b}$ from \cref{eq:penalty_qh}, let
    \begin{align}
      \label{eq_def:Uh}
      \Uh[\bm{w}]{hq:\bm{b}}       &\defeq         \left.         \frac{\partial}{\partial \thetahcomp{h}{q}}         \log \left(\psi\!\left(\yh{h}; \thetah{h}\right)\right)         \right|_{\left(\yh{h};\, \thetah{h}\right) =         \left(\bm{w};\, \thetahb{h}{b}\right) }.
    \end{align}
              \item     \label{item_def:Uh_truncated}
        For \mbox{$L\geq0$}, define the following lower and upper
    truncated versions of $\Uh[\bm{w}]{hq:\bm{b}}$,
    \begin{subequations}
      \label{eq:definition_of_UhL}
      \begin{align}
        \label{eq:definition_of_UhL_leq}
        \UhL[\bm{w}]{hq:\bm{b}}{\leq L} 
        &\defeq           \Uh[\bm{w}]{hq:\bm{b}}\cdot \Ind{\absp{\Uh[\bm{w}]{hq:\bm{b}}} \leq L},  \\
        \label{eq:definition_of_UhL_>}
        \UhL[\bm{w}]{hq:\bm{b}}{> L} 
        &\defeq           \Uh[\bm{w}]{hq:\bm{b}}\cdot \Ind{\absp{\Uh[\bm{w}]{hq:\bm{b}}} > L}. 
      \end{align}
    \end{subequations}
    Obviously;
    \mbox{$\Uh[\bm{w}]{hq:\bm{b}} = \UhL[\bm{w}]{hq:\bm{b}}{\leq L}
      + \UhL[\bm{w}]{hq:\bm{b}}{> L}$} and 
    \mbox{$\UhL[\bm{w}]{hq:\bm{b}}{\leq L} \cdot
      \UhL[\bm{w}]{hq:\bm{b}}{> L} = 0$}.
      \item     \label{item_def:Uh_limit}
        Let $\Uh[\bm{w}]{hq}$ be as in \cref{item_def:Uh}, with the
    difference that the limit $\blimit$ of the parameters
    $\thetahb{h}{b}$ are used in the definition.\footnote{      The limit of the parameters $\thetahb{h}{b}$ will exist under
      assumptions that implies that the four requirements of the
      Klimko-Nelson approach are satisfied, cf.\
      \citet{Tjostheim201333} for details.}     Let $\UhL[\bm{w}]{hq}{\leq L}$ and $\UhL[\bm{w}]{hq}{> L}$ be
    the     truncated versions of~$\Uh[\bm{w}]{hq}$.
  \end{enumerate}
\end{definition}

The following simple observations will be useful later on.

\begin{lemma}
  \label{th:Uh_finite}
  For the point $\LGp$, the following holds for the functions
  introduced in \cref{def:Uh}.
  \begin{enumerate}[label=(\alph*)]
  \item    \label{th:Uh_finite_sup}
      \mbox{$\sup_{\scriptscriptstyle hq}      \parenAbs{\Uh[\LGp]{hq:\bm{b}}} < \infty$} and 
   \mbox{$\sup_{\scriptscriptstyle hq}      \parenAbs{\Uh[\LGp]{hq}} < \infty$}.
    \item    \label{th:Uh_threshold_limit}
      When $L$ is large enough,
   \mbox{$\UhL[\LGp]{hq:\bm{b}}{\leq L} = \Uh[\LGp]{hq:\bm{b}}$}
   and \mbox{$\UhL[\LGp]{hq}{\leq L} = \Uh[\LGp]{hq}$}.
 \end{enumerate}
\end{lemma}

\begin{proof}
  By definition, the functions~$\Uh[\bm{w}]{hq:\bm{b}}$
  and~$\Uh[\bm{w}]{hq}$ will all be bivariate polynomials of order two
  (in the variables $\wz{1}$ and~$\wz{2}$), which implies that they
  are well defined for any point~$\LGp$.  Since the parameters in
  these polynomials originates from a local Gaussian approximation of
  $\gh[\yh{h}]{h}$ at the point~$\LGp$, and since
  \myref{assumption_Yt}{assumption_Yt_strong_mixing} ensures that the
  bivariate densities $\gh[\yh{h}]{h}$ will approach the product of
  the marginal densities when $\hlimit$, it follows that the estimated
  parameters must stabilise when $h$ becomes large.  This rules out
  the possibility that any of the parameters can grow to infinitely
  large values, which implies that the supremums in
  \cref{th:Uh_finite_sup} are finite.  \Cref{th:Uh_threshold_limit}
  follows as a direct consequence of this, the statement holds true
  for any threshold value $L$ that is larger than the supremums given
  in \cref{th:Uh_finite_sup}.
\end{proof}

The bivariate kernel to be used in the present approach will be the
same as the one used in \citet{Tjostheim201333}, i.e.\ it will be the
product kernel based on two standard normal kernels.  The following
definition enables a more general approach to be used in the
theoretical investigation,\footnote{  Differences in the computational cost implies that the product
  normal kernel is used for practical~purposes.} while capturing the desirable properties
that will be satisfied for the product normal~kernel.

\begin{definition}
  \label{def:kernel}
  From a bivariate, non-negative, and bounded kernel function
  $K(\bm{w})$, that satisfies
  \begin{subequations}
    \label{eq:kernel_integrals}
    \begin{align}
      \label{eq:kernel_integral_one}
      &\intss{\RRn{2}}{} \!\! K\!\left(\wz{1},\wz{2}\right)
        \d{\wz{1}}\!\d{\wz{2}} = 1,\\
              \label{eq:kernel_integrals_left_wellbehaved}
      &\mathcalKz{1:k} \! \left(\wz{2}\right)         \defeq         \intss{\RRn{1}}{} \!\!         K\!\left(\wz{1},\wz{2}\right)         \wz[k]{1} \d{\wz{1}}  \qquad         \text{is bounded for }         k \in \TSR{0,1,2}{}{}, \\
              \label{eq:kernel_integrals_right_wellbehaved}
      &\mathcalKz{2:\ell} \! \left(\wz{1}\right)         \defeq         \intss{\RRn{1}}{} \!\!         K\!\left(\wz{1},\wz{2}\right)         \wz[\ell]{2} \d{\wz{2}}  \qquad         \text{is bounded for }         \ell \in \TSR{0,1,2}{}{}, \\
              \label{eq:kernel_integrals_finite}
      &\intss{\RRn{2}}{} \!\! K\!\left(\wz{1},\wz{2}\right)         \absp{\wz[k]{1} \wz[\ell]{2}}         \d{\wz{1}}\!\d{\wz{2}} < \infty, \qquad k, \ell
        \geq 0 \text{ and } k + \ell \leq  2\cdot\ceil{\nu},
                    \end{align}
  \end{subequations}
  where \mbox{$\nu>2$} is from
  \myref{assumption_Yt}{assumption_Yt_strong_mixing} (and
  $\ceil{\cdot}$ is the ceiling function), define
  \begin{align}
    \label{eq:definition_of_K}
    \Khb[\yh{h}-\LGp]{h}{\bm{b}}
    &\defeq       \frac{1}{\bz{1}\bz{2}} K\left(      \frac{\yz{h}-\vz{1}}{\bz{1}},
      \frac{\yz{0}-\vz{2}}{\bz{2}}
      \right).
  \end{align}
\end{definition}

It turns out, see \cref{app:integrals_kernel_score} for details, that
the asymptotic results needed later on mainly depends upon the
properties of the kernel $K(\bm{w})$ and the components
$\Uh[\bm{w}]{hq:\bm{b}}$ of the score functions.

Some vector and matrix notation is needed in order to make the
expressions later on more~tractable.
\begin{definition}
  \label{def:matrices_related_to_the_new_penalty_function}
  With $\gh[\yh{h}]{h}$, $\Uh[\bm{w}]{hq:\bm{b}}$ and $K(\bm{w})$ as
  given in \cref{def:Yht_and_YMt,def:kernel,def:Uh}, let
  \mbox{$\mathfrakUhb{h}{b} \defeq
    \Vector[\prime]{\Uh[\bm{v}]{h1:\bm{b}},\dotsc,
      \Uh[\bm{v}]{hp:\bm{b}}}$}, and define the following matrices. 
  \begin{subequations}
    \label{eq_def:matrices_related_to_the_new_penalty_function}
    \begin{align}
      \label{eq_def:matrices_related_to_the_new_penalty_function_Whb}
      \WhN{}{h}{\bm{b}}
      &\defeq \mathfrakUhb{h}{b} \mathfrakUhb[\prime]{h}{b} \cdot
        \gh[\LGp]{h} 
        \intss{\RRn{2}}{} \power[2]{}{K(\bm{w})}\d{\bm{w}}, \\
      \label{eq_def:matrices_related_to_the_new_penalty_function_WMb}
      \WMbN{m}{b}{} 
      &\defeq \oplusss{h=1}{m} \WhN{}{h}{\bm{b}}.
    \end{align}
  \end{subequations}
  Matrices $\WhN{}{h}{}$ and $\WMbN{m}{}{}$ can be defined in a
  similar manner, using the $\blimit$ versions $\Uh[\bm{w}]{hq}$ from
  \myref{def:Uh}{item_def:Uh_limit}.
    Note that $\WhN{}{h}{\bm{b}}$ and $\WhN{}{h}{}$ will have rank one,
  whereas $\WMbN{m}{}{\bm{b}}$ and $\WMbN{m}{}{}$ will have rank~$m$.
    Furthermore, note that if \mbox{$\bmaz{h} \in \RRn{5}$} and
  \mbox{$\baM{m} =
    \Vector[\prime]{\bmaz{1},\dotsc,\bmaz{m}}$}, then
  \mbox{$\baM[\prime]{m} \WMbN{m}{}{\bm{b}} \baM{m} = \sumss{h=1}{m}
    \bmaz[\prime]{h} \WhN{}{h}{\bm{b}} \bmaz{h}$}.
\end{definition}

The time is due for the introduction of the random variables.
\begin{definition}
  \label{def:Xht_*Lc}
          Based on $\Yht{h}{t}$, $\Uh[\bm{w}]{hq:\bm{b}}$ and
  $\Khb[\yh{h}-\LGp]{h}{\bm{b}}$ from
  \cref{def:Yht_and_YMt,def:Uh,def:kernel}, define new bivariate
  random variables as follows,
  \begin{subequations}
    \label{eq:definition_of_Xht_*Lc}
    \begin{align}
      \label{eq:definition_of_XNht*}
      \XNht{n}{hq}{t}\!\left(\LGp\right)       &\defeq
        \sqrt{\bz{1}\bz{2}}         \Khb[\Yht{h}{t}-\LGp]{h}{\bm{b}} \Uh[\Yht{h}{t}]{hq:\bm{b}}, \\
      \label{eq:definition_of_XhtL}
      \XNht[\leq L]{n}{hq}{t}\!\left(\LGp\right)       &\defeq
        \sqrt{\bz{1}\bz{2}}         \Khb[\Yht{h}{t}-\LGp]{h}{\bm{b}} \UhL[\Yht{h}{t}]{hq:\bm{b}}{\leq L}, \\
      \label{eq:definition_of_Xhtc}
      \XNht[>L]{n}{hq}{t}\!\left(\LGp\right)       &\defeq
        \sqrt{\bz{1}\bz{2}}         \Khb[\Yht{h}{t}-\LGp]{h}{\bm{b}} \UhL[\Yht{h}{t}]{hq:\bm{b}}{> L}.
    \end{align}
  \end{subequations}
  Obviously;   \mbox{    $\XNht{n}{hq}{t}\!\left(\LGp\right) =     \XNht[\leq L]{n}{hq}{t}\!\left(\LGp\right) +     \XNht[>L]{n}{hq}{t}\!\left(\LGp\right)$}   and \mbox{    $\XNht[\leq L]{n}{hq}{t}\!\left(\LGp\right) \cdot     \XNht[>L]{n}{hq}{t}\!\left(\LGp\right) = 0$}.
\end{definition}

Since the point $\LGp$ will be fixed for the remainder of this
discussion, $\LGp$ will be suppressed and only $\XNht{n}{hq}{t}$ will
be used when referring to \cref{eq:definition_of_XNht*}, and~$\LGp$
will also be suppressed for the new random variables derived
from~$\XNht{n}{hq}{t}$.

Note: A comparison of $\XNht{n}{hq}{t}$ against the components
occurring in the expression for $\nablah{h}\QhN[\thetah{h}]{h}{n}$,
see \cref{eq:QhN_derivatives}, implies that the following adjusted
variable should be included,
\begin{align}
  \label{eq:Xht_adjusted}
  \tildeXNht{n}{hq}{t} 
  &\defeq     \XNht{n}{hq}{t} - \sqrt{\bz{1}\bz{2}}\intss{\RRn{2}}{}     \Khb[\yh{h}-\LGp]{h}{\bm{b}} \Uh[\yh{h}]{hq:\bm{b}}     \psi\!\left(\yh{h}; \thetah{h}\right) \d{\bm{\yh{h}}},
\end{align}
but the arguments later on will use a mean adjusted approach similar
to the one used in \citet{masry95:_nonpar_estim_ident_nonlin_arch},
see the definitions of $\ZNht{n}{hq}{t}$ and $\QNh{n}{hq}$ below, and
the only place $\tildeXNht{n}{hq}{t}$ is needed is in the proof of
\cref{eq:connection_RV_nabla_QMN}.

\begin{definition}
  \label{def:ZNht_and_QNh}
  Based on the bivariate random variables $\XNht{n}{hq}{t}$ from
  \cref{def:Xht_*Lc} define the following bivariate and
  \mbox{$(m+1)$-variate} random variables,
  \begin{subequations}
    \label{eq:definition_ZNht_and_QNh}
    \begin{align}
      \label{eq:definition_ZNht}
      \ZNht{n}{hq}{t}         &\defeq \XNht{n}{hq}{t} -           \E{\XNht{n}{hq}{t}}, \\
                \label{eq:definition_QNh}
      \QNh{n}{hq}         &\defeq 
                    \sumss{t=1}{n} \ZNht{n}{hq}{t}.
    \end{align}
  \end{subequations}
  Similarly, $\ZNht[\geq L]{n}{hq}{t}$, $\ZNht[< L]{n}{hq}{t}$,
  $\QNh[\geq L]{n}{hq}$ and $\QNh[< L]{n}{hq}$ can be defined in the
  natural manner, with the obvious connections   \mbox{$\ZNht{n}{hq}{t} =     \ZNht[\geq L]{n}{hq}{t} + \ZNht[< L]{n}{hq}{t}$},   \mbox{$\ZNht[\geq L]{n}{hq}{t} \cdot \ZNht[< L]{n}{hq}{t} = 0$},   and   \mbox{$\QNh{n}{hq} = \QNh[\geq L]{n}{hq} + \QNh[< L]{n}{hq}$}
  holding for all~$L$.   Moreover: \mbox{$\Cov{       \ZNht{n}{hq}{i}}{       \ZNht{n}{j}{k} } = \E{       \ZNht{n}{hq}{i} \cdot       \ZNht{n}{j}{k} } = \Cov{       \XNht{n}{hq}{i}}{       \XNht{n}{jr}{k} }$}.
\end{definition}

The last batch of random variables can now be introduced.
\begin{definition}
  \label{def:ZNMt_and_QNM}
  Based upon the bivariate $\ZNht{n}{hq}{t}$ from
  \cref{def:ZNht_and_QNh}, and for
  \mbox{$\bm{a} \defeq \baM{m} \in \RRn{5\times m}$},
      define the following \mbox{$(m+1)$-variate} random~variables,
  \begin{subequations}
    \label{eq:definition_ZNMt_and_QNM}
    \begin{align}
            \label{eq:definition_ZNMt}
      \ZNMt{n}{m}{t}(\bm{a})       &\defeq \sumss{h=1}{m} \sumss{q=1}{5} \az{hq} \ZNht{n}{hq}{t}         = \bmaz[\prime]{} \ZNhtvec{n}{\overbar{m}:t}, \\
            \label{eq:definition_QNM}
      \QNM{n}{m}(\bm{a})       &\defeq \sumss{h=1}{m} \sumss{q=1}{5} \az{hq} \QNh{n}{hq}         = \bmaz[\prime]{}  \QNhvec{n}{\overbar{m}},
    \end{align}
  \end{subequations}
  where $\ZNhtvec{n}{\overbar{m}:t}$ and $\QNhvec{n}{\overbar{m}}$ are defined in the
  obvious manner.  
        \end{definition}

\begin{lemma}
  \label{eq:connection_RV_nabla_QMN}
  $\QNhvec{n}{\overbar{m}}$ and $\sqrt{\bz{1}\bz{2}}\nablaM{m}\QMN[\thetaMb{m}{b}]{m}{n}$ share
  the same limiting distribution.
\end{lemma}

\begin{proof}
  The only difference between $\QNhvec{n}{\overbar{m}}$ and
  $\sqrt{\bz{1}\bz{2}}\nablaM{m}\QMN[\thetaMb{m}{b}]{m}{n}$ is that
  the first use $\ZNht{n}{hq}{t}$ where the second use
  $\tildeXNht{n}{hq}{t}$.  The difference between these components are
  \begin{align}
    \ZNht{n}{hq}{t} - \tildeXNht{n}{hq}{t}     = \sqrt{\bz{1}\bz{2}} \cdot \intss{\RRn{2}}{}     \Khb[\yh{h}-\LGp]{h}{\bm{b}} \Uh[\yh{h}]{hq:\bm{b}}     \parenC{\gh[\yh{h}]{h} -     \psi\!\left(\yh{h}; \thetah{h}
    \right) }  \d{\bm{\yh{h}}},
  \end{align}
            and this difference will not only approach zero but in fact be
  identical to zero when the bandwidth $\bm{b}$ is smaller than
  $\bmbz{0}$, since the population value $\thetahb{h}{b}$ in that case
  satisfies \cref{eq:def_of_theta_hb}.  The result now follows from
  \citet[Th.~25.4]{Billingsley12:_probab_measur}.
\end{proof}

The purpose of the new random variables introduced in
\cref{def:Xht_*Lc,def:ZNht_and_QNh,def:ZNMt_and_QNM} is to find under
which conditions the fourth requirement of the Klimko-Nelson approach
is satisfied in the general situation where $\mlimit$ and $\blimit$
when $\nlimit$.

The part that does require some effort to investigate is the fourth
requirement of \cref{th:Klimko_Nelson_1978},
which (using the notation introduced here) means that it is necessary
to verify that $\nz[-1/2]{}\,\QNhvec{n}{\overbar{m}}$ approaches a
normal distribution when $\bm{b}$ goes to zero when $n$ and~$m$ are
\enquote{large enough}.  The proof will be presented in a step by step
manner, that builds upon the asymptotic behaviour
of~$\E{\XNht{n}{hq}{i}\cdot\XNht{n}{jr}{k}}$.
The computation of this expectation will (depending on the indices
$h$, $i$, $j$ and~$k$) either require a bivariate, trivariate or
tetravariate integral.

\begin{table}[h]
  \begin{center}
    \begin{tabular}{l||c|c||c|c}
      Combinations       &$\LGp$ &$\bm{b}$  &$\Yht{h}{i}$ &$\Yht{j}{k}$ \\
      \hline
      First argument of $\Khb{h}{\bm{b}}$       &$\vz{1}$ &$\bz{1}$ &$\Yz{h+i}$ &$\Yz{j+k}$ \\
            Second argument of $\Khb{h}{\bm{b}}$       &$\vz{2}$ &$\bz{2}$ &$\Yz{i}$ &$\Yz{k}$
    \end{tabular}
    \caption{Factors deciding bivariate, trivariate or tetravariate.}
    \label{table:b12_Yhijk_variants}
  \end{center}
\end{table}

\Cref{table:b12_Yhijk_variants} lists the combinations that must be
taken into account when
computing~$\E{\XNht{n}{hq}{i}\cdot\XNht{n}{jr}{k}}$, i.e.\ the
presence of $\LGp$ and $\bm{b}$ and the dependence on $\Yz{t}$ in the
kernel~functions --- and it is evident from this table that the amount
of overlap in the indexing set
\mbox{$  \parenC{i, h+i, k, j+k}$} will decide if the resulting
integral turns out to be bi-, tri- or~tetravariate.
Note that \cref{eq:App:hatlgsd_definition_main_document} of
\myref{def:lgsd_estimator}{def:lgsd_esitimator_folded} implies that
only positive indices are required, so the bivariate case can thus
only occur when \mbox{$i=k$} and \mbox{$h=j$}.  It will be seen later
on that these bivariate components are the only ones that adds
non-negligible contributions to the asymptotic~behaviour.

\makeatletter{}

\subsubsection{The asymptotic results --- basic part}
\label{app:the_asymptotic_results}

The analysis of the asymptotic properties of~$\XNht{n}{hq}{i}$, from
\cref{def:Xht_*Lc}, would be quite simple if either the kernel
function~$K(\bm{w})$ or the score-function
components~$\Uh[\bm{w}]{hq:\bm{b}}$ had bounded support, since the
finiteness requirements of
\myref{assumption_Yt}{assumption:h_x:y_x:y:z_finite_expectations} then
would follow directly from
\cref{th:integrals_kernel_and_score_components}, and the proof of
\cref{th:expectations_of_XNht*} would be rather trivial.  However, in
the present analysis, $K(\bm{w})$ and~$\Uh[\bm{w}]{hq:\bm{b}}$ both
have~$\RRn{2}$ as their support, which implies that extra care must be
taken when working with the densities under consideration.

\begin{lemma}
  \label{th:expectations_of_XNht*}
  When $\Yz{t}$   satisfies \cref{assumption_Yt}, and
  $\Uh[\bm{w}]{hq:\bm{b}}$ and $K(\bm{w})$ are as given in
  \cref{def:kernel,def:Uh}, then the random variables
  $\XNht{n}{hq}{t}$ from \cref{def:Xht_*Lc} satisfies
  \begin{enumerate}[label=(\alph*)]
  \item     \label{eq:E(Xhi)}
        $\E{\XNht{n}{hq}{i}} = \Oh{\sqrt{\bz{1}\bz{2}}}$.
  \item     \label{eq:E(|Xhi|.nu)}
        $\power[1/\nu]{}{\E{\absp[\nu]{\XNht{n}{hq}{i}}}} =
    \Oh{\absp[(2-\nu)/2\nu]{\bz{1}\bz{2}}}$.
  \item     \label{eq:E(Xhi.Xjk)}
        $\E{\XNht{n}{hq}{i}\cdot\XNht{n}{jr}{k}} =     \begin{cases}
            \Uh[\LGp]{hq:\bm{b}} \Uh[\LGp]{hr:\bm{b}} \gh[\LGp]{h}       \intss{\RRn{2}}{} \power[2]{}{K(\bm{w})}\d{\bm{w}} 
            + \Oh{\bz{1}\vee\bz{2}}       &\text{when bivariate},\\
      \Oh{\bz{1}\wedge\bz{2}}       &\text{when trivariate},\\
      \Oh{\bz{1}\bz{2}}       &\text{when tetravariate},
    \end{cases}$ \\
                    where bivariate, trivariate and tetravariate refers to how many
    different $\Yz{t}$ the four indices $h$, $i$, $j$ and~$k$ gives,
    cf.\ \cref{table:b12_Yhijk_variants} for details.
  \end{enumerate}
\end{lemma}

\begin{proof}
                              The expectations in \cref{eq:E(Xhi),eq:E(|Xhi|.nu),eq:E(Xhi.Xjk)}
  are all finite due to
  \myref{assumption_Yt}{assumption:h_x:y_x:y:z_finite_expectations}
  and they do in addition correspond to integrals whose integrands are
  of the form \mbox{$\mathcal{V}\cdot g$}, where $g$ is a density
  function and $\mathcal{V}$ is an integrand of the type discussed in
  \cref{eq:int_K*U,eq:int_(KU)^nu,eq:int_KKUU} of
  \cref{th:integrals_kernel_and_score_components}, i.e.\ $\mathcal{V}$
  collects everything that only depends on the functions
  $\Uh[\bm{w}]{hq:\bm{b}}$ and $K(\bm{w})$.  The substitutions used in
  the proof of \cref{th:integrals_kernel_and_score_components} can be
  applied to the different cases under investigation, and it follows
  that these substitutions will create new integrals with the desired
  function of $\bz{1}$ and $\bz{2}$ as a scaling factor.    This proves \cref{eq:E(Xhi),eq:E(|Xhi|.nu)} and it also takes care
  of the trivariate and tetravariate cases of \cref{eq:E(Xhi.Xjk)}.

              \Cref{eq_assumption:gh_differentiable_at_LGp} from
  \myref{assumption_Yt}{assumption:gh_differentiable_at_LGp} is
  needed for the bivariate case of \cref{eq:E(Xhi.Xjk)}, i.e.\ the
  Taylor expansion of $\gh[\yh{h}]{h}$ around the point~$\LGp$
  allows the integral of interest to be written as the sum of the
  following three~integrals:
  \begin{subequations}
    \label{eq:E(Xhi.Xjk)_bivariate}
    \begin{align}
      \label{eq:E(Xhi.Xjk)_bivariate_constant}
      \mathcalJz{1} 
      &\defeq \intss{\RRn{2}}{}         \mathcal{V}\left(\yh{h}\right) \cdot \gh[\LGp]{h}         \d{\yh{h}}, \\
            \label{eq:E(Xhi.Xjk)_bivariate_linear}
      \mathcalJz{2} 
      &\defeq \intss{\RRn{2}}{}         \mathcal{V}\left(\yh{h} \right) \cdot 
        \left(         \power[\prime]{}{\bmmathfrakgz{h}(\LGp)}         \left[\yh{h}-\LGp \right]  \right)         \d{\yh{h}}, \\ 
            \label{eq:E(Xhi.Xjk)_bivariate_remainder}
      \mathcalJz{3} 
      &\defeq \intss{\RRn{2}}{}         \mathcal{V}\left(\yh{h} \right) \cdot 
        \left(         \power[\prime]{}{\bmmathfrakRz{h}(\yh{h})}         \left[\yh{h}-\LGp \right] \right)         \d{\yh{h}}.
    \end{align}
  \end{subequations}

  The bivariate case of
  \myref{th:integrals_kernel_and_score_components}{eq:int_KKUU} shows
  that the term $\mathcalJz{1}$ gives the desired
  result, so it remains to prove that the terms
  $\mathcalJz{2}$ and $\mathcalJz{3}$ are
  $\Oh{\bz{1}\vee\bz{2}}$.  For this investigation, the substitution
  \mbox{$\wz{1} = \left(\yz{h}-\vz{1}\right)/\bz{1}$} and
  \mbox{$\wz{2} = \left(\yz{0}-\vz{2}\right)/\bz{2}$} must be applied,
  which in particular replaces the vector
  \mbox{$\left[\yh{h}-\LGp \right]$} with the vector
  \mbox{$\Vector[\prime]{\bz{1}\wz{1},\bz{2}\wz{2}}$}.
  In order to compactify the notation, let $\az{1}$ and $\az{2}$
  denote the two components
  of~$\bmmathfrakgz{h}(\LGp)$, let $\mathcal{W}$ be
  the substituted version of $\mathcal{V}$, let
  $\mathfrakRz{h1}$ and $\mathfrakRz{h2}$
  be the two components of the remainder function and finally let
  $\mathfrakTz{h1}$ and $\mathfrakTz{h2}$
  be the substituted versions of
  $\mathfrakRz{h1}\mathcal{W}$ and
  $\mathfrakRz{h2}\mathcal{W}$.

  With this notation, the substitution used upon
  $\mathcalJz{2}$ gives
  \begin{align}
    \label{E(Xhi.Xjk)_substitution_linear}
    \mathcalJz{2}     &= \az{1}\bz{1} \intss{\RRn{2}}{}       \wz{1} \cdot \mathcal{W}(\bm{w})        \d{\bm{w}}
      + \az{2}\bz{2} \intss{\RRn{2}}{}       \wz{2} \cdot \mathcal{W}(\bm{w})       \d{\bm{w}},
  \end{align}
  whose integrands include an extra factor of $\wz{1}$ or $\wz{2}$
  compared to the integrands encountered in the proof of
  \cref{th:integrals_kernel_and_score_components}.  This is however no
  problem, since \myref{th:integrals:K_w1w2}{eq:int_|K2w_i|} implies
  that the finiteness conclusion still holds true in these cases,
  which implies that $\mathcalJz{2}$ is $\Oh{\bz{1}\vee\bz{2}}$

  Since
  \myref{assumption_Yt}{assumption:h_x:y_x:y:z_finite_expectations}
  ensures that the sum of the three integrals
  $\mathcalJz{1}$, $\mathcalJz{2}$ and
  $\mathcalJz{3}$ is finite, and the above discussion
  shows that the two first of them are finite, it follows that
  $\mathcalJz{3}$ also is~finite.  An inspection
  of~$\mathcalJz{3}$ after substitution,~i.e.\
    \begin{align}
      \label{E(Xhi.Xjk)_substitution_remainder}
      \mathcalJz{3}       &=  \intss{\RRn{2}}{}         \left[
        \bz{1}\wz{1} \cdot         \mathfrakTz{h1}(\bm{y}(\bm{w}))         + \bz{2}\wz{2} \cdot
        \mathfrakTz{h2}(\bm{y}(\bm{w}))         \right] \d{\bm{w}},
    \end{align}
    then reveal that the maximum of $\bz{1}$ and $\bz{2}$ can be
    factorised out of the integrand.  This implies that
    $\mathcalJz{3}$ is $\Oh{\bz{1}\vee\bz{2}}$, and thus concludes the
    proof of \cref{th:expectations_of_XNht*}
\end{proof}

The following corollary is handy when the covariance is the target of interest.

\begin{corollary}
  \label{th:covariances_of_XNht*}
  When $\Yz{t}$   satisfies \cref{assumption_Yt}, and
  $\Uh[\bm{w}]{hq:\bm{b}}$ and $K(\bm{w})$ are as given in
  \cref{def:kernel,def:Uh}, then the random variables
  $\XNht{n}{hq}{t}$ from \cref{def:Xht_*Lc} satisfies
  \begin{align}
    \label{eq:covariances_of_XNht*}
    \Cov{\XNht{n}{hq}{i}}{\XNht{n}{jr}{k}} =     \begin{cases}
      \Uh[\LGp]{hq:\bm{b}} \Uh[\LGp]{hr:\bm{b}} \gh[\LGp]{h}
      \intss{\RRn{2}}{} \power[2]{}{K(\bm{w})}\d{\bm{w}}       + \Oh{\bz{1}\vee\bz{2}}       &\text{when bivariate},\\
      \Oh{\bz{1}\wedge\bz{2}}       &\text{when trivariate},\\
      \Oh{\bz{1}\bz{2}}       &\text{when tetravariate}.
    \end{cases}
  \end{align}
        \end{corollary}

\begin{proof}
  Since \mbox{$\Cov{\XNht{n}{hq}{i}}{\XNht{n}{jr}{k}} =     \E{\XNht{n}{hq}{i}\cdot\XNht{n}{jr}{k}} -
    \E{\XNht{n}{hq}{i}}\cdot\E{\XNht{n}{jr}{k}}$}, the result follows
  immediately from an inspection of \cref{eq:E(Xhi),eq:E(Xhi.Xjk)} of
  \cref{th:expectations_of_XNht*}.
\end{proof}

The next corollary is needed in the proof of
\cref{th:convergent_covariance_ZNMt}.

\begin{corollary}
  \label{th:Lp_properties_for_ZNht_and_ZNMt}
  When $\Yz{t}$   satisfies \cref{assumption_Yt}, and
  $\Uh[\bm{w}]{hq:\bm{b}}$ and $K(\bm{w})$ are as given in
  \cref{def:kernel,def:Uh}, then the random variables
  $\ZNht{n}{hq}{t}$ and $\ZNMt{n}{m}{t}(\bm{a})$ from
  \cref{def:ZNht_and_QNh} satisfies
  \begin{enumerate}[label=(\alph*)]
  \item     \label{th:Lp_properties_for_ZNht_limit}
        \mbox{$\power[1/\nu]{}{\E{\absp[\nu]{\ZNht{n}{hq}{t}}}} =
      \Oh{\absp[(2-\nu)/2\nu]{\bz{1}\bz{2}}}$}.
      \item     \label{th:Lp_properties_for_ZNMt_limit}
            \mbox{$\power[1/\nu]{}{\E{\absp[\nu]{\ZNMt{n}{m}{t}(\bm{a})}}}
      =   \Oh{m\absp[(2-\nu)/2\nu]{\bz{1}\bz{2}}}$}.
  \end{enumerate}
\end{corollary}

\begin{proof}
  The connection between expectations and $\Lp{\nu}$-spaces discussed in
  \cref{app:sigma_algebras_Lp_spaces}, see
  \cref{eq:expectation_and_Lph}, can be applied here, which in essence
  reduces the proof to a simple application of Minkowski's inequality.
    For \cref{th:Lp_properties_for_ZNht_limit}, note that
  \cref{th:expectations_of_XNht*} gives the following result
    \begin{subequations}
    \begin{align}
      \label{th:Lp_properties_for_ZNht_limit_proof}
      \power[1/\nu]{}{\E{\absp[\nu]{\ZNht{n}{hq}{t}}}}       &=          \power[1/\nu]{}{\E{\absp[\nu]{\XNht{n}{hq}{t} -
        \E{\XNht{n}{hq}{t}} }}} \\
      &\leq         \power[1/\nu]{}{\E{\absp[\nu]{\XNht{n}{hq}{t}}}} +
        \power[1/\nu]{}{\E{\absp[\nu]{\E{\XNht{n}{hq}{t}}}}} \\ 
      &=         \Oh{\absp[(2-\nu)/2\nu]{\bz{1}\bz{2}}} + 
        \Oh{\sqrt{\bz{1}\bz{2}}} \\
      &= \Oh{\absp[(2-\nu)/2\nu]{\bz{1}\bz{2}}}.
    \end{align}
  \end{subequations}
    \Cref{th:Lp_properties_for_ZNMt_limit} now follows from
  \cref{th:Lp_properties_for_ZNht_limit} and \cref{th:Lp_expectation},
  due to the following inequality,
  \begin{subequations}
    \begin{align}
      \label{th:Lp_properties_for_ZNMt_limit_proof}
      \power[1/\nu]{}{\E{\absp[\nu]{\ZNMt{n}{m}{t}(\bm{a})}}}       &=          \power[1/\nu]{}{\E{\absp[\nu]{\sumss{h=1}{m} \sumss{q=1}{5}
        \az{hq} \ZNht{n}{hq}{t}} }} \\
      &\leq          \sumss{h=1}{m} \sumss{q=1}{5} \absp{\az{hq}}         \power[1/\nu]{}{\E{\absp[\nu]{\ZNht{n}{hq}{t}}}} \\
      &\leq          \sumss{h=1}{m} \sumss{q=1}{5}         \Az{\overbar{m}} \cdot
                                                \Oh{\absp[(2-\nu)/2\nu]{\bz{1}\bz{2}}} \\
      &= \Oh{m \absp[(2-\nu)/2\nu]{\bz{1}\bz{2}}}.
    \end{align}
  \end{subequations}
  where $\Az{\overbar{m}}$ is the maximum of $\absp{\az{hq}}$.
\end{proof}

\makeatletter{}
\subsubsection{The asymptotic results --- final part}
\label{app:the_asymptotic_results_advanced}

This section will present the final steps toward the verification of
the fourth requirement of the Klimko-Nelson approach for the case
where $\mlimit$ and $\blimit$ when $\nlimit$.  Note that
\cref{th:AN_of_QNh_and_QNM} (the main theorem) requires both a large
block - small block argument and a truncation argument, and the
technical details related to these components will be taken care of in
\cref{th:convergent_covariance_ZNMt} and
\cref{th:eta_variance_expressions}.

The large block - small block argument requires that quite a few
components must be verified to be asymptotically negligible.  The
following lemma, which extends an argument encountered in the proof
of~\citet[Lemma~4.3(b)]{masry95:_nonpar_estim_ident_nonlin_arch},
shows that the asymptotic negligibility of all the \enquote{off the
  diagonal} components can be taken care of in one operation.

\begin{lemma}
  \label{th:convergent_covariance_ZNMt}
  When $\Yz{t}$   satisfies \cref{assumption_Yt}, when $n$, $m$ and $\bm{b}$ are as
  specified in \cref{assumption_Nmb}, and when
  $\Uh[\bm{w}]{hq:\bm{b}}$ and $K(\bm{w})$ are as given in
  \cref{def:kernel,def:Uh} --- then the random variables
  $\ZNMt{n}{m}{t}(\bm{a})$ from \cref{def:ZNMt_and_QNM} satisfies
  \begin{align}
    \label{eq:convergent_covariance_ZNMt_general_M}
    \frac{1}{n} \sumss{\stackrel{i,k = 1}{\scriptscriptstyle i\neq k}}{n}     \absp{\E{\ZNMt{n}{m}{i}(\bm{a}) \cdot     \ZNMt{n}{m}{k}(\bm{a}) }} &= \oh{1}.
  \end{align}
\end{lemma}

\begin{proof}
  \Myref{assumption_Yt}{assumption_Yt_strictly_stationary}, i.e.\ the
  strict stationarity of $\TSR{\Yz{t}}{t\in\ZZ}{}$, implies that the
  double sum in \cref{eq:convergent_covariance_ZNMt_general_M} can be
  reduced to a single sum, i.e.\
  \begin{align}
    \label{eq:single_sum_o(1)}
    \frac{1}{n} \sumss{\stackrel{i,k = 1}{\scriptscriptstyle i\neq k}}{n}     \absp{\E{\ZNMt{n}{m}{i}(\bm{a}) \cdot     \ZNMt{n}{m}{k}(\bm{a}) }}      &= 2 \sumss{\ell=1}{n-1} \left(1
      - \frac{\ell}{n}\right)       \INMl{n}{m}{\ell}(\bm{a}),
  \end{align}
  where the terms $\INMl{n}{m}{\ell}(\bm{a})$ are given by
  \begin{subequations}
    \label{eq:definiton_and_rewriting_of_INMl}
    \begin{align}
      \label{eq:definition_of_INMl}
      \INMl{n}{m}{\ell}(\bm{a})       &\defeq \absp{
        \E{\ZNMt{n}{m}{0}(\bm{a}) \cdot         \ZNMt{n}{m}{\ell}(\bm{a})} } \\
      &= \absp{ \E{        \sumss{h=1}{m} \sumss{q=1}{5} \vechcomp{hq}{a} \ZNht{n}{hq}{0}         \cdot         \sumss{j=1}{m} \sumss{r=1}{5} \vechcomp{jr}{a}
        \ZNht{n}{jr}{\ell} }} \\
      \label{eq:rewriting_of_INMl}
      &= \absp{ \sumss{h=1}{m} \sumss{j=1}{m} \sumss{q=1}{5} \sumss{r=1}{5}         \vechcomp{hq}{a} \vechcomp{jr}{a} \E{        \ZNht{n}{hq}{0}         \cdot         \ZNht{n}{jr}{\ell} }} \\
          \label{eq:proof_o(1)_INMl_to_INhl}
      &\leq \sumss{h=1}{m} \sumss{j=1}{m}  \sumss{q=1}{5} \sumss{r=1}{5}         \absp{\vechcomp{hq}{a}}\! \absp{\vechcomp{jr}{a}}         \INhjl{n}{hq}{jr}{\ell},
    \end{align}
  \end{subequations}
  where \mbox{$\INhjl{n}{hq}{jr}{\ell} \defeq \absp{\E{         \ZNht{n}{hq}{0} \cdot         \ZNht{n}{jr}{\ell} }} = \absp{\Cov{         \XNht{n}{hq}{0}}{         \XNht{n}{jr}{\ell} }}$}.
  
  Introducing integers $\kz{n}$ (to be specified later on) such that
  \mbox{$\kz{n}\rightarrow\infty$} and
  \mbox{$\kz{n}\mz[2]{}\bz{1}\bz{2}\rightarrow0$} as
  \mbox{$n\rightarrow\infty$}, \cref{eq:single_sum_o(1)}~can be
  written as the sum of the following three sums,
  \begin{subequations}
    \label{eq:proof_o(1)_J}
    \begin{align}
      \label{eq:proof_o(1)_J1}
      \Jz{1} &\defeq 2 \sumss{\ell = 1}{m} \left(1 -
            \ell/n\right) \INMl{n}{m}{\ell}(\bm{a}), \\
      \label{eq:proof_o(1)_J2}
      \Jz{2} &\defeq 2 \sumss{\makebox[0pt]{$\scriptstyle \ell =
            m+1$}}{\makebox[0pt]{$\scriptstyle \kz{n} + m$}} \,\left(1 -
            \ell/n\right) \INMl{n}{m}{\ell}(\bm{a}), \\
      \label{eq:proof_o(1)_J3}
      \Jz{3} &\defeq 2 \sumss{\makebox[0pt]{$\scriptstyle \ell =
            \kz{n}+m+1$}}{n-1} \, \left(1 -
            \ell/n\right) \INMl{n}{m}{\ell}(\bm{a}).
    \end{align}
  \end{subequations}

  From the definition of $\INMl{n}{m}{\ell}(\bm{a})$ it is seen that
  in $\Jz{1}$ there will be some overlap between those $\Yz{t}$ that
  are a part of $\ZNMt{n}{m}{0}(\bm{a})$ and those that are a part of
  \mbox{$\ZNMt{n}{m}{\ell}(\bm{a})$}, and moreover that this will not
  be the case for the two sums $\Jz{2}$ and~$\Jz{3}$.

        \Cref{eq:proof_o(1)_INMl_to_INhl,eq:proof_o(1)_J1} implies that a
  squeeze argument can be used when dealing with $\Jz{1}$, i.e.\
  \begin{align}
    \label{eq:bound_on_J1}
    0 \leq \Jz{1} \leq     2\cdot \left(\max_{\stackrel{    \scriptscriptstyle h \in \parenC{1,\dotsc,m}}{
    \scriptscriptstyle q \in \parenC{1,\dotsc,5}}}
    \absp[2]{\az{hq}} \right) \cdot     \sumss{\ell=1}{m} \sumss{h=1}{m} \sumss{j=1}{m}  \sumss{q=1}{5} \sumss{r=1}{5} \absp{\Cov{     \XNht{n}{hq}{0}}{     \XNht{n}{jr}{\ell} }},
  \end{align}
  and \cref{th:covariances_of_XNht*} can be used to determine how the
  summand behaves in the limit.        \Cref{table:b12_Yhijk_variants}, page
  \pageref{table:b12_Yhijk_variants}, shows that the bivariate case
  never occurs, that $h$ must be equal to $\ell$ or \mbox{$j+\ell$} in
  order for a trivariate case to occur, and that the rest of the cases
  must be~tetravariate.  It is not hard (but a bit tedious) to
  explicitly compute the number of trivariate terms that occur in
  \cref{eq:bound_on_J1}, but for the present asymptotic analysis it is
  sufficient to note that the number of trivariate terms is of
  order~$\mz[2]{}$, whereas the number of tetravariate terms is of
  order~$\mz[3]{}$.    \Cref{th:covariances_of_XNht*} thus gives that the bivariate and
  tetravariate parts of the bound for $\Jz{1}$ respectively are
  $\Oh{\mz[2]{}\!\left(\bz{1}\wedge\bz{2}\right)}$ and
  $\Oh{\mz[3]{}\bz{1}\bz{2}}$.

  \mbox{$\Jz{1}=\oh{1}$} now follows from
  \myref{assumption_Nmb}{eq:assumption_m(b1 join b2)} and the
  following two simple observations;
  \begin{subequations}
    \label{eq:J1_dominating_order}
    \begin{align}
      \label{eq:J1_dominating_order_bivariate}
      \mz[2]{}\!\left(\bz{1}\wedge\bz{2}\right)       & \leq \mz[2]{}\!\left(\bz{1}\vee\bz{2}\right), \\
            \label{eq:J1_dominating_order_trivariate}
      \mz[3]{}\bz{1}\bz{2}       &\leq \mz[-1]{} \cdot
        \mz[4]{}
        \! \parenRz[2]{}{\bz{1}\vee\bz{2}}
        =  \mz[-1]{} \cdot 
        \parenRz[2]{}{\mz[2]{}\left(\bz{1}\vee\bz{2}\right)}.
    \end{align}
  \end{subequations}

  For $\Jz{2}$, a squeeze similar to the one in \cref{eq:bound_on_J1}
  can be used.  The situation becomes simpler since \mbox{$\ell>M$}
  ensures that only the tetravariate case is present, and the order of
  $\Jz{2}$~becomes
  \begin{align}
    \label{eq:order_of_J2}
    \Jz{2} 
    &= \Oh{\kz{n} \mz[2]{} \bz{1}\bz{2}}.
  \end{align}
  Since \mbox{$\kz{n}\mz[2]{}\bz{1}\bz{2}\rightarrow0$} (with a choice
  of $\kz{n}$ to be specified below), it follows
  that~\mbox{$\Jz{2}=\oh{1}$}.

  For $\Jz{3}$, the Corollary of Lemma~2.1
  in~\citet{Davydov:1968:CDG} will be used to get an upper bound
  on~$\INMl{n}{m}{\ell}(\bm{a})$, such that a squeeze-argument can
  be used for $\Jz{3}$ too.  The requirements needed for
  Davydov's result are covered as follows: The strong mixing
  requirement is covered by \cref{assumption_Yt}, and (for a given $m$
  and $\bm{b}$) the requirement about finite expectations follows from
  \myref{th:Lp_properties_for_ZNht_and_ZNMt}{th:Lp_properties_for_ZNMt_limit}.

  The \mbox{$\sigma$-algebras} to be used follows from the comment
  stated after \cref{eq:sigma_algebra_Yht_range}, i.e.\ that
  \mbox{$\ZNMt{n}{m}{0}(\bm{a}) \in \sigmaYhtRange{0}{m}$},
  whereas~\mbox{$\ZNMt{n}{m}{\ell}(\bm{a}) \in
    \sigmaYhtRange{\ell}{\ell+m} \subset \sigmaYhtRange{m +
      (\ell-m)}{\infty}$}.   Thus, for \mbox{$\ell > \kz{n} + m$}, the following bound
  is obtained on~$\INMl{n}{m}{\ell}(\bm{a})$,
 \begin{subequations}
    \label{eq:bound_on_INMl}
    \begin{align}
      \label{eq:bound_on_INMl_definition}       \INMl{n}{m}{\ell}(\bm{a}) 
      &=         \absp{\E{\ZNMt{n}{m}{0}(\bm{a}) \cdot \ZNMt{n}{m}{\ell}(\bm{a})} } \\
      \label{eq:bound_on_INMl_zero_mean}       &=         \absp{\E{\ZNMt{n}{m}{0}(\bm{a}) \cdot
        \ZNMt{n}{m}{\ell}(\bm{a})} -         \E{\ZNMt{n}{m}{0}(\bm{a})} \cdot
        \E{\ZNMt{n}{m}{\ell}(\bm{a})} } \\
      \label{eq:bound_on_INMl_Davydov}       &\leq         12
        \parenRz[1/\nu]{}{\E{\absp[\nu]{\ZNMt{n}{m}{0}(\bm{a})}}}
        \cdot
        \parenRz[1/\nu]{}{\E{\absp[\nu]{\ZNMt{n}{m}{\ell}(\bm{a})}}}
        \cdot         \parenSz[1-1/\nu-1/\nu]{}{\alpha(\ell-m)}\\
      \label{eq:bound_on_INMl_stationarity}       &=         12
        \parenRz[2]{}{        \parenRz[1/\nu]{}{        \E{\absp[\nu]{\ZNMt{n}{m}{0}(\bm{a})}}}} \cdot         \parenSz[1-2/\nu]{}{\alpha(\ell-m)} \\
      \label{eq:bound_on_INMl_boundedness}       &= 12 \parenRz[2]{}{
        \Oh{m\absp[(2-\nu)/2\nu]{\bz{1}\bz{2}}}}
        \cdot         \parenSz[1-2/\nu]{}{\alpha(\ell-m)} \\
      \label{eq:bound_on_INMl_boundedness_II}       &\leq \mathcal{C} \cdot \mz[2]{} \cdot
        \absp[(2-\nu)/\nu]{\bz{1}\bz{2}}\cdot         \parenSz[1-2/\nu]{}{\alpha(\ell-m)},
    \end{align}
  \end{subequations}
  where \cref{eq:bound_on_INMl_zero_mean} follows since the mean of
  $\ZNMt{n}{m}{t}(\bm{a})$ by construction is zero,
    where \cref{eq:bound_on_INMl_Davydov} is Davydov's result, where
  \cref{eq:bound_on_INMl_stationarity} use the strict stationarity of
  the process $\parenC{\Yz{t}}$, where
  \cref{eq:bound_on_INMl_boundedness} is due to
  \myref{th:Lp_properties_for_ZNht_and_ZNMt}{th:Lp_properties_for_ZNMt_limit},
  and
      finally \cref{eq:bound_on_INMl_boundedness_II} is an equivalent
  statement, using a suitable constant $\mathcal{C}$ to express the
  upper~bound.

  A squeeze for $\Jz{3}$ can now be stated in the following manner
  \begin{align}
    \label{eq:squeeze_J3}
    0 \leq \Jz{3} \leq \mathcalCz{3} \cdot     \sumss{\makebox[0pt]{$\scriptstyle j = \kz{n} + 1$}}{\infty}     \ \left( \mz[2]{} \cdot
    \absp[(2-\nu)/\nu]{\bz{1}\bz{2}} \right)\cdot     \parenSz[1-2/\nu]{}{\alpha(j)},
  \end{align}
  where $\mathcalCz{3}$ is a constant, where the index has been
  shifted by introducing \mbox{$j=\ell-m$}, and where the sum
  from \cref{eq:proof_o(1)_J3} has been extended to infinity (adding
  only non-negative summands).
  
  A comparison of \cref{eq:squeeze_J3} with the finiteness requirement
  that the strong mixing coefficients should satisfy, see
  \myref{assumption_Yt}{assumption_Yt_strong_mixing}, indicates that
  if
  \mbox{$\jz[a]{} \geq \mz[2]{} \cdot
    \absp[(2-\nu)/\nu]{\bz{1}\bz{2}}$} for \mbox{$j \geq \kz{n} + 1$},
  then that could be used to get a new upper bound in
  \cref{eq:squeeze_J3}.  Taking the $\ith{a}$ root on both sides,
      it is clear that the desired inequality can be obtained when
  \mbox{$\kz{n} + 1 = \ceil{      \mz[2/a]{} \cdot
      \absp[(2-\nu)/a\nu]{\bz{1}\bz{2}} }$}, which gives the new bound
                \begin{align}
    \label{eq:squeeze_J3_improved}
    0 \leq \Jz{3} \leq \mathcalCz{3} \cdot     \sumss{\makebox[0pt]{$\scriptstyle j = \kz{n} + 1$}}{\infty}     \ \jz[a]{} \parenSz[1-2/\nu]{}{\alpha(j)},
  \end{align}
  and if \mbox{$\kz{n}\rightarrow\infty$} when
  $\nlimit$, the finiteness assumption from
  \myref{assumption_Yt}{assumption_Yt_strong_mixing} gives
  that~\mbox{$\Jz{3} = \oh{1}$}.

  Finally, \cref{th:k_N_limit_for_off_diagonal_components} verifies
  that $\kz{n}$ satisfies the two limits
  \mbox{$\kz{n} \mz[2]{} \bz{1}\bz{2}
    \rightarrow 0$} (needed for the \mbox{$\Jz{2}$-term})
  and \mbox{$\kz{n} \rightarrow \infty$} (needed for the
  \mbox{$\Jz{3}$-term}).  Altogether, this shows that
  \cref{eq:convergent_covariance_ZNMt_general_M} can be rewritten as
  \mbox{$\Jz{1}+\Jz{2}+\Jz{3}$}, all
  of which are $\oh{1}$, and the proof is~complete.
\end{proof}

The following observations are needed in the truncation argument of \cref{th:AN_of_QNh_and_QNM}.

\begin{corollary}
      \label{th:var(ZNMt)}
  \label{th:eta_variance_expressions}
  When $\Yz{t}$   satisfies \cref{assumption_Yt},   when $n$, $m$ and $\bm{b}$ are as specified in
  \cref{assumption_Nmb},         and with
  \mbox{$\WMbN{m}{}{\bm{b}}=\oplusss{h=1}{m} \WhN{}{h}{\bm{b}}$} and
  \mbox{$\bm{a} = \baM{m} =
    \Vector[\prime]{\bmaz{1},\dotsc,\bmaz{m}}$} (with
  \mbox{$\bmaz{h} \in \RRn{5}$}) as given in
  \cref{def:matrices_related_to_the_new_penalty_function}, then the
  random variable $\ZNMt{n}{m}{t}(\bm{a})$ from
  \cref{def:ZNMt_and_QNM} satisfies
  \begin{enumerate}[label=(\alph*)]
  \item     \label{th:var(ZNMt)=expression}
        \label{th:var(ZNMt)=O(m)}
        $\Var{\ZNMt{n}{m}{t}(\bm{a})} =         \baM[\prime]{m} \WMbN{m}{}{\bm{b}}\baM{m} +
    \Oh{\mz[2]{}\cdot
      \parenR{\bz{1}\vee\bz{2}} } =     \sumss{h=1}{m} \bmaz[\prime]{h} \WhN{}{h}{\bm{b}} \bmaz{h}  +     \Oh{\mz[2]{}\cdot
      \parenR{\bz{1}\vee\bz{2}} }      = \Oh{m}$.
              \end{enumerate}
    Furthermore, with \mbox{$r\defeq\rz{n}$} a sequence of integers that
  goes to~$\infty$ when $\nlimit$, and for a given threshold
  value~$L$, the following holds for the random variables
  \mbox{$\etaz{1:r} \defeq \sumss{t=1}{r} \ZNMt{n}{m}{t}(\bm{a})$},
  \mbox{$\etaz[\leq L]{1:r} \defeq \sumss{t=1}{r} \ZNMt[\leq
    L]{n}{m}{t}(\bm{a})$} and
  \mbox{$\etaz[> L]{1:r} \defeq \sumss{t=1}{r}
    \ZNMt[>L]{n}{m}{t}(\bm{a})$}.
    \begin{enumerate}[label=(\alph*),resume]
  \item     \label{th:eta_variance_expressions_not_truncated}
        \mbox{$\Var{\etaz{1:r}} =       r \cdot \parenC{\sumss{h=1}{m} \bmaz[\prime]{h}
        \WhN{}{h}{\bm{b}} \bmaz{h} + \oh{1} }$}. 
  \item     \label{th:eta_variance_expressions_truncated}
        When $L$ is large enough,     \mbox{$\Var{\etaz[\leq L]{1:r}} =       r \cdot \parenC{\sumss{h=1}{m} \bmaz[\prime]{h}
        \WhN{}{h}{\bm{b}} \bmaz{h} + \oh{1} }$} and     \mbox{$\Var{\etaz[> L]{1:r}} = r \cdot \oh{1}$}.
  \end{enumerate}
\end{corollary}

\begin{proof}
  For \cref{th:var(ZNMt)=expression}, note that it follows from
  \cref{def:ZNht_and_QNh,def:ZNMt_and_QNM} that
  \begin{subequations}
    \begin{align}
      \label{eq_th:var(ZNMt)=expression_first}
      \small \Var{\ZNMt{n}{m}{t}(\bm{a})}       &{\small =         \sumss{h=1}{m} \sumss{j=1}{m} \sumss{q=1}{5} \sumss{r=1}{5}
        \az{hq}\az{jr} \Cov{\XNht{n}{hq}{t}}{\XNht{n}{jr}{t}}} \\
      \label{eq_th:var(ZNMt)=expression_second}
      &{\small = 
        \sumss{h=1}{m} \sumss{q=1}{5} \sumss{r=1}{5}
        \az{hq}\az{hr} \Cov{\XNht{n}{hq}{t}}{\XNht{n}{hr}{t}} +        \sumss{\stackrel{h,j = 1}{\scriptscriptstyle h\neq j}}{m} \sumss{q=1}{5} \sumss{r=1}{5}
        \az{hq}\az{jr} \Cov{\XNht{n}{hq}{t}}{\XNht{n}{jr}{t}}}.
    \end{align}
  \end{subequations}
  The bivariate case of \cref{th:covariances_of_XNht*} can be applied
  to the \enquote{diagonal part} of the sum in
  \cref{eq_th:var(ZNMt)=expression_second}, whereas the trivariate and
  tetravariate cases can be applied to the \enquote{off-diagonal
    part}.  The \enquote{diagonal part} can thus be written as the sum
  of \\
  {$\sumss{h=1}{m} \sumss{q=1}{5} \sumss{r=1}{5} \az{hq}\az{hr}     \Uh[\LGp]{hq:\bm{b}} \Uh[\LGp]{hr:\bm{b}} \gh[\LGp]{h}
    \intss{\RRn{2}}{} \power[2]{}{K(\bm{w})}\d{\bm{w}}$} (which is
  equal to 
  {$\baM[\prime]{} \WMbN{m}{}{\bm{b}} \baM{}=\sumss{h=1}{m}
    \bmaz[\prime]{h} \WhN{}{h}{\bm{b}} \bmaz{h}$}) and a sum that is
  \mbox{$\Oh{m\cdot\left(\bz{1}\vee\bz{2}\right)}$}.  For the
  \enquote{off-diagonal part} the result is
  \mbox{$\Oh{\mz[2]{}\cdot\left(\bz{1}\wedge\bz{2}\right)}$}.  Both of
  these asymptotically negligible terms are covered by
  \mbox{$\Oh{\mz[2]{}\cdot\left(\bz{1}\vee\bz{2}\right)}$}, and this
  gives the two first equalities of \cref{th:var(ZNMt)=expression}.
  The last equality follows since the summands
          $\bmaz[\prime]{h} \WhN{}{h}{\bm{b}} \bmaz{h}$ are finite.
  
  For \cref{th:eta_variance_expressions_not_truncated}, note that the
  variance can be expressed as
  \begin{align}
    \Var{\etaz{1:r}} = \sumss{i=1}{r}
    \Var{\ZNMt{n}{m}{i}(\bm{a})} +     \sumss{\stackrel{i,k = 1}{\scriptscriptstyle i\neq k}}{r}
    \E{\ZNMt{n}{m}{i}(\bm{a}) \cdot \ZNMt{n}{m}{k}(\bm{a})}.
  \end{align}
  The \enquote{on diagonal} part of this sum equals
  \mbox{$r\cdot \Var{\ZNMt{n}{m}{1}(\bm{a})}$} due to
  \myref{assumption_Yt}{assumption_Yt_strictly_stationary}, while the
  \enquote{off diagonal} part due to
  \cref{th:convergent_covariance_ZNMt} becomes \mbox{$r\cdot\oh{1}$}.
  Together with the result from \cref{th:var(ZNMt)=expression}, this
  gives the statement in
  \cref{th:eta_variance_expressions_not_truncated}.
  
  The truncated cases in \cref{th:eta_variance_expressions_truncated}
  use the same arguments as those encountered in
  \cref{th:eta_variance_expressions_not_truncated}, with the effect
  that the \mbox{$\Uh[\LGp]{hq:\bm{b}} \Uh[\LGp]{hr:\bm{b}}$} that
  occurs in $\WhN{}{h}{\bm{b}}$ either are replaced by
  \mbox{$\UhL[\LGp]{hq:\bm{b}}{\leq L}\UhL[\LGp]{hr:\bm{b}}{\leq L}$}
  or by \mbox{$\UhL[\LGp]{hq:\bm{b}}{>L}\UhL[\LGp]{hr:\bm{b}}{>L}$}.
    \Myref{th:Uh_finite}{th:Uh_threshold_limit} gives that
  \mbox{$\UhL[\LGp]{hq:\bm{b}}{\leq L}=\Uh[\LGp]{hq:\bm{b}}$} when~$L$
  is large enough (and thus \mbox{$\UhL[\LGp]{hq:\bm{b}}{>L}=0$}),
  which completes the proof.
\end{proof}

The main theorem can now be stated, i.e.\ this result can be used to
verify the fourth requirement of the Klimko-Nelson approach for the
penalty function $\QMN[\thetaMb{m}{b}]{m}{n}$, from which it follows
an asymptotic normality result for
$\widehatbmthetaz{\LGp|\overbar{m}|\bm{b}}$, that finally gives the
asymptotic normality result of $\hatlgsdM[p]{\LGp}{\omega}{m}$.
(Confer \cref{note:interpretation_large_m} for an interpretation of
the $m$ that occurs in the limiting distributions.)

\begin{theorem}
  \label{th:AN_of_QNh_and_QNM}
  For a given point \mbox{$\LGp=\LGpoint$}:
  When $\Yz{t}$   satisfies \cref{assumption_Yt,assumption_score_function},   when $n$, $m$ and $\bm{b}$ are as specified in
  \cref{assumption_Nmb},         and with
  \mbox{$\WMbN{m}{}{\bm{b}}=\oplusss{h=1}{m} \WhN{}{h}{\bm{b}}$} and
  {$\bm{a} = \baM{m} =
    \Vector[\prime]{\bmaz{1},\dotsc,\bmaz{m}}$} (with
  \mbox{$\bmaz{h} \in \RRn{5}$}) as given in
  \cref{def:matrices_related_to_the_new_penalty_function},
      then the random variables
  $\QNM{n}{m}(\bm{a})$ and $\QNhvec{n}{\overbar{m}}$ from
  \cref{def:ZNMt_and_QNM} will for small $\bm{b}$ and large $m$ and
  $n$ satisfy
  \begin{enumerate}[label=(\alph*)]
  \item     \label{th:AN_of_QNM}
        $\nz[-1/2]{}\,\QNM{n}{m}(\bm{a}) 
    \stackrel{\scriptscriptstyle d}{\longrightarrow}     \UVN{0}{\sumss{h=1}{m} \bmaz[\prime]{h} \WhN{}{h}{\bm{b}} \bmaz{h}}$, i.e.\
    asymptotically univariate normal.   \item     \label{th:AN_of_QNh_vector}
        $\nz[-1/2]{}\,\QNhvec{n}{\overbar{m}}
    \stackrel{\scriptscriptstyle d}{\longrightarrow}     \UVN{\bm{0}}{\oplusss{h=1}{m} \WhN{}{h}{\bm{b}}}$, i.e.\ asymptotically
    \mbox{$5m$-variate} normal.   \end{enumerate}
\end{theorem}

\begin{proof}
  For the proof of \cref{th:AN_of_QNM}, note the following connection
  between $\QNM{n}{m}(\bm{a})$ and $\ZNMt{n}{m}{t}(\bm{a})$ which
  follows directly from \cref{def:ZNht_and_QNh,def:ZNMt_and_QNM},
  \begin{subequations}
    \begin{align}
      \label{eq:QNM_vs_ZNMt}
      {\small \QNM{n}{m}(\bm{a})       = \sumss{h=1}{m} \sumss{q=1}{5} \az{hq} \QNh{n}{hq} 
      = \sumss{h=1}{m} \sumss{q=1}{5} \az{hq} 
      \left[ \sumss{t=1}{n}
      \ZNht{n}{hq}{t} \right]
      =  \sumss{t=1}{n} \left[ \sumss{h=1}{m} \sumss{q=1}{5} \az{hq} 
      \ZNht{n}{hq}{t} \right] 
      =  \sumss{t=1}{n} \ZNMt{n}{m}{t}(\bm{a}).}
    \end{align}
  \end{subequations}
    A large block - small block argument can be used to analyse this,
  i.e.\ the index set \mbox{$\parenC{1,\dotsc,n}$} will be
  partitioned into large blocks and small blocks, such that
  $\QNM{n}{m}(\bm{a})$ can be expressed as the sum of $\SNp{n}{1}$,
  $\SNp{n}{2}$ and $\SNp{n}{3}$ (to be defined below).  The asymptotic
  distribution of $\QNM{n}{m}(\bm{a})$ will be shown to
  coincide with the asymptotic distribution of $\SNp{n}{1}$, the
  summands of $\SNp{n}{1}$ will be shown to be asymptotically
  independent, and finally the Lindeberg conditions for asymptotic
  normality of $\SNp{n}{1}$ will be verified.

  Use~$\ell$,~$r$, and~$s$ from
  \myref{th:block_sizes_for_main_result}{th:block_sizes_for_main_result_r_ell_limits}
  to divide the indexing set \mbox{$\parenC{1,\dotsc,n}$} into
  \mbox{$2\ell + 1$} subsets of large blocks and small blocks (and one
  reminder block), defined as follows
  \begin{subequations}
    \label{eq:large_small_remainder_blocks}
    \begin{align}
      \label{eq:large_blocks}
      \mathcalAz{j}       &\defeq \parenC{(j-1)\left(r + s\right) + 1,
        \dotsc, (j-1)\left(r + s\right) + r }, \text{ for
        }j = 1,\dotsc, \ell, \\
              \label{eq:small_blocks}
      \mathcalBz{j} 
      &\defeq \parenC{(j-1)\left(r + s\right) + r +
        1, \dotsc, j\left(r + s\right)  }, \text{ for
        }j = 1,\dotsc, \ell, \\
      \label{eq:reminder_block}
      \mathcalCz{\ell} 
      &\defeq
        \begin{cases}
          \parenC{\ell\left(r + s\right) + 1, \dotsc, n }
          &\text{when } \ell\left(r + s\right) < n, \\
          \emptyset 
          &\text{when } \ell\left(r + s\right) = n.
        \end{cases}
    \end{align}
  \end{subequations}
  In order to avoid   iterated sums later on, introduce the following unions,
  \begin{subequations}
    \begin{align}
      \mathcal{A}^{\circ} 
      \defeq \cupss{j = 1}{\ell}  \mathcalAz{j}, \qquad
      \mathcal{B}^{\circ} 
      \defeq          \cupss{j = 1}{\ell}  \mathcalBz{j}.
    \end{align}
  \end{subequations}
  Note that the number of elements in $\mathcal{A}^{\circ}$ and
  $\mathcal{B}^{\circ}$ will be \mbox{$\ell r$} and \mbox{$\ell s$}
  respectively.  The number of elements in $\mathcalCz{\ell}$ will
  be \mbox{$n-\ell(r+s)$}, and this can vary between~0 and
  \mbox{$r+s-1<2r$}.

  Use these subsets of \mbox{$\parenC{1,\dotsc,n}$ }to define
  the following variables,            \begin{subequations}
    \label{eq:eta_xi_zeta}
    \begin{align}
      \label{eq:eta}
      \etaz{j} 
      &\defeq \sumss{t \in \mathcalAz{j}}{}
        \ZNMt{n}{m}{t}(\bm{a}), \text{ for }j = 1,\dotsc, \ell,
      &\SNp{n}{1} 
      &\defeq \sumss{j=1}{\ell} \etaz{j}         = \sumss{\scriptscriptstyle t\in\mathcal{A}^{\circ}}{}
        \ZNMt{n}{m}{t}(\bm{a}), \\
              \label{eq:xi}
      \xiz{j}
      &\defeq \sumss{t \in \mathcalBz{j}}{}
        \ZNMt{n}{m}{t}(\bm{a}), \text{ for }j = 1,\dotsc, \ell,
      &\SNp{n}{2} 
      &\defeq \sumss{j=1}{\ell} \xiz{j}        = \sumss{\scriptscriptstyle t\in\mathcal{B}^{\circ}}{}
        \ZNMt{n}{m}{t}(\bm{a}), \\
              \label{eq:zeta}
      \zetaz{\ell} 
      &\defeq \sumss{t \in \mathcalCz{\ell}}{}
        \ZNMt{n}{m}{t}(\bm{a}), 
      &\SNp{n}{3} 
      &\defeq \zetaz{\ell},
    \end{align}
  \end{subequations}
  such that 
  \begin{align}
    \label{eq:QNM_rewritten}
    \nz[-1/2]{}\, \QNM{n}{m}(\bm{a}) =     \nz[-1/2]{} \parenC{\SNp{n}{1} + \SNp{n}{2} + \SNp{n}{3}  }.
  \end{align}
  
  The expectation of these quantities are by construction equal to
  zero, which gives 
  \begin{align}
    \label{eq:var_QNM_S}
    \Var{\nz[-1/2]{}\,\QNM{n}{m}(\bm{a})} 
    &= \frac{1}{n} \E{\QNM{n}{m}(\bm{a}) \cdot
      \QNM{n}{m}(\bm{a}) }  
      = \frac{1}{n} \sumss{p=1}{3} \sumss{q=1}{3}       \E{\SNp{n}{p} \cdot \SNp{n}{q}}.
  \end{align}
  
  When \mbox{$p\neq q$}, there will be no overlap between the indexing
  sets that occur in the two sums, and the following inequality, here
  illustrated by the case \mbox{$p=1$} and~\mbox{$q=2$}, is obtained
  \begin{subequations}
    \label{eq:SNp_SNq_and_p<>q_negligible}
    \begin{align}
      \absp{\frac{1}{n} \E{\SNp{n}{1} \cdot \SNp{n}{2} } }
      &=         \absp{\frac{1}{n} \E{        \left(\sumss{\scriptscriptstyle i\in\mathcal{A}^{\circ}}{}
        \ZNMt{n}{m}{i}(\bm{a}) \right)         \cdot \left( \sumss{\scriptscriptstyle
        k\in\mathcal{B}^{\circ}}{} \ZNMt{n}{m}{k}(\bm{a})
        \right) }}\\
          &\leq         \frac{1}{n}         \sumss{\scriptscriptstyle i\in\mathcal{A}^{\circ}}{}         \sumss{\scriptscriptstyle k\in\mathcal{B}^{\circ}}{}         \absp{\E{        \ZNMt{n}{m}{i}(\bm{a}) \cdot         \ZNMt{n}{m}{k}(\bm{a}) }}\\
          &\leq         \frac{1}{n} \sumss{\stackrel{i,k = 1}{\scriptscriptstyle i\neq k}}{n}         \absp{\E{\ZNMt{n}{m}{i}(\bm{a}) \cdot         \ZNMt{n}{m}{k}(\bm{a}) }}.
    \end{align}
  \end{subequations}
  \Cref{th:convergent_covariance_ZNMt} thus gives that the expectation
  of all the cross-terms are asymptotically negligible.

  For the case \mbox{$p=q=2$}, i.e.\ the small blocks, the same
  strategy as in \cref{eq:SNp_SNq_and_p<>q_negligible} shows that the
  internal cross-terms are asymptotically negligible.
  \Myref{th:var(ZNMt)}{th:var(ZNMt)=O(m)} states that the remaining
  summands all are $\Oh{m}$, which results in the following bound
  \begin{subequations}
    \label{eq:SN2_SN2_negligible}
    \begin{align}
      \frac{1}{n} \E{ \SNp{n}{2} \cdot \SNp{n}{2} }       &=         \frac{1}{n} \sumss{\scriptscriptstyle i, k \in
        \mathcal{B}^{\circ}}{} \E{\ZNMt{n}{m}{i}(\bm{a}) \cdot
        \ZNMt{n}{m}{k}(\bm{a}) }\\             &=         \frac{1}{n} \sumss{\scriptscriptstyle i \in \mathcal{B}^{\circ}}{}
        \E{\ZNMt{n}{m}{i}(\bm{a}) \cdot
        \ZNMt{n}{m}{i}(\bm{a}) } +
        \frac{1}{n} \sumss{\stackrel{i, k \in
        \mathcal{B}^{\circ}}{\scriptscriptstyle i\neq k}}{}
        \E{\ZNMt{n}{m}{i}(\bm{a}) \cdot
        \ZNMt{n}{m}{k}(\bm{a}) }\\             &=         \frac{1}{n} \sumss{\scriptscriptstyle i \in
        \mathcal{B}^{\circ}}{} \Oh{m} + \oh{1}\\             &=         \Oh{\frac{m\ell s}{n}}.
    \end{align}
  \end{subequations}

  For the case \mbox{$p=q=3$}, i.e.\ the residual block, a similar
  argument gives
  \begin{align}
    \label{eq:SN3_SN3_negligible}
    \frac{1}{n} \E{\SNp{n}{3} \cdot \SNp{n}{3} }     = \Oh{\frac{m\left(n - \ell(r +s)\right)}{n}} 
    < \Oh{\frac{mr}{n}}. 
  \end{align}

  \Myref{th:block_sizes_for_main_result}{th:block_sizes_for_main_result_r_ell_limits}
  ensures that $(m\ell s)/n$ and $mr/n$ goes to zero, so the terms
  investigated in \cref{eq:SN2_SN2_negligible} and
  \cref{eq:SN3_SN3_negligible} are asymptotically~negligible.  This
  implies that   \mbox{$\nz[-1/2]{}\!\left(\QNM{n}{m}(\bm{a}) -
      \SNp{n}{1}\right) \Rightarrow 0$}, and
  \citet[Theorem~25.4]{Billingsley12:_probab_measur} states that there
  thus is a common limiting distribution
  for~$\nz[-1/2]{}\,\QNM{n}{m}(\bm{a})$ and~$\nz[-1/2]{}\,\SNp{n}{1}$.

  The arguments used for $\SNp{n}{2}$ also gives the simple
  observation below, which is needed later on,
  \begin{align}
    \label{eq:SN1_SN1_expression}
    \Var{ \nz[-1/2]{}\,\SNp{n}{1} }     = \frac{1}{n} \sumss{j=1}{\ell} \Var{\etaz{j}} + \oh{1}.   \end{align}

  The next step is to show that the random variables $\etaz{j}$ are
  asymptotically independent, which formulated relative to the
  characteristic functions corresponds to~showing
  \begin{align}
    \label{eq:characteristic_functions_SN1_eta}
    \absp{\E{\exp\!\left(it\SNp{n}{1}\right)} -     \prodss{j=1}{\ell} \E{\exp\left(it\etaz{j}\right)}}
    \rightarrow 0.
  \end{align}
  The validity of this statement follows from Lemma~1.1 in
  \citet[p.~180]{Volkonskii:1959:SLT}, by introducing random variables
  \mbox{$\Vz{\!j} = \exp\!\left(it\etaz{j}\right)$}, for
  \mbox{$j=1,\dotsc,\ell$}.  By construction, the $\Vz{\!j}$ trivially
  satisfies the requirement \mbox{$\absp{\Vz{\!j}}\leq 1$}, so it only
  remains to identify the corresponding \mbox{$\sigma$-algebras} and
  the distance between them.  From the definitions of $\etaz{j}$,
  $\mathcalAz{j}$ and $\ZNMt{n}{m}{t}(\bm{a})$, it is easy to see that
  \mbox{$\Vz{\!j} \in \sigmaYhtRange{(j-1)(r+s)+1}{(j-1)(r+s)+r+m}$},
  and from this it follows that the distance between the highest index
  in the \mbox{$\sigma$-algebra} corresponding to $\Vz{\!j}$ and the
  lowest index in the \mbox{$\sigma$-algebra} corresponding to
  $\Vz{\!j+1}$, is given by
  \begin{align}
    \vartheta 
    &\defeq \parenC{((j+1)-1)(r+s)+1} -       \parenC{(j-1)(r+s)+r+m}       = s - m + 1.
  \end{align}
  
  \Myref{assumption_Nmb}{eq:assumption_m=o(s)}, i.e.\
  \mbox{$m=\oh{s}$}, ensures that there (asymptotically) will be no
  overlap between these \mbox{$\sigma$-algebras}, and the result from
  \citet{Volkonskii:1959:SLT} thus gives
  \mbox{$16(\ell-1)\alpha(\vartheta)$} as an upper bound on the left
  side of \cref{eq:characteristic_functions_SN1_eta}.
      \Myref{th:block_sizes_for_main_result}{th:block_sizes_for_main_result_r_ell_limits}
  says that this bound goes   to zero, which shows that the $\etaz{j}$ are
  asymptotically~independent.

  It remains to verify the Lindeberg condition, for which an
  expression for
  \mbox{$\mathfraksz[2]{\ell}
    \defeq \sumss{j=1}{\ell} \Var{\etaz{j}}$} is needed.
  From \myref{assumption_Yt}{assumption_Yt_strictly_stationary} and 
  \myref{th:eta_variance_expressions}{th:eta_variance_expressions_not_truncated},
  it follows that
  \begin{align}
    \label{eq:Lindeberg_sn_squared}
    \mathfraksz[2]{\ell}     =  \sumss{j=1}{\ell} \Var{\etaz{j}}     = \ell \cdot \Var{\etaz{1}}     = \ell \cdot r \cdot \parenC{\sumss{h=1}{m} \bmaz[\prime]{h}
    \WhN{}{h}{\bm{b}} \bmaz{h} + \oh{1} }, 
  \end{align}
    and assuming
  \mbox{$\mathfraksz[2]{\ell}>0$}, the condition to verify is
  \begin{align}
    \label{eq:Lindeberg_condition}
    \forall\ \epsilon > 0\qquad     \lim_{n\rightarrow\infty} \sumss{j=1}{\ell} \frac{1}{    \mathfraksz[2]{\ell}}     \E{ \etaz[2]{j}     \cdot
    \Ind{\absp{\etaz{j}} \geq \epsilon\sqrt{    \mathfraksz[2]{\ell}
    }}}     \longrightarrow 0.
  \end{align}
    This holds trivially if the sets occurring in the indicator
  functions, i.e.\
  \mbox{$\TSR{\absp{\etaz{j}} \geq \epsilon\sqrt{        \mathfraksz[2]{\ell}}}{}{}$}, becomes empty when $n$ is large
  enough.  It is thus of interest to see if an upper bound
  for~\mbox{$\absp{\etaz{j}}$} can be found, and if the
  limit of this upper bound becomes smaller than the limit of the
  right-hand side~\mbox{$\epsilon\sqrt{      \mathfraksz[2]{\ell}}$}.

  Keeping in mind the definitions
  of~\mbox{$\XNht{n}{hq}{t}$},~\mbox{$\ZNht{n}{hq}{t}$}
  and~\mbox{$\etaz{j}$}, see
  \cref{eq:definition_of_XNht*,eq:definition_ZNht,eq:eta}, it is clear
  that an upper bound for $\absp{\etaz{j}}$ might be
  deduced~from,
  \begin{subequations}
    \label{eq:eta_bound}
    \begin{align}
      \absp{\etaz{j}} 
      &=         \absp{\sumss{ t\in\mathcalAz{j}}{}         \sumss{h=1}{m} \sumss{q=1}{5} \az{hq} \ZNht{n}{hq}{t}}         \leq         \sumss{ t\in\mathcalAz{j}}{}         \sumss{h=1}{m} \sumss{q=1}{5} \absp{\az{hq}} \absp{\ZNht{n}{hq}{t}}, \\
              \absp{\ZNht{n}{hq}{t}} 
      &=         \absp{\XNht{n}{hq}{t} - \E{\XNht{n}{hq}{t}}}         \leq         \absp{\XNht{n}{hq}{t}} + \Oh{\sqrt{\bz{1}\bz{2}}}, \\
              \absp{\XNht{n}{hq}{t}} 
      &=
        \absp{\sqrt{\bz{1}\bz{2}}\cdot \frac{1}{\bz{1}\bz{2}}         \Kh[\frac{\Yz{t+h}-\vz{1}}{\bz{1}},         \frac{\Yz{t}-\vz{2}}{\bz{2}}]{h}         \Uh[\Yht{h}{t}]{h:\bm{b}} }.
    \end{align}
  \end{subequations}

  If all of the functions~$\Uh[\bm{w}]{hq:\bm{b}}$ are bounded, or if
  the kernel functions~$\Khb[\bm{w}-\LGp]{h}{\bm{b}}$ have bounded
  support, then the present framework will be sufficient to reach the
  desired conclusion.  However, no such conditions are assumed, and a
  truncation argument must thus be introduced in order to deal with
  this~problem --- in particular, the expression
  \mbox{$\QNM{n}{m}(\bm{a}) = \QNM[\leq L]{n}{m}(\bm{a}) +
    \QNM[>L]{n}{m}(\bm{a})$} will be~used.

  \Myref{th:Uh_finite}{th:Uh_finite_sup} implies that a large enough
  value for the threshold~$L$ will ensure that 
        all constructions and arguments based upon the ordinary
  functions~$\Uh[\bm{w}]{hq:\bm{b}}$ also works nicely for the
  truncated functions $\UhL[\bm{w}]{hq:\bm{b}}{\leq L}$
  and~$\UhL[\bm{w}]{hq:\bm{b}}{>L}$.
    With regard to the limiting distributions, first note that
  $\nz[-1/2]{}\,\QNM[>L]{n}{m}(\bm{a})$ and
  $\nz[-1/2]{}\,\power[|>L]{}{\Sz[(1)]{n}}$ shares the
  same limiting distribution, and then observe that the upper
  truncated versions of
  \cref{eq:SN1_SN1_expression,eq:Lindeberg_sn_squared} together with
  the result from
  \myref{th:eta_variance_expressions}{th:eta_variance_expressions_truncated},
  gives the following bound when $L$ is large enough:
  \begin{align}
    \Var{ \nz[-1/2]{}\,\power[|>L]{}{\Sz[(1)]{n}}}         = \frac{1}{n} \sumss{j=1}{\ell} \Var{\etaz[>L]{j}} +
    \oh{1}     = \frac{\ell r}{n} \cdot \oh{1}.
  \end{align}
  Since \mbox{$\ell r \asymp n$}, it follows that
  \mbox{$\nz[-1/2]{}\,\QNM[>L]{n}{m}(\bm{a})
    \Rightarrow 0$}, so the limiting distributions of
  $\nz[-1/2]{}\,\QNM{n}{m}(\bm{a})$ and
  $\nz[-1/2]{}\,\QNM[\leq L]{n}{m}(\bm{a})$
  coincide when $L$ is large~enough.\footnote{    Truncation arguments often requires the threshold value~$L$ to go
    to~$\infty$ in order for a conclusion to be obtained for the
    original expression, but this is not required for the present case
    under investigation (due to \cref{th:Uh_finite}).}     Next, observe that the random variable $\absp{\etaz[\leq L]{j}}$
  obviously will have an upper bound, since the truncated polynomial
  $\UhL[\bm{w}]{hq:\bm{b}}{\leq L}$ will occur in the lower truncated
  version of \cref{eq:eta_bound}.
    Since the kernel function $K(\bm{w})$ by definition is bounded by
  some constant $\mathcal{K}$, it follows
  that~$\absp{\etaz[\leq L]{j}}$ is bounded~by
  \begin{align}
    \label{eq:upper_bound_eta_truncated}
    \absp{\etaz[\leq L]{j}} 
    &\leq       5rm \left(\max \absp{\az{hq}}\right)       \left( \frac{\mathcal{K}}{\sqrt{\bz{1}\bz{2}}}       L + \Oh{\sqrt{\bz{1}\bz{2}}} \right)       < \mathcal{C}L \frac{r m}{\sqrt{\bz{1}\bz{2}}},
  \end{align}
  where $\mathcal{C}$ is a constant that is independent of the
  index~$j$.

  It remains to verify that the indicator functions
  $\Ind{\absp{\etaz[\leq L]{j}} \geq \epsilon\sqrt{      \left(\mathfraksz[2]{\ell}\right)^{{}_{\leq L}}}}$, from the lower truncated version
  of \cref{eq:Lindeberg_condition}, becomes zero when $\nlimit$, which
  can be done by checking that the upper bound of
  $\absp{\etaz[\leq L]{j}}$ from
  \cref{eq:upper_bound_eta_truncated} in the limit gives a smaller
  value than the lower truncated version of
  \mbox{$\left(\mathfraksz[2]{\ell}\right)^{{}_{\leq L}}$} from
  \cref{eq:Lindeberg_sn_squared}.  
        This in turn can be done by dividing both of them with
  $\sqrt{\ell r m}$, and then compare their limits.  Assuming that the
  threshold value~$L$ is high enough to allow
  \myref{th:eta_variance_expressions}{th:eta_variance_expressions_truncated}
  to be used, i.e.\ that
  \mbox{$\left(\mathfraksz[2]{\ell}\right)^{{}_{\leq L}}$} and
  $\mathfraksz[2]{\ell}$ share the same asymptotic expression, this
  becomes,
  \begin{subequations}
    \label{eq:lhs_and_rhs_indicator_function_truncated}
    \begin{align}
      \label{eq:lhs_indicator_function_truncated}
      \frac{\absp{\etaz[\leq L]{j}}}{\sqrt{\ell r m}} 
      &\leq         \mathcal{C}L
        \sqrt{\frac{mr}{\ell\bz{1}\bz{2}}}         \longrightarrow 0,         \qquad 
        \text{due to  \myref{th:block_sizes_for_main_result}{th:block_sizes_for_main_result_r_ell_limits},
        } \\
              \label{eq:rhs_indicator_function_truncated}
      \frac{\epsilon\sqrt{      \left(\mathfraksz[2]{\ell}\right)^{{}_{\leq L}}}}{\sqrt{\ell r m}}   
      &=         \epsilon \cdot \sqrt{\frac{1}{m} \parenC{\sumss{h=1}{m}
        \bmaz[\prime]{h} \WhN{}{h}{\bm{b}} \bmaz{h} + \oh{1} }}                  \asymp         \epsilon \cdot \sqrt{\frac{1}{m} \sumss{h=1}{m} \bmaz[\prime]{h} \WhN{}{h}{\bm{b}} \bmaz{h}}.      \end{align}
  \end{subequations}
    \Myref{assumption_score_function}{assumption_score_function_limit_h}
    ensures that $\WhN{}{h}{\bm{b}}$
  (from \cref{def:matrices_related_to_the_new_penalty_function})
  converges to some non-zero matrix (as $\hlimit$ and $\blimit$), and
  this implies that the limit of
  \mbox{$\frac{1}{m} \sumss{h=1}{m} \bmaz[\prime]{h} \WhN{}{h}{\bm{b}}
    \bmaz{h}$} in \cref{eq:rhs_indicator_function_truncated} will be
  nonzero, from which it follows that the indicator function in
  \cref{eq:Lindeberg_condition} becomes zero in the limit, i.e.\ that
  the Lindeberg condition is satisfied.

  This implies that 
  \begin{align}
    \label{eq:asymptotic_normality_for_truncated_version}
    \frac{\sumss{j=1}{\ell} \etaz[\leq L]{j}}{    \sqrt{\mathfraksz[2]{\ell}}} \longrightarrow N(0,1),
  \end{align}
  which due to~\mbox{$\ell r \asymp n$} can be re-expressed as
  \begin{align}
    \label{eq:asymptotic_normality_for_truncated_version_cleaned}
    \nz[-1/2]{} \sumss{j=1}{\ell} \etaz[\leq L]{j}    \longrightarrow \UVN{0}{\sumss{h=1}{m} \bmaz[\prime]{h}
        \WhN{}{h}{\bm{b}} \bmaz{h}}.
  \end{align}
    The proof of \cref{th:AN_of_QNM} is now complete, since the four
  random variables $\nz[-1/2]{}\QNM{n}{m}(\bm{a})$,
  $\nz[-1/2]{}\QNM[\leq L]{n}{m}(\bm{a})$, \mbox{$\nz[-1/2]{}
    \!\left(\SNp{n}{1}\right)^{{}_{\leq L}}$} and \mbox{$\nz[-1/2]{}
    \sumss{j=1}{\ell} \etaz[\leq L]{j}$} all share the same limiting
  distribution (when $L$ is large enough).

  The proof of \cref{th:AN_of_QNh_vector} follows from the
  Cram\'{e}r-Wold theorem.
\end{proof}

  \label{note:interpretation_large_m}
  The statements in \cref{th:AN_of_QNh_and_QNM} has to be interpreted
  as an approximate asymptotic distributions valid for large $m$ and
  $n$ and small $\bm{b}$.  One part of the \enquote{asymptotic
    problem} is the interpretation of an infinite-variate Gaussian
  distribution, but the main problem is the occurrence of the kernel
  function $K(\bm{w})$, which in the limit gives a degenerate Gaussian
  distribution in \myref{th:AN_of_QNh_and_QNM}{th:AN_of_QNh_vector}.
  This degeneracy in itself would not have been any issue if the
  target of interest had been the asymptotic behaviour of
  $\nz[-1/2]{}\,\QNhvec{n}{\overbar{m}}$, but it requires some
  additional rescaling before the Klimko-Nelson approach in
  \cref{th:Klimko_Nelson_1978} can be used to investigate the
  asymptotic properties of the estimates~$\hatthetaMN{m}{n}$, see
  \cref{app:asymptotics_for_theta} for details.

\begin{corollary}
  \label{th:A4_requirement_final_case}
  Given the same assumptions as in \cref{th:AN_of_QNh_and_QNM}, the
  following asymptotic result holds true
  \begin{align}
    \nz[-1/2]{}\sqrt{\bz{1}\bz{2}} \nablaM{m}\QMN[\thetaMb{m}{b}]{m}{n}
    \stackrel{\scriptscriptstyle d}{\longrightarrow}     \UVN{\bm{0}}{\oplusss{h=1}{m} \WhN{}{h}{\bm{b}}},
  \end{align}
  i.e.\ asymptotically \mbox{$5m$-variate} normal.
\end{corollary}

\begin{proof}
  \Cref{eq:connection_RV_nabla_QMN} states that
  $\QNhvec{n}{\overbar{m}}$ and
  $\sqrt{\bz{1}\bz{2}}\nablaM{m}\QMN[\thetaMb{m}{b}]{m}{n}$ have the
  same limiting distribution, and the result thus follows from
  \myref{th:AN_of_QNh_and_QNM}{th:AN_of_QNh_vector}.
\end{proof}

\makeatletter{}
\subsection{The asymptotic results for
  $\widehatbmthetaz{\LGp|\overbar{m}|\bm{b}}$}
\label{app:asymptotics_for_theta}

The final details needed for the investigation of the asymptotic
properties of $\hatlgsdM[p]{\LGp}{\omega}{m}$ will now be presented.
(Confer \cref{note:interpretation_large_m} for an
interpretation of the $m$ that occurs in the limiting distribution.)

\begin{theorem}
  \label{th:asymptotics_for_hatlgtheta}
  Under the same assumptions as in \cref{th:AN_of_QNh_and_QNM}, the
  estimated parameter vector
  $\widehatbmthetaz{\LGp|\overbar{m}|\bm{b}}$ converges towards the
  true parameter vector $\bmthetaz{\LGp|\overbar{m}}$ in the
  following manner.
  \begin{align}
    \label{eq:th:asymptotics_for_hatlgtheta}
    \sqrt{n \!\parenRz[3]{}{\bz{1}\bz{2}}} \cdot     \parenR{\widehatbmthetaz{\LGp|\overbar{m}|\bm{b}} -
    \bmthetaz{\LGp|\overbar{m}}}     \stackrel{\scriptscriptstyle d}{\longrightarrow}
    \UVN{\bm{0}}{\Sigmaz{\LGp|\overbar{m}}},
  \end{align}
  where
  \mbox{$\Sigmaz{\LGp|\overbar{m}} \defeq
    \oplusss{h=1}{m}\Sigmaz{\LGp|h}$}, i.e.\
  $\Sigmaz{\LGp|\overbar{m}}$ is the direct sum of the covariance
  matrices $\Sigmaz{\LGp|h}$ that corresponds to
  \mbox{$\sqrt{n \!\parenRz[3]{}{\bz{1}\bz{2}}} \cdot
    \parenR{\widehatbmthetaz{\LGp|h|\bm{b}}-\bmthetaz{\LGp|h}}$}.
\end{theorem}

\begin{proof}
  Under the given assumptions, \cref{th:A4_requirement_final_case}
  states that the fourth requirement of \cref{th:Klimko_Nelson_1978}
  (the Klimko-Nelson approach) holds true for the local penalty
  function $\QMN[\bmthetaz{\LGp|\overbar{m}|\bm{b}}]{m}{n}$ in the
  general case where $\mlimit$ and $\blimit$ when $\nlimit$.  The
  three remaining requirements holds true by the same arguments that
  was used in \cref{app:new_penalty_function}, so the Klimko-Nelson
  approach can be used to obtain an asymptotic result for the
  difference of the estimate
  $\widehatbmthetaz{\LGp|\overbar{m}|\bm{b}}$ and the true parameter
  $\bmthetaz{\LGp|\overbar{m}}$.

  As in \citet{Tjostheim201333}, it will be instructive to first
  consider the simpler case where $m$ and $\bm{b}$ were fixed.  In
  this case, the asymptotic result obtained from
  \cref{th:Klimko_Nelson_1978} takes the form,
  \begin{align}
    \label{eq:th:asymptotics_for_hatlgtheta_original}
    \sqrt{n} \cdot     \left(\widehatbmthetaz{\LGp|\overbar{m}|\bm{b}} -
    \bmthetaz{\LGp|\overbar{m}} \right)     \stackrel{\scriptscriptstyle d}{\longrightarrow}
    \UVN{\bm{0}}{    \Sigmaz{\LGp|\overbar{m}}},
  \end{align}
  with
  \mbox{$\Sigmaz{\LGp|\overbar{m}} \defeq
    \Vz[-1]{\LGp|\overbar{m}} \Wz{\LGp|\overbar{m}}
    \Vz[-1]{\LGp|\overbar{m}}$}, where the \mbox{$5m\times 5m$}
  matrices $\Vz{\LGp|\overbar{m}}$ and $\Wz{\LGp|\overbar{m}}$ can
  be represented as
  \begin{align}
    \label{eq:Vn_Wm_matrices_lgsd}
    \Vz{\LGp|\overbar{m}}
    = \oplusss{h=1}{m}     \Vz{\LGp|h}, \qquad
    \Wz{\LGp|\overbar{m}}
    = \oplusss{h=1}{m} \Wz{\LGp|h},
  \end{align}
  i.e.\ they are the direct sums of the \mbox{$5\times 5$} matrices
  $\Vz{\LGp|h}$ and $\Wz{\LGp|h}$ that corresponds to the bivariate
  penalty functions used for the investigation of the parameter
  vectors $\bmthetaz{\LGp|h|\bm{b}}$.

  Since $\Vz{\LGp|\overbar{m}}$ is the direct sum of the invertible
  matrices $\Vz{\LGp|h}$, it follows that
  $\Vz[-1]{\LGp|\overbar{m}}$ is the direct sum of
  $\Vz[-1]{\LGp|h}$ (see e.g.\ \citet[p.31]{Horn:2012:MA:2422911}).
    This implies that the matrix of interest can be expressed as
  \mbox{$\Sigmaz{\LGp|\overbar{m}} =
    \oplusss{h=1}{m}\Sigmaz{\LGp|h}$}, where
  \mbox{$\Sigmaz{\LGp|h} \defeq
    \Vz[-1]{\LGp|h}\Wz{\LGp|h}\Vz[-1]{\LGp|h}$} are the
  covariance matrices that corresponds to
  \mbox{$\sqrt{n} \cdot
    \parenR{\widehatbmthetaz{\LGp|h|\bm{b}}-\bmthetaz{\LGp|h}}$},
  i.e.\ a bivariate result like the one in
  \citet[Th.~1]{Tjostheim201333}.

  For the general situation, when $\mlimit$ and $\blimit$ when
  $\nlimit$, it is necessary with an additional scaling in order to
  get a covariance matrix with finite entries.  Obviously, a factor
  $\sqrt{\bz{1}\bz{2}}$ must be included in order to balance the
  effect of the kernel function $\Khb{h}{\bm{b}}$.
  
  Moreover, since the limiting matrices of $\Vz{\LGp|h}$ and
  $\Wz{\LGp|h}$ turns out to have rank one, an additional scaling is
  required in order to obtain a covariance matrix with finite entries.
  This case is treated in \citet[Th.~3]{Tjostheim201333}, from which
  it follows that the scaling factor must be
  $\sqrt{\parenRz[3]{}{\bz{1}\bz{2}}}$.
\end{proof}

\subsection{An extension to two different points, i.e.\ both
  $\LGp$ and $\LGpd$}
\label{app:theta_extension_different_local_points}
The previous analysis was restricted to the case where one point was
used throughout, which is sufficient for the investigation of the
asymptotic properties of the \mbox{$m$-truncated} estimates
$\hatlgsdM[p]{\LGp}{\omega}{m}$ for a point $\LGp$ that lies upon the
diagonal (see \cref{th:asymptotics_for_hatlgsd}) or for general points
\mbox{$\LGp\in\RRn{2}$} when the time series under investigation is
time reversible (see \cref{th:asymptotics_for_hatlgsd_reversible}).

An investigation of the \mbox{$m$-truncated} estimates
$\hatlgsdM[p]{\LGp}{\omega}{m}$ for points~\mbox{$\LGp=\LGpoint$} that
lies off the diagonal, i.e.\ \mbox{$\LGpnotdiagonal$}, requires some
minor modifications of the setup leading to
\cref{th:asymptotics_for_hatlgtheta}, as discussed in
the proof of the following theorem.

\begin{theorem}
  \label{th:asymptotics_for_hatlgtheta_LGp_and_LGpd}
  Consider the same setup as in \cref{th:AN_of_QNh_and_QNM}, but with
  the modification that the point
  \mbox{$\LGp=\LGpoint$} lies off the diagonal, and
  with the added requirement that the bivariate densities
  $\gh[\yh{h}]{h}$ does not possess diagonal symmetry.  
    With \mbox{$\LGpd=\parenR{\LGpi{2},\LGpi{1}}$} the diagonal
  reflection of $\LGp$, the two parameter vectors
  $\widehatbmthetaz{\LGp|\overbar{m}|\bm{b}}$ and
  $\widehatbmthetaz{\LGpd|\overbar{m}|\bm{b}}$ can be combined to a
  vector
  \mbox{$\widehatbmThetaz{\OUbar{m}|\bm{b}}\!\parenR{\LGp,\LGpd} =
    \Vector[\prime]{\widehatbmthetaz[\prime]{\LGp|\overbar{m}|\bm{b}},
      \widehatbmthetaz[\prime]{\LGpd|\overbar{m}|\bm{b}}}$},
  possessing the following asymptotic behaviour.
    \begin{align}
    \label{eq:th:asymptotics_for_hatlgtheta_LGp_and_LGpd}
    \sqrt{n \!\parenRz[3]{}{\bz{1}\bz{2}}} \cdot     \left(\widehatbmThetaz{\OUbar{m}|\bm{b}}\!\parenR{\LGp,\LGpd} -
    \bmThetaz{\OUbar{m}}\!\parenR{\LGp,\LGpd} \right)     \stackrel{\scriptscriptstyle d}{\longrightarrow}
    \UVN{\bm{0}}{    \begin{bmatrix}
      \Sigmaz{\LGp|\overbar{m}} &0 \\
      0 &\Sigmaz{\LGpd|\overbar{m}}
    \end{bmatrix}},
  \end{align}
  where the matrices $\Sigmaz{\LGp|\overbar{m}}$ and
  $\Sigmaz{\LGpd|\overbar{m}}$ are as given in
  \cref{th:asymptotics_for_hatlgtheta}.
\end{theorem}

\begin{proof}
  This result follows when the Klimko-Nelson approach is used upon the
  local penalty-function
  \begin{align}
    \label{eq:combined_penalty_function_LGp_and_LGpd}
    \Qz{\OUbar{m}:n}\!\parenR{\bmThetaz{\OUbar{m}|\bm{b}}\!\parenR{\LGp,\LGpd}}
    \defeq \QMN[\bmthetaz{\LGp|\overbar{m}|\bm{b}}]{m}{n} +
    \QMN[\bmthetaz{\LGpd|\overbar{m}|\bm{b}}]{m}{n},
  \end{align}
        i.e.\ the four requirements in
  \cref{th:Klimko_Nelson_1978_A1,th:Klimko_Nelson_1978_A2,th:Klimko_Nelson_1978_A3,th:Klimko_Nelson_1978_A4}
  of \cref{th:Klimko_Nelson_1978} must be verified for this new
  penalty function.  The function $\QMN{m}{n}$ on the right side of
  \cref{eq:combined_penalty_function_LGp_and_LGpd} is the penalty
  function encountered in the investigation of
  $\bmthetaz{\LGp|\overbar{m}|\bm{b}}$, i.e.\ the same observations
  $\TSR{\Yz{t}}{t=1}{n}$ occurs in both the first and second term, but
  the point of interest will be $\LGp$ in the first one and
  $\LGpd$ in the second one.  

  The requirement that $\LGp$ lies off the diagonal together with the
  requirement that none of the bivariate densities $\gh[\yh{h}]{h}$
  possess diagonal symmetry implies that different approximating local
  Gaussian densities occurs for the different points and different
  lags, so it can be assumed that there is no common parameters in
  $\bmthetaz{\LGp|\overbar{m}|\bm{b}}$ and
  $\bmthetaz{\LGpd|\overbar{m}|\bm{b}}$.  
                                        This implies that the arguments used to verify the three first
  requirements of \cref{th:Klimko_Nelson_1978} for the penalty
  function $\QMN{m}{n}$ (see
  \cref{Res:QMN_A1_requirement,Res:QMN_A2_requirement,Res:QMN_A3_requirement}),
  also will work upon the combined penalty function
    $\Qz{\OUbar{m}:n}$, and it will in particular be the case that the
  Hessian matrix $\Vz{\OUbar{m}|\bm{b}:n}$ occurring in
  \cref{Res:QMN_A2_requirement} can be written as the direct sum of
  the matrices that corresponds to
  $\QMN[\bmthetaz{\LGp|\overbar{m}|\bm{b}}]{m}{n}$ and
  $\QMN[\bmthetaz{\LGpd|\overbar{m}|\bm{b}}]{m}{n}$, i.e.\
  \mbox{$\Vz{\OUbar{m}|\bm{b}}\!\parenR{\LGp,\LGpd} =
    \VMbN{m}{\bm{b}}{n}\!\parenR{\LGp} \oplus
    \VMbN{m}{\bm{b}}{n}\!\parenR{\LGpd}$}, where the points of
  interest have been included in the notation to keep track of the
  components.
  
    The investigation of the fourth requirement of the Klimko-Nelson
  approach for the penalty function $\QMN{\Ubar{m}}{n}$ requires some minor
  modifications of the constructions that was encountered in
  \cref{app:final_building_blocks}.  Both
  $\XNht{n}{hq}{t}\!\left(\LGp\right)$ and
  $\XNht{n}{hq}{t}\!\left(\LGpd\right)$ (for \mbox{$h=1,\dotsc,m$} and
  \mbox{$q=1,\dotsc,5$}) are needed, and the final random variable
  will include both $\LGp$ and $\LGpd$ versions of the variables
  $\ZNht{n}{hq}{t}$, $\QNh{n}{hq}$, $\ZNMt{n}{m}{t}(\bm{a})$,
  $\ZNhtvec{n}{\overbar{m}:t}$, $\QNM{n}{m}(\bm{a})$ and
  $\QNhvec{n}{\overbar{m}}$.
    
  A minor revision of \cref{eq:connection_RV_nabla_QMN} proves that
  the same limiting distribution occurs for the
  \mbox{$\sqrt{\bz{1}\bz{2}}$}-scaled gradient of
  $\Qz{\OUbar{m}:n}\!\parenR{\bmThetaz{\OUbar{m}|\bm{b}}\!\parenR{\LGp,\LGpd}}$
  and for the random variable
  {$\QNhvec{n}{\OUbar{m}}(\LGp,\LGpd) \defeq \Vector[\prime]{
      \subp{\QNhvec{n}{\overbar{m}}(\LGp)}{}{}{}{\prime},
     \subp{\QNhvec{n}{\overbar{m}}(\LGpd)}{}{}{}{\prime}}$},
    and it is easy to see that
  \mbox{$\ZNMt{n}{\Ubar{m}}{t}(\bmaz{1},\bmaz{2};\LGp,\LGpd) \defeq
    \ZNMt{n}{m}{t}(\bmaz{1};\LGp) + \ZNMt{n}{m}{t}(\bmaz{2};\LGpd)$}
  must take the place of $\ZNMt{n}{m}{t}(\bm{a})$ in the existing
  proofs.
    The key ingredient for the asymptotic investigation of
  $\ZNMt{n}{\Ubar{m}}{t}(\bmaz{1},\bmaz{2};\LGp,\LGpd)$ is a simple
  extension of \myref{th:expectations_of_XNht*}{eq:E(Xhi.Xjk)} such
  that it also covers the \enquote{cross-term} cases
  $\E{\XNht{n}{hq}{i}(\LGp)\cdot\XNht{n}{jr}{k}(\LGpd)}$ and verifies
  that these cases are asymptotically negligible.  This follows from
  the results stated in
  \cref{th:integrals_kernel_and_score_components_LGP_and_LGpd}
    
  The statement for $\ZNMt{n}{m}{t}(\bm{a})$ given in
  \myref{th:Lp_properties_for_ZNht_and_ZNMt}{th:Lp_properties_for_ZNMt_limit}
  extends trivially to the present case, since the asymptotic
  behaviour are unaffected by the adjustment that a sum of length $m$
  is replaced by two sums of length~$m$.  The statement in
  \cref{th:convergent_covariance_ZNMt} remains the same too, but some
  minor adjustments are needed in the proof: First of all, from the
  definition of $\ZNMt{n}{\Ubar{m}}{t}(\bmaz{1},\bmaz{2};\LGp,\LGpd)$,
  it follows that
  \begin{align}
    \nonumber
    \ZNMt{n}{\Ubar{m}}{i}(\bmaz{1},\bmaz{2};\LGp,\LGpd) \cdot
    \ZNMt{n}{\Ubar{m}}{k}(\bmaz{1},\bmaz{2};\LGp,\LGpd)
    & =       \ZNMt{n}{m}{i}(\bmaz{1};\LGp) \cdot
      \ZNMt{n}{m}{k}(\bmaz{1};\LGp)
      +        \ZNMt{n}{m}{i}(\bmaz{1};\LGp) \cdot
      \ZNMt{n}{m}{k}(\bmaz{2};\LGpd) \\
    &\phantom{=\ } +       \ZNMt{n}{m}{k}(\bmaz{1};\LGp) \cdot
      \ZNMt{n}{m}{i}(\bmaz{2};\LGpd)
      +        \ZNMt{n}{m}{i}(\bmaz{2};\LGpd) \cdot
      \ZNMt{n}{m}{k}(\bmaz{2};\LGpd),
  \end{align}
  and only the parts that contains both $\LGp$ and $\LGpd$ needs to be
  investigated (since the other terms already are covered by the
  existing results).  The statement that must be verified reduces to
  \begin{align}
    \label{eq:convergent_covariance_ZNMt_general_M_both_LGp_and_LGpd}
    \frac{1}{n} \sumss{\stackrel{i,k = 1}{\scriptscriptstyle i\neq k}}{n}     \absp{\ZNMt{n}{m}{i}(\bmaz{1};\LGp) \cdot
      \ZNMt{n}{m}{k}(\bmaz{2};\LGpd) } &= \oh{1},
  \end{align}
  and it is straightforward to verify that this sum can be realised as
  \begin{align}
    \label{eq:convergent_covariance_final_step_LGp_and_LGpd}
    \sumss{\ell=1}{n-1} \left(1
    - \frac{\ell}{n}\right)     \INMl{n}{m}{\ell}(\bmaz{1},\bmaz{2};\LGp,\LGpd) 
    +     \sumss{\ell=1}{n-1} \left(1
    - \frac{\ell}{n}\right)     \INMl{n}{m}{\ell}(\bmaz{2},\bmaz{1};\LGpd,\LGp),
  \end{align}
  where
  \mbox{$\INMl{n}{m}{\ell}(\bmaz{1},\bmaz{2};\LGp,\LGpd) \defeq \absp{
      \E{\ZNMt{n}{m}{0}(\bmaz{1},\LGp) \cdot         \ZNMt{n}{m}{\ell}(\bmaz{2},\LGpd)} } $}, with
  $\INMl{n}{m}{\ell}(\bmaz{2},\bmaz{1};\LGpd,\LGp)$ defined in the
  obvious manner by interchanging the parameters and the points.
    The desired result follows from this, since the remaining part of
  the proof of \cref{th:convergent_covariance_ZNMt} (using the
  adjusted version of \myref{th:expectations_of_XNht*}{eq:E(Xhi.Xjk)})
  gives that the two sums in
  \cref{eq:convergent_covariance_final_step_LGp_and_LGpd} both
  are~$\oh{1}$.
  
  The investigation of the variance of
  $\ZNMt{n}{\Ubar{m}}{t}(\bmaz{1},\bmaz{2};\LGp,\LGpd)$ is straight
  forward, i.e.\ the standard formula for the variance of a sum of
  random variables gives
  \begin{align}
    \nonumber
                                    \Var{\ZNMt{n}{\Ubar{m}}{t}(\bmaz{1},\bmaz{2};\LGp,\LGpd)} 
      = \Var{\ZNMt{n}{m}{t}(\bmaz{1},\LGp)}  
     + 2
      \Cov{\ZNMt{n}{m}{t}(\bmaz{1},\LGp)}{ 
      \ZNMt{n}{m}{t}(\bmaz{2},\LGpd)}  
        + \Var{\ZNMt{n}{m}{t}(\bmaz{2},\LGpd)},
  \end{align}
  and the revised version of
  \myref{th:expectations_of_XNht*}{eq:E(Xhi.Xjk)} implies that the
  covariance part of this expression is asymptotically negligible.
  The two variances are already covered by the existing version of
  \myref{th:var(ZNMt)}{th:var(ZNMt)=expression}, and from this it is
  clear that the asymptotically non-negligible parts can be written as
  \begin{align}
    \label{eq:asympotically_non_negligible_main_matrix_LGp_and_LGpd}
    \baM[\prime]{\Ubar{m}} \cdot \WMbN{\Ubar{m}}{}{\bm{b}} \cdot
    \baM{\Ubar{m}}     \defeq
    \Vector{\bmaz[\prime]{1},\bmaz[\prime]{2}} \cdot     \parenR{\WMbN{m}{}{\bm{b}}(\LGp) \oplus
    \WMbN{m}{}{\bm{b}}(\LGpd)} \cdot     \begin{bmatrix}
      \bmaz{1} \\ 
      \bmaz{2}
    \end{bmatrix}     = \bmaz[\prime]{1} \cdot \WMbN{m}{}{\bm{b}}(\LGp)  \cdot \bmaz{1}     + \bmaz[\prime]{2} \cdot \WMbN{m}{}{\bm{b}}(\LGpd) \cdot \bmaz{2},
  \end{align}
  whereas the asymptotically negligible parts of
  \myref{th:var(ZNMt)}{th:var(ZNMt)=expression} remains as before.
  This is sufficient for the revision of \cref{th:var(ZNMt)} (since
  \cref{th:eta_variance_expressions_not_truncated,th:eta_variance_expressions_truncated}
  follows from \cref{th:var(ZNMt)=expression} and
  \cref{th:convergent_covariance_ZNMt})
  
              Finally, \cref{th:AN_of_QNh_and_QNM} can now be updated based on the
  matrix
  \mbox{$\WMbN{\Ubar{m}}{}{\bm{b}} \defeq \WMbN{m}{}{\bm{b}}(\LGp)
    \oplus \WMbN{m}{}{\bm{b}}(\LGpd)$},
      and with some minor adjustments of the proof, i.e.\ new cross-terms
  are asymptotically negligible and sums of length $m$ are replaced
  with two sums of length~$m$, it follows that
  \begin{align}
    \label{eq:asymptotical_result_penalty_function_LGp_and_LGpd}
    \nz[-1/2]{}\, \Qz{\OUbar{m}:n}\!\parenR{\bmThetaz{\OUbar{m}|\bm{b}}\!\parenR{\LGp,\LGpd}}
    \stackrel{\scriptscriptstyle d}{\longrightarrow}     \UVN{\bm{0}}{\WMbN{m}{}{\bm{b}}(\LGp)
    \oplus \WMbN{m}{}{\bm{b}}(\LGpd)}.
  \end{align}
  
  The revised version of \cref{th:A4_requirement_final_case} is as
  before trivial to prove, which completes the investigation of the
  fourth requirement needed in order to use the Klimko-Nelson
  approach.  Basic linear algebra together with
  \cref{th:asymptotics_for_hatlgtheta}
      now finishes the proof.
\end{proof}

\medskip

  \label{note:alternative_smoothing_strategy}
  The arguments above could (under suitable assumptions) have been
  formulated in a more general setup, leading to a result that shows
  that the parameter vectors
  $\widehatbmthetaz{\LGpz{i}|\overbar{m}|\bm{b}}$ corresponding to
  different points $\TSR{\LGpz{i}}{i=1}{\nu}$ will be jointly
  asymptotically normal and pairwise asymptotically independent.  The
  asymptotically independent property are inherited by the
  corresponding estimated local Gaussian spectral densities
  $\hatlgsdM[p]{\LGpz{i}}{\omega}{m}$, and this enables an alternative
  smoothing strategy for the estimated local Gaussian spectral
  densities at a given point~$\LGp$, see
  \cref{sec:alternative_smoothing_strategy}.  However, the added
  computational cost incurred by such an estimation approach may make
  this a less interesting topic of investigation.

\subsection{An alternative smoothing strategy?}
\label{sec:alternative_smoothing_strategy}
The previously defined estimates $\hatlgsdM[5]{\LGp}{\omega}{m}$ of
$\lgsd[5]{\LGp}{\omega}$ was based on a weighting function
$\lambdaz{m}(h)$ that worked upon the estimated
values~$\hatlgacr[5]{\LGp}{h}$, but it should for the record be noted
that an alternative approach could have been applied too.

The point is that it is possible to extend the result of
\cref{app:theta_extension_different_local_points} to show that the
estimated \mbox{$m$-truncated} local Gaussian spectral densities
$\hatlgsdM[5]{\LGpz{i}}{\omega}{m}$ corresponding to different points
$\TSR{\LGpz{i}}{i=1}{\nu}$ will be jointly asymptotically normal and
pairwise asymptotically independent (when $\mlimit$ and $\blimit$ as
$\nlimit$).  This enables an alternative smoothing strategy, where an
estimate $\hatlgsdM[5]{\LGp}{\omega}{m}$ for a given point~$\LGp$
could be based on a weighting of the values of
$\hatlgsdM[5]{\LGpz{i}}{\omega}{m}$ in a grid of points
surrounding~$\LGp$.

This alternative approach shares some superficial similarities with
the one used when the ordinary global spectrum $f(\omega)$ is computed
based on the periodogram, see e.g.\ \citet{Brockwell:1986:TST:17326}
for details.  However, the efficiency of the periodogram-approach in
the estimation of~$f(\omega)$ is due to the \textit{Fast Fourier
  Transform}, which implies that the periodogram can be computed
directly from the observations without the need for an explicit
computation of all of the estimated
autocovariances~$\widehatrhoz{}(h)$, and that shortcut is not
available for the local Gaussian case.  The computational load would
thus become much larger for the local Gaussian case if such an
averaging-approach was applied.

\section{Technical details}
\label{app:Tecnhical_details}
\setcounter{figure}{0}

This section collects some technical details that would have impeded
the flow of the main argument if they had been included throughout the
paper.  A brief overview: \Cref{app:_diagonal_argument} discuss the
\textit{diagonal folding property} of the local Gaussian
autocorrelations~$\lgacr[p]{\LGp}{h}$ and \cref{app:_time_reversible}
considers the special case of time-reversible time series.
\Cref{app:limit_theorems} collects technical results related to the
asymptotic relationship between $n$, $m$ and~$\bm{b}$, whereas
\cref{app:integrals_kernel_score} shows that the assumptions on the
kernel function $K(\bm{w}$) and the score
functions~$\Uh[\bm{w}]{hq:\bm{b}}$ implies that some integrals are
finite (which implies that
\myref{assumption_Yt}{assumption:h_x:y_x:y:z_finite_expectations} will
be trivially satisfied if the bivariate densities $\gh[\yh{h}]{h}$ are
finite).  \Cref{app:sigma_algebras_Lp_spaces} contains a few basic
definitions/comments related to \mbox{$\alpha$-mixing},
\mbox{$\sigma$-algebras} and \mbox{$\Lp{\nu}$-spaces}.

\makeatletter{}
\subsection{The diagonal folding property of $\lgacr[p]{\LGp}{h}$}
\label{app:_diagonal_argument}

The following simple observation about $\lgacr[p]{\LGp}{h}$ is of
interest both for theoretical and computational aspects of the local
Gaussian spectral density $\lgsd[p]{\LGp}{\omega}$.

\begin{lemma}
  \label{th:diagonal_folding_property}
  For a strictly stationary time series $\TSR{\Yz{t}}{t\in\ZZ}{}$ and
  a point~\mbox{$\LGp=\LGpoint$}, the following symmetry property
  (diagonal folding) holds for the local Gaussian autocorrelation,
  \begin{align}
    \label{eq:th:diagonal_folding_property}
    \lgacr[p]{\LGp}{-h} = \lgacr[p]{\LGpd}{h},
  \end{align}
  where \mbox{$\LGpd=\parenR{\LGpi{2},\LGpi{1}}$} is the diagonal
  reflection of $\LGp$.
\end{lemma}

\begin{proof}
  This is a simple consequence of the symmetrical nature of the
  bivariate random variables
  \mbox{$\Yht{h}{t}\defeq\left(\Yz{h},\Yz{0}\right)$} and
  \mbox{$\Yht{-h}{t}\defeq\left(\Yz{-h},\Yz{0}\right)$}, which due to
  the connection between the corresponding cumulative density
  functions
  \begin{align}
    \nonumber
    \label{eq:G_diagonal_folding_property}
    \Gz{\!-h}\!\left(\yz{-h},\yz{0} \right)
    &= \Prob{\Yz{-h} \leq \yz{-h}, \Yz{0} \leq \yz{0}} 
      = \Prob{\Yz{0} \leq \yz{0}, \Yz{-h} \leq \yz{-h}}  
      = \Prob{\Yz{h} \leq \yz{0}, \Yz{0} \leq \yz{-h}} \\
    &= \Gz{h}\!\left(\yz{0},\yz{-h} \right)
  \end{align}
  gives the following property\footnote{    This must not be confused with the property that $\gz{h}$ and
    $\gz{-h}$ themselves are symmetric around the diagonal, for that
    will in general not be the case.}   for the probability density functions,
  \begin{align}
    \gz{\!-h}\!\left(\yz{-h},\yz{0} \right) = \gz{h}\!\left(\yz{0},\yz{-h} \right).
  \end{align}

  This implies that \mbox{$\gz{\!-h}(\LGp) = \gz{h}(\LGpd)$}, and the
  symmetry does moreover induce a symmetrical relation between the
  parameters $\bmthetaz{-h}\!\left(\LGp\right)$ of the local Gaussian
  approximation of $\gz{-h}$ at $\LGp$ and the parameters
  $\bmthetaz{h}\!\left(\LGpd\right)$ of the local Gaussian
  approximation of $\gz{h}$ at~$\LGpd$, i.e.\ if
  \mbox{$\bmthetaz{-h}\!\left(\LGp\right)=\Vector[\prime]{\muz{1},\muz{2},\sigmaz{11},\sigmaz{22},\rho}$}
  then
  \mbox{$\bmthetaz{h}\!\left(\LGpd\right)=\Vector[\prime]{\muz{2},\muz{1},\sigmaz{22},\sigmaz{11},\rho}$}.
    \Cref{eq:th:diagonal_folding_property} follows since $\rho$ in these
  two vectors respectively represents $\lgacr[p]{\LGp}{-h}$ and
  $\lgacr[p]{\LGpd}{h}$, and this completes the proof.
\end{proof}

  A trivial consequence of the \textit{diagonal folding property} in
  \cref{th:diagonal_folding_property} is that the local Gaussian
  autocorrelation becomes an even function of the lag~$h$ when
  $\LGpdiagonal$.

\makeatletter{}
\subsection{Time-reversible time series}
\label{app:_time_reversible}

Additional symmetry properties are present for time reversible time
series, which implies that the local Gaussian spectral densities
$\lgsd[p]{\LGp}{\omega}$ always are real-valued for such time series,
see \cref{def:Yt_reversible,th:asymptotics_for_hatlgsd_reversible}.

The following simple result follows immediately from
\cref{def:Yt_reversible}.
\begin{lemma}
  \label{th:Yt_time_reversible}
  If $\TSR{\Yz{t}}{t\in\ZZ}{}$ is time reversible, then  
  \begin{align}
    \label{eq:def:Yt_reversible}
    \gh[\LGpi{1},\LGpi{2}]{h} =  \gh[\LGpi{2},\LGpi{1}]{h}
  \end{align}
  for all points \mbox{$\LGp=\LGpoint \in \RRn{2}$} and all
  \mbox{$h\in\NN$}, which implies 
  \begin{align}
    \lgacr[p]{\LGp}{-h} = \lgacr[p]{\LGp}{h}.
  \end{align}
\end{lemma}

\begin{proof}
  The time reversibility of $\TSR{\Yz{t}}{t\in\ZZ}{}$ implies that
  \mbox{$\parenR{\Yz{h},\Yz{0}}$} and \mbox{$\parenR{\Yz{-h},\Yz{0}}$}
  have the same joint distribution, i.e.\
  \begin{align}
    \nonumber
    \Gz{\!-h}\!\left(\yz{-h},\yz{0} \right)
    &= \Prob{\Yz{-h} \leq \yz{-h}, \Yz{0} \leq \yz{0}} 
      = \Prob{\Yz{h} \leq \yz{-h}, \Yz{0} \leq \yz{0}} 
      = \Gz{h}\!\left(\yz{-h},\yz{0} \right).
  \end{align}
  Together with the observation in
  \cref{eq:G_diagonal_folding_property}, this gives the diagonal
  symmetry stated in \cref{eq:def:Yt_reversible}.  The statement for
  the local Gaussian autocorrelations follows by the same reasoning as
  in the proof of \cref{th:diagonal_folding_property}.
\end{proof}

\makeatletter{}

\subsection{Two limit theorems --- and one comment}
\label{app:limit_theorems}

This section contains two lemmas and one comment.
\Cref{th:block_sizes_for_main_result} combines a check of the internal
consistency of \cref{assumption_Nmb} with the limits needed for the
small block-large block argument in \cref{th:AN_of_QNh_and_QNM},
whereas \cref{th:k_N_limit_for_off_diagonal_components} takes care of
the two limits needed in order to prove that the \textit{off the
  diagonal components} in \cref{th:convergent_covariance_ZNMt} are
asymptotically negligible.  The comment at the end of this section has
been included due to the remark at the end of
\cref{pR2bp2_number_of_points_in_the_window_main} in the main
document.

\begin{lemma}
  \label{th:block_sizes_for_main_result}
  Under \cref{assumption_Nmb}, the following holds.
  \begin{enumerate}[label=(\alph*)]
  \item     \label{th:block_sizes_for_main_result_M_and_awful_exponent}
        There exists integers~$s$ that makes
    \cref{eq:assumption_m=o((Nb1b2)^tau/(2+5tau)-lambda),eq:assumption_m=o(s)}
    of \cref{assumption_Nmb} compatible.
      \item     \label{th:block_sizes_for_main_result_existence_c}
        There exists integers $s$ and constants \mbox{$c \defeq
      \cNi{n}\rightarrow\infty$}, such that
    \begin{align}
      \label{eq:cn_properties}
      c\cdot s = \oh{\sqrt{n\bz{1}\bz{2}/m}},       \qquad                         \sqrt{nm/\bz{1}\bz{2}}\cdot c\cdot \alpha(s-m+1)
      \longrightarrow 0.
    \end{align}
      \item     \label{th:block_sizes_for_main_result_r_ell_limits}
        There exists integers $s$ and constants~$c$, such that with $r$,
    $\ell$ and~$\vartheta$ given as the integers
    \begin{align}  
      \label{eq:definition_r_and_ell}
      r = \rNi{n} \defeq \floor{\frac{\sqrt{n\bz{1}\bz{2}/m}}{c}},       \qquad       \ell = \ellNi{n} \defeq \floor{\frac{n}{r + s}},       \qquad       \vartheta = \varthetaz{\!n} \defeq s-m+1,
    \end{align}
        the following limits occur when \mbox{$n\rightarrow\infty$}:
    \begin{align}
      \label{eq:five_limits}
      \frac{s}{r} \longrightarrow0;\qquad       \ell \alpha(\vartheta) \longrightarrow 0; \qquad       \frac{mr}{n} \longrightarrow 0;\qquad       \frac{mr}{\ell\bz{1}\bz{2}} \longrightarrow 0; \qquad       \frac{m\ell s}{n} \longrightarrow 0.
    \end{align}
  \end{enumerate}
\end{lemma}

\begin{proof}
  \Cref{th:block_sizes_for_main_result_M_and_awful_exponent} will be
  established by first observing that it is possible to find
  integers~$s$ that ensures that
  \myref{assumption_Nmb}{eq:assumption_m=o(s)} is compatible with the
  requirement~\mbox{$m=\oh{\parenRz[\xi]{}{n\bz{1}\bz{2}}}$},
  for any~\mbox{$\xi\in\left(0,\tfrac{1}{3}\right)$}, and then
  checking that the exponent~\mbox{$\tau/(2+5\tau)-\lambda$} lies in
  this~interval.

  Observe that it is impossible to have~\mbox{$m=\oh{s}$}
  and~\mbox{$s=\oh{\sqrt{n\bz{1}\bz{2}/m}}$}
  when~\mbox{$m\geq \sqrt{n\bz{1}\bz{2}/m}$}, which implies
  \mbox{$m<\sqrt{n\bz{1}\bz{2}/m}$}, which is equivalent
  to~\mbox{$m<\parenRz[1/3]{}{n\bz{1}\bz{2}}$}.  Some extra leeway is
  needed in order to construct the desired integers~$s$, so consider
  the requirement
  \begin{align}
    m &= \oh{\parenRz[1/3-\zeta]{}{n\bz{1}\bz{2}}},         \qquad \text{for some } \zeta \in \left(0,\tfrac{1}{3}\right).
  \end{align}
    Define the integers~$s$ by~\mbox{$s\defeq m\cdot\mathfrak{s}$},
  where~\mbox{$\mathfrak{s} \defeq 1\vee
    \floor{\parenRz[\zeta/2]{}{n\bz{1}\bz{2}}}$}, and note that this
  construction ensures that~$s$ goes to~$\infty$.  Further,
  \mbox{$m=\oh{s}$} holds
  since~\mbox{$m/s = 1/\mathfrak{s} \rightarrow 0$},
  and~\mbox{$s=\oh{\sqrt{n\bz{1}\bz{2}/m}}$} holds since
  \begin{align}
    \nonumber
    \frac{s}{\sqrt{n\bz{1}\bz{2}/m}}     &\asymp       \frac{m\cdot\parenRz[\zeta/2]{}{n\bz{1}\bz{2}}}{      \parenRz[1/2]{}{n\bz{1}\bz{2}/m}}       = \frac{\mz[3/2]{}}{\parenRz[(1-\zeta)/2]{}{n\bz{1}\bz{2}}}           = \parenSz[3/2]{}{\frac{m}{
      \parenRz[(1-\zeta)/3]{}{n\bz{1}\bz{2}}}} \\
        &= \parenSz[3/2]{}{\frac{1}{      \parenRz[2\zeta/3]{}{n\bz{1}\bz{2}}}       \cdot \frac{m}{      \parenRz[1/3-\zeta]{}{n\bz{1}\bz{2}}}}       \rightarrow \parenSz[3/2]{}{\frac{1}{\infty}\cdot 0}       = 0.
  \end{align}
  This implies that the desired integers~$s$ can be found
  whenever~\mbox{$m=\oh{\parenRz[\xi]{}{n\bz{1}\bz{2}}}$},
  with~\mbox{$\xi\in\left(0,\tfrac{1}{3}\right)$}.  Since the value
  of~\mbox{$\tau/(2+5\tau)-\lambda$} lies in the
  interval~\mbox{$\left(0,\tfrac{1}{5}\right)$}, the proof
  of~\cref{th:block_sizes_for_main_result_M_and_awful_exponent}
  is~complete.

  For~\cref{th:block_sizes_for_main_result_existence_c,th:block_sizes_for_main_result_r_ell_limits},
  the integers~$s$ and constants~$c$ can e.g.\ be defined as
  \begin{align}
    \label{eq:definition_of_s_and_c}
    s =     1\vee\floor{    \parenRz[1-\eta]{}{\sqrt{n\bz{1}\bz{2}/m}}},     \qquad                 c =     \parenRz[\eta/2]{}{\sqrt{n\bz{1}\bz{2}/m}},     \quad     \text{for some \mbox{$\eta\in(0,1)$}.}
  \end{align}
  Since \mbox{$1-\eta$} and $\eta/2$ are in \mbox{$(0,1)$}, it follows
  from \myref{assumption_Nmb}{eq:assumption_Nb1b2/m} that~$s$ and~$c$
  goes to~$\infty$ as~required.  A quick inspection reveals that the
  product~\mbox{$c\cdot s$} is $\oh{\sqrt{n\bz{1}\bz{2}/m}}$, proving
  the first part of \cref{eq:cn_properties}.  For the second part of
  \cref{eq:cn_properties}, keep in mind the similarity with
  \myref{assumption_Nmb}{eq:assumption_alpha(s)_and_(N/b1b2)^1/2},
  and observe that $c$ in the limit is asymptotically equivalent to
  $\sz[\eta/2(1-\eta)]{}$.  Since $\eta$ can be selected such
  that the exponent \mbox{$\eta/2(1-\eta)$} becomes smaller than any
  \mbox{$\tau>0$}, the second statement holds too, which completes the
  proof of~\cref{th:block_sizes_for_main_result_existence_c}.

  In order to
  prove~\cref{th:block_sizes_for_main_result_r_ell_limits}, note that
  a floor-function~\mbox{$\floor{x}$} in a denominator can be ignored
  in the limit \mbox{$x\rightarrow\infty$},
  since~\mbox{$x\asymp\floor{x}$}, that
  is~\mbox{$\lim x/\floor{x} = 1$}.  Moreover, observe that
  \myref{assumption_Nmb}{eq:assumption_Nb1b2/m} implies that $n/m$
  goes to~$\infty$.  With these observations, all except the last
  limit in \cref{eq:five_limits} are trivial to~prove,~i.e.\
    \begin{subequations}
    \begin{align}
      \frac{s}{r}       &\asymp \frac{s}{\frac{\sqrt{n\bz{1}\bz{2}/m}}{c}}         = \frac{c\cdot s}{\sqrt{n\bz{1}\bz{2}/m}}         \rightarrow 0, \\
              \ell \alpha(\vartheta)       &\leq \frac{n}{r+s} \alpha(\vartheta)         \asymp \frac{n}{r} \alpha(\vartheta)         \asymp \frac{n}{\frac{\sqrt{n\bz{1}\bz{2}/m}}{c}}
        \alpha(\vartheta)         = \sqrt{nm/\bz{1}\bz{2}} \cdot c\cdot \alpha(\vartheta)         \rightarrow 0, \\
              \frac{mr}{n}       &\leq \frac{\frac{\sqrt{n\bz{1}\bz{2}/m}}{c}}{n/m}         = \frac{\sqrt{\bz{1}\bz{2}}}{c\sqrt{n/m}}         \rightarrow \frac{0}{\infty\cdot\infty}         = 0, \\
              \frac{mr}{\ell\bz{1}\bz{2}}       &\asymp \frac{mr}{\frac{n}{r+s}\bz{1}\bz{2}}         =  \frac{r(r+s)}{n\bz{1}\bz{2}/m}         \asymp \frac{\rz[2]{}}{n\bz{1}\bz{2}/m}         \leq \frac{\frac{n\bz{1}\bz{2}/m}{        \cz[2]{}}}{n\bz{1}\bz{2}/m}         = \frac{1}{\cz[2]{}}
        \rightarrow 0.
    \end{align}
  \end{subequations}

  For the proof of~\mbox{$m\ell s/n\rightarrow 0$}, the explicit
  expressions for~$s$ and~$c$ from \cref{eq:definition_of_s_and_c}
  will be needed, i.e.\
  \begin{align}
    \nonumber
    \frac{m\ell s}{n}     &\leq \frac{m \frac{n}{r+s} s}{n}       = m\frac{s}{r+s}       \asymp m\frac{s}{r}       \asymp m\frac{c\cdot s}{\sqrt{n\bz{1}\bz{2}/m}}       \leq m\frac{      \parenRz[1-\eta/2]{}{\sqrt{n\bz{1}\bz{2}/m}}}{      \sqrt{n\bz{1}\bz{2}/m}}  \\
        \label{eq:Mls/N_requirement}
    &= \frac{m}{\parenRz[\eta/4]{}{n\bz{1}\bz{2}/m}}       = \frac{\mz[1+\eta/4]{}}{      \parenRz[\eta/4]{}{n\bz{1}\bz{2}}}       = \parenRz[(4+\eta)/4]{}{\frac{m}{      \parenRz[\eta/(4+\eta)]{}{n\bz{1}\bz{2}}}}      \!\!\!\!\!\!\!\!\!\!\!\!\!\!\!.
  \end{align}
    \Myref{assumption_Nmb}{eq:assumption_m=o((Nb1b2)^tau/(2+5tau)-lambda)}
  states
  that~\mbox{$m =
    \oh{\parenRz[\tau/(2+5\tau)-\lambda]{}{n\bz{1}\bz{2}}}$},
  and it is consequently sufficient to show that an~$\eta$ can be
  found which
  gives~\mbox{$\tau/(2+5\tau)-\lambda \leq p(\eta) \defeq
    \eta/(4+\eta)$}.  Since
  \mbox{$\pz[\prime]{}(\eta) =
    4/\parenRz[2]{}{4+\eta} > 0$}, the highest value
  of~$p(\eta)$ will be found at the upper end of the interval of
  available arguments.  From the proof
  of~\cref{th:block_sizes_for_main_result_existence_c} it is known
  that   \mbox{$\eta/2(1-\eta)<\tau$}, which gives the
  requirement~\mbox{$\eta < 2\tau/(1+2\tau)$}.      The value of~$p(\eta)$ at the upper end of this interval
  is~\mbox{$\tau/(2+5\tau)$}, and since~\mbox{$\lambda>0$} it is 
  possible to find an~$\eta$ that
  satisfies~{$\tau/(2+5\tau)-\lambda\leq
    p(\eta)<\tau/(2+5\tau)$}, which concludes the proof.
\end{proof}

\begin{lemma}
  \label{th:k_N_limit_for_off_diagonal_components}
  Under \cref{assumption_Nmb}, the sequence of integers defined by
  $\kz{n} + 1 \defeq \ceil{\mz[2/a]{} \cdot
    \absp[(2-\nu)/a\nu]{\bz{1}\bz{2}} }$ satisfies the following two
  limit requirements.
  \begin{enumerate}[label=(\alph*)]
  \item     \label{th:k_N_limit_for_off_diagonal_components__infty_limit}
    $\kz{n} \longrightarrow \infty$.
  \item     \label{th:k_N_limit_for_off_diagonal_components__zero_limit}
    $\kz{n} \mz[2]{} \bz{1}\bz{2} \longrightarrow 0$.
  \end{enumerate}
\end{lemma}

\begin{proof}
  The key requirements \mbox{$\nu>2$} and \mbox{$a>1-2/\nu$}
  (inherited from \myref{assumption_Yt}{assumption_Yt_strong_mixing})
  ensures that \mbox{$2/a>0$} and \mbox{$(2-\nu)/a\nu<0$}.  As
  $\mlimit$ and $\blimit$ when $\nlimit$, it follows that
  \mbox{$\kz{n}\rightarrow\infty$}, which proves
  \cref{th:k_N_limit_for_off_diagonal_components__infty_limit}.

  For \cref{th:k_N_limit_for_off_diagonal_components__zero_limit},
  observe that \mbox{$\kz{n} = \ceil{      \mz[2/a]{} \cdot \absp[(2-\nu)/a\nu]{\bz{1}\bz{2}} } - 1 <
    \mz[2/a]{} \cdot \absp[(2-\nu)/a\nu]{\bz{1}\bz{2}}$} implies
  \begin{subequations}
    \begin{align}
      \kz{n} \mz[2]{} \bz{1}\bz{2} 
      &<         \left( \mz[2/a]{} \cdot
        \absp[(2-\nu)/a\nu]{\bz{1}\bz{2}} \right)         \cdot  \mz[2]{} \bz{1}\bz{2} \\
      &=         \mz[2(1+1/a)]{} \cdot
        \absp[1+ (2-\nu)/a\nu]{\bz{1}\bz{2}} \\
      &\leq         \mz[2(1+1/a)]{} \cdot
        \absp[1+ (2-\nu)/a\nu]{        \parenRz[2]{}{\bz{1}\vee\bz{2}}} \\
      &=         \parenCz[2(1+ (2-\nu)/a\nu)]{}{         \mz[\parenC{1+1/a}/\parenC{1+(2-\nu)/a\nu}]{} \cdot
        \left(\bz{1}\vee\bz{2}\right)} \\
      &=         \parenCz[2(1+ (2-\nu)/a\nu)]{}{
        \mz[\parenC{\nu(a+1)}/\parenC{\nu(a-1) +2}]{} \cdot
        \left(\bz{1}\vee\bz{2}\right) }.    \end{align}
  \end{subequations}
  An inspection of the outermost exponent reveals
  \begin{align}
    2\cdot\left(1+ \frac{(2-\nu)}{a\nu}\right) =         2\cdot\frac{a - (1-2/\nu)}{a} > 0,
  \end{align}
  which together with \myref{assumption_Nmb}{eq:assumption_m(b1 join
    b2)} concludes the proof of
  \cref{th:k_N_limit_for_off_diagonal_components__zero_limit}.
\end{proof}

\textbf{A comment related to the remark at the end of
  \cref{pR2bp2_number_of_points_in_the_window_main}:} 
\label{pR2bp2_number_of_points_in_the_window_SM}
It is not required for the theoretical investigation, but it might
still be of interest to mention the following observation: Consider
a combination of a given point $\LGp$, a small bandwidth vector
$\bm{b}=\parenR{\bz{1},\bz{2}}$, and a \textit{large} sample of size
$n$ from a univariate time series $\TSR{\Yz{t}}{t\in\ZZ}{}$ that
satisfies \cref{assumption_Yt}.  The number of lag-$h$ pairs in the
vicinity of $\LGp$ will then, for each $h=1,\dotsc,m$, be of order
$n\bz{1}\bz{2}\cdot\gh[\LGp]{h}$ --- and this will, when
$\gh[\LGp]{h}>0$, go to infinity when $\nlimit$ and $\blimit$.

Only a sketch of the argument will be given here, since the
asymptotic theory does not build upon this observation: First select
a $\bm{b}$-dependent region $\mathcalVz{\bm{b}}(\LGp)$ around $\LGp$
to be the \enquote{$\bm{b}$-vicinity of $\LGp$}, i.e.\
$\mathcalVz{\bm{b}}(\LGp)$ should shrink when $\blimit$.  The area
of $\mathcalVz{\bm{b}}(\LGp)$ should be given by some constant
$\mathcal{A}$ times $\bz{1}\bz{2}$.
From a sample of size $n$
there will be a total of $n-h$ lag-$h$ pairs, and the expected
number of those in the region $\mathcalVz{\bm{b}}(\LGp)$ will be
$(n-h)\cdot\iint_{\mathcalVz{\bm{b}}(\LGp)}\gh[\yh{h}]{h}\dyh{h}$.
\Myref{assumption_Yt}{assumption:gh_differentiable_at_LGp} implies
that the bivariate density functions $\gh[\yh{h}]{h}$ are continuous
at $\LGp$, and it is thus clear that both
$\inf_{\yh{h}\in \mathcalVz{\bm{b}}(\LGp)} \gh[\yh{h}]{h}$ and
$\sup_{\yh{h}\in \mathcalVz{\bm{b}}(\LGp)} \gh[\yh{h}]{h}$ go to
$\gh[\LGp]{h}$ when $\blimit$.  The integral
$\iint_{\mathcalVz{\bm{b}}(\LGp)}\gh[\yh{h}]{h}\dyh{h}$ will thus be
of order $\mathcal{A}\cdot\bz{1}\bz{2}\cdot\gh[\LGp]{h}$ when
$\blimit$, and the result follows.

This shows why it even for rather large
samples might be hard to obtain good estimates of the local Gaussian
spectral densities in the tails, where the densities $\gh[\LGp]{h}$
are low.

\subsection{Integrals based on the kernel and the score functions}
\label{app:integrals_kernel_score}

The asymptotic properties of the random variables introduced in
\cref{def:Xht_*Lc,def:ZNht_and_QNh,def:ZNMt_and_QNM} does of course
depend upon the properties of the time series
$\TSR{\Yz{t}}{t\in\ZZ}{}$ upon which they have been defined, but quite
a few of the required properties does in fact only depend upon
$K(\bm{w})$ and~$\Uh[\bm{w}]{hq:\bm{b}}$.
Note that the treatment in this section exploits the property that the
functions~$\Uh[\bm{w}]{hq:\bm{b}}$ all are quadratic polynomials in
the variables $\wz{1}$ and $\wz{2}$, which implies that the
inequalities from \cref{th:integrals:K_w1w2} is sufficient for the
proofs of the asymptotic results given in
\cref{th:integrals_kernel_and_score_components}.  

\begin{lemma}
  \label{th:integrals:K_w1w2}
  For $K(\bm{w})$ from \cref{def:kernel} (page \pageref{def:kernel}),
  and \mbox{$\nu>2$} from
  \myref{assumption_Yt}{assumption_Yt_strong_mixing} (page
  \pageref{assumption_Yt_strong_mixing}), the following~holds:
  \begin{enumerate}[label=(\alph*)]
  \item     \label{eq:int_|Kw_i|}
        $\absp{      \intss{\RRn{2}}{} K\!\left(\wz{1},\wz{2}\right)       \wz[k]{1} \wz[\ell]{2}       \d{\wz{1}}\!\d{\wz{2}}} < \infty, \qquad k, \ell
    \geq 0 \text{ and } k + \ell \leq 5$.
          \item     \label{eq:int_|K2w_i|}
        $\absp{      \intss{\RRn{2}}{} \power[2]{}{K\!\left(\wz{1},\wz{2}\right)}       \wz[k]{1} \wz[\ell]{2}       \d{\wz{1}}\!\d{\wz{2}}} < \infty, \qquad k, \ell
    \geq 0 \text{ and } k + \ell \leq 5$.
  \item     \label{eq:Kw_Lnu}
        $K\!\left(\wz{1},\wz{2}\right)     \wz[k]{1} \wz[\ell]{2} \in \Lp{\nu},     \qquad k, \ell \geq 0 \text{ and } k + \ell \leq 2$.
  \end{enumerate}
\end{lemma}

\begin{proof}
    Since the kernel function by definition is non-negative, it follows
  that
  \begin{align}
    \label{eq:int_|Kw_i|_proof}
    \absp{    \intss{\RRn{2}}{} K\!\left(\wz{1},\wz{2}\right)     \wz[k]{1} \wz[\ell]{2}     \d{\wz{1}}\!\d{\wz{2}}} \leq     \intss{\RRn{2}}{} \!\! K\!\left(\wz{1},\wz{2}\right)     \absp{\wz[k]{1} \wz[\ell]{2}}     \d{\wz{1}}\!\d{\wz{2}},
  \end{align}
  which proves \cref{eq:int_|Kw_i|}, since
  \cref{eq:kernel_integrals_finite} of \cref{def:kernel} implies that
  this is finite for the specified range of $k$ and~$\ell$.
  
    Since the kernel function is bounded, there is some constant
  $\mathcal{C}$ such that \mbox{$K(\bm{w}) \leq \mathcal{C}$}, which
  implies~that
  \begin{align}
    \label{eq:int_|K2w_i|_proof}
    \absp{    \intss{\RRn{2}}{} \power[2]{}{K\!\left(\wz{1},\wz{2}\right)}     \wz[k]{1} \wz[\ell]{2}     \d{\wz{1}}\!\d{\wz{2}}}     \leq     \mathcal{C} \absp{    \intss{\RRn{2}}{} K\!\left(\wz{1},\wz{2}\right)     \wz[k]{1} \wz[\ell]{2}     \d{\wz{1}}\!\d{\wz{2}}},
  \end{align}
  which due to \cref{eq:int_|Kw_i|} is finite, thus
  \cref{eq:int_|K2w_i|} holds true.
  
    Next, note that $\absp[\nu]{    K\!\left(\wz{1},\wz{2}\right) \wz[k]{1} \wz[\ell]{2}} =
  \absp[(\nu-1)]{ K\!\left(\wz{1},\wz{2}\right)}
  \absp{K\!\left(\wz{1},\wz{2}\right)} \absp[\nu]{\wz[k]{1}
    \wz[\ell]{2}} \leq \mathcalCz[(\nu-1)]{}
  K\!\left(\wz{1},\wz{2}\right) \absp[\nu]{\wz[k]{1} \wz[\ell]{2}}$,
  which gives the following inequality,
  \begin{align}
    \label{eq:Kw_Lnu_proof}
    \parenRz[1/\nu]{}{    \intss{\RRn{2}}{} \absp[\nu]{    K\!\left(\wz{1},\wz{2}\right)     \wz[k]{1} \wz[\ell]{2}} \d{\wz{1}}\!\d{\wz{2}} }    \leq     \mathcalCz[(\nu-1)/\nu]{}     \parenRz[1/\nu]{}{
    \intss{\RRn{2}}{} K\!\left(\wz{1},\wz{2}\right)     \absp[\nu]{\wz[k]{1} \wz[\ell]{2}} \d{\wz{1}}\!\d{\wz{2}} },
  \end{align}
  from which it is clear that a proof of the finiteness of the right
  hand side of \cref{eq:Kw_Lnu_proof} will imply \cref{eq:Kw_Lnu}.
  Since the region of integration can be divided into
  \mbox{$\mathcalAz{k\ell} = \TSR{\bm{w} : \absp{\wz[k]{1}
        \wz[\ell]{2}} \leq 1}{}{}$} and
  \mbox{$\mathcalAz[c]{k\ell} = \RRn{2} \setminus \mathcalAz{k\ell}$},
  it follows from the non-negativeness of $K(\bm{w})$, and
  \cref{eq:kernel_integral_one,eq:kernel_integrals_finite} of
  \cref{def:kernel}, that
  \begin{subequations}
    \begin{align}
      \label{eq:Kw_Lnu_proof_finite_region}
      \intss{\mathcalAz{k\ell}}{} K\!\left(\wz{1},\wz{2}\right)       \absp[\nu]{\wz[k]{1} \wz[\ell]{2}} \d{\wz{1}}\!\d{\wz{2}}         &\leq           \intss{\mathcalAz{k\ell}}{} K\!\left(\wz{1},\wz{2}\right)           \d{\wz{1}}\!\d{\wz{2}}            \leq           \intss{\RRn{2}}{} K\!\left(\wz{1},\wz{2}\right)           \d{\wz{1}}\!\d{\wz{2}} = 1, \\%
               \nonumber
      \intss{\mathcalAz[c]{k\ell}}{} K\!\left(\wz{1},\wz{2}\right)       \absp[\nu]{\wz[k]{1} \wz[\ell]{2}} \d{\wz{1}}\!\d{\wz{2}}       &\leq         \intss{\mathcalAz[c]{k\ell}}{} K\!\left(\wz{1},\wz{2}\right)         \absp[\ceil{\nu}]{\wz[k]{1} \wz[\ell]{2}} \d{\wz{1}}\!\d{\wz{2}} \\       \label{eq:Kw_Lnu_proof_finite_region_II}
      &\leq         \intss{\RRn{2}}{} K\!\left(\wz{1},\wz{2}\right)         \absp{\wz[k\ceil{\nu}]{1} \wz[\ell\ceil{\nu}]{2}}
        \d{\wz{1}}\!\d{\wz{2}} < \infty,    \end{align}
  \end{subequations}
  where the last inequality follows since the assumption
  \mbox{$k+\ell\leq2$} ensures that
  \mbox{$k\ceil{\nu} + \ell\ceil{\nu} \leq 2\ceil{\nu}$}.  
    The expression in \cref{eq:Kw_Lnu_proof} is thus finite --- and, as
  stated in \cref{eq:Kw_Lnu},
  \mbox{$K\!\left(\wz{1},\wz{2}\right) \wz[k]{1} \wz[\ell]{2} \in
    \Lp{\nu}$}.
\end{proof}

\begin{lemma}
  \label{th:integrals_kernel_and_score_components}
  The following holds for $\Uh[\bm{w}]{hq:\bm{b}}$ and
  $\Khb[\yh{h}-\LGp]{h}{\bm{b}}$ from \cref{def:Uh,def:kernel}, and
  \mbox{$\nu>2$} from
  \myref{assumption_Yt}{assumption_Yt_strong_mixing}:
  \begin{enumerate}[label=(\alph*)]
  \item     \label{eq:int_K*U}
        $\intss{\RRn{2}}{}     \sqrt{\bz{1}\bz{2}}\Khb[\bm{\zeta}-\LGp]{h}{\bm{b}}     \Uh[\bm{\zeta}]{hq:\bm{b}} \d{\bm{\zeta}} =     \Oh{\sqrt{\bz{1}\bz{2}}}$.
  \item     \label{eq:int_(KU)^nu}
        $\parenRz[1/\nu]{}{\intss{\RRn{2}}{}         \absp[\nu]{\sqrt{\bz{1}\bz{2}}\Khb[\bm{\zeta}-\LGp]{h}{\bm{b}}           \Uh[\bm{\zeta}]{hq:\bm{b}}}
        \!\!\d{\bm{\zeta}}} =
    \Oh{\absp[(2-\nu)/2\nu]{\bz{1}\bz{2}}}$.
  \item     \label{eq:int_KKUU}
        Let
    \mbox{$\mathcalKz{qr,hj:\bm{b}}\!\left(\zetah{1},\zetah{2}\right)
      \defeq       \Khb[\zetah{1}-\LGp]{h}{\bm{b}}       \Khb[\zetah{2}-\LGp]{j}{\bm{b}}       \Uh[\zetah{1}]{hq:\bm{b}}       \Uh[\zetah{2}]{jr:\bm{b}}$}, where     $\zetah{1}$ and $\zetah{2}$ either coincide completely
    (bivariate), have one common component (trivariate), or have no
    common components (tetravariate).     Let $\kappa$ be the number of variates, and let
    $\d{\bm{\zeta}(\kappa})$ represent the corresponding
    \mbox{$\kappa$-variate} differential.  Then,\\
        $\intss{\RRn{\kappa}}{}     \left(\bz{1}\bz{2}\right)     \mathcalKz{qr,hj:\bm{b}}\!\left(\zetah{1},\zetah{2}\right)     \d{\bm{\zeta}(\kappa}) =     \begin{cases}
      \Uh[\LGp]{hq:\bm{b}} \Uh[\LGp]{jr:\bm{b}}       \intss{\RRn{2}}{} \power[2]{}{K(\bm{w})} \d{\bm{w}} + \Oh{\bz{1}\vee\bz{2}}       &\kappa = 2,\\
      \Oh{\bz{1}\wedge\bz{2}}       &\kappa = 3,\\
      \Oh{\bz{1}\bz{2}}       &\kappa = 4.
    \end{cases}$
  \end{enumerate}
\end{lemma}

\begin{proof}
    Recalling the definition of $\Khb[\yh{h}-\LGp]{h}{\bm{b}}$ from
  \cref{eq:definition_of_K}, the integral in \cref{eq:int_K*U}
  can be written~as
  \begin{align}
    \intss{\RRn{2}}{}     \sqrt{\bz{1}\bz{2}} \cdot \frac{1}{\bz{1}\bz{2}} K\left(    \frac{\zetahi{}{1}-\vz{1}}{\bz{1}},
    \frac{\zetahi{}{2}-\vz{2}}{\bz{2}}
    \right) 
    \Uh[\zetahi{}{1},\zetahi{}{2}]{hq:\bm{b}}     \d{\zetahi{}{1}} \d{\zetahi{}{2}},
  \end{align}
  which implies that the substitutions
  \mbox{$\wz{1} = \left(\zetahi{}{1}-\vz{1}\right)/\bz{1}$} and
  \mbox{$\wz{2} = \left(\zetahi{}{1}-\vz{2}\right)/\bz{2}$} gives
  the~integral
  \begin{align}
    \nonumber
    &\intss{\RRn{2}}{}       \frac{\sqrt{\bz{1}\bz{2}}}{\bz{1}\bz{2}} K\left(      \wz{1}, \wz{2} \right)       \Uh[\bz{1}\wz{1}+\vz{1},\bz{2}\wz{2}+\vz{2}]{hq:\bm{b}}       \left(\bz{1}\!\d{\wz{1}}\right)\left(\bz{2}\!\d{\wz{2}}\right)\\
    \label{eq:int_K*U_intermediate}
    &= \sqrt{\bz{1}\bz{2}} \cdot       \intss{\RRn{2}}{}       K\left(\wz{1}, \wz{2} \right)       \Uh[\bz{1}\wz{1}+\vz{1},\bz{2}\wz{2}+\vz{2}]{hq:\bm{b}}       \d{\wz{1}}\!\d{\wz{2}}.
  \end{align}
  Since $\Uh[\bm{w}]{hq:\bm{b}}$ is a bivariate polynomial, it is
  clear that
  $\Uh[\bz{1}\wz{1}+\vz{1},\bz{2}\wz{2}+\vz{2}]{hq:\bm{b}}$ can be
  written~as
  \begin{align}
    \label{eq:U_simple_sum}
    \Uh[\vz{1},\vz{2}]{hq:\bm{b}} +     \bz{1}\cz{1}\wz{1} +
    \bz{2}\cz{2}\wz{2} + 
    \bz[2]{1}\cz{11}\wz[2]{1}+
    \bz{1}\bz{2}\cz{12}\wz{1}\wz{2} +
    \bz[2]{2}\cz{22}\wz[2]{2},
  \end{align}
  for suitable constants $\cz{1}$, $\cz{2}$, $\cz{11}$, $\cz{12}$
  and~$\cz{22}$.  The integral in \cref{eq:int_K*U_intermediate} can
  thus be expressed as a sum of integrals like those occurring in
  \myref{th:integrals:K_w1w2}{eq:int_|Kw_i|}, all of which are finite.
  The dominant term becomes $\Oh{\sqrt{\bz{1}\bz{2}}}$ when $\blimit$,
  and the conclusion of \cref{eq:int_K*U}~follows.
  
    The substitution used in \cref{eq:int_K*U} can also be applied for
  \cref{eq:int_(KU)^nu}, resulting~in
  \begin{align}
    \nonumber
    &\parenRz[1/\nu]{}{\intss{\RRn{2}}{}       \absp[\nu]{ \sqrt{\bz{1}\bz{2}} \cdot\frac{1}{\bz{1}\bz{2}} K\left(      \wz{1}, \wz{2} \right)       \Uh[\bz{1}\wz{1}+\vz{1},\bz{2}\wz{2}+\vz{2}]{hq:\bm{b}}}       \left(\bz{1}\!\d{\wz{1}}\right)      \left(\bz{2}\!\d{\wz{2}}\right) } \\
    \label{eq:int_(KU)^nu_intermediate}
    &= \absp[(2-\nu)/2\nu]{\bz{1}\bz{2}}       \parenRz[1/\nu]{}{\intss{\RRn{2}}{}       \absp[\nu]{ K\left(      \wz{1}, \wz{2} \right)       \Uh[\bz{1}\wz{1}+\vz{1},\bz{2}\wz{2}+\vz{2}]{hq:\bm{b}}}       \d{\wz{1}}\!\d{\wz{2}} }.
  \end{align}
  Note that this represent the norm in $\Lp{\nu}$-space, and that
  \cref{eq:U_simple_sum} implies that it can be realised as the norm
  of a sum of the simpler components encountered in
  \myref{th:integrals:K_w1w2}{eq:Kw_Lnu}.  It is now clear that
  Minkowski's inequality can be used to obtain a bound for the
  expression in \cref{eq:int_(KU)^nu_intermediate}.  In particular,
  constants $\ez{1}$, $\ez{2}$, $\ez{11}$, $\ez{12}$ and~$\ez{22}$ can
  be found that realises this bound as
  \begin{align}
    \label{eq:int_(KU)^nu_final}
    \absp[(2-\nu)/2\nu]{\bz{1}\bz{2}}     \left(        \Uh[\vz{1},\vz{2}]{hq:\bm{b}} +     \bz{1}\ez{1}\wz{1} +     \bz{2}\ez{2}\wz{2} +     \bz[2]{1}\ez{11}\wz[2]{1} +     \bz{1}\bz{2}\ez{12}\wz{1}\wz{2} +     \bz[2]{2}\ez{22}\wz[2]{2}     \right),   
  \end{align}
      which is dominated by the
  \mbox{$\absp[(2-\nu)/2\nu]{\bz{1}\bz{2}}$-term} when $\blimit$, as
  stated in \cref{eq:int_(KU)^nu}.
  
    The investigation of \cref{eq:int_KKUU} requires different
  substitutions depending on the $\kappa$ for the configuration under
  investigation.  Noting that the integrand in addition to the scaling
  factor $\bz{1}\bz{2}$ always contains the product
  $\Khb[\zetah{1}-\LGp]{h}{\bm{b}}\Khb[\zetah{2}-\LGp]{j}{\bm{b}}$, it
  follows that it regardless of the value of $\kappa$ will be a factor
  $1/\bz{1}\bz{2}$ that will be adjusted by the $\bz{1}$- and
  \mbox{$\bz{2}$-factors} that originates from the substituted
  differentials.  It is easy to check that the new differentials
  becomes   $\bz{1}\bz{2}\d{\wz{1}}\!\d{\wz{2}}$ when \mbox{$\kappa=2$},   $\bz[2]{1}\bz{2}\d{\wz{1}}\!\d{\wz{2}}\!\d{\wz{3}}$ or
  $\bz{1}\bz[2]{2}\d{\wz{1}}\!\d{\wz{2}}\!\d{\wz{3}}$ when
  \mbox{$\kappa=3$},   and
  $\bz[2]{1}\bz[2]{2}\d{\wz{1}}\!\d{\wz{2}}\!\d{\wz{3}}\!\d{\wz{4}}$
  when \mbox{$\kappa=4$}.  

  For the bivariate case, the substitution from \cref{eq:int_K*U}
  gives an expression of the following form,
  \begin{align}
    \label{eq:int_KKUU_bivariate}
    \intss{\RRn{2}}{} \power[2]{}{K\!\left(\wz{1},\wz{2}\right)}     \cdot \mathcal{U}\!\left(\wz{1},\wz{2}\right) \d{\wz{1}}\!\d{\wz{2}},
  \end{align}
  where $\mathcal{U}\!\left(\wz{1},\wz{2}\right)$ is a product whose
  factors both are of the form encountered in \cref{eq:U_simple_sum},
  i.e.\ it will be a quartic polynomial in the variables
  $\left(\bz{1}\wz{1}\right)$ and $\left(\bz{2}\wz{2}\right)$, and its
    constant term will be $\Uh[\LGp]{hq:\bm{b}} \Uh[\LGp]{jr:\bm{b}}$.
  From
  \myref{th:integrals_kernel_and_score_components}{eq:int_|K2w_i|} it
  follows that this will be a finite integral, and as $\blimit$ the
  result will be as given for the \mbox{$\kappa=2$} case of
  \cref{eq:int_KKUU}.

  For the trivariate case, the overlap between $\zetah{1}$ and
  $\zetah{2}$ will belong to one of the following configurations,   (i) \mbox{$\zetah{1} = \left(\zetahi{}{1}, \zetahi{}{2}\right)$}
  and \mbox{$\zetah{2} = \left(\zetahi{}{1}, \zetahi{}{3}\right)$},   (ii) \mbox{$\zetah{1} = \left(\zetahi{}{1}, \zetahi{}{2}\right)$}
  and \mbox{$\zetah{2} = \left(\zetahi{}{3}, \zetahi{}{1}\right)$},   (iii) \mbox{$\zetah{1} = \left(\zetahi{}{1}, \zetahi{}{2}\right)$} and
  \mbox{$\zetah{2} = \left(\zetahi{}{2}, \zetahi{}{3}\right)$},   or   (iv) \mbox{$\zetah{1} = \left(\zetahi{}{1}, \zetahi{}{2}\right)$}
  and \mbox{$\zetah{2} = \left(\zetahi{}{3}, \zetahi{}{2}\right)$}.     The reasoning is identical for the four cases, so it is sufficient
  to consider case~(i), which gives the following product of kernel
  functions in the original integral,
  \begin{align}
    \label{eq:trivaritate_kernel_product}
    K\!\left(\left(\zetahi{}{1}-\vz{1}\right)\!/\bz{1},     \left(\zetahi{}{2}-\vz{2}\right)\!/\bz{2}\right)     \cdot     K\!\left(\left(\zetahi{}{2}-\vz{1}\right)\!/\bz{1},     \left(\zetahi{}{3}-\vz{2}\right)\!/\bz{2}\right).                       \end{align}
  When the substitution   \begin{align}
    \label{eq:trivariate_substitution}
    \wz{1} = \left(\zetahi{}{1}-\vz{1}\right)\!/\bz{1}, \qquad 
    \wz{2} = \left(\zetahi{}{2}-\vz{2}\right)\!/\bz{2}, \qquad
    \wz{3} = \left(\zetahi{}{3}-\vz{2}\right)\!/\bz{2}, 
  \end{align}
  is used, the following component occurs in the transformed integrand,
  \begin{align}
    \label{eq:messy_trivaritate_kernel_product}
    \mathcal{K}\left(\wz{1},\wz{2},\wz{3}\right) \defeq    K\!\left(\wz{1}, \wz{2}\right)     \cdot     K\!\left(
    \left[\left(\bz{2}\wz{2}+\vz{2}\right)-\vz{1}\right]/\bz{1}, \wz{3} \right).   \end{align}
  The argument
  \mbox{$\left[\left(\bz{2}\wz{2}+\vz{2}\right)-\vz{1}\right]/\bz{1}$}
  does not pose a problem due to the boundedness requirement from
  \cref{eq:kernel_integrals_finite} in \cref{def:kernel}, and the
  following inequality thus holds for
  \mbox{$\ell \in \TSR{0,1,2}{}{}$},
  \begin{subequations}
    \label{eq:messy_trivaritate_kernel_product_resolved}
    \begin{align}
      \intss{\RRn{1}}{} \mathcal{K}\left(\wz{1},\wz{2},\wz{3}\right)
      \wz[\ell]{3} \d{\wz{3}} 
      &=         K\!\left(\wz{1}, \wz{2}\right) \cdot        \intss{\RRn{1}}{}     K\!\left(
        \left[\left(\bz{2}\wz{2}+\vz{2}\right)-\vz{1}\right]/\bz{1}, \wz{3} \right)
        \wz[\ell]{3} \d{\wz{3}} \\
      &=       K\!\left(\wz{1}, \wz{2}\right) \cdot         \mathcalKz{2:\ell} \! \left(        \left[\left(\bz{2}\wz{2}+\vz{2}\right)-\vz{1}\right]/\bz{1}
        \right) \\
      &\leq         \mathcalDz{2:\ell} \cdot K\!\left(\wz{1}, \wz{2}\right),
    \end{align}
  \end{subequations}
    where $\mathcalDz{2:\ell}$ is a constant that bounds the function
  $\mathcalKz{2:\ell}$.  

  Since the substitution in \cref{eq:trivariate_substitution}
  transforms the integral of interest~into
  \begin{align}
    \label{eq:int_KKUU_trivariate}
    \bz{2}\intss{\RRn{3}}{} \mathcal{K}\left(\wz{1},\wz{2},\wz{3}\right)     \cdot \mathcal{U}\!\left(\wz{1},\wz{2},\wz{3}\right)     \d{\wz{1}}\!\d{\wz{2}}\!\d{\wz{3}},
  \end{align}
  where $\mathcal{U}\!\left(\wz{1},\wz{2},\wz{3}\right)$ is a
  quadratic polynomial in the variables $\left(\bz{1}\wz{1}\right)$
  and $\left(\bz{2}\wz{3}\right)$, and a quartic polynomial in
  $\wz{2}$ (with coefficients having suitable powers of $\bz{1}$ and
  $\bz{2}$ as factors), the observation in
  \cref{eq:messy_trivaritate_kernel_product_resolved} implies that an
  iterated approach to the integral (starting with the
  \mbox{$\wz{3}$-variable}) can be used to show that each part of the
  sum will be bounded by a constant times an integral of the form
  encountered in
  \myref{th:integrals_kernel_and_score_components}{eq:int_|Kw_i|}.
  The trivariate integral in \cref{eq:int_KKUU} can thus be bounded by
  a sum of finite integrals having coefficients based on powers of
  $\bz{1}$ and $\bz{2}$.  From the $\bz{2}$ factor in
  \cref{eq:int_KKUU_trivariate}, it follows that the trivariate
  integral in this case is $\Oh{\bz{2}}$ when $\blimit$.  Note that
  \mbox{$\wz{2} = \left(\zetahi{}{2}-\vz{1}\right)\!/\bz{1}$} could
  have been used as an alternative substitution in
  \cref{eq:trivariate_substitution}, which by the obvious
  modifications of the arguments implies that the integral also will
  be $\Oh{\bz{1}}$ when $\blimit$ --- and from this if follows that
    the integral is $\Oh{\bz{1}\wedge\bz{2}}$, which completes the proof
  for the \mbox{$\kappa=3$} case of \cref{eq:int_KKUU}.

  The case \mbox{$\kappa=4$} is quite simple, since no common
  components in $\zetah{1}$ and $\zetah{2}$ implies that the
  tetravariate integral, after the obvious substitution, corresponds
  to an expression of the form
  \begin{align}
    \label{eq:int_KKUU_tetravariate}
    \bz{1}\bz{2} \left( 
    \intss{\RRn{2}}{} K(\bm{w})
    \Uh[\bm{\zeta}(\bm{w})]{hq:\bm{b}}
        \d{\bm{w}}
    \right) \cdot   \left( 
    \intss{\RRn{2}}{} K(\bm{w})
    \Uh[\bm{\zeta}(\bm{w})]{jr:\bm{b}}
    \d{\bm{w}}
    \right),
  \end{align}
  where
  \mbox{$\bm{\zeta}(\bm{w}) =
    \left(\bz{1}\wz{1}+\vz{1},\bz{2}\wz{2}+\vz{2}\right)$}.    The integrals occurring in this product are similar to those
  encountered in the bivariate case discussed above, and it is clear
  that the result will be $\Oh{\bz{1}\bz{2}}$ when $\blimit$, which
  concludes the proof of \cref{eq:int_KKUU}.
\end{proof}

Note that the bivariate case of
\myref{th:integrals_kernel_and_score_components}{eq:int_KKUU} only
considers the configuration where the components of $\zetah{1}$ and
$\zetah{2}$ coincide completely, while the configuration where
\mbox{$\bmzetaz{1} = \left(\zetahi{}{1},\zetahi{}{2}\right)$} and
$\bmzetaz{2}$ is the diagonal reflection
\mbox{$\left(\zetahi{}{2},\zetahi{}{1}\right)$} has been left out.
This restriction does not pose a problem for the asymptotic
investigation of $\hatlgsdM[p]{\LGp}{\omega}{m}$ when the point
\mbox{$\LGp=\parenR{\LGpi{1},\LGpi{2}}$} lies upon the diagonal, i.e.\
when \mbox{$\LGpi{1} = \LGpi{2}$}, since the diagonal folding property
ensures that it is sufficient to consider positive lags for the
point~$\LGp$ in this case.  For the general case, where
\mbox{$\LGpi{1} \neq \LGpi{2}$}, the following adjusted version of
\myref{th:integrals_kernel_and_score_components}{eq:int_KKUU} is
needed, where one of the kernels use~$\LGp$ and the other use the
diagonally reflected point~\mbox{$\LGpd=\parenR{\LGpi{2},\LGpi{1}}$}.

\begin{lemma}
  \label{th:integrals_kernel_and_score_components_LGP_and_LGpd}
  The following holds for $\Uh[\bm{w}]{hq:\bm{b}}$ and
  $\Khb[\yh{h}-\LGp]{h}{\bm{b}}$ from \cref{def:Uh,def:kernel}, when
  the point \mbox{$\LGp=\parenR{\LGpi{1},\LGpi{2}}$} does not
  coincide with its diagonal reflection
  \mbox{$\LGpd=\parenR{\LGpi{2},\LGpi{1}}$}, i.e.\
  \mbox{$\LGpi{1} \neq \LGpi{2}$}.\\
    Let
  \mbox{$\mathcalKz{qr,hj:\bm{b}}\!\left(\zetah{1},\zetah{2};\LGp,\LGpd\right)
    \defeq     \Khb[\zetah{1}-\LGp]{h}{\bm{b}}     \Khb[\zetah{2}-\LGpd]{j}{\bm{b}}     \Uh[\zetah{1}]{hq:\bm{b}}     \Uh[\zetah{2}]{jr:\bm{b}}$}, where   $\zetah{1}$ and $\zetah{2}$ either are diagonal reflections of each
  other (bivariate), have one common component (trivariate), or have
  no common components (tetravariate).   Let $\kappa$ be the number of variates, and let
  $\d{\bm{\zeta}(\kappa})$ represent the corresponding
  \mbox{$\kappa$-variate} differential.  Then,\\
      $$\intss{\RRn{\kappa}}{}     \left(\bz{1}\bz{2}\right)     \mathcalKz{qr,hj:\bm{b}}\!\left(\zetah{1},\zetah{2};\LGp,\LGpd\right)     \d{\bm{\zeta}(\kappa}) =     \begin{cases}
      \oh{1}       &\kappa = 2,\\
      \Oh{\bz{1}\wedge\bz{2}}       &\kappa = 3,\\
      \Oh{\bz{1}\bz{2}}       &\kappa = 4.
    \end{cases}$$
\end{lemma}

\begin{proof}
  The statements for the trivariate and tetravariate cases are
  identical to those in
  \myref{th:integrals_kernel_and_score_components}{eq:int_KKUU}, and
  so are the proofs, i.e.\ the same substitutions can be applied for
  the present cases of interest.
  
  For the bivariate case, the substitution
  \mbox{$\wz{1} = \left(\zetahi{}{1}-\vz{1}\right)/\bz{1}$} and
  \mbox{$\wz{2} = \left(\zetahi{}{1}-\vz{2}\right)/\bz{2}$} gives that
  the integral
  \mbox{$\intss{\RRn{2}}{} \subp{K(\wz{1},\wz{2})}{}{}{}{2}\cdot
    \mathcal{U}\!\left(\wz{1},\wz{2}\right) \d{\wz{1}}\d{\wz{2}}$}
  from \cref{eq:int_KKUU_bivariate} is replaced with a sum of
  integrals of the form,
  \begin{align}
    \label{eq:kernels_going_to_zero}
    \intss{\RRn{2}}{}
    K(\wz{1}+(\vz{1}-\vz{2})/\bz{1},
    \wz{2}+(\vz{2}-\vz{1})/\bz{2}) \cdot K(\wz{1},\wz{2}) \wz[k]{1}\wz[\ell]{2}
    \d{\wz{1}}\d{\wz{2}},
  \end{align}
  where \mbox{$k, \ell \geq 0$} and \mbox{$k+\ell \leq 4$}.  and the
  integrands of these integrals goes to zero when $\blimit$, due to
  the assumption that \mbox{$\LGpi{1} \neq \LGpi{2}$}.  To clarify:
  For a kernel function $K$ whose nonzero values occurs on a bounded
  region of $\RRn{2}$, the integrand of
  \cref{eq:kernels_going_to_zero} will become identical to zero when
  \mbox{$(\vz{1}-\vz{2})/\bz{1}$} and \mbox{$(\vz{2}-\vz{1})/\bz{2}$}
  are large enough to ensure that at least one of the factors in the
  integrand must be zero.  
  For the general case, first observe that the factors
  $K(\wz{1},\wz{2})\wz[k]{1}\wz[\ell]{2}$ are the integrands that
  occurs in \myref{th:integrals:K_w1w2}{eq:int_|Kw_i|}, and the
  finiteness of those integrals implies that these factors must go to
  zero at a sufficiently high rate when $\wz{1}$ and $\wz{2}$ are far
  from origo.  The rate at which the individual kernel
  $K(\wz{1},\wz{2})$ goes to zero will of course be faster than that
  of the product $K(\wz{1},\wz{2})\wz[k]{1}\wz[\ell]{2}$, and together
  this implies that the integrand in \cref{eq:kernels_going_to_zero}
  must go to zero when $\blimit$, and the integral thus becomes
  asymptotically negligible.
\end{proof}

  It is a straightforward (albeit somewhat tedious) exercise to verify
  that \cref{eq:kernels_going_to_zero} goes towards zero at an
  exponential rate when the kernel function $K(\bm{w})$ is the product
  normal kernel.  The observation that the bivariate case of
  \cref{th:integrals_kernel_and_score_components_LGP_and_LGpd} is
  $\oh{1}$ can also be derived from the realisation that
  $\Khb[\zetah{1}-\LGp]{h}{\bm{b}}$ and
  $\Khb[\zetah{2}-\LGpd]{j}{\bm{b}}$ are entities that converge
  towards two different bivariate Dirac delta functions, and the limit
  of the integral becomes zero since these delta functions sifts out
  different~points.

\makeatletter{}
\subsection{A few details related to $\sigma$-algebras,
  $\alpha$-mixing and $\Lp{\nu}$-spaces}
\label{app:sigma_algebras_Lp_spaces}

The following general definitions and basic observations are needed
when e.g.\ results from~\citet{Davydov:1968:CDG} and
\citet{Volkonskii:1959:SLT} are used.

\textbf{Related $\sigma$-algebras}\\
The $\sigma$-algebras related to the process
$\TSR{\Yz{t}}{t\in\ZZ}{}$, will be denoted
\begin{align}
  \label{eq:sigma_algebra_Yht_range} 
  \sigmaYhtRange{t}{s} &\defeq 
  \sigma\!\left(\Yz{t},\dotsc,\Yz{s}\right),
\end{align}
where $t$ and $s$ are allowed to take the values $-\infty$ and
$+\infty$ respectively.

Note in particular, that if a new random variable is defined by means
of a measurable function $\xi\!\left(\yM{m}\right)$ from $\RRn{m+1}$
to $\RR$, i.e.\
\mbox{$\mathcalYz{m:t} \defeq
  \xi\!\left(\YMt{m}{t}\right)$},
then
\mbox{$\mathcalYz{m:t} \in \sigmaYhtRange{t}{t+m}$}.

\textbf{Inheritance of $\alpha$-mixing} \\
The coefficients in the strong mixing property mentioned in
\myref{assumption_Yt}{assumption_Yt_strong_mixing}, is given by
\begin{align}
  \label{eq:alpha_mixing_definition_for_Yt}
  \alpha\!\left(s\,|\,\Yz{t} \right)  \defeq   \sup \parenC{  \parenAbs{\Prob{A\cap B} - \Prob{A}\Prob{B}}\,:\,
  -\infty < t < \infty,   \ 
  A \in \sigmaYhtRange{\!-\infty}{t},   \ 
  B \in \sigmaYhtRange{t+s}{\infty} },
\end{align}
from which it is an easy task to verify that a derived process, like
the $\mathcalYz{m:t}$ mentioned above, will have an
inherited $\alpha$-mixing coefficient that satisfies
\begin{align}
  \label{eq:alpha_mixing_definition_derived_from_Yt}
  \alpha\!\left(s\,|\, \mathcalYz{m:t} \right)  \leq   \alpha\!\left(s-m\,|\,\Yz{t} \right).
\end{align}
This implies that the finiteness requirement in
\cref{eq:alpha_requirement} will be inherited by the process
$\mathcalYz{m:t}$, i.e.\ with $\nu$ and $a$ as
introduced in \myref{assumption_Yt}{assumption_Yt_strong_mixing}, the
following holds true
\begin{align}
  \label{eq:alpha_requirement_derived_from_Yt}
  \sumss{j=1}{\infty} \jz[a]{}   \parenSz[1-2/\nu]{}{\alpha(j\,|\,\mathcalYz{m:t})} < \infty.
\end{align}

\textbf{Related \mbox{$\Lp{\nu}$-spaces}} \\
Some inequalities are needed in the main proofs, and these
inequalities can be verified by means of the simple connection between
expectations and \mbox{$\Lp{\nu}$-spaces} outlined below.\footnote{
  These definitions are normally presented with $p$ used instead of
  $\nu$.}

First of all, when a measure space \mbox{$\left(\Omega, \mathcal{G},
    \mu\right)$} is given, then for \mbox{$1\leq \nu < \infty$}, the
space \mbox{$\Lp{\nu} \defeq \Lp{\nu}\!\left(\Omega, \mathcal{G},
    \mu\right)$} is defined to be the class of measurable real
functions $\zeta$ for which $\absp[\nu]{\zeta}$ is integrable, that is,
\begin{align}
  \label{eq:Lp_def}
  \zeta(z) \in \Lp{\nu} \quad \defarrow \quad
  \intss{\Omega}{} \absp[\nu]{\zeta(z)} \d{\mu} < \infty.
\end{align}
The \mbox{$\Lp{\nu}$-spaces} related to the processes $\Yht{h}{t}$ and
$\YMt{m}{t}$ will henceforth be denoted by
\begin{subequations}
  \label{eq:definition_Lp-Yht-YMt}
  \begin{align} 
  \label{eq:definition_Lp-Yht}
    \Lph{h}{\nu} \quad &\text{--- the $\Lp{\nu}$ spaces related to the
                       densities~$\gh{h}$,} \\
  \label{eq:definition_Lp-YMt}
    \LpM{m}{\nu} \quad &\text{--- the $\Lp{\nu}$ space related to the
                       density $\gM{m}$.}        \end{align}
\end{subequations}

These $\Lp{\nu}$ spaces are in fact Banach spaces, see e.g.\
\citet[Section~19]{Billingsley12:_probab_measur} for details, which
means that they are complete normed vector spaces, with a
\mbox{$\nu$-norm} defined~by
\begin{align}
  \label{eq:Lp_norm}
  \parenABSz{\nu}{\zeta(z)} 
  &\defeq     \parenRz[1/\nu]{}{    \intss{\Omega}{} \absp[\nu]{\zeta(z)} \d{\mu}}     = \parenRz[1/\nu]{}{ \E{\absp[\nu]{\zeta(Z)}}}
\end{align}
and the Minkowski's inequality (i.e.\ the triangle inequality for
\mbox{$\Lp{\nu}$-spaces}) will play a central role in the investigation
later~on,
\begin{align}
  \label{eq:Minkowski's_inequality}
  \parenABSz{\nu}{\zetaz{1}(z) + \zetaz{2}(z)} &\leq
  \parenABSz{\nu}{\zetaz{1}(z)} +
  \parenABSz{\nu}{\zetaz{2}(z)}.
\end{align}

The main reason for the introduction of these \mbox{$\Lp{\nu}$-spaces}
are the following observation: With $Z$ a random variable on
\mbox{$\left(\Omega, \mathcal{G}, \mu\right)$}, the definitions of
expectation and \mbox{$\Lp{\nu}$-spaces} gives a sequence of
equivalences
\begin{align}
  \label{eq:expectation_and_Lph}
  \E{\absp[\nu]{\zeta(Z)}} < \infty \quad   \Longleftrightarrow \quad   \intss{\Omega}{}   \absp[\nu]{\zeta(z)} \d{\mu} < \infty   \quad \Longleftrightarrow \quad   \zeta(z) \in \Lp{\nu}.
\end{align}

\begin{lemma}
  \label{th:Lp_expectation}
          For a univariate time series $\TSR{\Yz{t}}{t\in\ZZ}{}$, with
  $\Yht{h}{t}$ and $\YMt{m}{t}$ as defined in \cref{def:Yht_and_YMt},
  and with $m$ bivariate functions
  \mbox{$\zetahi{h}{}:\RRn{2} \longrightarrow \RRn{1}$}
    \begin{center}
    If \mbox{$\E{\absp[\nu]{ \zetahi[\Yht{h}{t}]{h}{}}} < \infty$} for
    \mbox{$h=1,\dotsc,m$}, then \mbox{$\parenRz[1/\nu]{}{        \E{\absp[\nu]{\sumss{h=1}{m} \az{h} \zetahi[\Yht{h}{t}]{h}{} }}
      } \leq       \sumss{h=1}{m} \absp{\az{h}}
      \parenRz[1/\nu]{}{\E{\absp[\nu]{\zetahi[\Yht{h}{t}]{h}{}}}}  <
      \infty$}.
  \end{center}
\end{lemma}

\begin{proof}
  From \cref{eq:expectation_and_Lph} it follows that
  \mbox{$\E{\absp[\nu]{ \zetahi[\Yht{h}{t}]{h}{}}} < \infty$} implies
  \mbox{$\zetahi[\yh{h}]{h}{} \in \Lph{h}{\nu}$} for
  \mbox{$h=1,\dotsc,m$}.  With $\tildezetahi[\yM{m}]{h}{}$ the
  corresponding trivial extensions to \mbox{$(m+1)$-variate}
  functions, it follows from \cref{eq:expectation_hierarchy} that
  \mbox{$\tildezetahi[\yM{m}]{h}{} \in \LpM{m}{\nu}$} for
  \mbox{$h=1,\dotsc,m$}.    From the vector space property of \mbox{$\Lp{\nu}$-spaces} it follows
  that
  \mbox{$\sumss{h=1}{m} \az{h} \zetahi[\Yht{h}{t}]{h}{} \in
    \LpM{m}{\nu}$}, and Minkowski's inequality then gives the
  desired~result.
\end{proof}

\section{Sensitivity analysis of the tuning parameters}
\label{app:sensitivity_analysis}
\setcounter{figure}{0} 

This section will investigate how sensitive
$\hatlgsdM[5]{\LGp}{\omega}{m}$ is to changes in the tuning
parameters (and the point $\LGp$).  This will be done by the
distance function $D$ introduced in
\cref{app:method_for_sensitivity_analysis}, together with plots that
reveal information about the frequency-dimension.

\Cref{app:Point_sensitivity,app:Bandwidth_sensitivity} respectively
consider the sensitivity of the point $\LGp$ and the bandwidth
$\bm{b}$, whereas the sensitivity of the truncation level $m$ is
discussed in \cref{app:Truncation_level_sensitivity}.
The effect the value of the block length $L$ has upon the
bootstrap-based pointwise confidence intervals is discussed in
\cref{app:Block_length_sensitivity}, since that gives the most
natural flow.

\Cref{How.to.select.the.tuning.parameters?} contains a discussion
related to the \textit{selection of tuning parameters} for
$\hatlgsdM[5]{\LGp}{\omega}{m}$,
and it also contains some references to the related problem of
selecting the bandwidth when a local Gaussian correlation is to be
estimated from a sample.

The scripts required for the replication of the results in this
section are contained in the \Rpackage \lgsdRpackage, and these
scripts can be used as templates for those that would like to
investigate other time series in a similar manner.  See
\cref{app:data_details} for details.

\subsection{Sensitivity analysis - the distance function}
\label{app:method_for_sensitivity_analysis}

An investigation of the sensitivity requires a tool that can measure
the differences that occur in the resulting estimates when the
tuning parameters are adjusted.  
Many techniques have been developed in order to deal with distances
between spectral functions, cf.\ e.g.\
\citet[Section~7]{BASSEVILLE2013621} and
\citet[Section~1]{georgiou07:_distan_rieman_metric_spect_densit_funct}.
Some approaches are based on proper distance functions, whereas other
use divergence/distortion measures where symmetry and the triangular
identity no longer are present.

A natural (and easy to implement) candidate for the case of interest
in this paper is the distance function inherited from the complex
Hilbert space of Fourier series on the interval
$\parenS{-\tfrac{1}{2},\tfrac{1}{2}}$, cf.\ e.g.\
\citet[Ch.~2.8]{Brockwell:1986:TST:17326}, i.e.\ for
$f(\omega)=\sumss{h=-\infty}{\infty}\rho(h)\ez[-2\pi i h]{}$ the norm
is defined by
$\subp{||f(\omega)||}{}{}{}{2} = \intss{-1/2}{1/2}
f(\omega)\overbar{f(\omega)}\d{\omega} = \sumss{h=-\infty}{\infty}
\subp{\rho(h)}{}{}{}{2}$.  This motivates the following definition.

\begin{definition}
  \label{def:sensitivity_analysis_distance_function}
  Given two spectra
  $\fz{1}(\omega)=\sumss{h=-\infty}{\infty}\rhoz{1}(h)\ez[-2\pi i
  h]{}$ and
  $\fz{2}(\omega)=\sumss{h=-\infty}{\infty}\rhoz{2}(h)\ez[-2\pi i
  h]{}$, the distance between them is denoted by
  \begin{equation}
    \label{eq:def:sensitivity_analysis_distance_function}
    D(\fz{1}(\omega),\fz{2}(\omega)) \defeq \sqrt{\sumss{h=-\infty}{\infty}
      \subp{\parenR{\rhoz{1}(h)-\rhoz{2}(h)}}{}{}{}{2} }.
  \end{equation}
  Furthermore: The notation $D(\fz{1}(\omega))$ will be interpreted
  as $D(\fz{1}(\omega),0)$, which implies that
  $D(\fz{1}(\omega),\fz{2}(\omega))$ also can be written as
  $D(\fz{1}(\omega) -\fz{2}(\omega))$ (which is used in \cref{fig:m_sensitivity_dmbp_percentages}).
\end{definition}

Note that $D$ will work both for real-valued and complex-valued
spectra, and the latter is of importance both with regard to the
univariate case when the point $\LGp$ lies of the diagonal, and with
regard to the multivariate case treated in
\myblind{\citet{jordanger17:_lgcsd}}{blinded reference}.

The obvious adjustment must be done when $D$ is used on
$m$-truncated estimates $\hatlgsdM{\LGp}{\omega}{m}$, i.e.\
$\rho(h)$ should be replaced with
$\lambdazM{h}{m}\cdot\hatlgacr{\LGp}{h}$ when
$|h|\leq m$, and with 0 when $|h|>m$.

The distance function in
\cref{def:sensitivity_analysis_distance_function} is not applicable
in the FFT-periodogram based approach to the estimation of spectral
densities, since that approach does not explicitly compute the
coefficients needed in
\cref{eq:def:sensitivity_analysis_distance_function}.
However, note that the deviance measure that is based on the
\textit{root mean squared error} (RMSE), cf.\ e.g.\
\citet[Section~3.2]{chen2019semiparametric}, is closely related to the
one used in the present paper.  To emphasise: If $\tildefz{1}(\omega)$
and $\tildefz{2}(\omega)$ are two periodogram-based estimates of the
spectral densities $\fz{1}(\omega)$ and $\fz{2}(\omega)$, then the
RMSE-distance is given by
\begin{equation}
  \label{eq:distance_RMSE}
  \Dz{\operatorname{RMSE}}\!\parenR{\tildefz{1}(\omega),
    \tildefz{2}(\omega)} \defeq \sqrt{\frac{1}{n}\sumss{l=0}{n-1}
    \subp{\parenS{\tildefz{1}(\omegaz{l}) -
        \tildefz{2}(\omegaz{l})}}{}{}{}{2}},
\end{equation}
where the summation is over all the Fourier-frequencies
$\omegaz{l}=l/n$ in the interval $[0,1)$.  A quick inspection of the
expression under the square-root in \cref{eq:distance_RMSE} reveals
that this is a Riemann-sum approximation of the integral
$\intss{0}{1} \parenR{\fz{1}(\omega) -
  \fz{2}(\omega)}\overbar{\parenR{\fz{1}(\omega) -
    \fz{2}(\omega)}}\d{\omega}$.  This will, when $\nlimit$, converge
towards $\subp{||\fz{1}(\omega) - \fz{2}(\omega)||}{}{}{}{2}$, which
shows the close connection with the distance function from
\cref{def:sensitivity_analysis_distance_function}.

Another more commonly used divergence measure is also considered in
\citet{chen2019semiparametric}, and that is the divergence measure
based on the Kullback-Leibler (KL) divergence \citet{kullback1951}.
For the periodogram-based approach this can be written as
\begin{equation}
  \label{eq:distance_KL}
  \Dz{\operatorname{KL}}\!\parenR{\tildefz{1}(\omega),
    \tildefz{2}(\omega)} \defeq \sumss{l=0}{n-1}
  \tildefz{1}(\omegaz{l}) \ln
  \parenR{\frac{\tildefz{1}(\omegaz{l})}{\tildefz{2}(\omegaz{l})}}.
\end{equation}
An implementation of the KL-approach in this paper was briefly
considered, but it was discarded since the local Gaussian spectral
densities $\lgsd{\LGp}{\omega}$ in general will be complex-valued
functions, and it was thus not clear how to adjust
\cref{eq:distance_KL} in a proper manner.

\textbf{Regarding the frequency-dimension:} 
A distance measure like the $D$ defined in
\cref{def:sensitivity_analysis_distance_function} does not contain
any information about the frequencies, and completely different
spectral densities can have the same distance-value.  It is thus,
for the purpose of sensitivity analysis, important to combine
distance-based plots with plots that reveal something about the
frequency-component too.

\subsection{Sensitivity analysis: The point $\LGp$}
\label{app:Point_sensitivity}

The bandwidth $\bm{b}$ 
is a central tuning parameter when an estimate of the $m$-truncated
local Gaussian spectral density $\lgsdM[p]{\LGp}{\omega}{m}$ is
desired for a given point $\LGp=\LGpoint$.  The point $\LGp$ itself is
not a tuning parameter of the estimation algorithm, but an
investigator will obviously be interested in information about how
$\lgsdM[p]{\LGp}{\omega}{m}$ varies with $\LGp$, and it is thus also
natural to consider the sensitivity of the estimate
$\hatlgsdM[p]{\LGp}{\omega}{m}$ relatively the selected point.

Two plots related to this particular investigation was included in
the main part, i.e.\
\cref{fig:dmbp_heatmap_and__levels_vs_norm,fig:dmt_heatmap_and__levels_vs_norm},
which respectively considered the \textit{local trigonometrical
  example} and the \texttt{dmbp}-data.  It is preferable to have a
plot available for the present discussion too, and
\cref{fig:apARCH_heatmap_and__levels_vs_norm} shows an example based
on one single simulation from the apARCH$(2,3)$ that was fitted to
the \texttt{dmbp}-data, cf.\ \cref{sec:Real_data}.

\begin{figure}[h]
  {\centering \includegraphics[width=1\textwidth]{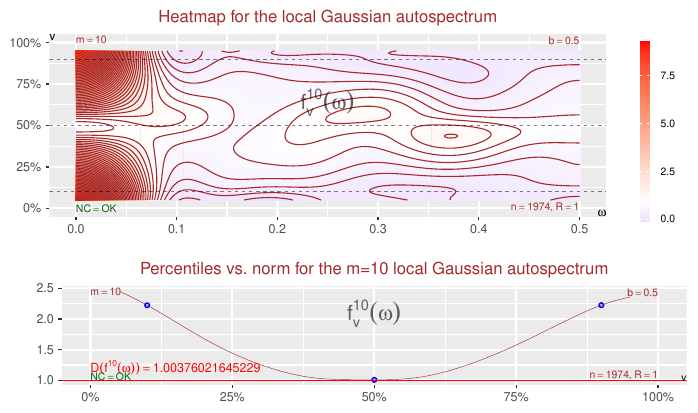}
  }
  \caption[]{Heatmap and corresponding distance-based plots for the
    the apARCH$(2,3)$-model that was fitted to the
    \texttt{dmbp}-data, showing how $\hatlgsdM{\LGp}{\omega}{10}$
    varies with the percentiles for the diagonal-points $\LGp$.  The
    percentiles used in \cref{fig:GARCH}, i.e.\ 10\%, 50\% and 90\%,
    have been highlighted with
    lines/points.}\label{fig:apARCH_heatmap_and__levels_vs_norm}
\end{figure}

The point $\LGp=\LGpoint$ is bivariate, but the present
investigation will restrict its attention to the diagonal cases.
The requirement $\LGpi{1} = \LGpi{2}$ is used for simplicity since
it ensures that the resulting local Gaussian spectral densities will
be real valued.

This restriction implies that the point $\LGp$ is allowed to vary
continuously along a one dimensional line (the diagonal), and a
heatmap can be used to see how $\hatlgsdM[p]{\LGp}{\omega}{m}$
varies with $\LGp$ (for a fixed~$\bm{b}$).  It is also of interest
to use the distance function $D$ from
\cref{def:sensitivity_analysis_distance_function} to create a
distance-based plot that shows how the norm
$D\!\parenR{\hatlgsdM[p]{\LGp}{\omega}{m}}$ varies with $\LGp$.

The points $\LGp$ in \cref{fig:apARCH_heatmap_and__levels_vs_norm}
ranges from the 5\% percentile to the 95\% percentile of the
standard normal distribution, increasing in steps of 0.5\%
(altogether 91 different points).  This percentile based selection
implies that the corresponding points are not equally spaced along
the actual diagonal, and the plots in
\cref{fig:apARCH_heatmap_and__levels_vs_norm} have thus used the
option that the points $\LGp$ have been presented according to their
underlying percentile-values --- which implies that these plots
primarily reveals information about the 
copula-structure of the time series under investigation.

It can be seen from \cref{fig:apARCH_heatmap_and__levels_vs_norm}
that $\hatlgsdM{\LGp}{\omega}{m}$ near the 50\% percentile is quite
close to an i.i.d.\ white noise situation --- and it also seems to
be a very clear symmetry around the 50\% percentile.  This is in
stark contrast to the situation seen for the \texttt{dmbp}-data,
cf.\ \cref{fig:dmbp_heatmap_and__levels_vs_norm}, which indicates an
asymmetry around the 50\% percentile

Note that the 5\% and 95\% percentiles are quite far out in the
tails of the distribution, and it is thus natural to assume that the
selected bandwidth 
in those cases might fail to work properly --- the small sample
variation of the points closest to the point $\LGp$ might simply
render the estimated local Gaussian autocorrelations rather dubious.
It is possible to counter this problem by selecting a larger
bandwidth for percentiles in the tails, but it is then important to
keep in mind that a too large bandwidth might completely miss the
desired local structure at the point of investigation.

\textbf{Heatmap-plots for the estimates $\hatlgacr{\LGp}{h}$:} The
construction of the two plots in
\cref{fig:apARCH_heatmap_and__levels_vs_norm} requires the
computation of all of the underlying estimates $\hatlgacr{\LGp}{h}$,
for $h=1,\dotsc,m$.  It is thus also possible to create
heatmap-based plots that can visualise how these estimates changes
as the point $\LGp$ moves from the 5\% to the 95\% percentile, and
this can for the apARCH$(2,3)$-example be seen in
\cref{fig:apARCH_v_heatmap_lgacr}.

\begin{figure}[h]
  {\centering \includegraphics[width=\textwidth]{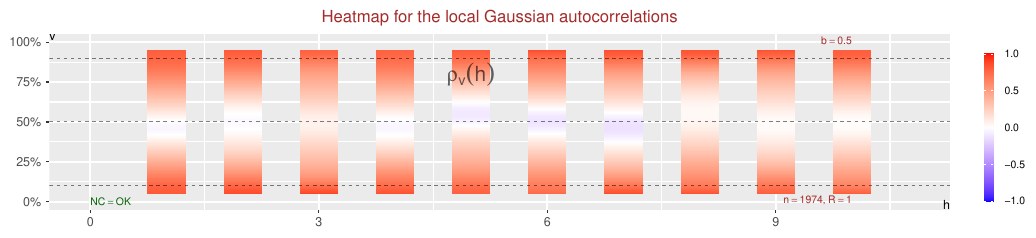}
  }
  \caption[]{Heatmap for $\hatlgacr{\LGp}{h}$, for the
    apARCH$(2,3)$-model fitted to the \texttt{dmbp}-data.  The
    percentiles used in \cref{fig:GARCH}, i.e.\ 10\%, 50\% and 90\%,
    have been highlighted with lines.}
  \label{fig:apARCH_v_heatmap_lgacr}
\end{figure}

It is clear from \cref{fig:apARCH_v_heatmap_lgacr} that the
estimated values $\hatlgacr{\LGp}{h}$ are near symmetric around the
50\% percentile, which thus explains the corresponding symmetry for
$\hatlgsdM{\LGp}{\omega}{m}$ seen in
\cref{fig:apARCH_heatmap_and__levels_vs_norm}.  For the
\texttt{dmbp}-data, see \cref{fig:dmbp_v_heatmap_lgacr}, a similar
level of symmetry is not to the same extent present.  It might from
such plots be possible to identify if it is the contribution from
some particular lags $h$ that drives the asymmetry of the
corresponding estimated spectrum $\hatlgsdM{\LGp}{\omega}{m}$.

\begin{figure}[h]
  {\centering \includegraphics[width=\textwidth]{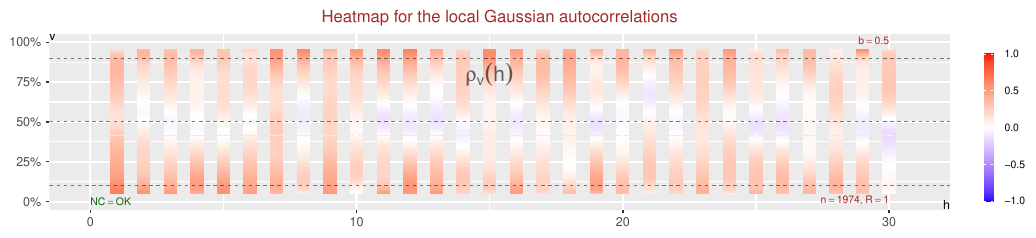}
  }
  \caption[]{Heatmap for $\hatlgacr{\LGp}{h}$, for the
    \texttt{dmbp}-data.  The percentiles used in \cref{fig:dmbp},
    i.e.\ 10\%, 50\% and 90\%, have been highlighted with
    lines.}\label{fig:dmbp_v_heatmap_lgacr}
\end{figure}

For completeness, \cref{fig:dmt_v_heatmap_lgacr} has been included
in order to show how the situation looks like for the \textit{local
  trigonometric example} seen in
\cref{fig:dmt_heatmap_and__levels_vs_norm}.

\begin{figure}[h]
  {\centering \includegraphics[width=\textwidth]{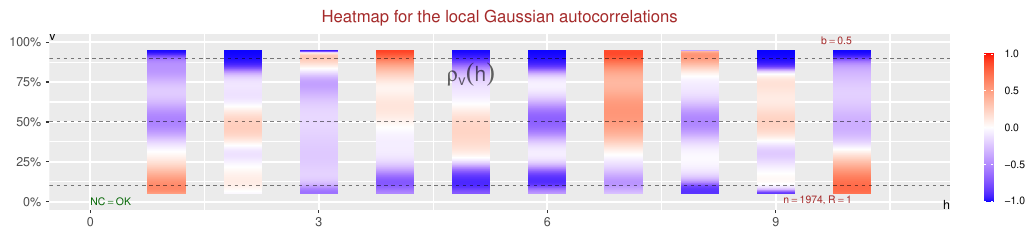}
  }
  \caption[]{Heatmap for $\hatlgacr{\LGp}{h}$, for the \textit{local trigonometric case}.
    The percentiles used in \cref{fig:trigonometric}, i.e.\ 10\%,
    50\% and 90\%, have been highlighted with
    lines.}\label{fig:dmt_v_heatmap_lgacr}
\end{figure}

\subsection{Sensitivity analysis: The bandwidth $\bm{b}$}
\label{app:Bandwidth_sensitivity}

The bandwidth $\bm{b}=\parenR{\bz{1},\bz{2}}$ is bivariate, but it
is natural to assume $\bz{1}=\bz{2}$ when a univariate time series
is investigated.  With this restriction it follows that the
sensitivity of $\hatlgsdM{\LGp}{\omega}{m}$ due to changes in the
bandwidth $\bm{b}$ can be investigated in a similar manner to the
one used in the preceding section for the diagonal points $\LGp$.

The bandwidth $\bm{b}$ should be selected according to the
Goldilocks principle, i.e.\ it should neither be \enquote{too low}
nor \enquote{too high}, it must be \enquote{just right}.  The
heatmap and distance-based plots from
\cref{fig:apARCH_heatmap_and__levels_vs_norm} can easily be adjusted
to visualise the problems that occur when the bandwidth does not
belong to the \enquote{just right} region.  The plots shown in
\cref{fig:heatmap_distance_v} does once more consider the
\texttt{dmbp}-data, and in this case the bandwidth ranges from 0.25
to 1.5 in steps by 0.005 (altogether 251 different bandwidths).  The
bandwidth $b=0.5$ has been highlighted since it was that value that
was used in \cref{fig:dmbp}.

\begin{figure}[h]
  {\centering \includegraphics[width=.9\textwidth]{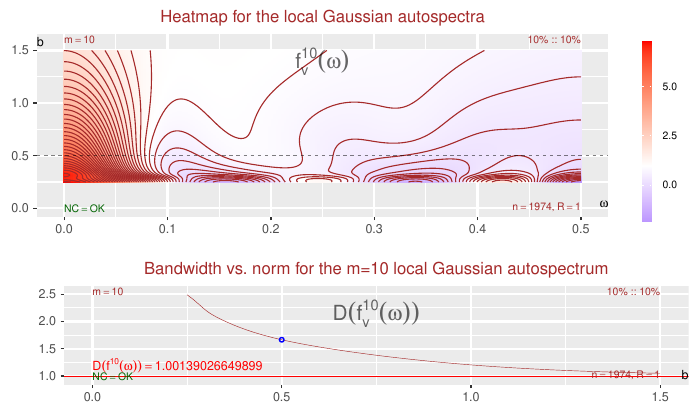}
  }
  \caption[]{Heatmap and corresponding distance-based plots based on
    the \texttt{dmbp-data}, showing how $\hatlgsdM{\LGp}{\omega}{10}$
    varies with the bandwidth $b$.  The bandwidth used in
    \cref{fig:dmbp}, i.e.\ $b=0.5$, has been
    highlighted with a line/point.}\label{fig:heatmap_distance_v}
\end{figure}

The problem when $\bm{b}$ becomes too large is that the estimated
local Gaussian autocorrelations $\hatlgacr{\LGp}{h}$ no longer will
capture the local structure of interest, and the corresponding
estimated local Gaussian spectral density
$\hatlgsdM{\LGp}{\omega}{m}$ (which no longer deserves to be
referred to as \enquote{local}) will then be indistinguishable from
the ordinary spectral density.  It is clear from
\cref{fig:heatmap_distance_v} that a bandwidth of $b=1.5$ is far too
large for the present investigation.

The expected behaviour when a too low bandwidth is used is that it
will trigger a degeneration of the estimated local Gaussian
autocorrelations, i.e.\ $\hatlgacr[p]{\LGp}{h}$ will tend towards
either $+1$ or $-1$ regardless of the actual structure of the
underlying density distributions.

The reason for this is that $\hatlgacr[p]{\LGp}{h}$ will, due to the
kernel function from the density estimation algorithm, become
increasingly sensitive to the position of the $h$-lagged pairs
$\parenR{\Yz{t+h},\Yz{t}}$ that lies nearest to the point
$\LGp=\LGpoint$.  To clarify, for a given point $\LGp$ there will be
a collection of Euclidean distances to the $h$-lagged pairs
$\parenR{\Yz{t+h},\Yz{t}}$ in the sample, and these distances could
(after a re-indexing) be sorted in ascending order
$\TSR{\dz{i}}{i=1}{n-h}$.

Under the assumption that it is the product normal kernel that is
used, the contribution from a lag-$h$ pair
$\parenR{\Yz{t+h},\Yz{t}}$ that lies a distance of $\dz{i}$ from
$\LGp$ will be weighted by
$\wz{i:\bm{b}} \defeq \tfrac{1}{2\pi \bz[2]{}}\ez[-d_i^2/2b^2]{}$
--- and it is now natural to consider the set of all the weights
$\mathcalWz{\LGp:\bm{b}}\defeq\TSR{\wz{i:\bm{b}}}{i=1}{n-h}$.

The primary detail of interest is how much larger the weights are
for the pairs that lies closest to $\LGp$, and it thus necessary to
consider the fraction
$\rz{ij:\bm{b}}\defeq\wz{i:\bm{b}}/\wz{j:\bm{b}} =
\parenRz[1/b^2]{}{\ez[d_j^2-d_i^2]{}}$.  The number $\rz{ij:\bm{b}}$
will, when $\dz{i}<\dz{j}$, grow to $\infty$ when $\blimit$, and
this implies that the estimation algorithm for small $b$-values will
become increasingly sensitive to the $h$-lagged pairs that lies
closest to the point $\LGp$ when the bandwidth shrinks --- and in
the end it would thus be natural to have a degeneration of the
estimated value $\lgacr{\LGp}{h}$ to either $+1$ or $-1$.

Note that the corresponding
$D\!\parenR{\hatlgsdM[p]{\LGp}{\omega}{m}}$ will grow
when this degeneration happens, as can be seen for $b=0.25$ in the
distance-based plot in \cref{fig:heatmap_distance_v}.

\textbf{Heatmap-plots for the estimates $\hatlgacr{\LGp}{h}$:} It is
here, as it was for the investigation of the diagonal points $\LGp$,
possible to also consider a heatmap based investigation of the
underlying estimates $\hatlgacr{\LGp}{h}$, for $h=1,\dotsc,m$.  Such
a plot is given in \cref{fig:dmbp_b_heatmap_lgacr}, and it can there
be observed that it for some of $\hatlgacr{\LGp}{h}$-estimates is
the case that the estimates first switch sign from positive to
negative --- and then they grows quickly towards $-1$.  This kind of
behaviour is expected to occur when the bandwidth $\bm{b}$ has
shrunk to a level that implies that the kernel function in the local
penalty function, cf.\ \cref{eq:QhN}, gives high weights to the few
observations $\parenR{\Yz{t+h},\Yz{t}}$ nearest $\LGp$, and very low
weights elsewhere.

\begin{figure}[h]
  {\centering \includegraphics[width=\textwidth]{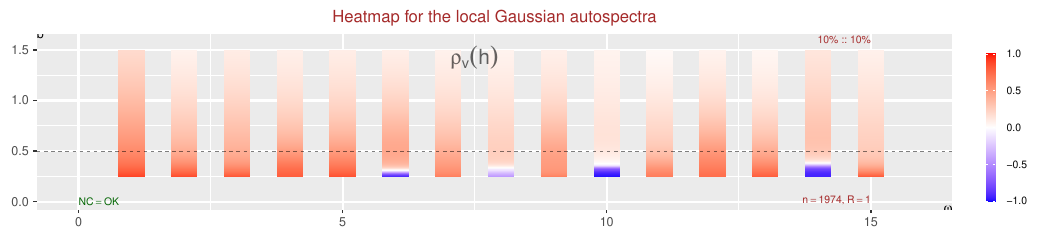}
  }
  \caption[]{Heatmap for $\hatlgacr{\LGp}{h}$, for the
    \texttt{dmbp}-data, showing the effect of different bandwidths.
    The bandwidth used in \cref{fig:dmbp}, i.e.\ 0.5 has been
    highlighted with a line.}\label{fig:dmbp_b_heatmap_lgacr}
\end{figure}

Note that \cref{fig:heatmap_distance_v} considers the situation
where $\LGp$ is the diagonal point corresponding to the
lower tail, but similar plots could have been included for the cases
where $\LGp$ corresponds to either the center or the upper
tail.\footnote{
  The interested reader can use the scripts in the \Rpackage
  \lgsdRpackage\ to get access to these plots for the center and
  upper tail, cf.\ \cref{app:data_details} for details.}  A comparison of the
distance-based plots for these three points is presented in
\cref{fig:three_distances_v}, in order to show how the
bandwidth-sensitivity of $\hatlgsdM[p]{\LGp}{\omega}{m}$ also
depends on the selected point $\LGp$.  A common scale has been used
for the three subplots in order to emphasise the asymmetry between
the lower and upper tail.

\begin{figure}[h]
  {\centering \includegraphics[width=\textwidth]{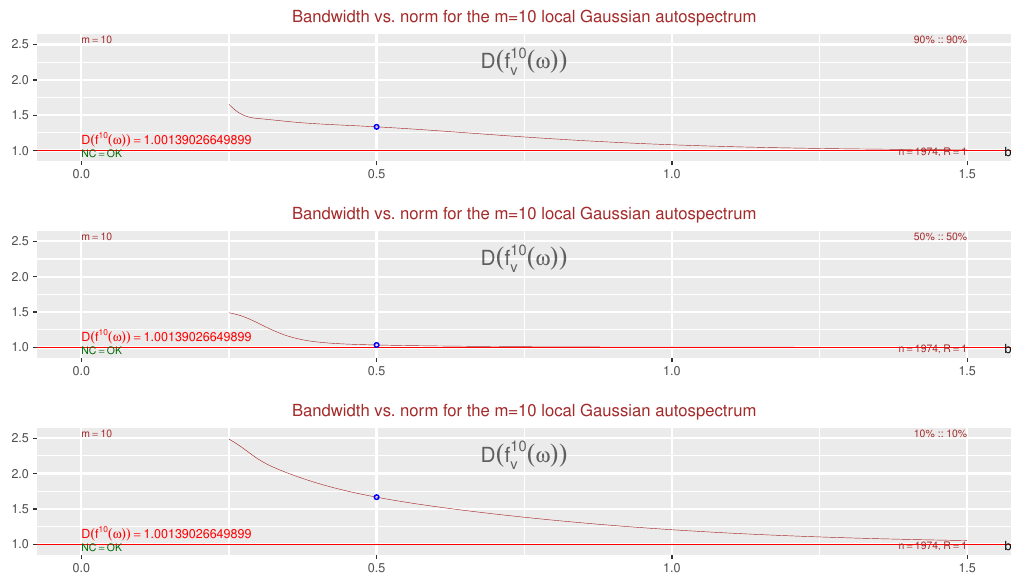}
  }
  \caption[]{Distance-based plots for the \texttt{dmbp-data}, showing
    how $\hatlgsdM{\LGp}{\omega}{10}$ varies with the bandwidth $b$
    for the three percentiles used in \cref{fig:dmbp}, i.e.\ 10\%,
    50\% and 90\%.  The bandwidth used in \cref{fig:dmbp}, i.e.\
    $b=0.5$, has been highlighted by a point.}\label{fig:three_distances_v}
\end{figure}

Note that the center plot of \cref{fig:three_distances_v} reveals that
the \enquote{too low bandwidth problem} occurs a bit slower in a high
density region, but it will even there eventually create a situation
where the estimated local Gaussian autocorrelations degenerate towards
either $+1$ or $-1$.

The heatmap and distance-based plots in
\cref{fig:heatmap_distance_v,fig:dmbp_b_heatmap_lgacr,fig:three_distances_v} can detect the
clearly undesirable regions for the bandwidth $\bm{b}$, but they do
not reveal what the \enquote{just right} value for the bandwidth
should be.  Nevertheless, it is still possible to gain some insight
into how sensitive the estimate of $\lgsdM[p]{\LGp}{\omega}{m}$ will
be for minor variations of the bandwidth $\bm{b}$, and that can be
useful with regard to the selection of a few bandwidths that can be
used when e.g.\ a bootstrap-investigation is to be performed.

The framework used in the \Rpackage \lgsdRpackage\ ensures that it is
trivial to compute and investigate a wide range of bandwidths
simultaneously, and the key idea is that knowledge of the local
dependency structure can still be obtained even if the selected
bandwidths are not spot on the \enquote{just right} value for the
bandwidth.

\subsection{Sensitivity analysis: The truncation level $m$}
\label{app:Truncation_level_sensitivity}

The shape of $\lgsdM{\LGp}{\omega}{m}$ for a low truncation level
can be different from the shape seen when a higher truncation level
is used.  It is thus of interest to investigate how sensitive the
estimates $\hatlgsdM{\LGp}{\omega}{m}$ are to changes in the
truncation level $m$.

This issue can easily be probed by performing an initial
investigation with a high value for the maximum lag to be computed,
since the computational cost is not too large when only a single
sample (like the \texttt{dmbp}-data) is investigated. 
It did e.g.\ not take a long time to estimate $\lgacr{\LGp}{h}$ for
$h=1,\dotsc,200$, which was needed for the construction of
\cref{fig:dmbp_lag} in the main document --- and with these
estimates it is trivial to compare $\widehatfz[m]{}(\omega)$ and
$\hatlgsdM{\LGp}{\omega}{m}$ for $m$ up to 200, since the integrated
\Rref{shiny}-application of the \Rpackage \lgsdRpackage\ can animate
the changes that occur in the spectra when $m$ grows from 0 to 200.

The computational costs can become rather large when it is necessary
to find pointwise confidence intervals, since a high number of
replicates then must be investigated with the same configuration of
tuning parameters.  It is then important to figure out a sufficient
truncation level $m$, and restrict the attention to the estimates of
$\lgacr{\LGp}{h}$ for $h=1,\dotsc,m$.

A drawback with the \Rref{shiny}-based approach in \lgsdRpackage\ is
that it requires an inspection of many different plots.  It could
thus be of interest to also consider summary-plots that either use
the distance function $D$ from
\cref{def:sensitivity_analysis_distance_function}, or some
heatmap-based alternative visualisation of
$\hatlgsdM{\LGp}{\omega}{m}$, similar to those used for
$\hatlgacr{\LGp}{h}$ in
\cref{fig:dmbp_v_heatmap_lgacr,fig:dmt_v_heatmap_lgacr,fig:apARCH_v_heatmap_lgacr,fig:dmbp_b_heatmap_lgacr}.

\textbf{Distance plots:} 
It is possible to investigate the $m$-sensitivity by distance-based
plots, but those plots are less 
useful in this case.
One reason for this is that the norms
$D\!\parenR{\hatlgsdM{\LGp}{\omega}{m}}$ are monotonically
increasing as functions of $m$. This can easily be seen by first
recalling (cf.\
\myref{def:lgsd_estimator}{def:lgsd_esitimator_folded}) that the
estimates 
$\hatlgsdM{\LGp}{\omega}{m}$ are given by
$$\hatlgsdM[p]{\LGp}{\omega}{m} \defeq 1 + \sumss{h=1}{m}
\lambdazM{h}{m}\cdot \hatlgacr{\LGpd}{h} \cdot \ez[+2\pi i\omega h]{}
+ \sumss{h=1}{m} \lambdazM{h}{m}\cdot \hatlgacr{\LGp}{h} \cdot
\ez[-2\pi i\omega h]{},$$ 
and then keeping in mind that the lag-window function
$\lambdazM{h}{m}$ satisfies $\lambdazM{h}{m+1}\geq\lambdazM{h}{m}$.
It follows that
$D\!\parenR{\hatlgsdM{\LGp}{\omega}{m+1}} \geq
D\!\parenR{\hatlgsdM{\LGp}{\omega}{m}}$, which does not provide any
useful new information.

Instead of a plot of the norms
$D\!\parenR{\hatlgsdM[p]{\LGp}{\omega}{m}}$, it is slightly more
interesting to consider a plot that shows
$D\!\parenR{\hatlgsdM[p]{\LGp}{\omega}{m+1}-\hatlgsdM[p]{\LGp}{\omega}{m}}$,
i.e.\ the distances between $\hatlgsdM[p]{\LGp}{\omega}{m+1}$ and
$\hatlgsdM[p]{\LGp}{\omega}{m}$ in the Hilbert room of Fourier
series.  This idea is shown in
\cref{fig:m_sensitivity_dmbp_percentages} for the three diagonal
points and 200 lags that was included in \cref{fig:dmbp_lag}.
Note that \cref{fig:m_sensitivity_dmbp_percentages} takes into
account the scaling due to the lag-window function
$\lambdazM{h}{m}$, and as such it does provide some new information
compared to that contained in the plot showing the estimated local
Gaussian autocorrelations.

\begin{figure}[h]
  {\centering \includegraphics[width=\textwidth]{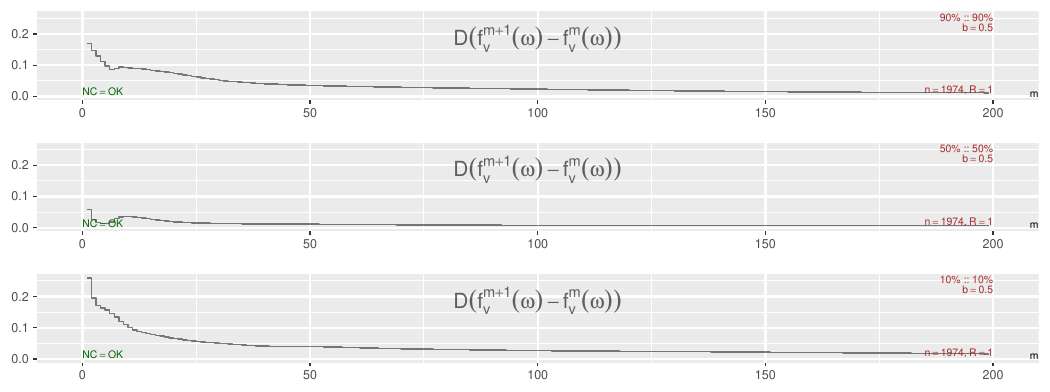}
  }
  \caption[]{Distances between successive $m$-truncations of the local
    spectra,
    \texttt{dmbp}-data.}\label{fig:m_sensitivity_dmbp_percentages}
\end{figure}

The three subplots of \cref{fig:m_sensitivity_dmbp_percentages}
shows that
$D\!\parenR{\hatlgsdM[p]{\LGp}{\omega}{m+1}-\hatlgsdM[p]{\LGp}{\omega}{m}}$
rather quickly starts to decrease monotonically, which is as
expected given the presence of the lag-window function
$\lambdazM{h}{m}$.  This decrease implies that the effect of a
change in the truncation level from $m$ to $m+1$ becomes smaller as
$m$ grows, and the sensitivity is thus largest when $m$ is small. 

\Cref{fig:m_sensitivity_dmbp_percentages} might indicate that the
$m=10$ used in the main part is a bit to small.  However, the
purpose of that particular truncation level was simply to show
that even a low truncation level could be used to detect the
presence of nonlinear dependency structures in the time series under
investigation, i.e.\ structures not detected by the ordinary spectrum.

It is natural to assume that two \textit{successive} local Gaussian
spectra $\hatlgsdM{\LGp}{\omega}{m}$ and
$\hatlgsdM{\LGp}{\omega}{m+1}$ should be similar in shape when $m$
has grown a bit, but this does not imply that the
\textit{accumulated} changes to $\hatlgsdM{\LGp}{\omega}{m}$ are
negligible.  It is thus important to also inspect the
frequency-dimension, and this can as mentioned above easily be done
by the interactive \Rref{shiny}-application in the
\lgsdRpackage-package.

\textbf{Heatmap plots:} 
The truncation level $m$ is a discrete tuning parameter, and an
inspection based on a heatmap-based approach could thus follow the
setup used for the estimated $\hatlgacr{\LGp}{h}$-values seen in
\cref{fig:dmbp_v_heatmap_lgacr,fig:dmt_v_heatmap_lgacr,fig:apARCH_v_heatmap_lgacr,fig:dmbp_b_heatmap_lgacr}.
The \Rpackage \lgsdRpackage\ contains a script that can be used to
create such a heatmap-based plot for $\hatlgsdM{\LGp}{\omega}{m}$,
with the frequencies $\omega$ along one axis and the truncation
levels $m$ along the other.

The resulting heatmap-based plot clearly showed that the peak seen
in \cref{fig:dmbp} at $\omega=0$ (for $m=10$ and a point either in
the lower or upper tail) became even more dominating as $m$
increased, and the peak dominated to such an extent that the
heatmap-based plot did not reveal anything about the other
frequencies.  This plot has thus not been included here, but the
script is 
available in \lgsdRpackage, cf.\ \cref{app:data_details} for
details.

\section{How to select the tuning parameters?}
\label{How.to.select.the.tuning.parameters?}
\setcounter{figure}{0} 

Several tuning parameters are required in order to compute the
$m$-truncated estimate $\hatlgsdM[5]{\LGp}{\omega}{m}$ of the local
Gaussian spectrum $\lgsd{\LGp}{\omega}$, for a given point $\LGp$.
In addition to the truncation level $m$, there is a bandwidth
$\bm{b}$ (to be used when estimating the local Gaussian
autocorrelations $\lgacr{\LGp}{h}$, for $h\in\parenC{1,\dotsc,m}$).
There is also a lag-window function $\lambdaz{m}(h)$ used for
smoothing.

The sensitivity analysis in \cref{app:sensitivity_analysis}
considered the effect of minor changes to the tuning parameters
$\bm{b}$ and $m$, and it did also discuss the sensitivity of
$\hatlgsdM[5]{\LGp}{\omega}{m}$ that is due to the position of the
point $\LGp$ --- which is of interest to know when a given
sample/model is to be investigated.

The task of finding \enquote{optimal tuning parameters} lies beyond
the scope of this paper, and the focal point of interest in this
section will be to give some advice with regard to how the \Rpackage
\lgsdRpackage\ can be used to investigate a given sample/model, cf.\
\cref{Using.the.lgsdRpackage} for the details.  A few comments
related to the selection of the bandwidth $\bm{b}$ is given in
\cref{sec:Some.comments.regarding.the.bandwidth}, primarily in order
to give some pointers to papers that have discussed bandwidth
selection for the estimation of the local Gaussian correlation
$\lgcor{\LGp}$.

\subsection{Using the \Rpackage \lgsdRpackage}
\label{Using.the.lgsdRpackage}

The \Rpackage \lgsdRpackage\ can compute
$\hatlgsdM[5]{\LGp}{\omega}{m}$ for a wide range of tuning
parameters, and for a huge selection of different points $\LGp$.
The integrated \Rref{shiny}-application enables an easy interactive
investigation of the resulting estimates, with an interface that
makes it trivial to switch between visualisations of the estimated
local Gaussian autocorrelations $\hatlgacr{\LGp}{h}$ and the
corresponding estimated local Gaussian spectral densities
$\hatlgsdM[5]{\LGp}{\omega}{m}$.

The computational cost for one single estimate of the local Gaussian
correlation $\lgacr{\LGp}{h}$, for a given lag $h$, a given
bandwidth $\LGp$ and a given point $\LGp$, is usually not that high
(depends on the sample size $n$).  The computational cost does
however quickly escalate when a huge combination of points $\LGp$,
bandwidths $\bm{b}$ and large truncation level $m$ is used.  It
becomes even worse if it is of interest to produce pointwise
confidence intervals, since it then will be necessary to have $R$
replicates of every configuration of these tuning parameters.

This implies that it for a practical investigation is natural to
first do the computations on a single sample, a few bandwidths
$\bm{b}$ and a wide range of points $\LGp$.  The truncation level
$m$ could in this initial investigation probably be rather low,
e.g.\ $m=30$, since the key observation is that it is differences
between the $m$-truncated ordinary and local Gaussian spectra that
can reveal the presence of non-Gaussian dependency structures in the
sample.

The next step of the investigation is the inspection of the heatmap-
and distance-based plots of the estimates
$\hatlgsdM[5]{\LGp}{\omega}{m}$, and from this it is then possible
to figure out if there are some subset of the points $\LGp$ that it
would be of particular interest to investigate further.  If such
points are identified, then it is possible to restrict another
investigation to these points, and then perform e.g.\ $R=100$
replicates in order to produce the pointwise confidence intervals.

This procedure was used in
\cref{app:fig:trigonometric.C1.component}, where the aim of the
investigation was to show that for a sufficiently large sample from
the \textit{local trigonometric} model used in
\cref{sec:Deterministic_trigonometric_models}, it should be possible
to detect the $\Cz{1}(t)$ component that only occurred with a
probability of $\pz{1}=0.05$.  In this case a range of diagonal
points $\LGp$ were selected from the lower tail, and one sample was
used as the basis for the heatmap- and distance-based plots seen in
\cref{fig:dmt_for_extreme_tail_heatmap_and__levels_vs_norm}.  From
this it was then easy to identify a suitable point $\LGp$ that could
be used to create the plot in
\cref{fig:local.trigonometric.C1-component}, where the pointwise
confidence intervals also are present.

This kind of investigation is easy to reproduce for other samples,
since the scripts in the \Rpackage \lgsdRpackage\ can be modified in
order to deal with similar investigations, cf.\ the discussion in
\cref{app:data_details} for further details.

\subsection{Some comments regarding the bandwidth $\bm{b}$}
\label{sec:Some.comments.regarding.the.bandwidth}

The bandwidth \mbox{$\bm{b}=(.5,.5)$} used as default in
\cref{sec:Examples} of the main part was selected based on the fact
that \mbox{$b=.5$} is quite close to the value obtained when the
formula \mbox{$b\approx 1.75\nz[-1/6]{}$} was given the value
\mbox{$n=1974$} (the length of the \texttt{dmbp}-data).  This
formula, due to H{\aa}kon Otneim, is based on an empirical
comparison with a cross-validation bandwidth algorithm used in
\citet{otneim2017locally}, and it has been applied here even though
it originates from a bandwidth-selection algorithm aimed at
computing density estimates based on the one-free-parameter local
Gaussian approximation employed in that paper.

There does exist a leave-one-out cross-validation algorithm for the
selection of the bandwidth to be used when estimating the local
Gaussian correlation based on independent observations, see
\citet[Section~3.4]{Berentsen2014:departure_from_independence} for
details.  However, the estimation of the local Gaussian spectral
density $\lgsdM[1]{\LGp}{\omega}{m}$ requires the estimation of $m$
different local Gaussian autocorrelations $\lgacr[1]{\LGp}{h}$, and
such cross-validation algorithms then becomes quite time
consuming\footnote{Tests were performed to see if it might be
  possible to only use the bandwidth-algorithm for the case
  \mbox{$h=1$}, and then let the higher lags inherit the estimated
  bandwidth~--- but it turned out that that assumption was not a
  viable one.  In particular, the bandwidths estimated for the
  higher lags did not need to be close to the one estimated for the
  first lag.}  --- in particular if it in addition is necessary to
use bootstrapping in order to obtain pointwise confidence intervals
for the estimates.  Moreover, it may be a bit questionable to apply
an algorithm developed for independent observations in a time series
setting.  In particular, the leave-one-out cross-validation has some
flaws if the aim is model selection based upon dependent data, see
\citet{shao1993linear,burman1994cross,racine2000consistent}, where
the concepts leave-$\nz{\nu}$-out cross-validation, \mbox{$h$-block}
cross validation, and \mbox{$hv$-block} cross-validation were
introduced as better tools for the dependent case.

\section{Regarding sampling and resampling}
\label{app:regarding_resampling}
\setcounter{figure}{0} 

This section will discuss sampling related issues, both with regard
to the parametric and the nonparametric cases.  Details related to
the trivial case of sampling from parametric models are given in
\cref{app:simulations.from.a.parametric.model}.
\Cref{app:resampling_parametric_bootstrap} discusses the approach
based on parametric bootstrapping, which can be of interest in order
to see if samples from a model fitted to a given data-set have the
same dependency structure as the original data.  This section
includes a plot similar to one of the diagnostic plots used in
\citet{BIRR2019122}, in which points $\LGp$ both on and off the
diagonal have been used in the investigation.

Nonparametric and model free bootstrap strategies are discussed in
\cref{app:nonparametric_resampling_strategies,app:adjusted_resampling_algorithm},
and it is there seen that a slightly adjusted version of the block
bootstrap, cf.\ \cref{def:index_based_block_bootstrap_for_tuples} on
page \pageref{def:index_based_block_bootstrap_for_tuples}, can be a
useful resampling strategy for the estimators that are used to find
the local Gaussian spectral densities.

A sensitivity analysis of the block length argument $L$ (used in the
adjusted resampling algorithm) is given in
\cref{app:Block_length_sensitivity}, and a few additional comments
related to problematic issues with the initial approach are given in
\cref{app:What.about.the.ordinary.block.bootstrap?}.

\subsection{Simulations from a parametric model}
\label{app:simulations.from.a.parametric.model}

Simulations are trivial for parametric time series models, since new
independent samples (of the same length $n$) can be made directly
from the model.  The estimates of $\lgsdM{\LGp}{\omega}{m}$ (for the
specified values of $m$ and $\bm{b}$) are then computed for each of
these samples, the mean of the resulting estimated spectra is used
as the proxy for the true spectra, whereas pointwise confidence
intervals are constructed directly from the collection of estimated
spectra. 

\subsection{Parametric bootstrap and local sanity-testing of models}
\label{app:resampling_parametric_bootstrap}

A parametric bootstrap approach can be used to investigate models
fitted to real data, and this is e.g.\ used in \citet{BIRR2019122}.
The idea behind the parametric bootstrap is that a parametric model
first is fitted to the original sample, and then that fitted model
is used when resampling --- which implies that the second step in
this procedure is identical to the one described in
\cref{app:simulations.from.a.parametric.model}.

This approach can be used to perform a local sanity-test of the
fitted model, since it becomes possible to identify
points/frequencies with a clear mismatch between the local
structures detected in the original sample and those seen in samples
from the fitted model.  The plot presented in
\cref{fig:aparch_dmbp_comparison}, which is similar in structure to
one of the plots in \citet{BIRR2019122}, shows how such a comparison
can be performed for the \texttt{dmbp}-data and the
apARCH$(2,3)$-model that was seen in \cref{fig:GARCH,fig:dmbp} in
\cref{sec:Real_data}.

\begin{figure}[h]
  {\centering \includegraphics[width=1\linewidth]{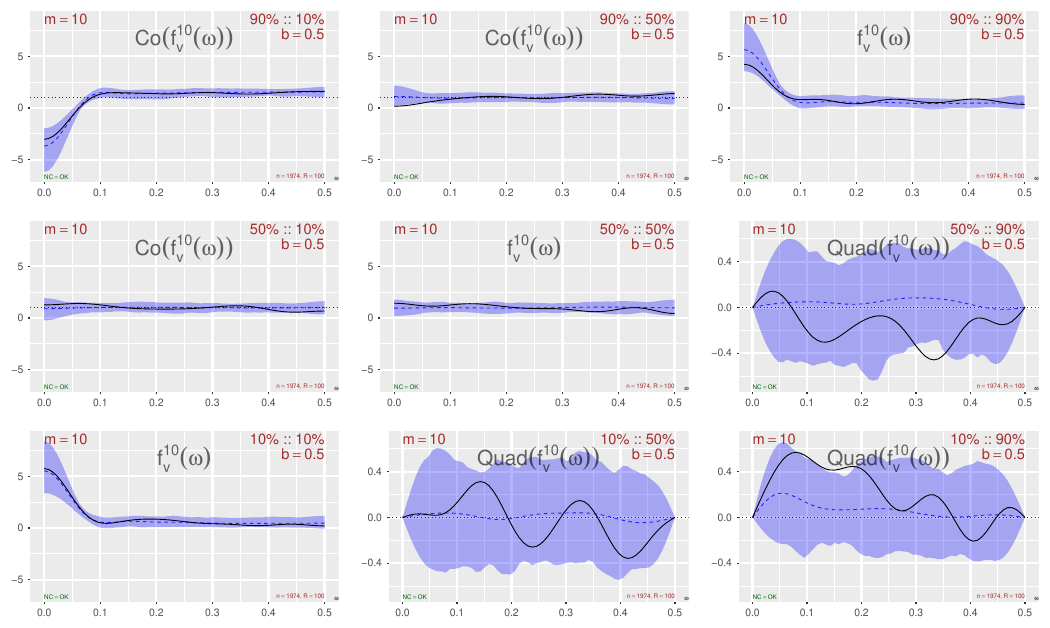}
  }
  \caption[]{The estimates of $\lgsdM{\LGp}{\omega}{m}$ based on the
    \texttt{dmbp}-data (solid lines) have been superimposed on the
    corresponding estimates based on samples from the fitted
    apARCH$(2,3)$-model.  The off-diagonal points $\LGp$ give
    complex-valued $\lgsdM{\LGp}{\omega}{m}$, see main text for
    explanation.}\label{fig:aparch_dmbp_comparison}
\end{figure}

The key idea in \cref{fig:aparch_dmbp_comparison} is that estimates
of $\lgsdM{\LGp}{\omega}{m}$ based on the original sample can be
superimposed on the plots based on parametric bootstrapping from the
fitted model, and this makes it easy to compare them.

Nine different points $\LGp=\LGpoint$ are considered in
\cref{fig:aparch_dmbp_comparison}, and these are based on the
combinations that can be created when $\LGpi{1}$ and $\LGpi{2}$
varies over the 10\%, 50\% and 90\% percentiles of the standard
normal distribution.  The corresponding plots are ordered in a grid
in accordance with the position of these nine points in the plane,
as can be seen by the information about $\LGp$ in the upper right
corner of the respective plots.

The estimates of $\lgsdM{\LGp}{\omega}{m}$ for the three diagonal
points are real-valued, and this is thus in essence the same plots
that was seen in \cref{fig:GARCH} --- but the information about the
global spectrum has been removed and the solid lines from
\cref{fig:dmbp} have been added to the plots.

The estimates $\lgsdM{\LGp}{\omega}{m}$ are complex-valued for the
six off-diagonal points, and in this case the \Rpackage
\lgsdRpackage\ follows the convention used for the complex-valued
cross-spectra, viz.\
$\operatorname{Co}\!\parenR{\lgsdM{\LGp}{\omega}{m}} =
\operatorname{Re}\!\parenR{\lgsdM{\LGp}{\omega}{m}}$ and
$\operatorname{Quad}\!\parenR{\lgsdM{\LGp}{\omega}{m}} =
- \operatorname{Im}\!\parenR{\lgsdM{\LGp}{\omega}{m}}$.

The off-diagonal points are symmetric around the diagonal, i.e.\
both $\LGp=\LGpoint$ and its diagonal reflection $\LGpd=\LGpointd$
are present.  It is the case that
$\lgsd[p]{\LGp}{\omega} = \overline{\lgsd[p]{\LGpd}{\omega}}$, cf.\
\myref{th:lgsd_properties}{th:lgsd_reflection_property},
so it is sufficient to plot 
$\operatorname{Co}\!\parenR{\lgsdM{\LGp}{\omega}{m}}$ on one side of
the diagonal and 
$\operatorname{Quad}\!\parenR{\lgsdM{\LGp}{\omega}{m}}$ on the other
side.

Finally, the same scale is used for all plots showing real values,
whereas another scale is used for the plots related to the imaginary
parts.  This distinction is natural since the scale needed for the
imaginary part can be much smaller,
as can be seen in \cref{fig:aparch_dmbp_comparison}.

A comparison of the dashed and solid lines in
\cref{fig:aparch_dmbp_comparison} can now be used to see if there
might be any faults with the apARCH$(2,3)$-model that was fitted to
the \texttt{dmbp}-data.  The plots related to the real parts does
not give any indications that something is off, with a possible
minor exception near $\omega=0$ for the point at the upper tail (as
also observed in \cref{sec:GARCH_model} in the main part).  The
plots related to the imaginary parts might (when seen isolated)
imply that the model did not catch all of the dependency structure
--- but it is here important to keep in mind that different scales
are used for the two groups of plots, and as such it seems natural
that a good match at the dominating scale might be accompanied with
a more messy situation at the other scale.

It seems natural to conclude that the selected apARCH$(2,3)$-model
performs rather well, which is as expected since it was one of the
better models from a testing procedure that tried out
several thousand different variations of the GARCH-type models
implemented in the \Rref{rugarch}-package.

A final comment to \cref{fig:aparch_dmbp_comparison}: Note the shape
seen for the points $\LGp=\LGpoint$ on the outer tails of the
anti-diagonal, viz.\ when $\LGpi{1}$ corresponds to the 10\%
percentile and $\LGpi{2}$ to the 90\% percentile (or vice versa).
For these points, 
$\operatorname{Co}\!\parenR{\lgsdM{\LGp}{\omega}{m}}$ does have a
deep trough near $\omega=0$, which is rather natural since in a
volatile situation it can be the case that a large decrease is
followed by a somewhat larger increase (like a \enquote{Sucker
  Rally} in the stock-market).

\subsection{Nonparametric bootstrapping techniques}
\label{app:nonparametric_resampling_strategies}

This section will first explain why the block bootstrap could be a
reasonable resampling technique for a statistic like the
$m$-truncated estimates of the local Gaussian spectra.  It will then
be seen that after all there are some issues with the block
bootstrap for the present case, and that motivates the quest for a
slightly modified resampling strategy.

Technical details related to the bootstrap and block bootstrap are
collected in \cref{sec:bootstrap_block_bootstrap}, whereas
\cref{sec:edge_effects_between_blocks} discuss some problems related
to edge-effects between the blocks in the resampled time series.
\Cref{sec:natural_solution_to_edge_effect_issue?} discuss one
potential solution to the edge-effect issue, and explains why this
approach was discarded for the investigation performed in the
present paper.  \Cref{sec:block_of_blocks_bootstrap} presents the
ideas behind the block-of-blocks bootstrap (where edge-effects does
not occur), and it explains why a direct application of that method
might not be an optimal approach when the statistic of interest is
computed by means of an algorithm that contains a kernel function.

\textbf{Justification for the block bootstrap:} First of all, recall
from \myref{def:lgsd_estimator}{def:lgsd_esitimator_folded} (page
\pageref{def:lgsd_esitimator_folded} in the main part) that the
$m$-truncated estimates $\hatlgsdM[p]{\LGp}{\omega}{m}$ of the local
Gaussian spectral densities $\lgsd[p]{\LGp}{\omega}$, are
constructed as follows:
\begin{align}
  \hatlgsdM[p]{\LGp}{\omega}{m} \defeq 1 +       \sumss{h=1}{m}
  \lambdazM{h}{m}\cdot \hatlgacrb[p]{\LGpd}{h}{\bmbzh{h}}
  \cdot \ez[+2\pi i\omega h]{} +       \sumss{h=1}{m}
  \lambdazM{h}{m}\cdot \hatlgacrb[p]{\LGp}{h}{\bmbzh{h}} \cdot
  \ez[-2\pi i\omega h]{},
\end{align}
where the point \mbox{$\LGpd=\LGpointd$} is the diagonal reflection
of \mbox{$\LGp=\LGpoint$}, and $\bmbzh{h}$ is the bandwidth-vector
used for the lag-$h$ pairs (the $\bmbzh{h}$ will henceforward be
dropped from the notation).

Note that the estimates $\hatlgacr[p]{\LGpd}{h}$ and
$\hatlgacr[p]{\LGp}{h}$, for $h=1,\dotsc,m$, and also the
$m$-truncated estimate $\hatlgsdM[p]{\LGp}{\omega}{m}$, all are
estimated by a local likelihood approach --- and the asymptotic
properties of these estimates were developed in the present paper
using the procedure from \citet{klimko1978}, cf.\ the discussion in
\cref{App:local_penalty_function_Klimko_Nelson_approach}.

A statistic obtained from the Klimko-Nelson procedure was explicitly
mentioned by K{\"u}nsch as an example for which the block bootstrap
method would be applicable, cf.\ \citet[Example 2.4, p.\
1219-20]{kuensch89:_jackk_boots_gener_station_obser}, 
and a resampling based on the block bootstrap was thus initially
used for the construction of the pointwise confidence intervals for
the \texttt{dmbp}-example seen in \cref{fig:dmbp}.

Comments received during the review-process initiated an 
investigation of the following problem: Estimates based on the block
bootstrap method can suffer from edge-effect noise when it is used
on smaller sample sizes, cf.\ the discussion in
\cref{app:corrupt.tuples.edge.effect.block.bootstrap}.  This
motivated an investigation of possible replacements, that in the end
lead to the slightly adjusted version of the block bootstrap given
in \cref{app:adjusted_resampling_algorithm}, see
\cref{def:index_based_block_bootstrap_for_tuples} on page
\pageref{def:index_based_block_bootstrap_for_tuples}.

\subsubsection{The bootstrap and the block bootstrap}
\label{sec:bootstrap_block_bootstrap}

The bootstrap introduced in \citet{efron1979} use sampling with
replacement from an i.i.d.\ sample $\TSR{\Xz{i}}{i=1}{n}$ to create
a collection of $B$ bootstrapped samples
$\TSR{\TSR{\Xz[*]{i:b}}{i=1}{n}}{b=1}{B}$,
and then a nonparametric estimator of the variance of a statistic
$\Tz{n} \defeq T\!\parenR{\Xz{1},\dotsc,\Xz{n}}$ can be computed
from the estimates in $\TSR{\Tz[*]{n:b}}{b=1}{B}$, where
$\Tz[*]{n:b} \defeq T\!\parenR{\Xz[*]{1:b},\dotsc,\Xz[*]{n:b}}$.
The block bootstrap introduced in
\citet{kuensch89:_jackk_boots_gener_station_obser} enables a similar
investigation to be performed when the statistic $T_n$ is computed
on a set of observations $\TSR{\Xz{i}}{i=1}{n}$ from a stationary
process, and in this case the resampled sets
$\TSR{\Xz[*]{i:b}}{i=1}{n}$ are created by the following procedure:
(1) Create the set of $L$-sized blocks of consecutive observations
from $\TSR{\Xz{i}}{i=1}{n}$, i.e.\ $\TSR{\bmYz{i}}{i=1}{n-(L-1)}$,
where $\bmYz{i}=\parenR{\Xz{i},\dotsc,\Xz{i+(L-1)}}$.  (2) Sample
with replacement $\ceil{n/L}$ of these blocks, to obtain a set
$\TSR{\bmYz[*]{i:b}}{i=1}{\ceil{n/L}}$.  (3) Concatenate the
selected blocks to one block of size $\ceil{n/L}\cdot L$, and
truncate it at length $n$ to obtain the desired resampled version
$\TSR{\Xz[*]{i:b}}{i=1}{n}$.

\citet{kuensch89:_jackk_boots_gener_station_obser} lists a wide
range of different types of statistics that can be based on
$\TSR{\Xz[*]{i:b}}{i=1}{n}$, and it is for the purpose of the
present paper of particular interest to note that statistics based
on the Klimko-Nelson procedure is specifically mentioned as a case,
which as mentioned above is the case for the estimators in this
paper.

\subsubsection{Corrupt tuples and edge-effect noise for the block
  bootstrap}
\label{app:corrupt.tuples.edge.effect.block.bootstrap}
\label{sec:edge_effects_between_blocks}

A problematic issue with the block bootstrap is that it will
introduce a bit of \textit{edge-effect noise} into the estimation
procedure.  For example, if a time series $\TSR{\Yz{t}}{t=1}{n}$ of
length $n$ is given, then an estimate of $\lgacr{\LGp}{h}$ will be
based on the bivariate set
$\mathcalYz{h}\defeq\TSR{\parenR{\Yz{t+h},\Yz{t}}}{t=1}{n-h}$ of size
$n-h$.  When the block bootstrap is used with some block length $L$,
then there will be a resampled sequence $\TSR{\Yz[*]{t}}{t=1}{n}$
and the idea is that an estimate of $\lgacr{\LGp}{h}$ now should be
computed based on the bivariate set
$\mathcalYz[*]{h:L}\defeq\TSR{\parenR{\Yz[*]{t+h},\Yz[*]{t}}}{t=1}{n-h}$.

However, the set $\mathcalYz[*]{h:L}$ will contain \textit{corrupt
  tuples} that do not exist in $\mathcalYz{h}$, i.e.\ the first and
second component of $\parenR{\Yz[*]{t+h},\Yz[*]{t}}$ can belong to
different blocks, and this will add a bit of \textit{edge-effect
  noise} into the estimation process.  The edge-effect noise is
negligible in the asymptotic situation (very large sample sizes $n$
and large block lengths $L$), but it can 
make an impact when smaller samples are investigated.

For the present paper, it is of particular interest to consider the
amount of corrupt tuples that occur when the block bootstrap is used
on the \texttt{dmbp}-data ($n=1974$ unique observations, i.e.\ no
ties).  The plots in \cref{fig:dmbp} used the truncation level
$m=10$ for $\lgsdM{\LGp}{\omega}{m}$, and it is thus natural to
focus on the estimation of $\lgacr{\LGp}{h}$ for $h=1,\dotsc,10$.

It is easy to see that the expected number $\mathcalEz[*]{h:L}$ of
corrupt tuples in $\mathcalYz[*]{h:L}$ to a close
approximation\footnote{
  It is possible that two neighbouring blocks can join perfectly (no
  edge-effect noise), so the correct formula for the expected number of
  corrupt tuples is slightly less than the numbers given in
  \cref{eq:th:block_expected_number_corrupt}, but this level of
  precision is not needed for the present discussion.} 
will be a simple formula of the number of blocks $q\defeq\ceil{n/L}$
and the length $r\defeq n - (q-1)\cdot L$ of the last block, i.e.\
\begin{equation}
  \label{eq:th:block_expected_number_corrupt}
  \mathcalEz[*]{h:L} \approx
  \begin{cases}
    h \cdot (q-1)     & h \leq r \leq L \\
    h \cdot (q-2) + r & 1 \leq r<h.
  \end{cases}
\end{equation}
A total of $n-h$ tuples $\parenR{\Yz[*]{t+h},\Yz[*]{t}}$ are
included in $\mathcalYz[*]{h:L}$, and the expected fraction of
corrupt tuples is thus given by $\mathcalEz[*]{h:L}/(n-h)$.
It is enlightening to compute the expected fractions of corrupt
tuples for the \texttt{dmbp}-data for the two block lengths $L=25$
and $L=100$, and the results (given as percentages) are listed in
\cref{app:table:corrupt.tuples.block.bootstrap}.
\begin{table}[ht]
  \centering
  \begingroup\tiny
  \begin{tabular}{r|llllllllll}
    \hline
    $L$ \textbackslash\ $h$ & 1 & 2 & 3 & 4 & 5 & 6 & 7 & 8 & 9 & 10 \\ 
    \hline
    25 & 4.0\% & 7.9\% & 11.9\% & 15.8\% & 19.8\% & 23.8\% & 27.8\% & 31.7\% & 35.7\% & 39.7\% \\ 
    100 & 1.0\% & 1.9\% & 2.9\% & 3.9\% & 4.8\% & 5.8\% & 6.8\% & 7.7\% & 8.7\% & 9.7\% \\
    \hline
  \end{tabular}
  \endgroup
  \caption{The expected fraction of corrupt tuples when
    $\lgacr{\LGp}{h}$ are estimated from block bootstrap replicates
    of the \texttt{dmbp}-data ($n=1974$), when $L\in\parenC{25,100}$
    and $h\in\parenC{1,\dotsc,10}$.}
  \label{app:table:corrupt.tuples.block.bootstrap}
\end{table}

It is evident, based on
\cref{app:table:corrupt.tuples.block.bootstrap}, that the expected
fraction of corrupt tuples can become rather large when
$\lgacr{\LGp}{h}$ is estimated for high lags $h$.  The problem for
estimates of $\lgsdM{\LGp}{\omega}{m}$ is slightly reduced since the
estimates $\hatlgacr{\LGp}{h}$ are weighted with the lag-window
functions $\lambdaz{m}(h)$ when $\hatlgsdM{\LGp}{\omega}{m}$ is
computed, which implies that the estimates $\hatlgacr{\LGp}{h}$
suffering from the highest levels of edge-effect noise do not
contribute that much to the final result.

\Cref{app:table:corrupt.tuples.block.bootstrap} indicates that it
could be of interest to find an adjusted resampling technique,
preferably one that completely (or at least partially) removes the
corrupt tuples from the estimation algorithm.  Two different
approaches that completely avoids the corrupt tuples are presented
in
\cref{sec:natural_solution_to_edge_effect_issue?,sec:block_of_blocks_bootstrap},
but there are some issues with these two methods that make them less
interesting to implement.

It is however possible to reduce the number of corrupt tuples by
slightly tweaking the way the block bootstrap algorithm is used when
applied to smaller sample sizes.  The key idea is to move the
primary focus to the indices of the original sample, and then apply
a simple adjustment that selects the $h$-lag pairs in a manner that
is more in line with the way these pairs would have been selected if
the methods from
\cref{sec:natural_solution_to_edge_effect_issue?,sec:block_of_blocks_bootstrap}
had been used.  The technical details are given in
\cref{app:adjusted_resampling_algorithm}, see in particular
\cref{def:index_based_block_bootstrap_for_tuples}.

The corrupt tuples do not disappear with the adjusted resampling
strategy from \cref{def:index_based_block_bootstrap_for_tuples}, but
the expected fraction of such tuples (for a given combination of
sample size $n$, block length $L$ and lag $h$) is significantly
lower than those seen in
\cref{app:table:corrupt.tuples.block.bootstrap}.  It can e.g.\ be
seen from \cref{app:table:corrupt.tuples.ibb.bootstrap} (page
\pageref{app:table:corrupt.tuples.ibb.bootstrap}) that for the
$h=10$ case it will be a reduction from 39.7\% to 0.11\% when
$L=25$, and a reduction from 9.67\% to 0.028\% when $L=100$.

\subsubsection{A \enquote{natural} solution to the edge-effect
  issue?}
\label{sec:natural_solution_to_edge_effect_issue?}

Obviously, if the aim of the investigation is restricted to
$\hatlgacr{\LGp}{h}$ for a single value of $h$, then it is trivial
to completely avoid the problem of corrupt tuples in
$\mathcalYz[*]{h:L}$.  The solution in that case would simply be to
realise $\mathcalYz{h}$ as a sample from a bivariate time series,
and then apply the block bootstrap method on $\mathcalYz{h}$ instead
of the original sample.  The situation becomes a bit more
complicated when it is necessary to estimate $\hatlgacr{\LGp}{h}$,
for $h=1,\dotsc,m$, since an approach where each of these estimates
are computed from its own $\mathcalYz{h}$ might fail to capture some
of the temporal dependency structure from the original sample
$\TSR{\Yz{i}}{i=1}{n}$.

The temporal dependency structure between $\hatlgacr{\LGp}{h}$ will
be taken care of if the estimation of
$\TSR{\lgacr{\LGp}{h}}{h=1}{m}$ is based on (the relevant parts of)
the $(m+1)$-tuples in the derived time series
$\mathcalYz{\overbar{m}} =
\TSR{\parenR{\Yz{i+m},\dotsc,\Yz{i+1},\Yz{i}}}{i=1}{n-m}$, but this
approach is slightly wasteful since the estimation of
$\lgacr{\LGp}{h}$ for an $h<m$ in this case discards the last $m-h$
observations that would have been used if the estimate had been
based on $\mathcalYz{h}$ instead.  The effect of this wastefulness
will of course not be severe when a large sample is investigated,
but it is present.

Moreover, this approach implies that the estimates of
$\lgacr{\LGp}{h}$, for $1\leq h\leq m$, will depend on the selected
value $m$.  For a strict regime of reproducibility, like the one
implemented in the \Rpackage \lgsdRpackage\ , this implies that
everything must be recomputed if the initial truncation level
$\mz{1}$ is changed to $\mz{2}$.  The computational cost related to
the estimate of $\lgacr{\LGp}{h}$ (for a fixed point $\LGp$ and a
fixed bandwidth $\bm{b}$) is usually not that high, but a local
Gaussian investigation will typically involve a wide range of lags
$h$, many points $\LGp$, different values of the bandwidth $\bm{b}$,
and a huge number of replicates.  This implies that the number of
cases to recompute might increase to the tens of thousands, which
makes the \enquote{resampling from $\mathcalYz{\overbar{m}}$ seen as
  an $(m+1)$-variate time series} approach far from desirable to
implement.

The new estimation algorithm introduced in
\cref{app:adjusted_resampling_algorithm} are inspired by the
resampling from tuples outlined above, and for the cost of a tiny
percentage of edge-effects it will completely avoid the problematic
issues mentioned.  In particular: The estimation of the local
Gaussian autocorrelations $\lgacr{\LGp}{h}$ will use all the
available information in $\mathcalYz{h}$, and the estimated values
$\hatlgacr{\LGp}{h}$ will be the same regardless of the value of the
truncation level $m$.

{
  \textbf{The role of the block length $L$ when resampling from
    $\mathcalYz{\overbar{m}}$:} 
  The discussion in \cref{app:adjusted_resampling_algorithm} will
  reveal that the block length $L$ plays a different role when the
  block bootstrap is used on the $(m+1)$-variate tuples in
  $\mathcalYz{\overbar{m}}$, since both $m$ and $L$ then contribute
  to the capturing of the desired dependency structure.  This is
  different from the situation seen when the ordinary block
  bootstrap is used on $\TSR{\Yz{t}}{t=1}{n}$, since then it only is
  the block length $L$ that decides to what extent the temporal
  dependency structure of the original sample is preserved in the
  resampled data $\TSR{\Yz[*]{t}}{t=1}{n}$.  In particular, a too
  short block length will simply destroy all of the dependency
  structure that it is of interest to investigate.

  The situation changes when the block bootstrap is used on
  $\mathcalYz{\overbar{m}}$ (regarded as an $(m+1)$-dimensional time
  series), since it for some estimators then might be the case that
  even a very short block length $L$ can give decent results (in
  particular for an estimator that focus solely on the content
  captured in the $(m+1)$-variate tuples).  For example: If $L=1$,
  then the block bootstrap used on $\mathcalYz{\overbar{m}}$ is
  equivalent to uniform sampling from the tuples in
  $\mathcalYz{\overbar{m}}$.  For an estimator that does not care
  about the internal order of the resampled tuples, e.g.\ the local
  likelihood estimator used in this paper, it might then in fact be
  sufficient to use such a short block length.

  The block length argument $L$ is for this particular situation
  reduced to a tuning parameter that governs the expected number of
  times the different tuples occur in the resampled version of
  $\mathcalYz{\overbar{m}}$.  A higher value of the block length $L$
  will slightly reduce the fraction of tuples sampled from the start
  and the end of $\mathcalYz{\overbar{m}}$, whereas the majority of
  the tuples will have a tiny increase in the expected number of
  occurrences, cf.\
  \cref{app:block.length.and.expected.content.of.resampled.data}.

  The reduction in the expected number of tuples sampled from the
  end of the time series can be of interest for the adjusted
  resampling strategy given in
  \cref{app:adjusted_resampling_algorithm}, since it will induce a
  corresponding reduction in the expected number of corrupt tuples,
  which is desirable since it removes some of the expected
  edge-effect noise from the estimation.  See the discussion in
  \cref{app:adjusted_resampling_algorithm} for further details.

}

\subsubsection{The block-of-blocks bootstrap}
\label{sec:block_of_blocks_bootstrap}

Another tuple-based bootstrapping approach that should be mentioned
is the block-of-blocks bootstrap introduced in
\citet{politis92:_gener_resam_schem_trian_array}.  This method
completely avoids the edge-effect issue that was mentioned for the
block bootstrap, which makes it an interesting alternative to
consider.

The key idea in the block-of-blocks bootstrap is that two levels of
blocks are created, and resampling is made from the second level.
The first level of blocks are created as follows: For a strictly
stationary and weakly dependent $d$-variate time series
$\TSR{\bmXz{i}}{i=1}{n}$, let
$\Bz{i,m,L}\defeq\parenR{\bmXz{(i-1)L+1},\dotsc,\bmXz{(i-1)L+m}}$.
The block $\Bz{i,m,L}$ contains $m$ consecutive observations, and it
can be considered the result of a \enquote{window} of width $m$ that
is \enquote{moving} at lags $L$ at a time.  There are
$Q=\ceil{(n-m)/L}$ of these blocks, and for each block a statistic
$\Tz{i,m,L}$ is defined by a function
$\phiz{m}:\RR^{dm}\rightarrow\RR$, i.e.\
$\Tz{i,m,L}\defeq\phiz{m}\!\parenR{\Bz{i,m,L}}$.  Note that the set
$\TSR{\Tz{i,m,L}}{i=1}{Q}$ actually is a sample from a strictly
stationary univariate time series (derived from the original time
series through $\phiz{m}$), and note that the mean of
$\TSR{\Tz{i,m,L}}{i=1}{Q}$, i.e.\
$\subp{\overline{T}}{}{n}{}{}\defeq\tfrac{1}{Q}\sumss{i=1}{Q}\Tz{i,m,L}$,
gives an estimate of the true value of the statistic given by the
aforementioned function $\phiz{m}$.  It is thus of interest to do a
block bootstrap on the sample $\TSR{\Tz{i,m,L}}{i=1}{Q}$ in order to
investigate the properties of the estimator
$\subp{\overline{T}}{}{n}{}{}$ --- and this motivates the creation
of the second level of blocks $\mathcalBz{j}$, which are created
from $\TSR{\Tz{i,m,L}}{i=1}{Q}$ by means of a \enquote{window} of
width $L$ that is \enquote{moving} at lags $h$ at a time:
$\mathcalBz{j}\defeq\parenR{\Tz{(j-1)h+1,m,L},\dotsc,\Tz{(j-1)h+L,m,L}}$
is constructed by taking $L$ consecutive observations from
$\TSR{\Tz{i,m,L}}{i=1}{Q}$, and there are $q=\ceil{(Q-L)/h}$ of
these blocks.  \citet[p.\
1993]{politis92:_gener_resam_schem_trian_array} explain how sampling
with replacement ($k$ times), followed by a concatenation, can be
used to construct resampled sets $\Tz[*]{1},\dotsc,\Tz[*]{kL}$, and
they give the required theoretical results that connects the mean
$\subp{\overline{T}}{}{}{}{*}$ of this sample with the mean
$\subp{\overline{T}}{}{n}{}{}$ --- which thus gives the algorithm
for the block-of-blocks bootstrapping.

The block-of-blocks bootstrap completely avoids the edge-effect
problem that occurs when the block bootstrap is used, since the
statistic of interest (given by the function $\phiz{m}$) are
computed on the individual blocks $\Bz{i,m,L}$.
This restriction to individual blocks can be an excellent idea for
many statistics of interest, but it is a somewhat questionable
approach for the estimates $\hatlgacr[p]{\LGp}{h}$ of the local
Gaussian autocorrelations.  The reason for this is that the
bandwidth argument $\bm{b}$ in the kernel function
$\Kh[\bm{w}-\LGp]{\bm{b}}$ must be much larger if the estimation
algorithm is to be used on only a subset of the observations --- and
the local structures of interest might then not be detected at all.

It would of course be of interest to implement the block-of-block
bootstrap for the estimates of the local Gaussian spectra if very
large samples are encountered, i.e.\ when the individual blocks
contains several thousand consecutive observations --- but for
shorter samples (like the \texttt{dmbp}-example) it seems better to
use something else.

\subsection{A slightly adjusted resampling algorithm}
\label{app:adjusted_resampling_algorithm}

This section will present a minor adjustment of the ordinary block
bootstrap.  The adjusted approach will by construction return the
same results as those obtained from the ordinary block bootstrap
when the sample size $n$ and the block length $L$ are large.  The
situation is different for smaller sample sizes, since the adjusted
approach then will remove the majority of the corrupt tuples that
adds edge-effect noise into the estimation of the local Gaussian
autocorrelations $\lgacr{\LGp}{h}$.

{
  Note that this adjusted resampling strategy is designed
  to take care of statistics that are constructed from pairs
  $\parenR{\Yz{t+h},\Yz{t}}$, and it does this by mimicking key
  features of the optimal resampling strategy described in
  \cref{sec:natural_solution_to_edge_effect_issue?}.  In
  contradistinction to the adjusted block bootstrap, the ordinary
  block bootstrap is not restricted to statistics based on pairs
  $\parenR{\Yz{t+h},\Yz{t}}$, nor is it specially designed for such
  a case.

  The block length $L$ plays a different role when the resampling is
  done on $(m+1)$-tuples, and it can be considered as a tuning
  parameter that governs the expected number of times the different
  tuples will occur in the resampled set, cf.\ the discussion at the
  end of \cref{sec:natural_solution_to_edge_effect_issue?}.  The
  sensitivity analysis of the block length $L$ in
  \cref{app:Block_length_sensitivity} indicates that the selection
  of 
  $L$ should not be a problematic issue when the samples are large
  enough.
}

\subsubsection{A toy example to illustrate the principle}

It will be a bit easier to digest the definitions and the algorithm
that are given later on in this section, if a simple toy-example is
investigated first: Consider a situation with a time series having
five unique observations $\Yz{1},\Yz{2},\Yz{3},\Yz{4},\Yz{5}$ and
assume that there is an interest for an estimate based on the four
lag-1 tuples in
$\mathcalYz{1} = \TSR{\parenR{\Yz{t+1},\Yz{t}}}{t=1}{4}$.  If a
block bootstrap with block length $L=2$ is used, the resampled time
series might e.g.\ look like
$\Yz[*]{1}=\Yz{4},\Yz[*]{2}=\Yz{5},\Yz[*]{3}=\Yz{3},\Yz[*]{4}=\Yz{4},\Yz[*]{5}=\Yz{2}$,
and the corresponding set of lag-1 tuples would be
$\mathcalYz[*]{1:2} = \TSR{\parenR{\Yz[*]{t+1},\Yz[*]{t}}}{t=1}{4}$.
It is easy to see that $\mathcalYz[*]{1:2}$ in this case will contain
the two corrupt tuples $\parenR{\Yz{3},\Yz{5}}$ and
$\parenR{\Yz{2},\Yz{4}}$, i.e.\ tuples that are not present in
$\mathcalYz{1}$.

The key idea in the adjusted algorithm is to move the focus to the
indices of the original sample, i.e.\ $1,2,3,4,5$, and then use the
block bootstrap to sample from these.  The resampled set of indices
for the example above would be $4,5,3,4,2$, and from these it is
possible to construct the \textit{cyclically $h=1$ shifted} set of
indices $5, 1, 4, 5, 3$.  The method is simply to add the lag $h=1$
to all the resampled indices --- and to start back on 1 if a value
exceeds $n=5$.  The four desired lag-1 tuples
$\mathcalYz[\sharp]{1:2} =
\TSR{\parenR{\Yz[\sharp]{t+1},\Yz[\sharp]{t}}}{t=1}{4}$ are now
created by using the resampled set of indices in the
$\Yz[\sharp]{t}$-component, whereas the cyclically $h=1$ shifted
indices are used for the $\Yz[\sharp]{t+h}$-component.  This results
in the following four tuples,
$\mathcalYz[\sharp]{1:2}=\parenC{\parenR{\Yz{5},\Yz{4}},
  \parenR{\Yz{1},\Yz{5}}, \parenR{\Yz{4},\Yz{3}},
  \parenR{\Yz{5},\Yz{4}}}$, and it is easy to see that the only
corrupt tuple in $\mathcalYz[\sharp]{1:2}$ is
$\parenR{\Yz{1},\Yz{5}}$.  Note: It could in principle now also be
added a fifth tuple $\parenR{\Yz{3},\Yz{2}}$ to
$\mathcalYz[\sharp]{1:2}$, but that is not of interest since there
are only four tuples in $\mathcalYz{1}$.

The adjusted resampling algorithm is thus quite simple in structure,
and it only needs to be formalised.  This is taken care of in
\cref{def:slightly_tweaked_modulo_function,def:indices_of_tuple_given_starting_index,def:actual_tuple_given_starting_index,def:index_based_block_bootstrap_for_tuples}.

It is easy to compute the expected number of corrupt tuples in
$\mathcalYz[\sharp]{h:L}$ for a given combination of sample size
$n$, lag $h$, and block length $L$, and this is done in
\cref{th:ibbb_expected_number_corrupt} in
\cref{sec:edge_effects_for_the_adjusted_resampling_algorithm}.  It
can from this easily be seen how much the edge-effect noise is
reduced for estimates based on the \texttt{dmbp}-data, cf.\
\cref{app:table:corrupt.tuples.ibb.bootstrap} on page
\pageref{app:table:corrupt.tuples.ibb.bootstrap}.

\subsubsection{Three definitions and one algorithm}
\label{app.three.def.one.algorithm}

\begin{definition}
  \label{def:slightly_tweaked_modulo_function}
  For $n$ and $i$ positive integers, and $h$ a non-negative integer,
  define the new index $\mathcal{M}(i,h;n)$ as follows:
  \begin{equation}
    \label{eq:def:slightly_tweaked_modulo_function}
    \mathcal{M}(i,h;n) \defeq 1 + \parenS{(i+h-1)\!\!\! \mod n} =
    (i + h) - n\cdot\floor{\frac{i+h-1}{n}}
  \end{equation}
\end{definition}
The result of $\mathcal{M}(i,h;n)$ will always be a number in the
set $\parenC{1,\dotsc,n}$, and $\mathcal{M}(i,0;n)=i$ when
$i\leq n$.

\begin{definition}
  \label{def:indices_of_tuple_given_starting_index}
  For fixed positive integers $m$ and $n$, with $m<n$, and any
  {starting index} $i\in\parenC{1,\dotsc,n}$, define the
  $(m+1)$-tuple $\mathfrak{M}(i;m,n)$ as follows:
  \begin{equation}
    \label{eq:def:indices_of_tuple_given_starting_index}
    \mathfrak{M}(i;m,n) \defeq \parenR{\mathcal{M}(i,m,n), \dotsc, 
      \mathcal{M}(i,1,n),i}
  \end{equation}
\end{definition}

The result of $\mathfrak{M}(i;m,n)$ will be referred to as an
$(m+1)$-variate tuple of indices.  It will have the desirable form
$\parenR{i+m,\dotsc,i+1,i}$ when $i\leq n-m$.  The result will be
\textit{cyclically shifted} when $i\in\parenC{n-m+1,\dotsc,n}$,
i.e.\ the indices will in that case have the form
$\parenR{\mathcal{M}(i,m,n),\dotsc, 1, n,\dotsc, i}$.  Note that it
is trivial to tweak the definition of $\mathfrak{M}(i;m,n)$, if only
a subset of the resulting indices is required.  This is e.g.\ the
case for the indices needed when estimating $\hatlgacr[p]{\LGp}{h}$,
where it only is the bivariate pairs
$\parenR{\mathcal{M}(i,h,n),i}$ that it is of interest to consider.

\begin{definition}
  \label{def:actual_tuple_given_starting_index}
  For a sample $\TSR{\Yz{i}}{i=1}{n}$ of length
  $n$, an integer $m<n$ and any {starting index}
  $i\in\parenC{1,\dotsc,n}$, use the indices from
  $\mathfrak{M}(i;m,n)$ to define an $(m+1)$-variate tuple
  $\bm{Y}(i;m,n)$ as follows:
  \begin{equation}
    \label{eq:def:actual_tuple_given_starting_index}
    \bm{Y}(i;m,n) \defeq
    \parenR{\Yz{\mathcal{M}(i,m,n)}, \dotsc,\Yz{\mathcal{M}(i,1,n)},\Yz{i}}
  \end{equation}
  The resulting tuple will be referred to as \enquote{desirable}
  when $i\leq n-m$, whereas it will be referred to as
  \enquote{corrupt} when $i\in\parenC{n-m+1,\dotsc,n}$.
\end{definition}

If a starting index $i$ is selected randomly from
$\parenC{1,\dotsc,n}$, then there is a probability of
$p=\tfrac{n-m}{n}$ that the tuple $\bm{Y}(i;m,n)$ will be desirable,
and a probability of $1-p=\tfrac{m}{n}$ that the tuple will be
corrupt.

With these definitions, it is now time to present the adjusted
resampling algorithm.

\begin{algorithm}[Circular index-based block bootstrap for tuples]
  \label{def:index_based_block_bootstrap_for_tuples}
  \ \newline Given a sample
  $\TSR{\Yz{i}}{i=1}{n}$ of length $n$ from a
  strictly stationary time series, and a statistic $\Tz{n}$ that is
  given as a function $\varphiz{n}$ of the $(m+1)$-variate set
  $\mathcalYz{\overbar{m}}\defeq\TSR{\parenR{\Yz{i+m},\dotsc,\Yz{i+1},\Yz{i}}}{i=1}{n-m}$,
  i.e. $\Tz{n}\defeq\varphiz{n}\!\parenR{\mathcalYz{\overbar{m}}}$.
  For a given block length $L$, let $q$ be the number
  $\ceil{n/L}$, and define a resampled set
  $\mathcalYz[\sharp]{\overbar{m}:L}$, and
  $\Tz[\sharp]{n}\defeq\varphiz{n}\!\parenR{\mathcalYz[\sharp]{\overbar{m}:L}}$,
  as follows:
  \begin{enumerate}[label=(\alph*)]
  \item 
    \label{alg:ibbb_nq}
    Sample with replacement $q$ numbers $\nz{1},\dotsc,\nz{q}$ from
    the index set $\parenC{1,\dotsc,n-(L-1)}$.
  \item 
    \label{alg:ibbb_Iq}
    For $j\in\parenC{1,\dotsc,q}$, let $\mathcalIz[\sharp]{j:L}$ be
    the $L$-sized tuple
    $\parenR{\nz{j},\nz{j}+1,\dotsc,\nz{j}+L-1}$.
  \item 
    \label{alg:ibbb_I*}
    Let $\mathcalIz[\sharp]{n}=\parenR{\iz[\sharp]{1},\dotsc,\iz[\sharp]{n}}$ be
    the $n$-sized tuple that occurs when the $q$ tuples
    $\mathcalIz[\sharp]{1:L},\dotsc,\mathcalIz[\sharp]{q:L}$ first are
    concatenated into one tuple, and then truncated at length $n$.
  \item 
    \label{alg:ibbb_G*}
    Use the first $n-m$ indices from $\mathcalIz[\sharp]{n}$ as
    starting indices, and let $\mathcalYz[\sharp]{\overbar{m}:L}$ be
    given by
    \begin{equation}
      \label{eq:alg:ibbb_G*}
      \mathcalYz[\sharp]{\overbar{m}:L} \defeq \TSR{\bm{Y}(\iz[\sharp]{j};m,n)}{j=1}{n-m}.
    \end{equation}
  \item 
    \label{alg:ibbb_T*}
    Use the function $\varphiz{n}$ to define the estimate
    $\Tz[\sharp]{n}$, i.e.\
    $\Tz[\sharp]{n}\defeq\varphiz{n}\!\parenR{\mathcalYz[\sharp]{\overbar{m}:L}}$.
  \end{enumerate}
\end{algorithm}

The index set $\mathcalIz[\sharp]{n}$ from
\myref{def:index_based_block_bootstrap_for_tuples}{alg:ibbb_I*} is
the same set of indices that would occur if the block bootstrap was
used to obtain a resampled version $\TSR{\Yz[*]{i}}{i=1}{n}$ of the
original sample $\TSR{\Yz{i}}{i=1}{n}$.
This implies (assuming reasonable
values for $L$ and $m$) that the majority of the tuples in
$\mathcalYz[\sharp]{\overbar{m}:L}$
also will be present in
$\mathcalYz[*]{\overbar{m}:L}\defeq\TSR{
  \parenR{\Yz[*]{i+m},\dotsc,\Yz[*]{i+1},\Yz[*]{i}}}{i=1}{n-m}$,
where the latter is the one that would have been used to get an
estimate
$\Tz[*]{n}\defeq\varphiz{n}\!\parenR{\mathcalYz[*]{\overbar{m}:L}}$
if the ordinary block bootstrap was used.

All the desirable tuples in $\mathcalYz[*]{\overbar{m}:L}$
will also be contained in $\mathcalYz[\sharp]{\overbar{m}:L}$, and
it is easy to see, cf.\ similar discussion in
\cref{app:corrupt.tuples.edge.effect.block.bootstrap}, that the
number of desirable $(m+1)$-variate tuples in
$\mathcalYz[*]{\overbar{m}:L}$ at least must be
$\parenR{n-m\cdot\ceil{n/L}}/(n-m)$.  This fraction converges
towards 1, given reasonable assumptions with regard to how fast
$L\rightarrow\infty$ and $\mlimit$ when $\nlimit$, which thus
implies that the content of $\mathcalYz[\sharp]{\overbar{m}:L}$ and
$\mathcalYz[*]{\overbar{m}:L}$ in essence coincide when $\nlimit$
--- and it is thus natural to anticipate that the asymptotic
behaviour of the estimates $\Tz[\sharp]{n}$ and $\Tz[*]{n}$ should
be quite similar.

As previously mentioned, the block bootstrap is viable for a
statistic based on the Klimko-Nelson procedure, cf.\ \citet[Example
2.4, p.\ 1219-20]{kuensch89:_jackk_boots_gener_station_obser}, and
it is thus in particular applicable when estimating the local
Gaussian autocorrelations $\lgacr{\LGp}{h}$ and the $m$-truncated
local Gaussian spectra $\lgsdM{\LGp}{\omega}{m}$.  The previously
mentioned overlap between $\mathcalYz[\sharp]{\overbar{m}:L}$ and
$\mathcalYz[*]{\overbar{m}:L}$ indicates that the \textit{circular
  index-based block bootstrap for tuples} from
\cref{def:index_based_block_bootstrap_for_tuples} also should be a
viable alternative for the statistics of interest for the present
paper.

\subsubsection{The block length $L$ and the expected content of
  $\mathcalYz[\sharp]{m:L}$}
\label{app:block.length.and.expected.content.of.resampled.data}

The purpose of the adjusted resampling strategy is to provide the
required data $\mathcalYz[\sharp]{m:L}$, that can replace the
$(m+1)$-variate tuples in $\mathcalYz{m}$ when the pointwise
confidence intervals are to be estimated for the original estimate
of $\lgsdM{\LGp}{\omega}{m}$.  A sensitivity analysis of the block
length $L$ is included in \cref{app:Block_length_sensitivity}, and
it is thus of interest to add some comments about the effect the
block length $L$ has on the expected content of
$\mathcalYz[\sharp]{m:L}$.

It is with regard to this discussion of interest to point out that
the temporal connection \textit{between} the $(m+1)$-variate tuples
in $\mathcalYz{m}$ and $\mathcalYz[\sharp]{m:L}$ does not affect the
resulting estimates of $\lgsdM{\LGp}{\omega}{m}$.  The reason for
this is that the algorithm that estimates the local Gaussian
autocorrelations $\lgacr{\LGp}{h}$ only cares about the points in
the plane that are defined by the bivariate lag-$h$ tuples, that
again are derived from these $(m+1)$-tuples.  To clarify: The
temporal aspect is pivotal with regard to the construction of the
$(m+1)$-variate tuples in $\mathcalYz[\sharp]{m:L}$, but the order
does not matter anymore when these tuples first have been
constructed.

The main detail of interest is thus to figure out the expected
number of times the different tuples will occur in
$\mathcalYz[\sharp]{m:L}$.

The first detail to note is that the content of $\mathcalYz{m}$ and
$\mathcalYz[\sharp]{m:L}$ correspond to \textit{starting indices}
given by $(n-m)$-tuples from the index-set $\parenC{1,\dotsc,n}$.
For $\mathcalYz{m}$ it is simply the tuple $\parenR{1,\dotsc,n-m}$,
whereas it for $\mathcalYz[\sharp]{m:L}$ is the $n-m$ first indices
from the tuple $\mathcalIz[\sharp]{n}$ that was introduced in
\myref{def:index_based_block_bootstrap_for_tuples}{alg:ibbb_I*}.

A brief inspection of \cref{alg:ibbb_nq,alg:ibbb_Iq} of
\cref{def:index_based_block_bootstrap_for_tuples} reveals that
$\mathcalIz[\sharp]{n}$ is built from $q=\ceil{n/L}$ tuples
$\mathcalIz[\sharp]{j:L}=\parenR{\nz{j},\nz{j}+1,\dotsc,\nz{j}+L-1}$,
where the index $\nz{j}$ has been sampled uniformly from the
index-set $\parenC{1,\dotsc,n-(L-1)}$.  The length of the $q-1$
first of these tuples are $L$, whereas the last tuple might be
shorter since it has to be truncated to the length $r=n-(q-1)\cdot L$
in order for $\mathcalIz[\sharp]{n}$ to have the length $n$.

The expected content of $\mathcalYz[\sharp]{m:L}$ is thus related to
the expected number of times different starting indices $k$ will
occur in $\mathcalIz[\sharp]{n}$, which again is related to the
probability that the building blocks $\mathcalIz[\sharp]{j:L}$
contains $k$.  The situation for the $q-1$ first of these building
blocks is the simplest.
The basic observation is that the event
\enquote{$\mathcalIz[\sharp]{j:L}$ contains $k$}
is equivalent to \enquote{$\nz{j}\leq k \leq \nz{j}+L-1$}, which can
be rewritten as \enquote{$k- (L-1) \leq \nz{j} \leq k$}.  The number
$\nz{j}$ must lie in the index set $\parenC{1,\dotsc,n-(L-1)}$, so
this latter event is equivalent to \enquote{$1\vee (k- (L-1))\leq
  \nz{j} \leq k \wedge (n-(L-1))$}. 
This implies that the probability that the $L$-length tuple
$\mathcalIz[\sharp]{j:L}$ contains $k$ can be written out as
\begin{align}
  \label{app:eq:prob.k.in.L.length.block}
  \Prob{\mathcalIz[\sharp]{j:L} \text{ contains\ } k, \text{ for\ } j=1,\dotsc,q-1 }
  &= \
    \begin{cases}
      \frac{k}{n-(L-1)} & 1 \leq k < L \\
      \frac{L}{n-(L-1)} & L \leq k < n-L \\
      \frac{n-(k-1)}{n-(L-1)} & n -L \leq k \leq n.
    \end{cases}
\end{align}
The argument for the last block is similar, but the truncation to
length $r$ implies that it can not contain any indices above the
value $n-(L-r)$.
\begin{align}
  \label{app:eq:prob.k.in.r.length.block}
  \Prob{\mathcalIz[\sharp]{q:L} \text{ contains\ } k}
  &= \
    \begin{cases}
      \frac{k}{n-(L-1)} & 1 \leq k < r \\
      \frac{r}{n-(L-1)} & r \leq k < n-L \\
      \frac{n-(L-r)-(k-1)}{n-(L-1)} & n -L \leq k \leq n - (L-r) \\
      0 & n - (L-r) < k \leq n.
    \end{cases}
\end{align}

The expected number of occurrences of an index $k$ in the index set
$\mathcalIz[\sharp]{n}$ can be found by simply summing the expected
number of occurrences in the $q$ building blocks
$\mathcalIz[\sharp]{j:L}$, and this is easy to find from
\cref{app:eq:prob.k.in.L.length.block,app:eq:prob.k.in.r.length.block}.
For the purpose of the sensitivity analysis in
\cref{app:Block_length_sensitivity}, it is sufficient to observe
that the expected number of occurrences of an index $k$ that lies in
the set $\parenC{L,\dotsc,n-L}$ is given by
$(q-1)\cdot \frac{L}{n-(L-1)} + 1\cdot \frac{r}{n-(L-1)}$, and it
follows from $r=n-(q-1)\cdot L$ that this is the number
$\frac{n}{n-(L-1)}$.

This shows how the block length $L$ affects the expected number of
times different indices $k$ occurs in $\mathcalIz[\sharp]{n}$, which
as mentioned above reveals the expected number of times the
corresponding $(m+1)$ tuple will occur in $\mathcalYz[\sharp]{m:L}$.
It is clear from the fraction $\frac{n}{n-(L-1)}$ that it for a
large enough $n$ will be a rather negligible effect on these
expectations when $L$ is modified from e.g.\ 10 to 69 (which is the
case in \cref{app:Block_length_sensitivity}).

\subsubsection{Edge-effect noise for the adjusted resampling
  algorithm}
\label{sec:edge_effects_for_the_adjusted_resampling_algorithm}

This section will investigate the edge-effect noise that occurs when
the adjusted resampling algorithm is applied, and this will in
\cref{sec:dmbp.and.adjusted.resampling.algorithm} be used to check
that the fraction of corrupt tuples becomes minuscule when this
algorithm is used on the \texttt{dmbp}-data ($n=1974$ unique
observations, i.e.\ no ties).

\begin{lemma}
  \label{th:ibbb_expected_number_corrupt}
  Given a sample $\TSR{\Yz{i}}{i=1}{n}$ from a continuous-valued
  time series, and the corresponding derived time series of
  $(m+1)$-tuples
  $\mathcalYz{\overbar{m}} \defeq
  \TSR{\parenR{\Yz{i+m},\dotsc,\Yz{i+1},\Yz{i}}}{i=1}{n-m}$.  For a
  given block length $L>m$, let $q\defeq \ceil{n/L}$ be the number of
  blocks used in the construction of the resampled version
  $\mathcalYz[\sharp]{\overbar{m}:L}$ (introduced in
  \cref{def:index_based_block_bootstrap_for_tuples}), and let
  $r\defeq n-L\cdot(q-1)$ be the length of the last block.
  Let $\mathcalEz[\sharp]{\overbar{m}:L}$ denote the expected number
  of corrupt tuples in $\mathcalYz[\sharp]{\overbar{m}:L}$, i.e.\
  tuples not occurring in $\mathcalYz{\overbar{m}}$.  The number
  $\mathcalEz[\sharp]{\overbar{m}:L}$ is then given by the following
  formula:
  \begin{equation}
    \label{eq:th:ibbb_expected_number_corrupt}
    \mathcalEz[\sharp]{\overbar{m}:L} =
    \begin{cases}
      \frac{1}{2} (q-1) \frac{m(m+1)}{n-(L-1)} & m \leq r\leq L\\
      \frac{1}{2} (q-2) \frac{m(m+1)}{n-(L-1)} + \frac{1}{2}
      \frac{r(r+1)}{n-(L-1)} & 1 \leq r<m 
    \end{cases}
  \end{equation}
\end{lemma}

\begin{proof}

  The continuity-requirement implies that there are no ties (as is
  the case for the \texttt{dmbp}-data).  Further, there is no need
  to adjust the result for the possibility that a corrupt index-set
  (of length $m+1$) can concatenate observations from the two ends
  of $\TSR{\Yz{i}}{i=1}{n}$ into a sequence that already exists as a
  sub-sequence of $\TSR{\Yz{i}}{i=1}{n}$.  To clarify: This
  requirement ensures e.g.\ that no proper tuple
  $\parenR{\Yz{i+1},\Yz{i}}$ can be equal to
  $\parenR{\Yz{1},\Yz{n}}$, so the formulas in
  \cref{eq:th:ibbb_expected_number_corrupt} are thus exact and not
  only approximate.

  The blocks used in the construction of
  $\mathcalYz[\sharp]{\overbar{m}:L}$ are uniquely identified by
  the starting indices given in
  $\mathcalIz[\sharp]{j:L}=\parenR{\nz{j},\nz{j}+1,\dotsc,\nz{j}+(L-1)}$,
  where the initial numbers $\nz{1},\dotsc,\nz{q}$ are sampled
  uniformly from $\parenC{1,\dotsc,n-(L-1)}$.  This implies that
  $\mathcalEz[\sharp]{\overbar{m}:L}$ can be expressed as the sum
  of the expected number of corrupt tuples in the individual blocks.

  It was mentioned in \cref{def:actual_tuple_given_starting_index}
  that zero corrupt tuples would occur for a starting index in
  $\parenC{n-(m-1),\dotsc,n}$, and it follows from this that a block
  will contain 0 corrupt tuples when $\nz{j}\leq n-m-(L-1)$.  This
  implies that the probability for 0 corrupt tuples in a block is
  given by $\tfrac{n-m-(L-1)}{n-(L-1)}$.  It is easy to check that a
  starting tuple given by $\nz{j}=n-(m-k)-(L-1)$ for some
  $k\in\parenC{1,\dotsc,m}$ must correspond to a block that contains
  $k$ corrupt tuples, and each of these outcomes have the same
  probability $\tfrac{1}{n-(L-1)}$.  It follows from this that the
  expected number of corrupt tuples in a block is given by
  $\tfrac{1}{2}\tfrac{m(m+1)}{n-(L-1)}$.

  The expected number of corrupt tuples for the individual blocks
  can now be used to compute $\mathcalEz[\sharp]{\overbar{m}:L}$,
  i.e.\ the expected number of corrupt tuples in
  $\mathcalYz[\sharp]{\overbar{m}:L}$.  Note that only the first
  $n-m$ indices from $\mathcalIz[\sharp]{n}$ are used in the
  computation of $\mathcalYz[\sharp]{\overbar{m}:L}$, cf.\
  \myref{def:index_based_block_bootstrap_for_tuples}{alg:ibbb_G*},
  and that implies that any potential corrupt tuples from the last
  block will be discarded due to this truncation.  The length $r$ of
  the last block will thus influence whether or not some potential
  corrupt tuples from the second to last block also might be removed
  in this truncation, and the formula for the expected number of
  corrupt tuples in $\mathcalYz[\sharp]{\overbar{m}:L}$ must thus
  take the value of $r$ into account.  By construction, $r$ will be
  a number in the set $\parenC{1,\dotsc,L}$.

  The case where $r\geq m$ is the simplest, since the truncation to
  length $n-m$ in this case does not affect the second to last
  block.  The expected number of corrupt tuples in
  $\mathcalYz[\sharp]{\overbar{m}:L}$ is thus simply the sum of the
  expected number of corrupt tuples from the $q-1$ first blocks,
  which gives the result
  $\mathcalEz[\sharp]{\overbar{m}:L}=\frac{1}{2} (q-1)
  \frac{m(m+1)}{n-(L-1)}$ when $r\geq m$.

  The situation for the case $r<m$ is slightly more complicated.
  The effect of truncation to length $n-m$ will in this case
  completely eliminate the last block of $\mathcalIz[\sharp]{n}$,
  and the second to last block will have its last $(m-r)$ indices
  removed.  This implies that the highest possible number of corrupt
  tuples from block number $q-1$ is reduced from $m$ to $r$, which
  implies that the expected number of corrupt tuples from this block
  becomes $\tfrac{1}{2}\tfrac{r(r+1)}{n-(L-1)}$.  The stated result
  follows when this expected number is added together with the
  contribution from the $q-2$ first blocks, i.e.\
  $\mathcalEz[\sharp]{\overbar{m}:L}=\frac{1}{2} (q-2)
  \frac{m(m+1)}{n-(L-1)} + \frac{1}{2} \frac{r(r+1)}{n-(L-1)}$ when
  $r<m$.
\end{proof}

The result in \cref{th:ibbb_expected_number_corrupt} is stated for
$m+1$ tuples of the form $\parenR{\Yz{i+m},\dotsc,\Yz{i+1},\Yz{i}}$,
but it is easy to see that the expected number of corrupt tuples
remains the same if it is restated for bivariate lag-$m$ tuples 
$\parenR{\Yz{i+m},\Yz{i}}$.  This implies that the formula in
\cref{eq:th:ibbb_expected_number_corrupt} can be used for the
\texttt{dmbp}-data investigation given in the next section.

The continuity requirement in \cref{th:ibbb_expected_number_corrupt}
was included in order to avoid additional technicalities in the
proof, but the resulting expression for the expected number of
corrupt tuples would for most cases remain the same even if some
observations were repeated.

A minor warning should be added with regard to the corrupt tuples
that actually do occur when the resampling strategy from
\cref{def:index_based_block_bootstrap_for_tuples} is used: The way
the tuples $\parenR{\Yz[\sharp]{t+h},\Yz[\sharp]{t}}$ is constructed
implies that the corrupt tuples always will occur at the exact same
positions. For example, the lag-1 corrupt tuple will always be the
tuple $\parenR{\Yz{1},\Yz{n}}$, the lag-2 corrupt tuples will always
be from the set
$\parenC{\parenR{\Yz{1},\Yz{n-1}}, \parenR{\Yz{2},\Yz{n}}}$, and so
on.  In a worst case scenario, some of these tuples might be very
close to the point $\LGp$ for which the local Gaussian
autocorrelation $\lgacr{\LGp}{h}$ is to be computed (this can easily
be checked by plotting the relevant tuples).  Given the low expected
fraction of corrupt tuples, cf.\
\cref{app:table:corrupt.tuples.ibb.bootstrap} in the next section,
it seems likely that this effect should not turn out to be a too big
problem.

\subsubsection{The \texttt{dmbp}-data and corrupt tuples for the
  adjusted resampling algorithm}
\label{sec:dmbp.and.adjusted.resampling.algorithm}

It was seen in \cref{app:corrupt.tuples.edge.effect.block.bootstrap}
that the ordinary block bootstrap could produce a high fraction of
corrupt tuples when it was used on smaller samples.
The \texttt{dmbp}-data ($n=1974$) was used as a test case, and
\cref{app:table:corrupt.tuples.block.bootstrap} on page
\pageref{app:table:corrupt.tuples.block.bootstrap} listed the
approximate fractions of corrupt tuples that was expected to occur
in
$\mathcalYz[*]{h:L}=\TSR{\parenR{\Yz[*]{t+h},\Yz[*]{t}}}{t=1}{n-h}$
when $h\in\parenC{1,\dotsc,10}$ and $L=\parenC{25,100}$.
It is now of interest to create a similar table for the
\textit{circular index-based block bootstrap for tuples} from
\cref{def:index_based_block_bootstrap_for_tuples}, in order to see
to what extent the adjusted resampling strategy manages to reduce the
expected number of corrupt tuples
$\mathcalYz[\sharp]{h:L}=\TSR{\parenR{\Yz[\sharp]{t+h},\Yz[\sharp]{t}}}{t=1}{n-h}$.

The counting formula from \cref{eq:th:ibbb_expected_number_corrupt}
can, as mentioned after the proof of
\cref{th:ibbb_expected_number_corrupt}, be used for the present case
of interest too.  The length of the last blocks will for the two
cases $L=25$ and $L=100$ respectively be 24 and 74, and this implies
(since both of them are larger than $h=10$), that it is the version
$\frac{1}{2} (q-1) \frac{h(h+1)}{n-(L-1)}$ that should be used to
find the expected number of corrupt tuples in
$\mathcalYz[\sharp]{h:L}$ when $h\in\parenC{1,\dotsc,10}$.
The data in \cref{app:table:corrupt.tuples.block.bootstrap} was
given as fractions of the total number of tuples $n-h$, and
\cref{app:table:corrupt.tuples.ibb.bootstrap} has thus used the same
adjustment.

\begin{table}[ht]
  \centering
  \begingroup\tiny
  \begin{tabular}{r|llllllllll}
    \hline
    $L$ \textbackslash\ $h$ & 1 & 2 & 3 & 4 & 5 & 6 & 7 & 8 & 9 & 10 \\ 
    \hline
    25 & 0.002\% & 0.006\% & 0.012\% & 0.020\% & 0.030\% & 0.043\% & 0.057\% & 0.073\% & 0.092\% & 0.112\% \\ 
    100 & 0.001\% & 0.002\% & 0.003\% & 0.005\% & 0.008\% & 0.011\% & 0.014\% & 0.019\% & 0.023\% & 0.028\% \\ 
    \hline
  \end{tabular}
  \endgroup
  \caption{The expected amount of corrupt tuples when
    $\lgacr{\LGp}{h}$ are estimated for the \texttt{dmbp}-data
    by the \textit{circular index-based block bootstrap for tuples},
    cf.\ \cref{def:index_based_block_bootstrap_for_tuples}.}
  \label{app:table:corrupt.tuples.ibb.bootstrap}
\end{table}

A comparison with \cref{app:table:corrupt.tuples.block.bootstrap},
with focus on the entries in the $h=10$ column, shows that the
numbers have been reduced from 39.7\% to 0.112\% when $L=25$, and it
has been a reduction from 9.7\% to 0.028\% when $L=100$.  This
implies that the edge-effect noise for the adjusted resampling
strategy can be considered rather negligible, and it could also be
the case that estimates based on $\mathcalYz[\sharp]{h:L}$ might be
less sensitive to changes in the block length $L$, cf.\ the
sensitivity analysis in \cref{app:Block_length_sensitivity}.

The relation between the entries in
\cref{app:table:corrupt.tuples.block.bootstrap,app:table:corrupt.tuples.ibb.bootstrap}
can be found by comparing the counting formulas for the expected
number of corrupt tuples in $\mathcalYz[*]{h:L}$ and
$\mathcalYz[\sharp]{h:L}$, cf.\ respectively
\cref{eq:th:block_expected_number_corrupt,eq:th:ibbb_expected_number_corrupt},
and this results in\footnote{
  The result is only approximate since
  \cref{eq:th:block_expected_number_corrupt} did not adjust for the
  possibility that neighbouring blocks in some rare cases could join
  perfectly (no edge-effect noise), but the actual fraction should
  be very close to the one given by the expressions in
  \cref{eq:th:fraction.expected.corrupt.ibb.vs.block}.  }
\begin{equation}
  \label{eq:th:fraction.expected.corrupt.ibb.vs.block}
  \frac{\mathcalEz[\sharp]{h:L}}{\mathcalEz[*]{h:L}}
  \approx
  \begin{cases}
    \frac{1}{2} \frac{h+1}{n-(L-1)} & h \leq r\leq L \\
    \frac{1}{2} \frac{(q-2)\cdot h\cdot(h+1) + r\cdot(r+1)
    }{(n-(L-1))\cdot(h \cdot (q-2) + r)} & 1\leq r<h.
  \end{cases}
\end{equation}
It follows from \cref{eq:th:fraction.expected.corrupt.ibb.vs.block}
that the $L=25$ and $h=10$ entry in
\cref{app:table:corrupt.tuples.ibb.bootstrap} is 0.282\% of the
corresponding value in
\cref{app:table:corrupt.tuples.block.bootstrap} --- and it can
similarly be seen that the same relation for the entry $L=100$ and
$h=10$ is given by 0.293\%.

\subsection{Sensitivity analysis: The block length $L$}
\label{app:Block_length_sensitivity}

The block length sensitivity for the adjusted resampling strategy
from \cref{def:index_based_block_bootstrap_for_tuples} will now be
investigated --- and the computations will, as for the other tuning
parameters investigated in \cref{app:sensitivity_analysis}, be based
on the \texttt{dmbp}-data.

The tool for this investigation will be the distance function $D$
that was defined in \cref{app:method_for_sensitivity_analysis},
i.e.\ the distance function inherited from the complex Hilbert space
of Fourier series on the interval
$\parenS{-\tfrac{1}{2},\tfrac{1}{2}}$.  This distance function does
not reveal anything about the frequency-component of the cases under
investigation, so it is also necessary to include a plot that focus
on that aspect for a few of the block lengths $L$.

The block length $L$ takes integer values, and one possible way to
gain some insight into the sensitivity of this argument is to use a
sequence of box-plots to show the status for different values of
$L$.  This approach has been used in
\cref{fig:local__block_length_norms_and_changes_in_percentage},
where the panel at the top contains the results as $L$ increases in
steps of 1 from $L=10$ to $L=69$.

\begin{figure}[h]
  {\centering \includegraphics[width=\textwidth]{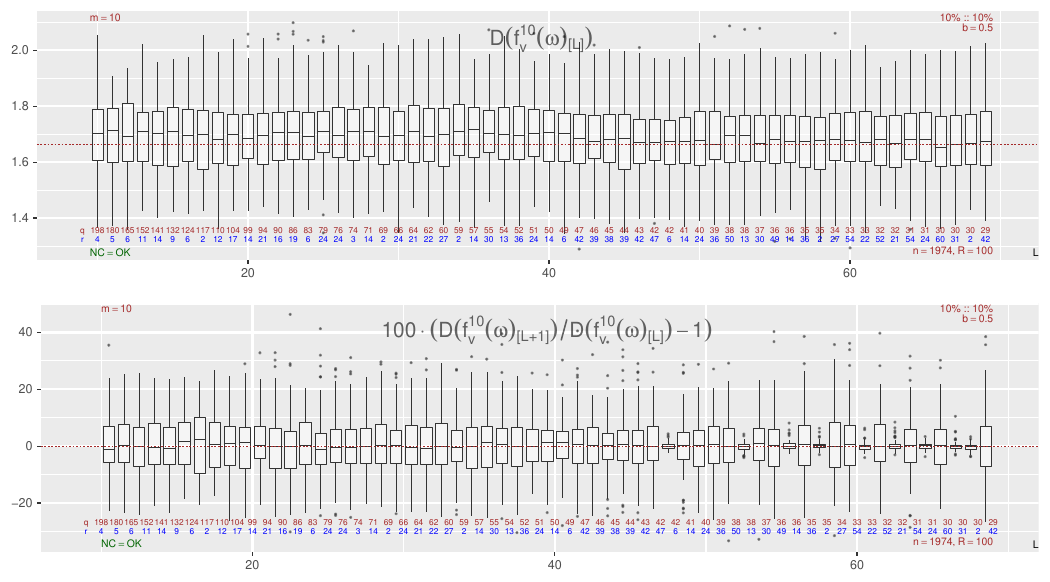}
  }
  \caption[]{Distance based box-plots for the investigation of the
    sensitivity of the block length $L$ for the adjusted resampling
    strategy from \cref{def:index_based_block_bootstrap_for_tuples}.
    The numbers in the two bottom rows show $q=\ceil{n/L}$ and
    $r=n-(q-1)\cdot L$, i.e.\ the number of blocks and the length of
    the last block. }
  \label{fig:local__block_length_norms_and_changes_in_percentage}
\end{figure}

\textbf{The panel at the top of
  \cref{fig:local__block_length_norms_and_changes_in_percentage}:} 
A box-plot for the $D\!\parenR{\lgsdM{\LGp}{\omega}{10}_ {[L]}}$-values
(based on $R=100$ replicates) is given for each block length $L$.  A
horizontal red dashed line has been added that shows the
$D\!\parenR{\lgsdM{\LGp}{\omega}{10}}$-value for the original
sample.  It can be seen that the medians of the box-plots tend to be
slightly larger than the horizontal line that corresponds to the
value based on the original sample, they seem to approach the line
as $L$ increases, but these medians are based on $R=100$ replicates
--- and another realisation might thus look slightly different.  It
does not seem to be any pattern here with regard to how these
box-plots changes when $L$ increases.

\textbf{The panel at the bottom of
  \cref{fig:local__block_length_norms_and_changes_in_percentage}:} 
These box-plots shows the percent-wise changes in the distances when
the block length goes from $L$ to $L+1$, and everything else is kept
identical, i.e.\
$100 \cdot \parenR{D\!\parenR{\lgsdM{\LGp}{\omega}{10}_
    {[L+1]}}/D\!\parenR{\lgsdM{\LGp}{\omega}{10}_ {[L]}}-1}$.  This
is possible to do since the reproducibility setup enables a tracking
for each individual realisation.

A horizontal red dashed line has been added at 0, and it is clear
that the median-part of these box-plots are quite close to this
horizontal line.  It can also be observed that some of these
box-plots are more compact than the other ones, and a simple
investigation of the numbers given at the bottom of the plots
reveals that this phenomenon occurs when an increase from $L$ to
$L+1$ does not reduce the number of blocks that are needed, i.e.\
they occur when $\ceil{n/L} = \ceil{n/(L+1)}$.

For the individual bootstrapped time series, this indicates that the
changes are minimal when the number of blocks remains the same ---
whereas the changes are much larger when the increase of $L$
triggers a reduction in the number of blocks.  However, as is
evident from an inspection of the panel at the top of
\cref{fig:local__block_length_norms_and_changes_in_percentage}, this
effect is only on the level of the individual replicates, and it is
averaged away when a collection of replicates is considered.

Note that the effect noticed in the bottom panel of
\cref{fig:local__block_length_norms_and_changes_in_percentage} also
is present for the global spectral densities (based on these
bootstrapped samples), so this phenomenon is thus not an artefact of
the way the local Gaussian spectral densities are estimated.

\textbf{The frequency-component:} 
\Cref{fig:local__block_length_norms_and_changes_in_percentage}
indicates that the block length sensitivity, as measured by
$D\!\parenR{\lgsdM{\LGp}{\omega}{m}}$, for the adjusted resampling
strategy from \cref{def:index_based_block_bootstrap_for_tuples} is
rather small.  But does this imply that these block lengths should
be considered \textit{equally good} or \textit{equally bad}?  That
can not be concluded from
\cref{fig:local__block_length_norms_and_changes_in_percentage}
alone, and it is thus necessary to also consider a plot that takes
the frequency-dimension into account, as is done in
\cref{fig:local__block_length_cibb_four_cases} for the four block
lengths $L\in\parenC{10,25,50,69}$.

\begin{figure}[h]
  {\centering
    \includegraphics[width=\textwidth]{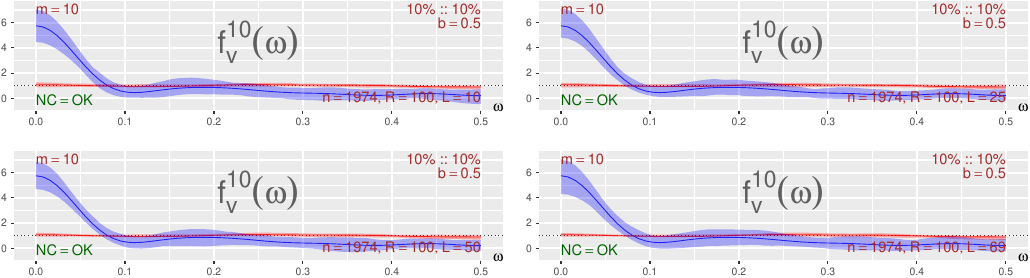}
  }
  \caption[]{Four different block lengths $L$ (from those
    investigated in
    \cref{fig:local__block_length_norms_and_changes_in_percentage})
    have here been used in the resampling strategy given in
    \cref{def:index_based_block_bootstrap_for_tuples}.  The values
    of $L$ are 10, 25, 50 and 69, and this information is plotted at
    the lower right corner of the plots.
  }\label{fig:local__block_length_cibb_four_cases}
\end{figure}

It is clear from \cref{fig:local__block_length_cibb_four_cases} that
the differences between these estimates are rather small, and it is
necessary to look closely in order to see that the pointwise
confidence intervals are slightly narrower near $\omega=0$ for the
two cases $L=25$ and $L=50$.  Moreover, the situation with minimal
differences between the estimates remains unchanged even if the
number of lags are increased to e.g.\ $h=50$.

This might at first sight be somewhat surprising (and a source of
concern), since it seems natural to assume that the block length $L$
should have a larger impact on the results.  However, this result is
actually quite natural to anticipate when the discussion from
\cref{app:block.length.and.expected.content.of.resampled.data} is
taken into account.  It was there noted that the algorithm that
estimates $\lgsdM{\LGp}{\omega}{m}$ does not use the temporal
connection \textit{between} the $(m+1)$-tuples in
$\mathcalYz[\sharp]{h:L} =
\TSR{\parenR{\Yz[\sharp]{t+m},\dotsc,\Yz[\sharp]{t+1},\Yz[\sharp]{t}}}{t=1}{n-m}$,
and that the important detail thus should be the expected number of
times the different tuples would occur in $\mathcalYz[\sharp]{h:L}$.

It was seen in
\cref{app:block.length.and.expected.content.of.resampled.data} that
the majority of the tuples were expected to occur
$\tfrac{n}{n-(L-1)}$ times, and this value hardly changes when
$n=1974$ and $L$ goes from 10 to 69.  There are of course also
differences with regard to the expected number of corrupt tuples for
different block-lengths, cf.\
\cref{eq:th:ibbb_expected_number_corrupt} in
\cref{th:ibbb_expected_number_corrupt}, but the data in
\cref{app:table:corrupt.tuples.ibb.bootstrap} clearly indicates that
this effect also can be considered minuscule.

The effect of different block lengths $L$ will of course be larger
if this resampling strategy is used on a short sample, but for such
samples it might not really be natural to compute the local Gaussian
spectrum in the first place (since the bandwidth $\bm{b}$ in such
cases must be large, and this tends to blot out local differences in
the spectrum).

\textbf{An additional example:} The preceding discussion about the
anticipated outcome is completely general in nature, but one might
still wonder if the results in
\cref{fig:local__block_length_norms_and_changes_in_percentage,fig:local__block_length_cibb_four_cases}
would have looked significantly differently if another case than the
\texttt{dmpb}-data had been used for the investigation.  This is
easy to investigate since the relevant scripts trivially can be
adjusted to investigate other cases too, like e.g.\ a single
realisation from the \textit{local trigonometric} time series, cf.\
\cref{fig:trigonometric,fig:dmt_heatmap_and__levels_vs_norm} in
\cref{sec:Deterministic_trigonometric_models} and the discussion in
\cref{app:fig:trigonometric}.

The result for the distance based box-plots for this new
investigation was (as expected) very similar to the result seen in
\cref{fig:local__block_length_norms_and_changes_in_percentage}.  The
analogue of \cref{fig:local__block_length_cibb_four_cases} is shown
in
\cref{fig:local__block_length_cibb_four_cases_local.trigonometric},
and it seems in fact to be the case that the differences between the
pointwise confidence intervals in this case is even smaller than
those observed in \cref{fig:local__block_length_cibb_four_cases}.

\begin{figure}[h]
  {\centering
    \includegraphics[width=\textwidth]{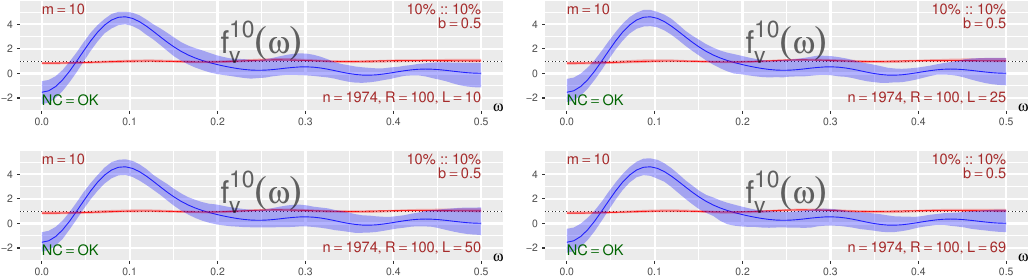}
  }
  \caption[]{A plot similar to
    \cref{fig:local__block_length_cibb_four_cases_local.trigonometric},
    that shows that the adjusted resampling strategy from
    \cref{def:index_based_block_bootstrap_for_tuples} also works
    when used on a single realisation from the local trigonometric
    time series used in
    \cref{fig:trigonometric,fig:dmt_heatmap_and__levels_vs_norm} in
    \cref{sec:Deterministic_trigonometric_models}, see
    \cref{app:fig:trigonometric} for details.
  }\label{fig:local__block_length_cibb_four_cases_local.trigonometric}
\end{figure}

\textbf{Conclusion:} The preceding discussion (based on the
\texttt{dmbp}-data and a \textit{local trigonometric} example)
indicates that the block length $L$ does not seem to have a major
impact on the estimates and pointwise confidence intervals obtained
from the adjusted resampling strategy given in
\cref{def:index_based_block_bootstrap_for_tuples}.  This simplifies
the task described in \cref{sec:Real_data}, i.e.\ to figure out if
it for a given sample of size $n$ seems reasonable to claim that an
observed difference between estimates of $\fz[m]{}(\omega)$ and
$\lgsdM{\LGp}{\omega}{m}$ is large enough to conclude that the
sample at hand do have a non-Gaussian dependency structure.

For other cases, it seems natural to recommend an approach where
different block lengths $L$ are tested (like seen in
\cref{fig:local__block_length_cibb_four_cases,fig:local__block_length_cibb_four_cases_local.trigonometric}),
in order to safeguard against the possibility that the present
examples for some reason both should be exceptional cases.

\textbf{Reproducibility:} The scripts needed for the reproduction of
the plots in this section
are 
included in the \Rpackage \lgsdRpackage\ (see \cref{app:data_details}
for further details), and the interested reader can there easily
adjust the range of the block lengths to be used.  It is also
possible to adjust all the other tuning parameters needed for the
estimation of $\lgsdM{\LGp}{\omega}{m}$, and it is even possible to
perform the computations with the ordinary block bootstrap if so
should be desired.

\subsection{What about the ordinary block bootstrap?}
\label{app:What.about.the.ordinary.block.bootstrap?}

It was originally the ordinary block bootstrap that was used as the
resampling strategy in this paper, and it could be of interest to
include a few comments related to this method.

First of all, note that there are data-driven methods for the
selection of the block length to be used with the block bootstrap,
see e.g.\
\citet{politis1994limit,buhlmann1999block,politis2004automatic,patton2009correction,lahiri2007nonparametric,nordman2014convergence}
--- but these methods do not produce good results when used upon
data with a nonlinear structure and a flat (ordinary) spectrum.

The \enquote{problem} is easily detected from an inspection of the
selection algorithms in sections 3.2 and 3.3 in
\citet{politis2004automatic}, as they all have a factor
\mbox{$G \defeq \sumss{h=-\infty}{\infty} |h|R(h)$} where $R(h)$ is
the lag~$h$ autocovariance of the series under investigation.  For a
time series whose ordinary spectrum is flat, the only nonzero $R(h)$
occurs when \mbox{$h=0$}, and the sum $G$ thus becomes zero in this
case.  This implies that the data-driven block length algorithms
(both for the stationary and for the circular bootstrap) considers a
very short block length to be suitable when bootstrapping the
\texttt{dmbp} data --- and that would obviously destroy all
nonlinear structures in the data.

To the best of the authors' knowledge, it does not exist any
adjustments of the algorithms used for block length selection that
is suited for a situation with a flat global spectrum, and the block
length $L$ thus had to be selected manually.  A sensitivity analysis
of the block length argument for the ordinary block bootstrap showed
something similar to
\cref{fig:local__block_length_norms_and_changes_in_percentage} when
$L$ was large enough, e.g.\ the range from $L=75$ to $L=135$.
It was mentioned on page \pageref{page:L=100} in the main part that
a block length of $L=100$ had been used in an earlier draft of this
paper (selected due to a visual inspection of the
$\hatlgacr{\LGp}{h}$-values seen in \cref{fig:dmbp_lag}, and after
the testing of a few alternatives), and it can be noted that the
pointwise confidence intervals then looked very similar to those
based on the adjusted block bootstrap, cf.\ \cref{fig:dmbp}.

The results for shorter block lengths could on the other hand be
rather bad, but that is hardly surprising based on the observations
in \cref{app:table:corrupt.tuples.block.bootstrap} (in
\cref{app:corrupt.tuples.edge.effect.block.bootstrap}) about the
fractions of corrupt tuples that occurs when the block bootstrap is
used on a short sample.

In retrospect it is clear that the optimal resampling strategy would
have been to use the block bootstrap on the derived $(m+1)$-variate
time series
$\TSR{\parenR{\Yz{t+m},\dotsc,\Yz{t+1},\Yz{t}}}{t=1}{n-m}$, since
that would have eliminated all of the edge-effect noise --- see
\cref{sec:natural_solution_to_edge_effect_issue?} for the details,
and an argument in favour of using the adjustment from
\cref{def:index_based_block_bootstrap_for_tuples} instead.  Note
that a resampling of the derived $(m+1)$-variate time series will
have the same properties as those discussed in
\cref{app:block.length.and.expected.content.of.resampled.data}, and
the sensitivity of the block-length argument $L$ should in this case
be similar to those seen in
\cref{fig:local__block_length_norms_and_changes_in_percentage,fig:local__block_length_cibb_four_cases}
for the adjusted resampling strategy.

The ordinary block bootstrap (working on $\TSR{\Yz{t}}{t=1}{n}$) is
available as a resampling strategy in the \Rpackage \lgsdRpackage,
but the default for this task is the adjusted block bootstrap from
\cref{def:index_based_block_bootstrap_for_tuples}.

\section{Scripts and details related to the examples}
\label{app:data_details}
\setcounter{figure}{0} 

The reproducibility of all the examples in this paper can be done by
the scripts contained in the \Rpackage \lgsdRpackage, and
\cref{app:The.scripts.in.lgsdRpackage} explains how the interested
reader can extract these scripts.
Additional details about the GARCH$(1,1)$-example seen in
\cref{fig:GARCH_intro_example_short}, and the apARCH$(2,3)$-example
seen in \cref{fig:GARCH}, are given
in \cref{app:fig:GARCH_intro_example_short,app:fig:GARCH}.

\Cref{app:fig:trigonometric} investigates the \textit{local
  trigonometric} example seen in
\cref{fig:trigonometric,fig:dmt_heatmap_and__levels_vs_norm}.  It
starts with a theoretical investigation of the general construction
of which the \textit{local trigonometric} example is a particular
realisation, and it then gives the heuristic arguments that enables
this example to be used for the sanity testing of the implemented
estimation algorithm.

The last part of \cref{app:fig:trigonometric} verifies that it for a
large sample is possible to detect an elusive 
component that only occurs with probability 0.05, 
and it ends with some comments related to issues that can occur
(under specific circumstances) when the local Gaussian machinery is
used on a time series whose global spectrum does not look like white
noise.

\subsection{The scripts in the \Rpackage \lgsdRpackage}
\label{app:The.scripts.in.lgsdRpackage}

All the examples in this paper (and the related multivariate paper
\myblind{\citet{jordanger17:_lgcsd}}{blinded reference}) can be
reproduced by the scripts in the \Rpackage \lgsdRpackage.  This
\Rpackage can
be installed 
by using \enquote{\devtoolgithublgsdRpackage}.  
The simplest way to extract the scripts from the internal storage of
this \Rpackage is to use the \Rfunction
\enquote{\Rref{LG\_extract\_scripts()}} after the package has been
installed.

These scripts can either be used as they are (reproduction of the
examples in this paper), or they can be used as templates for
similar investigations of other samples/models that the user would
like to investigate.  In the latter case, see
\cref{How.to.select.the.tuning.parameters?} for some comments
related to the selection of the tuning parameters of the estimation
algorithm.

The reproduction of the figures requires two different scripts.  The
first scripts contain the code needed for the estimation of the
local Gaussian autocorrelations $\lgacr{\LGp}{h}$ for all the
specified combinations of the tuning parameters, whereas the second
scripts contain the code that creates the particular visualisations
seen in the figures in this paper. 
Note that it is sufficient to use the first type of scripts in order
to use the integrated \Rref{shiny}-application that enables an easy
interactive investigation of the resulting estimates.  The second
type of scripts is first needed when one wants to put many figures
into one larger grid.

\subsection{The GARCH$(1,1)$-example in
  \cref{fig:GARCH_intro_example_short}}
\label{app:fig:GARCH_intro_example_short}
A GARCH$(1,1)$ example was in \cref{fig:GARCH_intro_example_short}
used to show that the local Gaussian spectral density could
detect 
dependency structures that the ordinary spectral density did
not detect.

The following description of the standard GARCH model,
introduced in \citet{bollerslev86:_gener}, is taken from the
vignette for the \Rref{rugarch}-package \citet{ghalanos15_rugarch},
\begin{align}
  \label{eq:GARCH_general_model}
  \sigmaz[2]{t} 
  &= \parenR{\omega + \sumss{j=1}{m} \zetaz{j} \vz{jt}} 
    + \sumss{i=1}{q} \alphaz{j} \varepsilonz[2]{t-j}  
    + \sumss{i=1}{p} \betaz{j} \sigmaz[2]{t-j}, 
\end{align}
with $\sigmaz[2]{t}$ denoting the conditional variance, $\omega$ the
intercept and $\varepsilonz[2]{t}$ the residuals from the mean
filtration process.  The GARCH order is defined by $(q,p)$ (ARCH,
GARCH), and it can also be $m$ external regressors $\vz{jm}$ which
are passed \textit{pre-lagged}.  Consult
\citet[sec.~2.2.1]{ghalanos15:_rugarch} for further details.

The \Rcode below specifies the parameters for the GARCH$(1,1)$ model
in
\cref{fig:GARCH_intro_example_short}.  
\begin{kframe}
\begin{alltt}
\hlkwd{library}\hlstd{(rugarch)}
\hlstd{.spec} \hlkwb{<-} \hlkwd{ugarchspec}\hlstd{(}
    \hlkwc{variance.model}\hlstd{=}\hlkwd{list}\hlstd{(}\hlkwc{model}\hlstd{=}\hlstr{"sGARCH"}\hlstd{,}
                        \hlkwc{garchOrder}\hlstd{=}\hlkwd{c}\hlstd{(}\hlnum{1}\hlstd{,}\hlnum{1}\hlstd{)),}
    \hlkwc{mean.model}\hlstd{=}\hlkwd{list}\hlstd{(}\hlkwc{armaOrder}\hlstd{=}\hlkwd{c}\hlstd{(}\hlnum{0}\hlstd{,}\hlnum{0}\hlstd{),}
                    \hlkwc{include.mean}\hlstd{=}\hlnum{TRUE}\hlstd{),}
    \hlkwc{distribution.model}\hlstd{=}\hlstr{"norm"}\hlstd{,}
    \hlkwc{fixed.pars}\hlstd{=}\hlkwd{list}\hlstd{(}\hlkwc{mu}\hlstd{=}\hlnum{0.001}\hlstd{,}
                    \hlkwc{omega}\hlstd{=}\hlnum{0.00001}\hlstd{,}
                    \hlkwc{alpha1}\hlstd{=}\hlnum{0.02}\hlstd{,}
                    \hlkwc{beta1}\hlstd{=}\hlnum{0.95}\hlstd{))}
\end{alltt}
\end{kframe}

\subsection{The apARCH$(2,3)$-example in
  \cref{fig:GARCH}}
\label{app:fig:GARCH}

The apARCH(2,3)-example seen in \cref{fig:GARCH} (see also
\cref{fig:apARCH_heatmap_and__levels_vs_norm,fig:apARCH_v_heatmap_lgacr})
had coefficients that were fitted to the \texttt{dmbp}-data by the
help of the \Rref{rugarch}-package \citet{ghalanos15_rugarch}, and
this particular model was selected after a testing procedure that
tried out several thousand different variations of the GARCH-type
models implemented in the \Rref{rugarch}-package.

The apARCH$(p,q)$ model (for observations \epsilonz{t}) was in
\citet{ding1993long} introduced as
\begin{subequations}
  \label{eq:apARCH_general_model}
  \begin{align}
    &\epsilonz{t} = \sz{t}\ez{t},\qquad \ez{t}\sim \UVN{0}{1}, \\
    &\sz[\delta]{t} = \alphaz{0} + \sumss{i=1}{p} \alphaz{i}
      \subp{\parenR{\absp{\epsilonz{t-i}}-\gammaz{i}\epsilonz{t-i}}}{}{}{}{\delta}
      + \sumss{j=1}{q}\betaz{j}\sz[\delta]{t-i},
  \end{align}
\end{subequations}
where $\alphaz{0} > 0$, $\delta \geq 0$, $\alphaz{i} \geq 0$ and
$-1 < \gammaz{i} < 1$ for $i=1\dotsc,p$, and $\betaz{j} \geq 0$ for
$j=1\dotsc,q$.

The description of this model in the \Rref{rugarch}-package, cf.\
\citet[sec.~2.2.5]{ghalanos15:_rugarch},
is slightly different: The constant $\alphaz{0}$ is there replaced
with $\parenR{\omega + \sumss{j=1}{m} \zetaz{j} \vz{jt}}$, which is
the same term that was used in \cref{eq:GARCH_general_model}, see
the previous section for details.

\subsection{The local trigonometric example in
  \cref{fig:trigonometric,fig:dmt_heatmap_and__levels_vs_norm}}
\label{app:fig:trigonometric}

This section will discuss some topics related to the \textit{local
  trigonometric} example, whose local Gaussian spectral density was
investigated in
\cref{fig:trigonometric,fig:dmt_heatmap_and__levels_vs_norm} of
\cref{sec:Deterministic_trigonometric_models}.  A few basic results
related to the general construction are given in
\cref{app:fig:trigonometric.general.properties}, whereas
\cref{app:fig:trigonometric.heuristic.argument} presents the arguments
that enabled this example to be used for the sanity testing of the
implemented estimation algorithm.

It was noted in \cref{sec:Deterministic_trigonometric_models} that the
first component of the \textit{local trigonometric} example could not
be detected in the short sample investigated in that section, but it
is possible to detected it when the sample-size is large enough, cf.\
\cref{fig:local.trigonometric.C1-component} in
\cref{app:fig:trigonometric.C1.component}.

Finally, \cref{sec:lgch:beware_of_global_structures} highlights
issues that can occur (under specific circumstances) when this
machinery is used on a time series whose global spectrum does not
look like white noise, as seen in
\cref{fig:trigonometric_one_cosine_tiny_noise} where the
$m$-truncated local Gaussian spectrum has been estimated for samples
from a deterministic function perturbed by very low random
fluctuations.

\subsubsection{Some properties of the general construction}
\label{app:fig:trigonometric.general.properties}

Recall that the \textit{local trigonometric} example is a particular
case of a general construction, in which a new time series
$\TSR{\Yz{t}}{t\in\ZZ}{}$ is constructed by the following scheme:
\begin{enumerate}
\item 
  Select $r$ time series $\TSR{\Cz{i}(t)}{i=1}{r}$.
\item 
  Select a random variable $I$ with values in the set
  $\parenC{1,\dotsc,r}$, and use this to sample a collection of
  indices $\TSR{\Iz{t}}{t\in\ZZ}{}$ (i.e.\ for each $t$ an
  independent realisation of $I$ is taken).  Let
  $\pz{i}\defeq\Prob{\Iz{i}=i}$ denote the probabilities for the
  different outcomes.
\item 
  Define $\Yz{t}$ by
  means of the equation
  \begin{align}
    \label{app:eq:local_trigonometric_sum}
    \Yz{t} \defeq \sumss{i=1}{r} \Ind{\Iz{t} = i}\cdot\Cz{i}(t).
  \end{align}
\end{enumerate}

The basic properties of $\TSR{\Yz{t}}{t\in\ZZ}{}$ can be expressed
relatively those of $\TSR{\Cz{i}(t)}{i=1}{r}$, as seen in the
following result.

\begin{lemma} ½
  \label{app:lemma:properties.of.the.artificial.time.series}
  With $\TSR{\Yz{t}}{t\in\ZZ}{}$ as defined above, it follows that:
  \begin{enumerate}[label=(\alph*)]
  \item 
    \label{app:lemma:properties.of.the.artificial.time.series..E}
    $\E{\Yz{t}} = \sumss{i=1}{r} \pz{i}\cdot\E{\Cz{i}(t)}$
  \item 
    \label{app:lemma:properties.of.the.artificial.time.series..E2}
    $\E{\Yz{t+h}\cdot\Yz{t}} =
    \begin{cases}
      \sumss{i=1}{r}\sumss{j=1}{r} \pz{i}\cdot\pz{j}\cdot\E{\Cz{i}(t+h)\cdot\Cz{i}(t)}
      & h\neq 0 \\
      \sumss{i=1}{r} \pz{i}\cdot\E{\Cz{i}(t)^2} & h=0
    \end{cases}$
  \item 
    \label{app:lemma:properties.of.the.artificial.time.series..Cov}
    $\Cov{\Yz{t+h}}{\Yz{t}} =
    \begin{cases}
      \sumss{i=1}{r} \sumss{j=1}{r} \pz{i}\cdot\pz{j}\cdot
      \Cov{\Cz{i}(t+h)}{\Cz{j}(t)} & h \neq 0 \\
      \sumss{i=1}{r} \pz{i}\cdot\E{\Cz{i}(t)^2} -
      \parenR{\sumss{i=1}{r} \pz{i}\cdot\E{\Cz{i}(t)}}^2 & h = 0
    \end{cases}$
  \item 
    \label{app:lemma:properties.of.the.artificial.time.series..Cov..independent}
    The additional assumption that $\Cz{i}(t)$ and $\Cz{j}(t)$ are
    independent when $i\neq j$, simplifies the $h\neq0$ case to:
    $\Cov{\Yz{t+h}}{\Yz{t}}=\sumss{i=1}{r} \pz[2]{i}\cdot
    \Cov{\Cz{i}(t+h)}{\Cz{i}(t)}$.
  \end{enumerate}
\end{lemma}

\begin{proof}
  The random variable $\Iz{t}$ that produces the set of indices is
  independent of $\Cz{i}(t)$, and
  \cref{app:lemma:properties.of.the.artificial.time.series..E} thus
  follows without further ado.  For the $h\neq0$ case of
  \cref{app:lemma:properties.of.the.artificial.time.series..E2} it is
  sufficient to note that $\Iz{t+h}$ and $\Iz{t}$ then are
  independent, and it follows that
  $\E{\Ind{\Iz{t+h}=i}\cdot\Ind{\Iz{t}=j}} =
  \E{\Ind{\Iz{t+h}=i}}\cdot\E{\Ind{\Iz{t}=j}} =
  \Prob{\Iz{t+h}=i}\cdot\Prob{\Iz{t}=j} = \pz{i}\cdot\pz{j}$. For the
  $h=0$ case of
  \cref{app:lemma:properties.of.the.artificial.time.series..E2} it is
  enough to note that $\Ind{\Iz{t}=i}\cdot\Ind{\Iz{t}=j}=0$ when
  $i\neq j$, which together with
  $\Ind{\Iz{t}=i}\cdot\Ind{\Iz{t}=i}=\Ind{\Iz{t}=i}$ gives the
  required expression. 
  The statements in
  \cref{app:lemma:properties.of.the.artificial.time.series..Cov,app:lemma:properties.of.the.artificial.time.series..Cov..independent}
  follows trivially from those in
  \cref{app:lemma:properties.of.the.artificial.time.series..E,app:lemma:properties.of.the.artificial.time.series..E2}.
\end{proof}

The key idea in the local trigonometric example is that the $r$ time
series $\Cz{i}(t)$ all should be \enquote{cosines with some noise},
since this implies (given a reasonable parameter configuration) that
it should be possible to present a decent guesstimate with regard to
the expected shape of 
the $m$-truncated local Gaussian spectral density (for some carefully
selected tuning parameters of the estimation algorithm).  The global
spectrum in this case will not be flat, but it will for low truncation
levels be \enquote{flat enough} for the purpose of showing that the
global spectrum does not detect the underlying frequencies whereas the
local Gaussian spectral density function can do that task.

The following result reiterates the $\Cz{i}(t)$-definition used in
the local trigonometric example, and it presents some basic properties
related to this definition.

\begin{lemma}
  \label{app:lemma:properties.of.the.artificial.time.series..cosine.part}
  Let
  $\Cz{i}(t) = \Lz{i} + \Az{i}(t) \cdot \cos \left(2\pi\alphaz{i} t +
    \varphiz{i} \right)$,
  be defined in the following manner: 
  $\Lz{i}$ and $\alphaz{i}$ are constants that respectively defines
  the horizontal base-line and the frequency.
  The amplitude $\Az{i}(t)$ are for each $t$ uniformly distributed on
  an interval $\parenS{\az{i},\bz{i}}$, and $\Az{i}(t+h)$ and
  $\Az{i}(t)$ are independent when $h\neq0$. 
  The phase-adjustment $\varphiz{i}$ are uniformly drawn (one time for
  each realisation) from the interval between~$0$ and~$2\pi$, and it
  is moreover assumed that the stochastic processes $\varphiz{i}$ and
  $\Az{i}(t)$ are independent of each other.

  \begin{enumerate}[label=(\alph*)]
  \item 
    \label{app:lemma:properties.of.the.artificial.time.series..cosine.part.E}
    $\E{\Cz{i}(t)} = \Lz{i}$
  \item 
    \label{app:lemma:properties.of.the.artificial.time.series..cosine.part.E2}
    $\E{\Cz{i}(t+h)\cdot\Cz{i}(t)} =
    \begin{cases}
      \Lz[2]{i} + \frac{\pi}{4}\cdot
      \parenR{\az[2]{i}+2\az{i}\bz{i}+\bz[2]{i}}\cdot
      \cos(2\pi\alphaz{i}\cdot h) & h \neq 0 \\
      \Lz[2]{i} + \frac{\pi}{3}\cdot
      \parenR{\az[2]{i}+\az{i}\bz{i}+\bz[2]{i}} & h = 0
    \end{cases}$
  \item  
    \label{app:lemma:properties.of.the.artificial.time.series..cosine.part.Cov}
    $\Cov{\Cz{i}(t+h)}{\Cz{i}(t)} =
    \begin{cases}
      \frac{\pi}{4}\cdot
      \parenR{\az[2]{i}+2\az{i}\bz{i}+\bz[2]{i}}\cdot
      \cos(2\pi\alphaz{i}\cdot h) &h\neq 0\\
      \frac{\pi}{3}\cdot
      \parenR{\az[2]{i}+\az{i}\bz{i}+\bz[2]{i}} &h= 0
    \end{cases}$
  \end{enumerate}  
\end{lemma}

\begin{proof}
  This is a consequence of the independence of the two stochastic
  processes $\Az{i}(t)$ and $\varphiz{i}$, and the basic observations:
  $\E{\Az{i}(t)}=\tfrac{1}{2}\cdot\parenR{\az{i}+\bz{i}}$,
  $\E{\Az[2]{i}(t)}=\tfrac{1}{3} \cdot
  \parenR{\az[2]{i}+\az{i}\bz{i}+\bz[2]{i}}$,
  $\E{\cos(2\pi\alphaz{i}t + \varphiz{i})}= 0$ and
  $\E{\cos(2\pi\alphaz{i}(t+h) + \varphiz{i})\cdot\cos(2\pi\alphaz{i}t
    + \varphiz{i})}= \pi\cdot\cos(2\pi\alphaz{i}\cdot h)$.  
  The proof of
  \cref{app:lemma:properties.of.the.artificial.time.series..cosine.part.E}
  is trivial.  For
  \cref{app:lemma:properties.of.the.artificial.time.series..cosine.part.E2}
  it suffices to observe that the $h\neq 0$ case contains
  $\E{\Az{i}(t)}^2$ as a factor, whereas the $h=0$ case contains
  $\E{\Az{i}(t)^2}$ as a factor.
  \Cref{app:lemma:properties.of.the.artificial.time.series..cosine.part.Cov}
  follows from
  \cref{app:lemma:properties.of.the.artificial.time.series..cosine.part.E,app:lemma:properties.of.the.artificial.time.series..cosine.part.E2}.
\end{proof}

Finally, the \textit{local trigonometric} example is obtained by using
$r$ time series $\Cz{i}(t)$, of the form given in
\cref{app:lemma:properties.of.the.artificial.time.series..cosine.part},
in the construction of the time series $\Yz{t}$, i.e.\
\begin{align}
  \label{eq:THE.local.trigonometric.example}
  \Yz{t} = \sumss{i=1}{r} \Ind{\Iz{t} = i}\cdot \parenR{\Lz{i} +
  \Az{i}(t) \cdot \cos \left(2\pi\alphaz{i} t + \varphiz{i} \right)},
\end{align}
where it furthermore is assumed that the $i$-indexed stochastic
variables $\Az{i}(t)$ and $\varphiz{i}$ are independent of the
$j$-indexed variants when $i\neq j$.  It now follows from
\cref{app:lemma:properties.of.the.artificial.time.series,app:lemma:properties.of.the.artificial.time.series..cosine.part}
that the $h\neq 0$ correlation of the time series $\Yz{t}$ in
\cref{eq:THE.local.trigonometric.example} is given by
\begin{align}
  \label{app:eq:correlation.of.the.artificial.time.series}
  \rhoz{Y}(h)
  &= \frac{\frac{\pi}{4}\cdot\sumss{i=1}{r}
    \pz[2]{i}\cdot\parenR{\az[2]{i} + 2\az{i}\bz{i} +
    \bz[2]{i}}\cdot\cos(2\pi\alphaz{i}\cdot h) }
    {\sumss{i=1}{r} \pz{i}\cdot \parenS{\Lz[2]{i} + \frac{\pi}{3}\cdot
    \parenR{\az[2]{i}+\az{i}\bz{i}+\bz[2]{i}}} -
    \parenRz[2]{}{\sumss{i=1}{r} \pz{i}\cdot\Lz{i}} }.
\end{align}

An inspection of
\cref{app:eq:correlation.of.the.artificial.time.series} reveals that
it is fairly easy to find a parameter configuration for which the
numerator is rather small compared to the denominator.  This is of
course not white noise, but the key idea is that it is close enough
to white noise to make it impossible to deduce anything about the
underlying frequencies $\alphaz{i}$ based on the ordinary spectrum.

\subsubsection{The heuristic argument that motivates the local
  trigonometric example}
\label{app:fig:trigonometric.heuristic.argument}

This section starts with an outline that shows how it is possible to
select the parameters of the \textit{local trigonometric} time series
from \cref{eq:THE.local.trigonometric.example} in such a manner that
some specified key features should be present after the
pseudo-normalisation of a sample.  It is with regard to this also
necessary to take into account the tuning-parameters of the estimation
algorithm (i.e.\ the bandwidth $\bm{b}$), since these must be adjusted
relative to the size $n$ of the sample.  The last part of this section
considers the example used in
\cref{sec:Deterministic_trigonometric_models}, and the discussion
related to \cref{fig:dmt_example_TS} will pinpoint why this
hand-waving approach actually works.

\textbf{The heuristic argument:} The amplitude $\Az{i}(t)$ is
uniformly distributed on $\parenS{\az{i},\bz{i}}$, and it thus follows
that all the observations from the $\Cz{i}$ component lies in the
interval $\mathcalIz{i}=\parenS{\Lz{i}-\bz{i},\Lz{i}+\bz{i}}$.  The
first key requirement is that the $r$ intervals $\mathcalIz{i}$ should
have a minimal amount of overlap, and it is moreover for simplicity
natural to require that the base-lines are ordered as follows
$\Lz{1} < \Lz{2} < \dots < \Lz{r}$.

The base lines $\Lz{i}$ do occur in the denominator of $\rhoz{Y}(h)$,
cf.\ \cref{app:eq:correlation.of.the.artificial.time.series}, but for
the purpose of the local Gaussian spectral density investigation it is
the corresponding values after the pseudo-normalisation that is of
interest.  This implies that the values of $\Lz{i}$, $\az{i}$ and
$\bz{i}$ are somewhat irrelevant, since minor modifications of them
will return exactly the same pseudo-normalised sample.

The key ingredient with regard to the pseudo-normalised version of the
sampled values that lies in a given interval $\mathcalIz{i}$, is the
the specification of the probability $\pz{i}\defeq\Prob{\Iz{t}=i}$.
Assuming that the intervals $\mathcalIz{i}$ does not overlap, it is
clear that it for a sample of size $n$ will be natural to assume that
approximately $\pz{i}\cdot n$ of the observations should lie in the
interval $\mathcalIz{i}$ --- and the symmetry of the cosine around its
baseline then implies that approximately one half of these
$\pz{i}\cdot n$ observations should lie below $\Lz{i}$ and the other
half above it.

It follows from this that the base-line $\Lz{1}$ of the $\Cz{1}(t)$
component should occur near the $\LGpi{1} \defeq \tfrac{1}{2}\pz{1}$
percentile of the sample, the \enquote{border-line} between
$\Cz{1}(t)$ and $\Cz{2}(t)$ near the $\pz{1}$ percentile, the
base-line $\Lz{2}$ of the $\Cz{2}(t)$ near the
$\LGpi{2} \defeq \parenR{\pz{1} + \tfrac{1}{2}\pz{2}}$ percentile,
and so on.  This implies that the part of the sample that lies in
the interval $\Iz{i}$ should correspond to the observations between
the two percentiles $\vz{i} - \tfrac{1}{2}\pz{i}$ and
$\vz{i} + \tfrac{1}{2}\pz{i}$,
and this part should moreover look like
a random selection of $\pz{i}\cdot n$ observations from the
$\Cz{i}(t)$ component.

The idea now is that the pseudo-normalisation of the $\Iz{i}$-part
of the sample still should contain a structure that reveals the
frequency $\alphaz{i}$ of the underlying cosine, and that it thus
(given a suitable combination of point $\LGp$ and bandwidth
$\bm{b}$) should be possible to get a result that looks
approximately like the result obtained when
$\lgsdM{\LGp}{\omega}{m}$ is estimated for a single cosine with a
frequency equal to $\alphaz{i}$.

It is possible to select the probabilities $\TSR{\pz{i}}{i=1}{r}$
such that the percentile $\vz{i}$ of the base-line $\Lz{i}$
corresponds directly to the diagonal point $\LGp$ for which
$\lgsdM{\LGp}{\omega}{m}$ should be estimated, but this does not
take into account that the pseudo-normalised version of the
$\Iz{i}$-part of the sample in general will not be symmetric around
$\Phiz[-1]{}(\vz{i})$ (with one exception when $\vz{i}=0.5$).  The
probabilities $\TSR{\pz{i}}{i=1}{r}$ should be selected such that
the $\vz{i}$ percentile lies closer to the center than the
percentile corresponding to $\LGp$.

Given a configuration of probabilities $\TSR{\pz{i}}{i=1}{r}$, and
furthermore assuming that the probability $\pz{i}$ for the
$\Cz{i}(t)$ component is sufficiently large relatively the sample
size $n$, it will now be possible to find a point $\LGp$ and a
bandwidth $\bm{b}$ such that the estimate of
$\lgsdM{\LGp}{\omega}{m}$ has the predicted shape with a peak at the
frequency $\alphaz{i}$ that is used in $\Cz{i}(t)$

It is possible to construct \textit{local trigonometric} examples
where some of the components $\Cz{i}(t)$ are impossible to detect
for a given sample size $n$, but they could still be detected when a
larger sample is used, cf.\ the example discussed in
\cref{app:fig:trigonometric.C1.component}.

\textbf{The case investigated in
  \cref{sec:Deterministic_trigonometric_models}:} 
The heuristic arguments outlined above were used in order to create
the \textit{local trigonometric} example in
\cref{sec:Deterministic_trigonometric_models}.  The initial
requirements for the construction of that particular example were that
the sample should have the same length as the \texttt{dmbp}-data,
i.e.\ $n=1974$, that the bandwidth should be $0.5$, and that the
investigation should be performed at the three diagonal points
corresponding to the 10\%, 50\% and 90\% percentiles of the standard
normal distribution.

The initial approach used three $\Cz{i}(t)$-components with equal
probability of being selected.  An additional new first component
$\Cz{1}(t)$ was then added, with $\pz{1}=0.05$ (and with a
corresponding reduction of the next probability to
$\pz{2}=\tfrac{1}{3}-0.05$).  This adjustment was done in order to get
more mileage out of the example, since it then also could be used to
highlight that the local Gaussian spectral density in some cases might
not have enough observations available to detect all the local
features.  Note that the elusive first component can be detected when
the sample size increases, cf.\
\cref{app:fig:trigonometric.C1.component} for details.

The explicit expression of the \textit{local trigonometric} time
series used in \cref{sec:Deterministic_trigonometric_models} is given
by the following equation,
\begin{align}
  \label{supplementary:eq:local_trigonometric_sum} 
  \Yz{t} \defeq \sumss{i=1}{4} \Ind{\Iz{t} = i} \cdot \parenR{\Lz{i} +
  \Az{i}(t) \cdot   \cos\! \left(2\pi\alphaz{i} t + \varphiz{i}
  \right)}, 
\end{align}
where the probabilities $\pz{i}\defeq\Prob{\Iz{t}=i}$ are given by
$(0.05, 1/3-0.05, 1/3, 1/3)$, and the frequencies $\alphaz{i}$ are
given by $(0.267, 0.091, 0.431, 0.270)$.  The base-lines $\Lz{i}$ are
given by the values $(-2, -1, 0, 1)$, and the lower and upper ranges
for the uniforms sampling of the amplitudes $\Az{i}(t)$ are
respectively given by {$(0.5, 0.2, 0.2, 0.5)$} and
{$(1.0, 0.5, 0.3, 0.6)$}.  Recall that these latter values are not
really of interest with regard to the pseudo-normalised version of the
sample, and the only requirement regarding these should be that they
are selected in order to give a minimal amount of overlap between the
different components.  The phase-adjustments $\varphiz{i}$ are
uniformly selected from the interval $[0,2\pi)$, one time for each
realisation of a sample from $\Yz{t}$.

\begin{knitrout}\footnotesize
  \definecolor{shadecolor}{rgb}{0.969, 0.969, 0.969}\color{fgcolor}\begin{figure}[h]

    {\centering \includegraphics[width=1\linewidth]{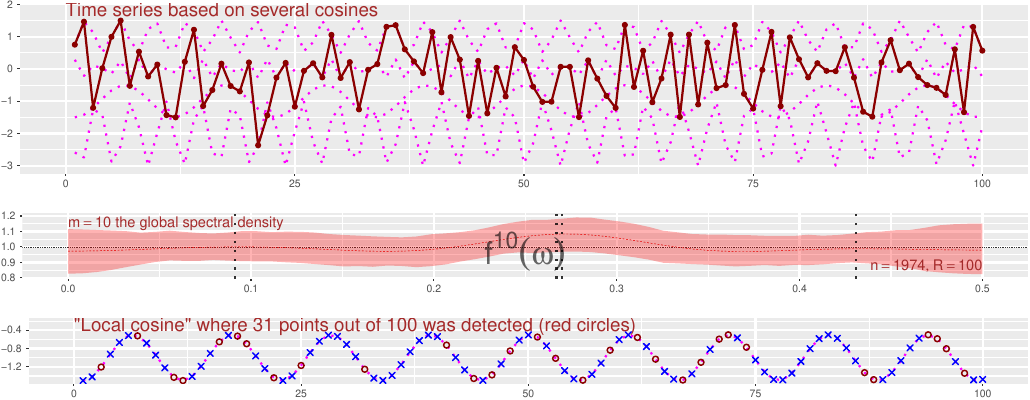} 

    }

    \caption{Top: Short excerpt from artifical example based on
      \textit{hidden trigonometric components}.  Center: Estimated
      (truncated) global spectral density (\textit{hidden frequencies}
      indicated with vertical lines).  Bottom: Local cosine showing the
      detected points at the local level centered at -1. Further details
      in the main text.}\label{fig:dmt_example_TS}
  \end{figure}

\end{knitrout}

\Cref{fig:dmt_example_TS} shows a \textit{simplified} excerpt of
length 100 from one realisation of $\Yz{t}$.  The amplitudes
$\Az{i}(t)$ have here for the simplicity of the present discussion
been fixed to the values {$(1.0, 0.5, 0.3, 0.5)$} since it is of
importance to emphasise which one of the underlying \enquote{hidden}
components $\Cz{i}(t)$ (shown as dotted curves in the top panel) that
was selected in this case (the phase-adjustments $\varphiz{i}$ in this
particular realisation are \mbox{$(0.52, 2.57, 3.24, 2.49)$}).
The center panel of \cref{fig:dmt_example_TS} shows an estimate of the
$m$-truncated (global) spectral density $\fz[m]{}(\omega)$, based on
100 independent samples of length~1974 and with a 90\% pointwise
confidence interval that shows that it is viable to claim that this
particular process behaves almost like white noise.  Note that the
vertical lines in the center panel shows the frequencies~$\alphaz{i}$
that was used in~\cref{eq:local_trigonometric_components}.

The bottom panel of \cref{fig:dmt_example_TS} is the one of major
interest for the present discussion, i.e.\ it is the one from which it
is possible to provide an explanation for the expected shape of the
local Gaussian spectral density, at some particularly designated
points~$\LGp$ (given a suitable bandwidth~$\bm{b}$).  
First of all, the bottom panel shows \textit{one} of the cosines from
the top panel, the
circles represents the points from the top panel that happened to lie
on this particular cosine~--- and the
crosses represents all the remaining points (at integer valued
times~$t$) of the cosine.  Recall that these points are from the
simplified realisation where the amplitudes $\Az{i}(t)$ are constant,
and that the actual values thus would be distorted a bit from those
observed here.

The
circles can be considered as a randomly selected collection of
points from a time series like the one investigated in
\cref{fig:trigonometric_one_cosine} (single cosine function with
some white noise), and the main point of interest is that it (for a
sufficiently long time series, and a sufficiently large
bandwidth~$\bm{b}$) will be the case that the estimated local
Gaussian autocorrelations based on this scarce subset might be quite
close to the estimates obtained if all the points had been
available.  The rationale for this claim is related to the way that
the local Gaussian auto-correlation at lag~$h$ (at a given
point~$\LGp$) is computed from the sets of bivariate points
$\parenR{\Yz{t+h},\Yz{t}}$.  In particular: It might not have a
detrimental effect on the resulting estimate if some of these
lag~$h$ pairs are removed at random, as long as the remaining number
of pairs is large enough.  Based on this idea, it can thus be argued
that the local Gaussian spectral density estimated from the
collection of the circled-marked points should be fairly close to
the situation shown in \cref{fig:trigonometric_one_cosine}, at least
if the time series under investigation is sufficiently~long.

This final heuristic graphical argument is the reason for the
guesstimate that the $m$-truncated local Gaussian spectral densities
(for the three points $\LGp$ corresponding to the 10\%, 50\% and 90\%
percentiles) should have an overall shape that resembles the one seen
for the single cosine example seen in
\cref{fig:trigonometric_one_cosine}.

It did turn out, cf.\
\cref{fig:trigonometric,fig:dmt_heatmap_and__levels_vs_norm}, that the
guesstimate based on these heuristic arguments in fact did hold true,
and the $m$-truncated local Gaussian spectral densities did in fact
detect the specified frequencies $\alphaz{i}$ in
\mbox{$(0.091, 0.431, 0.270)$} at the three targeted points $\LGp$.

Note that the frequency $0.267$ corresponding to the $\Cz{1}(t)$
component could not be detect based on only $n=1974$ observations,
since the probability $\pz{1}=0.05$ requires an investigation far out
in the lower tail.  It is however possible to detect it whit a much
larger sample size, cf.\ the discussion in the next section.

\subsubsection{Detecting the $\Cz{1}(t)$ component of the local
  trigonometric example}
\label{app:fig:trigonometric.C1.component}

The \textit{local trigonometric} example seen in
\cref{fig:trigonometric,fig:dmt_heatmap_and__levels_vs_norm} of
\cref{sec:Deterministic_trigonometric_models}, cf.\
\cref{supplementary:eq:local_trigonometric_sum} for the definition,
contains a component $\Cz{1}(t)$ that goes undetected when the sample
size of $n=1974$ is used.  The reason for the elusiveness of the
$\Cz{1}(t)$ component is that it has a probability of $\pz{1}=0.05$ of
being selected, which implies that it is expected to only find 98.7
observations from this component when $n=1974$.

The \enquote{border} between the observations from the $\Cz{1}(t)$
and $\Cz{2}(t)$ components should occur near the 5\% percentile, but
it is necessary to \enquote{zoom in} on a point $\LGp$ that lies
farther out in the tail than $\pz{1}/2=0.025$.  This requirement
occurs since the estimate of $\lgsdM{\LGp}{\omega}{m}$ should avoid
\enquote{contamination} from the observations from the $\Cz{2}(t)$
component.

Based on the idea that it might be necessary to go all the way out
to the 1\%, it seemed natural to attempt an investigation based on
$n=25000$ observations.  Since the point $\LGp$ now is far out in
the lower tail, e.g.\ the 1\% percentile of the standard normal
distribution is $-2.326$, it seemed reasonable to use the bandwidth
$\bm{b} = (0.4, 0.4)$.

The heatmap and distance plots in
\cref{fig:dmt_for_extreme_tail_heatmap_and__levels_vs_norm} is based
on an investigating of a single realisation, that included
percentiles based on values starting from 2 bandwidths below the 5\%
percentile and ending at 1/2 bandwidth below the 5\% percentile,
i.e.\ the diagonal points starts at approximately the 0.72\%
percentile and ends at the 3.25\% percentile.

\begin{knitrout}\footnotesize
  \definecolor{shadecolor}{rgb}{0.969, 0.969, 0.969}\color{fgcolor}\begin{figure}[h]

    {\centering \includegraphics[width=1\linewidth]{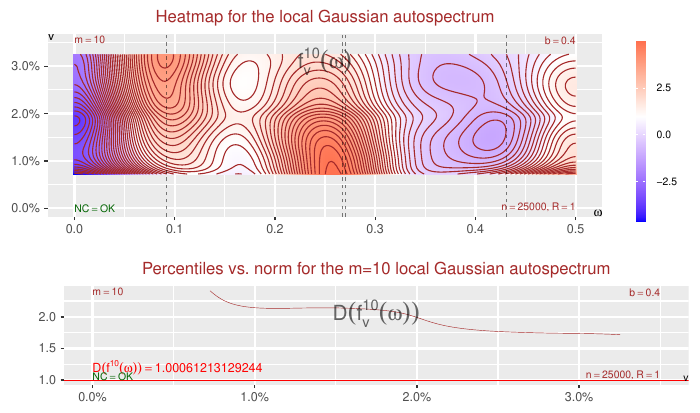} 

    }

    \caption{A heatmap+distance plot used to search for an
      \enquote{optimal} percentile that can reveal the $\alphaz{1}$
      frequency in the lower tail of the \textit{local
        trigonometric}.}
    \label{fig:dmt_for_extreme_tail_heatmap_and__levels_vs_norm}
  \end{figure}
\end{knitrout}

It is no surprise that the $\Cz{2}(t)$ component completely
dominates at the 3.25\% percentile, and it can be seen that it is
necessary to go down to at least the 1\% percentile in order to
detect a peak close to the frequency $\alphaz{1}=0.267$ of the
$\Cz{1}(t)$ component.  Note that
\cref{fig:dmt_for_extreme_tail_heatmap_and__levels_vs_norm} is based
on only 1 single realisation, and other realisations might look
slightly different.

\Cref{fig:local.trigonometric.C1-component} shows the situation when
$R=100$ replicates are used to estimate $\lgsdM{\LGp}{\omega}{m}$ at
the diagonal point $\LGp$ that corresponds to the 1\% percentile.
This shows that $\hatlgsdM{\LGp}{\omega}{m}$
in this case has the expected \enquote{cosine}-shape, and the peak
is at
the frequency $\alphaz{1}=0.267$ of the
$\Cz{1}(t)$ component.

\begin{knitrout}\footnotesize
  \definecolor{shadecolor}{rgb}{0.969, 0.969, 0.969}\color{fgcolor}\begin{figure}[h]

    {\centering \includegraphics[width=1\linewidth]{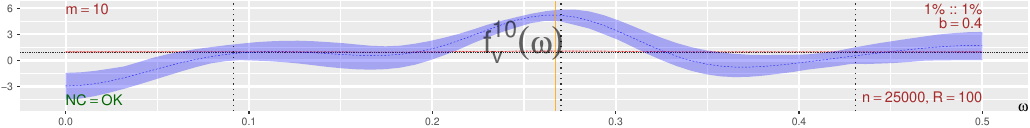} 

    }

    \caption{The detection of the $\alphaz{1}$ frequency in the lower tail
      of the \textit{local trigonometric} example requires a large sample
      and an investigation far out in the lower tail.}
    \label{fig:local.trigonometric.C1-component}
  \end{figure}
\end{knitrout}

\subsubsection{Beware of deterministic global structures under small
  noise}
\label{sec:lgch:beware_of_global_structures}

It was seen in the \textit{local trigonometric} example investigated
in the preceding sections, that a peak of
$\hatlgsdM{\LGp}{\omega}{m}$, like the one seen in
\cref{fig:dmt_for_extreme_tail_heatmap_and__levels_vs_norm},
corresponded to a frequency $\alpha$ of some underlying
cosine-function (that $\widehatfz[m]{}(\omega)$ did not detect).  It
should here be emphasised that the \textit{local trigonometric}
example was fine tuned in order to test the sanity of the
implemented estimation algorithm --- and it would thus be a logical
fallacy to conclude that a similar peak of
$\hatlgsdM{\LGp}{\omega}{m}$ (not present in
$\widehatfz[m]{}(\omega)$) always could be interpreted in the same
manner for general non-Gaussian time series.

An investigation of this issue can be seen in
\cref{fig:trigonometric_one_cosine_tiny_noise}, where an extreme
version of the case investigated in
\cref{fig:trigonometric_one_cosine} are presented.  The setup is
similar to the one from \cref{fig:trigonometric_one_cosine}, i.e.\ the
plots are based on 100 samples of length~1974 from a model of the form
\mbox{$\Yz{t}=\cos\!\parenR{2\pi\alpha t + \varphi} + \wz{t}$}, where
\mbox{$\alpha=0.302$} (as before), but the standard deviation of the
Gaussian white noise $\wz{t}$ has now been reduced to
\mbox{$\sigma = 0.05$}.

\begin{knitrout}\footnotesize
  \definecolor{shadecolor}{rgb}{0.969, 0.969, 0.969}\color{fgcolor}\begin{figure}[h]

    {\centering \includegraphics[width=1\linewidth]{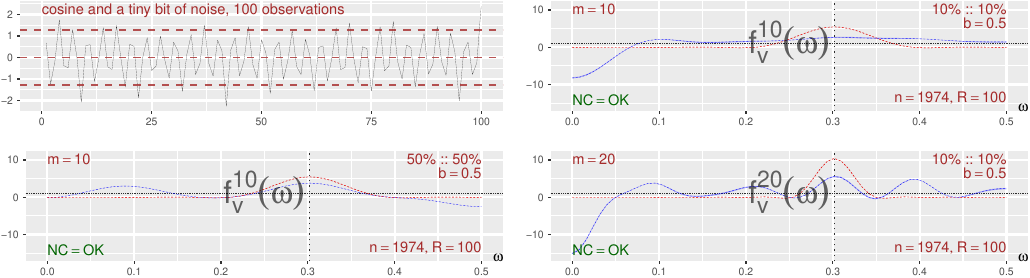} 

    }

    \caption[Pseudo-normalised single cosine with small
    noise]{Pseudo-normalised single cosine with small
      noise.}\label{fig:trigonometric_one_cosine_tiny_noise}
  \end{figure}
\end{knitrout}

The low value of the standard deviation~$\sigma$ implies that samples
from this model have a very clear periodic behaviour, as can be seen
from the plots in \cref{fig:trigonometric_one_cosine_tiny_noise},
where the 90\% confidence intervals are almost indistinguishable from
the mean of the estimates.  This clear periodicity is also evident
from the trace shown in the upper left panel of
\cref{fig:trigonometric_one_cosine_tiny_noise}, where the 100 first
pseudo-normalised observations of one of the samples are~presented.

\Cref{fig:trigonometric_one_cosine_tiny_noise} contains estimates of
the local and global spectra, with focus on the points in the lower
tail and the center for the truncation level \mbox{$m=10$}, and for
the lower tail when \mbox{$m=20$}.  The additional peaks seen at the
center is due to the kernel function $K$ that is used in the
estimation algorithm --- in particular $K$ works on the $h$-lagged
pairs $\parenR{\Yz{t+h},\Yz{t}}$, the contribution becomes
negligible for pairs far away from $\LGp$, the pairs that do
contribute give the impression that several \enquote{local
  frequencies} might be present, but the underlying model has only
one single frequency.

The case in \cref{fig:trigonometric_one_cosine_tiny_noise} is
extreme since the noise-term is minuscule.  Because the local
Gaussian correlation is based on a continuous distribution
assumption and the use of a kernel function, similar difficulties
can be expected for other deterministic functions embedded in low
noise. One possible way out might be to consider an approach where a
parametric model is fitted first to the data and then examine the
residuals with a global and a local spectral analysis.

\clearpage{}
  }
  \putbib     
\end{bibunit}


\begin{thebibliography}{48}
\expandafter\ifx\csname natexlab\endcsname\relax\def\natexlab#1{#1}\fi
\expandafter\ifx\csname url\endcsname\relax
  \def\url#1{\texttt{#1}}\fi
\expandafter\ifx\csname urlprefix\endcsname\relax\def\urlprefix{URL }\fi

\bibitem[{Berentsen et~al.(2017)Berentsen, Cao, Francisco-Fern{\'a}ndez, and
  Tj{\o}stheim}]{berentsen16:_some_proper_local_gauss_correl}
Berentsen, G.~D., Cao, R., Francisco-Fern{\'a}ndez, M., Tj{\o}stheim, D., 2017.
  {Some Properties of Local Gaussian Correlation and Other Nonlinear Dependence
  Measures}. Journal of Time Series Analysis 38~(2), 352--380.
\newline\urlprefix\url{http://dx.doi.org/10.1111/jtsa.12183}

\bibitem[{Berentsen et~al.(2014{\natexlab{a}})Berentsen, Kleppe, and
  Tj{\o}stheim}]{Berentsen:2014:ILR}
Berentsen, G.~D., Kleppe, T.~S., Tj{\o}stheim, D.~B., Feb. 2014{\natexlab{a}}.
  {Introducing \texttt{localgauss}, an R Package for Estimating and Visualizing
  Local Gaussian Correlation}. j-J-STAT-SOFT 56~(12).
\newline\urlprefix\url{http://www.jstatsoft.org/v56/i12}

\bibitem[{Berentsen and
  Tj{\o}stheim(2014)}]{Berentsen2014:departure_from_independence}
Berentsen, G.~D., Tj{\o}stheim, D., 2014. Recognizing and visualizing
  departures from independence in bivariate data using local {G}aussian
  correlation. Statistics and Computing 24~(5), 785--801.
\newline\urlprefix\url{http://dx.doi.org/10.1007/s11222-013-9402-8}

\bibitem[{Berentsen et~al.(2014{\natexlab{b}})Berentsen, Tj{\o}stheim, and
  Nordb{\o}}]{Berentsen2014:Copula}
Berentsen, G.~D., Tj{\o}stheim, D., Nordb{\o}, T., 2014{\natexlab{b}}.
  {Recognizing and visualizing copulas: An approach using local Gaussian
  approximation}. Insurance: Mathematics and Economics 57, 90 -- 103.
\newline\urlprefix\url{http://www.sciencedirect.com/science/article/pii/S0167668714000432}

\bibitem[{Birr et~al.(2019)Birr, Kley, and Volgushev}]{BIRR2019122}
Birr, S., Kley, T., Volgushev, S., 2019. {Model assessment for time series
  dynamics using copula spectral densities: A graphical tool}. {Journal of
  Multivariate Analysis} 172, 122 -- 146, {Dependence Models}.
\newline\urlprefix\url{http://www.sciencedirect.com/science/article/pii/S0047259X18301842}

\bibitem[{Bollerslev(1986)}]{bollerslev86:_gener}
Bollerslev, T., 1986. Generalized autoregressive conditional
  heteroskedasticity. Journal of Econometrics 31~(3), 307 -- 327.
\newline\urlprefix\url{http://www.sciencedirect.com/science/article/pii/0304407686900631}

\bibitem[{Bollerslev and Ghysels(1996)}]{bollerslev96:_period}
Bollerslev, T., Ghysels, E., 1996. {Periodic Autoregressive Conditional
  Heteroscedasticity}. Journal of Business \& Economic Statistics 14~(2),
  139--151.
\newline\urlprefix\url{http://amstat.tandfonline.com/doi/abs/10.1080/07350015.1996.10524640}

\bibitem[{Brillinger(1984)}]{Brillinger:1984:CWJa}
Brillinger, D.~R. (Ed.), 1984. The collected works of {John W. Tukey}. {Volume
  I}. {Time} series: 1949--1964. Wadsworth Statistics\slash Probability Series.
  Wadsworth, Pacific Grove, CA, USA, with introductory material by William S.
  Cleveland and Frederick Mosteller.

\bibitem[{Brillinger(1991)}]{brillinger1991some}
Brillinger, D.~R., 1991. Some history of the study of higher-order moments and
  spectra. Statistica Sinica 1~(465-476), 24J.
\newline\urlprefix\url{http://www3.stat.sinica.edu.tw/statistica/j1n2/j1n23/..\\j1n210\\j1n210.htm}

\bibitem[{Carrasco and Chen(2002)}]{carrasco_chen_2002}
Carrasco, M., Chen, X., 2002. Mixing and moment properties of various garch and
  stochastic volatility models. Econometric Theory 18~(1), 17–39.

\bibitem[{Chang et~al.(2017)Chang, Cheng, Allaire, Xie, and
  McPherson}]{chang16_shiny}
Chang, W., Cheng, J., Allaire, J., Xie, Y., McPherson, J., 2017. shiny: Web
  Application Framework for R. R package version 1.0.3.
\newline\urlprefix\url{https://CRAN.R-project.org/package=shiny}

\bibitem[{Davis and Mikosch(2009)}]{davis2009}
Davis, R.~A., Mikosch, T., 11 2009. {The extremogram: A correlogram for extreme
  events}. Bernoulli 15~(4), 977--1009.
\newline\urlprefix\url{https://doi.org/10.3150/09-BEJ213}

\bibitem[{Ghalanos(2020)}]{ghalanos15_rugarch}
Ghalanos, A., 2020. {rugarch: Univariate GARCH models.} R package version
  1.4-2.
\newline\urlprefix\url{https://cran.r-project.org/package=rugarch}

\bibitem[{Hagemann(2011)}]{hagemann2011robust}
Hagemann, A., November 2011. {Robust Spectral Analysis}.
\newline\urlprefix\url{https://ssrn.com/abstract=1956581}

\bibitem[{Han et~al.(2016)Han, Linton, Oka, and Whang}]{Han2016251}
Han, H., Linton, O., Oka, T., Whang, Y.-J., 2016. The cross-quantilogram:
  Measuring quantile dependence and testing directional predictability between
  time series. Journal of Econometrics 193~(1), 251 -- 270.
\newline\urlprefix\url{http://www.sciencedirect.com/science/article/pii/S0304407616300458}

\bibitem[{Hjort and Jones(1996)}]{hjort96:_local}
Hjort, N.~L., Jones, M.~C., 08 1996. Locally parametric nonparametric density
  estimation. Ann. Statist. 24~(4), 1619--1647.
\newline\urlprefix\url{http://dx.doi.org/10.1214/aos/1032298288}

\bibitem[{Hong(1999)}]{hong1999hypothesis}
Hong, Y., 1999. {Hypothesis Testing in Time Series via the Empirical
  Characteristic Function: A Generalized Spectral Density Approach}. Journal of
  the American Statistical Association 94~(448), 1201--1220.
\newline\urlprefix\url{http://tandfonline.com/doi/abs/10.1080/01621459.1999.10473874}

\bibitem[{Hong(2000)}]{hong2000generalized}
Hong, Y., 2000. Generalized spectral tests for serial dependence. Journal of
  the Royal Statistical Society: Series B (Statistical Methodology) 62~(3),
  557--574.
\newline\urlprefix\url{http://onlinelibrary.wiley.com/doi/10.1111/1467-9868.00250/abstract}

\bibitem[{Jordanger and Tj{\o}stheim(2017)}]{jordanger17:_lgcsd}
Jordanger, L.~A., Tj{\o}stheim, D., 2017. Nonlinear cross-spectrum analysis via
  the local gaussian correlation.
\newline\urlprefix\url{https://arxiv.org/abs/1708.02495}

\bibitem[{Klimko and Nelson(1978)}]{klimko1978}
Klimko, L.~A., Nelson, P.~I., 05 1978. {On Conditional Least Squares Estimation
  for Stochastic Processes}. Ann. Statist. 6~(3), 629--642.
\newline\urlprefix\url{http://dx.doi.org/10.1214/aos/1176344207}

\bibitem[{Kl{\"u}ppelberg and
  Mikosch(1994)}]{kl94:_some_limit_theor_self_normal}
Kl{\"u}ppelberg, C., Mikosch, T., 1994. {Some Limit Theory for the
  Self-Normalised Periodogram of Stable Processes}. {Scandinavian Journal of
  Statistics} 21~(4), 485--491.
\newline\urlprefix\url{http://www.jstor.org/stable/4616332}

\bibitem[{K{\"u}nsch(1989)}]{kuensch89:_jackk_boots_gener_station_obser}
K{\"u}nsch, H.~R., 1989. {The Jackknife and the Bootstrap for General
  Stationary Observations}. The Annals of Statistics 17~(3), 1217--1241.
\newline\urlprefix\url{http://www.jstor.org/stable/2241719}

\bibitem[{Lacal and Tj{\o}stheim(2017)}]{lacal2017local}
Lacal, V., Tj{\o}stheim, D., 2017. {Local Gaussian Autocorrelation and Tests
  for Serial Independence}. Journal of Time Series Analysis 38~(1), 51--71,
  10.1111/jtsa.12195.
\newline\urlprefix\url{http://dx.doi.org/10.1111/jtsa.12195}

\bibitem[{Lacal and Tj{\o}stheim(2018)}]{lacal:_estim}
Lacal, V., Tj{\o}stheim, D., 2018. {Estimating and Testing Nonlinear Local
  Dependence Between Two Time Series}. Journal of Business \& Economic
  Statistics 0~(0), 1--13.
\newline\urlprefix\url{https://doi.org/10.1080/07350015.2017.1407777}

\bibitem[{Li et~al.(2016)Li, Zhong, and Park}]{li2016generalized}
Li, H., Zhong, W., Park, S.~Y., 2016. {Generalized cross-spectral test for
  nonlinear Granger causality with applications to money–output and
  price–volume relations}. Economic Modelling 52, Part B, 661 -- 671.
\newline\urlprefix\url{http://www.sciencedirect.com/science/article/pii/S0264999315002916}

\bibitem[{Li(2008)}]{li08:_laplac_period_time_series_analy}
Li, T.-H., 2008. {Laplace Periodogram for Time Series Analysis}. Journal of the
  American Statistical Association 103~(482), 757--768.
\newline\urlprefix\url{http://dx.doi.org/10.1198/016214508000000265}

\bibitem[{Li(2010{\natexlab{a}})}]{li10:_nonlin_method_robus_spect_analy}
Li, T.-H., May 2010{\natexlab{a}}. {A Nonlinear Method for Robust Spectral
  Analysis}. IEEE Transactions on Signal Processing 58~(5), 2466--2474.
\newline\urlprefix\url{http://ieeexplore.ieee.org/abstract/document/5406102/}

\bibitem[{Li(2010{\natexlab{b}})}]{li10:_robus}
Li, T.-H., Aug 2010{\natexlab{b}}. {Robust coherence analysis in the frequency
  domain}. In: Signal Processing Conference, 2010 18th European. IEEE, pp.
  368--371.
\newline\urlprefix\url{http://ieeexplore.ieee.org/abstract/document/7096642/}

\bibitem[{Li(2010{\natexlab{c}})}]{Li20102133}
Li, T.-H., 2010{\natexlab{c}}. A robust periodogram for high-resolution
  spectral analysis. Signal Processing 90~(7), 2133 -- 2140.
\newline\urlprefix\url{http://www.sciencedirect.com/science/article/pii/S0165168410000137}

\bibitem[{Li(2012{\natexlab{a}})}]{li12:_detec}
Li, T.-H., March 2012{\natexlab{a}}. Detection and estimation of hidden
  periodicity in asymmetric noise by using quantile periodogram. In: 2012 IEEE
  International Conference on Acoustics, Speech and Signal Processing (ICASSP).
  pp. 3969--3972.
\newline\urlprefix\url{http://ieeexplore.ieee.org/abstract/document/6288787/}

\bibitem[{Li(2012{\natexlab{b}})}]{li2012robust}
Li, T.-H., 2012{\natexlab{b}}. On robust spectral analysis by least absolute
  deviations. Journal of Time Series Analysis 33~(2), 298--303.
\newline\urlprefix\url{http://dx.doi.org/10.1111/j.1467-9892.2011.00760.x}

\bibitem[{Li(2012{\natexlab{c}})}]{li2012quantile}
Li, T.-H., 2012{\natexlab{c}}. {Quantile Periodograms}. Journal of the American
  Statistical Association 107~(498), 765--776.
\newline\urlprefix\url{http://dx.doi.org/10.1080/01621459.2012.682815}

\bibitem[{Li(2014)}]{li2014quantile}
Li, T.-H., 2014. {Quantile Periodogram and Time-Dependent Variance}. Journal of
  Time Series Analysis 35~(4), 322--340.
\newline\urlprefix\url{http://dx.doi.org/10.1111/jtsa.12065}

\bibitem[{Li(2019)}]{li19:_quant_frequen_analy_spect_diver}
Li, T.-H., 2019. {Quantile-Frequency Analysis and Spectral Divergence Metrics
  for Diagnostic Checks of Time Series With Nonlinear Dynamics}. Papers,
  arXiv.org.
\newline\urlprefix\url{https://EconPapers.repec.org/RePEc:arx:papers:1908.02545}

\bibitem[{Linton and Whang(2007)}]{Linton2007250}
Linton, O., Whang, Y.-J., 2007. The quantilogram: With an application to
  evaluating directional predictability. Journal of Econometrics 141~(1), 250
  -- 282, semiparametric methods in econometrics.
\newline\urlprefix\url{http://www.sciencedirect.com/science/article/pii/S0304407607000152}

\bibitem[{Nelsen(2006)}]{Nelsen2006}
Nelsen, R.~B., 2006. An {I}ntroduction to {C}opulas -, 2nd Edition. Springer,
  Berlin, Heidelberg.

\bibitem[{Otneim and Tj{\o}stheim(2017)}]{otneim2017locally}
Otneim, H., Tj{\o}stheim, D., 2017. {The locally Gaussian density estimator for
  multivariate data}. Statistics and Computing 27~(6), 1595--1616.
\newline\urlprefix\url{https://doi.org/10.1007/s11222-016-9706-6}

\bibitem[{Otneim and Tj{\o}stheim(2018)}]{otneim2018conditional}
Otneim, H., Tj{\o}stheim, D., 2018. {Conditional density estimation using the
  local Gaussian correlation}. Statistics and Computing 28~(2), 303--321.
\newline\urlprefix\url{http://dx.doi.org/10.1007/s11222-017-9732-z}

\bibitem[{Silvapulle and Granger(2001)}]{silvapulle01:_large}
Silvapulle, P., Granger, C., 2001. Large returns, conditional correlation and
  portfolio diversification: a value-at-risk approach. Quantitative Finance
  1~(5), 542--551.
\newline\urlprefix\url{https://doi.org/10.1080/713665877}

\bibitem[{Sklar(1959)}]{sklar1959:_copula}
Sklar, A., 1959. {{Fonctions de R\'epartition \`a $n$ dimensions et leurs
  Marges}}. Publications de l'Institut de Statistique de l'Universit\'e de
  {P}aris 8, 229--231.

\bibitem[{St{\o}ve and Tj{\o}stheim(2014)}]{stoeve14:_measur}
St{\o}ve, B., Tj{\o}stheim, D., April 2014. Measuring asymmetries in financial
  returns: an empirical investigation using local gaussian correlation. In:
  Haldrup, N., Meitz, M., Saikkonen, P. (Eds.), {Essays in Nonlinear Time
  Series Econometrics}. No. 9780199679959 in OUP Catalogue. Oxford University
  Press, pp. 307--329.

\bibitem[{St{\o}ve et~al.(2014)St{\o}ve, Tj{\o}stheim, and
  Hufthammer}]{stove2014using}
St{\o}ve, B., Tj{\o}stheim, D., Hufthammer, K.~O., 2014. Using local gaussian
  correlation in a nonlinear re-examination of financial contagion. Journal of
  Empirical Finance 25~(C), 62--82.

\bibitem[{Sz{\'e}kely and Rizzo(2009)}]{szekely2009}
Sz{\'e}kely, G.~J., Rizzo, M.~L., 12 2009. Brownian distance covariance. Ann.
  Appl. Stat. 3~(4), 1236--1265.
\newline\urlprefix\url{https://doi.org/10.1214/09-AOAS312}

\bibitem[{Ter{\"a}svirta et~al.(2010)Ter{\"a}svirta, Tj{\o}stheim, Granger,
  et~al.}]{terasvirta2010modelling}
Ter{\"a}svirta, T., Tj{\o}stheim, D., Granger, C.~W., et~al., 2010. Modelling
  nonlinear economic time series. OUP Catalogue.

\bibitem[{Tj{\o}stheim and Hufthammer(2013)}]{Tjostheim201333}
Tj{\o}stheim, D., Hufthammer, K.~O., 2013. {Local Gaussian correlation: A new
  measure of dependence}. Journal of Econometrics 172~(1), 33 -- 48.
\newline\urlprefix\url{http://www.sciencedirect.com/science/article/pii/S0304407612001741}

\bibitem[{Tong(1990)}]{tong1990non}
Tong, H., 1990. Non-linear time series: a dynamical system approach. Oxford
  University Press.

\bibitem[{Tukey(1959)}]{Tukey:1959:IMS}
Tukey, J.~W., 1959. An introduction to the measurement of spectra. In:
  Grenander, U. (Ed.), Probability and Statistics, The {Harald Cram{\'e}r}
  Volume. Almqvist and Wiksell, Stockholm, Sweden, pp. 300--330.

\bibitem[{Wang and Hong(2017)}]{wang2017characteristic}
Wang, X., Hong, Y., 2017. {Characteristic function based testing for
  conditional independence: A nonparametric regression approach}. Econometric
  Theory, 1--35.
\newline\urlprefix\url{https://doi.org/10.1017/S026646661700010X}

\end{thebibliography}


\begin{thebibliography}{35}
\expandafter\ifx\csname natexlab\endcsname\relax\def\natexlab#1{#1}\fi
\expandafter\ifx\csname url\endcsname\relax
  \def\url#1{\texttt{#1}}\fi
\expandafter\ifx\csname urlprefix\endcsname\relax\def\urlprefix{URL }\fi

\bibitem[{Basseville(2013)}]{BASSEVILLE2013621}
Basseville, M., 2013. {Divergence measures for statistical data processing ---
  An annotated bibliography}. Signal Processing 93~(4), 621 -- 633.
\newline\urlprefix\url{http://www.sciencedirect.com/science/article/pii/S0165168412003222}

\bibitem[{Berentsen and
  Tj{\o}stheim(2014)}]{Berentsen2014:departure_from_independence}
Berentsen, G.~D., Tj{\o}stheim, D., 2014. Recognizing and visualizing
  departures from independence in bivariate data using local {G}aussian
  correlation. Statistics and Computing 24~(5), 785--801.
\newline\urlprefix\url{http://dx.doi.org/10.1007/s11222-013-9402-8}

\bibitem[{Billingsley(2012)}]{Billingsley12:_probab_measur}
Billingsley, P., 2012. Probability and {M}easure, {A}niversary Edition. Wiley.

\bibitem[{Birr et~al.(2019)Birr, Kley, and Volgushev}]{BIRR2019122}
Birr, S., Kley, T., Volgushev, S., 2019. {Model assessment for time series
  dynamics using copula spectral densities: A graphical tool}. {Journal of
  Multivariate Analysis} 172, 122 -- 146, {Dependence Models}.
\newline\urlprefix\url{http://www.sciencedirect.com/science/article/pii/S0047259X18301842}

\bibitem[{Bollerslev(1986)}]{bollerslev86:_gener}
Bollerslev, T., 1986. Generalized autoregressive conditional
  heteroskedasticity. Journal of Econometrics 31~(3), 307 -- 327.
\newline\urlprefix\url{http://www.sciencedirect.com/science/article/pii/0304407686900631}

\bibitem[{Brockwell and Davis(1986)}]{Brockwell:1986:TST:17326}
Brockwell, P.~J., Davis, R.~A., 1986. Time {S}eries: {T}heory and {M}ethods.
  Springer-Verlag New York, Inc., New York, NY, USA.

\bibitem[{B{\"u}hlmann and K{\"u}nsch(1999)}]{buhlmann1999block}
B{\"u}hlmann, P., K{\"u}nsch, H.~R., 1999. Block length selection in the
  bootstrap for time series. Computational Statistics \& Data Analysis 31~(3),
  295--310.
\newline\urlprefix\url{http://www.sciencedirect.com/science/article/pii/S0167947399000146}

\bibitem[{Burman et~al.(1994)Burman, Chow, and Nolan}]{burman1994cross}
Burman, P., Chow, E., Nolan, D., 1994. {A Cross-Validatory Method for Dependent
  Data}. Biometrika 81~(2), 351--358.
\newline\urlprefix\url{http://www.jstor.org/stable/2336965}

\bibitem[{Chen et~al.(2019)Chen, Sun, and Li}]{chen2019semiparametric}
Chen, T., Sun, Y., Li, T.-H., 2019. {A Semi-Parametric Estimation Method for
  the Quantile Spectrum with an Application to Earthquake Classification Using
  Convolutional Neural Network}.

\bibitem[{Davydov(1968)}]{Davydov:1968:CDG}
Davydov, Y.~A., 1968. {Convergence of Distributions Generated by Stationary
  Stochastic Processes}. Theory of Probability and Application 13~(4),
  691--696.
\newline\urlprefix\url{http://dx.doi.org/10.1137/1113086}

\bibitem[{Ding et~al.(1993)Ding, Granger, and Engle}]{ding1993long}
Ding, Z., Granger, C.~W., Engle, R.~F., 1993. A long memory property of stock
  market returns and a new model. Journal of Empirical Finance 1~(1), 83--106.
\newline\urlprefix\url{http://www.sciencedirect.com/science/article/pii/092753989390006D}

\bibitem[{Efron(1979)}]{efron1979}
Efron, B., 01 1979. Bootstrap methods: Another look at the jackknife. Ann.
  Statist. 7~(1), 1--26.
\newline\urlprefix\url{https://doi.org/10.1214/aos/1176344552}

\bibitem[{Fan and Yao(2003)}]{fan03:_nonlin_time_series}
Fan, J., Yao, Q., 2003. Nonlinear {T}ime {S}eries: {N}onparametric and
  {P}arametric {M}ethods. Springer.

\bibitem[{{Georgiou}(2007)}]{georgiou07:_distan_rieman_metric_spect_densit_funct}
{Georgiou}, T.~T., Aug 2007. {Distances and Riemannian Metrics for Spectral
  Density Functions}. IEEE Transactions on Signal Processing 55~(8),
  3995--4003.

\bibitem[{Ghalanos(2020{\natexlab{a}})}]{ghalanos15:_rugarch}
Ghalanos, A., 2020{\natexlab{a}}. Introduction to the rugarch package (Version
  1.4-2).
\newline\urlprefix\url{https://cran.r-project.org/web/packages/rugarch/vignettes/Introduction_to_the_rugarch_package.pdf}

\bibitem[{Ghalanos(2020{\natexlab{b}})}]{ghalanos15_rugarch}
Ghalanos, A., 2020{\natexlab{b}}. {rugarch: Univariate GARCH models.} R package
  version 1.4-2.
\newline\urlprefix\url{https://cran.r-project.org/package=rugarch}

\bibitem[{Hjort and Jones(1996)}]{hjort96:_local}
Hjort, N.~L., Jones, M.~C., 08 1996. Locally parametric nonparametric density
  estimation. Ann. Statist. 24~(4), 1619--1647.
\newline\urlprefix\url{http://dx.doi.org/10.1214/aos/1032298288}

\bibitem[{Horn and Johnson(2012)}]{Horn:2012:MA:2422911}
Horn, R.~A., Johnson, C.~R., 2012. Matrix Analysis, 2nd Edition. Cambridge
  University Press, New York, NY, USA.

\bibitem[{Jordanger and Tj{\o}stheim(2017)}]{jordanger17:_lgcsd}
Jordanger, L.~A., Tj{\o}stheim, D., 2017. Nonlinear cross-spectrum analysis via
  the local gaussian correlation.
\newline\urlprefix\url{https://arxiv.org/abs/1708.02495}

\bibitem[{Klimko and Nelson(1978)}]{klimko1978}
Klimko, L.~A., Nelson, P.~I., 05 1978. {On Conditional Least Squares Estimation
  for Stochastic Processes}. Ann. Statist. 6~(3), 629--642.
\newline\urlprefix\url{http://dx.doi.org/10.1214/aos/1176344207}

\bibitem[{Kullback and Leibler(1951)}]{kullback1951}
Kullback, S., Leibler, R.~A., 03 1951. On information and sufficiency. Ann.
  Math. Statist. 22~(1), 79--86.
\newline\urlprefix\url{http://dx.doi.org/10.1214/aoms/1177729694}

\bibitem[{K{\"u}nsch(1989)}]{kuensch89:_jackk_boots_gener_station_obser}
K{\"u}nsch, H.~R., 1989. {The Jackknife and the Bootstrap for General
  Stationary Observations}. The Annals of Statistics 17~(3), 1217--1241.
\newline\urlprefix\url{http://www.jstor.org/stable/2241719}

\bibitem[{Lahiri et~al.(2007)Lahiri, Furukawa, and
  Lee}]{lahiri2007nonparametric}
Lahiri, S.~N., Furukawa, K., Lee, Y.-D., 2007. A nonparametric plug-in rule for
  selecting optimal block lengths for block bootstrap methods. Statistical
  Methodology 4~(3), 292--321.
\newline\urlprefix\url{http://www.sciencedirect.com/science/article/pii/S1572312706000505}

\bibitem[{Masry and
  Tj{\o}stheim(1995)}]{masry95:_nonpar_estim_ident_nonlin_arch}
Masry, E., Tj{\o}stheim, D., 1995. Nonparametric {E}stimation and
  {I}dentification of {N}onlinear {ARCH} {T}ime {S}eries {S}trong {C}onvergence
  and {A}symptotic {N}ormality: {S}trong {C}onvergence and {A}symptotic
  {N}ormality. Econometric Theory 11~(02), 258--289.
\newline\urlprefix\url{http://EconPapers.repec.org/RePEc:cup:etheor:v:11:y:1995:i:02:p:258-289\_00}

\bibitem[{Nordman and Lahiri(2014)}]{nordman2014convergence}
Nordman, D.~J., Lahiri, S.~N., 05 2014. Convergence rates of empirical block
  length selectors for block bootstrap. Bernoulli 20~(2), 958--978.
\newline\urlprefix\url{http://dx.doi.org/10.3150/13-BEJ511}

\bibitem[{Otneim and Tj{\o}stheim(2017)}]{otneim2017locally}
Otneim, H., Tj{\o}stheim, D., 2017. {The locally Gaussian density estimator for
  multivariate data}. Statistics and Computing 27~(6), 1595--1616.
\newline\urlprefix\url{https://doi.org/10.1007/s11222-016-9706-6}

\bibitem[{Patton et~al.(2009)Patton, Politis, and White}]{patton2009correction}
Patton, A., Politis, D.~N., White, H., 2009. {Correction to “Automatic
  Block-Length Selection for the Dependent Bootstrap” by D. Politis and H.
  White}. Econometric Reviews 28~(4), 372--375.
\newline\urlprefix\url{http://dx.doi.org/10.1080/07474930802459016}

\bibitem[{Politis and Romano(1992)}]{politis92:_gener_resam_schem_trian_array}
Politis, D.~N., Romano, J.~P., 1992. {A General Resampling Scheme for
  Triangular Arrays of $\alpha$-Mixing Random Variables with Application to the
  Problem of Spectral Density Estimation}. The Annals of Statistics 20~(4),
  1985--2007.
\newline\urlprefix\url{http://www.jstor.org/stable/2242377}

\bibitem[{Politis and Romano(1994)}]{politis1994limit}
Politis, D.~N., Romano, J.~P., 1994. {Limit theorems for weakly dependent
  Hilbert space valued random variables with application to the stationary
  bootstrap}. Statistica Sinica 4~(2), 461--476.
\newline\urlprefix\url{http://www.jstor.org/stable/24305527}

\bibitem[{Politis and White(2004)}]{politis2004automatic}
Politis, D.~N., White, H., 2004. {Automatic Block-Length Selection for the
  Dependent Bootstrap}. Econometric Reviews 23~(1), 53--70.
\newline\urlprefix\url{http://dx.doi.org/10.1081/ETC-120028836}

\bibitem[{Racine(2000)}]{racine2000consistent}
Racine, J., 2000. Consistent cross-validatory model-selection for dependent
  data: $hv$-block cross-validation. Journal of Econometrics 99~(1), 39--61.
\newline\urlprefix\url{http://www.sciencedirect.com/science/article/pii/S0304407600000300}

\bibitem[{Shao(1993)}]{shao1993linear}
Shao, J., 1993. {Linear Model Selection by Cross-validation}. Journal of the
  American statistical Association 88~(422), 486--494.
\newline\urlprefix\url{http://www.tandfonline.com/doi/abs/10.1080/01621459.1993.10476299}

\bibitem[{Taniguchi and
  Kakizawa(2000)}]{taniguchi00:_asymp_theor_statis_infer_time_series}
Taniguchi, M., Kakizawa, Y., 2000. Asymptotic {T}heory of {S}tatistical
  {I}nference for {T}ime {S}eries. Springer.

\bibitem[{Tj{\o}stheim and Hufthammer(2013)}]{Tjostheim201333}
Tj{\o}stheim, D., Hufthammer, K.~O., 2013. {Local Gaussian correlation: A new
  measure of dependence}. Journal of Econometrics 172~(1), 33 -- 48.
\newline\urlprefix\url{http://www.sciencedirect.com/science/article/pii/S0304407612001741}

\bibitem[{Volkonskii and Rozanov(1959)}]{Volkonskii:1959:SLT}
Volkonskii, V.~A., Rozanov, Y.~A., 1959. {Some Limit Theorems for Random
  Functions. I}. Theory of Probability and Application 4~(2), 178--197.
\newline\urlprefix\url{http://epubs.siam.org/doi/10.1137/1104015}

\end{thebibliography}
\end{document}